\documentclass{book}

\usepackage{amssymb,amsmath,amsthm}
\usepackage{mathrsfs}
\usepackage{enumerate}
\usepackage[scr=boondox]{mathalpha}         

\usepackage{makeidx} 

\usepackage{graphicx}        
\usepackage{hyperref}

\usepackage{a4wide}

\theoremstyle{plain}
\newtheorem{thm}{Theorem}[chapter]
\newtheorem{prop}[thm]{Proposition}
\newtheorem{lemma}[thm]{Lemma}
\newtheorem{cor}[thm]{Corollary}
\theoremstyle{definition}
\newtheorem{definition}[thm]{Definition}
\newtheorem{remark}[thm]{Remark}

\newcommand{\tn}[1]{\ensuremath{\mathbb{T}^{#1}}}

\newcommand{\rn}[1]{\ensuremath{\mathbb{R}^{#1}}}
\newcommand{\nn}[1]{\ensuremath{\mathbb{N}^{#1}}}

\newcommand{\sn}[1]{\ensuremath{\mathbb{S}^{#1}}}

\newcommand{\ro}{\mathbb{R}}
\newcommand{\ldr}[1]{\langle #1\rangle}
\newcommand{\g}{\gamma}
\newcommand{\de}{\delta}
\renewcommand{\a}{\alpha}

\renewcommand{\b}{\beta}
\newcommand{\mfE}{\mathfrak{E}}
\newcommand{\mfD}{\mathfrak{D}}
\newcommand{\mfP}{\mathfrak{P}}
\newcommand{\mfR}{\mathfrak{R}}
\newcommand{\mfg}{\mathfrak{g}}
\newcommand{\mfh}{\mathfrak{h}}
\newcommand{\mff}{\mathfrak{f}}
\newcommand{\mfT}{\mathfrak{T}}
\newcommand{\mft}{\mathfrak{t}}
\newcommand{\mfI}{\mathfrak{I}}

\newcommand{\mfB}{\mathfrak{B}}

\newcommand{\mfS}{\mathfrak{S}}
\newcommand{\mfs}{\mathfrak{s}}
\newcommand{\mfr}{\mathfrak{r}}

\newcommand{\mfG}{\mathfrak{G}}

\newcommand{\bq}{\bar{q}}
\newcommand{\bA}{\bar{A}}
\newcommand{\fbar}{\bar{f}}
\newcommand{\be}{\bar{e}}

\newcommand{\e}{\epsilon}
\newcommand{\bK}{\bar{K}}
\newcommand{\bn}{\bar{n}}

\newcommand{\bnu}{\bar{\nu}}

\newcommand{\bp}{\bar{p}}
\newcommand{\bR}{\bar{R}}

\newcommand{\biota}{\bar{\iota}}
\newcommand{\bvarphi}{\bar{\varphi}}
\newcommand{\bth}{\bar{\theta}}

\newcommand{\bvka}{\bar{\varkappa}}
\newcommand{\bka}{\bar{\kappa}}
\newcommand{\bPhi}{\bar{\Phi}}

\newcommand{\bOmega}{\bar{\Omega}}

\newcommand{\bo}{\bar{o}}
\newcommand{\bP}{\bar{P}}
\newcommand{\bphi}{\bar{\phi}}
\newcommand{\bpsi}{\bar{\psi}}
\newcommand{\bsi}{\bar{\sigma}}

\newcommand{\bS}{\bar{S}}
\newcommand{\bge}{\bar{g}}
\newcommand{\bk}{\bar{k}}
\newcommand{\bM}{\bar{M}}
\newcommand{\bx}{\bar{x}}

\newcommand{\sgn}{\mathrm{sgn}}
\newcommand{\roRic}{\mathrm{Ric}}
\newcommand{\ovRic}{\overline{\mathrm{Ric}}}

\newcommand{\rosc}{\mathrm{sc}}
\newcommand{\roqu}{\mathrm{qu}}
\newcommand{\iso}{\mathrm{iso}}

\newcommand{\roId}{\mathrm{Id}}

\newcommand{\roper}{\mathrm{per}}
\newcommand{\ropar}{\mathrm{par}}
\newcommand{\ropre}{\mathrm{pre}}
\newcommand{\rogen}{\mathrm{gen}}
\newcommand{\roLRS}{\mathrm{LRS}}

\newcommand{\roap}{\mathrm{ap}}
\newcommand{\ropp}{\mathrm{pp}}
\newcommand{\roev}{\mathrm{ev}}
\newcommand{\roodd}{\mathrm{odd}}

\newcommand{\rodiv}{\mathrm{div}}
\newcommand{\rodS}{\mathrm{dS}}

\newcommand{\rod}{\mathrm{d}}

\newcommand{\rodiag}{\mathrm{diag}}

\newcommand{\rotr}{\mathrm{tr}}
\newcommand{\roc}{\mathrm{c}}
\newcommand{\rond}{\mathrm{nd}}

\newcommand{\roloc}{\mathrm{loc}}

\newcommand{\mrI}{\mathrm{I}}
\newcommand{\mrII}{\mathrm{II}}
\newcommand{\mrVIz}{\mathrm{VI}_{0}}
\newcommand{\mrVIIz}{\mathrm{VII}_{0}}
\newcommand{\mrVIII}{\mathrm{VIII}}
\newcommand{\mrIX}{\mathrm{IX}}
\newcommand{\bah}{\bar{h}}
\newcommand{\msH}{\mathscr{H}}

\newcommand{\hf}{\hat{f}}

\newcommand{\mV}{\mathcal{V}}
\newcommand{\mZ}{\mathcal{Z}}

\newcommand{\ma}{\mathcal{A}}
\newcommand{\mB}{\mathcal{B}}
\newcommand{\mN}{\mathcal{N}}

\newcommand{\mS}{\mathcal{S}}
\newcommand{\mC}{\mathcal{C}}
\newcommand{\mH}{\mathcal{H}}
\newcommand{\md}{\mathcal{D}}
\newcommand{\mD}{\mathcal{D}}

\newcommand{\me}{\mathcal{E}}
\newcommand{\mX}{\mathcal{X}}
\newcommand{\mfX}{\mathfrak{X}}

\newcommand{\mK}{\mathcal{K}}

\newcommand{\msK}{\mathscr{K}}

\newcommand{\sfA}{\mathsf{A}}
\newcommand{\sfB}{\mathsf{B}}
\newcommand{\sfG}{\mathsf{G}}

\newcommand{\sfQ}{\mathsf{Q}}
\newcommand{\sfF}{\mathsf{F}}
\newcommand{\sfP}{\mathsf{P}}
\newcommand{\sfN}{\mathsf{N}}
\newcommand{\sfS}{\mathsf{S}}
\newcommand{\sfU}{\mathsf{U}}
\newcommand{\sfV}{\mathsf{V}}
\newcommand{\sfW}{\mathsf{W}}
\newcommand{\sfR}{\mathsf{R}}

\newcommand{\tr}{\mathrm{tr}}
\renewcommand{\d}{\partial}
\newcommand{\bfI}{\mathbf{I}}

\newcommand{\Sp}{\Sigma_{+}}
\newcommand{\Sm}{\Sigma_{-}}
\newcommand{\No}{N_{1}}
\newcommand{\bN}{\bar{N}}
\newcommand{\bNo}{\bar{N}_{1}}
\newcommand{\Nt}{N_{2}}
\newcommand{\bNt}{\bar{N}_{2}}
\newcommand{\Nth}{N_{3}}
\newcommand{\bNth}{\bar{N}_{3}}

\newcommand{\bSp}{\bar{\Sigma}_{+}}

\newcommand{\bSm}{\bar{\Sigma}_{-}}

\makeindex

\makeindex

\begin{document}

\author{Hans Ringstr\"{o}m}
\title{On the structure of big bang singularities in spatially homogenous solutions to the Einstein non-linear scalar field equations}

\maketitle

\frontmatter

\include{dedicbook}

\setcounter{tocdepth}{1}
\tableofcontents

\mainmatter

\chapter{Introduction}

The subject of this article is the structure of big bang singularities in spatially homogeneous solutions to the \textit{Einstein non-linear scalar
  field equations}.
\index{Einstein non-linear scalar field equations}%
Given a \textit{potential} $V\in C^{\infty}(\ro)$,
\index{Potential}%
we are, in other words, interested in solving
\begin{subequations}\label{seq:ENLSFE}
  \begin{align}
    G = & T,\\
    \Box_{g}\phi = & V'\circ\phi
  \end{align}
\end{subequations}
for a Lorentz manifold $(M,g)$ and a function $\phi\in C^{\infty}(M)$, where $\Box_{g}$ is the wave operator associated with $g$, 
\begin{align*}
  G := & \mathrm{Ric}-\tfrac{1}{2}Sg,\\
  T := & d\phi\otimes d\phi-\left[\tfrac{1}{2}|d\phi|_g^2+V\circ\phi\right]g.
\end{align*}
Here $\mathrm{Ric}$ and $S$ denote the Ricci tensor and scalar curvature of $g$ respectively. Moreover,  $\phi$ is referred to as the \textit{scalar field};
\index{Scalar field}%
$G$ as the \textit{Einstein tensor};
\index{Einstein tensor}%
and $T$ as the \textit{stress energy tensor}.
\index{Stress energy tensor}%
Solutions to (\ref{seq:ENLSFE}), or systems obtained by including additional matter sources, have been studied by many authors. The reason is that they
can model both inflation and dark energy. However, the questions we consider in this article differ from the ones normally asked by physicists.
Our main goal in this article is to obtain a global picture of the dynamics in the spatially homogeneous 
setting. Using the corresponding results, one can prove that there are large classes of spatially locally homogeneous solutions that are globally
non-linearly stable (in the absence of symmetries), both to the future and to the past. However, our main motivation stems from a desire to obtain
a more general understanding of the structure of big bang singularities, as we now explain. 

\textbf{The structure of cosmological singularities; the BKL proposal and initial data on the big bang.} One important proposal concerning the structure
of big bang singularities is due to Belinski\v{\i}, Khalatnikov and Lifschitz, BKL for short; see, e.g., \cite{bkl1,bkl15,bkl2}. Very roughly speaking,
their main hypothesis is that big bang singularities are local, spacelike and oscillatory. In particular, they suggest that ``spatial derivatives''
are asymptotically negligible, so that the evolution equations reduce to a family of ODE's; that observers going into the singularity typically
lose the ability to communicate; and that the eigenvalues of the expansion normalised Weingarten maps of an appropriate foliation oscillate according
to a chaotic one-dimensional map. Recall here that the Weingarten map of a spacelike hypersurface is the second fundamental form with one index raised
by the induced metric. Moreover, the expansion normalised Weingarten map is the Weingarten map divided by the mean curvature. There are, however,
exceptions to this hypothesis. For example, particular types of matter, such as a stiff fluid or a scalar field, are expected to suppress the
oscillations so that the expansion normalised Weingarten map converges. The corresponding setting is referred to as convergent or quiescent. Non-linear
scalar fields belong to this category of matter models, and we are here therefore mainly interested in the quiescent setting. Several attempts to
formalise the ideas of BKL have been formulated in the literature; see, e.g., \cite{bkl1,bkl15,bkl2,ELS,JimandVince,luw,huar,HUL,dadB,dhn,dah} and
the references cited therein. In \cite{ELS,JimandVince}, the authors develop the idea of a solution being \textit{asymptotically velocity term dominated
  near the singularity} (AVTDS); a solution is AVTDS if it asymptotes to a solution of the so-called \textit{velocity term dominated} (VTD) system.
Roughly speaking, the VTD system is a truncated system of equations obtained by dropping some of the spatial derivatives. However, it is unfortunately
not geometric in nature. In \cite{RinQC}, a substitute in the form of a geometric
notion of initial data on the singularity is introduced; see Definition~\ref{def:ndvacidonbbssh} for a slight generalisation in the present setting.
In the special case of a Gaussian foliation, a solution to the VTD system can be thought of as initial data on the singularity; see \cite{RinQC} for
a justification of this statement. However, the notion of initial data on the singularity is not tied to Gaussian foliations. In fact, beyond 
covariance, this notion has several advantages: that a development induces data on the singularity can be interpreted as a covariant formulation
of the statement that it is AVTDS; if there are initial data on the singularity in a certain symmetry class, it is natural to conjecture that generic
solutions induce data on the singularity in this symmetry class, so that generic solutions are quiescent; if there are no initial data on the
singularity in a given symmetry class, generic solutions can be expected to be oscillatory in this symmetry class; and in quiescent settings, one
could ideally hope to parametrise solutions by data on the singularity. There are several caveats to these expectations. However, it is clear that
initial data on the singularity is a natural starting point for understanding the structure of singularities, and this is the perspective we take in
this article. 

\textbf{The map from initial data on the singularity to developments.}
Given a notion of initial data on the singularity, the first question to ask is: do data uniquely determine corresponding developments? That the answer
is yes in the $3+1$-dimensional Einstein non-linear scalar field setting is verified in \cite{andres}, a result that generalises
\cite{aarendall,klinger,fal}. The results of \cite{aarendall,klinger,fal,andres} are all fundamental. On the other hand, they do not provide a global
picture of the dynamics. In fact, it can be argued
that they do not even give a local picture of the dynamics in the sense that it is not clear from these results that there is an open set of regular
initial data whose developments give rise to data on the singularity. The reason for this deficiency is that the results of, e.g.,
\cite{aarendall,klinger,fal,andres} do not include any statements concerning the regularity of the map from data on the singularity to the corresponding
developments. If one were, e.g., able to prove that the map, say $F$, from the set of isometry classes of initial data on the singularity to the set of
isometry classes of developments is a local diffeomorphism, one could immediately conclude, e.g., that developments that induce data on the singularity
are stable and that there is an open set of regular initial data whose developments induce data on the singularity. If $F$ is a local diffeomorphism, one
could also argue that exceptional subsets of initial data on the singularity correspond to exceptional
subsets of developments. Turning to the global properties of $F$, the ideal situation is when $F$ is
a diffeomorphism. In that setting, isometry classes of initial data on the singularity do not only provide optimal asymptotic information, they
also parametrise the set of isometry classes of developments. However, even when $F$ is not a global diffeomorphism, the image of $F$ is a natural
reference and one can focus on the complement of the image of $F$. This complement might, e.g., be non-generic, and thereby irrelevant; it could include
an open set of solutions that do not have crushing singularities and are therefore irrelevant in the analysis of singularities etc. If one can verify that
the topology of the set of isometry classes of data on the singularity differs from the topology of the set of isometry classes of developments, it is
clear that $F$ cannot be a diffeomorphism. This can be used to make deductions concerning the global dynamics of solutions; below we provide
examples of this. 

\textbf{Spikes.}
In the spatially homogeneous setting, it is natural to conjecture generic developments to belong to the image of $F$. We prove several such results
in this article. However, in the spatially inhomogeneous setting, this is not to be expected. This is partly related to what is referred to as
``spikes'' in the $\tn{3}$-Gowdy symmetric setting. The reason for the terminology is that the metric is invariant under an effective action of
$\mathbb{T}^2$ in this setting, with the consequence that the metric components only depend on $(t,\vartheta)\in (0,\infty)\times\sn{1}$ (where the
$t$-coordinate of a spacetime point $\xi$ is the area of the orbit of $\xi$ under the action of $\tn{2}$). Moreover, the limits, in the direction of
the singularity, of the eigenvalues of the expansion normalised Weingarten maps of the constant-$t$ hypersurfaces are functions on $\sn{1}$. The
limits are typically smooth, except for a discrete set of elements of $\sn{1}$ at which they make jumps. Since the limits are smooth when removing the
discrete set of points, and make finite jumps at the exceptional points, the discontinuities are called spikes. However, due to the symmetry reduction,
the spikes actually correspond to $2$-tori on the singularity (when representing the singularity as a $3$-torus). From that point of view, it is more
natural to think of the spikes as ``scars''; two dimensional hypersurfaces on the singularity. In the $\tn{3}$-Gowdy symmetric setting, the
solutions with spikes can be built from developments arising from data on the singularity; see \cite{raw}. However, even more generally, it is natural
to expect scars to be subsets of the zero sets of the limits of structure coefficients associated with suitably normalised eigenvector fields of the
expansion
normalised Weingarten map. If one can demonstrate that the relevant limits are sufficiently regular and generically have the property that their
derivatives are non-zero when they vanish, one can hope to prove that the scars correspond to hypersurfaces more generally. Note, however, that any
such analysis is crucially dependent on understanding the properties of the map from regular initial data to asymptotic data. Note also
that if one ignores the spikes, generic $\tn{3}$-Gowdy symmetric vacuum solutions induce data on the remaining part of the singularity; see
\cite{RinQCSymm}. To summarise, even in the presence of spikes, it can be expected to be of crucial importance to understand the regularity properties
of the map from regular data to data on the singularity. 

\textbf{The attractor.} Concerning oscillatory solutions, there are only partial results in the spatially homogeneous setting, and no results in the
spatially inhomogeneous setting. However, there are some features of the analysis in the spatially homogeneous setting that may be relevant more
generally. For instance, some of the results in the spatially homogeneous setting are based on a detailed analysis of how solutions approach the
Kasner circle along a stable manifold and then diverge away from it along an unstable manifold. The first step of the analysis is thus to understand
what the stable and unstable manifolds are. In the spatially homogeneous setting, this is straightforward. However, in the spatially inhomogeneous
setting, it is more complicated. Nevertheless, it is quite conceivable that they can be obtained as sets of developments corresponding to initial
data on the singularity as in \cite{fal}, \cite{klinger} or \cite{andres}. However, to understand the corresponding subset of regular initial data, we
need to understand the regularity of the map $F$. Similarly, but more ambitiously, one can hope to construct an attractor. Again, this can be expected
to require an understanding of the regularity of the map $F$.

\textbf{Motivation, goals.} With the above observations in mind, it is clear that it is not only of interest to understand the structure of cosmological
singularities in the non-linear scalar field setting. It is of particular interest to do so from the perspective of initial data on the singularity.
Moreover, it is of central interest to understand the regularity of the map from initial data on the singularity to developments. The main purpose of
this article is to provide a global picture of the dynamics in the spatially homogenous setting, with particular focus on the regularity properties of
the map $F$. This yields a detailed analysis in a model case which can then be used as a reference when turning to situations with less symmetry.
Here we mainly focus on Bianchi class A solutions; i.e., the case that the relevant spacetimes arise from left invariant initial data on unimodular
$3$-dimensional Lie groups. However, we also discuss the $k=-1$ FLRW setting. 

Concerning the Einstein scalar field case, \cite{RinQCSymm} contains a demonstration of the fact that Bianchi class A solutions induce data on the
singularity unless they are vacuum solutions of type VIII or IX; see Table~\ref{table:bianchiA} for the relevant classification. Moreover, vacuum Bianchi
type VIII and IX solutions are either oscillatory or have a Cauchy horizon. On the other hand, \cite{RinQCSymm} does not contain a discussion of the
regularity of the map $F$. Moreover, even though the scalar field setting is of interest, the non-linear scalar field setting is more important, since
it, in the physics literature, is used to model inflation and used as a mechanism to generate accelerated expansion in the late universe. Note, however,
that when physicists study inflation, they are not interested in the asymptotics in the direction in the singularity. Instead, they are interested in
a regime which is close to the singularity but distinct from it. Moreover, in the expanding direction, the interesting situation is the transition from
a matter dominated phase to a phase where the dark energy takes over. Since we are here interested in the asymptotics in the direction of the singularity
and not in an intermediate phase, our analysis is not directly relevant to the study of inflation, and since we do not include additional matter, our solutions do
not model the transition from a matter dominated phase to where the dark energy takes over. On the other hand, due to our analysis, we are able to prove
global non-linear stability of a large class of solutions which are good models of the early and late universe. Moreover, the analysis is a natural starting
point for including mattter. 

As big bang formation has been analysed in the Einstein scalar field setting, one might expect analysing the past asymptotics to be
straightforward in the non-linear scalar field setting. However, the presence of the potential destroys a key monotonicity property. This makes the analysis
substantially more complicated. 

\textbf{Previous results.} Since the literature in the non-linear scalar field setting is vast, it is not possible to give an exhaustive list of
references. However, \cite{foster,KaMZ,KaM,reninflation,aslowroll,akessence,hayoung} are some examples of articles on this topic containing
mathematical results. See also the book \cite{coley}. 
More recent results in the spatially flat FLRW setting with matter, using a dynamical systems perspective, are to be found in \cite{ahu,alu,aauz,aauo}. 
We refer the interested reader to these articles as a starting point for entering the literature. The goal of many of the mathematical studies
is to model the accelerated expansion of the universe. A natural second step, once the future asymptotics are understood, is then to prove future
global non-linear stability of the corresponding solutions. Some examples of this are given by \cite{HaR,RinInv,RinPL,s,stab,FSF,LaI}. On the other
hand, to the best of our knowledge, the asymptotics in the direction of the singularity in the anisotropic setting have not been analysed in detail,
and this is the main purpose of the present article. Combining the detailed analysis with previous future global non-linear stability results and
a recent result concerning stable big bang formation then yields a large class of solutions that are past and future globally non-linearly stable.

\textbf{Results, outline of the introduction.} In order to derive the regularity propeties of the map $F$, we first need to endow isometry classes of data
on the singularity with a smooth structure and similarly for isometry classes of developments. This requires the introduction of a substantial amount of
terminology. In Section~\ref{ssection:initialdata}, we introduce the notion of regular Bianchi class A initial data. We also discuss corresponding
symmetry types; isotropic, locally rotationally symmetric etc. Fixing not only the Bianchi type, but also a symmetry type, turns out to be crucial
in order for it to be possible to endow the corresponding set of isometry classes of initial data with, say, fixed mean curvature with a natural
smooth structure. In Section~\ref{section:data on the singularity}, we introduce a notion of data on the singularity in the Bianchi class A setting,
see Definition~\ref{def:ndvacidonbbssh}. This notion is of course more restrictive than the one introduced in \cite{RinQC}, since it concerns a
spatially homogeneous setting. On the other hand, it is also more general in that we here do not require the eigenvalues of the limit of the expansion
normalised Weingarten map to be distinct. Moreover, we include data corresponding to Cauchy horizons as one possibility in
Definition~\ref{def:ndvacidonbbssh}. Again, we introduce symmetry types in this setting and notation for the relevant sets of initial data on the
singularity. In Section~\ref{ssection:developments} we turn to developments. In the setting of interest here, there are natural foliations of
the developments of Bianchi class A initial data. They are both Gaussian (the curves with constant space-coordinate parametrise unit timelike
geodesics orthogonal to the leaves of the foliation) and CMC (constant mean curvature). It is convenient to work with these developments. We
introduce them and the relevant terminology in Section~\ref{ssection:developments}. We also prove existence of developments, given regular initial
data and deduce their basic properties. In Subsection~\ref{ssection:dev ind data on sing}, we clarify what it means for a development to induce
data on the singularity; see Definition~\ref{def:ind data on sing}. We also verify that the developments of isometric regular initial data induce
isometric data on the singularity (assuming that they induce data on the singularity). In Theorem~\ref{thm:dataonsingtosolution} we prove that, under
suitable assumptions concerning the potential, data on the singularity give rise to unique developments. Moreover, we verify that isometric data on the
singularity give rise to isometeric developments. Before analysing the asymptotics in the direction of the singularity, it is important to keep in
mind that Bianchi type IX solutions are special in that
they include a large range of dynamics: solutions can induce data on the singularity; exhibit chaotic oscillations in the direction of the singularity;
have Cauchy horizons through which the solution can be extended in inequivalent ways; expand both to the future and the past (as in de Sitter space);
and have a static geometry (as in the Einstein static universe). If one is interested in studying cosmological singularities, the latter two
possibilities have to be excluded in some way. In Subsection~\ref{ssection:Bianchi type IX setting}, we introduce conditions that serve this purpose.
In Subsection~\ref{ssection:iso class dev}, we are then in a position to introduce appropriate terminology for isometry classes of developments. 

In Section~\ref{section:parametrising iso class dev}, we turn to the problem of endowing the set of isometry classes of developments with a smooth
structure. When doing so, we always focus on the simply connected setting; due to the form of the isometries induced on developments by isometries
of data, see Propositions~\ref{prop:iso id to iso sol} and \ref{prop:iso idos to iso dev}, one can always take quotients later. Next, a natural starting
point is to consider isometry classes of simply connected regular initial data of a fixed Bianchi type $\mfT$ (see Table~\ref{table:bianchiA}), fixed
symmetry type $\mfs$ (isotropic, locally rotationally symmetric etc.), fixed
mean curvature $\vartheta$ and corresponding to a potential $V$. We denote this set by ${}^{\rosc}\mfB_\mfT^\mfs[V](\vartheta)$ (in practice, we also
exclude certain exceptional cases, such as initial data for Minkowski space). This set can be endowed with a smooth structure; see
Section~\ref{ssection:param solns} for details. If, for a given Bianchi type, all developments have the property that the mean curvature is strictly
monotonically decaying, then the natural map from ${}^{\rosc}\mfB_\mfT^\mfs[V](\vartheta)$ to ${}^{\rosc}\mfD_\mfT^\mfs[V]$ (the set of isometry classes of
developments corresponding to simply connected regular initial data of a fixed Bianchi type $\mfT$, fixed symmetry type $\mfs$ and corresponding to a
potential $V$) is an injection, say $\iota_\vartheta$. Whatever smooth structure one wishes to endow ${}^{\rosc}\mfD_\mfT^\mfs[V]$ with, it would be desirable
if this injection were a smooth diffeomorphism onto its image. In fact, the union of the images of the $\iota_\vartheta$ yields ${}^{\rosc}\mfD_\mfT^\mfs[V]$.
This means that one can endow ${}^{\rosc}\mfD_\mfT^\mfs[V]$ with a smooth structure by first taking the disjoint union of the images of all the
$\iota_\vartheta$ and then taking the quotient of the union by the equivalence relation induced by the Einstein flow. Needless to say, for this to make
sense, the maps between subsets of isometry classes of regular initial data with distinct mean curvatures induced by the Einstein flow have to be
diffeomorphisms. Proving this statement does require an effort. Note also that, in favourable situations, there is one mean curvature, say $\vartheta_0$,
that is attained by all developments. In that case, ${}^{\rosc}\mfB_\mfT^\mfs[V](\vartheta_0)$ can be used to parametrise ${}^{\rosc}\mfD_\mfT^\mfs[V]$.
However, for general potentials, there is no such $\vartheta_0$, and the more complicated construction mentioned above becomes necessary. There are additional
complications in the Bianchi type IX and isotropic Bianchi type I cases. In particular, the above discussion does not quite apply; we refer the reader to
Section~\ref{section:parametrising iso class dev} for the details. 

Once we have introduced a smooth structure on the isometry classes of developments and isometry classes of data on the singularity, we can discuss the
regularity of the map from isometry classes of developments (that induce data on the singularity) to isometry classes of data on the singularity. In fact,
in Section~\ref{section:map iso dev to sdata}, we prove that this map is a diffeomorphism for several choices of Bianchi and symmetry type. Up to this point,
one could naively hope to generalise the results to much more general classes of spacetimes (with less or even no symmetry). However, it would be optimistic
to hope to obtain a global picture of the dynamics in general. In the Bianchi class A setting, we discuss the global dynamics in Section~\ref{ssection:asdirofsing}.
To begin with, we note that there is a dichotomy in the direction of the singularity. Either the solution is \textit{vacuum dominated} in the sense that
$\phi_t/\theta$ converges to zero (where $t$ denotes the Gaussian time coordinate, i.e. it measures proper time to the singularity, and $\theta$ denotes
the mean curvature of the constant-$t$ hypersurfaces)  or \textit{matter dominated} in the sense that $\phi_t/\theta$ converges to a non-zero limit in the
direction of the
singularity. The first result concerning the asymptotics is then contained in Theorem~\ref{thm:dev inducing data on the sing}. The conclusion is roughly
speaking the following: Ignoring isotropic Bianchi type I solutions for the moment, the remaining solutions induce data on the singularity unless they
are vacuum dominated, of Bianchi type VIII or IX and generic (in the sense that they are neither isotropic nor locally rotationally symmetric). Moreover,
the developments that do not induce data on the singularity are oscillatory. As a corollary, we can, e.g., conclude that for many choices of Bianchi 
and symmetry type, all developments induce data on the singularity, so that the map from isometry classes of developments to isometry classes of data on
the singularity is a global diffeomorphism. In the remaining cases, it is sufficient to remove the developments with oscillatory singularities in order
to obtain a diffeomorphism.

The isotropic Bianchi type I setting is different and requires a separate treatment. One issue is the fact that there are isotropic Bianchi type I
solutions that do not have crushing singularities. However, even excluding those, there is an additional complication. In order to explain it, note first
that in the isotropic Bianchi type I setting, we impose stricter conditions on the potential, so that there is a range of mean curvatures that is attained
by the CMC foliations of all the developments that have a crushing singularity. This means that there is a $\vartheta_0$ such that the isometry classes of
simply connected isotropic Bianchi type I developments with a crusing singularity and corresponding to a potential $V$, say
${}^{\rosc}\mfD_{\mrI,\roc}^{\iso}[V]$, can actually be identified with ${}^{\rosc}\mfB_{\mrI}^{\iso}[V](\vartheta_0)$. Naively, one would expect the Einstein flow
to define a diffeomorphism from ${}^{\rosc}\mfB_{\mrI}^{\iso}[V](\vartheta_0)$
to the corresponding set of isometry classes of initial data on the singularity, say ${}^{\rosc}\mfS_{\mrI}^{\iso}$. However, there is a problem: these two sets
typically do not even have the same topology. To begin with ${}^{\rosc}\mfS_{\mrI}^{\iso}$ is diffeomorphic to two copies of $\ro$. If the potential is bounded,
the same is true of ${}^{\rosc}\mfB_{\mrI}^{\iso}[V](\vartheta_0)$. In that case, one can therefore hope that all developments induce data on the singularity
and that the natural map from ${}^{\rosc}\mfB_{\mrI}^{\iso}[V](\vartheta_0)$ to ${}^{\rosc}\mfS_{\mrI}^{\iso}$ is a diffeomorphism. This is exactly what happens.
On the other hand, if the potential is bounded in one direction and unbounded in the other, ${}^{\rosc}\mfB_{\mrI}^{\iso}[V](\vartheta_0)$ is diffeomorphic to
$\ro$. There must thus be an exceptional solution that does not induce data on the singularity. Under suitable conditions on the potential, that is
exactly what happens. Finally, if the potential is unbounded in both directions, then ${}^{\rosc}\mfB_{\mrI}^{\iso}[V](\vartheta_0)$ is diffeomorphic to
$\sn{1}$. There must thus be two exceptional solutions that do not induce data on the singularity. Under suitable conditions on the potential, that is
exactly what happens. In this sense, we obtain conclusions concerning the dynamics by merely studying the topologies of the set of isometry classes
of data on the singularity and the set of isometry classes of developments. 

In Section~\ref{section:as sing k minus one intro}, we turn to the spatially homogeneous and isotropic negative spatial scalar curvature setting; i.e.,
the $k=-1$ FLRW setting. The motivation for doing this is that we wish to cover all the FLRW settings. We begin by defining the relevant notion of
regular initial data. In this case there is, in analogy with the isotropic Bianchi type I setting, a notion of trivial regular initial data. These are
the data that give rise to the Milne model or analogous solutions with a positive cosmological constant. In
Proposition~\ref{prop:Milne lambda as k minus one}, we prove that for every $\phi_\infty$, there is a solution such that $\phi(t)\rightarrow\phi_\infty$
and such that the geometry asymptotes to a vacuum solution with cosmological constant $\Lambda=V(\phi_\infty)$ (here we assume $V\geq 0$). It is important
to note that these solutions do not induce data on the singularity. In fact, $\bS/\theta^2$ does not converge to zero for these solutions, where
$\bS$ denotes the scalar curvature of the leaves of spatial homogeneity. In Definitions~\ref{def:ndvacidonbbssh k minus one} and
\ref{def:k eq minus one dev}, we then introduce the notions of data on the singularity and development of regular initial data. In
Definition~\ref{def:ind data on sing k negative} we clarify what it means for developments to induce data on the singularity and in
Proposition~\ref{prop:dos ind dev k minus one}, we verify that, given data on the singularity, there are unique developments inducing the
prescribed initial data. Next, in Proposition~\ref{prop:two outcomes intro k minus one}, we verify that developments either induce data on the
singularity or have Milne like asymptotics (or the analogous asymptotics with a positive cosmological constant). Finally, in
Proposition~\ref{prop:generic k minus one intro}, we verify that the latter case is Lebesgue and Baire non-generic.

In Section~\ref{ssection:futureasymptotics}, we then turn to the problem of deriving future asymptotics. In the case of an exponential potential,
this has already been done; see \cite{RinPL} and references cited therein. Under suitable assumptions on the potential, in particular that it has a
strictly positive lower bound,
we derive detailed asymptotics for generic Bianchi class A solutions in Proposition~\ref{prop:futureassteponeintro}. However, we exclude Bianchi
type IX solutions, since such solutions can both expand and contract in the future direction, and isolating the boundary between the two possibilities
is a subtle issue we do not sort out in this article. Moreover, we prove that the corresponding spatially locally homogeneous and spatially closed
solutions are future globally non-linearly stable in Corollary~\ref{cor:stabilitylb}. In Proposition~\ref{prop:generic k minus one intro exp},
we carry out an analogous analysis in the $k=-1$ FLRW setting.

In Section~\ref{ssection:sphom} we give a motivation of why we focus on Bianchi class A solutions. Finally, in Section~\ref{section:outline article},
we give an outline of the article. Section~\ref{section:acknoledgements} contains the acknowledgements.   

\section{Regular initial data}\label{ssection:initialdata}
In the present article, we mainly restrict our attention to solutions arising from the following type of initial data. 
\begin{definition}\label{def:Bianchiid}
  \textit{Bianchi class A initial data for the Einstein non-linear scalar field equations},
  \index{Bianchi class A!Initial data}%
  \index{Bianchi class A!Regular initial data}%
  \index{Einstein non-linear scalar field equations!Bianchi class A initial data}%
  \index{Initial data!Regular!Bianchi class A}%
  \index{Initial data!Bianchi class A}%
  \index{Initial data!Regular}%
  with potential $V\in C^{\infty}(\rn{})$, consist of the
  following: a connected $3$-dimensional unimodular Lie group $G$; a left invariant metric $\bge$ on $G$; a left invariant symmetric covariant
  $2$-tensor field $\bk$ on $G$; and two constants $\bphi_{0}$ and $\bphi_{1}$ satisfying the \textit{constraint equations}
  \index{Constraint equations}%
  \begin{subequations}\label{seq:constraintsBA}
    \begin{align}
      \bS-|\bk|_{\bge}^{2}+(\rotr_{\bge}\bk)^{2} = & \bphi_{1}^{2}+2V(\bphi_{0}),\label{eq:ham con original}\\
      d\rotr_{\bge}\bk-\rodiv_{\bge}\bk = & 0
    \end{align}
  \end{subequations}
  (these equations are also referred to as the \textit{Hamiltonian} and \textit{momentum constraint equations} respectively).
  \index{Constraint equations!Hamiltonian}%
  \index{Constraint equations!Momentum}%
  \index{Hamiltonian constraint}%
  \index{Momentum constraint}%
  Here $\bS$ denotes the scalar curvature associated with $\bge$. The data are said to be \textit{trivial}
  \index{Initial data!Trivial}%
  \index{Trivial initial data}%
  if $\bge$ is flat, $3\bk=(\tr_{\bge}\bk)\bge$, $\bphi_1=0$ and $V'(\bphi_0)=0$. 
\end{definition}
\begin{remark}\label{remark:unimodular}
  In order to define the notion of unimodularity, let $G$ be a Lie group and $\mfg$ the associated Lie algebra. Given $X\in \mfg$, define
  $\mathrm{ad}_{X}:\mfg\rightarrow\mfg$ by $\mathrm{ad}_{X}(Y)=[X,Y]$. Let $\eta_{G}\in \mfg^{*}$ be defined by $\eta_{G}(X)=\rotr\ \mathrm{ad}_{X}$.
  Then $G$ is \textit{unimodular}
  \index{Unimodular!Lie group}%
  \index{Lie group!Unimodular}%
  if $\eta_{G}=0$ and \textit{non-unimodular}
  \index{Non-unimodular!Lie group}%
  \index{Lie group!Non-unimodular}%
  if $\eta_{G}\neq 0$.
\end{remark}
  Let $G$ be a $3$-dimensional unimodular Lie group, $\mfg$ be the corresponding Lie algebra and $\{e_{i}\}$ be a basis of $\mfg$. Define $\g^{k}_{ij}$
  by the requirement that $[e_{i},e_{j}]=\g^{k}_{ij}e_{k}$. The condition of unimodularity implies that there is a unique symmetric matrix $\bn$ with
  components $\bn^{ij}$ such that $\g_{ij}^{k}=\e_{ijl}\bn^{lk}$, where $\e_{123}=1$ and $\e_{ijk}$ is antisymmetric in all of its indices; see \cite{ketal}
  for original references and, e.g., \cite[Lemma~19.3, p.~206]{RinCauchy} for a textbook presentation. It is convenient to introduce the following
  terminology.
  \begin{definition}\label{def:comm and Wein matr}
    Let $G$ be a connected $3$-dimensional unimodular Lie group and $\{e_i\}$ be a basis of the corresponding Lie algebra $\mfg$. Then $\bn$, defined
    above, is called the \textit{commutator matrix}
    \index{Commutator matrix}%
    associated with $\{e_i\}$. If, in addition, $V\in C^{\infty}(\ro)$, $(G,\bge,\bk,\bphi_0,\bphi_1)$
    are Bianchi class A initial data for the Einstein non-linear scalar field equations and $\{e_i\}$ is orthonormal with respect to $\bge$, then the
    matrix $K$ with components $\bk(e_i,e_j)$ is called the \textit{Weingarten matrix}
    \index{Weingarten matrix}%
    associated with $\{e_i\}$. 
  \end{definition}
  \begin{remark}
    Under the assumptions of the definition, in particular, the assumption that $\{e_i\}$ is orthonormal with respect to $\bge$, (\ref{seq:constraintsBA})
    is equivalent to 
    \begin{subequations}\label{seq:compverconstraints}
      \begin{align}
        -\rotr\bn^{2}+\frac{1}{2}(\rotr\bn)^{2}
        -\rotr K^{2}+(\rotr K)^{2} = & \bphi_{1}^{2}
                                       +2V(\bphi_{0}),\label{eq:nuvhc}\\
        K\bn-\bn K = & 0.\label{eq:macom}
      \end{align}
    \end{subequations}
    This statement follows from \cite[Lemma~19.13, p.~210]{RinCauchy}.      
\end{remark}
  Changing basis, if necessary, it can be ensured that $\bn$ is diagonal and has diagonal elements $\bn_i$ of exactly one of the forms listed in
  Table~\ref{table:bianchiA}. Moreover, this table yields a classification of the Lie algebras of $3$-dimensional unimodular Lie groups. We refer
  the interested reader to \cite[Lemma~19.8, p.~208]{RinCauchy} for a justification of these statements. For a discussion of corresponding simply
  connected Lie groups, see the proof of Lemma~\ref{lemma:realcond} below.  
\begin{table}
\caption{Bianchi class A.}\label{table:bianchiA}
\begin{center}
\begin{tabular}{@{}cccccc}
  Type & $\bn_{1}$ & $\bn_{2}$ & $\bn_{3}$ & \vline & $\tilde{G}$\\
\hline
I                   & 0 & 0 & 0 & \vline & $\rn{3}$\\
II                  & + & 0 & 0 & \vline & $\mathrm{Heis}$\\
V$\mathrm{I}_{0}$   & 0 & + & $-$ & \vline & $\mathrm{Sol}$\\
VI$\mathrm{I}_{0}$  & 0 & + & + & \vline & -\\
VIII                & $-$ & + & + & \vline & $\tilde{\mathrm{Sl}}(2,\rn{})$\\
IX                  & + & + & + & \vline & $\mathrm{SU}(2)$
\end{tabular}
\end{center}
\end{table}
\begin{definition}\label{def:Bianchi Type}
  Let $G$ be a $3$-dimensional unimodular Lie group. Then, using the notation introduced prior to the statement of the definition, $G$ is said to
  be of \textit{Bianchi type}
  \index{Bianchi type!Lie group}%
  \index{Lie group!Bianchi type}%
  \index{Bianchi class A!Type}%
  \index{Bianchi class A!Type!Lie group}%
  $\mfT\in \{\mrI,\mrII,\mrVIz,\mrVIIz,\allowbreak \mrVIII,\mrIX\}$ if there is a frame $\{e_i\}$ of $\mfg$ such that
  the diagonal elements of $\bn$ have the corresponding signs indicated in Table~\ref{table:bianchiA}. 
\end{definition}
Beyond the left invariance inherent in Definition~\ref{def:Bianchiid}, some initial data have additional symmetries. 
\begin{definition}\label{def:isotropic data}
  Let $V\in C^\infty(\ro)$ and $(G,\bge,\bk,\bphi_{0},\bphi_{1})$ be Bianchi class A initial data for the Einstein non-linear scalar field equations. If the
  Ricci tensor of $\bge$ is a multiple of $\bge$ and if $\bk$ is a multiple of $\bge$, then the initial data are said to be \textit{isotropic}.
  \index{Initial data!Regular!Isotropic}%
  \index{Initial data!Isotropic}%
  \index{Isotropic!Initial data!Regular}%
\end{definition}
\begin{remark}\label{remark:iso VIIz data are BI}
  Isotropic Bianchi type VII${}_0$ initial data are of Bianchi type I; see Lemma~\ref{lemma:LRS BVIIz is BI} below. For this reason, we do not consider
  such data in what follows. 
\end{remark}
Next, we introduce the notion of local rotational symmetry. 
\begin{definition}\label{def:LRS}
  Let $V\in C^\infty(\ro)$ and $(G,\bge,\bk,\bphi_{0},\bphi_{1})$ be Bianchi class A initial data for the Einstein non-linear scalar field equations. Then the
  data are said to be \textit{locally rotationally symmetric}
  \index{Initial data!Regular!Locally rotationally symmetric}%
  \index{Initial data!Locally rotationally symmetric}%
  \index{Locally rotationally symmetric!Initial data!Regular}%
  or \textit{LRS}
  \index{Initial data!Regular!LRS}%
  \index{Initial data!LRS}%
  \index{LRS!Initial data!Regular}%
  if they are not isotropic and if there is a family of Lie algebra isomorphisms
  $\Psi_{t}$, $t\in\rn{}$, and an orthonormal frame $\{e_{i}\}$ of $\mfg$ with respect to $\bge$ such that
  \begin{equation}\label{eq:Psit definition}
    \Psi_{t}e_{1}=e_{1},\ \ \
    \Psi_{t}e_{2}=\cos(t)e_{2}+\sin(t)e_{3},\ \ \
    \Psi_{t}e_{3}=-\sin(t)e_{2}+\cos(t)e_{3},
  \end{equation}
  and such that for all $X,Y\in\mfg$ and $t\in\ro$, the relation $\bk(\Psi_{t}X,\Psi_{t}Y)=\bk(X,Y)$ holds.  
\end{definition}
\begin{remark}\label{remark:LRSisometries}
  In the simply connected setting, the $\Psi_t$ arise from a family of Lie group isomorphisms, preserving the initial data; cf.
  \cite[Theorem~20.19, p.~530]{Lee}. 
\end{remark}
\begin{remark}\label{remark:LRS VIIz data are BI}
  LRS Bianchi type VII${}_0$ initial data are of Bianchi type I; see Lemma~\ref{lemma:LRS BVIIz is BI} below. For this reason, we do not consider
  such data in what follows. 
\end{remark}
Local rotational symmetry appears in Bianchi types I, II, VII${}_0$, VIII and IX. However, there is a related symmetry in the case of Bianchi type
VI${}_0$.
\begin{definition}\label{def:permutation symmetry}
  Let $V\in C^\infty(\ro)$ and $(G,\bge,\bk,\bphi_{0},\bphi_{1})$ be Bianchi type VI${}_0$ initial data for the Einstein non-linear scalar field equations.
  Then the data are said to be \textit{permutation symmetric}
  \index{Initial data!Regular!Permutation symmetric}%
  \index{Initial data!Permutation symmetric}%
  \index{Permutation symmetric!Initial data!Regular}%
  if there is a Lie algebra isomorphism $\Psi$ and an orthonormal frame $\{e_{i}\}$ of $\mfg$ with respect to
  $\bge$ such that
  \begin{equation}\label{eq:Psi definition}
    \Psi e_{1}=e_{1},\ \ \
    \Psi e_{2}=e_3,\ \ \
    \Psi e_{3}=e_2
  \end{equation}
  and such that for all $X,Y\in\mfg$, the relation $\bk(\Psi X,\Psi Y)=\bk(X,Y)$ holds.  
\end{definition}
It is convenient to introduce notation for the different types of initial data
\begin{definition}\label{def:BV}
  Let $V\in C^{\infty}(\ro)$. Then $\mB[V]$
  \index{$\a$Aa@Notation!Sets of regular initial data!mBV@$\mB[V]$}%
  denotes the set of all Bianchi class A initial data for the Einstein non-linear scalar field equations with
  potential $V$ which are  non-trivial and are neither of isotropic nor of LRS Bianchi type VII${}_0$. The corresponding
  subset of simply connected initial data is denoted ${}^{\rosc}\mB[V]$.
  \index{$\a$Aa@Notation!Sets of regular initial data!scmBV@${}^{\rosc}\mB[V]$}%
  If $\mfI\in \mB[V]$ is such that the Lie group of $\mfI$ is of Bianchi
  type $\mfT$, then the initial data $\mfI$ are said to be of \textit{Bianchi type} $\mfT$.
  \index{Bianchi type!Initial data}%
  \index{Initial data!Bianchi type}%
  The set of $\mfI\in \mB[V]$ of Bianchi type $\mfT$ is denoted
  $\mB_{\mfT}[V]$.
  \index{$\a$Aa@Notation!Sets of regular initial data!mBTV@$\mB_\mfT[V]$}%
  The set ${}^{\rosc}\mB_{\mfT}[V]$
  \index{$\a$Aa@Notation!Sets of regular initial data!scmBTV@${}^{\rosc}\mB_\mfT[V]$}%
  is defined similarly. In all of these definitions, the argument $[V]$ is omitted in case the potential $V$
  is clear from the context.
\end{definition}
\begin{remark}\label{remark:dev of trivial id}
  Consider trivial initial data $\mfI$; see Definition~\ref{def:Bianchiid}. Then (\ref{eq:ham con original}) implies that
  $V(\bphi_0)=(\tr_{\bge}\bk)^2/3\geq 0$. If $\tr_{\bge}\bk=0$, then $\mfI$ are initial data for a quotient of Minkowski space.
  If $\pm\tr_{\bge}\bk>0$, then $\mfI$ give rise to a development $(M_{\rodS},g_{\rodS})$, where $M_{\rodS}=G_0\times\ro$, $G_0$ is
  a Bianchi type I Lie group, and
  \index{Trivial initial data!Developments}%
  \begin{equation}\label{eq:g dS}
    g_{\rodS}=-dt\otimes dt+e^{\pm 2Ht}\bge_0.
  \end{equation}
  Here $\bge_0$ is a flat metric on $G_0$ and $H=[V(\bphi_0)/3]^{1/2}$. In particular, $(M_{\rodS},g_{\rodS})$ is a quotient of a part of de Sitter space.
  Moreover, it is not a particularly natural model for the asymptotics in the contracting direction (it is, e.g., unstable in the contracting
  direction). To summarise: Minkowski space is not interesting in a cosmological setting, and $(M_{\rodS},g_{\rodS})$ is an unnatural model when analysing
  the asymptotics in the direction of the singularity. Moreover, due to their additional symmetry, the corresponding initial data sets cause problems
  in the regularity analysis of the set of isometry classes of initial data. For these reasons, we exclude trivial initial data. On the other, it should
  be mentioned that $(M_{\rodS},g_{\rodS})$ is important in the study of inflation, which takes place in intermediate regimes distinct from the singularity;
  cf., e.g., \cite{aauo}.
\end{remark}
In order for the sets of isometry classes of initial data to be smooth manifolds, it is necessary to focus on one symmetry class at a time. It is
therefore important to introduce the following terminology.
\begin{definition}
  Let $V\in C^{\infty}(\ro)$ and fix a Bianchi class A type $\mfT$. The sets of isotropic, permutation symmetric and LRS elements of $\mB_{\mfT}[V]$
  are denoted $\mB^{\iso}_{\mfT}[V]$,
  \index{$\a$Aa@Notation!Sets of regular initial data!mBisoTV@$\mB_\mfT^{\iso}[V]$}%
  $\mB^{\roper}_{\mfT}[V]$
  \index{$\a$Aa@Notation!Sets of regular initial data!mBperTV@$\mB_\mfT^{\roper}[V]$}%
  and $\mB^{\roLRS}_{\mfT}[V]$
  \index{$\a$Aa@Notation!Sets of regular initial data!mBLRSTV@$\mB_\mfT^{\roLRS}[V]$}%
  respectively. The set of elements of $\mB_{\mfT}[V]$ which are
  neither isotropic, permutation symmetric nor LRS
  is denoted $\mB^{\rogen}_{\mfT}[V]$.
  \index{$\a$Aa@Notation!Sets of regular initial data!mBgenTV@$\mB_\mfT^{\rogen}[V]$}%
  For $\mfs\in \{\iso,\roLRS,\roper,\rogen\}$, ${}^{\rosc}\mB_{\mfT}^{\mfs}[V]$
  \index{$\a$Aa@Notation!Sets of regular initial data!scmBsTV@${}^{\rosc}\mB_\mfT^{\mfs}[V]$}%
  denotes the set of simply connected elements
  of $\mB_{\mfT}^{\mfs}[V]$. Moreover, ${}^{\rosc}\mfB_{\mfT}^{\mfs}[V]$
  \index{$\a$Aa@Notation!Sets of isometry classes of regular initial data!scmfBsTV@${}^{\rosc}\mfB_{\mfT}^{\mfs}[V]$}%
  denotes the set of isometry classes of elements in ${}^{\rosc}\mB_{\mfT}^{\mfs}[V]$. 
\end{definition}
\begin{remark}\label{remark:type and sym class}
  When we write $\mB_{\mfT}^{\mfs}[V]$ etc., we take it for granted that the Bianchi type $\mfT$ is such that the symmetry
  type $\mfs$ makes sense; e.g., if $\mfs=\roper$, then $\mfT=\mrVIz$. 
\end{remark}
\begin{remark}
  If $V\in C^\infty(\ro)$, the sets ${}^{\rosc}\mfB_{\mfT}^{\mfs}[V]$ can be parametrised by smooth manifolds; see
  Lemma~\ref{lemma:sc mfB mfT mfs param} and Remark~\ref{remark:sc mfB mfT mfs param}. In what follows we therefore assume these sets
  to have been endowed with the corresponding smooth structures.
\end{remark}

\section{Data on the singularity}\label{section:data on the singularity}
Next, we introduce a geometric notion of initial data on a singularity analogous to but more general than that given in
\cite[Definition~10, p.~1510]{RinQCSymm}. 
\begin{definition}\label{def:ndvacidonbbssh}
  Let $G$ be a $3$-dimensional unimodular Lie group, $\msH$ be a left invariant Riemannian metric on $G$, $\msK$ be a left invariant
  $(1,1)$-tensor field on $G$ and $(\Phi_{0},\Phi_{1})\in\rn{2}$. Then $(G,\msH,\msK,\Phi_{0},\Phi_{1})$ are
  \textit{quiescent Bianchi class A initial data on the singularity for the Einstein non-linear scalar field equations}
  \index{Bianchi class A!Initial data}%
  \index{Bianchi class A!Initial data on the singularity}%
  \index{Initial data!On singularity}%
  if
  \begin{enumerate}
  \item $\tr\msK=1$ and $\msK$ is symmetric with respect to $\msH$.
  \item $\tr\msK^{2}+\Phi_{1}^{2}=1$ and $\mathrm{div}_{\msH}\msK=0$.
  \item In case all the eigenvalues of $\msK$ are $<1$ and there is one eigenvalue, say $p_A$, satisfying $p_A\leq 0$, then the vector subspace of
    $\mfg$, say $\mfh$, perpendicular to the eigenspace of $p_A$ is a subalgebra of $\mfg$.
  \item If $1$ is an eigenvalue of $\msK$, there is an orthonormal basis $\{e_{i}\}$ of $\mfg$ with respect to $\msH$ such that
    $\msK e_{1}=e_{1}$ and such that if $\Psi_t$ is defined by (\ref{eq:Psit definition}), then $\Psi_t$ is a Lie algebra isomorphism
    for all $t$.
  \end{enumerate}  
\end{definition}
\begin{remark}
  The eigenvalues of $\msK$ are constants due to the left invariance of $\msK$.
\end{remark}
\begin{remark}
  If all the eigenvalues, say $p_A$, are $<1$ and there is one eigenvalue, say $p_1$, satisfying $p_1\leq 0$, then, due to the requirement that $p_A<1$
  and the fact that the $p_A$ sum up to $1$, it follows that $p_2,p_3>0$, so that the eigenspace corresponding to $p_1$ is $1$-dimensional and $\mfh$ is
  $2$-dimensional. 
\end{remark}
\begin{definition}\label{def:iso of data on sing}
  Let $\mfI_\a=(G_\a,\msH_\a,\msK_\a,\Phi_{\a,1},\Phi_{\a,0})$, $\a\in\{0,1\}$, be quiescent Bianchi class A initial data on the singularity for the
  Einstein non-linear scalar field equations. Then $\chi:G_0\rightarrow G_1$ is said to be an \textit{isometry}
  \index{Isometry!Initial data on singularity}%
  \index{Initial data!On singularity!Isometry of}%
  from $\mfI_0$ to $\mfI_1$ if $\chi$
  is a diffeomorphism such that
  \[
  \chi^*\msH_1=\msH_0,\ \ \
  \chi^*\msK_1=\msK_0,\ \ \
  \chi^*\Phi_{1,0}=\Phi_{0,0},\ \ \
  \chi^*\Phi_{1,1}=\Phi_{0,1}.
  \]
  Here $\chi^*\msK_1:=(\chi^{-1})_*\circ\msK_1\circ\chi_*$.
\end{definition}
It is of interest to introduce terminology for symmetry types of data on the singularity.
\begin{definition}\label{def:isotropic}
  Let $(G,\msH,\msK,\Phi_{0},\Phi_{1})$ be quiescent Bianchi class A initial data on the singularity for the Einstein non-linear scalar field equations.
  If $\msK=\mathrm{Id}/3$ and $\msH$ has constant curvature, then the initial data on the singularity are said to be \textit{isotropic}.
  \index{Initial data!On singularity!Isotropic}%
  \index{Initial data!Isotropic}%
  \index{Isotropic!Initial data!On singularity}%
  If the initial
  data on the singularity are not isotropic and if there is a family of Lie algebra isomorphisms $\Psi_{t}$,
  $t\in\rn{}$, and an orthonormal frame $\{e_{i}\}$ of $\mfg$ with respect to $\msH$ such that (\ref{eq:Psit definition}) and $\msK\Psi_t=\Psi_t\msK$
  both hold for all $t$, then the initial data on the singularity are said to be \textit{locally rotationally symmetric}
  \index{Initial data!On singularity!Locally rotationally symmetric}%
  \index{Initial data!Locally rotationally symmetric}%
  \index{Locally rotationally symmetric!Initial data!On singularity}%
  or \textit{LRS}.
  \index{Initial data!On singularity!LRS}%
  \index{Initial data!LRS}%
  \index{LRS!Initial data!On singularity}%
  If $G$ is of Bianchi type VI${}_0$ and if there is a Lie algebra isomorphism $\Psi$ and an orthonormal frame $\{e_{i}\}$ of $\mfg$ with
  respect to $\msH$ such that (\ref{eq:Psi definition}) and $\msK\Psi=\Psi\msK$ hold, then the initial data on the singularity are said to
  have a \textit{permutation symmetry}.
  \index{Initial data!On singularity!Permutation symmetric}%
  \index{Initial data!Permutation symmetric}%
  \index{Permutation symmetric!Initial data!On singularity}%
\end{definition}
\begin{remark}
  That all the eigenvalues of $\msK$ coincide is equivalent to $\msK=\mathrm{Id}/3$, since $\rotr\msK=1$.
\end{remark}
\begin{remark}
  If $\msK e_1=e_1$ (where $e_1$ is a unit vector), then the eigenvalues of $\msK$ are $1$, $0$ and $0$. This means that $\msK v=0$ for
  $v\perp e_1$. If $e_A$, $A=2,3$, are orthonormal vectors perpendicular to $e_1$, this means that $\Psi_t$ defined by (\ref{eq:Psit definition}) is
  a family of isometries such that $\Psi_t\msK=\msK\Psi_t$ for all $t$. Assuming, additionally, (as we do in case $1$ is an eigenvalue of $\msK$)
  that $\Psi_t$ is a Lie algebra isomorphism for all $t$, this means that the data on the singularity are LRS. By Theorem~\ref{thm:dataonsingtosolution}
  below, it follows that the development is LRS. 
\end{remark}
\begin{remark}\label{remark:LRS VIIz sing BI sing}
  Isotropic and LRS Bianchi type VII${}_0$ initial data on the singularity are of Bianchi type I; see Lemma~\ref{lemma:LRS VIIz is I sing} below. For
  this reason, we do not consider such data in what follows. 
\end{remark}
\begin{definition}\label{def:sets of id on singularity}
  The set of quiescent Bianchi class A initial data on the singularity for the Einstein non-linear scalar field equations which are neither of
  isotropic nor of LRS Bianchi type VII${}_0$ is denoted $\mS$.
  \index{$\a$Aa@Notation!Sets of initial data on the singularity!mS@$\mS$}%
  The corresponding set of simply connected initial data on the singularity is denoted ${}^{\rosc}\mS$.
  \index{$\a$Aa@Notation!Sets of initial data on the singularity!scmS@${}^{\rosc}\mS$}%
  Given a Bianchi class A type $\mfT$, the elements of $\mS$ and ${}^{\rosc}\mS$ of Bianchi type $\mfT$ are denoted $\mS_{\mfT}$
  \index{$\a$Aa@Notation!Sets of initial data on the singularity!mST@$\mS_\mfT$}%
  and ${}^{\rosc}\mS_{\mfT}$
  \index{$\a$Aa@Notation!Sets of initial data on the singularity!scmST@${}^{\rosc}\mS_\mfT$}%
  respectively.
  The sets of isotropic, permutation symmetric and LRS elements of $\mS_{\mfT}$ are denoted $\mS^{\iso}_{\mfT}$,
  \index{$\a$Aa@Notation!Sets of initial data on the singularity!mSisoT@$\mS_\mfT^{\iso}$}%
  $\mS^{\roper}_{\mfT}$
  \index{$\a$Aa@Notation!Sets of initial data on the singularity!mSperT@$\mS_\mfT^{\roper}$}%
  and $\mS^{\roLRS}_{\mfT}$
  \index{$\a$Aa@Notation!Sets of initial data on the singularity!mSLRST@$\mS_\mfT^{\roLRS}$}%
  respectively. Finally, the set of elements of $\mS_{\mfT}$ which are neither isotropic, permutation symmetric nor LRS
  is denoted $\mS^{\rogen}_{\mfT}$.
  \index{$\a$Aa@Notation!Sets of initial data on the singularity!mSgenT@$\mS_\mfT^{\rogen}$}%
  For $\mfs\in \{\iso,\roLRS,\roper,\rogen\}$, ${}^{\rosc}\mS_{\mfT}^{\mfs}$
  \index{$\a$Aa@Notation!Sets of initial data on the singularity!scmSsT@${}^{\rosc}\mS_\mfT^{\mfs}$}%
  denotes the set of simply connected elements of $\mS_{\mfT}^{\mfs}$. Finally, ${}^{\rosc}\mfS_{\mfT}^{\mfs}$
  \index{$\a$Aa@Notation!Sets of isometry classes of initial data on the singularity!scmfSsT@${}^{\rosc}\mfS_\mfT^{\mfs}$}%
  denotes the set of isometry classes of elements in ${}^{\rosc}\mS_{\mfT}^{\mfs}$.
\end{definition}
\begin{remark}
  Remark~\ref{remark:type and sym class} is equally relevant here. 
\end{remark}
\begin{remark}
  The sets ${}^{\rosc}\mfS_{\mfT}^{\mfs}$ can be parametrised by smooth manifolds; see Remark~\ref{remark:sc mfB mfT mfs param sing}.
  In what follows we therefore assume these sets to have been endowed with the corresponding smooth structures.
\end{remark}

\section{Developments}\label{ssection:developments}
Due to general results, we know that there is a unique maximal globally hyperbolic development associated with data as
in Definition~\ref{def:Bianchiid}; see, e.g., \cite[Lemma~23.2, p.~398]{stab} and \cite[Corollary~23.44, p.~418]{stab}. However, for the
purposes of the present discussion, it is convenient to work with a specific representation.

\begin{definition}\label{def:BianchiAdevelopment}
  Let $V\in C^{\infty}(\ro)$ and $\mfI=(G,\bar{h},\bk,\bphi_0,\bphi_1)\in\mB[V]$. A \textit{Bianchi class A non-linear scalar field development of}
  \index{Bianchi class A!Development}%
  \index{Development!Bianchi class A}%
  $\mfI$ is a triple $(M,g,\phi)$, where $M:=G\times J$, $J$ is an open interval, $\phi\in C^{\infty}(J)$,
  \begin{equation}\label{eq:gSH}
    g=-dt\otimes dt+\bge,
  \end{equation}
  \index{$\a$Aa@Notation!Time coordinates!t@$t$}%
  and $\bge$ denotes a family of left invariant Riemannian metrics on $G$. Finally, $(M,g,\phi)$ is required to satisfy the Einstein non-linear
  scalar field equations, and for some $t_{0}\in J$, the initial data induced on $G_{t_0}:=G\times \{t_{0}\}$ by $(M,g,\phi)$ and pulled back by
  $\iota_0$ are required to equal $\mfI$, where $\iota_0(p):=(p,t_0)$.

  In case the left endpoint of $J$ is $t_-$ ($t_-$ will always be assumed to equal $0$ if $t_->-\infty$) and $\theta(t)\rightarrow \infty$ as
  $t\downarrow t_-$ (where $\theta(t)$ denotes the mean curvature of $G_t$ with respect to $g$), then the development is said to have a
  \textit{crushing singularity}.
  \index{Crushing singularity}%

  If $(M,g,\phi)$ is a Bianchi class A non-linear scalar field development of isotropic/LRS/permutation symmetric initial data, then $(M,g,\phi)$ is
  said to be an \textit{isotropic/LRS/permutation symmetric development}.
  \index{Development!Isotropic}%
  \index{Isotropic!Development}%
  \index{Development!LRS}%
  \index{LRS!Development}%
  \index{Development!Permutation symmetric}%
  \index{Permutation symmetric!Development}%
\end{definition}
\begin{remark}
  In (\ref{eq:gSH}) $t$, by standard abuse of notation, denotes the time coordinate (mapping $(\bx,s)\in M$ to $s$) as well as the corresponding
  element of $J$.
\end{remark}
\begin{remark}
  We assign a Bianchi type to a Bianchi class A non-linear scalar field development according to the Bianchi type of $G$; see
  Definition~\ref{def:Bianchi Type}.
\end{remark}
Given regular initial data, there is, up to time translation, a unique corresponding maximal development.
\begin{prop}\label{prop:unique max dev}
  Let $V\in C^{\infty}(\ro)$ and $\mfI\in\mB[V]$. Then there is a unique (up to time translation) Bianchi class
  \index{Development!Existence}%
  \index{Development!Uniqueness}%
  A non-linear scalar field development of $\mfI$, say $(M,g,\phi)$ with $M=G\times J$ and $J=(t_-,t_+)$, which is maximal in the sense that
  either $t_\pm=\pm\infty$ or $|\theta|+|\phi|$ is unbounded as $t\rightarrow t_\pm$, where $\theta$ denotes the mean curvature of the constant-$t$
  hypersurfaces. Moreover, this development is globally hyperbolic. Let $\mfs\in \{\iso,\roLRS,\roper,\rogen\}$ and $\mfT$ be a Bianchi type
  consistent with $\mfs$. If $\mfI\in \mB_{\mfT}^{\mfs}[V]$ and $t\in J$, then the initial data induced on $G_t$ by $(M,g,\phi)$ also belong to
  $\mB_{\mfT}^{\mfs}[V]$. 
\end{prop}
\begin{remark}\label{remark:improved dichotomy}
  In case $V\geq 0$, the development is maximal in the sense that either $t_\pm=\pm\infty$ or $|\theta|$ is unbounded as $t\rightarrow t_\pm$.
\end{remark}
\begin{remark}
  In most of this article, we are interested in developments with $t_->-\infty$. In this case, we fix the time translation ambiguity by
  demanding that $t_-=0$. 
\end{remark}
\begin{proof}
  See Subsection~\ref{ssection:ex of dev step one} below.
\end{proof}
\begin{definition}
  Let $V\in C^{\infty}(\ro)$ and $\mfI\in\mB[V]$. Then the unique development obtained in Proposition~\ref{prop:unique max dev} is
  denoted $\md[V](\mfI)$
  \index{$\a$Aa@Notation!Sets of developments!mDVmfI@$\mD[V](\mfI)$}%
\end{definition}
The following observation will be of interest.
\begin{prop}\label{prop:iso id to iso sol}
  Let $\mfI_\a\in\mB[V]$, $\a=0,1$. Assume $\chi:G_0\rightarrow G_1$ to be an isometry from $\mfI_0$ to $\mfI_1$. Fix the time translation
  ambiguity in $\md[V](\mfI_\a)$ by demanding that the intial data $\mfI_\a$ be induced at the $t_0$ hypersurface, for some fixed $t_0\in\ro$. Then
  $\chi\times\roId$ is an isometry from $\md[V](\mfI_0)$ to $\md[V](\mfI_1)$ (in particular, the existence interval is the same for the two
  developments).  
\end{prop}
\begin{remark}
  Under the assumptions of the proposition, $\chi$ gives rise to an isometry between the initial data sets induced on $\{t\}\times G_0$ and
  $\{t\}\times G_1$ for any $t$ in the existence interval. 
\end{remark}
\begin{remark}
  If $\Gamma$ is a free and properly discontinuous subgroup of the isometry group of an initial data set, then the proposition implies that
  $\Gamma\times\roId$ is a free and properly discontinuous subgroup of the isometry group of the corresponding development. Taking the
  quotient yields a spatially locally homogeneous development of the initial data induced on the quotient of the initial manifold with $\Gamma$. 
\end{remark}
\begin{proof}
  The proof can be found just above the statement of Lemma~\ref{lemma:isometryinducesequal o P}.
\end{proof}

\subsection{Developments inducing data on the singularity}\label{ssection:dev ind data on sing}

In order to relate developments to data on the singularity, we first define what it means for a development to induce data on the singularity.
\begin{definition}\label{def:ind data on sing}
  Let $V\in C^{\infty}(\ro)$, $\mfI\in\mB[V]$ and $(M,g,\phi)=\mD[V](\mfI)$. Assume $(M,g,\phi)$ to have a crushing singularity. Using the terminology
  of Definition~\ref{def:BianchiAdevelopment}, let $\bge$ and $K$ be the metrics and Weingarten maps induced on $G_t$, $t\in J$, and define the
  \textit{expansion normalised Weingarten map}
  \index{Weingarten map!Expansion normalised}%
  \index{Expansion normalised!Weingarten map}%
  by $\mK:=K/\theta$. If there are $\mfI_\infty=(G,\msH,\msK,\Phi_{0},\Phi_{1})\in \mS$ such that $\mK\rightarrow\msK$,
  $\theta^{-1}\phi_t\rightarrow\Phi_{1}$ and $\phi+\theta^{-1}\phi_t\ln\theta\rightarrow\Phi_{0}$ as $t\downarrow t_-$; and such that the following
  limit holds for all $v,w\in\mfg$:
  \begin{equation}\label{eq:bAmetriclimit}
    \lim_{t\downarrow t_-}\bge(\theta^{\mK}v,\theta^{\mK}w)=\msH(v,w),
  \end{equation}
  then the development $(M,g,\phi)$ is said to \textit{induce $\mfI_\infty$ on the singularity}.
  \index{Development!Inducing data on singularity}%
  \index{Inducing data on singularity!Development}%
  \index{Initial data!On singularity!From development}%
\end{definition}
\begin{remark}
  One way to define the Weingarten map is as the second fundamental form with one index raised by the induced metric.
\end{remark}
\begin{remark}
  When we say that $\mK\rightarrow\msK$, we consider $\mK$ to be a family of tensor fields, indexed by $t$, on the fixed manifold $G$, and
  we mean that this family converges pointwise with respect to a fixed reference metric on $G$. 
\end{remark}
\begin{remark}
  Under the assumptions of the definition, $t_->-\infty$, so that we can assume $t_-=0$; see the proof of Theorem~\ref{thm:dataonsingtosolution}
  in Section~\ref{ssection:dataonsingularity}. 
\end{remark}
\begin{remark}
  In (\ref{eq:bAmetriclimit}), it is understood that
  \[
    \theta^{\mK}:=\textstyle{\sum}_{l=0}^{\infty}\tfrac{1}{l!}(\ln \theta)^l\mK^l.
  \]  
\end{remark}
\begin{lemma}\label{lemma:iso reg data iso dos}
  Under the assumptions of Proposition~\ref{prop:iso id to iso sol}, if $\md[V](\mfI_1)$ induces $\mfI_{1,\infty}$ on the singularity,
  then $\md[V](\mfI_0)$ also induces initial data, say $\mfI_{0,\infty}$, on the singularity. Moreover, $\chi$ is an isometry from
  $\mfI_{0,\infty}$ to $\mfI_{1,\infty}$; cf. Definition~\ref{def:iso of data on sing}. 
\end{lemma}
\begin{proof}
  Due to the conclusions of Proposition~\ref{prop:iso id to iso sol}, $\chi\times\roId$ is an isometry from $\md[V](\mfI_0)$ to $\md[V](\mfI_1)$.
  Denote the metric and scalar field of $\md[V](\mfI_\a)$ by $g_\a$ and $\phi_\a$ respectively. Denote the mean curvature of $G_\a\times \{t\}$ by
  $\theta_\a(t)$. Since $\theta_\a$ is independent of the spatial coordinate, and since $\chi\times\roId$ is an isometry, we have to have
  $\theta_1=\theta_0$. Next, since the scalar fields are independent of the spatial coordinate, the form of the isometry implies that
  $\phi_0=\phi_1$. Since $\theta_1=\theta_0$ and $\phi_1=\phi_0$, it follows that the scalar fields induce the same data on the singularity.
  It remains to consider the limits for the geometry. Let $\mK_\a$ denote the expansion normalised Weingarten map of the constant-$t$ hypersurfaces
  in $\md[V](\mfI_\a)$. Then
  \[
  \msK_0=\lim_{t\downarrow t_-}\mK_0=\lim_{t\downarrow t_-}(\chi^{-1})_*\circ\mK_1\circ\chi_*=(\chi^{-1})_*\circ\msK_1\circ\chi_*.
  \]
  This yields the desired statement for the limit of the expansion normalised Weingarten map. Finally
  \begin{equation*}
    \begin{split}
      \lim_{t\downarrow t_-}\bge_0(\theta_0^{\mK_0}v,\theta_0^{\mK_0}w) = & \lim_{t\downarrow t_-}\bge_0\big((\chi^{-1})_*\theta_1^{\mK_1}\chi_*v,
      (\chi^{-1})_*\theta_1^{\mK_1}\chi_* w\big)\\
      = & \lim_{t\downarrow t_-}\bge_1\big(\theta_1^{\mK_1}\chi_*v,\theta_1^{\mK_1}\chi_* w\big)=\msH_1(\chi_*v,\chi_*w)=\chi^{*}\msH_1(v,w).
    \end{split}
  \end{equation*}
  Thus $\msH_0=\chi^{*}\msH_1$. The lemma follows.
\end{proof}
In order to obtain results, we typically need to impose conditions on the potential. They can often be phrased in terms of the following terminology.
\begin{definition}\label{def:mfP a inf}
  Let $\a_V\in [0,\infty)$ and $k\in\nn{}$. Then the set of $V\in C^{\infty}(\ro)$ such that there is a constant $c_k<\infty$ with the property that
  \begin{equation}\label{eq:V k-derivatives exp estimate}
    \textstyle{\sum}_{l=0}^k|V^{(l)}(s)|\leq c_ke^{\sqrt{6}\a_V|s|}
  \end{equation}
  for all $s\in\ro$ is denoted $\mfP_{\a_V}^k$.
  \index{$\a$Aa@Notation!Spaces of potentials!mfPka@$\mfP_{\a_V}^k$}%
  Moreover, $\mfP_{\a_V}^{\infty}:=\cap_{l=0}^{\infty}\mfP_{\a_V}^{l}$.
  \index{$\a$Aa@Notation!Spaces of potentials!mfPinfa@$\mfP_{\a_V}^\infty$}%
\end{definition}
The next result states that there is a unique development corresponding to data on the singularity. An analogous result in the Einstein scalar
field setting can be found in \cite[Theorem~16, p.~1513]{RinQCSymm}. A much more general result, in the \textit{absence of symmetries}, is
contained in \cite[Theorem~1.9]{andres} and \cite[Theorem~1.19]{andres}; see also \cite{aarendall,klinger,fal}. However, the assumptions in
\cite{aarendall,klinger,fal,andres} are such that the following theorem is not a special case of the results of \cite{aarendall,klinger,fal,andres}.
\begin{thm}\label{thm:dataonsingtosolution}
  Let $\a_{V}\in (0,1)$ and $0\leq V\in \mfP_{\a_V}^{2}$; see Definition~\ref{def:mfP a inf}. Fix a Bianchi class A type $\mfT$ and
  $\mfs\in \{\iso,\roLRS,\roper,\rogen\}$. Assume that $\mfI_\infty:=(G,\msH,\msK,\Phi_{0},\Phi_{1})\in\mS_{\mfT}^\mfs$. Then there is a unique $\md[V](\mfI)$,
  arising from $\mfI=(G,\bge,\bk,\bphi_0,\bphi_1)\in\mB_{\mfT}[V]$, inducing $\mfI_\infty$ on the singularity.
  \index{Development!Initial data on singularity!Existence}%
  \index{Development!Initial data on singularity!Uniqueness}%
  Moreover, $\mfI\in\mB_{\mfT}^{\mfs}[V]$ and if $J$ is the existence interval of $\md[V](\mfI)$, the left endpoint of $J$ can be assumed to equal $0$.
  Finally, if $\mfI_\infty$ satisfy all the conditions of Definition~\ref{def:ndvacidonbbssh} but the last two; $1$ is an eigenvalue of $\msK$;
  and the final condition of Definition~\ref{def:ndvacidonbbssh} is violated, then there is no Bianchi class A non-linear scalar field development 
  inducing $\mfI_\infty$ on the singularity. 
\end{thm}
\begin{remark}\label{remark:V geq zero cond}
  It is only necessary to impose the condition $V\geq 0$ in case $1$ is an eigenvalue of $\msK$; i.e., if the last condition of
  Definition~\ref{def:ndvacidonbbssh} is satisfied. 
\end{remark}
\begin{remark}\label{remark:form of dev of dos}
  If $\mfI_\infty$ are isotropic, the metric of the development takes the form
  \begin{equation}\label{eq:iso form of development}
    g=-dt\otimes dt+a^2(t)\msH,
  \end{equation}
  where $a$ is a smooth and strictly positive function on the existence interval. 
  If $\mfI_\infty$ are LRS, there is an orthonormal frame $\{e_i\}$ of the Lie algebra $\mfg$ with respect to $\msH$, consisting of eigenvectors of
  $\msK$ (with eigenvalue $p_i$ corresponding to $e_i$) such that the associated commutator matrix is diagonal with diagonal components $o_i$; the
  $o_i$ have signs as in Table~\ref{table:bianchiA}; $p_2=p_3$; $o_2=o_3$; and either $o_1\neq o_2$ or $p_1\neq p_2$. Moreover, if $\{\xi^i\}$ is
  the frame dual to $\{e_i\}$, then the metric takes the form
  \begin{equation}\label{eq:LRS or PS form of development}
    g=-dt\otimes dt+a_1^2(t)\xi^1\otimes \xi^1+a_2^2(t)[\xi^2\otimes \xi^2+\xi^3\otimes \xi^3],
  \end{equation}
  where the $a_i$ are smooth and strictly positive functions on the existence interval.
  Moreover, this statement holds irrespective of the choice of frame $\{e_i\}$ with the above properties. In case $\mfI_\infty$ are permutation
  symmetric, a similar statement holds. The only difference concerning the frame is that $o_2=-o_3$ and that $o_1\neq o_2$ is automatically satisfied.
  The conclusion concerning the form of the metric is the same. 
\end{remark}
\begin{proof}
  See Section~\ref{ssection:dataonsingularity} for a proof of both the theorem and of Remark~\ref{remark:form of dev of dos}.
\end{proof}
Combining this result with Lemma~\ref{lemma:iso reg data iso dos} and Remark~\ref{remark:isometric data are of same type} yields the following
observation. 
\begin{lemma}\label{lemma:iso ri iso ids incl type}
  Let $\a_{V}\in (0,1)$ and $0\leq V\in \mfP_{\a_V}^{2}$. Fix a Bianchi class A type $\mfT$ and $\mfs\in \{\iso,\roLRS,\roper,\rogen\}$.
  Assume that $\mfI\in \mB_\mfT^\mfs[V]$ and that $\md[V](\mfI)$ induces data $\mfI_\infty$ on the singularity. Then $\mfI_\infty\in\mS_\mfT^\mfs$.
  Moreover, if $\mfI_1$ is isometric with $\mfI$, then $\mfI_1\in \mB_\mfT^\mfs[V]$; $\md[V](\mfI_1)$ induces data $\mfI_{\infty,1}$ on the
  singularity; $\mfI_{\infty,1}\in\mS_\mfT^\mfs$; and $\mfI_\infty$ and $\mfI_{\infty,1}$ are isometric. In particular, there is a well defined map
  from elements of $\mfB_\mfT^\mfs[V]$ whose developments induce data on the singularity to $\mfS_\mfT^\mfs$. 
\end{lemma}
\begin{remark}
  If one is prepared to exclude initial data on the singularity satisfying the last condition of Definition~\ref{def:ndvacidonbbssh}, then
  the requirement $V\geq 0$ can be omitted. 
\end{remark}
\begin{proof}
  The assumption that $\md[V](\mfI)$ induces data $\mfI_\infty$ means that $\mfI_\infty\in\mS_{\mfT}^{\mfs_1}$ for some symmetry type $\mfs_1$; see
  Definition~\ref{def:ind data on sing}. Then, due to Theorem~\ref{thm:dataonsingtosolution}, there is a unique $\md[V](\mfI_1)$, arising from
  some $\mfI_1\in\mB_{\mfT}^{\mfs_1}[V]$, inducing $\mfI_\infty$ on the singularity. In particular, $\md[V](\mfI_1)=\md[V](\mfI)$. However, due to
  Proposition~\ref{prop:unique max dev}, the symmetry type is preserved by the evolution. This means that $\mfs_1=\mfs$. The first statement
  of the lemma follows. The second statement is an immediate consequence of the first statement, Lemma~\ref{lemma:iso reg data iso dos} and
  Remark~\ref{remark:isometric data are of same type}. 
\end{proof}

In analogy with Proposition~\ref{prop:iso id to iso sol}, the following statement holds. 
\begin{prop}\label{prop:iso idos to iso dev}
  Let $\a_{V}\in (0,1)$ and $0\leq V\in \mfP_{\a_V}^{2}$; see Definition~\ref{def:mfP a inf}. Fix a Bianchi class A type $\mfT$ and
  $\mfs\in \{\iso,\roLRS,\roper,\rogen\}$. Assume that $\mfI_{\infty,\a}\in\mS_{\mfT}^\mfs$, $\a\in\{0,1\}$. If
  $\chi$ is an isometry from $\mfI_{\infty,0}$ to $\mfI_{\infty,1}$, see Definition~\ref{def:iso of data on sing}, then $\chi\times\roId$ is an
  isometry of the corresponding developments obtained in Theorem~\ref{thm:dataonsingtosolution}. 
\end{prop}
\begin{remark}
  Remark~\ref{remark:V geq zero cond} is equally relevant here.
\end{remark}
\begin{remark}
  As a consequence of the statement, isometry classes of data on the singularity induce isometry classes of developments.
\end{remark}
\begin{proof}
  The proof is to be found at the end of Appendix~\ref{section:isometryclasses}. 
\end{proof}
The asymptotics in Theorem~\ref{thm:dataonsingtosolution} are the best one can hope for, in the sense that they uniquely determine
the corresponding development. They therefore set the standard when deriving asymptotics in what follows. 

\subsection{The Bianchi type IX setting}\label{ssection:Bianchi type IX setting}

The Bianchi type IX setting is complicated by the fact that the set of developments include de Sitter space, which expands both to the future and
to the past; non-vacuum solutions to the Einstein scalar field equations which start with a big bang and end with a big crunch, both of which are
quiescent; Bianchi type IX vacuum solutions with a big bang and a big crunch, both of which are oscillatory; solutions with a static geometry
(by balancing the matter content with a cosmological constant and the spatial scalar curvature; cf. the Einstein static universe);
solutions that have a quiescent big bang singularity but then expand in a de Sitter like fashion etc. Here, we are interested in solutions with a
crushing big bang singularity. This means that we have to exclude some sets of initial data. One way of doing so is to begin by noting that
(\ref{eq:ham con original}) can be written
\begin{equation}\label{eq:ham con ito X}
  \theta^2-\tfrac{3}{2}\bphi_1^2-3V(\bphi_0)=\tfrac{3}{2}|\sigma|_{\bge}^{2}-\tfrac{3}{2}\bS,
\end{equation}
where $\sigma$ is the trace free part of $\bk$. Excluding isotropic Bianchi type I as well as Bianchi type IX initial data, the right hand side
of (\ref{eq:ham con ito X}) is strictly positive. However, in the case of isotropic Bianchi type I, the right hand side is identically zero.
Moreover, in the case of Bianchi type IX, the right hand side can have any sign. However, it is worth noting that, e.g., in the case of de Sitter
space, the right hand side is strictly negative on all hypersurfaces of spatial homogeneity. Moreover, in the case of the Einstein static universe,
the right hand side is strictly negative and constant. One way to exclude these developments is thus to insist on initial data such that both the
mean curvature and the right hand side of (\ref{eq:ham con ito X}) are strictly positive at some point in time (if these conditions are satisfied at
one point in time, then, due to Lemma~\ref{lemma:BianchiAdevelopment}, they are satisfied for the entire past). While this is a useful criterion, it
does exclude some solutions one would like to include; e.g. isotropic solutions with a quiescent big bang singularity. A different way to proceed
is to demand that the initial data be such that the development has the following properties: $\theta>0$ in a neighbourhood of $t_-$ (this condition
excludes both de Sitter space and the Einstein static universe); and for every $\e>0$, there is a $T$ in the existence interval such that
$X(t)/\theta^2(t)>-\e$ and $|V\circ\phi(t)|/\theta^2(t)<\e$ for all $t\leq T$, where $X$ denotes the right hand side of (\ref{eq:ham con ito X}). Since
we are interested in crushing singularities,
the assumption that $\theta(t)>0$ in a neighbourhood of the singularity is quite reasonable. Even in the oscillatory Bianchi type IX vacuum setting, the
positive part of $\bS/\theta^2$ converges to zero. The lower bound on $X/\theta^2$ is therefore quite natural. Finally, we are interested in situations
such that the contribution of the potential is asymptotically negligible in the direction of the crushing singularity. For that reason, the assumption
that $V\circ\phi/\theta^2$ converges to zero is reasonable. In the end, assuming suitable conditions on the potential, the first condition (i.e., that
$X$ and $\theta$ be positive at one point in time) implies the second. However, proving that this is the case is non-trivial, and we therefore
introduce both conditions here.

Before stating the formal definitions, it is convenient to introduce the following notation, given $V\in C^{\infty}(\ro)$:
\begin{equation}\label{eq:mfX def}
  \mfX(\theta,\bphi_0,\bphi_1):=\theta^2-\tfrac{3}{2}\bphi_1^2-3V(\bphi_0).
\end{equation}
\index{$\a$Aa@Notation!Functions!mfX@$\mfX$}%
\begin{definition}\label{definition:Bap Bp Bpp}
  Let $0\leq V\in C^{\infty}(\ro)$ and $\mfI\in \mB_{\mrIX}[V]$. Let $J$ and $\theta$ be the existence interval and mean curvature,
  respectively, of the corresponding development, see Proposition~\ref{prop:unique max dev}. If there is a $t_0\in J$ such that $\theta(t)>0$
  for all $t\leq t_0$; and if for every $\e>0$, there is a $T\in J$ such that $X(t)/\theta^2(t)>-\e$ and $V\circ\phi(t)/\theta^2(t)<\e$ for all
  $t\leq T$, where $X(t):=\mfX[\theta(t),\phi(t),\phi_t(t)]$, see (\ref{eq:mfX def}), then $\mfI$ are said to be \textit{almost positive}
  \index{Almost positive!Initial data}%
  \index{Initial data!Almost positive}%
  \index{Initial data!Regular!Almost positive}%
  initial data, written $\mfI\in \mB_{\mrIX,\roap}[V]$.
  \index{$\a$Aa@Notation!Sets of regular initial data!mBIXapV@$\mB_{\mrIX,\roap}[V]$}%
  Next, $\mfI$ are said to be \textit{positive}
  \index{Positive!Initial data}%
  \index{Initial data!Positive}%
  \index{Initial data!Regular!Positive}%
  initial data if there is a $t_0\in J$ such that $X(t_0)>0$ and $\theta(t_0)>0$, written $\mfI\in \mB_{\mrIX,+}[V]$.
  \index{$\a$Aa@Notation!Sets of regular initial data!mBIX+V@$\mB_{\mrIX,+}[V]$}%
  Finally $\mfI$ are said to be \textit{pseudo positive}
  \index{Pseudo positive!Initial data}%
  \index{Initial data!Pseudo positive}%
  \index{Initial data!Regular!Pseudo positive}%
  if they are either positive or almost positive, written $\mfI\in \mB_{\mrIX,\ropp}[V]$.
  \index{$\a$Aa@Notation!Sets of regular initial data!mBIXppV@$\mB_{\mrIX,\ropp}[V]$}%
  For $\mft\in \{+,\roap,\ropp\}$ and $\mfs\in \{\iso,\roLRS,\rogen\}$,
  \[
  \mB_{\mrIX,\mft}^{\mfs}[V]:=\mB_{\mrIX,\mft}[V]\cap \mB_{\mrIX}^{\mfs}[V],\ \ \
  {}^{\rosc}\mB_{\mrIX,\mft}^{\mfs}[V]:=\mB_{\mrIX,\mft}[V]\cap {}^{\rosc}\mB_{\mrIX}^{\mfs}[V].
  \]
  \index{$\a$Aa@Notation!Sets of regular initial data!mBIXtsV@$\mB_{\mrIX,\mft}^\mfs[V]$}%
  Finally, ${}^{\rosc}\mfB_{\mrIX,\mft}^{\mfs}[V]$
  \index{$\a$Aa@Notation!Sets of isometry classes of regular initial data!scmfBIXtsV@${}^{\rosc}\mfB_{\mrIX,\mft}^\mfs[V]$}%
  denotes the set of isometry classes of elements in ${}^{\rosc}\mB_{\mrIX,\mft}^{\mfs}[V]$. 
\end{definition}
\begin{remark}
  Let $0\leq V\in \mfP_{\a_V}^{0}$ for some $\a_V\in (0,1/3)$. Then $\mB_{\mrIX,+}[V]\subset \mB_{\mrIX,\roap}[V]$, so that
  $\mB_{\mrIX,\roap}[V]=\mB_{\mrIX,\ropp}[V]$; see Corollary~\ref{cor:ap eq pp}.
\end{remark}
Due to Lemmas~\ref{lemma:BianchiAdevelopment} and \ref{lemma:Bianchi IX remainder}, if $\mfI\in \mB_{\mrIX,\ropp}[V]$, then there is a $t_1\in J$ such
that $\theta(t)>0$ and $\theta_t(t)<0$ for $t\leq t_1$. Developments corresponding to initial data in $\mB_{\mrIX,\ropp}[V]$ therefore arise from the
following class of initial data.
\begin{definition}\label{def:plus ap pp nd}
  Let $0\leq V\in C^{\infty}(\ro)$, $\mfs\in \{\iso,\roLRS,\rogen\}$, $\mft\in \{+,\roap,\ropp\}$ and $\mfI\in \mB_{\mrIX,\mft}^{\mfs}[V]$. Let $J$ and
  $\theta$ be the existence interval and mean curvature, respectively, of the corresponding development, see Proposition~\ref{prop:unique max dev}.
  Let, moreover, $t_0$ be as in Definition~\ref{def:BianchiAdevelopment}. If $\theta(t_0)>0$ and $\theta_t(t)<0$ for all $t\leq t_0$, then
  $\mfI$ are said to be \textit{non-degenerate},
  \index{Non-degenerate!Initial data}%
  \index{Initial data!Non-degenerate}%
  \index{Initial data!Regular!Non-degenerate}%
  written $\mfI\in \mB_{\mrIX,\mft}^{\mfs,\rond}[V]$. 
  \index{$\a$Aa@Notation!Sets of regular initial data!mBIXtsndV@$\mB_{\mrIX,\mft}^{\mfs,\rond}[V]$}%
  The corresponding set of simply connected initial data sets is denoted ${}^{\rosc}\mB_{\mrIX,\mft}^{\mfs,\rond}[V]$,
  \index{$\a$Aa@Notation!Sets of regular initial data!scmBIXtsndV@${}^{\rosc}\mB_{\mrIX,\mft}^{\mfs,\rond}[V]$}%
  and ${}^{\rosc}\mfB_{\mrIX,\mft}^{\mfs,\rond}[V]$
  \index{$\a$Aa@Notation!Sets of isometry classes of regular initial data!scmfBIXtsndV@${}^{\rosc}\mfB_{\mrIX,\mft}^{\mfs,\rond}[V]$}%
  the isometry classes thereof. 
\end{definition}
The reason we introduce the notion of non-degenerate initial data is that it plays an important role in the definition of a smooth structure on
the set of isometry classes of developments in the Bianchi IX setting.

\subsection{Isometry classes of developments}\label{ssection:iso class dev}

Next, it is convenient to introduce the following terminology.
\begin{definition}\label{def:mDV}
  Let $V\in C^{\infty}(\ro)$. Let $\mfs\in \{\iso,\roLRS,\roper,\rogen\}$ and fix a Bianchi class A type $\mfT$. Then the set of isometry
  classes of developments $\mD[V](\mfI)$ arising from $\mfI\in {}^{\rosc}\mB_{\mfT}^{\mfs}[V]$ with strictly positive mean curvature
  is denoted ${}^{\rosc}\mfD_{\mfT}^{\mfs}[V]$.
  \index{$\a$Aa@Notation!Sets of isometry classes of developments!scmfDTsV@${}^{\rosc}\mfD_{\mfT}^{\mfs}[V]$}%
  If $V\geq 0$ and $\mft\in \{+,\roap,\ropp\}$, the set of isometry classes of developments $\mD[V](\mfI)$ for
  $\mfI\in {}^{\rosc}\mB_{\mrIX,\mft}^{\mfs}[V]$ is denoted ${}^{\rosc}\mfD_{\mrIX,\mft}^{\mfs}[V]$.
  \index{$\a$Aa@Notation!Sets of isometry classes of developments!scmfDIXtsV@${}^{\rosc}\mfD_{\mrIX,\mft}^{\mfs}[V]$}%
  The set of isometry classes of developments $\mD[V](\mfI)$ for $\mfI\in {}^{\rosc}\mB_{\mrI}^{\iso}[V]$ with a
  crushing singularity is denoted ${}^{\rosc}\mfD_{\mrI,\roc}^{\iso}[V]$.
  \index{$\a$Aa@Notation!Sets of isometry classes of developments!scmfDIcisoV@${}^{\rosc}\mfD_{\mrI,\roc}^{\iso}[V]$}%
\end{definition}

\section{Parametrising isometry classes of developments}\label{section:parametrising iso class dev}

We introduce notation for isometry classes of developments in Definition~\ref{def:mDV}. In order to discuss the regularity of the map from isometry
classes of developments to isometry classes of data on the singularity, we next note that ${}^{\rosc}\mfD_\mfT^\mfs[V]$ can be given the structure of a
smooth manifold. The relevant smooth structure arises from the natural smooth structures on sets of isometry classes of initial data with fixed mean
curvature. We therefore first introduce the following terminology.
\begin{definition}\label{def:id fixed vartheta zero}
  Fix $V\in C^\infty(\ro)$, $\mfs\in \{\iso,\roLRS,\roper,\rogen\}$, $\vartheta_0\in\ro$ and a Bianchi class A type $\mfT$. Then
  $\mB_{\mfT}^{\mfs}[V](\vartheta_0)$,
  \index{$\a$Aa@Notation!Sets of regular initial data!mBTsVthz@$\mB_{\mfT}^{\mfs}[V](\vartheta_0)$}%
  ${}^{\rosc}\mB_{\mfT}^{\mfs}[V](\vartheta_0)$
  \index{$\a$Aa@Notation!Sets of regular initial data!scmBTsV@${}^{\rosc}\mB_{\mfT}^{\mfs}[V](\vartheta_0)$}%
  and ${}^{\rosc}\mfB_{\mfT}^{\mfs}[V](\vartheta_0)$
  \index{$\a$Aa@Notation!Sets of isometry classes of regular initial data!scmfBTsVthz@${}^{\rosc}\mfB_{\mfT}^{\mfs}[V](\vartheta_0)$}%
  denote the subsets of $\mB_{\mfT}^{\mfs}[V]$, ${}^{\rosc}\mB_{\mfT}^{\mfs}[V]$ and ${}^{\rosc}\mfB_{\mfT}^{\mfs}[V]$ with mean curvature $\vartheta_0$
  respectively. Similarly, if $\vartheta_0>0$ and $\mft\in\{+,\roap,\ropp\}$, then $\mB_{\mrIX,\mft}^{\mfs,\rond}[V](\vartheta_0)$,
  \index{$\a$Aa@Notation!Sets of regular initial data!mBIXtsndVthz@$\mB_{\mrIX,\mft}^{\mfs,\rond}[V](\vartheta_0)$}%
  ${}^{\rosc}\mB_{\mrIX,\mft}^{\mfs,\rond}[V](\vartheta_0)$
  \index{$\a$Aa@Notation!Sets of regular initial data!scmBIXtsndVthz@${}^{\rosc}\mB_{\mrIX,\mft}^{\mfs,\rond}[V](\vartheta_0)$}%
  and ${}^{\rosc}\mfB_{\mrIX,\mft}^{\mfs,\rond}[V](\vartheta_0)$
  \index{$\a$Aa@Notation!Sets of isometry classes of regular initial data!scmfBIXtsndVthz@${}^{\rosc}\mfB_{\mrIX,\mft}^{\mfs,\rond}[V](\vartheta_0)$}%
  denote the subsets of $\mB_{\mrIX,\mft}^{\mfs,\rond}[V]$, ${}^{\rosc}\mB_{\mrIX,\mft}^{\mfs,\rond}[V]$ and
  ${}^{\rosc}\mfB_{\mrIX,\mft}^{\mfs,\rond}[V]$ with mean curvature $\vartheta_0$ respectively. In what follows it is also
  convenient to introduce the notation ${}^{\rosc}\mB_{\mrI,\rond}^{\iso}[V]$
  \index{$\a$Aa@Notation!Sets of regular initial data!scmBIndisoV@${}^{\rosc}\mB_{\mrI,\rond}^{\iso}[V]$}%
  and ${}^{\rosc}\mfB_{\mrI,\rond}^{\iso}[V]$
  \index{$\a$Aa@Notation!Sets of isometry classes of regular initial data!scmfBIndisoV@${}^{\rosc}\mfB_{\mrI,\rond}^{\iso}[V]$}%
  for the subset of ${}^{\rosc}\mB_{\mrI}^{\iso}[V]$ and ${}^{\rosc}\mfB_{\mrI}^{\iso}[V]$, respectively, with $\bphi_1\neq 0$.
  ${}^{\rosc}\mB_{\mrI,\rond}^{\iso}[V](\vartheta_0)$
  \index{$\a$Aa@Notation!Sets of regular initial data!scmBIndisoVthz@${}^{\rosc}\mB_{\mrI,\rond}^{\iso}[V](\vartheta_0)$}%
  and ${}^{\rosc}\mfB_{\mrI,\rond}^{\iso}[V](\vartheta_0)$
  \index{$\a$Aa@Notation!Sets of isometry classes of regular initial data!scmfBIndisoVthz@${}^{\rosc}\mfB_{\mrI,\rond}^{\iso}[V](\vartheta_0)$}%
  are the corresponding subsets with mean curvature $\vartheta_0$. 
\end{definition}
\begin{remark}
  The sets ${}^{\rosc}\mfB_{\mfT}^{\mfs}[V](\vartheta_0)$ and ${}^{\rosc}\mfB_{\mrI,\rond}^{\iso}[V](\vartheta_0)$ can be parametrised by smooth manifolds;
  see Remark~\ref{remark:par iso id fixed mc}. If $\a_V\in (0,1/3)$ and $0\leq V\in \mfP_{\a_V}^1$, then the sets
  ${}^{\rosc}\mfB_{\mrIX,\mft}^{\mfs,\rond}[V](\vartheta_0)$ can be parametrised by smooth manifolds, see Remark~\ref{remark:par iso id fixed mc} and
  Lemma~\ref{lemma:PimfT well def}. From now on, we tacitly assume that the above sets have been endowed with the corresponding smooth structures.
\end{remark}
In case $(\mfT,\mfs)\neq(\mrI,\iso)$, $\mfT\neq\mrIX$ and $0\leq V\in C^{\infty}(\ro)$, developments of initial data $\mfI\in\mB_\mfT^\mfs[V]$ are such
that $\theta_t<0$; cf. Lemma~\ref{lemma:BianchiAdevelopment}. Moreover, an element of ${}^{\rosc}\mfB_{\mfT}^{\mfs}[V](\vartheta_0)$ uniquely
determines an isometry class of developments; cf. Proposition~\ref{prop:iso id to iso sol}. These basic observations can be improved to yield the
following conclusion. 
\begin{lemma}\label{lemma:dev smo str nisoI nIX}
  Let $0\leq V\in C^{\infty}(\ro)$, $\vartheta\in I_V:=([3\inf_s V(s)]^{1/2},\infty)$ and fix a $(\mfT,\mfs)$ such that
  $(\mfT,\mfs)\neq(\mrI,\iso)$ and $\mfT\neq\mrIX$. Then the map taking
  $\mfI\in {}^{\rosc}\mB_{\mfT}^{\mfs}[V](\vartheta)$ to $\md[V](\mfI)$ induces an injective map
  \begin{equation}\label{eq:iota vartheta def}
    \iota_{\vartheta}:{}^{\rosc}\mfB_{\mfT}^{\mfs}[V](\vartheta)\rightarrow {}^{\rosc}\mfD_\mfT^\mfs[V].
  \end{equation}
  Moreover, ${}^{\rosc}\mfD_\mfT^\mfs[V]$ can be endowed with a topology and a smooth structure such that the following holds: the image of
  $\iota_{\vartheta}$ is an open subset of ${}^{\rosc}\mfD_\mfT^\mfs[V]$ for all $\vartheta\in I_V$; $\iota_{\vartheta}$ is a diffeomorphism onto its image
  for all $\vartheta\in I_V$; and the union of the images of the $\iota_{\vartheta}$ equals ${}^{\rosc}\mfD_\mfT^\mfs[V]$. 
\end{lemma}
\begin{proof}
  The proof can be found in Section~\ref{section:Id fix mean curv}.
\end{proof}
In the case of Bianchi type IX, the following holds. 
\begin{lemma}\label{lemma:dev smo str IX}
  Let $0\leq V\in \mfP_{\a_V}^1$, where $\a_V\in (0,1/3)$; $\mft\in \{+,\roap,\ropp\}$; and $0<\vartheta\in\ro$. Then the map taking
  $\mfI\in {}^{\rosc}\mB_{\mrIX,\mft}^{\mfs,\rond}[V](\vartheta)$ to $\md[V](\mfI)$ induces an injective map
  \begin{equation}\label{eq:iota vartheta def IX}
    \iota_{\vartheta}:{}^{\rosc}\mfB_{\mrIX,\mft}^{\mfs,\rond}[V](\vartheta)\rightarrow {}^{\rosc}\mfD_{\mrIX,\mft}^\mfs[V].
  \end{equation}
  Moreover, ${}^{\rosc}\mfD_{\mrIX,\mft}^{\mfs}[V]$ can be endowed with a topology and a smooth structure such that the following holds: the image of
  $\iota_{\vartheta}$ is an open subset of ${}^{\rosc}\mfD_{\mrIX,\mft}^\mfs[V]$ for all $\vartheta>0$; $\iota_{\vartheta}$ is a diffeomorphism onto its image
  for all $\vartheta>0$; and the union of the images of the $\iota_{\vartheta}$ equals ${}^{\rosc}\mfD_{\mrIX,\mft}^\mfs[V]$. 
\end{lemma}
\begin{remark}
  In case $\mft=+$, it is not necessary to assume that $V\in \mfP_{\a_V}^1$, it is sufficient to assume that $V\in C^{\infty}(\ro)$.
\end{remark}
\begin{proof}
  The proof can be found in Section~\ref{section:Id fix mean curv}.
\end{proof}
The case of isotropic Bianchi type I solutions is slightly more complicated. This is due to the existence of solutions whose metrics asymptote to
metrics of the form (\ref{eq:g dS}). Depending on the potential, there could in principle be infinitely many distinct solution with asymptotics of
this type. None of the corresponding solutions have a crushing singularity. On the other hand, in order for this to happen, the potential has to
be quite exotic; $V'$ has to have infinitely many distinct zeros with infinitely many distinct values of $V$. In the case of isotropic Bianchi type I
solutions, we, for these reasons, consider a more restricted class of potentials. 
\begin{definition}\label{def:mfPdef}
  Let $V\in C^{\infty}(\ro)$. If $V(s)$ converges to a finite number as $s\rightarrow \pm\infty$, denote the limit by $v_{\infty,\pm}$. If $V(s)$ does
  not converge to a finite number as $s\rightarrow \pm\infty$, let $v_{\infty,\pm}:=0$. Define
  \[
  v_{\max}(V):=\sup(\{v_{\infty,+},v_{\infty,-}\}\cup\{V(s_0)\, |\, s_0\in\ro,\, V'(s_0)=0\}).
  \]
  \index{$\a$Aa@Notation!Functions!vmax@$v_{\max}$}%
  Let $\mfP_{\ropar}$
  \index{$\a$Aa@Notation!Spaces of potentials!mfPpar@$\mfP_{\ropar}$}%
  denote the set of $V\in C^{\infty}(\ro)$ such that $V(s)\geq 0$ for all $s\in\ro$; $V'$ is bounded on every interval on which $V$
  is bounded; $V'(s)$ tends to a limit (finite or infinite) as $s\rightarrow \infty$ and as $s\rightarrow -\infty$; and $v_{\max}(V)<\infty$.
\end{definition}
\begin{remark}
  The set $\mfP_{\ropar}$ includes, e.g., the following three classes of potentials: non-negative polynomials; non-negative smooth functions such that
  $V>0$ outside a compact set and such that $V'/V$ converges to a non-zero limit as $s\rightarrow\pm\infty$; bounded non-negative smooth functions such
  that $V'(s)\rightarrow 0$ as $s\rightarrow\pm\infty$. 
\end{remark}
Under these additional assumptions on the potential, it is sufficient to use one fixed mean curvature to parametrise isometry classes of developments.
In fact, the following holds.
\begin{lemma}\label{lemma:parametr by single const th surf}
  Assume that the conditions of Lemma~\ref{lemma:dev smo str nisoI nIX} hold. Assume, in addition, that $V\in\mfP_{\ropar}$ and that
  $\vartheta>[3v_{\max}(V)]^{1/2}$. Then $\iota_{\vartheta}$ introduced in (\ref{eq:iota vartheta def}) is surjective. In particular,
  ${}^{\rosc}\mfB_{\mfT}^{\mfs}[V](\vartheta)$ and ${}^{\rosc}\mfD_\mfT^\mfs[V]$ are diffeomorphic.
\end{lemma}
\begin{proof}
  The surjectivity is an immediate consequence of Lemmas~\ref{lemma:dev smo str nisoI nIX}, \ref{lemma:BianchiAdevelopment} and
  \ref{lemma:bth large enough}. 
\end{proof}
Finally, we can state the result analogous to Lemmas~\ref{lemma:dev smo str nisoI nIX} and \ref{lemma:dev smo str IX} in the case of isotropic
Bianchi type I solutions.
\begin{lemma}\label{lemma:dev smo str isoI}
  Let $V\in\mfP_{\ropar}$ and $\vartheta\in I_V:=([3v_{\max}(V)]^{1/2},\infty)$. Then the map taking $\mfI\in {}^{\rosc}\mB_{\mrI,\rond}^{\iso}[V](\vartheta)$
  to $\md[V](\mfI)$ induces an injective map
  \begin{equation}\label{eq:iota vartheta nd def}
    \iota_{\vartheta}:{}^{\rosc}\mfB_{\mrI,\rond}^{\iso}[V](\vartheta)\rightarrow {}^{\rosc}\mfD_{\mrI,\roc}^{\iso}[V].
  \end{equation}
  Moreover, ${}^{\rosc}\mfD_{\mrI,\roc}^{\iso}[V]$ can be endowed with a topology and a smooth structure such that the following holds: the image of
  $\iota_{\vartheta}$ is an open subset of ${}^{\rosc}\mfD_{\mrI,\roc}^{\iso}[V]$ for all $\vartheta\in I_V$; $\iota_{\vartheta}$ is a diffeomorphism onto its
  image for all $\vartheta\in I_V$; and the union of the images of the $\iota_{\vartheta}$ equals ${}^{\rosc}\mfD_{\mrI,\roc}^{\iso}[V]$. 
\end{lemma}
\begin{remark}
  In contrast to Lemmas~\ref{lemma:dev smo str nisoI nIX} and \ref{lemma:dev smo str IX}, we here restrict our attention to initial data with
  $\bphi_1\neq 0$; see Definition~\ref{def:id fixed vartheta zero}. The reason for this is the following. The topology and smooth structure on
  ${}^{\rosc}\mfD_{\mrI,\roc}^{\iso}[V]$ are obtained by taking a quotient of the disjoint union of the images of the $\iota_\vartheta$. For this construction
  to lead to a smooth manifold, we need to verify that the map relating the different images (generated by the Einstein flow) is a local diffeomorphism.
  To this end, it is convenient to know that $\theta_t\neq 0$ at the constant-$t$ hypersurfaces to be identified. This is the main purpose served
  by the non degeneracy condition that $\bphi_1\neq 0$. 
\end{remark}
\begin{proof}
  The proof can be found in Section~\ref{section:Id fix mean curv}.
\end{proof}
Unfortunately, we do not obtain a result analogous to Lemma~\ref{lemma:parametr by single const th surf} in the case of isotropic Bianchi type I
solutions. Nevertheless, the following holds. 
\begin{lemma}\label{lemma:iso BI dev homeom one fix CMC}
  Let $V\in \mfP_{\ropar}$ and $\vartheta>[3v_{\max}(V)]^{1/2}$. Then ${}^{\rosc}\mfD_{\mrI,\roc}^{\iso}[V]$ is homeomorphic to
  ${}^{\rosc}\mfB_{\mrI}^{\iso}[V](\vartheta)$.
\end{lemma}
\begin{proof}
  The proof can be found in Section~\ref{section:Id fix mean curv}.
\end{proof}
\begin{remark}\label{remark:incompatible topologies}
  Due to Lemma~\ref{lemma:top of devel iso BI} and Remark~\ref{remark:par iso id fixed mc} it follows that ${}^{\rosc}\mfB_{\mrI}^{\iso}[V](\vartheta)$ has
  one of three possible topologies:
  two disjoint copies of $\ro$; $\ro$; and $\sn{1}$. On the other hand, due to Section~\ref{ssection:par data on sing}, it is clear that
  ${}^{\rosc}\mfS_{\mrI}^{\iso}$ is diffeomorphic to two copies of $\ro$. In general, there can therefore not be a diffeomorphism from
  ${}^{\rosc}\mfS_{\mrI}^{\iso}$ to ${}^{\rosc}\mfD_{\mrI,\roc}^{\iso}[V]$. In particular, this mere topological observation provides information
  concerning the dynamics. 
\end{remark}

\section[Map from developments to data on the singularity]{The map from isometry classes of developments to isometry classes
  of initial data on the singularity}\label{section:map iso dev to sdata}

Before turning to the regularity of the map from developments to data on the singularity, it is convenient to introduce the following terminology. 
\begin{definition}\label{def:mfR nBIX nIiso}
  Let $0\leq V\in\mfP_{\a_V}^2$ for some $\a_V\in (0,1)$ and fix a $(\mfT,\mfs)$ such that $(\mfT,\mfs)\neq(\mrI,\iso)$ and $\mfT\neq\mrIX$.
  Let $\vartheta\in I_V:=([3\inf_sV(s)]^{1/2},\infty)$ and ${}^{\rosc}_{\roqu}\mfB_\mfT^\mfs[V](\vartheta)$
  \index{$\a$Aa@Notation!Sets of isometry classes of regular initial data!scqumfBsTVth@${}^{\rosc}_{\roqu}\mfB_\mfT^\mfs[V](\vartheta)$}%
  denote the isometry classes of the elements in ${}^{\rosc}\mB_\mfT^\mfs[V](\vartheta)$ whose developments induce data on the singularity.
  Due to Lemma~\ref{lemma:iso ri iso ids incl type}, there is then a well defined map
  \begin{equation}\label{eq:mfR mfT mfs vartheta}
    \mfR_{\mfT,\mfs}^{\vartheta}:{}^{\rosc}_{\roqu}\mfB_\mfT^\mfs[V](\vartheta)\rightarrow {}^{\rosc}\mfS_\mfT^\mfs.
  \end{equation}
  \index{$\a$Aa@Notation!Maps!mfRTsth@$\mfR_{\mfT,\mfs}^{\vartheta}$}%
\end{definition}
With this terminology, the following holds. 
\begin{prop}\label{prop:dev to data on sing nIX}
  Assume that the conditions of Definition~\ref{def:mfR nBIX nIiso} are satisfied and that $V\in \mfP_{\a_V}^\infty$ for some $\a_V\in (0,1)$. Then
  ${}^{\rosc}_{\roqu}\mfB_\mfT^\mfs[V](\vartheta)$ is an open subset of
  ${}^{\rosc}\mfB_\mfT^\mfs[V](\vartheta)$ and the map (\ref{eq:mfR mfT mfs vartheta}) is a diffeomorphism onto its image. Moreover,
  using the notation introduced in Lemma~\ref{lemma:dev smo str nisoI nIX}, if $\vartheta_i\in I_V$, $i\in\{1,2\}$, and
  \begin{equation}\label{eq:iota vthonexo eq iota vthtwoextwo}
    \iota_{\vartheta_1}(x_1)=\iota_{\vartheta_2}(x_2)
  \end{equation}
  for some $x_i\in {}^{\rosc}_{\roqu}\mfB_\mfT^\mfs[V](\vartheta_i)$, $i\in\{1,2\}$, then $\mfR_{\mfT,\mfs}^{\vartheta_1}(x_1)=\mfR_{\mfT,\mfs}^{\vartheta_2}(x_2)$.
  Next, define
  \[
  {}^{\rosc}_{\roqu}\mfD_\mfT^\mfs[V]:=\bigcup_{\vartheta\in I_V}\iota_{\vartheta}\big({}^{\rosc}_{\roqu}\mfB_\mfT^\mfs[V](\vartheta)\big)    
  \]
  \index{$\a$Aa@Notation!Sets of isometry classes of developments!scqumfDTsV@${}^{\rosc}_{\roqu}\mfD_\mfT^\mfs[V]$}%
  and $\mfR_{\mfT,\mfs}:{}^{\rosc}_{\roqu}\mfD_\mfT^\mfs[V]\rightarrow {}^{\rosc}\mfS_\mfT^\mfs$
  \index{$\a$Aa@Notation!Maps!mfRTs@$\mfR_{\mfT,\mfs}$}%
  by
  \[
  \mfR_{\mfT,\mfs}(\iota_\vartheta(x))=\mfR_{\mfT,\mfs}^{\vartheta}(x).
  \]
  Then ${}^{\rosc}_{\roqu}\mfD_\mfT^\mfs[V]$ is an open subset of ${}^{\rosc}\mfD_\mfT^\mfs[V]$ and $\mfR_{\mfT,\mfs}$ is a diffeomorphism. 
\end{prop}
\begin{remark}\label{remark:map to sing diffeo}
  The statement that $\mfR_{\mfT,\mfs}$ is a diffeomorphism means, in particular, that the range of $\mfR_{\mfT,\mfs}$ is ${}^{\rosc}\mfS_\mfT^\mfs$.
  Note also that ${}^{\rosc}_{\roqu}\mfD_\mfT^\mfs[V]$ is the set of all isometry classes of simply connected developments of Bianchi type $\mfT$
  and symmetry type $\mfs$ (corresponding to initial data with strictly positive mean curvature), corresponding to a potential $V$, that give
  rise to data on the singularity. 
\end{remark}
\begin{remark}\label{remark:one CMC surf enough}
  If, in addition to the assumptions of the proposition, $V\in\mfP_{\ropar}$ and $\vartheta>[3v_{\max}(V)]^{1/2}$, then $\mfR_{\mfT,\mfs}^{\vartheta}$
  is a diffeomorphim. This follows by combining the proposition with Lemma~\ref{lemma:parametr by single const th surf}.
\end{remark}
\begin{proof}
  The proof is to be found in Section~\ref{section:Reg map from dev to dos}. 
\end{proof}

In the case of Bianchi type IX, we have to impose additional conditions; cf. Lemma~\ref{lemma:dev smo str IX}.
\begin{definition}\label{def:mfR BIX}
  Let $0\leq V\in\mfP_{\a_V}^2$ for some $\a_V\in (0,1)$; $\mft\in \{+,\roap,\ropp\}$; $\mfs\in \{\iso,\roLRS,\rogen\}$; and
  $0<\vartheta\in\ro$. Let ${}^{\rosc}_{\roqu}\mfB_{\mrIX,\mft}^{\mfs,\rond}[V](\vartheta)$
  \index{$\a$Aa@Notation!Sets of isometry classes of regular initial data!scqumfBIXtsndVth@${}^{\rosc}_{\roqu}\mfB_{\mrIX,\mft}^{\mfs,\rond}[V](\vartheta)$}%
  denote the isometry classes of
  the elements in ${}^{\rosc}\mB_{\mrIX,\mft}^{\mfs,\rond}[V](\vartheta)$ whose developments induce data on the singularity. Due to
  Lemma~\ref{lemma:iso ri iso ids incl type}, there is then a well defined map
  \begin{equation}\label{eq:mfR mfT mfs vartheta IX}
    \mfR_{\mrIX,\mfs,\mft}^{\vartheta}:{}^{\rosc}_{\roqu}\mfB_{\mrIX,\mft}^{\mfs,\rond}[V](\vartheta)\rightarrow {}^{\rosc}\mfS_{\mrIX}^\mfs.
  \end{equation}
  \index{$\a$Aa@Notation!Maps!mfRIXstth@$\mfR_{\mrIX,\mfs,\mft}^{\vartheta}$}%
\end{definition}
\begin{prop}\label{prop:dev to dos BIX}
  Assume that the conditions of Definition~\ref{def:mfR BIX} are satisfied and that, for some $\a_V\in (0,1/3)$, $V\in \mfP_{\a_V}^\infty$.
  Then ${}^{\rosc}_{\roqu}\mfB_{\mrIX,\mft}^{\mfs,\rond}[V](\vartheta)$ is an open subset of ${}^{\rosc}\mfB_{\mrIX}^{\mfs}[V](\vartheta)$
  and the map (\ref{eq:mfR mfT mfs vartheta IX}) is a diffeomorphism onto its image. Moreover, using the notation introduced in
  Lemma~\ref{lemma:dev smo str IX}, if $\vartheta_i>0$, $i\in\{1,2\}$, and
  \begin{equation}\label{eq:iota vthonexo eq iota vthtwoextwo IX}
    \iota_{\vartheta_1}(x_1)=\iota_{\vartheta_2}(x_2)
  \end{equation}
  for some $x_i\in {}^{\rosc}_{\roqu}\mfB_{\mrIX,\mft}^{\mfs,\rond}[V](\vartheta_i)$, $i\in\{1,2\}$, then
  $\mfR_{\mrIX,\mfs,\mft}^{\vartheta_1}(x_1)=\mfR_{\mrIX,\mfs,\mft}^{\vartheta_2}(x_2)$. Next, define
  \[
  {}^{\rosc}_{\roqu}\mfD_{\mrIX,\mft}^\mfs[V]:=\bigcup_{\vartheta>0}\iota_{\vartheta}\big({}^{\rosc}_{\roqu}\mfB_{\mrIX,\mft}^{\mfs,\rond}[V](\vartheta)\big)    
  \]
  \index{$\a$Aa@Notation!Sets of isometry classes of developments!scqumfDIXtsV@${}^{\rosc}_{\roqu}\mfD_{\mrIX,\mft}^\mfs[V]$}%
  and $\mfR_{\mrIX,\mfs,\mft}:{}^{\rosc}_{\roqu}\mfD_{\mrIX,\mft}^\mfs[V]\rightarrow {}^{\rosc}\mfS_{\mrIX}^\mfs$ by
  \[
  \mfR_{\mrIX,\mfs,\mft}\big(\iota_\vartheta(x)\big)=\mfR_{\mrIX,\mfs,\mft}^{\vartheta}(x).
  \]
  \index{$\a$Aa@Notation!Maps!mfRIXst@$\mfR_{\mrIX,\mfs,\mft}$}%
  Then ${}^{\rosc}_{\roqu}\mfD_{\mrIX,\mft}^\mfs[V]$ is an open subset of ${}^{\rosc}\mfD_{\mrIX,\mft}^\mfs[V]$ and $\mfR_{\mrIX,\mfs,\mft}$ is a diffeomorphism
  onto its image (which is a diffeomorphism in case $\mft\in\{\roap,\ropp\}$).  
\end{prop}
\begin{remark}
  An observation similar to Remark~\ref{remark:map to sing diffeo} holds in this case as well.
\end{remark}
\begin{remark}
  In case $\mft=+$, the condition on $\a_V$ can be relaxed to $\a_V\in (0,1)$. 
\end{remark}
\begin{proof}
  The proof is to be found in Section~\ref{section:Reg map from dev to dos}. 
\end{proof}
Finally, in the isotropic Bianchi type I setting, we introduce the following terminology. 
\begin{definition}\label{def:mfR BI}
  Let $V\in \mfP_{\ropar}\cap\mfP_{\a_V}^2$ for some $\a_V\in (0,1)$ and let $\vartheta>[3v_{\max}(V)]^{1/2}$. Let ${}^{\rosc}_{\roqu}\mfB_{\mrI,\rond}^{\iso}[V](\vartheta)$
  \index{$\a$Aa@Notation!Sets of isometry classes of regular initial data!scqumfBIndisoVth@${}^{\rosc}_{\roqu}\mfB_{\mrI,\rond}^{\iso}[V](\vartheta)$}%
  denote the isometry classes of
  the elements in ${}^{\rosc}\mB_{\mrI,\rond}^{\iso}[V](\vartheta)$ whose developments induce data on the singularity. Due to
  Lemma~\ref{lemma:iso ri iso ids incl type}, there is then a well defined map
  \begin{equation}\label{eq:mfR vartheta I iso}
    \mfR_{\mrI,\iso}^{\vartheta}:{}^{\rosc}_{\roqu}\mfB_{\mrI,\rond}^{\iso}[V](\vartheta)\rightarrow {}^{\rosc}\mfS_{\mrI}^{\iso}.
  \end{equation}
  \index{$\a$Aa@Notation!Maps!mfRIisoth@$\mfR_{\mrI,\iso}^{\vartheta}$}%
\end{definition}
\begin{prop}\label{prop:dev to dos iso BI}
  Assume that the conditions of Definition~\ref{def:mfR BI} are satisfied and that $V\in \mfP_{\a_V}^\infty$ for some $\a_V\in (0,1)$.
  Then ${}^{\rosc}_{\roqu}\mfB_{\mrI,\rond}^{\iso}[V](\vartheta)$ is an open subset of ${}^{\rosc}\mfB_{\mrI,\rond}^{\iso}[V](\vartheta)$
  and the map (\ref{eq:mfR vartheta I iso}) is a diffeomorphism onto its image. Moreover, using the notation introduced in
  Lemma~\ref{lemma:dev smo str isoI}, if $\vartheta_i>[3v_{\max}(V)]^{1/2}$, $i\in\{1,2\}$, $\vartheta_1<\vartheta_2$ and
  \begin{equation}\label{eq:iota vthonexo eq iota vthtwoextwo I iso}
    \iota_{\vartheta_1}(x_1)=\iota_{\vartheta_2}(x_2)
  \end{equation}
  for some $x_i\in {}^{\rosc}_{\roqu}\mfB_{\mrI,\rond}^{\iso}[V](\vartheta_i)$, $i\in\{1,2\}$, then
  $\mfR_{\mrI,\iso}^{\vartheta_1}(x_1)=\mfR_{\mrI,\iso}^{\vartheta_2}(x_2)$. Next, define
  \[
    {}^{\rosc}_{\roqu}\mfD_{\mrI,\roc}^{\iso}[V]
    :=\bigcup_{\vartheta>[3v_{\max}(V)]^{1/2}}\iota_{\vartheta}\big({}^{\rosc}_{\roqu}\mfB_{\mrI,\rond}^{\iso}[V](\vartheta)\big)    
  \]
  \index{$\a$Aa@Notation!Sets of isometry classes of developments!scqumfDIcisoV@${}^{\rosc}_{\roqu}\mfD_{\mrI,\roc}^{\iso}[V]$}%
  and $\mfR_{\mrI,\iso}:{}^{\rosc}_{\roqu}\mfD_{\mrI,\roc}^{\iso}[V]\rightarrow {}^{\rosc}\mfS_{\mrI}^{\iso}$ by
  \[
  \mfR_{\mrI,\iso}\big(\iota_\vartheta(x)\big)=\mfR_{\mrI,\iso}^{\vartheta}(x).
  \]
  \index{$\a$Aa@Notation!Maps!mfRIiso@$\mfR_{\mrI,\iso}$}%
  Then ${}^{\rosc}_{\roqu}\mfD_{\mrI,\roc}^{\iso}[V]$ is an open subset of ${}^{\rosc}\mfD_{\mrI,\roc}^{\iso}[V]$ and $\mfR_{\mrI,\iso}$ is a diffeomorphism.  
\end{prop}
\begin{remark}
  An observation similar to Remark~\ref{remark:map to sing diffeo} holds in this case as well.
\end{remark}
\begin{proof}
  The proof is to be found in Section~\ref{section:Reg map from dev to dos}. 
\end{proof}

\section[Asymptotics, Singularity, General Case]{Asymptotics in the direction of the singularity, the general case}\label{ssection:asdirofsing}
Our next goal is to derive the asymptotics in the direction of the singularity. In particular, we want to prove that solutions induce data on
the singularity in the sense of Definition~\ref{def:ndvacidonbbssh}. However, this does not always occur. One exception is, e.g., the case of
vacuum Bianchi type VIII and IX solutions, where, by vacuum, we mean that the potential and the scalar field vanish identically. In this case,
the asymptotics are oscillatory, if one excludes LRS solutions; see
\cite[Theorem~5, p.~727]{cbu}. On the other hand, in the case of the Einstein scalar field equations, there is a clear dichotomy: either the
solution is vacuum, or, if there is a non-trivial matter content, the matter dominates and gives rise to convergent behaviour; see
\cite[Theorem~19.1, p.~478]{BianchiIXattr} (keeping in mind that scalar field matter is equivalent to an orthogonal stiff perfect fluid). 
The reason for this dichotomy is that in the Einstein scalar field case, there is a monotonicity ensuring that the expansion normalised energy
density increases in the direction of the singularity. Introducing a non-zero potential, this monotonicity is destroyed. This causes significant
complications, and these complications are, in part, the cause of the length of this article. 

In the presence of a potential, there is still a dichotomy between matter and vacuum dominated asymptotic behaviour in the direction of the
singularity. However, the dichotomy is not characterised by initial data, but rather by the asymptotics. Moreover, in order to obtain the
dichotomy, we need to impose conditions on the potential. In what follows, we therefore assume that $V\in \mfP_{\a_V}^1$; see
Definition~\ref{def:mfP a inf}. The restrictions on $\a_{V}$ depend on the Bianchi type: $\a_{V}\in (0,1)$ in the case of anisotropic Bianchi type I
and non-LRS Bianchi type II; and $\a_{V}\in (0,1/3)$ otherwise. Note that we here, as always, exclude isotropic and LRS Bianchi type VII${}_0$.

\begin{prop}\label{prop:dichotomy}
  Let $\mfT$ be a Bianchi class A type, $\mfs\in\{\iso,\roLRS,\roper,\rogen\}$ and $V\in \mfP_{\a_V}^1$ be non-negative. Assume that $\a_{V}\in (0,1)$ in
  the case of anisotropic Bianchi type I and non-LRS Bianchi type II; and that $\a_{V}\in (0,1/3)$ otherwise. Let $\mfI\in \mB_{\mfT}^\mfs[V]$, assume
  that $\tr_{\bge}\bk\geq 0$; that $(\mfT,\mfs)\neq (\mrI,\iso)$; and that $\mfI\in \mB_{\mrIX,+}[V]$ in case $\mfT=\mrIX$. Let $(M,g,\phi)=\mD[V](\mfI)$.
  Then the associated existence interval is of the form $(0,t_+)$ and $\theta(t)\rightarrow\infty$ as $t\downarrow 0$. Moreover, there are two
  possibilities. Either there is a $t_{0}>0$ and a $C\in\rn{}$ such that $|\theta(t)\phi_{t}(t)|\leq C$ for all $t\leq t_{0}$; or $\phi_{t}(t)/\theta(t)$
  converges to a non-zero limit as $t\downarrow 0$. 
\end{prop}
\begin{remark}\label{remark:aV in zero one third}
  It is natural to conjecture that the result holds for all $\a_V\in (0,1)$. However, the assumption $\a_{V}\in (0,1/3)$ is needed for our method of
  proof to work for the higher Bianchi types. 
\end{remark}
\begin{proof}
  The statement follows from Lemma~\ref{lemma:BianchiAdevelopment} and Theorem~\ref{thm:dichotomy}, keeping the Wainwright-Hsu formulation of
  Section~\ref{ssection:whsuform} in mind.
\end{proof}
In addition, if $\mfI\in\mB_{\mrIX,\roap}[V]$, then either Proposition~\ref{prop:dichotomy} applies or the following conclusion holds.
\begin{prop}\label{prop:dichotomy roap}
  Let $0\leq V\in C^\infty(\ro)$ and $\mfI\in \mB_{\mrIX,\roap}[V]$ be such that $\mfI\notin \mB_{\mrIX,+}[V]$. Let $(M,g,\phi)=\mD[V](\mfI)$. Then
  $|\phi_t(t)/\theta(t)|$ converges to $(2/3)^{1/2}$ as $t\downarrow 0$. 
\end{prop}
\begin{proof}
  By assumption, $X(t)/\theta^2(t)\rightarrow 0$ and $V\circ\phi(t)/\theta^2(t)\rightarrow 0$ as $t\downarrow 0$. The conclusion follows. 
\end{proof}
With the above observations in mind, it is natural to introduce the following terminology. 
\begin{definition}\label{def:matter and vacuum dominated}
  A Bianchi class A non-linear scalar field development as in Proposition~\ref{prop:dichotomy} or Proposition~\ref{prop:dichotomy roap}  is said to
  be \textit{matter dominated}
  \index{Matter dominated!Development}%
  \index{Matter dominated!Singularity}%
  \index{Development!Matter dominated}%
  \index{Singularity!Matter dominated}%
  if $\phi_{t}(t)/\theta(t)$ converges to a non-zero limit as $t\downarrow 0$, and is said to be \textit{vacuum dominated}
  \index{Vacuum dominated!Development}%
  \index{Vacuum dominated!Singularity}%
  \index{Development!Vacuum dominated}%
  \index{Singularity!Vacuum dominated}%
  otherwise. 
\end{definition}
\begin{remark}
  It is natural to ask if only vacuum solution are vacuum dominated. However, this is not the case. In fact,
  Theorem~\ref{thm:dataonsingtosolution} applies with $\Phi_{1}=0$ in the case of Bianchi types I, II, VI${}_{0}$, VII${}_{0}$ and LRS Bianchi types
  VIII and IX. The corresponding solutions are clearly not matter dominated. Moreover, in the case that the potential is such that there are no
  constant solutions to the non-linear scalar field equation, and in the case that $\Phi_{0}$ is such that $V'(\Phi_{0})\neq 0$, we can clearly
  conclude that the scalar field is non-constant, so that the solution is not a vacuum solution. In the case of Bianchi type VIII and IX solutions
  which are neither isotropic nor LRS, all the solutions for which $\phi$ is constant are vacuum dominated, but proving the existence of vacuum
  dominated solutions for which $\phi$ is non-constant is more difficult. 
\end{remark}
Next, we state the conclusions concerning the asymptotics in the direction of the singularity.
\begin{thm}\label{thm:dev inducing data on the sing}
  Let $\mfT$ be a Bianchi class A type, $\mfs\in\{\iso,\roLRS,\roper,\rogen\}$ and $V\in \mfP_{\a_V}^1$ be non-negative, where $\a_V\in (0,1)$ in case
  of Bianchi type I and non-LRS Bianchi type II; and $\a_V\in (0,1/3)$ otherwise. Assume $(\mfT,\mfs)\neq (\mrI,\iso)$ and let
  $\mfI\in\mB_{\mfT}^\mfs[V]$ with $\tr_{\bge}\bk\geq 0$. In case
  $\mfT=\mrIX$ assume, in addition, that $\mfI\in\mB_{\mrIX,\ropp}^\mfs[V]$. Then the development $\mD[V](\mfI)$ induces initial data on the singularity
  unless it is vacuum dominated, $\mfs=\rogen$ and $\mfT\in\{\mrVIII,\mrIX\}$. Finally, if $\mfs=\rogen$, $\mfT\in\{\mrVIII,\mrIX\}$ and
  $\mD[V](\mfI)$ is vacuum dominated, then the expansion normalised Weingarten map $\mK$ does not converge. In fact, the $\a$-limit set of the
  eigenvalues of $\mK$ contains two distinct points on the Kasner circle and the line connecting them. Moreover, $\bS/\theta^2$ does not converge
  to zero. 
\end{thm}
\begin{remark}
  Remark~\ref{remark:aV in zero one third} is equally relevant in this case.
\end{remark}
\begin{remark}
  In the statement of the theorem, the \textit{Kasner plane}, say $P_K$, is represented by the set of $q\in\rn{3}$ satisfying $q_1+q_2+q_3=1$ and
  the \textit{Kasner circle} is represented by the $q\in P_K$ such that $q_1^2+q_2^2+q_3^2=1$. Moreover, by the $\a$-limit set of the eigenvalues
  of $\mK$, we mean the set of $q\in P_K$ such that there is a sequence $t_k\downarrow 0$ (where $t=0$ represents the singularity; i.e.,
  $\theta(t)\rightarrow\infty$ as $t\downarrow 0$) with the property that $[p_1(t_k),p_2(t_k),p_3(t_k)]$ converges to $q$. Here $p_i(t)$, $i=1,2,3$,
  are the eigenvalues of $\mK(t)$. However, we do not impose a specific order on the $p_i$'s, so that the $\a$-limit set is invariant under permutations.  
\end{remark}
\begin{proof}
  The statements concerning convergence follow from Propositions~\ref{prop:matter dom data on sing}, \ref{prop:ap not plus},
  \ref{prop:I and II not iso as}, \ref{prop:VIz and VIIz not iso as} and \ref{prop:LRS VIII and IX}. The statement concerning oscillations follows
  from Proposition~\ref{cor:BianchiVIIIvacuumas}. 
\end{proof}
Combining this observation with Proposition~\ref{prop:dev to data on sing nIX} yields the following conclusion.
\begin{cor}
  Let $\mfT$ be a Bianchi class A type, $\mfs\in\{\iso,\roLRS,\roper,\rogen\}$ and $V\in \mfP_{\a_V}^\infty$ be non-negative, where $\a_V\in (0,1)$ in case
  of Bianchi type I and non-LRS Bianchi type II; and $\a_V\in (0,1/3)$ otherwise. Assume that $(\mfT,\mfs)\neq(\mrI,\iso)$ and $\mfT\neq\mrIX$.
  Then, if $(\mfT,\mfs)\neq (\mrVIII,\rogen)$, ${}^{\rosc}_{\roqu}\mfD_{\mfT}^{\mfs}[V]={}^{\rosc}\mfD_{\mfT}^{\mfs}[V]$ and
  the map $\mfR_{\mfT,\mfs}$ introduced in the statement of Proposition~\ref{prop:dev to data on sing nIX} defines a diffeomorphism from the set of all
  isometry classes of developments of simply connected initial data to the set of all isometry classes of simply connected initial data on the
  singularity. Moreover, ${}^{\rosc}_{\roqu}\mfD_{\mrVIII}^{\rogen}[V]$ coincides with the isometry classes of matter dominated developments.  
\end{cor}
\begin{remark}\label{remark:matter dominated smooth subm}
  Assuming $(\mfT,\mfs)\neq (\mrVIII,\rogen)$, it is clear that the subset of ${}^{\rosc}\mfD_{\mfT}^{\mfs}[V]$ induced by the vacuum dominated
  developments is a smooth
  codimension one submanifold; this is an immediate consequence of the existence of the diffeomorphism and the structure of the set of isometry classes
  of simply connected initial data on the singularity; see Section~\ref{ssection:par data on sing}. Moreover, in case $(\mfT,\mfs)=(\mrVIII,\rogen)$,
  the set of isometry classes of matter dominated developments is an open subset of ${}^{\rosc}\mfD_{\mrVIII}^{\rogen}[V]$. 
\end{remark}
\begin{remark}\label{remark:eigenvalues deg}
  Assuming either $(\mfT,\mfs)\neq (\mrVIII,\rogen)$ or that we are in the matter dominated case with $(\mfT,\mfs)=(\mrVIII,\rogen)$, the following
  holds. The subset of ${}^{\rosc}_{\roqu}\mfD_{\mfT}^{\mfs}[V]$ inducing data on the singularity with $\msK=\roId/3$ is a smooth codimension $2$
  submanifold if $\mfs\notin \{\iso,\roLRS\}$ and a smooth codimension $1$ submanifold if $\mfs=\roLRS$. Moreover, the subset inducing data such that
  two of the eigenvalues of $\msK$ coincide and differ from the third is a smooth codimension $1$ submanifold if $\mfs\notin \{\iso,\roLRS\}$. Again,
  the argument is similar to the one presented in the case of Remark~\ref{remark:matter dominated smooth subm}. 
\end{remark}
\begin{remark}\label{remark:horizon case}
  Consider initial data on the singularity with the property that $1$ is an eigenvalue of $\msK$. Assuming the initial data on the singularity
  to correspond to a development, they have to be LRS; see Theorem~\ref{thm:dataonsingtosolution}. Moreover, the subset of
  ${}^{\rosc}\mfD_{\mfT}^{\roLRS}[V]$ inducing isometry classes of such data on the singularity is a smooth codimension $1$ submanifold. 
\end{remark}
Next, we turn to Bianchi type IX. Combining Proposition~\ref{prop:dev to dos BIX} with Theorem~\ref{thm:dev inducing data on the sing} yields the
following corollary.
\begin{cor}
  Let $\mfs\in\{\iso,\roLRS,\rogen\}$ and $V\in \mfP_{\a_V}^\infty$ be non-negative, where $\a_V\in (0,1/3)$. Then, if $\mfs\neq\rogen$,  
  ${}^{\rosc}_{\roqu}\mfD_{\mrIX,\ropp}^{\mfs}[V]={}^{\rosc}\mfD_{\mrIX,\ropp}^{\mfs}[V]$ and the map $\mfR_{\mrIX,\mfs,\ropp}$ introduced in the statement of
  Proposition~\ref{prop:dev to dos BIX} defines a diffeomorphism from the set of all isometry classes of developments of initial
  data in ${}^{\rosc}\mB^{\mfs}_{\mrIX,\ropp}[V]$ to the set of all isometry classes of simply connected initial data on the singularity. Moreover,
  ${}^{\rosc}_{\roqu}\mfD_{\mrIX,\ropp}^{\rogen}[V]$ coincides with the isometry classes of matter dominated developments. 
\end{cor}
\begin{remark}
  Observations similar to Remarks~\ref{remark:matter dominated smooth subm}--\ref{remark:horizon case} hold in this case as well. 
\end{remark}

\section[Asymptotics, $k=0$ FLRW]{Asymptotics in the direction of the singularity, isotropic and spatially flat setting}\label{section:k eq zero case}

The results of the previous two sections illustrate that the Einstein flow defines a diffeomorphism between isometry classes of developments and isometry
classes of data on the singularity in many situations. However, the isotropic Bianchi type I setting is special for several reasons. To begin
with, there are the solutions described in Remark~\ref{remark:dev of trivial id}. However, even if we exclude these solutions by restricting our attention
to developments with a crushing singularity, we have the following obstruction to the existence of a diffeomorphism: Due to
Lemma~\ref{lemma:iso BI dev homeom one fix CMC} and Remark~\ref{remark:incompatible topologies}, there are potentials such that the set of isometry
classes of simply connected isotropic Bianchi type I developments with a crushing singularity have a different topology from the set of isometry classes
of simply connected isotropic Bianchi type I data on the singularity. For this simple reason, there must, for such potentials, be exceptional developments
that have a crushing singularity but do not induce initial data on the singularity; cf. Proposition~\ref{prop:dev to dos iso BI}. This complicates the
analysis. Nevertheless, it is possible to draw conclusions, some of which depend on more detailed assumptions concerning the potential. The state
space of interest is ${}^{\rosc}\mfB_{\mrI}^{\iso}[V]$. We endow this space with a smooth structure in Remark~\ref{remark:sc mfB mfT mfs param}.
However, the corresponding parametrisation arises from a general perspective containing redundant information. It is therefore convenient
to first introduce the following simplified terminology.
\begin{definition}\label{def:M plus}
  Fix a $V\in C^{\infty}(\rn{})$ and a $\vartheta_0\in\ro$. Then $B^{\iso}_{\mrI,+}$ and $B^{\iso}_{\mrI,+}(\vartheta_0)$ are defined by 
  \begin{align*}
    B^{\iso}_{\mrI,+} := & \{(\theta_0,\phi_0,\phi_1)\in \rn{3}\, |\, \theta_0^2=3\phi_1^2/2+3V(\phi_0),\ (\phi_1,V'(\phi_0))\neq (0,0)\},\\
    B^{\iso}_{\mrI,+}(\vartheta_0) := & \{(\phi_0,\phi_1)\in \rn{2}\, |\, \vartheta_0^2=3\phi_1^2/2+3V(\phi_0),\ (\phi_1,V'(\phi_0))\neq (0,0)\}.    
  \end{align*}
  \index{$\a$Aa@Notation!Symmetry reduced sets of regular initial data!BisoI+@$B^{\iso}_{\mrI,+}$}%
  \index{$\a$Aa@Notation!Symmetry reduced sets of regular initial data!BisoI+thz@$B^{\iso}_{\mrI,+}(\vartheta_0)$}%
\end{definition}
\begin{remark}
  The Hamiltonian constraint is $\theta_0^2=3\phi_1^2/2+3V(\phi_0)$ in the isotropic Bianchi type I setting; see
  (\ref{eq:ham con original}). The condition $(\phi_1,V'(\phi_0))\neq (0,0)$ excludes trivial initial data.
\end{remark}
\begin{remark}\label{remark: B iso mrI plus smooth struct}
  The non-triviality condition ensures that $B^{\iso}_{\mrI,+}$ is a smooth $2$-dimensional manifold and that if $B^{\iso}_{\mrI,+}(\vartheta_0)$ is
  non-empty, then it is a smooth $1$-dimensional manifold. These manifolds can be identified (via a diffeomorphism) with
  ${}^{\rosc}\mfB_{\mrI}^{\iso}[V]$ and ${}^{\rosc}\mfB_{\mrI}^{\iso}[V](\vartheta_0)$ respectively; see Lemma~\ref{lemma:sc mfB mfT mfs param} and
  Remarks~\ref{remark:par iso id fixed mc} and \ref{remark:sfR I iso B I plus iso}. 
\end{remark}
\begin{remark}
  Since the conditions $\phi_1=0$ and $V'(\phi_0)=0$ are preserved by the evolution, it is clear that the non-triviality condition is preserved
  by the evolution. Thus a solution corresponding to initial data in $B^{\iso}_{\mrI,+}$ remains in $B^{\iso}_{\mrI,+}$.
\end{remark}
Assume now that $V\in\mfP_{\ropar}$ and that $\vartheta_0>[3v_{\max}(V)]^{1/2}$. Then, due to Lemma~\ref{lemma:iso BI dev homeom one fix CMC} and
Remark~\ref{remark: B iso mrI plus smooth struct}, ${}^{\rosc}\mfD_{\mrI,\roc}^{\iso}[V]$ and $B^{\iso}_{\mrI,+}(\vartheta_0)$ are homeomorphic. Moreover,
due to Lemma~\ref{lemma:top of devel iso BI} and Remark~\ref{remark:sfR I iso B I plus iso}, the following holds:
\[
B^{\iso}_{\mrI,+}(\vartheta_0)\cong\left\{\begin{array}{cl} \ro\sqcup\ro & \mathrm{if}\ V\ \mathrm{bounded},\\
\ro & \mathrm{if}\ \{V(s)|s\leq 0\}\ \mathrm{bounded\ and }\ \{V(s)|s\geq 0\}\ \mathrm{unbounded},\\
\ro & \mathrm{if}\ \{V(s)|s\leq 0\}\ \mathrm{unbounded\ and }\ \{V(s)|s\geq 0\}\ \mathrm{bounded},\\
\sn{1} & \mathrm{if}\ \{V(s)|s\leq 0\}\ \mathrm{unbounded\ and }\ \{V(s)|s\geq 0\}\ \mathrm{unbounded},
\end{array}\right.
\]
where $\cong$ signifies the existence of a diffeomorphism. On the other hand, due to Section~\ref{ssection:par data on sing}, it is clear that
${}^{\rosc}\mfS_{\mrI}^{\iso}$ is diffeomorphic to two copies of $\ro$. Only in the case of bounded potentials can we thus expect a diffeomorphism from
the set of isometry classes of developments with a crushing singularity to the set of isometry classes of data on the singularity. In fact, this
is exactly what happens.
\begin{thm}
  Assume $V\in C^{\infty}(\ro)$ to be bounded and to be such that $V\in\mfP_{\ropar}\cap\mfP_{\a_V}^\infty$ for some $\a_V\in (0,1)$. Then
  ${}^{\rosc}_{\roqu}\mfD_{\mrI,\roc}^{\iso}[V]={}^{\rosc}\mfD_{\mrI,\roc}^{\iso}[V]$
  and the map $\mfR_{\mrI,\iso}$ introduced in Proposition~\ref{prop:dev to dos iso BI} is a diffeomorphism from ${}^{\rosc}\mfD_{\mrI,\roc}^{\iso}[V]$
  to $\mfS_{\mrI}^{\iso}$. 
\end{thm}
\begin{remark}
  In order to prove that isotropic Bianchi type I developments with a crushing singularity induce data on the singularity, it is sufficient to assume
  that $0\leq V\in C^\infty(\ro)$ and that $V$ is bounded. The additional requirements on $V$ are there to ensure that $\mfR_{\mrI,\iso}$ is a diffeomorphism. 
\end{remark}
\begin{proof}
  The statement is an immediate consequence of Proposition~\ref{prop:dev to dos iso BI}, Corollary~\ref{cor: Bianchi I isotropic V bd},
  Lemma~\ref{lemma: Bianchi I isotropic} and Remark~\ref{remark:data induced on sing BI iso}.
\end{proof}
Turning, next, to the case that $V$ is bounded for $s\leq 0$ and unbounded for $s\geq 0$ (or vice versa), it is natural to first consider the
case of an exponential potential: $V(s)=V_0e^{\lambda s}$ for some constants $V_0,\lambda>0$. In this case, the topology of the set of isometry classes
of developments with a crushing singularity is $\ro$. A natural way to modify this set in order to obtain two copies of $\ro$ (the topology of the set
of isometry classes of data on the singularity) is to remove one point. In other words, the naive conjecture would be that there should be one
exceptional isometry class of developments with a crushing singularity that does not induce initial data on the singularity, and that the remaining
isometry classes do induce data on the singularity. This is what happens.
\begin{prop}\label{prop:exp pot asympt}
  Let $0<V_0\in\rn{}$ and $0<\lambda<\sqrt{6}$. Define $V(s):=V_0e^{\lambda s}$. Let $\theta\in C^{\infty}(J,(0,\infty))$ and $\phi\in C^{\infty}(J,\rn{})$
  be the mean curvature and the scalar field of a development corresponding to initial data in $B^{\iso}_{\mrI,+}$, where $J=(t_-,t_+)$ is the maximal
  existence interval. Then $t_->-\infty$, $t_+=\infty$ and $\theta(t)\rightarrow\infty$ as $t\downarrow t_-$. Moreover, either the solution induces
  initial data on the singularity; or it does not induce initial data on the singularity and it is the unique solution satisfying
  $3\phi_t(t)/\theta(t)=-\lambda$ for all $t\in J$. 
  In particular, removing the unique solution satisfying $3\phi_t(t)/\theta(t)=-\lambda$ for all $t$ from
  ${}^{\rosc}\mfD_{\mrI,\roc}^{\iso}[V]$ yields ${}^{\rosc}_{\roqu}\mfD_{\mrI,\roc}^{\iso}[V]$ (which is diffeomorphic to ${}^{\rosc}\mfS_\mrI^{\iso}$ via
  the map $\mfR_{\mrI,\iso}$ introduced in Proposition~\ref{prop:dev to dos iso BI}).
\end{prop}
\begin{proof}
  The statement follows from Proposition~\ref{prop:dev to dos iso BI}, Lemma~\ref{lemma:BianchiAdevelopment},
  Proposition~\ref{eq:gen isotropic Bianchi I exp pot},
  Lemma~\ref{lemma: Bianchi I isotropic}, Remark~\ref{remark:data induced on sing BI iso} and the following argument. The unique solution
  satisfying $3\phi_t(t)/\theta(t)=-\lambda$ for all $t$ has the property that $\phi_t/\theta$ converges to $\Phi_1:=-\lambda/3$ in the direction
  of the singularity. If this solution induced data on the singularity, we would, on the other hand, have $\Phi_1^2=2/3$; this follows from
  Definition~\ref{def:ndvacidonbbssh} and the fact that $\msK=\roId/3$ in the isotropic setting. Combining these observations yields
  $\lambda^2=6$, contradicting the fact that $\lambda\in (0,\sqrt{6})$. 
\end{proof}
Even though this result is of interest, the relevant class of potentials is very restricted. It would be of interest to prove a similar result for
a more general class of potentials. We, next, do so in a setting in which the potential is unbounded in both directions. In this case, the topology
of the set of isometry classes of developments with a crushing singularity is $\sn{1}$. A natural way to modify this set in order to obtain two copies
of $\ro$ (the topology of the set of isometry classes of data on the singularity) is to remove two points. In other words, the naive conjecture would
be that there should be two exceptional isometry classes of developments with a crushing singularity that do not induce initial data on the singularity,
and that the remaining isometry classes do induce data on the singularity. This is what happens. More specifically, the following result holds. 
\begin{thm}\label{thm:asympt as exp pot intro}
  Assume $0\leq V\in C^{\infty}(\rn{})$ and that there are constants $C_V$ and $M$ such that $V(s)>0$ and 
  \begin{equation}\label{eq:ln V biss est intro}
    \left|\left(\ln V\right)''(s)\right|\leq C_V\ldr{s}^{-2}
  \end{equation}
  for all $|s|\geq M$. This means that $(\ln V)'(s)$ converges to limits as $s\rightarrow\pm \infty$. Call the limits $\lambda_\pm$ and assume
  that $-\sqrt{6}<\lambda_-<0$ and that $0<\lambda_+<\sqrt{6}$. Let $\theta\in C^{\infty}(J,(0,\infty))$ and $\phi\in C^{\infty}(J,\rn{})$ be the
  mean curvature and the scalar field of a development corresponding to initial data in $B^{\iso}_{\mrI,+}$, where $J=(t_-,t_+)$ is the maximal
  existence interval. Assuming that $\theta$ is unbounded, there are the following, mutually exclusive, cases:
  \begin{enumerate}[(i)]
  \item The solution is such that
    \begin{equation}\label{eq:phi exc limit intro}
      \lim_{t\rightarrow t_-}\left[3\phi_t(t)/\theta(t)+(\ln V)'[\phi(t)]\right]=0
    \end{equation}
    holds and $\phi(t)\rightarrow\infty$ as $t\rightarrow t_-$. Up to time translation, there is exactly one such solution, and its image is a smooth
    submanifold of $B^{\iso}_{\mrI,+}$.
  \item The solution is such that (\ref{eq:phi exc limit intro}) holds and $\phi(t)\rightarrow-\infty$ as $t\rightarrow t_-$. Up to
    time translation, there is exactly one such solution, and its image is a smooth submanifold of $B^{\iso}_{\mrI,+}$.
  \item The solution has a crushing singularity and induces data on the singularity. 
  \end{enumerate}
  Moreover, assuming, in addition, $V\in\mfP_{\ropar}\cap\mfP_{\a_V}^\infty$ for some $\a_V\in (0,1)$
  and removing the two unique solutions mentioned in (i) and (ii) from the set of isometry classes ${}^{\rosc}\mfD_{\mrI,\roc}^{\iso}[V]$ yields
  ${}^{\rosc}_{\roqu}\mfD_{\mrI,\roc}^{\iso}[V]$ (which is diffeomorphic to ${}^{\rosc}\mfS_\mrI^{\iso}$ via the map $\mfR_{\mrI,\iso}$ introduced in
  Proposition~\ref{prop:dev to dos iso BI}).
\end{thm}
\begin{remark}
  Results related to Proposition~\ref{prop:exp pot asympt} and Theorem~\ref{thm:asympt as exp pot intro}, using somewhat different assumptions,
  are to be found in \cite{foster}.
\end{remark}
\begin{remark}
  There are also solutions that do not have crushing singularities. However, under suitable assumptions on the potential, such solutions are
  non-generic; cf. Proposition~\ref{prop:iso theta bounded} and Remark~\ref{remark:non crushing non generic}. 
\end{remark}
\begin{remark}\label{remark: only one direction V inf intro}
  If we replace the assumptions that $V(s)>0$ and that (\ref{eq:ln V biss est intro}) hold for $|s|\geq M$ with the assumption that these estimates
  hold only for $s\geq M$ and the assumption that $V$ is bounded for $s\leq 0$; and if we remove the assumption that $(\ln V)'(s)\rightarrow\lambda_-$ as
  $s\rightarrow -\infty$, then either \textit{(i)} or \textit{(iii)} have to hold.
\end{remark}
\begin{remark}\label{remark: only one direction V inf intro v2}
  If $\bar{V}(s)=V(-s)$ satisfies the conditions of Remark~\ref{remark: only one direction V inf intro}, then either \textit{(ii)} or \textit{(iii)} have
  to hold.
\end{remark}
\begin{remark}\label{remark:one exceptional soln}
  If $V\in\mfP_{\ropar}\cap\mfP_{\a_V}^\infty$ for some $\a_V\in (0,1)$ and the assumptions of Remark~\ref{remark: only one direction V inf intro} hold, then, removing
  the unique solution as in \textit{(i)} from the set of isometry classes ${}^{\rosc}\mfD_{\mrI,\roc}^{\iso}[V]$ yields ${}^{\rosc}_{\roqu}\mfD_{\mrI,\roc}^{\iso}[V]$
  (which is diffeomorphic to ${}^{\rosc}\mfS_\mrI^{\iso}$ via the map $\mfR_{\mrI,\iso}$ introduced in Proposition~\ref{prop:dev to dos iso BI}). There is an
  analogous statement in case the assumptions of Remark~\ref{remark: only one direction V inf intro} are replaced by the assumptions of
  Remark~\ref{remark: only one direction V inf intro v2}.
\end{remark}
\begin{proof}
  The statement follows from Proposition~\ref{prop:dev to dos iso BI}, Theorem~\ref{thm:asympt as exp pot}, 
  Lemma~\ref{lemma: Bianchi I isotropic}, Remark~\ref{remark:data induced on sing BI iso} and the following argument. The unique solutions
  satisfying (\ref{eq:phi exc limit intro}) are such that $\phi_t/\theta$ converges to $-\lambda_\pm/3$, which is incompatible with
  Definition~\ref{def:ndvacidonbbssh} due to the assumption that $0<\lambda_\pm^2<6$; cf. the proof of Proposition~\ref{prop:exp pot asympt}.
  The statements of Remarks~\ref{remark: only one direction V inf intro}, \ref{remark: only one direction V inf intro v2} and
  \ref{remark:one exceptional soln} follow by, in addition, appealing to Remark~\ref{remark: only one direction V inf}. 
\end{proof}

\section[Asymptotics, $k=-1$ FLRW]{Asymptotics in the direction of the singularity, the isotropic negative spatial
  curvature setting}\label{section:as sing k minus one intro}

Finally, we state results in the direction of the singularity in the isotropic negative spatial curvature setting. We begin by defining the
relevant class of regular initial data.

\begin{definition}\label{def:id k minus one}
  \textit{Locally homogeneous and isotropic negative curvature initial data for the Einstein non-linear scalar field equations},
  \index{Initial data!$k=-1$ FLRW}%
  \index{Initial data!Isotropic}%
  \index{Initial data!Regular!$k=-1$ FLRW}%
  \index{Initial data!Regular!Isotropic}%
  \index{Spatially locally homogeneous and isotropic!Negative curvature!Initial data}%
  with potential $V\in C^{\infty}(\rn{})$, consist of the following: a complete hyperbolic $3$-manifold $(\bM,\bge)$; a covariant
  $2$-tensor field $\bk$ on $\bM$ which is a non-negative constant multiple of $\bge$; and two constants $\bphi_{0}$ and $\bphi_{1}$ satisfying:
  \begin{equation}\label{eq:ham con id k minus one}
    \bS-|\bk|_{\bge}^{2}+(\rotr_{\bge}\bk)^{2} = \bphi_{1}^{2}+2V(\bphi_{0}).
  \end{equation}
  The data are said to be \textit{trivial}
  \index{Initial data!Trivial}%
  \index{Trivial!Initial data}%
  if $\bphi_1=0$ and $V'(\bphi_0)=0$. Let $\mN[V]$
  \index{$\a$Aa@Notation!Sets of regular initial data!mNV@$\mN[V]$}%
  denote the set of all locally homogeneous and isotropic negative curvature initial data for the Einstein non-linear scalar field
  equations with potential $V$. 
\end{definition}
\begin{remark}\label{remark:Milne and Milne Lambda solns}
  In case the initial data are trivial and $V\geq 0$, the corresponding maximal globally hyperbolic development is given by $(M,g,\phi)$, where
  $M:=\bM\times (0,\infty)$,
  \begin{equation}\label{eq:metric Lambda k minus one}
    g:=-dt\otimes dt+a_{\Lambda}^2(t)\bge_-
  \end{equation}
  and $\phi\equiv\bphi_0$. Here $\bge_-$ is the constant multiple of $\bge$ with scalar curvature $-6$, $\Lambda:=V(\bphi_0)$ and
  \begin{equation}\label{eq:a of t expl formula}
    a_\Lambda(t):=\left\{\begin{array}{cl} H^{-1}\sinh(Ht) & \Lambda>0,\\ t & \Lambda=0,\end{array}\right.
  \end{equation}
  where $H:=(\Lambda/3)^{1/2}$. In other words, if $\Lambda=0$, the maximal globally hyperbolic development is the \textit{Milne model}.
  \index{Milne model}%
  In case the initial data are trivial with $V(\bphi_0)>0$, we obtain a vacuum solution to Einstein's
  equations with a cosmological constant $\Lambda$. These statements are justified below the proof of Lemma~\ref{lemma:ex un k minus one}. 
\end{remark}
It is of interest to know if asymptotics of this form appear more generally in the direction of the singularity. This turns out to be the case.
\begin{prop}\label{prop:Milne lambda as k minus one}
  Let $0\leq V\in C^{\infty}(\ro)$. Given $\phi_\infty\in\ro$, there are unique smooth functions $a:(0,\infty)\rightarrow (0,\infty)$ and $\phi:(0,\infty)\rightarrow\ro$
  such that if $(\bM,\bge_-)$ is a complete hyperbolic manifold with scalar curvature $-6$, $M:=\bM\times (0,\infty)$ and $g$ is defined by
  \begin{equation}\label{eq:gSH bge minus}
    g=-dt\otimes dt+a^2(t)\bge_-,
  \end{equation}
  then $(M,g,\phi)$ is a solution to the Einstein non-linear scalar field equations with the property that $\phi(t)\rightarrow\phi_\infty$ and
  $\theta(t)\rightarrow\infty$ as $t\downarrow 0$, where $\theta$ denotes the mean curvature.

  In addition, the unique $(a,\phi)$ is such that
  \begin{equation}\label{eq:phi asymptotics vacuum hyp t ver spec}
    |\phi_t(t)/\theta(t)|+|\phi(t)-\phi_\infty|\leq Ct^2
  \end{equation}
  for all $t\leq 1$ and some constant $C$. Moreover, the metric asymptotes to a solution to the Einstein
  vacuum equations with a cosmological constant $\Lambda:=V(\phi_\infty)$ in the direction of the singularity, see Remark~\ref{remark:Milne and Milne Lambda solns},
  in the sense that there is a constant $C$ such that
  \begin{align}
    |a(t)-a_\Lambda(t)|\leq & Ct^5,\label{eq:a as vacuum setting spec}\\
    \left|\theta(t)-3\tfrac{\dot{a}_\Lambda(t)}{a_\Lambda(t)}\right|\leq & Ct^3\label{eq:theta as vacuum setting k minus one spec}
  \end{align}
  for all $t\leq 1$. 
\end{prop}
\begin{proof}
  The proof is to be found in Subsection~\ref{ssection:spec as Milne k minus one}. 
\end{proof}
It is of interest to contrast these solutions with the ones inducing data on the singularity. To this end, it is convenient to record the
relevant notion of data on the singularity in the current setting.
\begin{definition}\label{def:ndvacidonbbssh k minus one}
  Let $(\bM,\msH)$ be a complete $3$-dimensional hyperbolic manifold, $\msK$ be the $(1,1)$-tensor field $\msK=\roId/3$ on $\bM$ and
  $(\Phi_{0},\Phi_{1})\in\rn{2}$. Then $(\bM,\msH,\msK,\Phi_{0},\Phi_{1})$ are \textit{locally homogeneous and isotropic negative curvature initial
    data on the singularity for the Einstein non-linear scalar field equations}
  \index{Initial data!$k=-1$ FLRW}%
  \index{Initial data!Isotropic}%
  \index{Initial data!On singularity!$k=-1$ FLRW}%
  \index{Initial data!On singularity!Isotropic}%
  \index{Spatially locally homogeneous and isotropic!Negative curvature initial data}%
  if $\Phi_1^2=2/3$. The set of such data is denoted $\mS_-^{\iso}$.
  \index{$\a$Aa@Notation!Sets of initial data on the singularity!mS-iso@$\mS_-^{\iso}$}%
\end{definition}
Next, we define the notion of a development.
\begin{definition}\label{def:k eq minus one dev}
  Let $V\in C^{\infty}(\ro)$ and $\mfI=(\bM,\bge,\bk,\bphi_0,\bphi_1)\in\mN[V]$. A
  \textit{spatially locally homogeneous and isotropic non-linear scalar field development of}
  \index{Development!$k=-1$ FLRW}%
  \index{Development!Isotropic}%
  \index{Spatially locally homogeneous and isotropic!Negative curvature!Development}%
  $\mfI$ is a triple $(M,g,\phi)$, where $M:=\bM\times J$, $J$ is an open interval, $\phi\in C^{\infty}(J)$,
  \begin{equation}\label{eq:gSH k minus one}
    g=-dt\otimes dt+a^2(t)\bge,
  \end{equation}  
  and $a:J\rightarrow (0,\infty)$ is smooth. Finally, $(M,g,\phi)$ is required to satisfy the Einstein non-linear
  scalar field equations, and for some $t_{0}\in J$, the initial data induced on $\bM_{t_0}:=\bM\times \{t_{0}\}$ by $(M,g,\phi)$ and pulled back by
  $\iota_0$ are required to equal $\mfI$, where $\iota_0(p):=(p,t_0)$.

  In case the left endpoint of $J$ is $t_-$ ($t_-$ will always be assumed to equal $0$ if $t_->-\infty$) and $\theta(t)\rightarrow \infty$ as
  $t\downarrow t_-$ (where $\theta(t)$ denotes the mean curvature of $\bM_t$ with respect to $g$), then the development is said to have a
  \textit{crushing singularity}.
  \index{Crushing singularity}%
\end{definition}
\begin{remark}
  As noted in Section~\ref{section: k negative case}, if $V$ is non-negative and $\mfI\in\mN[V]$, there is a unique spatially locally homogeneous and
  isotropic non-linear scalar field development of $\mfI$ with a crushing singularity and a $J$ which can be assumed to equal $(0,\infty)$.
  We denote this development by $\mD[V](\mfI)$. 
\end{remark}
That developments induce data on the singularity means the following.
\begin{definition}\label{def:ind data on sing k negative}
  Let $V\in C^{\infty}(\ro)$, $\mfI\in\mN[V]$ and $(M,g,\phi)$ be a spatially locally homogeneous and isotropic non-linear scalar field development of
  $\mfI$. Using the terminology of Definition~\ref{def:k eq minus one dev}, assume the development to have a crushing singularity,
  let $K$ be the Weingarten maps induced on $\bM_t$, $t\in J$, and define the \textit{expansion normalised Weingarten map}
  \index{Weingarten map!Expansion normalised}%
  \index{Expansion normalised!Weingarten map}%
  by $\mK:=K/\theta$. If
  there are $\mfI_\infty=(\bM,\msH,\msK,\Phi_{0},\Phi_{1})\in \mS_-^{\iso}$ such that $\mK\rightarrow\msK$, $\theta^{-1}\phi_t\rightarrow\Phi_{1}$ and
  $\phi+\theta^{-1}\phi_t\ln\theta\rightarrow\Phi_{0}$ as $t\downarrow t_-$; and such that the following limit holds for all $v,w\in T\bM$:
  \begin{equation}\label{eq:bAmetriclimit k negative}
    \lim_{t\downarrow t_-}\bge(\theta^{\mK}v,\theta^{\mK}w)=\msH(v,w),
  \end{equation}
  then the development $(M,g,\phi)$ is said to \textit{induce $\mfI_\infty$ on the singularity}.
  \index{Development!Inducing data on singularity}%
  \index{Inducing data on singularity!Development}%
  \index{Initial data!On singularity!From development}%
\end{definition}
\begin{remark}\label{remark:Milne dos}
  The solutions obtained in Proposition~\ref{prop:Milne lambda as k minus one} do not induce data on the singularity in the sense of
  Definition~\ref{def:ind data on sing k negative}, since $\theta^{-1}\phi_t$ converges to zero for these solutions. In addition, and more importantly, the
  expansion normalised spatial scalar curvature of these solutions, $\bS/\theta^2$, converges to $-2/3$. This conclusion should be contrasted
  with the BKL proposal, in which spatial derivative terms are expected to become negligible in the limit. More specifically, the expansion
  normalised spatial scalar curuvature is expected to tend to zero (except for the bounces in the oscillatory setting). Another important feature of
  singularities in the BKL proposal is that
  they are spacelike, unlike the singularities in the solutions mentioned in Remark~\ref{remark:Milne and Milne Lambda solns}. Finally,
  rescaling the metrics induced on the leaves of the foliation by $\bge(\theta^{\mK}v,\theta^{\mK}w)$ does not result in a finite limit.
  On the other hand, the solutions obtained in Proposition~\ref{prop:Milne lambda as k minus one} do have the following properties: $\phi_t/\theta$
  converges (to zero); $\phi+\theta^{-1}\phi_t\ln\theta$ converges to a limit ($\phi_\infty$); $\mK$ converges to a limit ($\roId/3$);
  and $\theta^2\bge$ converges to a limit ($9\bge_-$). In this sense, one can think of $(\bM,9\bge_-,\roId/3,\phi_\infty,0)$ as data on the singularity.
  Moreover, Proposition~\ref{prop:Milne lambda as k minus one} can be thought of as a result stating that there is a unique corresponding
  development. 
\end{remark}
Given data on the singularity, it is of interest to know that there are corresponding developments.
\begin{prop}\label{prop:dos ind dev k minus one}
  Let $V\in \mfP_{\a_V}^2$ for some $\a_V\in (0,1)$ and $\mfI_\infty\in \mS_-^{\iso}$; see Definition~\ref{def:ndvacidonbbssh k minus one}. Then there is a unique
  (up to time translation) development in the sense of Definition~\ref{def:k eq minus one dev} which induces the data $\mfI_\infty$ on the singularity
  in the sense of Definition~\ref{def:ind data on sing k negative}.
\end{prop}
\begin{remark}\label{remark:as dos k minus one}
  The existence interval of the unique development (cf. Lemma~\ref{lemma:ex un k minus one}), say $J$, can be assumed to take the form $J=(0,t_+)$.
  Moreover, there are constants $a_\infty\in (0,\infty)$ and $C$ and a $t_0\in J$ such that 
  \begin{align*}
    |\theta^{1/3}a-a_\infty|+|t\theta-1| \leq & Ct^{\g_V},\\
    |\theta^{-1}\phi_t-\bPhi_1| \leq & Ct^{\g_V},\\
    |\phi+\theta^{-1}\phi_t\ln\theta-\bPhi_0| \leq & C\ldr{\ln t}t^{\g_V}
  \end{align*}
  for $t\leq t_0$, where $\g_V:=\min\{4/3,2(1-\a_V)\}$ and $(\bPhi_1,\bPhi_0)$ are the data on the singularity for the scalar field. 
\end{remark}
\begin{proof}
  The proof of the proposition and the remark is to be found in Subsection~\ref{ssection:idos k minus one}. 
\end{proof}
At this stage, we know that the outcomes represented by Propositions~\ref{prop:Milne lambda as k minus one} and \ref{prop:dos ind dev k minus one} are
possible. However, it is also of interest to know if other asymptotics can occur. If we make slightly stronger assumptions concerning the potential,
we can conclude that these are the only possible outcomes.

\begin{prop}\label{prop:two outcomes intro k minus one}
  Let $0\leq V\in C^{\infty}(\ro)$ and $\mfI=(\bM,\bge,\bk,\bphi_0,\bphi_1)\in\mN[V]$. Then there are unique $a\in C^\infty(J,(0,\infty))$
  and $\phi\in C^\infty(J)$, where $J=(0,\infty)$, and a unique $t_0\in J$ such that if $g$ is given by
  (\ref{eq:hyperbolic version of metric}), where $\bge_-$ is the constant multiple of $\bge$ with scalar curvature $-6$, and
  $M:=\bM\times J$, then $(M,g,\phi)$ is a solution to the Einstein non-linear scalar field equations such that the initial data induced
  on $\bM_{t_0}:=\bM\times\{t_0\}$ by $(M,g,\phi)$ and pulled back by $\iota_{t_0}$ equal $\mfI$, where $\iota_{t_0}(p):=(p,t_0)$.
  Moreover, this solution has a crushing singularity at $t=0$.

  Assuming, in addition to the above, that $V\in\mfP_{\a_V}^1$ for some $\a_V\in (0,1/3)$, the development either induces data on the singularity
  in the sense of Definition~\ref{def:ind data on sing k negative} or it has asymptotics of the type obtained in
  Proposition~\ref{prop:Milne lambda as k minus one}.
\end{prop}
\begin{proof}
  The statement is an immediate consequence of Lemma~\ref{lemma:ex un k minus one} and Proposition~\ref{prop:two outcomes as k minus one}. 
\end{proof}
Fixing a complete hyperbolic $3$-manifold $(\bM,\bge_-)$, where the scalar curvature of $\bge_-$ is $-6$, the solutions with asymptotics as in
Proposition~\ref{prop:Milne lambda as k minus one} are parametrised by one parameter, the asymptotic value of the scalar field. However, the solutions
inducing initial data on the singularity are, in addition to $\Phi_1\in\{\sqrt{2/3},-\sqrt{2/3}\}$, parametrised by two parameters: $\Phi_0$ and the
constant multiplying $\bge_-$ in order to yield $\msH$. From this point of view, one would naively expect generic solutions to induce data on the
singularity. In order to prove a corresponding statement, it is convenient to first parametrise the set of initial data introduced in
Definition~\ref{def:id k minus one}.
Fixing $(\bM,\bge_-)$ as above, if $(\bM,\bge,\bk,\bphi_0,\bphi_1)\in\mN[V]$ are such that $\bge=\a\bge_-$ and $\bk=\a\b\bge_-$, the constraint
(\ref{eq:ham con id k minus one}) reads
\begin{equation}\label{eq:ham con al be k minus one}
  6\b^2=\bphi_1^2+\tfrac{6}{\a^2}+2V(\bphi_0).
\end{equation}
From this point of view, it is natural to introduce
\begin{equation}\label{eq:sfN minus V def}
  \sfN_-[V]:=\{(\b,\bphi_0,\bphi_1)\in [0,\infty)\times\rn{2}\ |\ 6\b^2-\bphi_1^2-2V(\bphi_0)>0\}.
\end{equation}
\index{$\a$Aa@Notation!Symmetry reduced sets of regular initial data!sfN-V@$\sfN_-[V]$}%
Given $(\b,\bphi_0,\bphi_1)\in \sfN_-[V]$, $\a>0$ is uniquely determined by (\ref{eq:ham con al be k minus one}). In this sense, $\sfN_-[V]$ parametrises
initial data, given $(\bM,\bge_-)$. With this terminology, the following holds.
\begin{prop}\label{prop:generic k minus one intro}
  Let $0\leq V\in C^\infty(\ro)$ and $(\bM,\bge_-)$ be a complete hyperbolic $3$-manifold with scalar curvature $-6$. Then the subset of $\sfN_-[V]$
  that gives rise to developments with the property that $\phi(t)$ converges to a finite number is contained in a countable union of codimension one submanifolds.
  In this sense, the outcome represented by Proposition~\ref{prop:Milne lambda as k minus one} corresponds to a set of initial data which is both Baire and
  Lebesgue non-generic. 
\end{prop}
\begin{proof}
  The proof is to be found in Subsection~\ref{subsection:generic behaviour k minus one}.
\end{proof}

\section{Future asymptotics}\label{ssection:futureasymptotics}

Next, we turn to the future asymptotics. Here we, roughly speaking, consider two situations: either $V$ is proportional to an exponential, or $V$
has a positive lower bound. We begin by considering the latter case. 

\begin{prop}\label{prop:futureassteponeintro}
  Assume $V\in C^{\infty}(\rn{})$ to have a strictly positive lower bound; $V(s)\rightarrow\infty$ as $|s|\rightarrow\infty$; and $V'$ to only
  have simple zeroes. Fix a Bianchi class A type $\mfT\neq\mrIX$ and an $\mfs\in\{\iso,\roLRS,\roper,\rogen\}$. Then there is a subset
  ${}^{\rosc}\mfG_\mfT^\mfs[V]$ of ${}^{\rosc}\mfB_\mfT^\mfs[V]\cap \{\theta>0\}$ which is open, dense and has full measure. In addition, if
  $\xi\in {}^{\rosc}\mfG_\mfT^\mfs[V]$; $\mfI\in \mB_\mfT^\mfs[V]$ is such that $[\mfI]=\xi$; and $(M,g,\phi)=\md[V](\mfI)$, then the following holds.
  First, $M$ and $g$ are of the form $M=G\times J$ and (\ref{eq:gSH}) respectively, where $G$ is a Bianchi type $\mfT$ Lie group and $J=(t_-,\infty)$.
  Moreover, there is a left invariant Riemannian metric $\bge_\infty$ on $G$; an $s_0\in\ro$ with the property that $V'(s_0)=0$ and $V''(s_0)>0$; and
  constants $C>0$ and $b>0$ such that
  \begin{subequations}
    \begin{align}
      |\phi(t)-s_{0}|+|\phi_{t}(t)| \leq & Ce^{-bt},\\
      |\bK(t)-H\roId|_{\bge_\infty}+|e^{-Ht}\bge(t)-\bge_\infty|_{\bge_\infty} \leq & Ce^{-bt}\label{eq:K m HId eHt bge inf intro}
    \end{align}
  \end{subequations}  
  for all $t\geq t_{0}$, where $H:=\sqrt{V(s_{0})/3}$, $t_{0}\in J$ and $\bge(t)$ and $\bK(t)$ are the metric and Weingarten map induced on
  $M_t=G\times\{t\}$ respectively. 
\end{prop}
\begin{proof}
  The statement follows from Propositions~\ref{prop:futureasstepone} and \ref{prop:futureas} below. 
\end{proof}
Due to the results of \cite{RinInv}, taking appropriate quotients of the spacetimes in Proposition~\ref{prop:futureassteponeintro} leads to
solutions that are future globally non-linearly stable.
\begin{cor}\label{cor:stabilitylb}
  Fix a solution to the Einstein non-linear scalar field equations as constructed in Proposition~\ref{prop:futureassteponeintro}. Let $\Gamma$
  be a co-compact subgroup of the isometry group of the initial data which acts freely and properly discontinuously and let $\Sigma$ be the
  corresponding compact
  quotient of the Lie group. Taking the corresponding quotient of the solution defines a solution on $\Sigma\times (t_-,\infty)$, here referred to
  as the background solution. Make a choice of Sobolev norms $\|\cdot\|_{H^{l}}$
  on tensor fields on $\Sigma$ and fix a $t_{0}\in (t_-,\infty)$. Let $(\Sigma,\bge,\bk,\bphi_{0},\bphi_{1})$ denote the initial data induced on
  $\Sigma\times \{t_{0}\}$ by the background solution. Then there is an $\e>0$ such that if $(\Sigma,\bah,\bka,\bvarphi_{1},\bvarphi_{0})$ are initial
  data for the Einstein non-linear scalar field equations with the same $V$ as in the case of the background solution and
  \[
  \|\bah-\bge\|_{H^{4}}+\|\bka-\bk\|_{H^{3}}+\|\bphi_{0}-\bvarphi_{0}\|_{H^{4}}+\|\bphi_{1}-\bvarphi_{1}\|_{H^{3}}\leq\e,
  \]
  then the maximal globally hyperbolic development corresponding to $(\Sigma,\bah,\bka,\bvarphi_{1},\bvarphi_{0})$ is future causally geodesically
  complete and there are expansions of the form given in the statement of \cite[Theorem~2, pp.~131--132]{RinInv} to the future.
\end{cor}
\begin{remark}
  Due to \cite{rav}, every unimodular $3$-dimensional Lie group has a co-compact subgroup.
\end{remark}
\begin{remark}
  One particular consequence of the corollary is that the background solution is future causally geodesically complete (note that this is not
  claimed in Proposition~\ref{prop:futureassteponeintro}). 
\end{remark}
\begin{remark}\label{remark:global nonlinear stability}
  Combining the corollary with the conclusions concerning the past asymptotics derived previously and the results of \cite{GPR} yields past
  and future global non-linear stability, curvature blow up in the direction of the singularity etc. for large classes of spatially locally
  homogeneous solutions. We refer the interested reader to \cite{GPR} for more detailed information. 
\end{remark}
\begin{proof}
  The proof of the statement is similar to, but slightly simpler than, the proof of \cite[Theorem~4, pp.~134-135]{RinInv}. In fact, it is sufficient
  to take the steps starting with the paragraph entitled stability at the bottom of \cite[p.~203]{RinInv} until the end of the proof. We leave
  the detailed verification to the reader. 
\end{proof}
The second case to consider is that of an exponential potential:
\begin{equation}\label{eq:Vexppot}
  V(\phi)=V_{0}e^{-\lambda\phi},
\end{equation}
where $0<V_{0}\in\rn{}$ and $\lambda\in (0,\sqrt{2})$. The reason we require $\lambda\in (0,\sqrt{2})$ is that we want the corresponding solutions
to the Einstein non-linear scalar field equations to exhibit accelerated expansion. In fact, under these assumptions, the class of spatially
homogeneous solutions of interest in this paper do not only exhibit accelerated expansion, they also exhibit isotropisation. The spatially
homogeneous setting is considered, e.g., in \cite{hall,KaM}; see also \cite[Section~12]{RinPL}. The future global non-linear
stability of suitable quotients of Bianchi class A solutions is considered in \cite{HaR,RinPL}. In fact, for spatially closed and spatially
locally homogeneous quotients of these spacetimes, there is a stability result analogous to Corollary~\ref{cor:stabilitylb}; see
\cite[Theorem~3, p.~162]{RinPL}. Moreover, we can, in analogy with Remark~\ref{remark:global nonlinear stability}, derive future and past global
non-linear stability results for large classes of spatially locally homogeneous solutions; see \cite{GPR} for details. 

Next, we consider the isotropic negative spatial curvature case. 
\begin{prop}\label{prop:generic k minus one intro exp}
  Assume $V\in C^{\infty}(\rn{})$ to have a strictly positive lower bound; $V(s)\rightarrow\infty$ as $|s|\rightarrow\infty$; and $V'$ to only
  have simple zeroes. Let $(\bM,\bge_-)$ be a complete hyperbolic $3$-manifold with scalar curvature $-6$. Then the subset of $\sfN_-[V]$,
  see (\ref{eq:sfN minus V def}), that gives rise to developments with the property that $\phi(t)\rightarrow s_0\in\ro$ with $V''(s_0)<0$ is
  contained in a countable union of one dimensional submanifolds. Call the complement $\sfG_-^+[V]$. Then, given initial data in $\sfG_-^+[V]$,
  the corresponding solution has the properties that there is an $s_0\in\ro$ such that $\phi(t)\rightarrow s_0\in\ro$; $V'(s_0)=0$; $V''(s_0)>0$;
  $\phi_t(t)\rightarrow 0$; and $\theta(t)\rightarrow [3V(s_0)]^{1/2}$ as $t\rightarrow \infty$. Moreover, the convergence is exponential. 
\end{prop}
\begin{remark}
  Combining this result with the conclusions of Section~\ref{section:as sing k minus one intro}, and the results of \cite{GPR} and \cite{RinInv}
  yields future and past global non-linear stability results for large classes of spatially closed solutions. 
\end{remark}
\begin{proof}
  The proof is to be found in Section~\ref{section:Future as k minus one setting}. 
\end{proof}

\section{Motivation for the restriction to Bianchi class A}\label{ssection:sphom}
In this paper we mainly restrict our attention to Bianchi class A initial data. It is of interest to describe how this restriction fits into the
general context of spatially homogeneous solutions. Defining spatially homogeneous solutions to be the maximal globally hyperbolic
developments of homogeneous initial data, there are, in $3$ dimensions, three possibilities: the initial data are left invariant on a
unimodular Lie group (Bianchi class A); the initial data are left invariant on a non-unimodular Lie group (Bianchi class B); or the initial
data induced on the universal covering space of the initial manifold are invariant under the isometry group of the standard
metric on $\rn{}\times\sn{2}$ (Kantowski-Sachs). Strictly speaking, the Bianchi classification is a classification of $3$-dimensional Lie algebras;
see \cite{ketal} for a more detailed discussion, references and a historical overview. Since we wish to construct solutions for which we can later prove
future and past global non-linear stability, and since our result concerning stable big bang formation is based on the assumption of spatial
compactness, see \cite{GPR}, we are here mainly interested in initial data that admit compact quotients. Compact quotients can arise
in two ways: either there is a co-compact subgroup, say $\Gamma$, of the Lie group $G$; or there is a co-compact subgroup $\Gamma$ of the isometry group
of the initial data which is not a subgroup of $G$. We are mainly interested in the first case, since the second case means that the isometry
group of the initial data typically has to be substantially larger than $G$. This entails additional symmetry assumptions on the initial data (and
therefore more limited classes of models). If $G$ is non-unimodular, there are no co-compact subgroups of $G$; see, e.g., \cite[Section~3.3]{AaS} and
the argument presented on \cite[pp.~865-866]{FIK}. On the other hand, in case $G$ is a unimodular $3$-dimensional Lie group, there are co-compact
subgroups $\Gamma$ of $G$; see \cite{rav}.

\section{Outline of the article}\label{section:outline article}

We begin our analysis in Chapter~\ref{chapter:parametrisations} by discussing how to parametrise isometry classes of regular initial data, initial
data on the singularity and developments by smooth manifolds. In Chapter~\ref{section:constructingadevelopment}, we then turn to the construction
of Bianchi class A developments, given initial data. We begin by deriving the equations in the symmetry reduced setting. We then prove existence of
developments and prove some preliminary results concerning the asymptotics. In particular, we provide conditions ensuring that there is a crushing
singularity and discuss the integrability of the mean curvature. Using this information, we are then in a position to introduce expansion normalised
variables. This is the topic of Section~\ref{ssection:whsuform}, where we introduce the Wainwright-Hsu formulation of the equations. These are the
equations we use throughout most of the article when analysing the asymptotics. In Chapter~\ref{section:asymptoticbehavsing}, we then turn to the
problem of analysing the asymptotics in the direction of the singularity. We begin by characterising the solutions that converge to a special point.
We also make some basic observation concerning the lower Bianchi types. However, the main result of the chapter is Theorem~\ref{thm:dichotomy}
in which we demonstrate that there is a dichotomy between vacuum dominated and matter dominated solutions. We also derive detailed asymptotics
for the matter dominated solutions. We are then left with the problem of analysing the asymptotics in the vacuum dominated setting. This is the
subject of Chapter~\ref{chapter:vacuum dom set}. We prove that the anisotropic and non-LRS vacuum dominated Bianchi type VIII and IX solutions exhibit oscillations
and derive detailed asymptotics in all the other cases. Using this information, we are then in a position to prove that, with the exception of
the solutions with oscillatory singularities, all developments induce data on the singularity. This is the subject of
Chapter~\ref{chapter:as to data on sing}. Next, we turn to the problem of specifying initial data on the singularity. In Chapter~\ref{section:specdatasing}
we address this topic from the point of view of the Wainwright-Hsu formulation. Proving that the map from regular initial data to data on the
singularity has the desired regularity properties turns out to be quite technical. In Chapter~\ref{chapter:map to sing}, we intiate the corresponding
discussion. To begin with, we discuss the regularity of the map from regular initial data to data on the singularity with respect to the
Wainwright-Hsu variables. However, there are many different perspectives. In particular, one can fix the mean curvature of the initial data or
consider general initial data. The former perspective gives rise to additional complications, since the time derivative of the mean curvature
along the Einstein flow need not always be non-zero. After introducing a substantial amount of notation, the first result is contained in
Lemma~\ref{lemma:reg of B}. Essentially as a consequence of this lemma, we obtain regularity of a representation of the map from isometry classes
of initial data with fixed mean curvature to isometry classes of initial data on the singularity; see Lemma~\ref{lemma:local repr map to sing}.
In order to define a smooth structure on the set of isometry classes of developments, we next need to prove that the map, induced by the Einstein
flow, relating isometry classes of initial data with different mean curvature is a local diffeomorphism. This is the topic of
Section~\ref{section:Id fix mean curv}. This section also contains proofs of Lemmas~\ref{lemma:dev smo str nisoI nIX}, \ref{lemma:dev smo str IX}
and \ref{lemma:dev smo str isoI}; i.e., the results endowing isometry classes of developments with a smooth structure. In
Chapter~\ref{chapter:geo res dir of sing}, we prove geometric existence and uniqueness results of developments corresponding to data on the
singularity. This chapter also contains a proof of Propositions~\ref{prop:dev to data on sing nIX}, \ref{prop:dev to dos BIX} and
\ref{prop:dev to dos iso BI}, proving, roughly speaking, that the Einstein flow generates a diffeomorphism between isometry classes of developments
that induce data on the singularity and isometry classes of initial data on the singularity.

In Chapter~\ref{chapter:k eq zero and minus one cases}, we turn to the $k=0$ and $k=-1$ FLRW solutions. The analysis in the isotropic Bianchi type I
setting turns out to be surprisingly technical. However, there is a simple explanation for this: not all solutions induce data on the singularity.
In other words, it is necessary to deal with exceptional cases. On the other hand, under rather strong assumptions on the potential, it is possible
to verify that there are exactly as many exceptional solutions as one would expect from a crude inspection of the topologies of the sets involved;
see the discussion in Section~\ref{section:k eq zero case} for more details. The analysis in the $k=-1$ FLRW setting is more straightforward, even
though there are exceptional cases in this situation as well. In Chapter~\ref{section:futureasymptotics}, we finally turn to future asymptotics.
In the appendix, we address the problem of parametrising solutions, regular data and data on the singularity. 

\section{Acknowledgements}\label{section:acknoledgements}
The author would like to thank Claes Uggla for comments on an earlier version of the article and Artur Alho for suggesting additional references. 
This research was funded by Vetenskapsr\aa det (the Swedish Research Council), dnr. 2017-03863 and 2022-03053.

\chapter[Parametrisations]{Parametrising solutions, data on the singularity and regular data}\label{chapter:parametrisations}

In this article, we are mainly interested in solutions with asymptotics as described in Theorem~\ref{thm:dataonsingtosolution}. It is therefore
of interest to identify the set of regular initial data leading to such asymptotics. Moreover, we wish to determine the regularity of
the map taking regular initial data to data on the singularity. In order for these questions to make sense, we need to parametrise the set of
isometry classes of regular initial data, as well as the set of isometry classes of data on the singularity, by smooth manifolds. 

\section{Parametrising regular initial data}\label{ssection:par reg in data}

In order to explain how to parametrise regular initial data, let $V\in C^{\infty}(\ro)$ and $(G,\bge,\bk,\bphi_{0},\bphi_{1})\in\mB[V]$. Let, moreover,
$\{e_{i}\}$ be a basis of $\mfg$ which is orthonormal with respect to $\bge$ and let $\bn$ and $K$ be the commutator and Weingarten matrices associated
with $\{e_{i}\}$; see Definition~\ref{def:comm and Wein matr}. Recall that $\bn$ and $K$ satisfy (\ref{seq:compverconstraints}). 
Due to (\ref{eq:macom}), there is a basis $\{e_{i}'\}$ for the Lie algebra which is orthonormal with respect to $\bge$ and such that 
$K$ and $\bn$ are diagonal; see \cite[Corollary~19.14, p.~211]{RinCauchy}. When convenient, we assume the diagonal elements of $\bn$
to have signs as indicated in Table~\ref{table:bianchiA}; this can be ensured by an argument similar to the proof of
\cite[Lemma~19.8, p.~208]{RinCauchy}. Given initial data $(G,\bge,\bk,\bphi_{0},\bphi_{1})$, we thus obtain
$(\bn,K,\bphi_{0},\bphi_{1})$, where $\bn=\rodiag(\bn_{1},\bn_{2},\bn_{3})$ and $K=\rodiag(k_{1},k_{2},k_{3})$. The resulting
$(\bn,K,\bphi_{0},\bphi_{1})$ in general depends on the choice of frame. For instance, in the case of Bianchi type I, $\bn=0$ and permuting the
diagonal elements of $K$ corresponds to permuting the labels of the elements of the frame. Going in the opposite direction, fix
$(\bn,K,\bphi_{0},\bphi_{1})$ satisfying (\ref{eq:nuvhc}). Assume, moreover, that $\bn$ and $K$ are diagonal and that the diagonal elements
of $\bn$ have signs as indicated in Table~\ref{table:bianchiA} corresponding to a $3$-dimensional unimodular Lie group $G$. Then there is a 
frame $\{e_i\}$ of $\mfg$ such that the associated commutator matrix is $\bn$. Given this frame, we define a metric $\bge$ by demanding that
$\{e_i\}$ be an orthonormal frame, and we define $\bk$ by demanding that the Weingarten matrix associated with $\{e_{i}\}$ be $K$. Then
$(G,\bge,\bk,\bphi_{0},\bphi_{1})$ are initial data such that the corresponding commutator and Weingarten matrices with respect to $\{e_{i}\}$
are $\bn$ and $K$ respectively. See Lemma~\ref{lemma:realcond} below for a justification of the above statements. Unfortunately, the resulting
initial data depend on the choice of frame. On the other hand, if $G$ is simply connected, two initial data sets obtained in this way are isometric;
see Lemma~\ref{lemma:nuKsameisometric} below. Next, note that the sum of the first two terms on the left hand side of (\ref{eq:nuvhc}) is the scalar
curvature of $\bge$, say $\bS$; see \cite[(19.6), p.~209]{RinCauchy}. Comparing this polynomial with Table~\ref{table:bianchiA}, it is clear that
$\bS\leq 0$ unless $G$ is of Bianchi type IX. Moreover, $\bS$ can only equal zero if the initial data are of Bianchi type I; isotropic or LRS Bianchi
type VII${}_0$; or Bianchi type IX. However, since we have excluded isotropic and LRS Bianchi type VII${}_0$ initial data, see
Remark~\ref{remark:LRS VIIz data are BI}, $\bS=0$ only occurs in the case of Bianchi types I and IX.

\textit{Maximal hypersurfaces.} In this article, we are not interested in trivial initial data; cf. Definitions~\ref{def:Bianchiid} and
\ref{def:BV} and Remark~\ref{remark:dev of trivial id}. In the case of a non-negative potential $V$ and initial data which are not of Bianchi type
IX, this means that we exclude maximal initial data, i.e. initial data with $\rotr_{\bge}\bk=0$. The reason for this is the following. Assume that
$G$ is not of Bianchi type IX, that $\rotr K=0$ and that $V$ is
non-negative. Then it follows from (\ref{eq:nuvhc}) that $\bS=0$, $\rotr K^2=0$, $\bphi_1^2=0$ and $V(\bphi_0)=0$. This means
that $K=0$, $\bphi_1=0$, $V'(\bphi_0)=0$ (since $V\geq 0$ and $V(\bphi_0)=0$) and that $G$ is of Bianchi type I. In other words, the data are trivial.
Next, we define the sets which will form the basis for the parametrisation of isometry classes of regular initial data. 
\begin{definition}\label{def:degenerate version of id}
  Fix $V\in C^{\infty}(\ro)$. Let $\sfB$
  \index{$\a$Aa@Notation!Symmetry reduced sets of regular initial data!sfB@$\sfB$}%
  be the set of $(\bn,K,\bphi_{0},\bphi_{1})\in M_{3}(\rn{})\times M_{3}(\rn{})\times\rn{2}$ such that (\ref{eq:nuvhc}) holds;
  $\bn=\rodiag(\bn_{1},\bn_{2},\bn_{3})$; $K=\rodiag(k_{1},k_{2},k_{3})$; and, if $\bn=0$ and $K$ is a multiple of the identity, then either $\bphi_1\neq 0$
  or $V'(\bphi_0)\neq 0$.
  Next, let $\sfB_{\mrI}$ denote the subset of $\sfB$ such that $\bn=0$; let $\sfB_{\mrII}$ denote the subset of $\sfB$ such that one of
  the $\bn_i$ is non-zero and the other two vanish; let $\sfB_{\mrVIz}$ denote the subset of $\sfB$ such that two of the $\bn_i$ are non-zero and have
  opposite signs, and the remaining $\bn_i$ vanishes;
  \index{$\a$Aa@Notation!Symmetry reduced sets of regular initial data!sfBT@$\sfB_{\mfT}$}%
  let $\sfB_{\mrVIIz}$ denote the subset of $\sfB$ such that two of the $\bn_i$ are non-zero and have
  the same signs, and the remaining $\bn_i$ vanishes; let $\sfB_{\mrVIII}$ denote the subset of $\sfB$ such that all the $\bn_i$ are non-zero, but they
  do not all have the same sign; and let $\sfB_{\mrIX}$ denote the subset of $\sfB$ such that all the $\bn_i$ are non-zero, and all have the same sign.
\end{definition}
\begin{remark}
  Here $M_{3}(\rn{})$ denotes the set of real valued $3\times 3$-matrices.
\end{remark}
\begin{remark}\label{remark:mfBmfT subm}
  Consider (\ref{eq:nuvhc}). Restricting to the case that $K$ and $\bn$ are diagonal, we can view the set defined by this equation as a subset of
  $\rn{8}$. Calculating the gradient (in $\rn{8}$) of the function obtained by subtracting the expression on the right hand side of (\ref{eq:nuvhc})
  from the left hand side, it can be deduced that if the gradient vanishes, then $\bn=0$, $K=0$, $\bphi_1=0$ and $V'(\bphi_0)=0$. Such points do not
  belong to $\sfB$. It follows that all the sets $\sfB_\mfT$ are manifolds. 
\end{remark}
It is convenient to keep the following observation in mind.
\begin{lemma}\label{lemma:x arising from mfI}
  Let $V\in C^\infty(\ro)$ and $\mfI=(G,\bge,\bk,\bphi_0,\bphi_1)\in\mB_\mfT[V]$, where $\mfT$ is a Bianchi class A type. Then there is an
  orthonormal basis $\{e_i\}$ of $\mfg$ with respect to $\bge$ such that if $\bn$ and $K$ are the commutator and Weingarten matrices associated with
  $\{e_i\}$ and $\mfI$, then $x:=(\bn,K,\bphi_0,\bphi_1)\in\sfB_\mfT$.
\end{lemma}
\begin{remark}\label{remark:x arising from mfI}
  If $x$ and $\mfI$ are related as in the statement of the lemma, we below say that $x$ arises from $\mfI$.
\end{remark}
\begin{proof}
  That there is an orthonormal basis $\{e_i\}$ of the desired form follows from the observations made at the beginning of the present section,
  keeping in mind that trivial initial data are excluded from $\mB[V]$ and characterised by $\bn=0$ (since we have excluded isotropic and LRS Bianchi
  type $\mrVIIz$ data), $K=\theta\roId/3$ (for some $\theta\in\ro$), $\bphi_1=0$ and $V'(\bphi_0)=0$. 
\end{proof}
Note that permuting the basis elements in $\{e_i\}$ (cf. Lemma~\ref{lemma:x arising from mfI}) transforms $x\in\sfB_\mfT$ to another element of
$\sfB_\mfT$. Similarly, changing $e_i$ to $-e_i$ for one $i\in\{1,2,3\}$ transforms $x\in\sfB_\mfT$ to another element of $\sfB_\mfT$. These operations
can be summarised by the maps
\[
\psi_\sigma^+:\sfB\rightarrow\sfB,\ \ \
\psi_\sigma^-:\sfB-\sfB_\mrI\rightarrow\sfB-\sfB_\mrI
\]
for $\sigma\in S_3$, defined as follows:
\begin{equation}\label{eq:psi sigma pm}
    \psi_{\sigma}^{\pm}(\bn_1,\bn_2,\bn_3,k_1,k_2,k_3,\bphi_0,\bphi_1) := 
    (\pm\bn_{\sigma(1)},\pm\bn_{\sigma(2)},\pm\bn_{\sigma(3)},k_{\sigma(1)},k_{\sigma(2)},k_{\sigma(3)},\bphi_0,\bphi_1);
\end{equation}
\index{$\a$Aa@Notation!Maps!psisigmapm@$\psi_{\sigma}^{\pm}$}%
recall that $S_3$ denotes the set of bijections from the set $\{1,2,3\}$ to itself. Here we tacitly think of 
$\sfB$ as a subset of $\rn{8}$. The map $\psi_{\sigma}^+$ ($\psi_{\sigma}^-$) corresponds to permuting the basis vectors according to $\sigma$ and then
flipping the sign of one of the basis vectors in case $\sigma$ is an odd (even) permutation. The reason we exclude $\sfB_\mrI$ from the domain of
definition of $\psi_\sigma^-$ is that if we do not, then $\psi_\sigma^-(x)=\psi_\sigma^+(x)$ for all $x\in\sfB_\mrI$ and $\sigma\in S_3$; i.e., $\psi_\sigma^-$
would then be redundant in the Bianchi type I setting, which causes unnecessary complications. 

Let $\Gamma$
\index{$\a$Aa@Notation!Symmetry groups!Gamma@$\Gamma$}%
denote the group consisting of the elements
$\psi_\sigma^\pm$. It is natural to quotient $\sfB_\mfT$ by this group. However, the result need not be a manifold. This is due to the fact that
non-trivial maps of the form $\psi_\sigma^\pm$ can have fixed points; in order to obtain a manifold, we would like the group action by which we quotient
to be free. In order to
identify the fixed points of non-trivial maps of the form $\psi_\sigma^\pm$, note that there are three types of elements of $S_3$; $\sigma\in S_3$
can have three fixed points ($\sigma=\mathrm{Id}$), one fixed point (there are $\{i,j,k\}=\{1,2,3\}$ such that $\sigma$ fixes $i$ and swaps
$j$ and $k$; i.e., $\sigma$ is a transposition interchanging $j$ and $k$) or no fixed points (in which case $\sigma$ can be written as a product of
two transpositions).
\begin{definition}\label{def:Gamma plus etc}
  Let $\Gamma^{\roev}$
  \index{$\a$Aa@Notation!Symmetry groups!Gammaev@$\Gamma^{\roev}$}%
  ($\Gamma^{\roodd}$)
  \index{$\a$Aa@Notation!Symmetry groups!Gammaodd@$\Gamma^{\roodd}$}%
  consist of the set of $\psi_\sigma^\pm$ such that $\sigma$ is even (odd). Let $\Gamma^{\pm}$
  \index{$\a$Aa@Notation!Symmetry groups!Gammapm@$\Gamma^{\pm}$}%
  be the subset of
  $\Gamma$ consisting of the maps $\psi_\sigma^\pm$ for $\sigma\in S_3$. Finally, let $\Gamma^{+,\roev}:=\Gamma^+\cap\Gamma^{\roev}$
  \index{$\a$Aa@Notation!Symmetry groups!Gammapmev@$\Gamma^{\pm,\roev}$, $\Gamma^{\pm,\roodd}$}%
  etc. 
\end{definition}
\begin{remark}
  Clearly, $\Gamma^{\roev}$ is a subgroup of $\Gamma$, but $\Gamma^{\roodd}$ is not. Moreover, $\Gamma^+$ is a subgroup of $\Gamma$, but $\Gamma^-$ is not.
\end{remark}
Next, we determine the fixed points of maps of the form $\psi_\sigma^\pm$. The map $\psi_{\roId}^+$ is of course the identity. The map $\psi_{\roId}^-$ does
not have any fixed points (since $\sfB_\mrI$ does not belong to the domain of definition of $\psi_{\roId}^-$). Next, let us consider the case of even and
non-trivial $\sigma\in S_3$.

\begin{lemma}\label{lemma:psi sigma plus isotropic}
  Let $V\in C^\infty(\ro)$ and assume $\sigma\in S_3$ to have no fixed points (i.e., $\sigma$ is even and non-trivial). Then $\psi_\sigma^-$
  has no fixed points. Moreover, if $x=(\bn_1,\bn_2,\bn_3,k_1,k_2,k_3,\bphi_0,\bphi_1)$ and $\psi_\sigma^+(x)=x$, then all the $\bn_i$ equal and
  all the $k_i$ equal, so that $\psi_\lambda^+(x)=x$ for all $\lambda\in S_3$. In particular, if $x\in\sfB_\mfT$ arises from $\mfI\in\mB_\mfT[V]$, cf.
  Remark~\ref{remark:x arising from mfI}, and is a fixed point of $\psi_\sigma^+$, then $\mfI$ are isotropic. Finally, let $\mfI\in\mB_{\mfT}^{\iso}[V]$
  and assume that $x\in\sfB$ arises from $\mfI$, cf. Remark~\ref{remark:x arising from mfI}. Then $\psi_\lambda^+(x)=x$ for all $\lambda\in S_3$.
\end{lemma}
\begin{proof}
  By assumption, there are $\{i,j,l\}=\{1,2,3\}$ such that
  $\sigma(i)=j$, $\sigma(j)=l$ and $\sigma(l)=i$. If $\psi_\sigma^-(x)=x$, then $\bn_i=-\bn_j$, $\bn_j=-\bn_l$ and $\bn_l=-\bn_i$. This means that
  $\bn=0$, so that $x\in\sfB_\mrI$; i.e., $x$ does not belong to the domain of definition of $\psi_\sigma^-$. Thus the map $\psi_\sigma^-$ does not have
  any fixed points. If $\psi_\sigma^+(x)=x$, then $\bn_j=\bn_i$, $\bn_l=\bn_j$ and $\bn_i=\bn_l$ so that all the $\bn_m$ equal. Similarly, all the $k_m$
  equal. Next, let $\mfI=(G,\bge,\bk,\bphi_0,\bphi_1)\in\mB_\mfT[V]$ and $\{e_a\}$ be an orthonormal basis of $\mfg$ with respect to $\bge$ such
  that if $\bn$ and $K$ are the commutator and Weingarten matrices associated with $\{e_a\}$ and $\mfI$, then $x:=(\bn,K,\bphi_0,\bphi_1)\in\sfB_\mfT$.
  Assume, moreover, that $x$ is a fixed point of $\psi_\sigma^+$ (so that all the $\bn_m$ equal and all the $k_m$ equal). Then $\bk$ is a multiple of
  $\bge$. In order to prove that the Ricci tensor of $\bge$ is a multiple of $\bge$, note that if $\bR_{ab}:=\overline{\mathrm{Ric}}(e_a,e_b)$, where
  $\overline{\mathrm{Ric}}$ denotes the Ricci tensor of $\bge$, then $\bR_{ab}=0$ for $a\neq b$. Moreover, denoting the diagonal components by $\bR_a$,
  then, if $\{a,b,c\}=\{1,2,3\}$, 
  \begin{equation}\label{eq:Riccidiagitocm iso}
    \bR_{a}:=-\tfrac{1}{2}(\bn_{b}-\bn_{c})^{2}+\tfrac{1}{2}\bn_{a}^{2}
  \end{equation}
  (no summation on any index). These statements are justified in \cite[Lemma~19.11, p.~209]{RinCauchy}. Since all the $\bn_a$ equal, it follows from
  (\ref{eq:Riccidiagitocm iso}) that all the $\bR_a$ equal. Thus the Ricci tensor of $\bge$ is a multiple of $\bge$, and we conclude that
  $\mfI$ are isotropic. Finally, let $\mfI=(G,\bge,\bk,\bphi_0,\bphi_1)\in\mB_{\mfT}^{\iso}[V]$ and $\{e_a\}$ be an orthonormal basis of $\mfg$ with
  respect to $\bge$ such that if $\bn$ and $K$ are the commutator and Weingarten matrices associated with $\{e_a\}$ and $\mfI$, then
  $x=(\bn,K,\bphi_0,\bphi_1)\in\sfB_\mfT$. Then all the $k_m$ equal; see Definition~\ref{def:isotropic data}. In order to prove that all the $\bn_m$
  equal, note that (\ref{eq:Riccidiagitocm iso}) holds in this case as well. That all the $\bR_a$ equal follows from Definition~\ref{def:isotropic data}.
  There are only two possibilities
  for all the $\bR_a$ to equal: either all the $\bn_a$ equal, or two of the $\bn_a$ equal and the third vanishes. In the latter case, $\mfI$ are of
  isotropic Bianchi type VII${}_0$. However, that possibility has been excluded in the definition of $\mB[V]$; see Definition~\ref{def:BV}. Thus all
  the $\bn_a$ equal, and the last statement of the lemma follows. 
\end{proof}

\begin{lemma}\label{lemma:sfB char per symm VIz}
  Let $V\in C^\infty(\ro)$, $\sigma\in S_3$ be a transposition and $x\in\sfB-\sfB_\mrI$ be a fixed point of $\psi_\sigma^-$. Then $x\in\sfB_{\mrVIz}$.
  In fact, there are $\{i,j,l\}=\{1,2,3\}$ such that if $x=(\bn_1,\bn_2,\bn_3,k_1,k_2,k_3,\bphi_0,\bphi_1)$, then $\bn_i=0$, $\bn_j=-\bn_l\neq 0$
  and $k_j=k_l$. In particular, if $x\in\sfB_{\mrVIz}-\sfB_\mrI$ arises from $\mfI\in\mB_{\mrVIz}[V]$, cf.
  Remark~\ref{remark:x arising from mfI}, and is a fixed point of $\psi_\sigma^-$, then $\mfI$ are permutation symmetric. Finally, let
  $\mfI\in\mB_{\mrVIz}^{\roper}[V]$. Then there is an $x\in\sfB_{\mrVIz}$ arising from $\mfI$, cf. Remark~\ref{remark:x arising from mfI}, and a
  transposition $\sigma$ such that $\psi_\sigma^-(x)=x$.
\end{lemma}
\begin{proof}
  Let $\{i,j,l\}=\{1,2,3\}$ be such that $\sigma$ fixes $i$ and swaps $j$ and $l$. A fixed point $x$ of $\psi_\sigma^-$ is then such that
  $\bn_i=-\bn_i$, $\bn_j=-\bn_l$ and $k_j=k_l$. Thus $\bn_i=0$. Moreover, since $x\notin\sfB_\mrI$,
  $\bn_j$ and $\bn_l$ have to be non-zero and have opposite signs. This means that if $x\in\sfB-\sfB_\mrI$ is a
  fixed point of $\psi_\sigma^-$, then $x\in\sfB_{\mrVIz}$. Moreover, if $x\in\sfB-\sfB_\mrI$ arises from $\mfI\in\mB[V]$ and is a fixed point of
  $\psi_\sigma^-$, then $\mfI$ are permutation symmetric initial data, see Definition~\ref{def:permutation symmetry}. Finally, let
  $\mfI=(G,\bge,\bk,\bphi_0,\bphi_1)\in\mB_{\mrVIz}^{\roper}[V]$. We then wish to prove that there is an $x$ as in the last statement of the lemma.
  Let, to this end, $\{e_{i}'\}$ be an orthonormal basis of the type assumed to exist in
  Definition~\ref{def:permutation symmetry}. Let $K$ and $\bn$ be the corresponding Weingarten and commutator matrices. Due to the properties of
  $\Psi$, see Definition~\ref{def:permutation symmetry}, $\bn^{11}=0$, $\bn^{13}=-\bn^{12}$, $\bn^{23}=0$ and $\bn^{33}=-\bn^{22}$. This means that there
  are only two independent components of $\bn$, namely $\bn^{12}$ and $\bn^{22}$.
  Similarly $\bk_{12}=\bk_{13}$ and $\bk_{22}=\bk_{33}$. The eigenvalues
  of $\bn$ are $\lambda_0=0$, $\lambda_+$ and $\lambda_-=-\lambda_+$. Since $\bn\neq 0$, we can assume $\lambda_+$ to satisfy $\lambda_+>0$.
  Due to the symmetry of $\bn$, there are orthonormal eigenvectors $v_0$, $v_+$ and $v_-$ corresponding to $\lambda_0$, $\lambda_+$ and $\lambda_-$
  respectively. Due to the momentum constraint, the matrices $\bn$ and $K$ have to commute; see (\ref{eq:macom}). This means that $v_0$ and $v_\pm$
  are eigenvectors of $K$ with eigenvalues $\kappa_0$ and $\kappa_\pm$ respectively. In order to verify that $\kappa_+=\kappa_-$, it is convenient to
  first determine an eigenvector of $\bn$ corresponding to $\lambda_0$ in terms of $\bn^{12}$ and $\bn^{22}$ (we leave the details to the reader). Then
  one can verify that $K$ is a multiple of the identity on the orthogonal space (in order to do so, it is useful to appeal to the already mentioned
  relations between the coefficients as well as the relations arising from the fact that $K$ and $\bn$ commute). To conclude, there is an
  orthonormal frame $\{\be_{i}\}$ of $\mfg$ with respect to $\bge$ such that if $\nu$ and $\bK$ are the corresponding commutator and Weingarten matrices,
  then $\nu=\mathrm{diag}(0,\nu_0,-\nu_0)$ for a $0<\nu_0\in\rn{}$. Moreover, $K=\mathrm{diag}(p_1,p_2,p_2)$. The last statement of the lemma follows. 
\end{proof}

\begin{lemma}\label{lemma:sfB char per symm VIIz}
  Let $V\in C^\infty(\ro)$, $\sigma\in S_3$ be a transposition and $x\in\sfB$ be a fixed point of $\psi_\sigma^+$. Then there are
  $\{i,j,l\}=\{1,2,3\}$ such that if $x=(\bn_1,\bn_2,\bn_3,k_1,k_2,k_3,\bphi_0,\bphi_1)$, then $\bn_j=\bn_l$ and $k_j=k_l$. In particular, if
  $x\in\sfB_{\mfT}$ arises from $\mfI\in\mB_{\mfT}[V]$, cf. Remark~\ref{remark:x arising from mfI}, is a fixed point of $\psi_\sigma^+$, but is not
  a fixed point of all the maps in $\Gamma^+$, then $\mfI$ are LRS. Finally, let $\mfI\in\mB_{\mfT}^{\roLRS}[V]$. Then there is an
  $x\in\sfB_{\mfT}$ arising from $\mfI$, cf. Remark~\ref{remark:x arising from mfI}, which is not a fixed point of all the maps in $\Gamma^+$,
  and a transposition $\sigma$ such that $\psi_\sigma^+(x)=x$.
\end{lemma}
\begin{remark}\label{remark:sigma but not lambda}
  If $x$ is as in the last statement of the lemma, then there is no transposition, say $\lambda$, different from $\sigma$ such that
  $\psi_\lambda^+(x)=x$. In order to prove this statement, assume that there is such a $\lambda$. Then $\lambda\circ\sigma$ is a non-trivial even
  element of $S_3$ under which $x$ is invariant. Due to Lemma~\ref{lemma:psi sigma plus isotropic}, this means that $x$ is invariant under $\Gamma^+$,
  in contradiction with the assumptions. 
\end{remark}
\begin{proof}
  Let $\{i,j,l\}=\{1,2,3\}$ be such that $\sigma$ fixes $i$ and swaps $j$ and $l$. Then $\bn_j=\bn_l$ and $k_j=k_l$. This means that if $x\in\sfB$ arises
  from $\mfI\in\mB[V]$ and is a fixed point of $\psi_\sigma^+$, then $\mfI$ are either isotropic or LRS, see Definition~\ref{def:LRS}. If $\mfI$ are
  isotropic, then $x$ is a fixed point of all the maps in $\Gamma^+$; see Lemma~\ref{lemma:psi sigma plus isotropic}. Thus, if $x\in\sfB$ arises from
  $\mfI\in\mB[V]$, cf. Remark~\ref{remark:x arising from mfI}, is a fixed point of $\psi_\sigma^+$ but is not a fixed point of all the maps in $\Gamma^+$,
  then $\mfI$ are LRS. Assume now that $\mfI\in\mB_{\mfT}^{\roLRS}[V]$. Let $\{e_{i}\}$ be an orthonormal
  basis of the type assumed to exist in Definition~\ref{def:LRS}. Let $K$ and $\bn$ be the corresponding Weingarten and commutator matrices. Then the
  assumptions of Definition~\ref{def:LRS} imply that $K$ and $\bn$ are diagonal. Moreover, if $k_i$ and $\bn_i$ denote the diagonal components of $K$
  and $\bn$ respectively, then $\bn_2=\bn_3$ and $k_2=k_3$. In particular, if $x=(\bn_1,\bn_2,\bn_3,k_1,k_2,k_3,\bphi_0,\bphi_1)$, then there is a
  transposition $\sigma$ such that $\psi_\sigma^+(x)=x$. Finally, since $\mfI$ are not isotropic, $x$ is not a fixed point of all the maps in $\Gamma^+$.
\end{proof}
\begin{definition}\label{def:BTs GTs}
  Fix $V\in C^\infty(\ro)$. Let $\mfT$ be a Bianchi class A type. Define $\sfB_\mfT^{\iso}$ to be the subset of $\sfB_\mfT$ consisting of points
  invariant under $\Gamma^+$. Define $\sfB_\mrI^{\roper}=\varnothing$ and define, for $\mfT\neq\mrI$, $\sfB_{\mfT}^{\roper}$ to be the set of $p\in\sfB_{\mfT}$
  such that there is a map $\psi_\sigma^-\in \Gamma^{-,\roodd}$ with $\psi_\sigma^-(p)=p$. Define 
  $\sfB_{\mfT}^{\roLRS}$ to be the set of $p\in\sfB_{\mfT}-\sfB_{\mfT}^{\iso}$ such that there is a map $\psi_\sigma^+\in \Gamma^{+,\roodd}$ with
  $\psi_\sigma^+(p)=p$. Define
  \[
  \sfB_{\mfT}^{\rogen}:=\sfB_{\mfT}-\sfB_{\mfT}^{\iso}-\sfB_{\mfT}^{\roper}-\sfB_{\mfT}^{\roLRS}.
  \]
  \index{$\a$Aa@Notation!Symmetry reduced sets of regular initial data!sfBTs@$\sfB_\mfT^\mfs$}%
  Define $\Gamma_\mrI^{\iso}:=\{\roId\}$ and $\Gamma_{\mrIX}^{\iso}:=\{\roId,\psi^{-}_{\roId}\}$. Define $\Gamma_{\mrVIz}^{\roper}:=\Gamma^{\roev}$,
  $\Gamma_{\mrI}^{\roLRS}:=\Gamma^{+,\roev}$ and $\Gamma_{\mfT}^{\roLRS}:=\Gamma^{\roev}$ for $\mfT\notin \{\mrI,\mrVIz\}$. Finally, define
  $\Gamma_{\mrI}^{\rogen}:=\Gamma^+$ and $\Gamma_\mfT^{\rogen}:=\Gamma$ for $\mfT\neq\mrI$.
  \index{$\a$Aa@Notation!Symmetry groups!GammaTs@$\Gamma_\mfT^{\mfs}$}%
\end{definition}
\begin{remark}
  Note that $\sfB_\mfT^{\iso}$ is empty unless $\mfT\in \{\mrI,\mrIX\}$; $\sfB_\mfT^{\roper}$ is empty unless $\mfT=\mrVIz$;
  and $\sfB_{\mrVIz}^{\roLRS}$ is empty.
\end{remark}
Due to the observations made prior to the statement of the definition we obtain the following conclusions.
\begin{lemma}\label{lemma:sfR mfT mfs}
  Fix $V\in C^\infty(\ro)$. Let $\mfT$ be a Bianchi class A type and $\mfs\in\{\iso,\roper,\roLRS,\rogen\}$ be such that $\sfB_{\mfT}^{\mfs}$ and
  $\Gamma_{\mfT}^{\mfs}$ are defined in Definition~\ref{def:BTs GTs} and are non-empty. Then $\Gamma_{\mfT}^{\mfs}$ acts freely and properly discontinuously on
  $\sfB_{\mfT}^{\mfs}$ and the corresponding quotient is smooth and denoted by $\sfR_{\mfT}^{\mfs}$.
    \index{$\a$Aa@Notation!Symmetry reduced sets of regular initial data!sfRTs@$\sfR_\mfT^\mfs$}%
\end{lemma}
\begin{proof}
  By an argument similar to Remark~\ref{remark:mfBmfT subm}, the sets $\sfB_\mfT^\mfs$ are smooth manifolds. Next, note that if $\sigma$ has no fixed points
  (i.e., $\sigma$ is a non-trivial even element of $S_3$), then $\psi_\sigma^-$ has no fixed points and $\psi_\sigma^+$ only fixes isotropic elements; see
  Lemma~\ref{lemma:psi sigma plus isotropic}. If $\sigma$ is odd, then $\sigma$ is a transposition. Moreover, if $\psi_\sigma^-(x)=x$, then $x\in\sfB_{\mrVIz}^{\roper}$,
  and if $\psi_\sigma^+(x)=x$, then $x\in \sfB_\mfT^{\roLRS}\cup\sfB_\mfT^{\iso}$; see Lemmas~\ref{lemma:sfB char per symm VIz} and \ref{lemma:sfB char per symm VIIz}. Combining these
  observations with Definition~\ref{def:BTs GTs}, it is clear that $\Gamma_{\mfT}^{\mfs}$ is a finite group which acts on $\sfB_{\mfT}^{\mfs}$ without fixed points.
  Thus $\Gamma_{\mfT}^{\mfs}$ acts freely and properly discontinuously on $\sfB_{\mfT}^{\mfs}$, so that the quotient is smooth.
\end{proof}
Finally, we are in a position to parametrise ${}^{\rosc}\mfB_\mfT^\mfs[V]$ by smooth manifolds.
\begin{lemma}\label{lemma:sc mfB mfT mfs param}
  Let $V\in C^\infty(\ro)$. Then there is a bijection from ${}^{\rosc}\mfB_\mfT^\mfs[V]$ to $\sfR_{\mfT}^{\mfs}$ for the following combinations of
  $\mfT$ and $\mfs$: if $\mfs=\iso$ and $\mfT\in\{\mrI,\mrIX\}$; if $\mfs=\roLRS$ and $\mfT\in \{\mrI,\mrII,\mrVIII,\mrIX\}$; if
  $(\mfT,\mfs)=(\mrVIz,\roper)$; and if $\mfs=\rogen$ and $\mfT$ is an arbitrary Bianchi class A type.
\end{lemma}
\begin{remark}\label{remark:sc mfB mfT mfs param}
  Combining the bijection obtained in the lemma with Lemma~\ref{lemma:sfR mfT mfs} yields a smooth structure on ${}^{\rosc}\mfB_\mfT^\mfs[V]$
  for the relevant combinations of $\mfT$ and $\mfs$. In what follows, we assume ${}^{\rosc}\mfB_\mfT^\mfs[V]$ to have been endowed with this
  smooth structure. 
\end{remark}
\begin{remark}\label{remark:eta BTs realised}
  One immediate consequence of the proof is that if $\eta\in \sfB_\mfT^\mfs$, with $\mfT$ and $\mfs$ as in the statement of the lemma, then
  there are $\mfI\in{}^{\rosc}\mB_\mfT^\mfs[V]$ such that $\eta$ arises from $\mfI$. Moreover, we can, with each $\mfI\in{}^{\rosc}\mB_\mfT^\mfs[V]$
  associate an $x\in \sfB_\mfT^\mfs$ so that $x$ arises from $\mfI$. 
\end{remark}
\begin{proof}
  Let $(\mfT,\mfs)$ be one of the combinations mentioned in the statement of the lemma and $\mfI\in{}^{\rosc}\mB_\mfT^\mfs[V]$. Then there is an
  $x\in\sfB_{\mfT}$ arising from $\mfI$; cf. Lemma~\ref{lemma:x arising from mfI} and Remark~\ref{remark:x arising from mfI}. If $\mfs=\iso$, then
  $x$ is invariant under $\Gamma^+$ due to Lemma~\ref{lemma:psi sigma plus isotropic}. This means that $x\in\sfB_\mfT^{\iso}$; cf.
  Definition~\ref{def:BTs GTs}. Next, if $\mfs=\roper$, then we can assume $x$ to belong to
  $\sfB_{\mrVIz}^{\roper}$ due to Lemma~\ref{lemma:sfB char per symm VIz} and Definition~\ref{def:BTs GTs}. If $\mfs=\roLRS$, then we can
  assume $x$ to be such that it is not invariant under all the maps in $\Gamma^+$, but such that there is a transposition $\sigma$ such that
  $\psi_\sigma^+(x)=x$; see Lemma~\ref{lemma:sfB char per symm VIIz}. This means that $x\in\sfB_{\mfT}^{\roLRS}$; see Definition~\ref{def:BTs GTs}.
  Finally, let $\mfs=\rogen$. We then wish to prove that $x\in \sfB_\mfT^\mfs$. Assume, to this end, that $x\in\sfB_\mfT^{\mfr}$ for some
  $\mfr\neq\rogen$. If $\mfr=\iso$, then $\mfs=\iso$ due to Lemma~\ref{lemma:psi sigma plus isotropic}, a contradiction. If $\mfr=\roper$, then
  $\mfs=\roper$ due to Lemma~\ref{lemma:sfB char per symm VIz}, a contradiction. If $\mfr=\roLRS$, then $\mfs=\roLRS$ due to
  Lemma~\ref{lemma:sfB char per symm VIIz}, a contradiction. To conclude: we can, with each $\mfI\in{}^{\rosc}\mB_\mfT^\mfs[V]$ associate an
  $x\in \sfB_\mfT^\mfs$, which, in its turn, gives rise to an element, say $[x]$, of $\sfR_\mfT^\mfs$; see Lemma~\ref{lemma:sfR mfT mfs}.
  If $x_i\in\sfB_\mfT^\mfs$, $i=1,2$, are both associated with $\mfI$ in this way, then $[x_1]=[x_2]$ due to Lemma~\ref{lemma:isometryinducesequalnuK}.
  In other words, there is a well defined map from ${}^{\rosc}\mB_\mfT^\mfs[V]$ to $\sfR_\mfT^\mfs$. Moreover, again due to
  Lemma~\ref{lemma:isometryinducesequalnuK}, this map descends to a map from ${}^{\rosc}\mfB_\mfT^\mfs[V]$ to $\sfR_\mfT^\mfs$. Next, if two initial data
  sets, say $\mfI_0$ and $\mfI_1$, are mapped to the same element, then $\mfI_0$ and $\mfI_1$ are isometric; see Lemma~\ref{lemma:nuKsameisometric}.
  Thus the map from ${}^{\rosc}\mfB_\mfT^\mfs[V]$ to $\sfR_\mfT^\mfs$ is well defined and injective. What remains is to prove surjectivity.

  In order to prove surjectivity, let $[x]\in\sfR_\mfT^\mfs$ for some $x\in\sfB_\mfT^\mfs$. Due to Lemma~\ref{lemma:realcond}, there are non-trivial simply
  connected regular initial data $\mfI=(G,\bge,\bk,\bphi_0,\bphi_1)$ and a basis $\{e_i\}$ of $\mfg$, orthonormal with respect
  to $\bge$, such that if $\bn$ and $K$ are the commutator and Weingarten matrices associated with $\{e_i\}$ and $\mfI$, then $x=(\bn,K,\bphi_0,\bphi_1)$.
  We need to prove that $\mfI\in{}^{\rosc}\mB_\mfT^\mfs[V]$. All that remains to be demonstrated is that $\mfI$ are of symmetry type $\mfs$ and that $\mfI$
  are neither of isotropic nor of LRS Bianchi type $\mrVIIz$; see Definition~\ref{def:BV}. Assume, for the moment, that $\mfI$ are neither of isotropic
  nor of LRS Bianchi type $\mrVIIz$. Then we know that $\mfI\in{}^{\rosc}\mB_\mfT[V]$. This means that $\mfI\in{}^{\rosc}\mB_\mfT^{\mfr}[V]$ for some symmetry
  type $\mfr$. By the arguments presented above, this means that there is a basis $\{e_i'\}$ of $\mfg$, orthonormal with respect to $\bge$, such that if
  $\bn'$ and $K'$ are the commutator and Weingarten matrices associated with $\{e_i'\}$ and $\mfI$, then $x'=(\bn',K',\bphi_0,\bphi_1)\in\sfB_{\mfT}^{\mfr}$.
  Combining these observations with Lemma~\ref{lemma:isometryinducesequalnuK} and the fact that $\sfB_{\mfT}^{\mfr}$ and $\sfB_{\mfT}^{\mfs}$ are disjoint if
  $\mfs\neq\mfr$, it follows that $\mfr=\mfs$. In other words, $\mfI\in{}^{\rosc}\mB_\mfT^\mfs[V]$. This proves surjectivity, with one caveat: we need to
  exclude the possibility that $\mfI$ are isotropic or LRS Bianchi type $\mrVIIz$. Assume, to this end, that they are. Then the Ricci tensor of $\bge$
  vanishes identically; see Lemma~\ref{lemma:LRS BVIIz is BI}. This means that $\bn$ has to be such that $\bn_i=0$ and $\bn_j=\bn_l$ for some $i,j,l$ such
  that $\{i,j,l\}=\{1,2,3\}$; see (\ref{eq:Riccidiagitocm iso}). If $k_j=k_l$, then $x\in\sfB_{\mrVIIz}^{\mfs}$ for $\mfs\in\{\iso,\roLRS\}$, in
  contradiction with the assumptions. Thus $k_j\neq k_l$. In particular, the initial data are not isotropic. If $k_i\notin\{k_j,k_l\}$, the initial data
  cannot be LRS either. Assume therefore, without loss of generality, that $k_2=k_3\neq k_1$, that $\bn_2=0$ and that $\bn_3=\bn_1$. In order for $\mfI$
  to be LRS, there must be an orthonormal basis $\{e_m'\}$ of $\mfg$ with respect to $\bge$ and a family $\Psi_t$ of Lie algebra isomorphisms such that
  the conditions of Definition~\ref{def:LRS} hold with $\{e_i\}$
  replaced by $\{e_i'\}$. Let $K'$ and $\bn'$ be the Weingarten and commutator matrices associated with $\{e_i'\}$ and $\mfI$. Then, due to the
  the conditions of Definition~\ref{def:LRS}, $K'$ and $\bn'$ are diagonal. Moreover, if $k_i'$ and $\bn_i'$ denote the diagonal
  components of $K'$ and $\bn'$ respectively, then $\bn_2'=\bn_3'$ and $k_2'=k_3'$. This means, in particular, that $k_A=k_A'$ for $A\in\{2,3\}$ and
  that $k_1=k_1'$. Thus $e_1'=\pm e_1$. Since the initial data are of Bianchi type $\mrVIIz$ and $\bn_2'=\bn_3'$, we also have to have $\bn_2'\neq 0$
  and $\bn_1'=0$. Since the plane spanned by $e_2'$ and $e_3'$ has to coincide with the plane spanned by $e_2$ and $e_3$ (they are both the orthogonal
  complement of $e_1'=\pm e_1$), there is also a $t_0$ such that $\Psi_{t_0}e_2'=e_2$ and $\Psi_{t_0}e_3'=\pm e_3$. This means that
  \[
  0=[e_2',\pm e_3']=[\Psi_{t_0}e_2',\pm\Psi_{t_0}e_3']=[e_2,e_3]=\bn_1 e_1.
  \]
  Thus $\bn_1=0$. Since $\bn_2=0$ and $\bn_3=\bn_1$, we conclude that the initial data are of Bianchi type I, in contradiction with the assumptions.
  To conclude, $\mfI$ are neither of isotropic nor of LRS Bianchi type $\mrVIIz$. This completes the proof of surjectivity. 
\end{proof}

\section{Parametrising developments}\label{ssection:param solns}

We introduce the terminology ${}^{\rosc}\mfD_{\mfT}^{\mfs}[V]$ etc. in Definition~\ref{def:mDV}. In order to prove that these sets can be
parametrised by smooth manifolds, it is useful to first focus on sets of isometry classes of initial data with a fixed mean curvature. 
\begin{definition}\label{def:B varthetaz}
  Fix $V\in C^\infty(\ro)$ and $\vartheta_0\in\ro$. Let $\mfs\in\{\iso,\roper,\roLRS,\rogen\}$ and $\mfT$ be a Bianchi class A type
  such that $\sfB_{\mfT}^{\mfs}$ and $\Gamma_{\mfT}^{\mfs}$ are defined in Definition~\ref{def:BTs GTs} and non-empty. Using the notation introduced
  in Lemma~\ref{lemma:sfR mfT mfs}, let $\sfB_{\mfT}^{\mfs}(\vartheta_0)$ and $\sfR_\mfT^\mfs(\vartheta_0)$
  \index{$\a$Aa@Notation!Symmetry reduced sets of regular initial data!sfRTsthz@$\sfR_\mfT^\mfs(\vartheta_0)$}%
  \index{$\a$Aa@Notation!Symmetry reduced sets of regular initial data!sfBTsthz@$\sfB_\mfT^\mfs(\vartheta_0)$}%
  be the subsets of $\sfB_{\mfT}^{\mfs}$ and $\sfR_\mfT^\mfs$ respectively with $\tr K=\vartheta_0$.
  Moreover, let $\sfR_{\mrI,\rond}^{\iso}(\vartheta_0)$
  \index{$\a$Aa@Notation!Symmetry reduced sets of regular initial data!sfRIndisothz@$\sfR_{\mrI,\rond}^{\iso}(\vartheta_0)$}%
  be the subset of $\sfR_{\mrI}^{\iso}$ with $\tr K=\vartheta_0$ and $\bphi_1\neq 0$.
\end{definition}
In order to analyse the regularity of these sets, we begin by reformulating (\ref{eq:nuvhc}). Let
\begin{equation}\label{eq:ki ito kp km}
  k_1=\tfrac{1}{3}\theta-\tfrac{2}{3}k_+,\ \ \
  k_2=\tfrac{1}{3}\theta+\tfrac{1}{3}k_++\tfrac{1}{\sqrt{3}}k_-,\ \ \
  k_3=\tfrac{1}{3}\theta+\tfrac{1}{3}k_+-\tfrac{1}{\sqrt{3}}k_-,
\end{equation}
where $\theta=\tr K$. With this notation, (\ref{eq:nuvhc}) can be written
\begin{equation}\label{eq:ham cons fixed theta}
  \theta^2=k_+^2+k_-^2+\tfrac{3}{4}(\bn_1^2+\bn_2^2+\bn_3^2-2\bn_1\bn_2-2\bn_2\bn_3-2\bn_3\bn_1)+\tfrac{3}{2}\bphi_1^2+3V(\bphi_0).
\end{equation}
Fixing the left hand side of (\ref{eq:ham cons fixed theta}) to be $\vartheta_0^2$ for some $\vartheta_0\in\ro$, the only situation in which
the gradient of the right hand side vanishes is when $(\mfT,\mfs)=(\mrI,\iso)$, $\bphi_1=0$, $V'(\bphi_0)=0$ and $\vartheta_0^2=3V(\bphi_0)$; i.e.,
if the initial data are trivial, see Definition~\ref{def:Bianchiid} and Remark~\ref{remark:dev of trivial id}. However, this possibility has
been excluded from $\sfB_\mfT^\mfs$, see Definition~\ref{def:degenerate version of id}. This observation leads to the following conclusion.
\begin{lemma}\label{lemma:sfRmfTmfsvartheta sm mfd}
  Fix $V\in C^\infty(\ro)$, $\vartheta_0\in\ro$ and $(\mfT,\mfs)$ as in Definition~\ref{def:B varthetaz}. Then $\sfR_\mfT^\mfs(\vartheta_0)$ is a
  smooth manifold.
\end{lemma}
\begin{remark}\label{remark:par iso id fixed mc}
  The bijection whose existence is guaranteed by Lemma~\ref{lemma:sc mfB mfT mfs param} preserves the mean curvature. It therefore induces
  a bijection from ${}^{\rosc}\mfB_\mfT^{\mfs}[V](\vartheta_0)$ to $\sfR_\mfT^\mfs(\vartheta_0)$ for the combinations of $\mfT$ and $\mfs$ appearing
  in Lemma~\ref{lemma:sc mfB mfT mfs param}. Similarly, there is a bijection from $\sfR_{\mrI,\rond}^{\iso}(\vartheta_0)$ to 
  \begin{equation}\label{eq:mfB I nd def}
    {}^{\rosc}\mfB_{\mrI,\rond}^{\iso}[V](\vartheta_0):={}^{\rosc}\mfB_{\mrI}^{\iso}[V](\vartheta_0)\cap \{\bphi_1\neq 0\}.
  \end{equation}
  Due to the lemma, these bijections define smooth structures on ${}^{\rosc}\mfB_\mfT^{\mfs}[V](\vartheta_0)$ and
  ${}^{\rosc}\mfB_{\mrI,\rond}^{\iso}[V](\vartheta_0)$. In Section~\ref{section:Id fix mean curv}, we explain how to combine these sets to define a smooth
  structure on the set of isometry classes of developments.
\end{remark}
\begin{proof}
  Due to the observations made prior to the statement of the lemma, the set of elements in $\sfB_\mfT^\mfs$ with fixed mean curvature
  $\vartheta_0$, say $\sfB_\mfT^\mfs(\vartheta_0)$, is a smooth manifold. The elements of $\Gamma_\mfT^\mfs$ preserve the mean curvature
  and therefore $\Gamma_\mfT^\mfs$ acts freely and properly discontinuously on $\sfB_\mfT^\mfs(\vartheta_0)$. This means that the quotient, i.e.
  $\sfR_\mfT^\mfs(\vartheta_0)$, is smooth. 
\end{proof}
In the case of isotropic Bianchi type I, it is unfortunately not sufficient to focus on ${}^{\rosc}\mfB_{\mrI,\rond}^{\iso}[V](\vartheta_0)$; we also have to
impose the condition that the corresponding developments have a crushing singularity. Moreover, in order to obtain a smooth structure for the
set of isometry classes of developments, we need the condition that there is a crushing singularity to be open. One way to ensure this is to demand
that $V\in\mfP_{\ropar}$, see Definition~\ref{def:mfPdef}. Imposing this condition also makes it possible to ensure that all solutions with a crushing
singularity attain one fixed mean curvature; i.e., it is possible to parametrise the set of elements in ${}^{\rosc}\mfD_\mfT^\mfs[V]$ with a crushing
singularity by ${}^{\rosc}\mfB_\mfT^{\mfs}[V](\vartheta_0)$ for a suitably chosen $\vartheta_0$. To justify these statements, we next make some observations
concerning potentials belonging to $\mfP_{\ropar}$.
\begin{lemma}\label{lemma:vmax bds limsup}
  Let $V\in\mfP_{\ropar}$; see Definition~\ref{def:mfPdef}. If $\{V(s)\, |\, s\geq 0\}$ is bounded, then $v_{\max}(V)\geq\limsup_{s\rightarrow\infty}V(s)$.
  Similarly, if $\{V(s)\, |\, s\leq 0\}$ is bounded, then $v_{\max}(V)\geq\limsup_{s\rightarrow-\infty}V(s)$.
\end{lemma}
\begin{proof}
  Note that $V\in\mfP_{\ropar}\Rightarrow \overline{V}\in\mfP_{\ropar}$ with $v_{\max}(\overline{V})=v_{\max}(V)$, where $\overline{V}(s):=V(-s)$. Thus
  the second statement of the lemma follows by applying the first statement of the lemma to $\overline{V}$. We therefore only prove the
  first statement. If $\lim_{s\rightarrow\infty}V(s)$ exists, the desired conclusion follows from the definition. Assume that the limit does not
  exist. Then $v_-<v_+$, where
  \[
    v_+:=\limsup_{s\rightarrow\infty}V(s),\ \ \
    v_-:=\liminf_{s\rightarrow\infty}V(s).
  \]
  There are thus, for each $n\in\nn{}$, $s_n\leq u_n$ such that $s_n,u_n\rightarrow\infty$, $V(s_n)\geq v_+-1/n$, 
  $V'(s_n)>0$ and $V'(u_n)<0$. This means that there is a $w_n\in [s_n,u_n]$ such that $V'(s)\geq 0$ in $[s_n,w_n]$
  and $V'(w_n)=0$. Since $V'\geq 0$ on $[s_n,w_n]$, we have to have $V(w_n)\geq v_+-1/n$. By the definition of $v_{\max}$, this means that
  $v_{\max}(V)\geq v_+-1/n$ for all $n\in\nn{}$, so that $v_{\max}(V)\geq v_+$. The lemma follows. 
\end{proof}
\begin{cor}\label{cor:vmax geq sup}
  Let $V\in\mfP_{\ropar}$; see Definition~\ref{def:mfPdef}. If $V$ is bounded, then $v_{\max}(V)\geq \sup_{s\in\ro}V(s)$. 
\end{cor}
\begin{proof}
  Let $\e>0$ and define
  \[
    V_+:=\limsup_{s\rightarrow\infty}V(s),\ \ \
    V_-:=\limsup_{s\rightarrow-\infty}V(s).
  \]
  Then there are $s_-<s_+$ such that $V(s)\leq V_-+\e$ for all $s\leq s_-$ and $V(s)\leq V_++\e$ for all $s\geq s_+$. If $V$ attains its maximum on
  $[s_-,s_+]$ at an interior point, say $s_0$, then $V'(s_0)=0$, so that
  \[
    v_{\max}(V)\geq V(s_0)=\max_{s\in[s_-,s_+]}V(s).
  \]
  If $V$ attains its maximum on $[s_-,s_+]$ on the boundary, then, due to Lemma~\ref{lemma:vmax bds limsup},
  \[
    v_{\max}(V)+\e\geq\max\{V_+,V_-\}+\e\geq \max_{s\in[s_-,s_+]}V(s).
  \]
  To conclude $v_{\max}(V)+\e\geq\sup_{s\in\ro}V(s)$. Since this holds for all $\e>0$, the desired conclusion follows. 
\end{proof}
It is also going to be useful to know the following concerning potentials $V\in\mfP_{\ropar}$.
\begin{lemma}\label{lemma:V prime sign}
  Let $V\in\mfP_{\ropar}$; see Definition~\ref{def:mfPdef}. If $\{V(s)\, |\, s\geq 0\}$ is unbounded, there is an $s_+$ such that $V'(s)>0$ for all
  $s\geq s_+$. If $\{V(s)\, |\, s\leq 0\}$ is unbounded, there is an $s_-$ such that $V'(s)<0$ for all $s\leq s_-$. 
\end{lemma}
\begin{remark}\label{remark:divergence to infinity V}
  If $\{V(s)\, |\, s\geq 0\}$ is unbounded, it follows from the conclusions of the lemma that $V(s)\rightarrow\infty$ as $s\rightarrow\infty$. There
  are two possibilities. Either $V'(s)>0$ for all $s$. In that case $V(s)>v_{\max}(V)$ for all $s\in\ro$. Define, in this case, $s_+(V):=-\infty$. The
  other possibility is that there is an $s_0\in\ro$ such that $V'(s_0)=0$. In that case, let $s_a:=\sup\{s\in\ro:V'(s)=0\}$. Due to the assumptions
  and the conclusions of the lemma, $s_a\in\ro$.
  There are two possibilities. Either $V(s_a)=v_{\max}(V)$, in which case we define $s_{+}(V):=s_a$; note that $s_+(V)$ is then the smallest $s$
  such that $V(u)>v_{\max}(V)$ for all $u>s$. The final possibility is that $V(s_a)<v_{\max}(V)$. In that case there is a smallest $s$, larger than
  $s_a$, such that $V(s)=v_{\max}(V)$. Call this smallest $s$ $s_b$ and define $s_{+}(V):=s_b$; note that $s_+(V)$ is then the smallest $s$
  such that $V(u)>v_{\max}(V)$ for all $u>s$. By the definition of $v_{\max}(V)$, it also follows that $V'(s)>0$ for all $s>s_+(V)$. One can
  similarly define $s_-(V)$ in case $\{V(s)\, |\, s\leq 0\}$ is unbounded. 
\end{remark}
\begin{proof}
  For the same reason as in the proof of Lemma~\ref{lemma:vmax bds limsup}, it is sufficient to prove the first statement.   
  Assume, in order to obtain a contradiction, that
  $\{V(s)\, |\, s\geq 0\}$ is unbounded, but there is a sequence $s_n\rightarrow\infty$ such that $V'(s_n)\leq 0$. If there is a subsequence
  $\{s_{n_k}\}$ such that $V(s_{n_k})\rightarrow\infty$, we conclude that $v_{\max}(V)=\infty$, contradicting the fact that $V\in\mfP_{\ropar}$. Choosing a
  subsequence, if necessary, we can thus assume that $V(s_n)$ converges to a finite limit, say $V_+$. By the assumptions, there is also a sequence
  $u_n\rightarrow\infty$ such that $V(u_n)\rightarrow\infty$. Combining the existence of these two sequences, an argument similar to the one presented
  in the proof of Lemma~\ref{lemma:vmax bds limsup} implies that there is a sequence $w_n\rightarrow\infty$ such that $V'(w_n)=0$ and
  $V(w_n)\rightarrow\infty$. Thus $v_{\max}(V)=\infty$, contradicting the fact that $V\in\mfP_{\ropar}$. 
\end{proof}
Fix a $V\in\mfP_{\ropar}$ and a $\vartheta_0>[3v_{\max}(V)]^{1/2}$. Then, in the Bianchi class A setting, Bianchi type $\mfT\neq$IX initial data with
initial mean curvature $\vartheta_0$ correspond to Bianchi type $\mfT$ non-linear scalar field developments with a crushing singularity; see
Lemma~\ref{lemma:bth large enough}. The condition that the solution have a crushing singularity of course excludes some isotropic Bianchi type I
solutions; see Remark~\ref{remark:Exceptional Bianchi type I} and Lemma~\ref{lemma:BianchiAdevelopment}. Next, we turn to the topology of
$\sfR_{\mrI}^{\iso}(\vartheta_0)$, which, for reasons explained in the introduction, is of particular interest.
\begin{lemma}\label{lemma:top of devel iso BI}
  Fix a $V\in\mfP_{\ropar}$ and a $\vartheta_0>[3v_{\max}(V)]^{1/2}$. If $V$ is bounded, $\sfR_{\mrI}^{\iso}(\vartheta_0)$ is diffeomorphic to two
  copies of $\ro$. If one of $\{V(s)\, |\, s\leq 0\}$ and $\{V(s)\, |\, s\geq 0\}$ is bounded and the other is unbounded, then
  $\sfR_{\mrI}^{\iso}(\vartheta_0)$ is diffeomorphic to $\ro$. If both are unbounded, then $\sfR_{\mrI}^{\iso}(\vartheta_0)$ is diffeomorphic to $\sn{1}$.
\end{lemma}
\begin{remark}\label{remark:sfR I iso B I plus iso}
  Due to Definition~\ref{def:degenerate version of id}, Definition~\ref{def:BTs GTs} and Lemma~\ref{lemma:sfR mfT mfs}, it is clear that
  \begin{equation*}
    \begin{split}
      \sfR_{\mrI}^{\iso} = & \sfB_{\mrI}^{\iso}/\Gamma^{\iso}_{\mrI}=\sfB_{\mrI}^{\iso}\\
      = & \{(0,\theta_0\roId/3,\phi_0,\phi_1)\, |\, \theta_0^2=3\phi_1^2/2+3V(\phi_0),\ (\phi_1,V'(\phi_0))\neq (0,0)\}. 
    \end{split}
  \end{equation*}  
  This set can clearly be identified with $B_{\mrI,+}^{\iso}$. Similarly, we can identify $\sfR_{\mrI}^{\iso}(\vartheta_0)$ with
  $B_{\mrI,+}^{\iso}(\vartheta_0)$.
\end{remark}
\begin{proof}
  Due to Remark~\ref{remark:sfR I iso B I plus iso}, it is sufficient to determine the topology of $B_{\mrI,+}^{\iso}(\vartheta_0)$. To begin with,
  it is useful to note that the condition of non-degeneracy is superfluous for our choice of $\vartheta_0$. Let, to this end, $(\phi_0,\phi_1)$ be such
  that
  \begin{equation}\label{eq:ham con vartheta zero BI iso}
    \vartheta_0^2=\tfrac{3}{2}\phi_1^2+3V(\phi_0).
  \end{equation}
  Then either $\phi_1\neq 0$ or $\vartheta_0^2=3V(\phi_0)$. In the latter case, we have to have $V'(\phi_0)\neq 0$, since
  $\vartheta_0>[3v_{\max}(V)]^{1/2}$. Thus, for our choice of $\vartheta_0$,
  \[
  B^{\iso}_{\mrI,+}(\vartheta_0) := \{(\phi_0,\phi_1)\in \rn{2}\, |\, \vartheta_0^2=3\phi_1^2/2+3V(\phi_0)\};
  \]
  i.e., the condition of non-degeneracy is automatically satisfied. 

  In case $V$ is bounded, then $\vartheta_0^2>3\sup_{s\in\ro}V(s)$, cf. Corollary~\ref{cor:vmax geq sup}, so that (\ref{eq:ham con vartheta zero BI iso})
  defines two disjoint manifolds, each diffeomorphic to $\ro$ (the manifolds are distinguished by the sign of $\phi_1$).
  
  Next, consider the case that $\{V(s)\, |\, s\leq 0\}$ is unbounded but that $\{V(s)\, |\, s\geq 0\}$ is bounded. Due to Lemma~\ref{lemma:vmax bds limsup}
  and the definition of $\vartheta_0$, we know that there is an $s_1\in\ro$ such that $3V(s)<\vartheta_0^2$ for all $s\geq s_1$. Since, in addition,
  $V(u)\rightarrow\infty$ as $u\rightarrow-\infty$, see Remark~\ref{remark:divergence to infinity V}, we conclude that 
  \[
    s_0:=\inf\{s\in\ro\, |\, 3V(u)\leq \vartheta_0^2\ \forall\, u\geq s\}
  \]
  is a real number. In particular, $3V(s_0)=\vartheta_0^2$ and $3V(s)\leq \vartheta_0^2$ for all $s\geq s_0$. Moreover, by the definition of $\vartheta_0$,
  $V'(s_0)\neq 0$. However, the possibility $V'(s_0)>0$ is excluded by the definition of $s_0$. This means that $V'(s_0)<0$. Assume, in order to obtain
  a contradiction, that there is an $s_2<s_0$ such that $3V(s_2)\leq\vartheta_0^2$. Then $V$ must attain a local maximum on $[s_2,s_0]$ at an $s_3\in(s_2,s_0)$.
  Then $V'(s_3)=0$, and $3V(s_3)>3V(s_0)=\vartheta_0^2$, in contradiction with the definition of $\vartheta_0$. Thus $3V(s)>\vartheta_0^2$
  for all $s<s_0$. The set of
  $\phi_0$ such that (\ref{eq:ham con vartheta zero BI iso}) has a solution is thus equal to $[s_0,\infty)$. Assume, in order to obtain a contradiction,
  that there is an $s_+>s_0$ such that $3V(s_+)=\vartheta_0^2$. Since $3V(s)\leq\vartheta_0^2$ for all $s\geq s_0$, this means that $s_+$ is a local
  maximum, so that $V'(s_+)=0$. However, this contradicts the definition of $\vartheta_0$. Thus $3V(s)<\vartheta_0^2$ for all $s>s_0$. We can thus
  define the following two sets:
  \[
    A_\pm:=\{(\phi_0,\phi_1)\in\rn{2}\, |\, \vartheta_0^2=3\phi_1^2/2+3V(\phi_0),\ \pm\phi_1\geq 0\}.
  \]
  By the above observations,
  \[
    A_\pm=\left\{(\phi_0,\phi_1)\in\rn{2}\, |\, \phi_0\in [s_0,\infty),\ \phi_1=\pm\left(2\vartheta_0^2/3-2V(\phi_0)\right)^{1/2}\right\}.
  \]
  Moreover, the sets $A_\pm$ are each homeomorphic to $[s_0,\infty)$, by projection to the $\phi_0$-axis, and they only intersect in $(s_0,0)$. This means
  that the union is homeomorphic to $\ro$. Since the union is $B_{\mrI,+}^{\iso}(\vartheta_0)$ and $B_{\mrI,+}^{\iso}(\vartheta_0)$ is a smooth
  $1$-dimensional submanifold of $\rn{2}$, we obtain the desired conclusion. The statement in case $\{V(s)\, |\, s\leq 0\}$ is bounded and
  $\{V(s)\, |\, s\geq 0\}$ is unbounded follows by applying the same argument to $\overline{V}$ introduced in the proof of Lemma~\ref{lemma:vmax bds limsup}.

  Next, assume that $\{V(s)\, |\, s\leq 0\}$ and $\{V(s)\, |\, s\geq 0\}$ are both unbounded. Then $V(s)\rightarrow\infty$ as $s\rightarrow\pm\infty$ due to
  Remark~\ref{remark:divergence to infinity V}, a fact we use tacitly in what follows. Fix an $s_0\in\ro$. Choose $s_a<s_0<s_b$ such that
  $V(s_a)>V(s_0)$ and $V(s_b)>V(s_0)$. Then $V$ has an interior minimum on $[s_a,s_b]$ attained at, say, $s_1$, so that $V'(s_1)=0$. Note that
  $\vartheta_0^2>3V(s_1)$ by the definition of $\vartheta_0$. Let
  \begin{align*}
    \sigma_- := & \inf\{s\leq s_1\, |\, 3V(u)\leq\vartheta_0^2\ \forall\, u\in [s,s_1]\},\\
    \sigma_+ := & \sup\{s\geq s_1\, |\, 3V(u)\leq\vartheta_0^2\ \forall\, u\in [s_1,s]\}.    
  \end{align*}
  Then $\sigma_\pm\in\ro$; $\sigma_-<s_1<\sigma_+$; $3V(s)\leq \vartheta_0^2$ on
  $[\sigma_-,\sigma_+]$; and $3V(s)=\vartheta_0^2$ for $s\in\{\sigma_-,\sigma_+\}$. Since $3V(\sigma_-)=\vartheta_0^2$, we cannot have
  $V'(\sigma_-)=0$. Moreover, by the definition of $\sigma_-$, we cannot have $V'(\sigma_-)>0$. This means that $V'(\sigma_-)<0$. Similarly,
  $V'(\sigma_+)>0$. Assume that there is a $\sigma_0<\sigma_-$ such that $3V(\sigma_0)\leq\vartheta_0^2$. Then $V$ must have a local maximum in
  $(\sigma_0,\sigma_-)$, say $\sigma_1$. Then $3V(\sigma_1)\geq 3V(\sigma_-)=\vartheta_0^2$ and $V'(\sigma_1)=0$. However, this contradicts the
  definition of $\vartheta_0$. Thus $3V(s)>\vartheta_0^2$ for all $s<\sigma_-$. Similarly, $3V(s)>\vartheta_0^2$ for all $s>\sigma_+$. Assume, in
  order to reach a contradiction, that there is an $s\in (\sigma_-,\sigma_+)$ such that $3V(s)=\vartheta_0^2$. Then $s$ is a local maximum of $V$
  so that $V'(s)=0$. This contradicts the definition of $\vartheta_0$. To conclude, $3V(s)<\vartheta_0^2$ for all $s\in (\sigma_-,\sigma_+)$. Define 
  \[
    B_\pm:=\{(\phi_0,\phi_1)\in\rn{2}\, |\, \vartheta_0^2=3\phi_1^2/2+3V(\phi_0),\ \pm\phi_1\geq 0\}.
  \]
  By the above observations,
  \[
    B_\pm=\left\{(\phi_0,\phi_1)\in\rn{2}\, |\, \phi_0\in [\sigma_-,\sigma_+],\ \phi_1=\pm\left(2\vartheta_0^2/3-2V(\phi_0)\right)^{1/2}\right\}.
  \]
  Moreover, the sets $B_\pm$ are each homeomorphic to $[\sigma_-,\sigma_+]$, by projection to the $\phi_0$-axis, and they only intersect in
  $(\sigma_\pm,0)$. This means that the union is homeomorphic to $\sn{1}$. Since the union is $B_{\mrI,+}^{\iso}(\vartheta_0)$ and $B_{\mrI,+}^{\iso}(\vartheta_0)$
  is a smooth $1$-dimensional submanifold of $\rn{2}$, we obtain the desired conclusion.
\end{proof}

\section{Parametrising data on the singularity}\label{ssection:par data on sing}

In analogy with $\sfB$, we next introduce $\sfS$ as a first step in the parametrisation of isometry classes of initial data on the singularity. 
\begin{definition}\label{def:mfS def}
  Let $\sfS$
  \index{$\a$Aa@Notation!Symmetry reduced sets of initial data on the singularity!sfS@$\sfS$}%
  be the set of $(o,P,\Phi_0,\Phi_1)\in M_3(\ro)\times M_3(\ro)\times\rn{2}$ such that
  \begin{equation}\label{eq:P Phi one cond on sing}
    \rotr P=1, \ \ \rotr P^2+\Phi_1^2=1;
  \end{equation}
  $o=\mathrm{diag}(o_1,o_2,o_3)$; $P=\mathrm{diag}(p_1,p_2,p_3)$; if $\{i,j,k\}=\{1,2,3\}$ and $p_i=1$, then $o_j=o_k$; and if $p_l<1$ for all $l$ and
  $o_i\neq 0$ for some $i$, then $p_i>0$. Next, let $\sfS_{\mrI}$ denote the subset of $\sfS$ such that all the $o_i$ vanish; let $\sfS_{\mrII}$ denote
  the subset of $\sfS$ such that one of
  the $o_i$ is non-zero and the other two vanish; let $\sfS_{\mrVIz}$ denote the subset of $\sfS$ such that two of the $o_i$ are non-zero and have
  opposite signs, and the remaining $o_i$ vanishes; let $\sfS_{\mrVIIz}$
  \index{$\a$Aa@Notation!Symmetry reduced sets of initial data on the singularity!sfST@$\sfS_\mfT$}%
  denote the subset of $\sfS$ such that two of the $o_i$ are non-zero and have
  the same signs, and the remaining $o_i$ vanishes; let $\sfS_{\mrVIII}$ denote the subset of $\sfS$ such that all the $o_i$ are non-zero, but they
  do not all have the same sign; and let $\sfS_{\mrIX}$ denote the subset of $\sfS$ such that all the $o_i$ are non-zero, and all have the same sign.
\end{definition}
Just as regular initial data give rise to elements of $\sfB$, see Lemma~\ref{lemma:x arising from mfI}, initial data on the singularity give rise to
elements of $\sfS$. 
\begin{lemma}\label{lemma:x arises from Iinf}
  Let $\mfI_\infty=(G,\msH,\msK,\Phi_{0},\Phi_{1})\in\mS$; see Definition~\ref{def:sets of id on singularity}. Then there is an orthonormal basis $\{e_i\}$
  of $\mfg$ with respect to $\msH$ such that if $o$ is the commutator matrix associated with $\{e_i\}$ and $P$ is the matrix with components
  $P_{ij}=\msH(\msK e_i,e_j)$, then $x:=(o,P,\Phi_0,\Phi_1)\in\sfS$. 
\end{lemma}
\begin{remark}\label{remark:x arising from mfIinf}
  If $x$ and $\mfI_\infty$ are related as in the statement of the lemma, we below say that $x$ arises from $\mfI_\infty$. If $P$ is related to $\mfI_\infty$
  and $\{e_i\}$ as in the statement of the lemma, we below say that $P$ is the Weingarten matrix associated with $\mfI_\infty$ and $\{e_i\}$. 
\end{remark}
\begin{remark}
  The converse statement, that elements of $\sfS$ arise from initial data on the singularity, can be found in Lemma~\ref{lemma:realcond sing}.
\end{remark}
\begin{proof}
  Let $\{e_i'\}$ be an orthonormal basis of $\mfg$ with respect to $\msH$; let $o'$ be the commutator matrix associated with $\{e_i'\}$; and let $P'$ be
  the matrix with components $P_{ij}'=\msH(\msK e_i',e_j')$. Then the condition $\rodiv_{\msH}\msK=0$ is equivalent to $o'$ and $P'$ commuting; see
  \cite[Lemma~19.13, p.~210]{RinCauchy}. Since $o'$ and $P'$ are also symmetric, this means that $o'$ and $P'$ can be simultaneously diagonalised
  by means of an orthogonal transformation. There is thus an orthonormal basis $\{e_i\}$ of $\mfg$ with respect to $\msH$ such that if $o$ is the
  corresponding commutator matrix and $P$ is the matrix with components $P_{ij}=\msH(\msK e_i,e_j)$, then $o$ and $P$ are diagonal. That
  $\tr\msK=1$ and $\tr\msK^2+\Phi_1^2=1$ hold implies that (\ref{eq:P Phi one cond on sing}) holds. Next, denote the diagonal components of $P$ by
  $p_i$ (note that the $p_i$ are the eigenvalues of $\msK$). There are two cases to consider. Either all the $p_i<1$ or there is one $p_i$ which equals $1$.

  Assume that one of the $p_i$, say $p_1$, equals $1$. Due to (\ref{eq:P Phi one cond on sing}), it follows that $p_2=p_3=0$. In particular,
  $\msK e_A=0$ for $A\in\{2,3\}$. Next, due to Definition~\ref{def:ndvacidonbbssh}, there is an orthonormal basis $\{\be_{i}\}$ of $\mfg$ with respect
  to $\msH$ such that $\msK \be_{1}=\be_{1}$ and such that if $\Psi_t$ is defined by (\ref{eq:Psit definition}) (with $e_i$ replaced by $\be_i$), then
  $\Psi_t$ is a Lie algebra isomorphism for all $t$. Clearly, $\be_1=\pm e_1$, so that $\{e_2,e_3\}$ and $\{\be_2,\be_3\}$ span the same plane. Thus
  $\msK\be_A=0$ for $A\in\{2,3\}$, and if
  $\bP$ is the matrix with components $\bP_{ij}=\msH(\msK \be_i,\be_j)$, then $\bP=\rodiag(1,0,0)$. Next, if $\bo$ is the commutator matrix associated
  with $\{\be_i\}$, then the existence of the family $\Psi_t$ of Lie algebra isomorphisms implies that $\bo=\rodiag(\bo_1,\bo_2,\bo_2)$. Thus
  $(\bo,\bP,\Phi_0,\Phi_1)\in\sfS$. 

  Next, assume that all the $p_l<1$. Then, since $\rotr P=1$, it follows that if $p_i\leq 0$ and $\{i,j,k\}=\{1,2,3\}$, then $p_j,p_k>0$. This means,
  since the $p_l$ are eigenvalues of $\msK$, that the eigenspace of $\msK$ corresponding to $p_i$ is perpendicular to the span of the eigenspaces
  corresponding to $\{p_j,p_k\}$, say $\mfh$. Due to Definition~\ref{def:ndvacidonbbssh} it follows that $\mfh$ is a subalgebra, which means that
  $[e_j,e_k]\in\mfh$. In particular, $[e_j,e_k]\perp e_i$. On the other hand $[e_j,e_k]=\pm o_ie_i$ (no summation). This means that $o_i=0$. Thus
  $(o,P,\Phi_0,\Phi_1)\in\sfS$. 
\end{proof}
Define $\sigma_\pm$ so that
\begin{equation}\label{eq:pi ito sp sm}
  p_1=\tfrac{1}{3}-\tfrac{2}{3}\sigma_+,\ \ \
  p_2=\tfrac{1}{3}+\tfrac{1}{3}\sigma_++\tfrac{1}{\sqrt{3}}\sigma_-,\ \ \
  p_3=\tfrac{1}{3}+\tfrac{1}{3}\sigma_+-\tfrac{1}{\sqrt{3}}\sigma_-;
\end{equation}
i.e., let $\sigma_+:=(1-3p_1)/2$ and $\sigma_-=\sqrt{3}(p_2-p_3)/2$. Then the first equality in (\ref{eq:P Phi one cond on sing}) is automatically
satisfied, and the second reads
\begin{equation}\label{eq:mfE def}
  \sigma_+^2+\sigma_-^2+\tfrac{3}{2}\Phi_1^2=1.
\end{equation}
In particular, (\ref{eq:P Phi one cond on sing}) defines a smooth manifold. Next, note that the maps introduced in (\ref{eq:psi sigma pm})
leave $\sfS_\mfT$ invariant. We can therefore also think of $\psi_\sigma^+$ as a map from $\sfS$ to $\sfS$ and $\psi_\sigma^-$ as a map from
$\sfS-\sfS_\mrI$ to $\sfS-\sfS_\mrI$. Strictly speaking, we should introduce different notation for these maps. However, for the sake of brevity,
we use the same notation and assume the domain of definition to be clear from the context. We also use the notation introduced in
Definition~\ref{def:Gamma plus etc}. As before, $\psi_{\roId}^+$ is the identity and $\psi_{\roId}^-$ has no fixed points.
\begin{lemma}\label{lemma:psi sigma plus isotropic sing}
  Assume $\sigma\in S_3$ to have no fixed points (i.e., $\sigma$ is even and non-trivial). Then $\psi_\sigma^-$, considered as a map from
  $\sfS-\sfS_\mrI$ to $\sfS-\sfS_\mrI$, has no fixed points. Moreover, if
  \[
    x=(o_1,o_2,o_3,p_1,p_2,p_3,\Phi_0,\Phi_1)\in\sfS
  \]
  and $\psi_\sigma^+(x)=x$,
  then all the $o_i$ equal and all the $p_i$ equal $1/3$, so that $\psi_\lambda^+(x)=x$ for all $\lambda\in S_3$. In particular, if $x\in\sfS_\mfT$
  arises from $\mfI_\infty\in\mS_\mfT$, cf. Remark~\ref{remark:x arising from mfIinf}, and is a fixed point of $\psi_\sigma^+$, then $\mfI_\infty$ are
  isotropic. Finally, let $\mfI_\infty\in\mS_{\mfT}^{\iso}$ and assume that $x\in\sfS$ arises from $\mfI_\infty$, cf.
  Remark~\ref{remark:x arising from mfIinf}. Then $\psi_\lambda^+(x)=x$ for all $\lambda\in S_3$.
\end{lemma}
\begin{proof}
  The proof is almost identical to the proof of Lemma~\ref{lemma:psi sigma plus isotropic}. The details are left to the reader.
\end{proof}

\begin{lemma}\label{lemma:sfB char per symm VIz sing}
  Let $\sigma\in S_3$ be a transposition and $x\in\sfS-\sfS_\mrI$ be a fixed point of $\psi_\sigma^-$. Then $x\in\sfS_{\mrVIz}$.
  In fact, there are $\{i,j,l\}=\{1,2,3\}$ such that if $x=(o_1,o_2,o_3,p_1,p_2,p_3,\Phi_0,\Phi_1)$, then $o_i=0$, $o_j=-o_l\neq 0$
  and $p_j=p_l$. In particular, if $x\in\sfS_{\mrVIz}$ arises from $\mfI_\infty\in\mS_{\mrVIz}$, cf.
  Remark~\ref{remark:x arising from mfIinf}, and is a fixed point of $\psi_\sigma^-$, then $\mfI_\infty$ are permutation symmetric. Finally, let
  $\mfI_\infty\in\mS_{\mrVIz}^{\roper}$. Then there is an $x\in\sfS_{\mrVIz}$ arising from $\mfI_\infty$, cf. Remark~\ref{remark:x arising from mfIinf},
  and a transposition $\sigma$ such that $\psi_\sigma^-(x)=x$.
\end{lemma}
\begin{proof}
  All the statements but the last one follow by arguments that are almost identical to the corresponding arguments in the proof
  of Lemma~\ref{lemma:sfB char per symm VIz}. In order to justify the last statement, let
  $\mfI_\infty=(G,\msH,\msK,\Phi_0,\Phi_1)\in\mS_{\mrVIz}^{\roper}$. We then wish to prove that there is an $x$ as in the last statement of the lemma.
  Let, to this end, $\Psi$ be a Lie algebra isomorphism and $\{e_{i}\}$ be an orthonormal basis of $\mfg$ with respect to $\msH$ such that
  (\ref{eq:Psi definition}) and $\msK\Psi=\Psi\msK$ hold; cf. Definition~\ref{def:isotropic}. Let $o$ be the commutator matrix associated with
  $\{e_i\}$ and let $P$ be the Weingarten matrix associated with $\mfI_\infty$ and $\{e_i\}$. That $o$ and $P$ commute follows from
  $\rodiv_{\msH}\msK=0$; cf. the proof of Lemma~\ref{lemma:x arises from Iinf}. Given this information, the proof of the last statement of the
  lemma is essentially identical to the
  proof of the corresponding statement in Lemma~\ref{lemma:sfB char per symm VIz}. Let us, however, note that $1$ cannot be an eigenvalue of
  $\msK$; this is not compatible with Bianchi type $\mrVIz$ due to the arguments presented in the proof of Lemma~\ref{lemma:x arises from Iinf}.
  Moreover, the final condition for membership in $\sfS$ is fulfilled due to the last part of the proof of Lemma~\ref{lemma:x arises from Iinf}.
\end{proof}
\begin{lemma}\label{lemma:sfB char per symm VIIz sing}
  Let $\sigma\in S_3$ be a transposition and $x\in\sfS$ be a fixed point of $\psi_\sigma^+$. Then there are
  $\{i,j,l\}=\{1,2,3\}$ such that if $x=(o_1,o_2,o_3,p_1,p_2,p_3,\Phi_0,\Phi_1)$, then $o_j=o_l$ and $p_j=p_l$. In particular, if
  $x\in\sfS_{\mfT}$ arises from $\mfI_\infty\in\mS_{\mfT}$, cf. Remark~\ref{remark:x arising from mfIinf}, is a fixed point of $\psi_\sigma^+$, but is not
  a fixed point of all the maps in $\Gamma^+$, then $\mfI_\infty$ are LRS. Finally, let $\mfI\in\mS_{\mfT}^{\roLRS}$. Then there is an
  $x\in\sfS_{\mfT}$ arising from $\mfI_\infty$, cf. Remark~\ref{remark:x arising from mfIinf}, which is not a fixed point of all the maps in $\Gamma^+$,
  and a transposition $\sigma$ such that $\psi_\sigma^+(x)=x$.
\end{lemma}
\begin{remark}
  An observation analogous to Remark~\ref{remark:sigma but not lambda} holds in this case as well.
\end{remark}
\begin{proof}
  The proof is almost identical to the proof of Lemma~\ref{lemma:sfB char per symm VIIz}. The details are left to the reader.
\end{proof}
In analogy with Definition~\ref{def:BTs GTs}, we introduce the following terminology. 
\begin{definition}\label{def:STs GTs}
  Let $\mfT$ be a Bianchi class A type and recall the notation introduced in Definition~\ref{def:Gamma plus etc}. Define $\sfS_\mfT^{\iso}$ to be the
  subset of $\sfS_\mfT$ invariant under $\Gamma^+$. Define $\sfS_{\mrI}^{\roper}=\varnothing$ and define, for $\mfT\neq\mrI$, $\sfS_{\mfT}^{\roper}$ to
  be the set of $p\in\sfS_{\mfT}$ such that there is a map $\psi_\sigma^-\in \Gamma^{-,\roodd}$ with $\psi_\sigma^-(p)=p$. Define $\sfS_{\mfT}^{\roLRS}$ to be the
  set of $p\in\sfS_{\mfT}-\sfS_{\mfT}^{\iso}$ such that there is a map $\psi_\sigma^+\in \Gamma^{+,\roodd}$ with $\psi_\sigma^+(p)=p$. Define
  \[
  \sfS_{\mfT}^{\rogen}:=\sfS_{\mfT}-\sfS_{\mfT}^{\iso}-\sfS_{\mfT}^{\roper}-\sfS_{\mfT}^{\roLRS}.
  \]
  \index{$\a$Aa@Notation!Symmetry reduced sets of initial data on the singularity!sfSTs@$\sfS_\mfT^\mfs$}%
\end{definition}
\begin{remark}
  Note that $\sfS_\mfT^{\iso}$ is empty unless $\mfT\in \{\mrI,\mrIX\}$; $\sfS_\mfT^{\roper}$ is empty unless $\mfT=\mrVIz$;
  and $\sfS_{\mrVIz}^{\roLRS}$ is empty.
\end{remark}
In analogy with Lemma~\ref{lemma:sfR mfT mfs}, the following statement holds. 
\begin{lemma}\label{lemma:sing mfT mfs}
  Let $\mfT$ be a Bianchi class A type and $\mfs\in\{\iso,\roper,\roLRS,\rogen\}$ be such that $\sfS_{\mfT}^{\mfs}$ and $\Gamma_{\mfT}^{\mfs}$,
  see Definitions~\ref{def:STs GTs} and \ref{def:BTs GTs}, are defined and non-empty. Then $\sfS_{\mfT}^{\mfs}$ is a smooth manifold,
  $\Gamma_{\mfT}^{\mfs}$ acts freely and properly discontinuously on $\sfS_{\mfT}^{\mfs}$ and the corresponding quotient is denoted $\sfU_{\mfT}^{\mfs}$.
  \index{$\a$Aa@Notation!Symmetry reduced sets of initial data on the singularity!sfUTs@$\sfU_\mfT^\mfs$}%
\end{lemma}
\begin{proof}
  If $\mfs=\iso$, then the second last condition in the list of requirements for membership in $\sfS$ is void and the last is automatically
  satisfied. Moreover, $P=\roId/3$, $\Phi_0$ is an arbitrary element of $\ro$ and $\Phi_1\in\{-(2/3)^{1/2},(2/3)^{1/2}\}$. If $\mfT=\mrI$,
  then $o=0$ and if $\mfT=\mrIX$, then $o=\nu_0\roId$ for some $0\neq\nu_0\in\ro$. Thus $\sfS_\mfT^{\iso}$ are smooth manifolds. Moreover, the
  groups $\Gamma_\mfT^{\iso}$ act freely and properly discontinuosly on them.

  Next, let $\mfs=\roper$. Then $\mfT=\mrVIz$ and $\sfS_{\mrVIz}^{\roper}$ is the union of three sets, corresponding to which $o_i$ vanishes, such
  that the closure of each of these sets has empty intersection with the other two. It is sufficient to focus on the set corresponding to
  $o_1=0$. If $x$ belongs to this set, then it is of the form
  $x=(0,o_2,-o_2,p_1,p_2,p_2,\Phi_0,\Phi_1)$ with $o_2\neq 0$. This means that (\ref{eq:mfE def}) takes the form $\sigma_+^2+3\Phi_1^2/2=1$
  (note that $p_1$ and $p_2$ are parametrised by (\ref{eq:pi ito sp sm}) in which $\sigma_-$ has been set to zero). There is no restriction
  on $(o_2,\Phi_0)$ other than $o_2\neq 0$. Note that $p_i$ cannot equal $1$ for any $i$, since the definition of $\sfS$ would then require
  $o_j=o_k$, where $\{i,j,k\}=\{1,2,3\}$, which is incompatible with $\mfT=\mrVIz$. This means that the second last condition in the list of
  requirements for membership in $\sfS$ is void. The final condition takes the form $p_2>0$; i.e., $\sigma_+>-1$. In particular,
  $\sfS_{\mrVIz}^{\roper}$ is a smooth manifold and $\Gamma_{\mrVIz}^{\roper}$ acts freely and properly discontinuously on this manifold.

  Next, let $\mfs=\roLRS$. In this setting, each $\sfS_{\mfT}^{\roLRS}$ consists of three disjoint sets, corresponding to a choice of $i$
  such that if $\{i,j,k\}=\{1,2,3\}$ then the corresponding elements satisfy $p_j=p_k$ and $o_j=o_k$; note that if $p_j=p_k$ and $o_j=o_k$,
  we cannot have both $p_i=p_j$ and $o_i=o_j$ since the corresponding elements are invariant under $\Gamma^+$. In fact, the closure of each
  of these sets has empty intersection with the other two. It is thus sufficient to focus on the case that $i=1$. If $x$ belongs to the
  corresponding set, then $x$ is of the form $x=(o_1,o_2,o_2,p_1,p_2,p_2,\Phi_0,\Phi_1)$. Whether $o_1$ or $o_2$ are zero or non-zero, and
  whether $o_1$ has the same sign as $o_2$ or the opposite sign (in case both are non-zero) depends on the symmetry type. However, the
  corresponding conditions are immaterial to the question at hand: is $\sfS_{\mfT}^{\roLRS}$ a smooth manifold? The condition that $x$ be LRS
  but not isotropic corresponds to the condition that $(o_1-o_2,p_1-p_2)\neq (0,0)$, an open condition. Next, (\ref{eq:mfE def}) takes the
  form $\sigma_+^2+3\Phi_1^2/2=1$ and there is no restriction on $\Phi_0$. If there is a $p_i$ which equals $1$, then it has to be $p_1$
  (and this corresponds to $\sigma_+=-1$). In this case, we automatically have $o_2=o_3$. In other words, the second last condition in the list of
  requirements for membership in $\sfS$ is automatically satisfied. The form of the final condition depends on the Bianchi type, but it in practice
  corresponds to an open condition on $\sigma_+$. In particular, $\sfS_{\mrVIIz}^{\roLRS}$ is a smooth manifold and $\Gamma_{\mrVIIz}^{\roLRS}$
  acts freely and properly discontinuously on this manifold.

  Finally, assume that $\mfs=\rogen$. The signs and vaninshing/non-vanishing of the $o_i$ depend on the symmetry type and are immaterial to the question
  at hand, as in the LRS setting. Removing isotropic, LRS and permutation symmetric elements corresponds to removing closed subsets from $\sfS_\mfT$; the
  set of fixed points of a finite number of continuous maps is clearly closed. There is no restriction on $\Phi_0$. Clearly, (\ref{eq:mfE def}) defines
  a smooth manifold. Next, note that if $p_i=1$ and $\{i,j,k\}=\{1,2,3\}$, then $p_j=p_k=0$
  due to (\ref{eq:P Phi one cond on sing}) and $o_j=o_k$ due to the definition of $\sfS$. This means that the corresponding elements in $\sfS$ are LRS,
  something we exclude here. This means that the second last condition in the list of requirements for membership in $\sfS$ is void. The final condition in
  the definition of $\sfS$ corresponds to an open condition on $\sigma_\pm$ depending on the Bianchi type. To conclude, $\sfS_\mfT^{\rogen}$ is a smooth
  manifold and $\Gamma_\mfT^{\rogen}$ acts freely and properly discontinuously on this manifold.
\end{proof}

Finally, we are in a position to parametrise ${}^{\rosc}\mfS_\mfT^\mfs$ by smooth manifolds.
\begin{lemma}\label{lemma:sc mfB mfT mfs param sing}
  There is a bijection from ${}^{\rosc}\mfS_\mfT^\mfs$ to $\sfU_{\mfT}^{\mfs}$ for the following combinations of
  $\mfT$ and $\mfs$: if $\mfs=\iso$ and $\mfT\in\{\mrI,\mrIX\}$; if $\mfs=\roLRS$ and $\mfT\in \{\mrI,\mrII,\mrVIII,\mrIX\}$; if
  $(\mfT,\mfs)=(\mrVIz,\roper)$; and if $\mfs=\rogen$ and $\mfT$ is an arbitrary Bianchi class A type.
\end{lemma}
\begin{remark}\label{remark:sc mfB mfT mfs param sing}
  Combining the bijection obtained in the lemma with Lemma~\ref{lemma:sing mfT mfs} yields a smooth structure on ${}^{\rosc}\mfS_\mfT^\mfs$
  for the relevant combinations of $\mfT$ and $\mfs$. In what follows, we assume ${}^{\rosc}\mfS_\mfT^\mfs$ to have been endowed with this
  smooth structure. 
\end{remark}
\begin{proof}
  Let $(\mfT,\mfs)$ be one of the combinations mentioned in the statement of the lemma and $\mfI_\infty\in{}^{\rosc}\mS_\mfT^\mfs$. Then there is an
  $x\in\sfS_{\mfT}$ arising from $\mfI_\infty$; cf. Lemma~\ref{lemma:x arises from Iinf} and Remark~\ref{remark:x arising from mfIinf}. If
  $\mfs=\iso$, then $x$ is invariant under $\Gamma^+$ due to Lemma~\ref{lemma:psi sigma plus isotropic sing}. This means that $x\in\sfS_\mfT^{\iso}$; cf.
  Definition~\ref{def:STs GTs}. Next, if $\mfs=\roper$, then we can assume $x$ to belong to
  $\sfS_{\mrVIz}^{\roper}$ due to Lemma~\ref{lemma:sfB char per symm VIz sing} and Definition~\ref{def:STs GTs}. If $\mfs=\roLRS$, then we can
  assume $x$ to be such that it is not invariant under all the maps in $\Gamma^+$, but such that there is a transposition $\sigma$ such that
  $\psi_\sigma^+(x)=x$; see Lemma~\ref{lemma:sfB char per symm VIIz sing}. This means that $x\in\sfS_{\mfT}^{\roLRS}$; see Definition~\ref{def:STs GTs}.
  Finally, let $\mfs=\rogen$. We then wish to prove that $x\in \sfS_\mfT^\mfs$. Assume, to this end, that $x\in\sfS_\mfT^{\mfr}$ for some
  $\mfr\neq\rogen$. If $\mfr=\iso$, then $\mfs=\iso$ due to Lemma~\ref{lemma:psi sigma plus isotropic sing}, a contradiction. If $\mfr=\roper$, then
  $\mfs=\roper$ due to Lemma~\ref{lemma:sfB char per symm VIz sing}, a contradiction. If $\mfr=\roLRS$, then $\mfs=\roLRS$ due to
  Lemma~\ref{lemma:sfB char per symm VIIz sing}, a contradiction. To conclude: we can, with each $\mfI\in{}^{\rosc}\mS_\mfT^\mfs$ associate an
  $x\in \sfS_\mfT^\mfs$, which, in its turn, gives rise to an element, say $[x]$, of $\sfU_\mfT^\mfs$; see Lemma~\ref{lemma:sing mfT mfs}.
  If $x_i\in\sfS_\mfT^\mfs$, $i=1,2$, are both associated with $\mfI_\infty$ in this way, then $[x_1]=[x_2]$ due to
  Lemma~\ref{lemma:isometryinducesequal o P}. In other words, there is a well defined map from ${}^{\rosc}\mS_\mfT^\mfs$ to $\sfU_\mfT^\mfs$. Moreover,
  again due to Lemma~\ref{lemma:isometryinducesequal o P}, this map descends to a map from ${}^{\rosc}\mfS_\mfT^\mfs$ to $\sfU_\mfT^\mfs$. Next, if two
  initial data sets, say $\mfI_{\infty,0}$ and $\mfI_{\infty,1}$, are mapped to the same element, then $\mfI_{\infty,0}$ and $\mfI_{\infty,1}$ are isometric;
  see Lemma~\ref{lemma:nuKsameisometric sing}. Thus the map from ${}^{\rosc}\mfS_\mfT^\mfs$ to $\sfU_\mfT^\mfs$ is well defined and injective. What remains
  is to prove surjectivity.

  In order to prove surjectivity, let $[x]\in\sfU_\mfT^\mfs$ for some $x\in\sfS_\mfT^\mfs$. Due to Lemma~\ref{lemma:realcond sing}, there are
  simply connected initial data on the singularity $\mfI_\infty=(G,\msH,\msK,\Phi_0,\Phi_1)$ and a basis $\{e_i\}$ of $\mfg$, orthonormal with respect to
  $\msH$, such that if $o$ and $P$ are the commutator and Weingarten matrices associated with $\{e_i\}$ and $\mfI_\infty$, then $x=(o,P,\Phi_0,\Phi_1)$.
  We need to prove that $\mfI_\infty\in{}^{\rosc}\mS_\mfT^\mfs$. All that remains to be demonstrated is that $\mfI_\infty$ are of symmetry type $\mfs$ and
  that $\mfI_\infty$ are neither of isotropic nor of LRS Bianchi type $\mrVIIz$; see Definition~\ref{def:sets of id on singularity}. Assume, for the moment,
  that $\mfI_\infty$ are neither of isotropic nor of LRS Bianchi type $\mrVIIz$. Then we know that $\mfI_\infty\in{}^{\rosc}\mS_\mfT$. This means that
  $\mfI_\infty\in{}^{\rosc}\mS_\mfT^{\mfr}$ for some symmetry type $\mfr$. By the arguments presented above, this means that there is a basis $\{e_i'\}$ of
  $\mfg$, orthonormal with respect to $\msH$, such that if $o'$ and $P'$ are the commutator and Weingarten matrices associated with $\{e_i'\}$ and
  $\mfI_\infty$, then $x'=(o',P',\Phi_0,\Phi_1)\in\sfS_{\mfT}^{\mfr}$. Combining these observations with Lemma~\ref{lemma:isometryinducesequal o P} and the
  fact that $\sfS_{\mfT}^{\mfr}$ and $\sfS_{\mfT}^{\mfs}$ are disjoint if
  $\mfs\neq\mfr$, it follows that $\mfr=\mfs$. In other words, $\mfI_\infty\in{}^{\rosc}\mS_\mfT^\mfs$. This proves surjectivity, with one caveat: we need to
  exclude the possibility that $\mfI_\infty$ are isotropic or LRS Bianchi type $\mrVIIz$. Assume, to this end, that they are. Then the Ricci tensor of $\msH$
  vanishes identically; see Lemma~\ref{lemma:LRS VIIz is I sing}. This means that $o$ has to be such that $o_i=0$ and $o_j=o_l$ for some $i,j,l$ such
  that $\{i,j,l\}=\{1,2,3\}$; see (\ref{eq:Riccidiagitocm iso}). If $p_j=p_l$, then $x\in\sfS_{\mrVIIz}^{\mfs}$ for $\mfs\in\{\iso,\roLRS\}$, in
  contradiction with the assumptions. Thus $p_j\neq p_l$. In particular, the initial data are not isotropic. If $p_i\notin\{p_j,p_l\}$, the initial data
  cannot be LRS either. Assume therefore, without loss of generality, that $p_2=p_3\neq p_1$, that $o_2=0$ and that $o_3=o_1$. In order for $\mfI_\infty$
  to be LRS, there must be an orthonormal basis $\{e_m'\}$ of $\mfg$ with respect to $\msH$ and a family $\Psi_t$ of Lie algebra isomorphisms such that
  the conditions of Definition~\ref{def:isotropic} hold with $\{e_i\}$ replaced by $\{e_i'\}$. Let $P'$ and $o'$ be the Weingarten and commutator matrices
  associated with $\{e_i'\}$ and $\mfI_\infty$. Then, due to the the conditions of Definition~\ref{def:isotropic}, $P'$ and $o'$ are diagonal. Moreover, if
  $p_i'$ and $o_i'$ denote the diagonal components of $P'$ and $o'$ respectively, then $o_2'=o_3'$ and $p_2'=p_3'$. This means, in particular, that
  $p_A=p_A'$ for $A\in\{2,3\}$ and that $p_1=p_1'$. Thus $e_1'=\pm e_1$. Since the initial data on the singularity are of Bianchi type $\mrVIIz$ and
  $o_2'=o_3'$, we also have to have $o_2'\neq 0$ and $o_1'=0$. Since the plane spanned by $e_2'$ and $e_3'$ has to coincide with the plane spanned by
  $e_2$ and $e_3$ (they are both the orthogonal complement of $e_1'=\pm e_1$), there is also a $t_0$ such that $\Psi_{t_0}e_2'=e_2$ and
  $\Psi_{t_0}e_3'=\pm e_3$. This means that
  \[
  0=[e_2',\pm e_3']=[\Psi_{t_0}e_2',\pm\Psi_{t_0}e_3']=[e_2,e_3]=o_1 e_1.
  \]
  Thus $o_1=0$. Since $o_2=0$ and $o_3=o_1$, we conclude that the initial data are of Bianchi type I, in contradiction with the assumptions.
  To conclude, $\mfI_\infty$ are neither of isotropic nor of LRS Bianchi type $\mrVIIz$. This completes the proof of surjectivity. 
\end{proof}

\chapter{Constructing a development}\label{section:constructingadevelopment}

Our next goal is to construct developments as in Definition~\ref{def:BianchiAdevelopment} corresponding to regular initial data $\mfI\in\mB[V]$;
i.e., to prove Proposition~\ref{prop:unique max dev}. As a starting point, we derive the relevant equations. 

\section{Equations, Bianchi class A development}\label{ssection:EquationsBclassAdev}

The construction of a development, given initial data, is based on solving (\ref{seq:EFEwrtvar})--(\ref{eq:phiddot}) below. Next, we explain how these
equations arise by expressing the Einstein non-linear scalar field equations for a Bianchi class A development as in
Definition~\ref{def:BianchiAdevelopment}. 

\begin{lemma}\label{lemma:FW frame eqns}
  Let $G$ be a connected $3$-dimensional unimodular Lie group, $I$ be an open interval and $g$ be a smooth metric on $M:=G\times I$ of the form
  (\ref{eq:gSH}), where $\bge$ is a smooth family of left invariant Riemannian metrics on $G$. Let $t_0\in I$ and let $\{e_i'\}$
  be an orthonormal basis of $\mfg$ with respect to $\bge(t_0)$. Let $e_0:=\d_t$. Define $e_i$ by the requirement that $e_i|_{t_0}=e_i'$ and
  $\nabla_{e_0}e_i=0$, where $\nabla$ is the Levi-Civita connection of $g$. Then $\{e_\a\}$ is a smooth global orthonormal frame for $(M,g)$, $e_i\perp e_0$
  for all $i\in\{1,2,3\}$ and $e_{i}|_t\in\mfg$ for all $t\in I$. Let $\ldr{\cdot,\cdot}:=g$, $\theta_{ij}:=\ldr{\nabla_{e_i}e_0,e_j}$,
  $\theta:=\de^{ij}\theta_{ij}$ and $\sigma_{ij}:=\theta_{ij}-\theta\delta_{ij}/3$. Let $n^{ij}$ be the components of the commutator matrix associated with
  $\{e_i\}$. Then $\theta_{ij}$ and $n^{ij}$ are independent of the spatial coordinate.

  Assume that $V\in C^{\infty}(\ro)$ and that $\phi\in C^{\infty}(I)$. Then $(M,g,\phi)$ satisfies the Einstein equations if and only if the following
  equations are satisfied:
  \begin{subequations}\label{seq:EFEwrtvar}
    \begin{align}
      \e_{mjl}n^{li}\theta_{ij} = & 0,\label{eq:Rzm eq zero}\\
      \d_t\sigma_{lm} = & -\theta\sigma_{lm}-s_{lm},\label{eq:silmd}\\
      \theta_t = & -\tfrac{3}{2}\sigma_{ij}\sigma^{ij}-\tfrac{1}{2}n_{ij}n^{ij}
      +\tfrac{1}{4}(\mathrm{tr}n)^{2}-\tfrac{3}{2}\phi_{t}^{2},\label{eq:thetad}\\
      \theta^{2} = & \tfrac{3}{2}\sigma_{ij}\sigma^{ij}+\tfrac{3}{2}n_{ij}n^{ij}-\tfrac{3}{4}(\mathrm{tr}n)^{2}
      +\tfrac{3}{2}\phi_{t}^{2}+3V\circ \phi,\label{eq:hamconfin}
    \end{align}
  \end{subequations}
  where indices are raised and lowered with the Kronecker delta and
  \begin{subequations}\label{seq:slmblmdef}
    \begin{align}
      b_{lm}  = & 2n_{m}^{\phantom{m}i}n_{il}-(\mathrm{tr}n)n_{lm},\label{eq:blmde}\\
      s_{lm}  = & b_{lm}-\tfrac{1}{3}(\mathrm{tr}b)\de_{lm}.\label{eq:slmde}  
    \end{align}
  \end{subequations}
  Moreover, the Jacobi identities yield
  \begin{equation}\label{eq:nijd}
    \d_t n_{ij} =  2\sigma^{k}_{\phantom{k}(i}n_{j)k}-\tfrac{1}{3}\theta n_{ij},
  \end{equation}
  where the parenthesis denotes symmetrization. Finally, the evolution equation for the scalar field can be written
  \begin{equation}\label{eq:phiddot}
    \phi_{tt}  = -\theta\phi_{t}-V'\circ\phi.
  \end{equation}

  Assume (\ref{seq:EFEwrtvar})--(\ref{eq:phiddot}) to be satisfied. Due to (\ref{eq:Rzm eq zero}), it follows that the $e_i'$ can be chosen so that
  $n^{ij}(t_0)=0$ for $i\neq j$ and $\theta_{ij}(t_0)=0$ for $i\neq j$.
  Given that $\{e_i'\}$ has been chosen in this way, there are smooth functions $a_i:I\rightarrow (0,\infty)$ such that
  $e_i=a_i^{-1}e_i'$. Moreover, $n^{ij}(t)=0$ for $i\neq j$ and $t\in I$; and $\theta_{ij}(t)=0$ for $i\neq j$ and $t\in I$.
\end{lemma}
\begin{remark}
  In what follows, we always take for granted that $e_i'$ has been chosen so that $n^{ij}(t_0)=0$ for $i\neq j$ and $\theta_{ij}(t_0)=0$ for $i\neq j$.
  Denoting the diagonal components of $n(t_0)$ by $\bn_i$, the matrix $n$ with components $n^{ij}$ satisfies $n=\mathrm{diag}(n_{1},n_{2},n_{3})$, where
  \begin{equation}\label{eq:chco}
    n_{1}=\frac{a_{1}}{a_{2}a_{3}}\bn_{1},\ \
    n_{2}=\frac{a_{2}}{a_{1}a_{3}}\bn_{2},\ \
    n_{3}=\frac{a_{3}}{a_{1}a_{2}}\bn_{3}.
  \end{equation}
\end{remark}
\begin{proof}
  Define $e_0':=\d_t$. Writing $e_i=e_i^\a e_\a'$, the equations for the $e_i$ can be written as system of linear ODE with smooth coefficients (depending
  only on time) for $e_i^\a$. This guarantees the existence of the smooth vector fields $e_i$. Due to the equation, the form of the metric, the fact that
  $e_0$ is a unit vector field and the fact that $e_i^0|_{t_0}=0$, it follows that $e_i^0=0$ for all $t$. Since $e_0\ldr{e_i,e_j}=0$, it is also clear that
  $\{e_i\}$ is an orthonormal frame for $\mfg$ with respect to $\bge(t)$. To conclude, $\{e_\a\}$ is a smooth global orthonormal frame for $(M,g)$ and
  $e_{i}|_t\in\mfg$ for all $t\in I$. Define $\Gamma^{\a}_{\b\de}$ by $\nabla_{e_\b}e_{\de}=\Gamma^{\a}_{\b\de}e_{\a}$. Due to the fact that $\nabla_{e_0}e_i=0$ it
  can then, using arguments similar to those presented in
  \cite[Subsection~B.2.1, p.~279]{RinCauchy}, be verified that $\Gamma^0_{00}=0$ and $\Gamma_{0i}^0=\Gamma_{i0}^0=\Gamma^{i}_{00}=0$. Moreover,
  $\Gamma^{i}_{0l}=0$, $\Gamma_{ij}^{0}=\theta_{ij}$ and $\Gamma^{i}_{j0}=\theta_{ij}$. Finally, $\Gamma^i_{jk}$ are the connection
  coefficients associated with the metric induced on the constant-$t$ hypersurfaces. Again, due to the equation for the $e_i$, the conclusions of
  \cite[Subsection~B.2.2, p.~280]{RinCauchy} and \cite[Subsection~B.2.3, pp.~280--281]{RinCauchy} hold. This means that
  \begin{subequations}
    \begin{align}
      \roRic(e_0,e_0) = & -\theta_t-\theta^{ij}\theta_{ij},\label{eq:riczz}\\
      \roRic(e_0,e_m) = & \e_{mjl}n^{li}\theta_{ij},\label{eq:riczm}\\
      \roRic(e_l,e_m) = & \d_t\theta_{lm}+\theta\theta_{lm}+\ovRic(e_l,e_m).\label{eq:riclm}
    \end{align}
  \end{subequations}
  On the other hand, due to \cite[Lemma~19.11, p.~209]{RinCauchy},
  \begin{equation}\label{eq:ovRic lm}
    \ovRic(e_l,e_m)=2n_{m}^{\phantom{m}i}n_{il}-n^{ij}n_{ij}\de_{lm}+\tfrac{1}{2}(\mathrm{tr}n)^{2}\de_{lm}-(\mathrm{tr}n)n_{lm}.
  \end{equation}  
  In the present setting, the Einstein equations can be written
  \begin{equation}\label{eq:eineqshba}
    R_{\a\b}-\nabla_{\a}\phi\nabla_{\b}\phi-V\circ\phi\cdot g_{\a\b}=0.
  \end{equation}
  Since $\phi$ is assumed to be independent of the spatial variables, it follows from (\ref{eq:riczm}) that the $0m$-components
  of (\ref{eq:eineqshba}) are equivalent to (\ref{eq:Rzm eq zero}). Due to (\ref{eq:riczz}), the $00$-component of (\ref{eq:eineqshba})
  is equivalent to
  \begin{equation}\label{eq:Riczzeq}
    \theta_t+\sigma_{ij}\sigma^{ij}+\tfrac{1}{3}\theta^{2}+\phi_{t}^{2}-V\circ\phi=0.
  \end{equation}
  Due to (\ref{eq:riclm}) and (\ref{eq:ovRic lm}), the $lm$ components of (\ref{eq:eineqshba}) are given by
  \begin{equation}\label{eq:Riclmeq}
    \d_t \theta_{lm}+\theta\theta_{lm}+2n_{m}^{\phantom{m}i}n_{il}
    -n^{ij}n_{ij}\de_{lm}+\tfrac{1}{2}(\mathrm{tr}n)^{2}\de_{lm}-(\mathrm{tr}n)n_{lm}-V\circ\phi\cdot\de_{lm}=0.
  \end{equation}
  It is convenient to divide the equation into its trace free parts and its trace. The trace free part is given by (\ref{eq:silmd}). 
  Taking the trace of (\ref{eq:Riclmeq}) yields
  \begin{equation}\label{eq:tracedottheij}
    \theta_t+\theta^{2}-n_{ij}n^{ij}+\tfrac{1}{2}(\mathrm{tr}n)^{2}-3V\circ\phi=0.
  \end{equation}
  Subtracting a third times (\ref{eq:tracedottheij}) from (\ref{eq:Riczzeq}) yields
  \[
  \tfrac{2}{3}\theta_t+\sigma_{ij}\sigma^{ij}+\tfrac{1}{3}n_{ij}n^{ij}
  -\tfrac{1}{6}(\mathrm{tr}n)^{2}+\phi_{t}^{2}=0,
  \]
  so that (\ref{eq:thetad}) holds. Finally, subtracting (\ref{eq:tracedottheij}) from (\ref{eq:Riczzeq}) yields the Hamiltonian constraint
  (\ref{eq:hamconfin}). It remains to be demonstrated is that if (\ref{seq:EFEwrtvar}) holds, then Einstein's equations are satisfied. We
  already know that the $0m$-components of (\ref{eq:eineqshba}) as well as the trace free part of the $lm$-equations are satisfied. That the
  $00$-equations and the trace of the $lm$-equations are satisfied then follows from (\ref{eq:thetad}) and (\ref{eq:hamconfin}).

  Next, due to the arguments presented in \cite[Subsection~B.3, pp.~281--282]{RinCauchy}, the Jacobi identities imply that 
  \begin{equation}\label{eq:jacobi}
    \d_t n^{ik}-\theta^{k}_{\phantom{k}l}n^{li}-\theta^{i}_{\phantom{i}j}n^{jk}+\theta n^{ki}=0,
  \end{equation}
  so that (\ref{eq:nijd}) holds. We leave the proof of (\ref{eq:phiddot}) to the reader.

  In order to prove the final statement of the lemma, assume that $\{e_i'\}$ has been chosen so that $n^{ij}(t_0)=0$ for $i\neq j$ and $\theta_{ij}(t_0)=0$
  for $i\neq j$. Then it follows from (\ref{eq:silmd}) and (\ref{eq:nijd}) that collecting the off-diagonal components of $\sigma_{ij}$ and the
  off-diagonal components of $n_{ij}$ into one vector, say $x$, then $x$ satisfies an equation of the form $x_t=Ax$. Since $x(t_0)=0$, it follows
  that $x(t)=0$ for all $t\in I$. Next, 
  \[
  0=\nabla_{e_0}e_i=[e_0,e_i]+\nabla_{e_i}e_0=e_0(e_i^j)e_j'+\textstyle{\sum}_j \theta_{ij}e_j^k e_k'.
  \]
  This means that
  \[
  e_0(e_i^k)=-\textstyle{\sum}_j \theta_{ij}e_j^k=-\theta_i e_i^k
  \]
  (no summation in the last expression), where $\theta_i:=\theta_{ii}$ (no summation). Since $e_i^k(t_0)=\de_i^k$, it follows that $e_i^k(t)=0$ for all
  $i\neq k$ and all $t\in I$. The final statement of the lemma follows. 
\end{proof}

\section{Existence of developments}\label{ssection:ex of dev step one}

Next, we prove Proposition~\ref{prop:unique max dev}. However, it is convenient to first make the following technical observation.
\begin{lemma}\label{lemma: bS low bd}
  Let $0\leq\a\in\ro$. Assume $0\leq n_i\in\ro$, $i=1,2,3$, to be such that $n_1n_2n_3\leq\a$. Then
  \[
    p:=n_1^2+n_2^2+n_3^2-2(n_1n_2+n_2n_3+n_3n_1)\geq -10\a^{2/3}. 
  \]  
\end{lemma}
\begin{proof}
  We can, without loss of generality, assume that $n_1\leq n_2\leq n_3$. If $n_3\leq \a^{1/3}$, then
  \[
    p\geq -2n_1(n_2+n_3)\geq-4\a^{2/3}.
  \]
  If $n_2\leq n_3/4$, then
  \[
    p\geq n_3^2-2n_3(n_1+n_2)\geq n_3^2-n_3^2=0.
  \]
  Assume now that $n_2\geq n_3/4$ and that $n_3\geq \a^{1/3}$. Then
  \[
    p\geq -2n_1n_2-2n_1n_3\geq -2\a^{2/3}-8\a^{2/3}.
  \]
  The lemma follows. 
\end{proof}

\begin{proof}[Proof of Proposition~\ref{prop:unique max dev}]
  Let $(G,\bge,\bk,\bphi_0,\bphi_1)=\mfI$. Due to Lemma~\ref{lemma:x arising from mfI}, there is an orthonormal basis $\{e_i'\}$ of $\mfg$ with
  respect to $\bge$ such that if $\bn$ and $K$ are the commutator and Weingarten matrices associated with $\{e_i'\}$ and $\mfI$, then
  $\bn$ and $K$ are diagonal. Moreover, due to (\ref{eq:nuvhc}), (\ref{eq:hamconfin}) is satisfied initially.
  Next, let $t_0\in \ro$ and solve (\ref{eq:silmd}), (\ref{eq:thetad}), (\ref{eq:nijd}) and (\ref{eq:phiddot}) with the following initial data:
  $\sigma_{ij}(t_0)$ are the components of the trace free part of $K$; $\theta(t_0):=\tr K$; $n(t_0):=\bn$; $\phi(t_0):=\bphi_0$ and
  $\phi_t(t_0):=\bphi_1$. Let $I$ denote the maximal existence interval. Due to (\ref{eq:silmd}), (\ref{eq:nijd}) and the fact that the off-diagonal
  components of $\sigma_{ij}$ and $n_{ij}$ vanish initially, it follows that $n(t)$ and $\sigma(t)$ are diagonal for all $t$. At this stage, we can
  define $a_i$ by
  \begin{equation}\label{eq:for the ai}
    \d_t a_{i}=(\sigma_{ii}+\theta/3)a_{i},\ \ \
    a_{i}(t_{0})=1
  \end{equation}
  (no summation). Let $\phi$ be the solution obtained above (by solving (\ref{eq:phiddot})), $M:=G\times I$, $e_0:=\d_t$ and $e_i:=a_i^{-1}e_i'$
  (no summation). Define a Lorentz metric on $M$ by the condition that $\{e_\a\}$ is an orthonormal frame. The proof of the statement that
  $(M,g,\phi)$ is a globally hyperbolic development of $\mfI$ is very similar to the proof of \cite[Proposition~20.3, p.~215]{RinCauchy}; the
  reader is encouraged to generalise the corresponding argument. In order to prove the uniqueness statement, assume that there are two developments
  with corresponding existence intervals $I_i$, $i=1,2$. By time translation, we can assume that the initial data are specified at $t=t_0$ for
  both developments. We can choose $\{e_i'\}$ as the initial frame (since the solutions induce the same initial data at $t_0$) for both developments.
  For both developments, we can introduce a Fermi-Walker propagated frame as in the statement of Lemma~\ref{lemma:FW frame eqns}. The corresponding
  $\sigma_{ij}$'s etc. then satisfy the same equations with the same initial data. This means that the $\sigma_{ij}$, the $\theta$, the $n_{ij}$ and
  the $\phi$ of the two solutions coincide. Moreover, both frames take the form $e_i:=a_i^{-1}e_i'$ (no summation). Since the $a_i$ have to satisfy
  (\ref{eq:for the ai}), it follows that the frames are the same. The only thing that remains to be verified is that the existence intervals are
  the same. However, the criteria that either $t_\pm=\pm\infty$ or $|\theta|+|\phi|$ is unbounded as $t\rightarrow t_\pm$ guarantee that the
  existence intervals are maximal. On the other hand, we need to verify that maximal solutions to (\ref{seq:EFEwrtvar})--(\ref{eq:phiddot})
  satisfy these criteria. Assume, to this end, that $t_+<\infty$ and that $\phi$ and $\theta$ remain bounded to the future. Excluding Bianchi
  type IX for the moment, (\ref{eq:hamconfin}) then implies that $\phi_t$ and the $\sigma_{ij}$ remain bounded to the future. Combining this observation
  with (\ref{eq:nijd}) leads to the conclusion that all the variables are bounded to the future. This means that the solution can be extended beyond
  $t_+$, contradicting the assumption that the existence interval is maximal. In the case of Bianchi type IX, assume, without loss of generality, that
  the $n_i$ are strictly positive and consider $f:=n_1n_2n_3$; here $n_i:=n_{ii}$ (no summation). Due to (\ref{eq:nijd}), $f_t=-\theta f$. This means,
  in particular, that $f$ remains bounded to the future; say $f\leq C$. Due to Lemma~\ref{lemma: bS low bd}, this means that the sum of the second and
  third terms on the right hand side of (\ref{eq:hamconfin}) is bounded from  below to the future. Moreover, combining this lower bound with
  (\ref{eq:hamconfin}) and the boundedness of $\theta$ and $\phi$ implies that $\sigma$ and $\phi_t$ are bounded to the future. Combining this bound with
  (\ref{eq:nijd}) yields the conclusion that all the variables remain bounded to the future, contradicting $t_+<\infty$.   
  This contradiction finishes the argument in the Bianchi type IX setting. The argument in the opposite time direction is similar. This leads
  to the desired dichotomy: either $t_\pm=\pm\infty$ or $|\theta|+|\phi|$ is unbounded as $t\rightarrow t_\pm$. In case $V\geq 0$, the above arguments lead
  to the improved dichotomy that either $t_\pm=\pm\infty$ or $|\theta|$ is unbounded as $t\rightarrow t_\pm$, as stated in Remark~\ref{remark:improved dichotomy}.

  In order to prove the final statement, assume, to begin with, that $\mfI\in\mB_{\mfT}^{\iso}[V]$. Due to Lemma~\ref{lemma:psi sigma plus isotropic}, we can
  then assume that all the $\sigma_{ij}$ vanish initially and all the $n_i$ equal initially. Combining these observations with (\ref{seq:EFEwrtvar})--(\ref{eq:nijd})
  yields the conclusion that $\sigma_{ij}(t)=0$ for all $t$ and that $n_i(t)=n_j(t)$ for all $t$ and all $i,j$. Thus the solution remains isotropic.
  Note also that (\ref{eq:for the ai}) yields the conclusion that all the $a_i$ equal. This means that the spacetime metric takes the form
  (\ref{eq:gSH}) with $\bge=a^2(t)\bar{h}$, where $\bar{h}$ is the initial metric. 

  Assume $\mfI\in\mB_\mfT^{\roLRS}[V]$. Then there is an $x\in \sfB_\mfT^{\roLRS}$ arising from $\mfI$; see
  Remark~\ref{remark:x arising from mfI}, Lemma~\ref{lemma:sfB char per symm VIIz} and Definition~\ref{def:BTs GTs}. We can, without loss of generality, assume
  that $x=(\bn_1,\bn_2,\bn_3,k_1,k_2,k_3,\bphi_0,\bphi_1)$ is such that $\bn_2=\bn_3$, $k_2=k_3$ and $(\bn_1,k_1)\neq (\bn_2,k_2)$. Next, (\ref{eq:silmd}) and
  (\ref{eq:nijd}) can be used to derive a homogeneous first order equation for $\sigma_{22}-\sigma_{33}$ and $n_{22}-n_{33}$. Since these quantities vanish initially,
  we conclude that they always vanish. If, in addition, $\sigma_{11}-\sigma_{22}$ and $n_{11}-n_{22}$ vanish at some later point, we similarly, by uniqueness, conclude
  that they always vanish, contradicting the properties of the initial data. Thus we always have $(\sigma_{11},n_{11})\neq (\sigma_{22},n_{22})$ for all $t$. These
  observations lead to the conclusion that that the initial data induced by the development at the constant-$t$ hypersurfaces are LRS. Next, since
  $\sigma_{22}=\sigma_{33}$, (\ref{eq:for the ai}) implies that $a_2=a_3$. This means that the metric can be written
  \begin{equation}\label{eq:LRS spacetime metric}
    g=-dt\otimes dt+a_1^2(t)\xi^1\otimes\xi^1+a_2^2(t)(\xi^2\otimes\xi^2+\xi^3 \otimes\xi^3),
  \end{equation}
  where $\{\xi^i\}$ is the frame dual to the frame used for the initial data, say $\{e_i'\}$.
  
  Assume $\mfI\in\mB_\mfT^{\roper}[V]$. Then there is an $x\in \sfB_\mfT^{\roper}$ arising from $\mfI$; see
  Remark~\ref{remark:x arising from mfI}, Lemma~\ref{lemma:sfB char per symm VIz} and Definition~\ref{def:BTs GTs}. We can, without loss of generality, assume
  that $x=(\bn_1,\bn_2,\bn_3,k_1,k_2,k_3,\bphi_0,\bphi_1)$ is such that $\bn_1=0$, $\bn_3=-\bn_2$ and $k_2=k_3$; i.e. $\bn_{11}=0$, $\bn_{33}=-\bn_{22}$ and
  $\bsi_{22}=\bsi_{33}$. Since the equation for $n_{11}$
  is homogeneous, this means that $n_{11}=0$. Moreover, (\ref{eq:silmd}) and (\ref{eq:nijd}) can be used to derive a homogeneous first order equation
  for $\sigma_{22}-\sigma_{33}$ and $n_{22}+n_{33}$. Since these quantities vanish initially, we conclude that they always vanish. This means that the
  permutation symmetry carries over to the development. Moreover, since $\sigma_{22}=\sigma_{33}$, (\ref{eq:for the ai}) implies that $a_2=a_3$. This
  means that the metric can again be written as in (\ref{eq:LRS spacetime metric}).

  Assume $\mfI\in\mB_\mfT^{\rogen}[V]$. Then, by the previous arguments, it cannot induce initial data with a higher degree of symmetry at another time.
  The final statement of the proposition follows. 
\end{proof}
In what follows, it will be of interest to know if the mean curvature is integrable or not. 
\begin{lemma}\label{lemma:BianchiAdevelopment}
  Let $V\in C^{\infty}(\rn{})$ be non-negative, $\mfT$ be a Bianchi class A type and $\mfI\in \mB_{\mfT}[V]$. Let $J$ and $\theta$ be the existence interval
  and mean curvature, respectively, of the corresponding development, see Proposition~\ref{prop:unique max dev}. Assume that $\theta$ is initially non-negative.
  If $\mfT\neq\mrIX$, there are the following two possibilities:
  \begin{itemize}  
  \item $\mfI\in \mB_{\mrI}^{\iso}[V]$ and $\theta(t)>0$ for all $t\in J$. Moreover, either $J=\rn{}$ or $J=(0,\infty)$ (after a time translation, if
    necessary). If $J=\rn{}$, then $\theta\notin L^{1}(-\infty,0)$ and $\theta\notin L^1(0,\infty)$. If $J=(0,\infty)$, then
    $\theta\notin L^{1}(0,1)$; $\theta\notin L^1(1,\infty)$; and $\theta\rightarrow\infty$ as $t\downarrow 0$. 
  \item $\mfI$ is not isotropic; $\theta(t)>0$ for all $t\in J$; up to a translation in $t$, $J=(0,\infty)$; $\theta(t)\rightarrow\infty$
    as $t\downarrow 0$; $\theta_t(t)<0$ for all $t$; $\theta\notin L^{1}(0,1)$; and $\theta\notin L^{1}(1,\infty)$.
  \end{itemize}
  In addition, $\theta$ converges to a non-negative real number as $t\rightarrow\infty$.

  If $\mfT=\mrIX$, $X(t):=\mfX[\theta(t),\phi(t),\phi_t(t)]$, see (\ref{eq:mfX def}), and $\mfI\in \mB_{\mrIX,+}[V]$, then $J=(0,t_+)$ and there is a $t_0\in J$
  such that $\theta(t)\rightarrow\infty$ as $t\downarrow 0$; $\theta_t(t)<0$ for all $t\leq t_0$; $X(t)>0$ for all $t\leq t_0$; and $\theta\notin L^{1}(0,t_0)$.
\end{lemma}
\begin{remark}\label{remark:Exceptional Bianchi type I}
  In order to illustrate that $J=\ro$ can occur, assume that there is an $s_0\in\ro$ such that $V'(s_0)=0$, $V(s_0)>0$ and $V''(s_0)<0$. There is then an
  isotropic Bianchi type I solution with $\theta(t)=[3V(s_{0})]^{1/2}$ and $\phi(t)=s_{0}$. The corresponding spacetime is a vacuum solution with a
  positive cosmological constant. Since it arises from trivial initial data, it is, however, excluded in the present setting. On the other hand, there are also
  non-vacuum solutions converging to the corresponding equilibrium; see the proof of Proposition~\ref{prop:futureasstepone} below. Note, in particular, that
  these solutions do not have crushing singularities, the setting we are interested in here. More generally,
  the isotropic Bianchi type I solutions  are non-generic; the complement of the corresponding set of initial data in Bianchi type I is open and dense
  (as well as of full measure). 
\end{remark}
\begin{proof}
  Let us first consider the case that $\mfT\neq\mrIX$ and that there is a $t_{1}$ in the existence interval such that $\theta(t_{1})=0$. Then, due to the
  discussion prior to Definition~\ref{def:degenerate version of id}, the initial data are trivial, a conclusion that contradicts $\mfI\in \mB[V]$. Since
  $\theta\geq 0$ initially, we conclude that $\theta(t)>0$ for all $t\in J$ if $\mfT\neq \mrIX$. Next, recall that the scalar curvatures of the constant-$t$
  hypersurfaces are given by
  \begin{equation}\label{eq:bSitonij}
    \bS:=-n_{ij}n^{ij}+\tfrac{1}{2}\tr(n)^{2};
  \end{equation}
  see, e.g., \cite[(19.6), p.~209]{RinCauchy}. Consider, moreover, the function
  \begin{equation}\label{eq:Xdef}
    X:=\theta^{2}-\tfrac{3}{2}\phi_{t}^{2}-3V\circ\phi,
  \end{equation}
  used in \cite{mas,reninflation}. Note that (\ref{eq:hamconfin}) and (\ref{eq:bSitonij}) imply that
  \begin{equation}\label{eq:Xaltversion}
    X=-\tfrac{3}{2}\bS+\tfrac{3}{2}\sigma_{ij}\sigma^{ij}.
  \end{equation}
  Since $\bS\leq 0$ for $\mfT\neq\mrIX$, $X$ is non-negative.

  Assume now that there is a $t_{1}$ such that $X(t_{1})=0$. Then $\sigma_{ij}(t_{1})=0$ and $\bS(t_{1})=0$. By arguments similar to those given prior to
  Definition~\ref{def:degenerate version of id}, this means that the solution is isotropic and of Bianchi type I. If $X$ ever vanishes, $X$ thus vanishes
  identically.  Due to (\ref{eq:Xdef}), this means that as long as $\theta$ is bounded, $\phi_{t}$ is bounded (so that $\phi$ cannot blow up in finite time).
  On the other hand, $\theta(t)>0$ for all $t$, so that (\ref{eq:thetad}) implies that $\theta$ decays to a non-negative value to
  the future. In particular, the solution cannot blow up in finite time to the future, so that $[t_{0},\infty)\subset J$. There are two possibilies
  concerning the existence interval. Either the existence interval is $\rn{}$, or it is $(0,\infty)$ after a time translation, if necessary. If $\theta$
  converges to a strictly positive number to the future, then it is clear that $\theta\notin L^1(1,\infty)$. If $\theta$ converges to zero, note that
  (\ref{eq:thetad}) and (\ref{eq:hamconfin}) imply that $\theta_t\geq -\theta^2$, so that
  \[
    \theta(t)\geq\theta(1)\exp\big(-\textstyle{\int}_1^t\theta(s)ds\big).
  \]
  For this to be consistent with $\theta$ converging to zero, we have to have $\theta \notin L^1(1,\infty)$. If the existence interval is $\rn{}$, then,
  since $\theta(t)\geq\theta(0)>0$ for all $t\leq 0$, it is clear that $\theta\notin L^1(-\infty,0)$. Finally, consider the case that
  the existence interval equals $(0,\infty)$. Then $\theta(t)\rightarrow\infty$ as $t\downarrow 0$, since $\theta$ increases to the past, and since the
  existence interval would be unbounded to the past in case $\theta$ were bounded to the past; see Remark~\ref{remark:improved dichotomy}. For a proof that
  $\theta\notin L^1(0,1)$ in this case, see below. If $\mfT\neq\mrIX$ we, from now on, therefore restrict our attention to anisotropic solutions such that
  $X(t)>0$ and $\theta(t)>0$ for all $t$. 
  
  If $\mfT=\mrIX$, we assume that $\mfI\in \mB_{\mrIX,+}[V]$. By Definition~\ref{definition:Bap Bp Bpp}, there is thus a $t_0\in J$ such that $\theta(t_0)>0$
  and $X(t_0)>0$. If $\mfT\neq\mrIX$, fix an arbitrary $t_{0}\in J$. Then $X(t_{0})>0$ and it can be computed that 
  \[
  X_t=\theta (\bS-3\sigma_{ij}\sigma^{ij})=-\tfrac{2}{3}\theta X-2\theta\sigma_{ij}\sigma^{ij};
  \]
  see \cite[(26.12), p.~452]{stab}, where we have set the terms associated with the Vlasov matter to zero. In case $\mfT\neq\mrIX$, we know that $\theta>0$
  and $X>0$. In case $\mfT=\mrIX$, the fact that $\theta(t_0)>0$ and $X(t_0)>0$ implies that $X_t(t_0)<0$, so that $X$ increases to the past. Combining this
  observation with (\ref{eq:Xdef}) implies that $\theta$ remains positive to the past. To summarise, if $\mfT=\mrIX$, then $X(t)>0$ and $\theta(t)>0$ for
  $t\leq t_0$. This means that for all Bianchi types, 
  \begin{equation}\label{eq:dotXest}
    X_t\leq-\tfrac{2}{3}\theta X
  \end{equation}
  for $t\leq t_0$. Letting
  \[
  \Theta(t)=\textstyle{\int}_{t_{0}}^{t}\theta(s)ds,
  \]
  (\ref{eq:dotXest}) implies that for $t\leq t_0$
  \[
  \tfrac{d}{dt}(e^{2\Theta/3}X)=e^{2\Theta/3}\left(\tfrac{2}{3}\theta X+X_t\right)\leq 0.
  \]
  In particular, due to (\ref{eq:Xdef}) and the fact that $V$ is non-negative, 
  \[
  e^{2\Theta(t)/3}\theta^{2}(t)\geq e^{2\Theta(t)/3}X(t)\geq X(t_{0}),
  \]
  for $t\leq t_{0}$, so that 
  \[
  \tfrac{1}{3}e^{\Theta(t)/3}\theta(t)\geq\tfrac{1}{3}X^{1/2}(t_{0}).
  \]
  Integrating this inequality from $t_{a}$ to $t_{b}$, where $t_{a}\leq t_{b}\leq t_{0}$, yields
  \begin{equation}\label{eq:qtbta}
    e^{\Theta(t_{b})/3}-e^{\Theta(t_{a})/3}\geq\tfrac{1}{3}X^{1/2}(t_{0})(t_{b}-t_{a}),
  \end{equation}
  so that 
  \[
  e^{\Theta(t_{b})/3}\geq\tfrac{1}{3}X^{1/2}(t_{0})(t_{b}-t_{a})+e^{\Theta(t_{a})/3}.
  \]
  If the solution existed globally to the past, we could let $t_{a}\rightarrow-\infty$ in this inequality, something which is clearly impossible.
  By a time translation, if necessary, we can thus assume that $J=(0,t_+)$. Combining this observation with Remark~\ref{remark:improved dichotomy}
  yields the conclusion that $\theta(t)$ is unbounded to the past. Next, note that (\ref{eq:thetad}) and (\ref{eq:Xaltversion}) imply that
  \begin{equation}\label{eq:theta t ito X}
    \theta_t=-\sigma_{ij}\sigma^{ij}-\tfrac{1}{3}X-\tfrac{3}{2}\phi_t^2.
  \end{equation}
  In other words, $\theta$ increases monotonically to the past for $t\leq t_0$, since $X(t)>0$ for $t\leq t_0$. Since $\theta$ is unbounded to the past,
  this means that $\theta(t)\rightarrow\infty$ as $t\downarrow 0$.
  
  For the remainder of the analysis, we consider the cases $\mfT=\mrIX$ and $\mfT\neq\mrIX$ separately. Assume $\mfT\neq\mrIX$. Since $\theta$ is decaying
  to the future and is bounded from below by $0$, we obtain global existence to the future by Remark~\ref{remark:improved dichotomy}. Thus the existence
  interval is $(0,\infty)$ and $\theta(t)\rightarrow\infty$ as $t\downarrow 0$. Combining (\ref{eq:thetad}) and (\ref{eq:hamconfin}) yields
  \begin{equation}\label{eq:theta t blow up}
    \theta_t=-\theta^{2}-\bS+3V\circ\phi.
  \end{equation}
  In particular, 
  \[
  \tfrac{d}{dt}(e^{\Theta}\theta)=e^{\Theta}(\theta^{2}+\theta_t)=e^{\Theta}[-\bS+3V\circ\phi]\geq 0,
  \]
  since $\bS\leq 0$ when $\mfT\neq\mrIX$. Thus $e^{\Theta}\theta$ decreases to the past, so that it converges to a constant $\a_{0}\geq 0$. We must thus have
  \[
  \lim_{t\downarrow 0}\Theta(t)=-\infty,
  \]
  so that $\theta\notin L^{1}(0,1)$ (this argument also applies if $(\mfT,\mfs)=(\mrI,\iso)$ and $J=(0,\infty)$). In order to prove that
  $\theta\notin L^{1}(1,\infty)$, note that (\ref{eq:Riczzeq}) implies that
  \begin{equation}\label{eq:dotthetaprelver}
    \theta_t=-\tfrac{1}{3}(1+q)\theta^{2},
  \end{equation}
  where
  \begin{equation}\label{eq:q def mid of pf}
    q:=3\tfrac{\sigma_{ij}\sigma^{ij}}{\theta^{2}}+3\tfrac{\phi_{t}^{2}}{\theta^{2}}-3\tfrac{V\circ\phi}{\theta^{2}}.
  \end{equation}
  Note that, due to (\ref{eq:hamconfin}), all the terms on the right hand side of (\ref{eq:q def mid of pf}) are bounded. In fact, $q\in [-1,2]$,
  so that $|q|\leq 2$. Combining this observation with (\ref{eq:dotthetaprelver}) yields, for $t\geq t_{0}$,
  \begin{equation}\label{eq:thetafutasnonint}
    \theta(t)=\theta(t_{0})\exp\big(-\tfrac{1}{3}\textstyle{\int}_{t_{0}}^{t}[1+q(s)]\theta(s)ds\big)
    \geq \theta(t_{0})\exp\big(-\textstyle{\int}_{t_{0}}^{t}\theta(s)ds\big).
  \end{equation}
  We know $\theta$ to be monotonically decaying. If it decays to a strictly positive number, then it is obvious that $\theta\notin L^{1}(1,\infty)$.
  If $\theta(t)\rightarrow 0$ as $t\downarrow 0$, then (\ref{eq:thetafutasnonint}) implies that $\theta\notin L^{1}(1,\infty)$.

  Next, consider the case that $\mfT=\mrIX$. All that remains to be demonstrated is that $\theta\notin L^{1}(0,t_0)$. Assuming the opposite, it follows
  that $f:=n_1n_2n_3$ remains bounded to the past; $f_t=-\theta f$ due to (\ref{eq:nijd}). Let $\a\in\ro$ be such that $(n_1n_2n_3)(t)\leq\a$ for all
  $t\leq t_0$. Due to Lemma~\ref{lemma: bS low bd}, 
  \[n_1^2+n_2^2+n_3^2-2(n_1 n_2+n_2 n_3+n_3 n_1) \geq  -10\a^{2/3}
  \]  
  for all $t\leq t_0$. In other words, $\bS$ is bounded from above to the past. On the other hand, 
  \begin{equation}\label{eq:exp two Theta theta}
    \tfrac{d}{dt}(e^{2\Theta}\theta)=e^{2\Theta}(\theta^2-\bS+3V\circ\phi),
  \end{equation}
  where we appealed to (\ref{eq:theta t blow up}). Since $\theta(t)\rightarrow\infty$ as $t\downarrow 0$ and $\bS$ is bounded from above to the past,
  there is a $t_1\in (0,t_0)$ such that the right hand side of (\ref{eq:exp two Theta theta}) is $\geq 0$ for $t\leq t_1$. This means that
  $e^{2\Theta}\theta$ converges to a non-negative real number to the past. Since $\theta(t)\rightarrow\infty$, this means that
  $\Theta(t)\rightarrow -\infty$; i.e., $\theta$ is not integrable to the past. 
\end{proof}
Beyond the case that both $X$ and $\theta$ are strictly positive, it is of interest to consider the following possibility. 
\begin{lemma}\label{lemma:Bianchi IX remainder}
  Let $0\leq V\in C^{\infty}(\rn{})$ and $\mfI\in \mB_{\mrIX,\roap}[V]$; see Definition~\ref{definition:Bap Bp Bpp}. Let $J$ and $\theta$ be the existence
  interval and mean curvature, respectively, of the corresponding development, see Proposition~\ref{prop:unique max dev}. Then, up to a time translation,
  there are $t_1\in J$ and $t_+>0$ such that $\theta(t)>0$ for $t\leq t_1$; $J=(0,t_+)$; $\theta(t)\rightarrow\infty$ as $t\downarrow 0$; $\theta_t(t)<0$
  for all $t\leq t_1$; and $\theta\notin L^{1}(0,t_1)$.
\end{lemma}
\begin{proof}
  By assumption, see Definition~\ref{definition:Bap Bp Bpp}, there is a $t_0\in J=(t_-,t_+)$ such that $\theta(t)>0$ for all $t\leq t_0$. If there is a $t\leq t_0$
  such that $X(t)>0$, then the desired conclusions follow from Lemma~\ref{lemma:BianchiAdevelopment}. We can therefore assume that $X(t)\leq 0$ for all
  $t\leq t_0$. Combining this assumption with the assumptions of the lemma yields $X(t)/\theta^2(t)\rightarrow 0$ as $t\downarrow t_-$. Next, again by
  the assumptions of the lemma, $V\circ\phi(t)/\theta^2(t)\rightarrow 0$ as $t\downarrow t_-$. On the other hand, (\ref{eq:Xdef}) can be reformulated to
  \begin{equation}\label{eq:X ver HC}
    1=\frac{X}{\theta^2}+\frac{3}{2}\frac{\phi_t^2}{\theta^2}+\frac{3V\circ\phi}{\theta^2},
  \end{equation}
  so that $3\phi_t^2(t)/\theta^2(t)\rightarrow 2$ as $t\downarrow t_-$. Combining the above observations and assumptions with (\ref{eq:theta t ito X}) yields a
  $T\in J$ such that $\theta_t\leq -\theta^2/2$ for all $t\leq T$. This means that $\theta$ blows up in finite time to the past, so that $t_-=0$. All that
  remains to be demonstrated is that $\theta\notin L^{1}(0,t_1)$. However, this follows by the same argument as at the end of the proof of
  Lemma~\ref{lemma:BianchiAdevelopment}.
\end{proof}
In the case of Bianchi type IX with a vanishing potential, we obtain the following improvement. 
\begin{lemma}
  Let $V=0$ and $\mfI\in \mB_{\mrIX}[V]$. Let $J$ and $\theta$ be the existence interval and mean curvature, respectively, of the corresponding
  development, see Proposition~\ref{prop:unique max dev}. Then $J=(t_{-},t_{+})$, where $t_{-},t_{+}\in\rn{}$. Moreover, if $t_{0}\in (t_{-},t_{+})$,
  then $\theta(t)\rightarrow\infty$ as $t\downarrow t_{-}$; $\theta(t)\rightarrow-\infty$ as $t\uparrow t_{+}$; $\theta\notin L^{1}(t_{-},t_{0})$;
  and $\theta\notin L^{1}(t_{0},t_{+})$.
\end{lemma}
\begin{proof}
  Note, to begin with, that, in the spatially homogeneous setting, orthogonal stiff fluids are equivalent to a scalar field. In fact, letting
  $\rho=\phi_{t}^{2}/2$, where $\phi$ is the scalar field, the stress energy tensor of the scalar field can be interpreted as the stress energy
  tensor of an orthogonal stiff fluid; see, e.g., \cite[Section~1.4]{RinQCSymm}. Given this correspondence, the statements are essentially immediate
  consequences of results in \cite{BianchiIXattr}. In particular, the statement
  concerning the existence interval follows from \cite[Lemma~21.8, p.~493]{BianchiIXattr}. In order to deduce the statements concerning the mean
  curvature, note that, according to \cite[Lemma~21.6, p.~492]{BianchiIXattr}, there is a $t_{0}\in (t_{-},t_{+})$ such that the mean curvature
  $\theta$ is strictly positive on the interval $I_{-}:=(t_{-},t_{0})$ and strictly negative on $I_{+}:=(t_{0},t_{+})$. Next, let the time coordinates
  $\tau_{\mp}$ be defined on $I_{\mp}$ (up to constants) by $d\tau_{\mp}/dt=\theta/3$, see \cite[(137), p.~487]{BianchiIXattr}. Then the asymptotic regime
  $t\rightarrow t_{\pm}$ corresponds to $\tau_{\pm}(t)\rightarrow -\infty$; see \cite[Lemma~22.5, p.~498]{BianchiIXattr}. Note that this implies the
  non-integrability statements of the proposition. Finally, since the mean curvature $\theta$ satisfies $\d_{\tau_{\pm}}\theta=-(1+q)\theta$, see
  \cite[(139), p.~487]{BianchiIXattr}, where $q=2\Omega+2\Sigma_{+}^{2}+2\Sigma_{-}^{2}$, $\Omega=3\rho/\theta^{2}$ and $\rho\geq 0$
  is the energy density of the scalar field, the statements concerning the divergence of the mean curvature follow. 
\end{proof}

When relating solutions to asymptotic data, it is useful to know that the following refinement of Lemma~\ref{lemma:BianchiAdevelopment} holds.

\begin{lemma}\label{lemma:bth large enough}
  Assume, in addition to the conditions of Lemma~\ref{lemma:BianchiAdevelopment}, that $\mfT\neq\mrIX$ and that $V\in\mfP_{\ropar}$; see
  Definition~\ref{def:mfPdef}. If $\bth>[3v_{\max}(V)]^{1/2}$, then, using the notation of Lemma~\ref{lemma:BianchiAdevelopment},
  $\{\theta(t)\, |\, t\in J\}$ contains $([3v_{\max}(V)]^{1/2},\infty)$. If $(\mfT,\mfs)\neq(\mrI,\iso)$, then $\theta_t<0$ for all $t$. If
  $(\mfT,\mfs)=(\mrI,\iso)$, then $\theta_t\leq 0$ and the zeros of $\theta_t$ are isolated. 
\end{lemma}
\begin{proof}
  If $(\mfT,\mfs)\neq(\mrI,\iso)$, then $\theta_t<0$ for all $t$ due to Lemma~\ref{lemma:BianchiAdevelopment}. If $(\mfT,\mfs)=(\mrI,\iso)$, then
  $\theta_t=-3\phi_t^2/2\leq 0$; see (\ref{eq:thetad}). In particular, if $\theta_t(t_1)=0$, then $\phi_t(t_1)=0$. If, in addition, $V'[\phi(t_1)]=0$,  
  then $\phi(t)=\phi(t_1)$ for all $t$; cf. (\ref{eq:phiddot}). This means that the original initial data had to be trivial, contradicting the assumptions.  
  If $\phi_t(t_1)=0$, we must thus have $V'[\phi(t_1)]\neq 0$, so that $\phi_{tt}(t_1)\neq 0$ due to (\ref{eq:phiddot}). Thus $\phi_t$ is non-zero in a
  punctured neighbourhood of $t_1$, so that $\theta_t<0$ in a punctured neighbourhood of $t_1$. What remains to be demonstrated is that the desired
  statement concerning the range of the mean curvature holds. 
  
  Due to \cite[Corollary~26.4, p.~450]{stab}, we know that there is a $0\leq V_{1}\in\rn{}$ such that $\theta(t)\rightarrow (3V_{1})^{1/2}$ and
  $\phi_{t}^{2}(t)+2V\circ\phi(t)\rightarrow 2V_{1}$ as $t\rightarrow\infty$. If $V_1=0$, the desired statement follows, assuming $\theta$ diverges to
  $\infty$ to the past. We can therefore assume that $V_1>0$. Since $\theta$ is monotonically decreasing, see (\ref{eq:thetad}), $\theta(t)\geq (3V_1)^{1/2}$
  for all $t$. Due to the arguments presented in the proof of \cite[Corollary~26.7, p.~452]{stab}, we also know that $\phi_{t}(t)\rightarrow 0$ as
  $t\rightarrow\infty$ and that $\phi(t)$ converges to a limit (finite or infinite). Moreover,
  $V'\circ\phi(t)\rightarrow 0$ as $t\rightarrow\infty$. Assume that $\phi(t)$ converges to a finite limit, say $s_0$. Then $V'(s_0)=0$ and it is clear
  that $\theta(t)$ converges to $[3V(s_0)]^{1/2}$. If we know that $\theta$ diverges to $\infty$ to the past, the desired statement follows in this case.
  If $\phi$ tends to an infinite limit, then either $V_1=v_{\infty,+}$ or $V_1=v_{\infty,-}$, using the terminology of Definition~\ref{def:mfPdef}. Again, the
  statement follows in this case, assuming $\theta$ diverges to $\infty$ to the past. If $(\mfT,\mfs)\neq(\mrI,\iso)$, $\theta$ diverges to $\infty$
  to the past due to Lemma~\ref{lemma:BianchiAdevelopment}. It therefore only remains to prove divergence for $(\mfT,\mfs)=(\mrI,\iso)$.

  Assume, in order to obtain a contradiction, that $\theta$ is bounded to the past. Since $\theta$ is monotonically decreasing, this means that $\theta$
  converges to $\theta_-\in (0,\infty)$ to the past. Moreover, $J=\ro$ due to Remark~\ref{remark:improved dichotomy}. Due to (\ref{eq:hamconfin}), $\phi_t$
  and $V\circ\phi$ (and thereby $V'\circ\phi$; see Definition~\ref{def:mfPdef}) remain bounded to the past. Due to (\ref{eq:phiddot}), this
  means that $\phi_t$ and $\phi_{tt}$ remain bounded to the past. Since $\phi_t^2$ is integrable to the past, this means that $\phi_t$ converges to zero to the
  past. Combining this observation with (\ref{eq:hamconfin}) and the fact that $\theta$ converges, we conclude that $V\circ\phi(t)$ converges to the past. By
  an argument similar to
  the proof of \cite[Corollary~26.7, p.~452]{stab}, we conclude that $\phi$ converges to a limit (finite or infinite). Due to the assumptions, this means that
  $V'\circ\phi$ converges to a limit. Moreover, this limit is a real number since $V\circ\phi$ remains bounded to the past. Since $\theta$ converges to a
  finite value, $\phi_t$ converges to zero and $V'\circ\phi$ converges, it follows from (\ref{eq:phiddot}) that $V'\circ\phi$ has to converge to zero. If
  $\phi(t)$ converges to a finite number, say $s_0$, we conclude that $V'(s_0)=0$. This means that $\theta(t)\rightarrow [3V(s_0)]^{1/2}$, contradicting the
  assumptions (note that $\theta$ is monotonically increasing to the past). Assume now that $\phi(t)\rightarrow\pm\infty$. Then $V\circ\phi(t)$ converges
  to $v_{\infty,\pm}$. This means that $\theta(t)\rightarrow (3v_{\infty,\pm})^{1/2}$, again contradicting the assumptions. To conclude, the assumption that
  $\theta$ does not diverge to infinity to the past yields a contradiction. 
\end{proof}

\section{Wainwright-Hsu formulation}\label{ssection:whsuform}
In order to analyse the asymptotics, it is convenient to expansion normalise the variables in analogy with the work of Wainwright and Hsu
\cite{wah}. We begin by changing the time coordinate.

\begin{lemma}\label{lemma:tautimecoord}
  Assume the conditions of Lemma~\ref{lemma:BianchiAdevelopment} to be satisfied with $\mfT\neq\mrIX$ and let $J$ denote the existence interval of the
  development. Let $\tau$ be a function satisfying
  \begin{equation}\label{eq:dtdtau}
    \frac{d \tau}{dt}=\frac{\theta}{3}.
  \end{equation}
  \index{$\a$Aa@Notation!Time coordinates!tau@$\tau$}%
  Then $\tau$ is a smooth diffeomorphism from $J$ to $\rn{}$.
\end{lemma}
\begin{proof}
  Due to Lemma~\ref{lemma:BianchiAdevelopment}, we know that $\theta$ is strictly positive, that $\theta\notin L^{1}(t_-,t_0)$ and that
  $\theta\notin L^{1}(t_0,\infty)$, where $J=(t_-,\infty)$ and $t_0\in J$. The statement follows.  
\end{proof}
In the Bianchi type IX setting, the following result holds.
\begin{lemma}\label{lemma:tautimecoord BIX}
  Assume that either the conditions of Lemma~\ref{lemma:BianchiAdevelopment} are satisfied with $\mfT=\mrIX$ or that the conditions of
  Lemma~\ref{lemma:Bianchi IX remainder} are satisfied. Let $J=(0,t_+)$ denote the existence interval of the development. Then there is a $t_0\in J$
  such that $\theta(t)>0$ for $t\leq t_0$. Define $\tau$ by $\tau(t_0)=0$ and (\ref{eq:dtdtau}). Then $\tau$ is a smooth diffeomorphism from $(0,t_0]$ to
  $(-\infty,0]$. 
\end{lemma}
\begin{proof}
  The conclusion follows from the definition of $\tau$ and Lemmas~\ref{lemma:BianchiAdevelopment} and \ref{lemma:Bianchi IX remainder}. 
\end{proof}
In what follows, we assume that either the conditions of Lemma~\ref{lemma:tautimecoord} or of Lemma~\ref{lemma:tautimecoord BIX} are satisfied, and we
restrict the time coordinate $\tau$ accordingly. Moreover, we expansion normalise the variables as in \cite{wah} (see also
\cite[Section~22.1, p.~233]{RinCauchy}). In other words, 
\index{$\a$Aa@Notation!Expansion normalised!Variables!Sigmaij@$\Sigma_{ij}$}%
\index{$\a$Aa@Notation!Expansion normalised!Variables!Nij@$N_{ij}$}%
\index{$\a$Aa@Notation!Expansion normalised!Variables!Sigmap@$\Sp$}%
\index{$\a$Aa@Notation!Expansion normalised!Variables!Sigmam@$\Sm$}%
\begin{equation}\label{eq:SigmaijNij}
  \Sigma_{ij}:=\frac{\sigma_{ij}}{\theta},\ \ N_{ij}:=\frac{n_{ij}}{\theta},\ \
  \Sp:=\frac{3}{2}(\Sigma_{22}+\Sigma_{33}),\ \ \Sm:=\frac{\sqrt{3}}{2}(\Sigma_{22}-\Sigma_{33}). 
\end{equation}
Let $N_{i}$ and $\Sigma_i$
\index{$\a$Aa@Notation!Expansion normalised!Variables!Sigmai@$\Sigma_{i}$}%
\index{$\a$Aa@Notation!Expansion normalised!Variables!Ni@$N_{i}$}%
be the diagonal elements of $N_{ij}$ and $\Sigma_{ij}$ respectively. In the derivations to follow, it is convenient to keep
in mind that
\begin{equation}\label{eq:Sigma sq}
  3(\Sigma_1^2+\Sigma_2^2+\Sigma_3^2)=3\Sigma_{ij}\Sigma^{ij}=2(\Sigma_{+}^{2}+\Sigma_{-}^{2}).
\end{equation}
Next, let 
\begin{subequations}\label{seq:PsiPhiphiiqdef}
  \begin{align}    
    \phi_{0} := & \phi,\ \ \ \phi_{1} := \phi',\label{eq:phizphiodef}\\
    \Psi := & \frac{1}{6}\phi_{1}^{2}+\frac{3V\circ\phi_{0}}{\theta^{2}},\ \ \ 
    \Omega := \frac{9}{2}\frac{V\circ\phi_{0}}{\theta^{2}},\label{eq:PsiOmegadef}\\
    q := & 2(\Sigma_{+}^{2}+\Sigma_{-}^{2}+\Psi-\Omega).\label{eq:qdef}
  \end{align}
\end{subequations}
\index{$\a$Aa@Notation!Expansion normalised!Variables!phiz@$\phi_0$}%
\index{$\a$Aa@Notation!Expansion normalised!Variables!phio@$\phi_1$}%
\index{$\a$Aa@Notation!Expansion normalised!Expressions!Psi@$\Psi$}%
\index{$\a$Aa@Notation!Expansion normalised!Expressions!Omega@$\Omega$}%
\index{$\a$Aa@Notation!Expansion normalised!Expressions!q@$q$}%
In (\ref{eq:phizphiodef}), the prime denotes a derivative with respect to $\tau$. With notation as in (\ref{seq:PsiPhiphiiqdef}), equation
(\ref{eq:Riczzeq}) yields
\begin{equation}\label{eq:thetaprime}
  \theta'=-(1+q)\theta.
\end{equation}
In terms of the expansion normalised variables, the equations (\ref{eq:nijd}) and (\ref{eq:silmd}) read
\begin{subequations}\label{seq:SigmaiNiprprel}
  \begin{align}
    N_{i}' = & (q+6\Sigma_{i})N_{i},\label{eq:Niprprel}\\
  \Sigma_{i}' = & -(2-q)\Sigma_{i}-3S_i\label{eq:Sigmaiprprel}
  \end{align}
\end{subequations}
respectively, where
\begin{equation}\label{eq:Si def}
  S_i:=B_i-\frac{1}{3}(B_1+B_2+B_3)
\end{equation}
\index{$\a$Aa@Notation!Expansion normalised!Expressions!Si@$S_i$}%
and, if $\{i,j,k\}=\{1,2,3\}$,
\[
B_i:=N_i(N_i-N_j-N_k);
\]
\index{$\a$Aa@Notation!Expansion normalised!Expressions!Bi@$B_i$}%
cf. (\ref{eq:slmde}). Since
\begin{equation}\label{eq:SigmaiitoSpm}
  6\Sigma_{1}=-4\Sp,\ \ \
  6\Sigma_{2}=2\Sp+2\sqrt{3}\Sm,\ \ \
  6\Sigma_{3}=2\Sp-2\sqrt{3}\Sm,
\end{equation}
(\ref{seq:SigmaiNiprprel}) yields
\begin{subequations}\label{seq:NoSmprime}
  \begin{align}
    N_i'=&f_i(q,\Sp,\Sm)N_i, \label{eq:nip} \\
    \Sigma_\pm'=& -(2-q)\Sigma_\pm-3S_{\pm},  \label{eq:spmp}
  \end{align}
\end{subequations}
where 
\begin{subequations}\label{seq:fidef}
  \begin{align}
    f_{1}(q,\Sp,\Sm) := & q-4\Sp,\label{eq:fodef}\\
    f_{2}(q,\Sp,\Sm) := & q+2\Sp+2\sqrt{3}\Sm,\label{eq:ftdef}\\
    f_{3}(q,\Sp,\Sm) := & q+2\Sp-2\sqrt{3}\Sm\label{eq:fthdef}
  \end{align}
\end{subequations}
and
\begin{subequations}\label{seq:SpSmev}
  \begin{align}    
    S_{+} := & \frac{1}{2}[(\Nt-\Nth)^2-\No(2\No-\Nt-\Nth)],\label{eq:Spdef}\\
    S_{-} := & \frac{\sqrt{3}}{2}(\Nth-\Nt)(\No-\Nt-\Nth). \label{eq:Smdef}
  \end{align}
\end{subequations}
\index{$\a$Aa@Notation!Expansion normalised!Expressions!Spm@$S_\pm$}%
Turning to the evolution equations for $\phi_{0}$ and $\phi_{1}$, see (\ref{eq:phiddot}), they can be written 
\begin{subequations}\label{seq:phizphioev}
  \begin{align}
    \phi_{0}' = & \phi_{1},\label{eq:phizev}\\
    \phi_{1}' = & -(2-q)\phi_{1}-\frac{9V'\circ\phi_{0}}{\theta^{2}}\label{eq:phioneev}
  \end{align}
\end{subequations}
Finally, the Hamiltonian constraint can be written
\begin{equation}\label{eq:constraint}
  \Sp^2+\Sm^2+\Psi+\frac{3}{4}[\No^2+\Nt^2+\Nth^2-2(\No\Nt+\Nt\Nth+\No\Nth)]=1.
\end{equation}
\index{Expansion normalised!Hamiltonian constraint}%
\index{Hamiltonian constraint!Expansion normalised}%
\index{Constraint equations!Hamiltonian!Expansion normalised}%
This relation can be used to deduce that
\begin{equation}\label{eq:altqdef}
  q=2-\frac{3}{2}[\No^2+\Nt^2+\Nth^2-2(\No\Nt+\Nt\Nth+\No\Nth)]-2\Omega.
\end{equation}
It is important to note that if we let $N_{i}$, $\Sigma_{\pm}$, $\phi_{0}$, $\phi_{1}$ and $\theta$ be solutions to (\ref{eq:thetaprime}) and
(\ref{seq:NoSmprime})--(\ref{seq:phizphioev}), where $q$ is given by (\ref{eq:altqdef}), and if we define $f$ by
\[
f:=\Sp^2+\Sm^2+\Psi+\frac{3}{4}[\No^2+\Nt^2+\Nth^2-2(\No\Nt+\Nt\Nth+\No\Nth)]-1,
\]
then
\begin{equation}\label{eq:fdiffeq}
  f'=-2(2-q)f.
\end{equation}
It is also of interest to note that the evolution equation for $\Psi$ can be written
\begin{equation}\label{eq:Psieq}
\Psi'=-2(2-q)\Psi+4\Omega.
\end{equation}
For future reference, it is useful to make the following observation.
\begin{remark}\label{remark:Sigmapm vs Sigmai}
  When analysing the asymptotics of solutions, it is convenient to work with the $\Sigma_\pm$ as opposed to the $\Sigma_i$, since this choice reduces
  the number of variables. However, the natural geometric assumptions and conclusions are phrased in terms of isometry classes of initial
  data, isometry classes of data on the singularity and isometry classes of solutions. This necessitates taking quotients, as in, e.g.,
  Lemma~\ref{lemma:sfR mfT mfs}. Moreover, the associated groups act naturally on solutions when expressed using the variables $\Sigma_i$.
  For this reason, it is useful to keep in mind that (\ref{eq:Sigma sq}) holds. In particular, $q$ introduced in (\ref{eq:qdef}) and the
  Hamiltonian constraint (\ref{eq:constraint}) are invariant
  under the maps $\psi_\sigma^\pm$ introduced in (\ref{eq:psi sigma pm}) (where the $\bn_i$ are replaced by the $N_i$, the $k_i$ are replaced by
  the $\Sigma_i$ and the $\bphi_i$ are replaced by the $\phi_i$). Note also that with this choice, (\ref{seq:SigmaiNiprprel}) replaces
  (\ref{seq:NoSmprime})--(\ref{seq:SpSmev}), so that $\psi_\sigma^\pm$ maps solutions to solutions. 
\end{remark}

We make the following observation concerning the existence interval of solutions to (\ref{eq:thetaprime}) and
(\ref{seq:NoSmprime})--(\ref{eq:constraint}).

\begin{lemma}\label{lemma:globalexistenceresceq}
  Assume that $V$ is a smooth non-negative function. Then the existence interval of solutions to (\ref{eq:thetaprime}) and
  (\ref{seq:NoSmprime})--(\ref{eq:constraint}) which are not of Bianchi type IX is $\rn{}$.
\end{lemma}
\begin{proof}
  Since we have excluded Bianchi type IX, the polynomial in the $N_{i}$ appearing on the left hand side of (\ref{eq:constraint}) is non-negative.
  Combining this observation with (\ref{eq:constraint}) yields the conclusion that $\Sp$, $\Sm$ and $\Psi$ are bounded. This means that
  $\phi_{1}$ is bounded, so that $\phi_{0}$ cannot blow up in finite time. From (\ref{seq:PsiPhiphiiqdef}) and (\ref{eq:constraint}),
  it follows that $-1\leq q\leq 2$. Since $-1\leq q\leq 2$
  and the $\Sigma_{\pm}$ are bounded, (\ref{eq:nip}) implies that the $N_{i}$ cannot blow up in finite time. Finally,
  (\ref{eq:thetaprime}) implies that $\theta$ can neither blow up nor go to zero in finite time. In short, the solution cannot blow up in finite
  time, and we obtain global existence. 
\end{proof}

\chapter{Asymptotic behaviour in the direction of the singularity}\label{section:asymptoticbehavsing}

Our next goal is to analyse the asymptotic behaviour of solutions in the direction of the singularity. We begin by introducing appropriate notation
and characterising the solutions that converge to a special point. 

\section{Characterising convergence to a special point}\label{ssection:charconvtosppoint}
We use the following terminology in the formulation of the results. 
\begin{definition}
  The unit circle in the $\Sp\Sm$-plane is referred to as the \textit{Kasner circle}.
  \index{Kasner circle}%
  On the Kasner circle, the points $(-1,0)$, $(1/2,\sqrt{3}/2)$
  and $(1/2,-\sqrt{3}/2)$ are referred to as the \textit{special (or Taub) points}.
  \index{Special points}%
  \index{Taub points}%
  In terms of the variables $\Sigma_{i}$ and $N_{i}$ introduced in
  connection with (\ref{eq:SigmaijNij}), the subspaces of locally rotationally symmetric (LRS) solutions are defined by $N_{i}=N_{j}$,
  $\Sigma_{i}=\Sigma_{j}$ and $(N_k,\Sigma_k)\neq (N_i,\Sigma_i)$ for some $i,j,k$ satisfying $\{i,j,k\}=\{1,2,3\}$; i.e., there are three
  such subspaces. 
\end{definition}
\begin{remark}
  The definition of local rotational symmetry coincides with the one for initial data introduced in Definition~\ref{def:BTs GTs} and
  corresponds to the one introduced in Definition~\ref{def:LRS}; see Lemma~\ref{lemma:sfB char per symm VIIz} and the definition of the
  expansion normalised variables.
\end{remark}
\begin{remark}
  Up to a rotation in the $\Sp\Sm$-plane and a suitable corresponding permutation of the $N_{i}$ (or, alternatively, the application of a
  $\psi_\sigma^+\in \Gamma^{+,\roev}$ when using the $\Sigma_i$, $N_i$ variables), the LRS subspaces are given by the conditions that
  $\Sm=0$ and $\Nt=\Nth$. Note also that the special points belong to the intersection of the LRS subspaces and the Kasner circle. 
\end{remark}

In the derivation of the asymptotics, it will be of interest to keep the following observation in mind. The main ideas of the proof go back to
\cite[pp.~721--722]{cbu}.

\begin{prop}\label{prop:limitcharsp}
  Let $0\leq V\in\mfP_{\b_V}^{0}$ for some $0<\b_V\in\ro$ and assume that either the conditions of Lemma~\ref{lemma:tautimecoord} or the conditions of
  Lemma~\ref{lemma:tautimecoord BIX} are satisfied. Consider the corresponding solution to (\ref{eq:thetaprime}) and
  (\ref{seq:NoSmprime})--(\ref{eq:constraint}). Then the
  existence interval includes $(-\infty,0]$. If, in addition, $[\Sp(\tau),\Sm(\tau)]\rightarrow (-1,0)$ as $\tau\rightarrow-\infty$, then the solution
  is contained in the invariant set defined by $\Sm=0$ and $\Nt=\Nth$. 
\end{prop}
\begin{remark}\label{remark:limitcharsp}
  By the symmetries of the equations, there are similar statements in case $(\Sp,\Sm)$ converges to one of the other special points. Moreover, in
  case $\mfT\neq\mrIX$, the assumption that $\Sm(\tau)\rightarrow 0$ as $\tau\rightarrow -\infty$ can be omitted. 
\end{remark}
\begin{remark}
  Due to this result, if one of the eigenvalues of the expansion normalised Weingarten map converges to $1$ and the other two converge to $0$ in the
  direction of the singularity, then the solution has an additional symmetry; it is LRS. This is the reason we impose the final condition in
  Definition~\ref{def:ndvacidonbbssh}.
\end{remark}
\begin{proof}
  Due to the fact that the conditions of Lemma~\ref{lemma:tautimecoord} or Lemma~\ref{lemma:tautimecoord BIX} are satisfied, the existence interval
  includes $(-\infty,0]$. If $\mfT=\mrIX$ and the conditions of Lemma~\ref{lemma:BianchiAdevelopment} are satisfied, then $X(\tau)>0$ for
  $\tau\leq \tau_0\leq 0$ and some $\tau_0$. This means that $3V\circ\phi_0/\theta^2\leq 1$ for all $\tau\leq \tau_0$, see (\ref{eq:Xdef}), so that
  \[
  \Psi-\Omega\geq-\tfrac{3}{2}\tfrac{V\circ\phi_0}{\theta^2}\geq-\tfrac{1}{2}
  \]
  for all $\tau\leq \tau_0$. Thus $q$ introduced in (\ref{eq:qdef}) satisfies $q\geq 2\Sp^2-1$. Since $\Sp$ tends to $-1$, this means that $q$ is bounded
  from below by $1$ in the limit. Next, assume that the conditions of Lemma~\ref{lemma:Bianchi IX remainder} are satisfied. Then $V\circ\phi_0/\theta^2$
  converges to zero in the direction of the singularity. This means that $\Omega\rightarrow 0$. Combining this observation with the definition of
  $q$ and the fact that $\Sp\rightarrow -1$, it follows that $q$ is bounded from below by $2$ in the limit. Combining these observations with
  (\ref{eq:nip}) yields the conclusion that in the Bianchi type IX setting, $\No$, $\No\Nt$ and $\No\Nth$ converge to zero exponentially in
  the direction of the singularity. In particular, in the case of Bianchi type IX, $-\No(\Nt+\Nth)$ converges to zero. On the other hand,
  (\ref{eq:constraint}) can be written
  \begin{equation}\label{eq:constraint NoNthNth}
    \Sp^2+\Sm^2+\Psi+\frac{3}{4}[\No^2+(\Nt-\Nth)^2]-\frac{3}{2}\No(\Nt+\Nth)=1.
  \end{equation}
  Due to the fact that the last term on the left hand side of (\ref{eq:constraint NoNthNth}) converges to zero and $\Sp$ converges to $-1$, it is
  clear that $\Psi\rightarrow 0$ (so that $\phi_1\rightarrow 0$ and $\Omega\rightarrow 0$ due to (\ref{eq:PsiOmegadef}) and the fact that $V$ is
  non-negative), and that the polynomial in the $N_{i}$ appearing on the left hand side of (\ref{eq:constraint}) converges to zero. When
  $\mfT\neq\mrIX$, we obtain the same conclusions, since the fourth term on the left hand side of (\ref{eq:constraint}) is non-negative. In fact,
  in this case, we can also deduce that $\Sm\rightarrow 0$ from the assumption that $\Sp\rightarrow -1$. Thus,
  irrespective of whether $\mfT=\mrIX$ or $\mfT\neq\mrIX$, the above observations, combined with (\ref{eq:altqdef}), yield the conclusion that $q$
  converges to $2$. Appealing to (\ref{eq:nip}), it furthermore follows that for every $\e>0$, there is a constant $C_{\e}$ such that
  \begin{equation}\label{eq:No etc exp decay}
    |\No(\tau)|+|(\No\Nt)(\tau)|+|(\No\Nth)(\tau)|\leq C_{\e}e^{(6-\e)\tau}
  \end{equation}
  for $\tau\leq 0$. Similarly, due to (\ref{eq:thetaprime}), $1/\theta(\tau)\leq C_{\e}e^{(3-\e)\tau}$ for all $\tau\leq 0$. Combining this estimate
  with the fact that $\phi'$ converges to zero and $V\in\mfP_{\b_V}^{0}$, it follows that $|\Omega(\tau)|\leq C_{\e}e^{(6-\e)\tau}$ for all $\tau\leq 0$. 

  Next, define
  \begin{equation}\label{eq:mff def}
    \mff:=\frac{4}{3}\Sm^{2}+(\Nt-\Nth)^{2}.
  \end{equation}
  Since the set defined by $\Sm=0$ and $\Nt=\Nth$ is invariant, $\mff$ is either identically zero or always strictly positive. It can also be computed
  that
  \begin{equation}\label{eq:mffpr}
    \mff'=-\frac{8}{3}(2-q)\Sm^{2}+2(q+2\Sp)(\Nt-\Nth)^{2}+4\sqrt{3}\No (\Nt-\Nth)\Sm.
  \end{equation}
  On the other hand, by the above observations, $2-q$, $q+2\Sp$ and $\No$ converge to zero in the direction of the singularity. It follows that for
  every $\e>0$, there is a $T$ such that $\mff'(\tau)\leq \e \mff(\tau)$ for all $\tau\leq T$. This means that
  \begin{equation}\label{eq:mffTest}
    \mff(T)\leq \mff(\tau)\exp[\e(T-\tau)]
  \end{equation}
  for all $\tau\leq T$. Next, (\ref{eq:spmp}) and (\ref{eq:altqdef}) yield
  \begin{equation}\label{eq:Sp prime LRS char}
      \Sp' = -\frac{3}{2}(\Nt-\Nth)^{2}(\Sp+1)+\frac{3}{2}\No^{2}(2-\Sp)-\frac{3}{2}\No(\Nt+\Nth)(1-2\Sp)-2\Omega\Sp.
  \end{equation}
  In case $\mfT\neq\mrIX$, then $\Sp+1\geq 0$, and the first term is negative. In case $\mfT=\mrIX$ and $\Sp+1\leq 0$, then
  (\ref{eq:constraint NoNthNth}) implies that
  \[
  -(\Sp+1)(-\Sp+1)+\Sm^2+\Psi+\frac{3}{4}[\No^2+(\Nt-\Nth)^2]=\frac{3}{2}\No(\Nt+\Nth).
  \]
  This means that
  \[
  |\Sp+1|\leq \frac{3}{4}\No(\Nt+\Nth)
  \]
  when $\Sp+1\leq 0$. Combining this observation with (\ref{eq:No etc exp decay}) yields the conclusion that $|\Sp(\tau)+1|\leq C_\e e^{(6-\e)\tau}$ for
  all $\tau\leq 0$, assuming $\Sp(\tau)+1\leq 0$. Combining this observation with (\ref{eq:No etc exp decay}), (\ref{eq:Sp prime LRS char}) and the
  fact that $|\Omega(\tau)|\leq C_{\e}e^{(6-\e)\tau}$ for all $\tau\leq 0$, it follows that, irrespective of whether $\Sp+1\leq 0$ or $\Sp+1\geq 0$,
  \[
  \Sp'\leq C_\e e^{(6-\e)\tau}
  \]
  for all $\tau\leq 0$. Integrating this estimate from $-\infty$ to $\tau$ yields $\Sp(\tau)+1\leq C_\e e^{(6-\e)\tau}$. Since the negative part of
  $\Sp+1$ satisfies a similar bound, it follows that
  \[
  |\Sp(\tau)+1|\leq C_{\e}e^{(6-\e)\tau}
  \]
  for all $\tau\leq 0$. Next, note that the constraint (\ref{eq:constraint NoNthNth}) yields
  \[
  \frac{3}{4}\mff\leq \Sm^{2}+\frac{3}{4}(\Nt-\Nth)^{2}+\Psi+\frac{3}{4}\No^{2}=\frac{3}{2}\No(\Nt+\Nth)+(1-\Sp)(1+\Sp)\leq C_{\e}e^{(6-\e)\tau}.
  \]
  Combining this estimate with (\ref{eq:mffTest}) yields the conclusion that
  \[
  \mff(T)\leq C_{\e}\exp[(6-\e)\tau+\e(T-\tau)]
  \]
  for all $\tau\leq T$. Choosing $\e$ small enough and letting $\tau\rightarrow-\infty$ yields the conclusion that the right hand side converges to
  zero. Thus $\mff(T)=0$ and the statement of the proposition follows.
\end{proof}
In what follows, it will also be convenient to have the following asymptotic characterisation of LRS Bianchi type $\mrVIIz$ solutions. 
\begin{prop}\label{prop:LRS BVIIz as char}
  Let $0\leq V\in\mfP_{\a_V}^{0}$ for some $\a_V\in (0,1)$. Consider a solution to (\ref{eq:thetaprime}) and
  (\ref{seq:NoSmprime})--(\ref{eq:constraint}) with $\No=0$ and $\Nt\Nth>0$. Assume that
  \begin{equation}\label{eq:limit assumptions}
    \lim_{\tau\rightarrow-\infty}(\Nt/\Nth,\Sp,\phi')(\tau)=(1,\sigma_+,\Phi_1), 
  \end{equation}
  where $\sigma_+^2+\Phi_1^2/6=1$ and $\sigma_+>-1$. Then the solution is contained in the invariant set defined by $\Sm=0$ and $\Nt=\Nth$;
  i.e., the solution is either LRS or isotropic. 
\end{prop}
\begin{proof}
  Combining the assumptions of the proposition with (\ref{eq:constraint}) yields the conclusion that $\Sm$, $\Omega$ and $\Nt-\Nth$ converge to zero as
  $\tau\rightarrow-\infty$. Due to these conclusions and (\ref{eq:qdef}), it follows that $q\rightarrow 2$ as
  $\tau\rightarrow -\infty$. Combining these observations with the assumptions and the equations, it is clear that $\Nt$ and $\Nth$ converge to zero
  exponentially. By the assumptions concerning $V$; the fact that $q$ converges to $2$; the fact that $|\phi'|\leq\sqrt{6}$ and (\ref{eq:thetaprime}),
  it is also clear that $\Omega$ decays exponentially. This means that $q-2$ converges to zero exponentially. Since $q-2$ and the $N_i$ converge to zero
  exponentially, (\ref{eq:spmp}) implies that $\Sp-\sigma_+$ and $\Sm$ decay exponentially. Returning to the equations with this information at hand yields
  \begin{equation}\label{eq:rate of convergence LRS VIIz}
    |N_i(\tau)|\leq C e^{2(1+\sigma_+)\tau},\ \ \
    |q(\tau)-2|+|\Sp(\tau)-\sigma_+|+|\Sm(\tau)|\leq Ce^{2\e\tau}
  \end{equation}
  for $i=2,3$ and $\tau\leq 0$, where $\e:=\min\{2(1+\sigma_+),3(1-\a_V)\}$. Next, note that
  \[
  \left(\tfrac{\Nt}{\Nth}\right)'=4\sqrt{3}\Sm\tfrac{\Nt}{\Nth}.
  \]
  Combining this observation with (\ref{eq:limit assumptions}) and (\ref{eq:rate of convergence LRS VIIz}) yields $\Nt(\tau)/\Nth(\tau)-1=O(e^{2\e\tau})$,
  so that
  \begin{equation}\label{eq:Delta N est}
    [\Nt(\tau)-\Nth(\tau)]^2=\Nth^2(\tau)[\Nt(\tau)/\Nth(\tau)-1]^2\leq C\exp[4(1+\sigma_+)\tau+4\e\tau]
  \end{equation}
  for all $\tau\leq 0$. Next, let 
  \[
  h(\tau):=\textstyle{\int}_{-\infty}^{\tau}[2-q(s)]ds.
  \]
  Then
  \[
  |(e^{h}\Sm)'(\tau)|=3e^{h(\tau)}|S_-(\tau)|\leq C\exp[4(1+\sigma_+)\tau+2\e\tau]
  \]
  for all $\tau\leq 0$. Integrating this estimate yields the conclusion that
  \begin{equation}\label{eq:Sm dec est}
    |\Sm(\tau)|\leq C\exp[4(1+\sigma_+)\tau+2\e\tau]
  \end{equation}
  for all $\tau\leq 0$. Defining $\mff$ as in (\ref{eq:mff def}), the estimates (\ref{eq:Delta N est}) and (\ref{eq:Sm dec est}) imply that
  \begin{equation}\label{eq:mff improved decay VIIz}
    \mff(\tau)\leq C\exp[4(1+\sigma_+)\tau+4\e\tau]
  \end{equation}
  for all $\tau\leq 0$. On the other hand, (\ref{eq:mffpr}) and the fact that $2-q\geq 0$ imply that $\mff'\leq 2(q+2\Sp)\mff$. Integrating
  this estimate from $\tau$ to some fixed $T\geq \tau$ and using the above observations yields
  \[
  \mff(T)\leq C\exp[4(1+\sigma_+)(T-\tau)]\mff(\tau)
  \]
  for all $\tau\leq T$. Combining this estimate with (\ref{eq:mff improved decay VIIz}) yields
  \[
  \mff(T)\leq C\exp[4(1+\sigma_+)T]e^{4\e\tau}
  \]
  for all $\tau\leq T$. Letting $\tau$ tend to $-\infty$ yields the conclusion that $\mff(T)=0$, so that the solution is contained in the invariant
  set defined by $\Sm=0$ and $\Nt=\Nth$. 
\end{proof}

\section{Preliminary observations concerning the asymptotics}\label{ssection:prelobsconas}

As a first step in the derivation of the asymptotics, we consider anisotropic Bianchi type I and non-LRS Bianchi type II solutions. 

\begin{lemma}\label{lemma:BIBIIasympt}
  Consider an anisotropic development satisfying the conditions of Lemma~\ref{lemma:BianchiAdevelopment}. Assume it to be either
  Bianchi type I or non-LRS Bianchi type II. In case it is
  of Bianchi type II, assume $\No\neq 0$. Then $(\Sp,\Sm)$ converges to a limit $(\sigma_{+},\sigma_{-})$ as $\tau\rightarrow-\infty$. In case the
  solution is of non-LRS Bianchi type I, $(\sigma_{+},\sigma_{-})$ does not belong to any of the line segments connecting a special point to its
  antipodal point. In the case of Bianchi type II, $\sigma_{-}\neq 0$ and $\sigma_{+}\leq 1/2$. Moreover, $\Psi$ converges to a limit $\Psi_{\infty}$
  and there is a constant $\theta_{\infty}>0$ such that $e^{3\tau}\theta(\tau)$ converges to $\theta_{\infty}$ as $\tau\rightarrow-\infty$. In the case
  of Bianchi type II, $\No(\tau)\rightarrow 0$ as $\tau\rightarrow -\infty$. 

  Assume, in addition to the above, that $V\in \mfP_{\a_V}^1$ for some $\a_V\in (0,1)$. Then there are constants $\a>0$, $\theta_{\infty}>0$, $\Phi_{1}$
  and $\Phi_{0}$ such that
  \begin{subequations}\label{seq:asBI}
    \begin{align}
      |\phi_{1}(\tau)-\Phi_{1}|+|\phi(\tau)-\Phi_{1}\tau-\Phi_{0}| \leq & Ce^{\a\tau},\label{eq:BIphias}\\
      |\Sp(\tau)-\sigma_{+}|+|\Sm(\tau)-\sigma_{-}| \leq & Ce^{\a\tau},\label{eq:BISpmas}\\
      |q(\tau)-2|+|\ln\theta(\tau)+3\tau-\ln\theta_{\infty}| \leq & Ce^{\a\tau}\label{eq:BIlnthetaas}      
    \end{align}
  \end{subequations}  
  for all $\tau\leq 0$. In the case of Bianchi type II, $(\sigma_{+},\sigma_{-})$ satisfies $\sigma_{+}<1/2$ and $\sigma_{-}\neq 0$ and there is,
  in addition to the above constants, an $m_{1}$ such that the following estimate also holds for all $\tau\leq 0$:
  \begin{equation}\label{eq:BIINzas}
    |\ln|\No(\tau)|-(2-4\sigma_{+})\tau-m_{1}| \leq Ce^{\a\tau}.
  \end{equation}
  In both cases, 
  \begin{equation}\label{eq:asymptotic hamiltonian constraint}
    \frac{1}{6}\Phi_{1}^{2}+\sigma_{+}^{2}+\sigma_{-}^{2}=1.
  \end{equation}
  Finally, in the case of Bianchi type I, $\a$ can be chosen to equal $6-\sqrt{6}\a_{V}|\Phi_{1}|$, and in the case of Bianchi type II,
  it can be chosen to equal $\min\{4-8\sigma_{+},6-\sqrt{6}\a_{V}|\Phi_{1}|\}$. 
\end{lemma}
\begin{proof}
  In the case of Bianchi type I, the expressions $S_{\pm}$ vanish and the expression in the $N_{i}$ appearing in
  (\ref{eq:constraint}) vanishes as well. Using the relation (\ref{eq:altqdef}), the equations (\ref{eq:spmp})
  thus take the form
  \begin{equation}\label{eq:SpmprBtI}
    \Sp'=-2\Omega\Sp,\ \ \
    \Sm'=-2\Omega\Sm.
  \end{equation}
  Moreover, since we assume that the solution is not isotropic, $(\Sp,\Sm)\neq (0,0)$. This means that $(\Sp,\Sm)$ moves away from the origin
  radially as time decreases. Due to the bounds on $(\Sp,\Sm)$, this implies that $(\Sp,\Sm)$ converges to a limit, say $(\sigma_{+},\sigma_{-})$
  as $\tau\rightarrow-\infty$. Moreover, in the non-LRS setting, this limit cannot belong to any of the line
  segments connecting a special point to its antipodal point. Combining this observation with the constraint yields the conclusion
  that $\Psi$ converges to a limit. Since one of $\Sigma_{\pm}$ has to be non-vanishing, the equations (\ref{eq:SpmprBtI}) imply that $\Omega$ has
  to be integrable to the past. Next, note that
  \[
  \theta(\tau)=\theta(\tau_{0})\exp\big(-\textstyle{\int}_{\tau_{0}}^{\tau}[1+q(s)]ds\big)
  =\theta(\tau_{0})\exp\big(-3(\tau-\tau_{0})+\textstyle{\int}_{\tau_{0}}^{\tau}2\Omega(s)ds\big).
  \]
  Since $\Omega$ is integrable to the past, we conclude that there is a $\theta_{\infty}>0$ such that $e^{3\tau}\theta(\tau)$ converges to
  $\theta_{\infty}$.

  Turning to the Bianchi type II case, note that $\Sm\neq 0$ (since $\Sm=0$ would imply that the solution is LRS). Moreover, $S_{-}=0$, so that
  combining (\ref{eq:spmp}) and (\ref{eq:altqdef}) yields
  \[
  \Sm'=-\tfrac{1}{2}(3\No^{2}+4\Omega)\Sm.
  \]
  This implies that both $\No^{2}$ and $\Omega$ are integrable to the past. Moreover, $|\Sm|$ increases to the past, so that $\Sm$, since it is
  bounded, has to converge to a non-zero value $\sigma_{-}$. Due to these observations, both of the terms on the right hand side of the equation
  for $\Sp$, see (\ref{eq:spmp}), 
  are integrable to the past, so that $\Sp$ has to converge to a limit, say $\sigma_{+}$. Next, since $\No^{2}$ is integrable and $\No\No'$ is
  bounded, it follows that $\No$ has to converge to zero. Since $q-2$ is integrable to the past, the desired conclusion concerning $\theta$
  follows by an argument similar to the one given in the case of Bianchi type I. Turning to limitations on $\sigma_{+}$, note that if
  $\sigma_{+}>1/2$, then (\ref{eq:nip}) and (\ref{eq:fodef}), combined with the fact that $q-2$ is integrable to the past, means that $\No^{2}$
  diverges to infinity, contradicting the conclusion that $\No$ converges to $0$. Thus $\sigma_{+}\leq 1/2$. Finally, the convergence of $\Psi$
  follows from the constraint.
  
  Assume now, in addition, that $V\in \mfP_{\a_V}^1$ for some $\a_V\in (0,1)$. Combining the Hamiltonian constraint (\ref{eq:constraint}) with the
  definition of $\Psi$, (\ref{eq:PsiOmegadef}), it is clear that $|\phi'|\leq \sqrt{6}$, so that there is a constant $C_{\phi}\geq 0$ with the property
  that
  \begin{equation}\label{eq:phi coarse estimate}
    |\phi(\tau)|\leq \sqrt{6}|\tau|+C_{\phi}
  \end{equation}
  for all $\tau\leq 0$. This means that there is a constant $C$ with the property that
  $\Omega(\tau)\leq Ce^{6(1-\a_{V})\tau}$ for all $\tau\leq 0$. Next, we wish to prove that $\No$ converges to zero exponentially. Note, to
  this end, that
  \begin{equation}\label{eq:SpprBtIIver}
    \Sp'=\tfrac{3}{2}\No^{2}(2-\Sp)-2\Omega\Sp\geq -2|\Omega|\geq -Ce^{6(1-\a_{V})\tau}
  \end{equation}
  for all $\tau\leq 0$, where we appealed to (\ref{eq:spmp}), (\ref{eq:Spdef}) and (\ref{eq:altqdef}). Integrating this estimate
  from $-\infty$ to $\tau$ yields the conclusion that
  \[
  \Sp(\tau)\geq \sigma_{+}-Ce^{6(1-\a_{V})\tau}
  \]
  for all $\tau\leq 0$. From the above, we know that $\sigma_{+}\leq 1/2$. We now wish to exclude the possibility that $\sigma_{+}=1/2$. Assume, therefore,
  that $\sigma_{+}=1/2$. Then
  \[
  q-4\Sp\leq q-2+Ce^{6(1-\a_{V})\tau}.
  \]
  However, we already know $q-2$ to be integrable to the past. This means that there is a constant $C$ such that
  \[
  |\No(\tau)|=|\No(0)|\exp\big(-\textstyle{\int}_{\tau}^{0}[q(s)-4\Sp(s)]ds\big)\geq |\No(0)| e^{-C}
  \]
  for all $\tau\leq 0$. However, this is incompatible with the fact that $\No$ converges to zero. Thus $\sigma_{+}<1/2$. This implies that
  $\No$ converges to zero exponentially. Combining this fact with (\ref{eq:SpprBtIIver}) yields the conclusion that $\Sp$ converges exponentially
  to its limit. This means that $q-4\Sp=2-4\sigma_{+}+\dots$, where the dots represent exponentially decaying terms. Thus $\No=O[e^{(2-4\sigma_{+})\tau}]$.
  Moreover, there is an $\a>0$ such that $\Sigma_{\pm}-\sigma_{\pm}=O(e^{\a\tau})$ and $q-2=O(e^{\a\tau})$. Next, note that $V'\circ\phi/\theta^{2}$ decays
  to zero exponentially at the same rate as $\Omega$; this follows from $V\in \mfP_{\a_V}^1$ and an argument similar to the above. Considering
  (\ref{eq:phioneev}), it is thus clear that $\phi_{1}$ converges exponentially to a limit, say $\Phi_{1}$. This can be integrated to conclude that
  $\phi-\Phi_{1}\cdot\tau$ converges exponentially to a limit. In fact, (\ref{eq:BIphias}) holds. The estimates (\ref{eq:BISpmas}) and
  (\ref{eq:BIlnthetaas}) follow from the equations and the improved bounds on $\Omega$ and $\No$. The estimate (\ref{eq:BIINzas}) follows by integrating
  (\ref{eq:nip}), keeping (\ref{eq:fodef}) and (\ref{seq:asBI}) in mind. The equality (\ref{eq:asymptotic hamiltonian constraint}) follows by taking the
  limit of (\ref{eq:constraint}). The final statement of the lemma concerning $\a$ follows by repeating the above arguments with the improved estimate for
  $\phi$ in mind. 
\end{proof}
Next, we make a technical observation concerning Bianchi type VII${}_{0}$ solutions.
\begin{lemma}\label{lemma:NtNthbdBtVIIZ}
  Consider a Bianchi type VII${}_{0}$ development satisfying the conditions of Lemma~\ref{lemma:BianchiAdevelopment}. Assume that
  $\No=0$. Then $(\Nt\Nth)(\tau)$ is bounded for $\tau\leq 0$.
\end{lemma}
\begin{proof}
  By assumption, the development is neither isotropic nor LRS. Thus $(\Nt-\Nth)^{2}+\Sm^{2}$ is always non-zero. Define
  \[
  Z_{-1}:=\tfrac{\tfrac{4}{3}\Sm^{2}+(\Nt-\Nth)^{2}}{\Nt\Nth};
  \]
  \index{$\a$Aa@Notation!Expansion normalised!Expressions!Zminusone@$Z_{-1}$}%
  see \cite[p.~1429]{wah} and \cite[p.~63]{bogo}. Then $Z_{-1}$ is strictly positive and it can be computed that
  \[
  Z_{-1}'=-\tfrac{16}{3}\tfrac{(\Sp+1)\Sm^{2}}{\tfrac{4}{3}\Sm^{2}+(\Nt-\Nth)^{2}}Z_{-1}.
  \]
  In particular, $Z_{-1}$ increases to the past. This means that for $\tau\leq 0$,
  \[
    (\Nt\Nth)(\tau)\leq \tfrac{1}{Z_{-1}(0)}\left(\tfrac{4}{3}\Sm^{2}+(\Nt-\Nth)^{2}\right)(\tau)\leq\tfrac{4}{3Z_{-1}(0)},
  \]
  where we appealed to (\ref{eq:constraint}) in the last step. The lemma follows. 
\end{proof}
The goal of the next lemma is to provide a condition ensuring that $\Omega$ decays to zero exponentially in the direction of
the singularity.
\begin{lemma}\label{lemma:X growth}
  Consider a development satisfying the conditions of Lemma~\ref{lemma:BianchiAdevelopment} which is not an isotropic Bianchi type I development.
  Then there is a constant $c_\theta>0$ such that $\theta(\tau)\geq c_\theta e^{-\tau}$ for all $\tau\leq 0$. Assuming, in addition, that
  $V\in \mfP_{\a_V}^{0}$ for some $\a_V\in (0,1/3)$, it follows that $\Omega$ decays exponentially. In fact, there is a constant $C_{\Omega}$ such that,
  for all $\tau\leq 0$,
  \begin{equation}\label{eq:Omegadec prel}
    \Omega(\tau)\leq C_{\Omega}e^{2(1-3\a_{V})\tau}.
  \end{equation}
  Assume, in addition to the above, that the solution is of Bianchi type VI${}_{0}$ or VII${}_{0}$ with $\No=0$. Then there are constants
  $\theta_\infty>0$ and $C_{\Omega}$ such that $e^{3\tau}\theta(\tau)\rightarrow\theta_\infty$ and 
  \begin{equation}\label{eq:Omegadec}
    \Omega(\tau)\leq C_{\Omega}e^{6(1-\a_V)\tau}
  \end{equation}
  for all $\tau\leq 0$. Finally, there is a $\sigma_{+}>-1$ such that
  \begin{equation}\label{eq:BVIzVIIzprelas}
    \lim_{\tau\rightarrow-\infty}\Sp(\tau)=\sigma_{+},\ \ \
    \lim_{\tau\rightarrow-\infty}q(\tau)=2,\ \ \
    \lim_{\tau\rightarrow-\infty}(\Nt-\Nth)(\tau)=0.
  \end{equation}  
\end{lemma}
\begin{remark}
  In the case of Bianchi type VI${}_{0}$, the last equality in (\ref{eq:BVIzVIIzprelas}) implies that both $\Nt$ and $\Nth$ converge to zero.
\end{remark}
\begin{proof}  
  Due to (\ref{eq:dotXest}), $X'\leq -2X$. Thus $X(\tau)\geq X(0)e^{-2\tau}$ for all $\tau\leq 0$. Since $X(0)>0$ due to the assumptions,
  see the proof of Lemma~\ref{lemma:BianchiAdevelopment}, and since $\theta^{2}\geq X$, it is clear that there is a constant $c_\theta>0$
  such that $\theta(\tau)\geq c_\theta e^{-\tau}$ for all $\tau\leq 0$. If (\ref{eq:phi coarse estimate}) holds, we conclude that
  (\ref{eq:Omegadec prel}) holds, since $V\in \mfP_{\a_V}^{0}$ and $\a_V\in (0,1/3)$. However, (\ref{eq:phi coarse estimate}) holds if
  $\mfT\neq\mrIX$. Moreover, if $\mfT=\mrIX$, the fact that $X(t)>0$ for all $t\leq t_0$ implies that $|\phi'|\leq\sqrt{6}$ for $\tau\leq 0$, so
  that (\ref{eq:phi coarse estimate}) holds for $\mfT=\mrIX$ as well. 

  Assume now that the solution is of Bianchi type VI${}_{0}$ or VII${}_{0}$ with $\No=0$. Note that 
  \begin{equation}\label{eq:SpprVIzVIIz}
    \Sp'=-\tfrac{3}{2}(\Nt-\Nth)^{2}(\Sp+1)-2\Omega\Sp
  \end{equation}  
  due to (\ref{eq:spmp}), (\ref{eq:Spdef}) and (\ref{eq:altqdef}). Due to (\ref{eq:Omegadec prel}), the second term on the right hand side
  converges to zero exponentially. Next, note that $\Sp(\tau)$ cannot converge to $-1$ as $\tau\rightarrow-\infty$ due to
  Proposition~\ref{prop:limitcharsp}.
  This means that there is an $\eta>0$ and a sequence $\tau_{k}\rightarrow-\infty$ such that $1+\Sp(\tau_{k})\geq\eta$. Consider (\ref{eq:SpprVIzVIIz}).
  The first term is non-positive and only causes $\Sp$ to increase towards the past. The second term converges to zero exponentially
  (since $|\Sp|\leq 1$). This means that for $k$ large enough
  \[
  \big|\textstyle{\int}_{-\infty}^{\tau_{k}}2\Omega(s)\Sp(s)ds\big|\leq \tfrac{\eta}{2}.
  \]
  In particular, for $k$ large enough, $\Sp(\tau)+1\geq\eta/2$ for all $\tau\leq \tau_{k}$. Combining this observation with (\ref{eq:SpprVIzVIIz}), we
  conclude that $(\Nt-\Nth)^{2}$ is integrable to the past. Combining the fact that $\Omega$ and $(\Nt-\Nth)^{2}$ are integrable to the past with
  (\ref{eq:altqdef}) yields the conclusion that $q-2$ is integrable to the past. Thus $e^{3\tau}\theta(\tau)\rightarrow\theta_\infty$ and 
  (\ref{eq:Omegadec}) hold by arguments similar to ones presented in the proof of Lemma~\ref{lemma:BIBIIasympt}.

  Combining (\ref{eq:SpprVIzVIIz}) with the fact that $(\Nt-\Nth)^{2}$ and $\Omega$ are integrable to the past, we conclude that $\Sp(\tau)$
  converges to a limit $\sigma_{+}$ as $\tau\rightarrow-\infty$. By the above, $\sigma_{+}>-1$. Since $(\Nt-\Nth)^{2}$ is integrable to the past
  and has bounded derivative to the past (this follows from the constraint (\ref{eq:constraint}) and Lemma~\ref{lemma:NtNthbdBtVIIZ}), it follows
  that $\Nt-\Nth$ converges to zero. Since $\Omega$ and $\Nt-\Nth$ converge to zero, (\ref{eq:altqdef}) yields the conclusion that $q$ converges to $2$.
  The lemma follows.  
\end{proof}
It is of interest to note that this lemma has the following corollary.
\begin{cor}\label{cor:ap eq pp}
  Let $0\leq V\in \mfP_{\a_V}^{0}$ for some $\a_V\in (0,1/3)$. Then $\mB_{\mrIX,+}[V]\subset \mB_{\mrIX,\roap}[V]$, so that
  $\mB_{\mrIX,\roap}[V]=\mB_{\mrIX,\ropp}[V]$; see Definition~\ref{definition:Bap Bp Bpp}. 
\end{cor}
\begin{proof}
  Since there is a $T$ such that $X(t)>0$ for all $t\leq T$ due to Lemma~\ref{lemma:BianchiAdevelopment} and the assumptions, all that remains
  to be proven is that $\Omega$ converges to zero. However, that $\Omega$ converges to zero follows from Lemma~\ref{lemma:X growth}.
\end{proof}

\section{The vacuum and matter dominated cases}\label{ssection:vacmatdom}

Next, we demonstrate that there is a dichotomy: either the expansion normalised energy density decays exponentially towards the past, or it
converges to a non-zero contribution.

\begin{thm}\label{thm:dichotomy}
  Consider a development satisfying the conditions of Lemma~\ref{lemma:BianchiAdevelopment} which is not an isotropic Bianchi type I development. Assume
  $V\in \mfP_{\a_V}^1$, where
  $\a_V\in (0,1)$ in case of Bianchi type I and non-LRS Bianchi type II; and $\a_V\in (0,1/3)$ otherwise. Then either there is a constant $C$ such that
  \begin{equation}\label{eq:Omegaphiprstrongdec}
    \Omega(\tau)+|\phi'(\tau)|\leq C\exp\big(-\textstyle{\int}_{\tau}^{0}2[1+q(s)]ds\big)
  \end{equation}
  for all $\tau\leq 0$; or there are constants $\a_{\infty}>0$, $\theta_{\infty}>0$, $C_{\infty}$, $\sigma_{\pm}$, $\Phi_{1}\neq 0$ and $\Phi_{0}$ such that
  \begin{subequations}\label{seq:SpmphiprNilimmatterdom}
    \begin{align}
      |\Sp(\tau)-\sigma_{+}|+|\Sm(\tau)-\sigma_{-}| \leq & C_{\infty}e^{\a_{\infty}\tau},\label{eq:SpSmlimmatterdom}\\
      |\phi'(\tau)-\Phi_{1}|+|\phi(\tau)-\Phi_{1}\tau-\Phi_{0}| \leq & C_{\infty}e^{\a_{\infty}\tau},\label{eq:phiprphilimmatterdom}\\
      |\No(\tau)|+|\Nt(\tau)|+|\Nth(\tau)| \leq & C_{\infty}e^{\a_{\infty}\tau},\label{eq:Nilimmatterdom}\\
      |\ln\theta(\tau)+3\tau-\ln\theta_{\infty}| \leq & C_{\infty}e^{\a_{\infty}\tau},\label{eq:lnthetalimmatterdom}\\
      \Omega(\tau) \leq & C_\infty e^{6(1-\a_V)\tau}
    \end{align}
  \end{subequations}  
  for all $\tau\leq 0$. Moreover, in the latter case, if $N_i\neq 0$, then $f_i(2,\sigma_+,\sigma_-)>0$.

  Assume, on the contrary, that the conditions of Lemma~\ref{lemma:BianchiAdevelopment} are not satisfied, but that the conditions of
  Lemma~\ref{lemma:Bianchi IX remainder} are. Assume, moreover, that $V\in \mfP_{\a_V}^1$ for some $\a_V\in (0,1)$. Then the second case described
  above holds. In particular, (\ref{seq:SpmphiprNilimmatterdom}) holds. Moreover, $\sigma_\pm=0$ and $\Phi_1^2=6$.  
\end{thm}
\begin{remark}
  Due to this result, Definition~\ref{def:matter and vacuum dominated} is meaningful, and below we speak of matter and vacuum dominated developments. 
\end{remark}
\begin{proof}
  Assume, to begin with, that the conditions of the first part of the lemma are satisfied. In particular, the conditions of
  Lemma~\ref{lemma:BianchiAdevelopment} are then satisfied, so that the development is not of LRS Bianchi type VII${}_0$. Let
  \begin{equation}\label{eq:Fdefsh}
    F(\tau)=\exp\big(-\textstyle{\int}_{\tau}^{0}2[2-q(s)]ds\big)\Psi(\tau).
  \end{equation}
  Then (\ref{eq:Psieq}) yields
  \begin{equation}\label{eq:F prime id}
    F'(\tau)=4\Omega(\tau)\exp\big(-\textstyle{\int}_{\tau}^{0}2[2-q(s)]ds\big).
  \end{equation}
  In particular, $F$ is bounded from below by $0$ and is decreasing to the past. There is thus a constant
  $F_{0}\geq 0$ such that
  \begin{equation}\label{eq:Fint}
    \begin{split}
      F(\tau)-F_{0} = & \textstyle{\int}_{-\infty}^{\tau}4\Omega(s)\exp\big(-\textstyle{\int}_{s}^{0}2[2-q(u)]du\big)ds\\
      = & 18\textstyle{\int}_{-\infty}^{\tau}\tfrac{V[\phi(s)]}{\theta^{2}(0)}\exp\big(-2\textstyle{\int}_{s}^{0}[1+q(u)]du\big)
      \exp\big(-\textstyle{\int}_{s}^{0}2[2-q(u)]du\big)ds\\
      = & \tfrac{18}{\theta^{2}(0)}\textstyle{\int}_{-\infty}^{\tau}V[\phi(s)]e^{6s}ds.
    \end{split}
  \end{equation}
  Since $V[\phi(\tau)]\leq Ce^{-6\a_{V}\tau}$ for $\tau\leq 0$, cf. the beginning of the proof of Lemma~\ref{lemma:X growth}, (\ref{eq:Fint}) yields
  \begin{equation}\label{eq:Fconvest}
    |F(\tau)-F_{0}|\leq Ce^{6(1-\a_V)\tau}
  \end{equation}
  for $\tau\leq 0$, so that 
  \[
  \big|\Psi(\tau)-F_{0}\exp\big(\textstyle{\int}_{\tau}^{0}2[2-q(s)]ds\big)\big|
  \leq C\exp\big(-\textstyle{\int}_{\tau}^{0}2[1+q(s)]ds-6\a_{V}\tau\big).
  \]
  Let us first consider the case that $F_{0}=0$. Then
  \begin{equation}\label{eq:Psiasympt}
    \Psi(\tau)\leq C\exp\big(-\textstyle{\int}_{\tau}^{0}2[1+q(s)]ds-6\a_{V}\tau\big).
  \end{equation}
  In the case of anisotropic Bianchi type I and non-LRS Bianchi type II, $\a_{V}\in (0,1)$ and $q$ converges to
  $2$; see Lemma~\ref{lemma:BIBIIasympt}. Combining these observations with (\ref{eq:Psiasympt}) yields the conclusion
  that $\Psi$ converges to zero exponentially. In the remaining cases, $\theta(\tau)\geq c_\theta e^{-\tau}$, see Lemma~\ref{lemma:X growth},
  and $\a_{V}\in (0,1/3)$. This, again, yields the conclusion that $\Psi$ converges to zero exponentially. This means that $\phi$ converges as
  $\tau\rightarrow-\infty$, so that $V\circ\phi$ is bounded to the past. Going through the above estimates with this information in mind yields the
  conclusion that
  \begin{equation}\label{eq:Psi est Fz zero}
    \Psi(\tau)\leq C\exp\big(-\textstyle{\int}_{\tau}^{0}2[1+q(s)]ds\big)
  \end{equation}
  for $\tau\leq 0$. On the other hand, (\ref{eq:PsiOmegadef}) and (\ref{eq:thetaprime}) yield
  \[
  \Psi= \tfrac{3}{2}\tfrac{\phi_{t}^{2}+2V\circ\phi}{\theta^{2}}\geq\tfrac{3V_{\inf}}{\theta^{2}(0)}
  \exp\big(-2\textstyle{\int}_{\tau}^{0}[1+q(s)]ds\big),
  \]
  where we use the fact that $V$ has a non-negative infimum, $V_{\inf}$. In particular, it is thus
  natural to introduce the quantity
  \[
  P(\tau)=\exp\big(\textstyle{\int}_{\tau}^{0}2[1+q(s)]ds\big)\Psi(\tau).
  \]
  Appealing to (\ref{eq:Psieq}) yields
  \begin{equation}\label{eq:Pprime}
    P'(\tau)=\exp\big(\textstyle{\int}_{\tau}^{0}2[1+q(s)]ds\big)(-6\Psi+4\Omega)=-9\tfrac{\phi_{t}^{2}}{\theta^{2}}
    \exp\big(\textstyle{\int}_{\tau}^{0}2[1+q(s)]ds\big)=-9\tfrac{\phi_{t}^{2}}{\theta^{2}(0)}.
  \end{equation}
  In particular, $P$ increases to the past (and, by the above, we know $P$ to be bounded). Thus there
  is a $P_{0}\in\ro$ such that $P(\tau)\rightarrow P_{0}$ as $\tau\rightarrow -\infty$. Since 
  \[
  P=\tfrac{3}{2}\tfrac{\phi_{t}^{2}+2V\circ\phi}{\theta^{2}(0)},
  \]
  and since $\phi$ converges to a limit, say $\phi_{0}$, it is clear that $\phi_{t}^{2}$ has to converge to a limit. Keeping (\ref{eq:Pprime})
  in mind, it is clear that $\phi_{t}$ has to converge to zero. Consider (\ref{eq:phiddot}). We know that $\phi$ converges to a limit, say
  $\phi_{0}$. Let $C$ be a constant such that $C\geq 2|V'(\phi_{0})|+1$. Then there is a $t_{0}>0$ such that $C>|V'[\phi(t)]|+1/2$ for
  $t\leq t_{0}$. Next, we wish to prove, by contradiction, that there is a $T\in (0,t_{0})$ such that $|(\theta\phi_{t})(t)|\leq C$ for all
  $t\leq T$. Assume, therefore, that there is a sequence $t_{k}\downarrow 0$ such that $|(\theta\phi_{t})(t_{k})|>C$. Fix a $t>0$ such that
  $\phi_{t}(t)\neq 0$. Then, due to (\ref{eq:phiddot}), (\ref{eq:dtdtau}) and (\ref{eq:thetaprime}),
  \[
  \d_{t}(\theta\phi_{t})=-\left(\tfrac{4}{3}+\tfrac{2}{3}(\Sp^{2}+\Sm^{2}+\Psi)
  -\tfrac{2}{3}\Omega+\tfrac{V'\circ\phi}{\theta\phi_{t}}\right)\theta^{2}\phi_{t}
  \]
  at $t$. Fix $T$ so that $\Omega(t)\leq 1/2$ for $t\leq T$. Then, assuming $t\leq \min\{T,t_{0}\}$ and $|(\theta\phi_{t})(t)|\geq C$,
  \[
  \tfrac{2}{3}\Omega+\left|\tfrac{V'\circ\phi}{\theta\phi_{t}}\right|\leq \tfrac{1}{3}+\tfrac{C-1/2}{C}=\tfrac{4}{3}-\tfrac{1}{2C}.
  \]
  This means that if $t\leq \min\{T,t_{0}\}$ and $|(\theta\phi_{t})(t)|\geq C$, then $|\theta\phi_{t}|$ increases to the past. This
  means that $|(\theta\phi_{t})(s)|\geq C$ for all $s\in (0,t)$. Combining this conclusion with (\ref{eq:phiddot}) and the fact that
  $|V'[\phi(s)]|\leq C-1/2$ for $s\in (0,t)$ yields the conclusion that $|\phi_{t}|$ is non-zero and increasing to the past. This means
  that it converges to a non-zero value, contradicting previous conclusions. Thus $|(\theta\phi_{t})(t)|<C$ for all $t\leq \min\{T,t_{0}\}$.
  This means that
  \[
  |\phi'(\tau)|\leq C\exp\left(-\textstyle{\int}_{\tau}^{0}2[1+q(s)]ds\right)
  \]
  for all $\tau\leq 0$. Combining this estimate with (\ref{eq:Psi est Fz zero}) yields (\ref{eq:Omegaphiprstrongdec}). 

  Next, assume $F_{0}>0$. Combining (\ref{eq:Fdefsh}) and (\ref{eq:Fconvest}) then yields
  \begin{equation}\label{eq:Psi ito Fz etc}
    \Psi(\tau)=\exp\left(\textstyle{\int}_{\tau}^{0}2[2-q(s)]ds\right)\left[F_{0}+O(e^{6(1-\a_V)\tau})\right].
  \end{equation}
  If $\mfT\neq\mrIX$, then $\Psi\leq 1$ due to the constraint. If $\mfT=\mrIX$, then $X(t)>0$ for $t\leq t_0$, see 
  Lemma~\ref{lemma:BianchiAdevelopment}. Due to (\ref{eq:Xdef}) and (\ref{eq:PsiOmegadef}), this means that $\Psi\leq 1$ in case
  $\mfT=\mrIX$. Thus $\Psi(\tau)\leq 1$ for all $\tau\leq 0$. If $\mfT\neq\mrIX$, then the first factor on the right hand side of
  (\ref{eq:Psi ito Fz etc}) is monotonically increasing to the past; note that $2-q\geq 0$ if $\mfT\neq\mrIX$ due to (\ref{eq:altqdef}).
  Combining this observation with (\ref{eq:Psi ito Fz etc}) and the fact that $\Psi\leq 1$ implies that $\Psi$ converges to a strictly
  positive number, and that
  \begin{equation}\label{eq:qmtint}
    \textstyle{\int}_{-\infty}^{0}2|2-q(s)|ds<\infty.
  \end{equation}
  Next, let $\mfT=\mrIX$ and define
  \begin{align*}
    Z_N := & (N_1N_2N_3)^{1/3}\exp\big(-\textstyle{\int}_{\tau}^{0}[2-q(s)]ds\big),\\
    Z := & \frac{(N_1N_2N_3)^{1/3}}{\Psi^{1/2}}=Z_N\big[F_{0}+O(e^{6(1-\a_V)\tau})\big]^{-1/2}.
  \end{align*}
  Then $Z_N'=2Z_N$, so that $Z_N(\tau)=e^{2\tau}Z_N(0)$. This means that
  \begin{equation}\label{eq:Ni prod est}
    (N_1N_2N_3)^{1/3}\leq e^{2\tau}Z_N(0)\big[F_{0}+O(e^{6(1-\a_V)\tau})\big]^{-1/2}\leq Ce^{2\tau}
  \end{equation}
  for all $\tau\leq 0$. Combining this estimate with Lemma~\ref{lemma: bS low bd} and (\ref{eq:altqdef}) yields
  \begin{equation}\label{eq:two minus q lower bound pf}
    2-q\geq -Ce^{4\tau}
  \end{equation}
  for all $\tau\leq 0$. Combining this estimate with (\ref{eq:Psi ito Fz etc}) and the fact that $\Psi\leq 1$ yields the conclusion that
  (\ref{eq:qmtint}) holds when $\mfT=\mrIX$. 
    
  Introduce $\omega=\Psi^{1/2}$. Then
  \[
  \omega'=-\left[(2-q)-\tfrac{2\Omega}{\omega^{2}}\right]\omega.
  \]
  Assuming $\Nt\Nth\neq 0$ and defining $g=|\Nt\Nth|/\omega$ then yields
  \[
  g'=\left[2(\Sp+1)^{2}+2\Sm^{2}+2\Psi-2\Omega-\tfrac{2\Omega}{\omega^{2}}\right]g.
  \]
  Since $\Omega$ and $\Omega/\omega^{2}$ decay to zero exponentially and $\Psi$ converges to a positive number, it is clear that $g$ converges to
  zero exponentially as $\tau\rightarrow-\infty$. Since $\omega$ is bounded, this implies that $\Nt\Nth\rightarrow 0$ exponentially as
  $\tau\rightarrow-\infty$. The arguments concerning $\No\Nt$ and $\No\Nth$ are similar; these quantities both converge to zero exponentially as
  $\tau\rightarrow-\infty$. Note also that, due to (\ref{eq:altqdef}),
  \begin{equation}\label{eq:qminustwoaltex}
  2-q=2\Omega+\tfrac{3}{2}[\No^2+\Nt^2+\Nth^2-2(\No\Nt+\Nt\Nth+
    \No\Nth)].
  \end{equation}
  Since $2-q, \Omega,\No\Nt,\No\Nth,\Nt\Nth\in L^{1}((-\infty,0])$, this equality implies that $N_{i}^{2}\in L^{1}((-\infty,0])$. Moreover, due to
  (\ref{eq:constraint}) and the fact that $\No\Nt,\No\Nth,\Nt\Nth$ all converge to zero, $\Psi$, $\Sigma_\pm$ and the $N_{i}$ are bounded to the past
  (so that $N_{i}'$ is bounded to the past). Thus $N_{i}(\tau)\rightarrow 0$ as $\tau\rightarrow-\infty$, so that $q(\tau)\rightarrow 2$ as
  $\tau\rightarrow-\infty$. By the above, all quadratic polynomials in the $N_{i}$'s are integrable to the past. Combining this observation with the
  fact that  $2-q$ is integrable to the past and the equations (\ref{eq:spmp}) yields the conclusion that $\Sp$ and $\Sm$
  converge to the past. Thus there are constants $\sigma_{\pm}$, $\psi$ such that 
  \[
  \lim_{\tau\rightarrow-\infty}\Sp(\tau)=\sigma_{+},\ \ \
  \lim_{\tau\rightarrow-\infty}\Sm(\tau)=\sigma_{-},\ \ \
  \lim_{\tau\rightarrow-\infty}\Psi(\tau)=\psi.
  \]
  Moreover, $\sigma_{+}^{2}+\sigma_{-}^{2}+\psi=1$. Since the $N_{i}$ converge to zero we know that if $N_i\neq 0$, then $f_i(2,\sigma_+,\sigma_-)\geq 0$.
  The question is then if equality can hold. Assume, therefore, without loss of generality, that $\sigma_+=1/2$ and that $\No\neq 0$. Then $\Nt$,
  $\Nth$, $\No\Nt$, $\No\Nth$ and $\Omega$ converge to zero exponentially (since $\sigma_+^2+\sigma_-^2=1-\psi<1$), so that, due to (\ref{eq:Spdef})
  and (\ref{eq:altqdef}),
  \[
  2-q=\tfrac{3}{2}\No^{2}+\dots,\ \ \
  S_{+}=-\No^{2}+\dots,
  \]
  where the dots signify exponentially decaying terms. Combining this with (\ref{eq:spmp}) yields
  \[
  \Sp'=\tfrac{3}{2}\No^{2}(2-\Sp)+\dots,
  \]
  where the dots signify exponentially decreasing terms. Since there, for every $\e>0$, is a constant $c_{\e}>0$ such that
  $\No^{2}(\tau)\geq c_{\e}e^{\e\tau}$ for $\tau\leq 0$, we conclude that $\Sp'$ is strictly positive in the limit. This means that there is a
  $T\leq 0$ such that $\Sp(\tau)>1/2$ for $\tau\leq T$. Moreover, due to (\ref{eq:altqdef}), $q(\tau)$ can be assumed to be strictly less
  than $2$ for $\tau\leq T$. Combining these observations, we conclude that $q(\tau)-4\Sp(\tau)<0$ for $\tau\leq T$. This is inconsistent
  with the fact that $\No$ converges to zero. We conclude that if $\No\neq 0$, then $\sigma_{+}<1/2$; i.e., $f_1(2,\sigma_+,\sigma_-)>0$.
  By applying the symmetries, we obtain the analogous conclusions concerning the other limiting cases. Given the above information, the
  asymptotics (\ref{seq:SpmphiprNilimmatterdom}) can finally be derived. We leave the details to the reader. The first part of the
  theorem follows.

  Assume now that the conditions of Lemma~\ref{lemma:BianchiAdevelopment} are not satisfied, but that the conditions of
  Lemma~\ref{lemma:Bianchi IX remainder} are. Due to (\ref{eq:X ver HC}) and the assumptions, $\phi_\tau^2$ then converges to $6$.
  Next, due to (\ref{eq:theta t ito X}), there is, for every $\e>0$, a $T\in J$ such that $\theta_t\leq -(1-\e)\theta^2$ for all $t\leq T$,
  so that
  \[
  \theta'\leq -3(1-\e)\theta
  \]
  for $t\leq T$. For every $\e>0$, there is thus a $\tau_\e$ such that
  \[
  \theta(\tau)\geq \exp[-3(1-\e)\tau],\ \ \
  |\phi(\tau)|\leq -\sqrt{6}(1+\e)\tau
  \]
  for all $\tau\leq \tau_\e$. Due to the assumptions, this means that $\Omega$ converges to zero exponentially. Combining this observation with
  (\ref{eq:Fdefsh}), (\ref{eq:F prime id}) and the fact that $\Psi\geq 1$, it follows that
  \[
  |F'(\tau)|\leq Ce^{\eta\tau}F(\tau)
  \]
  for some $\eta>0$ and $\tau\leq \tau_\e$. In particular, $F$ converges to a non-zero number to the past. Moreover, since
  $V\circ\phi(\tau)\leq Ce^{-6\a_V(1+\e)\tau}$ in the present setting, (\ref{eq:Fconvest}) is replaced by
  \begin{equation}\label{eq:Fconvest P pos}
    |F(\tau)-F_{0}|\leq Ce^{6[1-(1+\e)\a_V]\tau}
  \end{equation}
  for all $\tau\leq\tau_\e$. This means that (\ref{eq:Psi ito Fz etc}) is replaced by
  \begin{equation}\label{eq:Psi ito Fz etc P pos}
    \Psi(\tau)=\exp\left(\textstyle{\int}_{\tau}^{0}2[2-q(s)]ds\right)\left[F_{0}+O(e^{6[1-(1+\e)\a_V]\tau})\right].
  \end{equation}
  Combining this observation with the fact that $\Psi$ converges to $1$ and an argument similar to the proof of (\ref{eq:Ni prod est})
  and (\ref{eq:two minus q lower bound pf}) yields the conclusion that the polynomial on the left hand side of (\ref{eq:constraint}) is
  bounded from below by $-Ce^{4\tau}$. Since $\Psi$ converges to $1$, this means that $\Sigma_\pm\rightarrow 0$. Returning to
  (\ref{eq:qdef}) with this information at hand yields the conclusion that $q\rightarrow 2$. Combining this information with
  (\ref{eq:Niprprel}) yields the conclusion that all the $N_i$ are decaying exponentially. Due to (\ref{eq:altqdef}), it follows that
  $2-q$ decays exponentially. At this stage, it is straightforward to deduce the remaining conclusions. We leave the details to the reader. 
\end{proof}

\chapter{The vacuum dominated setting}\label{chapter:vacuum dom set}

Due to Theorem~\ref{thm:dichotomy}, there is a dichotomy: either $\phi_t/\theta$ converges to a non-zero limit, or it converges to zero. 
This leads to the notion of vacuum and matter dominated developments; see Definition~\ref{def:matter and vacuum dominated}. In the matter dominated
case, the conclusions of Theorem~\ref{thm:dichotomy} can be used to deduce detailed asymptotics. However, it is less clear what happens in the
vacuum dominated setting. The purpose of the present chapter is to derive the asymptotics in the vacuum dominated setting for anisotropic Bianchi types
I, II, VI${}_0$ and VII${}_0$ as well as for the isotropic and LRS Bianchi types VIII and IX. In the anisotropic and non-LRS Bianchi type VIII and IX
settings, we prove that the asymptotics are oscillatory. 

\section{Vacuum dominated Bianchi types I, II, VI${}_{0}$ and VII${}_{0}$}\label{ssection:vacdomIIIVIzVIIz}

Next, let us return to the case that (\ref{eq:Omegaphiprstrongdec}) holds. We first restrict our attention to Bianchi types I and II. 

\begin{cor}\label{cor:vBtIaII}
  Consider an anisotropic development satisfying the conditions of Lemma~\ref{lemma:BianchiAdevelopment}. Assume it to be of Bianchi type I or 
  II with $\No\neq 0$, and such that (\ref{eq:Omegaphiprstrongdec}) holds. Assume, moreover, that $V\in \mfP_{\a_V}^1$, where $\a_V\in (0,1)$ in
  case of Bianchi type I and non-LRS Bianchi type II; and $\a_V\in (0,1/3)$ otherwise. Then, in the case of Bianchi type I, there is a point
  $(\sigma_{+},\sigma_{-})$ on the Kasner circle (which is non-special if the solution is not LRS) and constants $\theta_{\infty}>0$ and $\Phi_{0}$
  such that
  \begin{subequations}\label{seq:asBIv}
    \begin{align}
      |\phi_{1}(\tau)|+|\phi(\tau)-\Phi_{0}| \leq & Ce^{6\tau},\label{eq:BIphiasv}\\
      |\Sp(\tau)-\sigma_{+}|+|\Sm(\tau)-\sigma_{-}| \leq & Ce^{6\tau},\label{eq:BISpmasv}\\
      |\ln\theta(\tau)+3\tau-\ln\theta_{\infty}| \leq & Ce^{6\tau}\label{eq:BIlnthetaasv}      
    \end{align}
  \end{subequations}  
  for all $\tau\leq 0$. In the the case of Bianchi type II, there is a point $(\sigma_{+},\sigma_{-})$ on the Kasner circle with
  $\sigma_{+}<1/2$ (which is non-special if the solution is not LRS) and constants $\theta_{\infty}>0$, $m_{1}$ and $\Phi_{0}$ such that
  \begin{subequations}\label{seq:asBIIv}
    \begin{align}
      |\phi_{1}(\tau)|+|\phi(\tau)-\Phi_{0}| \leq & Ce^{6\tau},\label{eq:BIIphiasv}\\
      |\Sp(\tau)-\sigma_{+}|+|\Sm(\tau)-\sigma_{-}| \leq & Ce^{\a\tau},\label{eq:BIISpmasv}\\
      |\ln\theta(\tau)+3\tau-\ln\theta_{\infty}| \leq & Ce^{\a\tau},\label{eq:BIIlnthetaasv}\\
      |\ln|\No(\tau)|-f_1(2,\sigma_+,\sigma_-)\tau-m_{1}| \leq & Ce^{\a\tau}\label{eq:BIINzasv}
    \end{align}
  \end{subequations}  
  for all $\tau\leq 0$, where $\a:=\min\{6,2f_1(2,\sigma_+,\sigma_-)\}$. 
\end{cor}
\begin{proof}
  Combining the results of Lemma~\ref{lemma:BIBIIasympt} with (\ref{eq:Omegaphiprstrongdec}), it
  is clear that in the case of Bianchi type I, $q-2=O(e^{6\tau})$ and $\phi'=O(e^{6\tau})$. Combining these observations with (\ref{eq:spmp}) and
  (\ref{eq:thetaprime}) yields (\ref{seq:asBIv}). Since $\Psi(\tau)\rightarrow 0$ as $\tau\rightarrow-\infty$, it is clear that $(\sigma_+,\sigma_-)$
  belongs to the Kasner circle. Due to the conclusions of Lemma~\ref{lemma:BIBIIasympt} it is also clear that $(\sigma_+,\sigma_-)$ can only be
  a special point when the solution is LRS. 

  Similarly, in the case of non-LRS Bianchi type II, Lemma~\ref{lemma:BIBIIasympt} implies that $q-2$ decays exponentially. Combining this observation
  with (\ref{eq:Omegaphiprstrongdec}) yields the conclusion that $\phi'=O(e^{6\tau})$. Integrating this estimate yields (\ref{eq:BIIphiasv}).
  Combining these observations with (\ref{eq:BIINzas}) yields the conclusion that $q-2=O(e^{\a\tau})$, where $\a=\min\{6,4(1-2\sigma_{+})\}$, and
  that $\No^{2}=O(e^{4(1-2\sigma_{+})\tau})$. Combining these observations with the equations yields (\ref{eq:BIISpmasv})--(\ref{eq:BIINzasv}). The
  remaining conclusions of the corollary in the non-LRS Bianchi type II setting follow from the fact that $\Psi$ converges to zero and the conclusions
  of Lemma~\ref{lemma:BIBIIasympt}.

  Finally, consider the LRS Bianchi type II setting. Then
  \[
  \Sp'=\tfrac{3}{2}\No^2(2-\Sp)-2\Omega\Sp.
  \]
  Since $\Omega$ converges to zero exponentially, see (\ref{eq:Omegadec prel}), and $|\Sp|\leq 1$, it is clear that $\No^2$ is integrable to the past
  (so that $\Sp(\tau)$ converges to a limit as $\tau\rightarrow-\infty$).
  This means that $q-2$ is integrable to the past. Combining this observation with (\ref{eq:Omegaphiprstrongdec}) yields (\ref{eq:BIIphiasv}).
  Since $\No^2$ is integrable and its derivative is bounded, we conclude that $\No$ converges to zero. Since, additionally, $\Psi$ converges to zero,
  due to (\ref{eq:Omegaphiprstrongdec}), and $\Sm=0$, the constraint (\ref{eq:constraint}) implies that $\Sp$ converges to either $1$ or $-1$. However,
  it cannot converge to $1$, since this would imply that $|\No|$ tends to infinity to the past. Combining these observations with the equations yields
  the remaining conclusions. 
\end{proof}

Next, let us turn to Bianchi types VI${}_{0}$ and VII${}_{0}$.

\begin{cor}\label{cor:vBtVIzaVIIz}
  Consider a development satisfying the conditions of Lemma~\ref{lemma:BianchiAdevelopment}. Assume it to be of Bianchi type
  VI${}_{0}$ or VII${}_{0}$ with $\No=0$, and such that (\ref{eq:Omegaphiprstrongdec}) holds. Assume, in addition, that $V\in \mfP_{\a_V}^{1}$
  for some $\a_V\in (0,1/3)$.
  Then there is a point $(\sigma_{+},\sigma_{-})$ on the Kasner circle with $\sigma_{+}>1/2$ and constants $\theta_{\infty}>0$, $m_{2}$, $m_{3}$ and
  $\Phi_{0}$ such that
  \begin{subequations}\label{seq:asBVIv}
    \begin{align}
      |\phi_{1}(\tau)|+|\phi(\tau)-\Phi_{0}| \leq & Ce^{6\tau},\label{eq:BVIaVIIphiasv}\\
      |\Sp(\tau)-\sigma_{+}|+|\Sm(\tau)-\sigma_{-}| \leq & Ce^{\a\tau},\label{eq:BVIaVIISpmasv}\\
      |\ln\theta(\tau)+3\tau-\ln\theta_{\infty}| \leq & Ce^{\a\tau},\label{eq:BVIaVIIlnthetaasv}\\
      |\ln|N_i(\tau)|-f_i(2,\sigma_+,\sigma_-)\tau-m_{i}| \leq & Ce^{\a\tau}\label{eq:BVIaVIINiasv}      
    \end{align}
  \end{subequations}  
  for all $\tau\leq 0$ and $i=2,3$, where $\a:=\min\{6,2f_2(2,\sigma_+,\sigma_-),2f_3(2,\sigma_+,\sigma_-)\}$.
\end{cor}
\begin{proof}
  Note that the conclusions of Lemma~\ref{lemma:X growth} hold in the present setting. 
  Combining (\ref{eq:constraint}) with (\ref{eq:BVIzVIIzprelas}) and (\ref{eq:Omegaphiprstrongdec}), it follows that $\Sm^{2}$ converges to
  $1-\sigma_{+}^{2}$. Since $\Sm$ is continuous, this implies that $\Sm$ converges to $\sigma_{-}$, where $(\sigma_{+},\sigma_{-})$ belongs to the
  Kasner circle. In order for this to be consistent with the fact that $\Nt$ and $\Nth$ are bounded to the past (in the case of Bianchi type
  VI${}_{0}$ this is an immediate consequence of the constraint, and in the case of Bianchi type VII${}_{0}$, it follows from
  Lemma~\ref{lemma:NtNthbdBtVIIZ} and the constraint), we have to have $1+\sigma_{+}+\sqrt{3}\sigma_{-}\geq 0$ and
  $1+\sigma_{+}-\sqrt{3}\sigma_{-}\geq 0$. Since $\sigma_{+}>-1$, this implies that $\sigma_{+}\geq 1/2$. However, assuming that $\sigma_{+}=1/2$,
  Proposition~\ref{prop:limitcharsp} and Remark~\ref{remark:limitcharsp} yield the conclusion that either $\Nt=0$ or $\Nth=0$, contradicting the
  assumptions. This means that $\sigma_{+}>1/2$. The remaining conclusions of the corollary follow by revisiting the equations using the above
  information. 
\end{proof}

\section{Vacuum dominated isotropic and LRS Bianchi types VIII and IX}\label{ssection:vacdomBVIII}

Next, we consider vacuum dominated LRS Bianchi type VIII solutions.
\begin{cor}\label{cor:LRS Bianchi type VIII}
  Consider a development satisfying the conditions of Lemma~\ref{lemma:BianchiAdevelopment} where $V\in\mfP_{\a_V}^1$ for some $\a_V\in (0,1/3)$. Assume
  it to be of LRS Bianchi type VIII with $\No<0$, $\Nt=\Nth$ and $\Sm=0$, and such that (\ref{eq:Omegaphiprstrongdec}) holds. Then there are constants
  $\theta_{\infty}>0$, $m_1$, $m_2=m_3$ and $\Phi_0$ such that
  \begin{subequations}
    \begin{align}
      |\Sp(\tau)+1| \leq & Ce^{6\tau},\label{eq:Sp limit LRS BVIII}\\
      |\ln |N_{i}(\tau)|-f_i(2,-1,0)\tau-m_i| \leq & Ce^{6\tau},\\
      |q(\tau)-2|+|\Omega(\tau)|+|\Psi(\tau)| \leq & Ce^{6\tau},\label{eq:q etc limit BVIII LRS}\\
      |\ln\theta(\tau)+3\tau-\ln\theta_{\infty}| \leq & Ce^{6\tau},\label{eq:ln theta as BVIII LRS}\\
      |\phi_{1}(\tau)|+|\phi(\tau)-\Phi_{0}| \leq & Ce^{6\tau}\label{eq:phi limit BVIII LRS}
    \end{align}
  \end{subequations}
  for all $\tau\leq 0$ and $i=1,2,3$. 
\end{cor}
\begin{proof}
  Combining Lemma~\ref{lemma:X growth} and (\ref{eq:Omegaphiprstrongdec}), it follows that $\Psi$ and $\Omega$ decay as $e^{2\tau}$. Next, note that
  (\ref{eq:spmp}), (\ref{eq:Spdef}), (\ref{eq:altqdef}) and the fact that $\Sm=0$ and $\Nt=\Nth$ yield
  \begin{equation}\label{eq:Sp prime BVIII}
    \Sp'=\tfrac{3}{2}\No^2(2-\Sp)-6\No\Nt\left(\tfrac{1}{2}-\Sp\right)-2\Omega\Sp.
  \end{equation}
  On the other hand, $|\Sp|\leq 1$ due to (\ref{eq:constraint}), so that $|2\Omega(\tau)\Sp(\tau)|\leq Ce^{2\tau}$ for all $\tau\leq 0$ and some constant
  $C>0$. Assume, for a moment, that there is a $\tau_1\leq 0$ such that $\Sp(\tau_1)\leq 1/2-Ce^{2\tau_1}$. We then wish to prove that $\Sp(\tau)\leq 1/2$
  for all $\tau\leq\tau_1$. Assume, in order to obtain a contradiction, that $\tau_2\leq\tau_1$ is such that $\Sp(\tau_2)=1/2$ and $\Sp(\tau)<1/2$ for
  $\tau_2<\tau\leq\tau_1$. Then the first two terms on the right hand side of (\ref{eq:Sp prime BVIII}) are non-negative on $[\tau_2,\tau_1]$. Thus
  $\Sp'(\tau)\geq -Ce^{2\tau}$ for all $\tau\in [\tau_2,\tau_1]$. Integrating this estimate yields
  \[
    \Sp(\tau_2)\leq \Sp(\tau_1)+\tfrac{C}{2}e^{2\tau_1}\leq \tfrac{1}{2}-\tfrac{C}{2}e^{2\tau_1}<\tfrac{1}{2},
  \]
  a contradiction. In other words, there are two possibilities. Either there is a $\tau_1$ such that $\Sp(\tau)\geq 1/2-Ce^{2\tau}$ for all $\tau\leq\tau_1$
  or there is a $\tau_1$ such that $\Sp(\tau)\leq s_+$ for all $\tau\leq \tau_1$ and some $s_+<1/2$. Let us consider the first case. Then
  \[
    q-4\Sp=2\Sp^2-4\Sp+2\Psi-2\Omega.
  \]
  Note that for $\Sp\in [1/2,1]$, $2\Sp^2-4\Sp\leq -3/2$. This means that $q-4\Sp$ is, in the limit as $\tau\rightarrow-\infty$, bounded from above by
  $-3/2$. This means that $|\No|$ tends to infinity exponentially. On the other hand, $|\No|\leq 2/\sqrt{3}$ due to (\ref{eq:constraint}). This contradiction
  leads to the conclusion that there is a $\tau_1$ such that $\Sp(\tau)\leq s_+$ for all $\tau\leq \tau_1$ and some $s_+<1/2$. Combining this observation
  with (\ref{eq:Sp prime BVIII}) leads to the conclusion that $\No^2$ and $\No\Nt$ are integrable to the past; note that $\Sp$ is bounded, the last term on
  the right hand side of (\ref{eq:Sp prime BVIII}) is integrable, $2-\Sp(\tau)\geq 2-s_+$ for all $\tau\leq\tau_1$ and $1/2-\Sp(\tau)\geq 1/2-s_+$ for
  all $\tau\leq\tau_1$. Since $\No$, $\No'$ and $(\No\Nt)'$ are bounded to the past, we conclude that $\No$ and $\No\Nt$ converge to zero. Combining this observation
  with (\ref{eq:constraint}) yields the conclusion that $\Sp$ converges to $-1$, so that $q$ converges to $2$. This means that $\No$ and $\No\Nt$ converge to zero
  exponentially, so that $q-2$ converges to zero exponentially. This means that $|\phi'(\tau)|\leq Ce^{6\tau}$, so that (\ref{eq:phi limit BVIII LRS}) follows.
  Note also that $\Sp$ converges exponentially to $-1$. Thus $|\No(\tau)|\leq Ce^{6\tau}$ and $|(\No\Nt)(\tau)|\leq Ce^{6\tau}$ for all $\tau\leq 0$. Moreover,
  $\Omega(\tau)\leq Ce^{6\tau}$ for all $\tau\leq 0$. This means that $|q(\tau)-2|\leq Ce^{6\tau}$ for all $\tau\leq 0$. Thus (\ref{eq:q etc limit BVIII LRS})
  and thereby (\ref{eq:ln theta as BVIII LRS}) hold. Combining the above also yields (\ref{eq:Sp limit LRS BVIII}). Combining the above with the equations
  for the $N_i$ yields the remaining conclusions. 
\end{proof}

In the vacuum dominated isotropic or LRS Bianchi type IX setting, the following holds. 
\begin{cor}\label{cor:LRS Bianchi type IX}
  Consider a development satisfying the conditions of Lemma~\ref{lemma:BianchiAdevelopment} where $V\in\mfP_{\a_V}^1$ for some $\a_V\in (0,1/3)$. Assume it to be
  of either isotropic or LRS Bianchi type IX with $\Nt=\Nth$ and $\Sm=0$, and such that (\ref{eq:Omegaphiprstrongdec}) holds. Then there are constants
  $\theta_{\infty}>0$, $m_1$, $m_2=m_3$ and $\Phi_0$ such that
  \begin{subequations}
    \begin{align}
      |\Sp(\tau)+1| \leq & Ce^{6\tau},\label{eq:Sp limit LRS BIX}\\
      |\ln N_{i}(\tau)-f_i(2,-1,0)\tau-m_i| \leq & Ce^{6\tau},\\
      |q(\tau)-2|+|\Omega(\tau)|+|\Psi(\tau)| \leq & Ce^{6\tau},\label{eq:q etc limit BIX LRS}\\
      |\ln\theta(\tau)+3\tau-\ln\theta_{\infty}| \leq & Ce^{6\tau},\label{eq:ln theta as BIX LRS}\\
      |\phi_{1}(\tau)|+|\phi(\tau)-\Phi_{0}| \leq & Ce^{6\tau}\label{eq:phi limit BIX LRS}
    \end{align}
  \end{subequations}
  for all $\tau\leq 0$ and $i=1,2,3$. 
\end{cor}
\begin{proof}
  Due to (\ref{eq:qdef}), (\ref{eq:spmp}) and (\ref{eq:constraint}),
  \begin{equation}\label{eq:Sp prime}
    \Sp'=(1-\Sp^2)(4-2\Sp)-\Psi(4-2\Sp)-2\Omega\Sp+9\No\Nt.
  \end{equation}
  Note that there is a constant $C$ such that $\Psi+\Omega\leq Ce^{2\tau}$ for all $\tau\leq\tau_0$; this follows from (\ref{eq:qdef}),
  (\ref{eq:Omegaphiprstrongdec}) and the Lemma~\ref{lemma:X growth}. Given $0<\e<1$, there is a $c_\e>0$ such that if
  $\Sp\in [-1+\e,1-\e]$, then the first term on the right hand side of (\ref{eq:Sp prime}) is bounded from below by $c_\e$. Let $T\leq\tau_0$ be such
  that the sum of the absolute values of the second and third terms on the right hand side of (\ref{eq:Sp prime}) is bounded from above by $c_\e/2$
  for $\tau\leq T$, assuming $\Sp(\tau)\in [-1+\e,1-\e]$. This means that if $\Sp(\tau)\in [-1+\e,1-\e]$ and $\tau\leq T$, then $\Sp'(\tau)\geq c_\e/2$.
  This means that $\Sp$ decreases to the past until it leaves $[-1+\e,1-\e]$. Moreover, it cannot return to this interval at a later time. To conclude,
  given $\e>0$, there are two possibilities. Either there is a time $T\leq \tau_0$ such that $\Sp(\tau)<-1+\e$ for all $\tau\leq T$ or a time
  $T_1\leq\tau_0$ such that $\Sp(\tau)> 1-\e$ for all $\tau\leq T_1$. In other words, either $\Sp$ converges to $-1$ or it converges to $1$. This
  means that $q$ converges to $2$. Assume that $\Sp$ converge to $1$. Then $\No$ tends to infinity exponentially and $\Nt$ tends to zero exponentially
  as $\tau\rightarrow-\infty$. This means that the polynomial on the left hand side of (\ref{eq:constraint}) tends to infinity, contradicting
  (\ref{eq:constraint}). The only possibility that remains is that $\Sp(\tau)\rightarrow-1$ as $\tau\rightarrow-\infty$. This means that $\No$ and
  $\No\Nt$ converge to zero exponentially. Since $\Sm=0$ and $\Psi$ converges to zero exponentially, this means that $\Sp+1$ decays exponentially,
  see (\ref{eq:constraint}),
  so that $q-2$ decays exponentially. Combining this observation with (\ref{eq:Omegaphiprstrongdec}), it is clear that $\Omega=O(e^{6\tau})$ and
  $\phi'=O(e^{6\tau})$. Moreover, $\No^2=O(e^{12\tau})$ and $\No\Nt=O(e^{6\tau})$. Returning to (\ref{eq:constraint}) and (\ref{eq:altqdef}) with this
  information at hand yields (\ref{eq:Sp limit LRS BIX}) and (\ref{eq:q etc limit BIX LRS}). Given this information, the remaining conclusions of
  the lemma can be deduced.
\end{proof}

\section{The generic Bianchi type VIII and IX settings}

Our next goal is to prove that if (\ref{eq:Omegaphiprstrongdec}) holds in the case of a Bianchi type VIII or IX solution which is neither isotropic
nor LRS, then the solution is oscillatory. To prove this, it is natural to return to the arguments presented in \cite{cbu}. Due to
the presence of the non-linear scalar field, complications arise. Nevertheless, suitable modifications of the arguments presented in \cite{cbu} yield
the desired conclusion.

\begin{lemma}\label{lemma:BVIIIcontrNiconcl}
  Consider a development satisfying the conditions of Lemma~\ref{lemma:BianchiAdevelopment}. Assume it to either be of Bianchi type VIII,
  with $\No<0$, $\Nt>0$ and $\Nth>0$; or of Bianchi type IX, with $N_i>0$ for all $i$. Assume, finally, that $V\in\mfP_{\a_V}^{1}$, where
  $\a_V\in (0,1/3)$ and that (\ref{eq:Omegaphiprstrongdec}) holds.  Then the set of $(\Sp,\Sm)(\tau)$ for $\tau\leq 0$ is bounded.
  Next, assume that $\|N(\tau)\|\rightarrow\infty$ as $\tau\rightarrow-\infty$.
  Then, up to a permutation of the variables in the case of Bianchi type IX, $\Nt(\tau)$, $\Nth(\tau)\rightarrow \infty$ and $\No(\tau)$,
  $\No(\tau)\Nt(\tau)$, $\No(\tau)\Nth(\tau)\rightarrow 0$ as $\tau\rightarrow-\infty$.
\end{lemma}
\begin{remark}
  We here use the notation $\|N\|:=(\No^{2}+\Nt^{2}+\Nth^{2})^{1/2}$.
\end{remark}
\begin{proof}
  Note that $(\No\Nt\Nth)'=3q(\No\Nt\Nth)$. Since $\Omega$ is integrable to the past, see Lemma~\ref{lemma:X growth}, this means that $|\No\Nt\Nth|$
  is bounded to the past; say $|\No\Nt\Nth|(\tau)\leq C$ for all $\tau\leq 0$ and some $C\geq 1$. Given this observation, the arguments presented in
  \cite[Lemmas~6, 7 and 15]{cbu} yield the desired conclusion. 
\end{proof}
Next, we wish to generalise \cite[Lemma~16, p.~724]{cbu} to the present setting.

\begin{lemma}\label{lemma:bounding Sp BIX contr arg}
  Consider a Bianchi type IX development satisfying the conditions of Lemma~\ref{lemma:BianchiAdevelopment} with $N_i>0$. Assume that $V\in\mfP_{\a_V}^{1}$,
  where $\a_V\in (0,1/3)$, and that (\ref{eq:Omegaphiprstrongdec}) holds. If $\tau_1\leq 0$, 
  \[
  -1+\tfrac{C_\Omega}{1-3\a_V}e^{2(1-\a_V)\tau_1}<\Sp(\tau_1)<\tfrac{1}{3},
  \]
  where $C_\Omega$ is the constant appearing in (\ref{eq:Omegadec prel}), and $9\No<\Nt+\Nth$ in $[\tau_2,\tau_1]$, then
  \begin{equation}\label{eq:Sptau lb}
    \Sp(\tau_1)-\tfrac{C_\Omega}{1-3\a_V}e^{2(1-\a_V)\tau_1}\leq \Sp(\tau)
  \end{equation}
  for all $\tau\in [\tau_2,\tau_1]$. 
\end{lemma}
\begin{proof}
  Combining (\ref{eq:spmp}) and (\ref{eq:altqdef}) yields $q-2\leq 3\No(\Nt+\Nth)$ and 
  \begin{equation}\label{eq:Sp prime BIX contr arg one}
    \Sp'=(q-2)(\Sp+1)+\tfrac{9}{2}\No(\No-\Nt-\Nth)+2\Omega. 
  \end{equation}
  This yields the conclusion that for $-1<\Sp<1/3$,
  \begin{equation}\label{eq:Sp prime BIX contr arg}
    \Sp'<\tfrac{1}{2}(9\No-\Nt-\Nth)\No+2\Omega.
  \end{equation}
  Assume now that that there is a $\tau\in [\tau_2,\tau_1]$ such that $\Sp(\tau)$ is strictly less than the left hand side of
  (\ref{eq:Sptau lb}). Then there is a subinterval $[\tau_b,\tau_a]$ of $[\tau_2,\tau_1]$ such that $-1<\Sp(\tau)\leq \Sp(\tau_1)$ for
  $\tau\in [\tau_b,\tau_a]$, $\Sp(\tau_a)=\Sp(\tau_1)$ and $\Sp(\tau_b)$ is strictly less than the left hand side of (\ref{eq:Sptau lb}). Integrating
  (\ref{eq:Sp prime BIX contr arg}), keeping (\ref{eq:Omegadec prel}) and the above observations in mind yields
  \[
  \tfrac{C_\Omega}{1-3\a_V}e^{2(1-\a_V)\tau_1}<\Sp(\tau_a)-\Sp(\tau_b)<\textstyle{\int}_{\tau_b}^{\tau_a}2\Omega(s)ds\leq \tfrac{C_\Omega}{1-3\a_V}e^{2(1-\a_V)\tau_a},
  \]
  which is a contradiction. 
\end{proof}
Next, we generalise \cite[Lemma~17, p.~724]{cbu} to the present setting.
\begin{lemma}\label{eq:Spdoesnotconvtomone}
  Consider a Bianchi type VIII development satisfying the conditions of Lemma~\ref{lemma:BianchiAdevelopment} with $\No<0$, $\Nt>0$ and $\Nth>0$. Assume
  that $V\in\mfP_{\a_V}^{1}$, where $\a_V\in (0,1/3)$, and that (\ref{eq:Omegaphiprstrongdec}) holds. Then there is an $\e_{0}>0$ with the following
  properties. Fix an $\e$
  such that $0<\e<\e_{0}$. Then there are constants $A>0$ and $T\leq 0$ (depending only on $\e$ and the solution) such that if
  $[\tau_{1},\tau_{0}]$ is an interval with $\tau_{0}\leq T$ and $\Nt(\tau)$, $\Nth(\tau)\geq A$ for all $\tau\in [\tau_{1},\tau_{0}]$,
  and if $\Sp(\tau_{0})\geq -1+2\e$, then $\Sp(\tau)\geq -1+\e$ for all $\tau\in [\tau_{1},\tau_{0}]$. 
\end{lemma}
\begin{proof}
  Note, to begin with, that combining (\ref{eq:spmp}) and (\ref{eq:altqdef}) yields
  \begin{equation}\label{eq:altsppBVIII}
    \Sp'=-(2-q)(\Sp+1)+\tfrac{9}{2}\No^{2}-\tfrac{9}{2}\No(\Nt+\Nth)+2\Omega.
  \end{equation}
  Assume now, in order to arrive at a contradiction, that there is a $\tau_{a}\in [\tau_{1},\tau_{0}]$ such that $\Sp(\tau_{a})=-1+\e$.
  This $\tau_{a}$ can, without loss of generality, be assumed to be the first $\tau$ before $\tau_{0}$ at which $\Sp$ attains the value
  $-1+\e$. Then there is a $\tau_{b}\in [\tau_{1},\tau_{0}]$ with $\tau_{b}\geq \tau_{a}$ such that $-1+\e\leq\Sp(\tau)\leq -1+2\e$ for all
  $\tau\in [\tau_{a},\tau_{b}]$. Moreover, $\Sp(\tau_{b})=-1+2\e$. Integrating (\ref{eq:altsppBVIII}) with this information in mind yields
  \begin{equation}\label{eq:contrBVIIISpineq}
    \e\leq \tfrac{9}{2}\textstyle{\int}_{\tau_{a}}^{\tau_{b}}\No^{2}(\tau)d\tau-\tfrac{9}{2}\textstyle{\int}_{\tau_{a}}^{\tau_{b}}[\No(\Nt+\Nth)](\tau)d\tau
    +2\textstyle{\int}_{\tau_{a}}^{\tau_{b}}\Omega(\tau)d\tau.
  \end{equation}
  Fixing $T$ close enough to $-\infty$, the bound depending only on $\e$ and the solution, it is clear that the last term on the right hand
  side of (\ref{eq:contrBVIIISpineq}) can be assumed to be smaller than $\e/6$; this is an immediate consequence of Lemma~\ref{lemma:X growth}.
  Assuming $\e_{0}$ to be small enough, we can assume
  $q-4\Sp$, $q-\Sp+\sqrt{3}\Sm$ and $q-\Sp-\sqrt{3}\Sm$ to be bounded from below by $2$ on $[\tau_{a},\tau_{b}]$. This means that $\No^{2}$,
  $\No\Nt$ and $\No\Nth$ are arbitrarily small at $\tau_{b}$ (by choosing $A$ large enough and using the fact that $|\No\Nt\Nth|(\tau)\leq C$
  for $\tau\leq 0$) and exponentially decaying backwards in time (at a fixed rate). Choosing $\e_{0}$ small enough and $A$ large enough (depending
  only on $\e$ and the solution), we conclude that the sum of the first two terms on the right hand side of (\ref{eq:contrBVIIISpineq}) are bounded
  from above by $\e/3$. This yields the inequality $\e\leq \e/2$, which is a contradiction. The lemma follows. 
\end{proof}
Our next goal is to prove the existence of an $\a$-limit point, as in \cite{cbu}. There is one problem in doing so: in order to obtain a closed
autonomous system, we here, as opposed to the vacuum setting studied in \cite{cbu}, need to include $\theta$, but $\theta$ diverges to $\infty$ in
the direction of the singularity. There is a simple remedy to this problem; we simply introduce the variable $u=1/\theta$ and replace all occurrences
of $1/\theta$ with $u$. Moreover, we replace (\ref{eq:thetaprime}) by $u'=(1+q)u$. We then obtain an autonomous system of equations for $u$, $\phi_0$,
$\phi_1$, $\Sp$, $\Sm$, $\No$, $\Nt$ and $\Nth$ satisfying the constraint (\ref{eq:constraint}). In what follows, we refer to this system as the
\textit{modified Wainwright-Hsu equations}.
\index{Modified Wainwright-Hsu equations}%
Note also that setting $u=0$ and $\phi_1=0$ leads to the vacuum version of the original Wainwright-Hsu equations; i.e., exactly the equations studied
in \cite{cbu}. Next, note that if (\ref{eq:Omegaphiprstrongdec}) and the conditions of Lemma~\ref{lemma:X growth} hold, then $\Omega$ and $\phi'$
decay as $e^{2\tau}$, so that there is a $\phi_\infty$ with the property that
\begin{equation}\label{eq:u phi one phi zero limit}
  \lim_{\tau\rightarrow-\infty}(u,\phi_0,\phi_1)(\tau)=(0,\phi_\infty,0).
\end{equation}
That there is an $\a$-limit point in the case of Bianchi type IX follows by an argument similar to the proof of \cite[Proposition~2, p.~724--725]{cbu}.
\begin{prop}\label{prop:alpexBIX}
  Consider a Bianchi type IX development satisfying the conditions of Lemma~\ref{lemma:BianchiAdevelopment} with $N_i>0$. Assume
  that $V\in\mfP_{\a_V}^{1}$, where $\a_V\in (0,1/3)$, and that (\ref{eq:Omegaphiprstrongdec}) holds. Then the corresponding solution to the modified
  Wainwright-Hsu equations has an $\a$-limit point. 
\end{prop}
\begin{proof}
  Under the assumptions of the proposition, there is a $\phi_\infty$ such that (\ref{eq:u phi one phi zero limit}) holds. In the present setting, it is
  therefore sufficient to prove that there is a sequence $\tau_{k}\rightarrow-\infty$ such that $(\Sp,\Sm,\No,\Nt,\Nth)(\tau_{k})$ converges.  
  Assume that there is no limit point of the desired type. Then, since $(\Sp,\Sm)$ remains bounded to the past, $\|N\|$ has to converge to $\infty$.
  Due to Lemma~\ref{lemma:BVIIIcontrNiconcl}, we can assume that $\No$, $\No\Nt$ and $\No\Nth$ converge to zero and that $\Nt$, $\Nth$ converge
  to infinity. Due to (\ref{eq:nip}), $(\Nt\Nth)'=2(q+2\Sp)\Nt\Nth$. Since $\Nt\Nth\rightarrow\infty$, this means that
  \begin{equation}\label{eq:q plus two Sp int}
    \textstyle{\int}_0^\tau(q+2\Sp)ds\rightarrow\infty
  \end{equation}
  as $\tau\rightarrow-\infty$. Combining this observation with (\ref{eq:qdef}) and the fact that $\Omega$ is integrable to the past, it is clear
  that
  \[
  \textstyle{\int}_\tau^0\Sp ds\rightarrow-\infty
  \]
  as $\tau\rightarrow-\infty$. Due to Lemma~\ref{lemma:bounding Sp BIX contr arg}, it follows that there is a $T\leq 0$ such that
  \[
  \Sp(\tau)\leq \tfrac{C_\Omega}{1-3\a_V}e^{2(1-\a_V)\tau}
  \]
  for all $\tau\leq T$. Next, assume that there is a $\tau_2\leq T$ such that
  \[
  -1+\tfrac{C_\Omega}{1-3\a_V}e^{2(1-\a_V)\tau_2}<\Sp(\tau_2).
  \]
  Then, due to Lemma~\ref{lemma:bounding Sp BIX contr arg},
  \[
  -1<\Sp(\tau_2)-\tfrac{C_\Omega}{1-3\a_V}e^{2(1-\a_V)\tau_2}\leq\Sp(\tau)
  \]
  for all $\tau\leq\tau_2$. For future reference, we denote the expression in the middle by $\a$. Note, moreover, that since $\Sp>-1$, $q-2\leq q+2\Sp$.
  Introduce $\ma:=\{\tau:q+2\Sp<0\}$ and $\ma_\tau:=\ma\cap [\tau,\tau_2]$. Due to (\ref{eq:Sp prime BIX contr arg}) and the fact that the first
  term on the right hand side of (\ref{eq:Sp prime BIX contr arg}) is negative, 
  \begin{equation*}
    \begin{split}
      \Sp(\tau_2)-\Sp(\tau) = & \textstyle{\int}_{\tau}^{\tau_2}\Sp'ds\leq \textstyle{\int}_{\ma_\tau}\Sp'ds+\textstyle{\int}_{\tau}^{\tau_2}2\Omega ds\\
      \leq & \textstyle{\int}_{\ma_\tau}(q-2)(\Sp+1)ds+\textstyle{\int}_{\tau}^{\tau_2}4\Omega ds,
    \end{split}
  \end{equation*}
  where we appealed to (\ref{eq:Sp prime BIX contr arg one}) in the last step. Since $q-2<q+2\Sp<0$ and $1+\Sp\geq 1+\a>0$ in $\ma_\tau$, it follows that
  \[
  \Sp(\tau_2)-\Sp(\tau)\leq (1+\a)\textstyle{\int}_{\ma_\tau}(q+2\Sp)ds+\textstyle{\int}_{\tau}^{\tau_2}4\Omega ds.
  \]
  However, due to (\ref{eq:q plus two Sp int}), the first term on the right hand side tends to $-\infty$ as $\tau\rightarrow-\infty$ and the second term
  on the right hand side is bounded due to (\ref{eq:Omegadec prel}). This implies that $\Sp\rightarrow\infty$, a contradiction. Next, assume that
  \[
  \Sp(\tau)\leq -1+\tfrac{C_\Omega}{1-3\a_V}e^{2(1-\a_V)\tau}
  \]
  for all $\tau\leq T$. Since $\No\Nt$ and $\No\Nth$ converge to zero, (\ref{eq:constraint}) then implies that $(\Sp,\Sm)\rightarrow (-1,0)$.
  Combining this observation with Proposition~\ref{prop:limitcharsp} yields the conclusion that $\Nt=\Nth$ and $\Sm=0$; i.e., that the solution is
  LRS. Then, due to Corollary~\ref{cor:LRS Bianchi type IX}, all the $N_i$ converge, contradicting the assumptions.
  To conclude, the assumption that $\|N\|\rightarrow\infty$ leads to a contradiction. Since $(\Sp,\Sm)$ remains bounded to the past, we conclude
  that the desired limit point exists. 
\end{proof}
Next, we prove the existence of an $\a$-limit point in the Bianchi VIII setting by carrying out an argument similar to the proof of \cite[Proposition~3, p.~725]{cbu}.
\begin{prop}\label{prop:alpexBVIII}
  Consider a Bianchi type VIII development with $\No<0$, $\Nt>0$ and $\Nth>0$, satisfying the conditions of Lemma~\ref{lemma:BianchiAdevelopment}.
  Assume that $V\in\mfP_{\a_V}^{1}$, where $\a_V\in (0,1/3)$, and that (\ref{eq:Omegaphiprstrongdec}) holds. Then the corresponding solution to the modified
  Wainwright-Hsu equations has an $\a$-limit point. 
\end{prop}
\begin{proof}
  As in the proof of Proposition~\ref{prop:alpexBIX}, it is sufficient to demonstrate that there is a sequence $\tau_{k}\rightarrow-\infty$ such that
  $(\Sp,\Sm,\No,\Nt,\Nth)(\tau_{k})$ converges. Due to the constraint (\ref{eq:constraint}), $\Sp$, $\Sm$ and $\No$ are bounded for all $\tau$.
  If there is no $\a$-limit point, we can therefore, appealing to Lemma~\ref{lemma:BVIIIcontrNiconcl}, assume that $\No$, $\No\Nt$ and $\No\Nth$
  converge to zero and that $\Nt$, $\Nth$ converge to infinity. In particular, $(\Nt\Nth)(\tau)\rightarrow\infty$ as $\tau\rightarrow-\infty$. Combining
  this observation with (\ref{eq:qdef}), (\ref{eq:nip}) and Lemma~\ref{lemma:X growth} yields the conclusion that
  \begin{equation}\label{eq:intdivBVIII}
    \textstyle{\int}_{\tau}^{0}(\Sp^{2}+\Sm^{2}+\Sp)(s)ds\rightarrow-\infty
  \end{equation}
  as $\tau\rightarrow -\infty$.
  Let us first consider the case that $\Sp(\tau)\rightarrow-1$ as $\tau\rightarrow -\infty$. Then Proposition~\ref{prop:limitcharsp} implies
  that $\Nt=\Nth$ and $\Sm=0$, so that Corollary~\ref{cor:LRS Bianchi type VIII} yields a contradiction to the assumption.
  If $\Sp$ does not converge to $-1$, we can appeal to Lemma~\ref{eq:Spdoesnotconvtomone} to conclude that there is an $\e>0$ and a $T\leq 0$
  such that $\Sp(\tau)\geq -1+\e/2$ for all $\tau\leq T$. Combining this observation with (\ref{eq:altsppBVIII}) and (\ref{eq:qdef}) yields,
  assuming $\tau\leq T$,
  \begin{equation*}
    \begin{split}
      \Sp' = & 2(\Sp^{2}+\Sm^{2}-1)(\Sp+1)+\tfrac{9}{2}\No^{2}-\tfrac{9}{2}\No(\Nt+\Nth)+2\Psi(\Sp+1)-2\Omega\Sp\\
      \leq & \e(\Sp^{2}+\Sm^{2}+\Sp-\Sp-1)+\tfrac{9}{2}\No^{2}-\tfrac{9}{2}\No(\Nt+\Nth)+2\Psi(\Sp+1)-2\Omega\Sp\\
      \leq & \e(\Sp^{2}+\Sm^{2}+\Sp)-\tfrac{\e^{2}}{2}+\tfrac{9}{2}\No^{2}-\tfrac{9}{2}\No(\Nt+\Nth)+2\Psi(\Sp+1)-2\Omega\Sp.
    \end{split}
  \end{equation*}
  Since the last four terms converge to zero as $\tau\rightarrow-\infty$ (where we used the conclusions of Lemma~\ref{lemma:BVIIIcontrNiconcl}),
  it follows that there is a $T_{1}$ such that
  \[
  \Sp'\leq \e(\Sp^{2}+\Sm^{2}+\Sp)
  \]
  for all $\tau\leq T_{1}$. Integrating this inequality from $\tau\leq T_{1}$ to $T_{1}$ yields
  \[
  \Sp(T_{1})-\Sp(\tau)\leq \textstyle{\int}_{\tau}^{T_{1}}\e(\Sp^{2}+\Sm^{2}+\Sp)(s)ds\rightarrow-\infty,
  \]
  where we appealed (\ref{eq:intdivBVIII}) in the last step. Since $\Sp$ is bounded, this conclusion leads to a contradiction. The proposition
  follows. 
\end{proof}

\begin{lemma}\label{lemma:NoNtNthconvtozero}
  Consider a Bianchi type VIII or IX development satisfying the conditions of Lemma~\ref{lemma:BianchiAdevelopment}. Assume that $V\in\mfP_{\a_V}^{1}$, where
  $\a_V\in (0,1/3)$, and that (\ref{eq:Omegaphiprstrongdec}) holds. Then $(\No\Nt\Nth)(\tau)$ tends to zero as $\tau\rightarrow -\infty$.
\end{lemma}
\begin{proof}
  Note that $(\No\Nt\Nth)'=3q\No\Nt\Nth$. Since $\Omega$ decays exponentially, there are two possibilities. Either $q$ is integrable to the past, and
  $\No\Nt\Nth$ converges to a non-zero number, or $q$ is non-integrable, and $\No\Nt\Nth$ converges to zero as $\tau\rightarrow-\infty$. What remains
  is thus to exclude the first possibility. Due to Propositions~\ref{prop:alpexBIX} and \ref{prop:alpexBVIII}, there is an $\a$-limit point, say
  $x_{*}$, of the solution to the modified Wainwright-Hsu equations. Assuming $\No\Nt\Nth$ to converge to a non-zero number, $x_{*}$ is of Bianchi type
  $\mfT\in\{\mrVIII,\mrIX\}$. Moreover, due to (\ref{eq:u phi one phi zero limit}), $x_*$ are effectively initial data for the vacuum version
  of the Wainwright-Hsu equations studied in \cite{cbu}. Moreover, applying the flow of the modified Wainwright-Hsu equations to $x_{*}$ yields a
  Bianchi type $\mfT$ vacuum solution, the closure of whose range is contained in the the $\a$-limit set of the original solution. However, since
  $\No\Nt\Nth$ is strictly monotonic along the vacuum flow (see the proof of \cite[Lemma~13, p.~722]{cbu} for a justification), we get a contradiction to the fact that the product
  $\No\Nt\Nth$ converges for the original solution. This means that $q$ is not integrable to the past, and the desired conclusion follows. 
\end{proof}
Next, we generalise \cite[Theorem~3, p.~726]{cbu} to the present setting. 
\begin{prop}
  Consider a Bianchi type $\mfT\in\{\mrVIII,\mrIX\}$ development satisfying the conditions of Lemma~\ref{lemma:BianchiAdevelopment}. Assume it to be non-LRS
  and with two $N_{i}$ positive and one negative in case $\mfT=\mrVIII$ and with all $N_i$ positive in case $\mfT=\mrIX$. Assume that $V\in\mfP_{\a_V}^{1}$,
  where $\a_V\in (0,1/3)$, and that (\ref{eq:Omegaphiprstrongdec}) holds. Then the corresponding solution to the modified Wainwright-Hsu equations has an $\a$-limit
  point, say $x_{*}=(0,\phi_\infty,0,\sigma_{+},\sigma_{-},n_{1},n_{2},n_{3})$, such that $(\sigma_{+},\sigma_{-})$ is not a special point on the Kasner circle and such that
  $n_{1}n_{2}n_{3}=0$.
\end{prop}  
\begin{proof}
  Let $Y_{*}=(0,\phi_\infty,0,s_{+},s_{-},\bn_{1},\bn_{2},\bn_{3})$ be an $\a$-limit point of the solution to the modified Wainwright-Hsu equations; see
  (\ref{eq:u phi one phi zero limit}) and Propositions~\ref{prop:alpexBIX} and \ref{prop:alpexBVIII}. Due to (\ref{eq:u phi one phi zero limit}), we, in what follows,
  focus on $y_{*}=(s_{+},s_{-},\bn_{1},\bn_{2},\bn_{3})$. Note that
  $\bn_{1}\bn_{2}\bn_{3}=0$ due to Lemma~\ref{lemma:NoNtNthconvtozero}. If $(s_{+},s_{-})$ is not a special point, we are thus done. Assume therefore,
  without loss of generality, that $(s_{+},s_{-})=(-1,0)$ (note that, in the present proposition, we do not assume $\No<0$ when $\mfT=\mrVIII$). Since the
  solution is not LRS, we cannot have $(\Sp,\Sm)\rightarrow (-1,0)$; see Proposition~\ref{prop:limitcharsp}. Let $\{\tau_{k}\}$ be a time sequence with
  $\tau_{k}\rightarrow-\infty$ such that $(\Sp,\Sm,\No,\Nt,\Nth)(\tau_{k})\rightarrow y_{*}$. Then there is an $0<\e<1/1000$ and a
  sequence $\{s_{k}\}$ with $s_{k}\leq \tau_{k}$ such that $(\Sp(s_{k}),\Sm(s_{k}))$ converges to the boundary of the ball of radius $\e$ and center $(-1,0)$.
  Moreover, $s_{k}$ can be assumed to
  be the first time $(\Sp,\Sm)$ reaches the boundary after $\tau_{k}$. If there is a subsequence of $N(s_{k})$ which is bounded (where
  $N=(\No,\Nt,\Nth)$), we can extract a convergent subsequence of the desired type. We can therefore assume that $\|N(s_{k})\|$ tends to infinity.
  As in the proof of \cite[Lemma~15]{cbu}, we can assume that two $N_{i}(s_{k})$ tend to infinity. On the other hand,
  it is, due to (\ref{eq:nip}) and (\ref{eq:fodef}), clear that $|\No|$ decays from $\tau_{k}$ to $s_{k}$ (note that $q-4\Sp\approx 6$ in the ball
  of radius $\e$ and center $(-1,0)$). This means that $\No<0$ if $\mfT=\mrVIII$. We can thus assume $\No<0$ and $\Nt,\Nth>0$ in the case of
  Bianchi type VIII. Moreover, $\Nt(s_k),\Nth(s_k)\rightarrow\infty$, $(\No\Nt)(s_k)\rightarrow 0$, $(\No\Nth)(s_k)\rightarrow 0$, irrespective of
  Bianchi type. 

  \textit{Bianchi type VIII.} Next, let $\e_{0}>0$ be the constant appearing in the statement of Lemma~\ref{eq:Spdoesnotconvtomone} and let $0<\eta<\e_{0}$
  be such that $\Sp(s_{k})\geq -1+2\eta$. Then there are corresponding constants $A$ and $T$ as in the statement of Lemma~\ref{eq:Spdoesnotconvtomone}.
  We can, without loss of generality, assume $A\geq 1$. We can also, by removing elements from the sequence if necessary, assume that $s_{k}\leq T$
  for all $k$ and that $\Nth(s_{k})\geq 100A$ for all $k$. By choosing $A$ appropriately, we can also assume that there is a sequence
  $v_{k}\leq s_{k}$ such that $\Nth(v_{k})=10A$ and $\Nth(\tau)\geq 10A$ in $[v_{k},s_{k}]$; the reason for this is that we know that there is
  an $\a$-limit point. Moreover, for $k$ large enough, we can assume that there are sequences $u_{k}$ and $t_{k}$ with
  $v_{k}\leq u_{k}\leq t_{k}\leq s_{k}$ such that $\Nth(t_{k})=10^{20}A$, $\Nth(u_{k})=10^{10}A$ and $10^{10}A\leq \Nth(\tau)\leq 10^{20}A$ for all
  $\tau\in [u_{k},t_{k}]$. Since $|\Nt-\Nth|\leq 2$ due to the constraint, this means that $|\Nt/\Nth-1|\leq 2\cdot 10^{-10}A^{-1}$ on
  $[u_{k},t_{k}]$. Next, since $q+2\Sp-2\sqrt{3}\Sm\leq 6$ (due to the constraint and the fact that $q\leq 2$),
  \[
  10\leq \ln\tfrac{\Nth(t_{k})}{\Nth(u_{k})}=\textstyle{\int}_{u_{k}}^{t_{k}}(q+2\Sp-2\sqrt{3}\Sm)(\tau)d\tau\leq 6(t_{k}-u_{k}).
  \]
  Thus $t_{k}-u_{k}\geq 1$. Finally, assume that $|\Sm|\geq 1/10$ in $[u_{k},t_{k}]$. Then
  \begin{equation}\label{eq:intSmuktk}
    \big|\textstyle{\int}_{u_{k}}^{t_{k}}4\sqrt{3}\Sm(\tau)d\tau\big|\geq \tfrac{2\sqrt{3}}{5}.    
  \end{equation}
  On the other hand
  \begin{equation}\label{eq:NtdbNthevol}
    \tfrac{\Nt(t_{k})}{\Nth(t_{k})}=\exp\left(\textstyle{\int}_{u_{k}}^{t_{k}}4\sqrt{3}\Sm(\tau)d\tau\right)\tfrac{\Nt(u_{k})}{\Nth(u_{k})}.
  \end{equation}
  However, (\ref{eq:intSmuktk}) and (\ref{eq:NtdbNthevol}) are not consistent with the fact that $|\Nt/\Nth-1|\leq 2\cdot 10^{-10}A^{-1}$ on
  $[u_{k},t_{k}]$. To conclude, there must thus be a sequence $r_{k}\in [u_{k},t_{k}]$ such that $|\Sm(r_{k})|\leq 1/10$, $\Sp(r_{k})\geq -1+\eta$
  and $|N_{i}(r_{k})|\leq 10^{20}A+2$. Taking an appropriate subsequence of $r_{k}$ yields the desired $\a$-limit point (note that the statement
  that $n_{1}n_{2}n_{3}=0$ follows from Lemma~\ref{lemma:NoNtNthconvtozero}).

  \textit{Bianchi type IX.} In this case there is, due to the constraint and the fact that $[\No(\Nt+\Nth)](s_k)\rightarrow 0$, an $\eta>0$ such that
  \[
  -1+\eta+\tfrac{C_\Omega}{1-3\a_V}e^{2(1-\a_V)s_k}<\Sp(s_k)\leq 0
  \]
  for $k$ large enough. Moreover, $9\No(s_k)<\Nt(s_k)+\Nth(s_k)$ for $k$ large enough. If the latter inequality holds for all $\tau\leq s_k$, then
  Lemma~\ref{lemma:bounding Sp BIX contr arg} implies that $\Sp(\tau)\geq -1+\eta$ for all $\tau\leq s_k$, contradicting the assumption that there
  is a sequence tending to $-\infty$ along which $\Sp$ converges to $-1$. Given this observation the rest of the proof can essentially be taken
  verbatim from the proof of \cite[Theorem~3, p.~726]{cbu}. 
\end{proof}

\begin{cor}\label{cor:BianchiVIIIvacuumas}
  Consider a Bianchi type VIII or IX development satisfying the conditions of Lemma~\ref{lemma:BianchiAdevelopment}. Assume that it is not LRS. Assume that
  $V\in\mfP_{\a_V}^{1}$, where $\a_V\in (0,1/3)$, and that (\ref{eq:Omegaphiprstrongdec}) holds.
  Then the $\a$-limit set of the solution contains the closure of a Bianchi type II vacuum orbit. In particular, it contains at least two points
  on the Kasner circle (with the $N_{i}$ all vanishing), at least one of which is non-special. 
\end{cor}
\begin{proof}
  Due to the proof of \cite[Lemma~14, p.~723]{cbu}, we know that there is a non-special $\a$-limit point on the Kasner circle. By an argument which is
  essentially identical to the proof of \cite[Proposition~6.1, p.~421]{BianchiIXattr}, it follows that the $\a$-limit set of the solution contains the
  closure of a Bianchi type II vacuum orbit.
\end{proof}

\chapter{From asymptotics to data on the singularity}\label{chapter:as to data on sing}

In the last two chapters, we derive the asymptotic behaviour of solutions in terms of the Wainwright Hsu variables, both in the matter and in the vacuum
dominated settings. Here we use this information to conclude that non-oscillatory solutions induce data on the singularity. 

\section{The matter dominated case}

We begin by considering the matter dominated case. 
\begin{prop}\label{prop:matter dom data on sing}
  Consider a development satisfying the conditions of Lemma~\ref{lemma:BianchiAdevelopment} which is not an isotropic Bianchi type I development. Assume
  that $V\in \mfP_{\a_V}^1$, where $\a_V\in (0,1)$ in case of Bianchi type I and non-LRS Bianchi type II; and $\a_V\in (0,1/3)$ otherwise. Assume that
  the development is matter dominated, i.e. that $\phi_t/\theta$ converges to a non-zero limit in the direction of the singularity; see
  Definition~\ref{def:matter and vacuum dominated}. Then the development induces data on the singularity; see Definition~\ref{def:ind data on sing}. 
\end{prop}
\begin{proof}
  In what follows, we make the assumptions and use the notation introduced in the proof of Proposition~\ref{prop:unique max dev}. Note, to begin with, that
  the $e_i'$ are eigenvector fields of the expansion normalised Weingarten map of the leaves of the foliation. The corresponding eigenvalues are
  $\ell_i:=\Sigma_i+1/3$. Next, note that (\ref{eq:for the ai}) holds. Recalling (\ref{eq:dtdtau}), this means that
  \begin{equation}\label{eq:dai dtau}
    \tfrac{da_i}{d\tau}=3\ell_ia_i
  \end{equation}
  and $a_i(0)=1$, where we have changed to $\tau$-time and we fix the translation ambiguity in $\tau$ by demanding that $\tau(t_0)=0$. Thus
  \[
  \theta^{\ell_i}a_i=\theta(0)^{\ell_i(\tau)}\exp\left(-\ell_i(\tau)\textstyle{\int}_0^\tau[1+q(s)]ds+\int_{0}^{\tau}3\ell_i(s) ds\right).
  \]
  Due to (\ref{eq:altqdef}) and (\ref{seq:SpmphiprNilimmatterdom}), it is clear that $q-2$ decays exponentially and that the difference between
  $\ell_i$ and its limit decays exponentially. This means that $\theta^{\ell_i}a_i$ converges to a strictly positive limit. In other words, there
  is a Riemannian metric $\msH$ such that (\ref{eq:bAmetriclimit}) holds. Since the $\ell_i$ converge, it is also clear that the expansion normalised
  Weingarten map converges to a limit, say $\msK$. Next, due to (\ref{seq:SpmphiprNilimmatterdom}), it is clear that $\theta^{-1}\phi_t$ converges
  to a limit, say $\Phi_1$ and that $\phi+\theta^{-1}\phi_t\ln\theta$ converges to a limit, say $\Phi_0$.

  That $\rotr\msK=1$ is an immediate consequence of the fact that $\mK\rightarrow\msK$. Since $\msK$ and $\msH$ are both diagonal with respect to
  $\{e_i'\}$, it is clear that $\msK$ is symmetric with respect to $\msH$. Since the commutator matrix associated with $\{e_i'\}$ is diagonal
  and the $e_i'$ are eigenvectors of $\msK$, it is clear that $\rodiv_{\msH}\msK=0$. That $\tr\msK^2+\Phi_1^2=1$ follows by taking the limit of
  (\ref{eq:constraint}), keeping (\ref{seq:SpmphiprNilimmatterdom}) in mind. Finally, since $\Phi_1\neq 0$, it is clear that $1$ cannot be an
  eigenvalue of $\msK$. On the other hand, due to Theorem~\ref{thm:dichotomy}, if $N_i\neq 0$, then $f_i(2,\sigma_+,\sigma_-)>0$. This means that
  condition 3 of Definition~\ref{def:ndvacidonbbssh} is fulfilled.

  It remains to be verified that the data on the singularity are neither of isotropic nor of LRS Bianchi type $\mrVIIz$; cf.
  Definition~\ref{def:sets of id on singularity}. Assume, to this end, that they are. This means, in particular, that the original regular initial data $\mfI$
  are of Bianchi type $\mrVIIz$. Moreover, there is an orthonormal frame $\{e_{i}\}$ of $\mfg$ with respect to $\msH$ and a family $\Psi_t$, $t\in\ro$, of Lie
  algebra isomorphisms such that (\ref{eq:Psit definition}) and $\msK\Psi_t=\Psi_t\msK$ both hold for all $t$. In the LRS setting, this is true by definition. In
  the isotropic setting, $\msK=\roId/3$, so that $\msK$ commutes with any Lie algebra isomorphism. As noted in the proof of
  Lemma~\ref{lemma:psi sigma plus isotropic}, there is only one way for $\msH$ to have constant curvature; given an orthonormal frame $\{\be_i\}$ of $\mfg$
  with respect to $\msH$ such that the corresponding commutator matrix, say $\bn$, is diagonal with $\bn=\rodiag(0,\bn_2,\bn_3)$, then $\bn_2$ and $\bn_3$
  have to equal. This means that, defining $\Psi_t$ by (\ref{eq:Psit definition}) with $e_i$ replaced by $\be_i$, $\Psi_t$ is a family of Lie algebra
  isomorphisms and $\msK\Psi_t=\Psi_t\msK$ holds for all $t$. 

  Combining the existence of the frame $\{e_i\}$ and the isomorphisms $\Psi_t$ with the fact that $\msK$ is symmetric with respect to $\msH$, it follows that the
  $e_i$ are eigenvector fields of $\msK$. Moreover, if $\msK e_i=p_i e_i$ (no summation), then $p_2=p_3$. Next, let $\nu$ be the commutator matrix associated
  with $\{e_i\}$. Then $\nu=\rodiag(\nu_1,\nu_2,\nu_2)$ due to the fact that the $\Psi_t$ are isomorphisms. Since the Lie group is of Bianchi type $\mrVIIz$, it
  follows that $\nu_1=0$. Concerning the original
  development, we can assume that $\No=0$. In other words, if $n'$ is the commutator matrix associated with $\{e_i'\}$, then $n'=\rodiag(0,n_2',n_3')$. Next,
  let $\a_i$ denote the limit of $\theta^{\ell_i}a_i$ and $\be_i:=\a_i^{-1}e_i'$. Then $\{\be_i\}$ is an orthonormal frame of $\mfg$ with respect to $\msH$. Let
  $\bn$ denote the commutator matrix associated with $\{\be_i\}$. Then $\bn=\rodiag(0,\bn_2,\bn_3)$. Let $A$ be a matrix such that $\be_i=A_i^{\phantom{i}j}e_j$.
  Then, due to the transformation law between commutator matrices associated with orthonormal bases, see \cite[(19.3), p.~207]{RinCauchy},
  \[
    \nu=(\det A)^{-1}A^t\bn A.
  \]
  Note that $A$ is orthogonal, so that $\det A=\pm 1$. On the other hand, exchanging $e_1$ by $-e_1$, if necessary, we can ensure that $\det A=1$. This
  means that $A\nu=\bn A$. From this, it can be deduced that $\nu_2=\bn_2=\bn_3$ and that $\be_1=\pm e_1$. Interchanging $e_1$ and $e_3$ with $-e_1$ and $-e_3$,
  respectively, and interchanging $\Psi_t$ with $\Psi_{-t}$, if necessary, we can assume that $\be_1=e_1$. Moreover, $\be_2$ and $\be_3$ are obtained by applying
  a rotation to $e_2$ and $e_3$. We can therefore, without loss of generality, assume that $\be_i=e_i$. Next, note that if $\{i,j,k\}=\{1,2,3\}$, then  
  \[
  N_k=\tfrac{1}{\theta}\tfrac{a_k}{a_i a_j}n_k'\ \ \ \mathrm{and}\ \ \ \bn_k=\tfrac{\a_k}{\a_i \a_j}n_k'.
  \]
  Since $\bn_2=\bn_3$, it follows that
  \[
    \tfrac{N_2}{N_3}=\tfrac{a_2^2}{a_3^2}\tfrac{n_2'}{n_3'}=\tfrac{a_2^2}{a_3^2}\tfrac{\a_2^{-2}}{\a_3^{-2}}\tfrac{\bn_2}{\bn_3}
    =\tfrac{a_2^2}{a_3^2}\tfrac{\a_2^{-2}}{\a_3^{-2}}. 
  \]
  Since $\theta^{\ell_i}a_i$ converges to $\a_i$, we conclude that
  \begin{equation}\label{eq:Nt Nth quotient limit}
    \lim_{\tau\rightarrow-\infty}\left(\tfrac{N_2}{N_3}\right)(\tau)=\lim_{\tau\rightarrow-\infty}\left(\theta^{2(\ell_3-\ell_2)}
      \tfrac{\theta^{2\ell_2}a_2^2}{\theta^{2\ell_3}a_3^2}\tfrac{\a_2^{-2}}{\a_3^{-2}}\right)(\tau)=\lim_{\tau\rightarrow-\infty}\theta^{2(\ell_3-\ell_2)}(\tau).
  \end{equation}
  Next, note that $\ell_i$ is the eigenvalue of $\mK$ corresponding to the eigenvector $e_i'$. Since $e_i'$ is parallel to $\be_i=e_i$; $\msK e_i=p_i e_i$
  (no summation); and $\mK$ converges to $\msK$, it is clear that $\ell_i$ converges to $p_i$. By the above, we, in addition, know that $\ell_i$
  converges exponentially to its limit and that $p_2=p_3$. This means that $\ell_2-\ell_3$ converges to zero exponentially. On the other hand, $\ln\theta$
  does not grow faster than linearly. Combining these observations with (\ref{eq:Nt Nth quotient limit}) yields the conclusion that $\Nt/\Nth$ converges
  to $1$ in the direction of the singularity. Since $\Sm$ converges to zero (note that $2\Sm=\sqrt{3}(\ell_2-\ell_3)$); $\Sp$ converges to a limit, say
  $\sigma_+$; $\phi'$ converges to a non-zero
  limit, say $\bPhi_1$; $\Omega$ converges to zero; and all the $N_i$ converge to zero, appealing to (\ref{eq:constraint}) yields the conclusion that
  $\sigma_+\in (-1,1)$ and that $\sigma_+^2+\bPhi_1^2/6=1$. Due to the above observations, the conditions of Proposition~\ref{prop:LRS BVIIz as char} are
  satisfied. This means that the original solution is contained in the invariant set given by $\Sm=0$ and $\Nt=\Nth$. This means that $a_2=a_3$, $\ell_2=\ell_3$
  and $\Nt=\Nth$, so that the original initial data are of either isotropic or LRS Bianchi type $\mrVIIz$. However, this contradicts the assumption
  that the original initial data belong to $\mB[V]$. The proposition follows. 
\end{proof}

\begin{prop}\label{prop:ap not plus}
  Assume that the conditions of Lemma~\ref{lemma:BianchiAdevelopment} are not satisfied, but that the conditions of Lemma~\ref{lemma:Bianchi IX remainder}
  are. Assume, moreover, that $V\in \mfP_{\a_V}^1$ for some $\a_V\in (0,1)$. Then the development induces data on the singularity.   
\end{prop}
\begin{proof}
  Due to the conclusions of Theorem~\ref{thm:dichotomy}, the statement follows by an argument which is essentially identical to the proof of
  Proposition~\ref{prop:matter dom data on sing}. 
\end{proof}

\section{The vacuum dominated case}

\begin{prop}\label{prop:I and II not iso as}
  Consider an anisotropic development satisfying the conditions of Lemma~\ref{lemma:BianchiAdevelopment}. Assume that $V\in \mfP_{\a_V}^1$, where
  $\a_V\in (0,1)$ in case of Bianchi type I and non-LRS Bianchi type II; and $\a_V\in (0,1/3)$ otherwise. If the development is vacuum dominated,
  see Definition~\ref{def:matter and vacuum dominated}, and of Bianchi type I or II, then it induces data on the singularity.   
\end{prop}
\begin{proof}
  That the development is vacuum dominated means that (\ref{eq:Omegaphiprstrongdec}) holds. This means that the conclusions of
  Corollary~\ref{cor:vBtIaII} hold, so that the desired conclusion follows by an argument which is essentially identical to the proof of
  Proposition~\ref{prop:matter dom data on sing}, with one exception: in this case, $1$ can be an eigenvalue of $\msK$. On the other hand, if
  $1$ is an eigenvalue of $\msK$, the limit point on the Kasner circle is special. Without loss of generality, we can assume
  that $(\Sp,\Sm)(\tau)\rightarrow (-1,0)$ as $\tau\rightarrow -\infty$. Due to Proposition~\ref{prop:limitcharsp}, we conclude that $\Sm=0$ and
  that $\Nt=\Nth$ for the entire solution, so that $\Nt=\Nth=0$. This means that $\ell_2=\ell_3$, using the notation of the proof of
  Proposition~\ref{prop:matter dom data on sing}, so that $a_2=a_3$ and $\msH(e_2',e_2')=\msH(e_3',e_3')$. Rescaling the frame $\{e_i'\}$
  appropriately yields a frame $\{e_i\}$ which is orthonormal with respect to $\msH$; such that $\msK e_1=e_1$; such that $\msK e_A=0$ for
  $A\in\{2,3\}$; and such that the commutator matrix associated with $\{e_i\}$ is $\rodiag(n,0,0)$. In other words, the frame $\{e_i\}$ is
  of the form required in condition 4 of Definition~\ref{def:ndvacidonbbssh}. 
\end{proof}

\begin{prop}\label{prop:VIz and VIIz not iso as}
  Consider an anisotropic development satisfying the conditions of Lemma~\ref{lemma:BianchiAdevelopment}. Assume that $V\in \mfP_{\a_V}^1$, where
  $\a_V\in (0,1/3)$. If the development is vacuum dominated, see Definition~\ref{def:matter and vacuum dominated}, and of Bianchi type VI${}_0$ or 
  VII${}_0$, then it induces data on the singularity.  
\end{prop}
\begin{proof}
  Due to Corollary~\ref{cor:vBtVIzaVIIz}, the desired conclusion follow by an argument similar to the above. Note, however, that $1$ cannot be
  an eigenvalue of $\msK$ in this case. 
\end{proof}
Finally, the following conclusions hold concerning vacuum dominated LRS Bianchi type VIII and IX developments.
\begin{prop}\label{prop:LRS VIII and IX}
  Consider a development satisfying the conditions of Lemma~\ref{lemma:BianchiAdevelopment}. Assume that $V\in \mfP_{\a_V}^1$, where
  $\a_V\in (0,1/3)$. If the development is vacuum dominated, see Definition~\ref{def:matter and vacuum dominated}, and of LRS Bianchi type VIII
  or of isotropic or LRS Bianchi type IX, then it induces data on the singularity.  
\end{prop}
\begin{proof}
  In this case, the desired conclusion follows from Corollaries~\ref{cor:LRS Bianchi type VIII} and \ref{cor:LRS Bianchi type IX} and arguments
  similar to the above.
\end{proof}

\chapter{Specifying data on the singularity}\label{section:specdatasing}

Our next goal is to specify data on the singularity. In the end, we wish to prove Theorem~\ref{thm:dataonsingtosolution}. However, as a first
step, we here specify data on the singularity for solutions to (\ref{eq:thetaprime}) and (\ref{seq:NoSmprime})--(\ref{eq:constraint}). The corresponding
geometric results
are to be found in Chapter~\ref{chapter:geo res dir of sing}. Note that similar results are derived in the non-degenerate Einstein stiff fluid
setting in \cite[Subsections~2.4--2.5]{RinQCSymm}. However, we here chose different variables, shortening the arguments in \cite{RinQCSymm}. On
the other hand, we here also deal with the case that the eigenvalues of $\msK$ coincide. This causes substantial additional complications in
comparison with \cite{RinQCSymm}. We begin by proving an abstract existence and uniqueness result. 

\section{Abstract setting}\label{section:abstract setting dos}

Fix $1\leq k\in\nn{}$ and $0\geq \tau_0\in\rn{}$ and let $\mX\in C^{\infty}(\rn{k}\times (-\infty,\tau_0],\rn{k})$. Consider 
\begin{equation}\label{eq:x prime equals mX}
  x'(\tau)=\mX[x(\tau),\tau].
\end{equation}
Fix an $\e>0$ and let $X_{\tau_{0}}$ be the space of continuous functions $x:(-\infty,\tau_{0}]\rightarrow\rn{k}$ such that
$\sup_{\tau\leq \tau_{0}}(e^{-\e\tau}|x(\tau)|)\leq 1$. For $x_{a},x_{b}\in X_{\tau_{0}}$ define
\begin{equation}\label{eq:d tau zero def}
  \begin{split}
    d_{\tau_0}(x_{a},x_{b})
    := & \sup_{\tau\leq\tau_{0}}(e^{-\e\tau}|x_a(\tau)-x_b(\tau)|).
  \end{split}
\end{equation}
Note that $(X_{\tau_{0}},d_{\tau_{0}})$ is a complete metric space. Assume that there is a constant $C_\mX\in\rn{}$ such that for every
$\tau_1\leq \tau_0$, 
\begin{equation}\label{eq:mX bd and lip bd}
  |\mX[x_{a}(\tau),\tau]|\leq C_\mX e^{3\e\tau/2},\ \ \
  |\mX[x_{a}(\tau),\tau]-\mX[x_{b}(\tau),\tau]|\leq C_\mX e^{2\e\tau}d_{\tau_{1}}(x_{a},x_{b})
\end{equation}
for all $\tau\leq\tau_1$ and all $x_a,x_b\in X_{\tau_1}$. 
\begin{lemma}\label{lemma:abs exist and unique}
  Given the above assumptions, there is a uniquely determined unbounded interval $I_{\max}\subset (-\infty,\tau_0]$ and a unique solution
  $x\in C^{1}(I_{\max},\rn{k})$ to (\ref{eq:x prime equals mX}) such that $e^{-\e\tau}|x(\tau)|\rightarrow 0$ as $\tau\rightarrow-\infty$ and
  such that $I_{\max}$ is the maximal existence interval of the solution $x$ to (\ref{eq:x prime equals mX}). Moreover, this unique solution
  is smooth.
\end{lemma}
\begin{proof}
  Let $\tau_1\leq\tau_0$, $x\in X_{\tau_1}$ and define $\Phi[x]$ by
  \[
  \Phi[x](\tau):=\textstyle{\int}_{-\infty}^\tau \mX[x(s),s]ds
  \]
  for $\tau\leq\tau_1$. Clearly, $\Phi[x]:(-\infty,\tau_1]\rightarrow\rn{k}$ is well defined and continuous; see (\ref{eq:mX bd and lip bd}).
  Moreover, 
  \begin{equation}\label{eq:Phi x decay}
    |\Phi[x](\tau)|\leq Ce^{3\e\tau/2}
  \end{equation}
  for all $\tau\leq\tau_1$, where $C$ only depends on $C_\mX$ and $\e$. For $\tau_1\leq \tau_0$ close enough to $-\infty$, the bound depending
  only on $\e$ and $C_\mX$, it is thus clear that $\Phi[x]\in X_{\tau_1}$ for all $x\in X_{\tau_1}$. Next, due to the second inequality in
  (\ref{eq:mX bd and lip bd}),
  \[
  d_{\tau_1}(\Phi[x_a],\Phi[x_b])\leq\frac{1}{2}d_{\tau_1}(x_a,x_b)
  \]
  for all $x_a,x_b\in X_{\tau_1}$, assuming $\tau_1\leq \tau_0$ to be close enough to $-\infty$, the bound depending only on $\e$ and $C_\mX$.
  By the Banach fixed point theorem, we conclude that $\Phi$ has a unique fixed point in $X_{\tau_1}$ for $\tau_1\leq \tau_0$ close enough to
  $-\infty$. By the definition of $\Phi$, the properties of $\mX$ and an iterative argument, it is clear that this fixed point, say $y$, is
  smooth and solves (\ref{eq:x prime equals mX}). Let $x$ be the maximal solution to (\ref{eq:x prime equals mX}) with initial data
  $y(\tau_2)$ at $\tau_2\leq\tau_1$ and let $I_{\max}$ be the corresponding existence interval. Then $x(\tau)=y(\tau)$ for $\tau\leq\tau_1$.
  Moreover, since $y$ is a fixed point of $\Phi$, (\ref{eq:Phi x decay}) implies that $e^{-\e\tau}|x(\tau)|\rightarrow 0$ as
  $\tau\rightarrow-\infty$. 

  Assume now $x_i\in C^{1}(I_{\max,i},\rn{k})$, $i=1,2$, to be two solutions satisfying the conditions of the lemma. Then, for
  $\tau_1\in I_{\max,1}\cap I_{\max,2}$ close enough to $-\infty$, it is clear that $x_i\in X_{\tau_1}$, $i=1,2$. Since $x_i$, $i=1,2$, both satisfy
  (\ref{eq:x prime equals mX}) and converge to $0$ as $\tau\rightarrow -\infty$, it follows from the first inequality in (\ref{eq:mX bd and lip bd})
  that $\Phi[x_i]=x_i$, if we consider the $x_i$ to be elements of $X_{\tau_1}$. However, for $\tau_1$ close enough to $-\infty$, we know that
  fixed points of $\Phi$ are unique. This means that $x_1(\tau)=x_2(\tau)$ for $\tau\leq\tau_1$. Since solutions to (\ref{eq:x prime equals mX}) are
  uniquely determined by initial data, it is clear that $x_1=x_2$ and $I_{\max,1}=I_{\max,2}$. The lemma follows. 
\end{proof}

\section{Data on the singularity and adapted variables}\label{ssection:dataonsingexpnormvar}
Our next goal is to specify data on the singularity for solutions to (\ref{eq:thetaprime}) and (\ref{seq:NoSmprime})--(\ref{eq:constraint}), keeping
Remark~\ref{remark:Sigmapm vs Sigmai} in mind. Before doing so, note that if
$\mu_{k}$ are functions satisfying
\begin{equation}\label{eq:mukder smoothness}
  \mu_{k}'=q+6\Sigma_k,
\end{equation}
\index{$\a$Aa@Notation!Expansion normalised!Variables!muk@$\mu_k$}%
then $N_{k}$ can be written $N_{k}=\e_{k}e^{\mu_{k}}$ (no summation), where $\e_{k}$ equals $-1$, $0$ or $1$, assuming the initial data for the
$\mu_{i}$ are chosen appropriately. We specify the following type of data on the singularity.
\begin{definition}
  Let $\sfV$ be the set of
  \index{$\a$Aa@Notation!Symmetry reduced sets of regular initial data!sfV@$\sfV$}%
  \[
  (\e_1,\e_2,\e_3,m_{1},m_{2},m_{3},\sigma_1,\sigma_2,\sigma_3,\varkappa_\infty,\Phi_{1},\Phi_{0})\in\{-1,0,1\}^3\times\rn{9}
  \]
  such that $\sigma_1+\sigma_2+\sigma_3=0$ and  
  \begin{equation}\label{eq:Phi sp sm as constraint}
    p_1^2+p_2^2+p_3^2+\tfrac{1}{9}\Phi_{1}^{2}=1,
  \end{equation}
  where $p_i=\sigma_i+1/3$, and such that one of the following two conditions hold: (i) If $\e_{i}\neq 0$, then $p_{i}>0$.
  (ii) There are $\{i,j,k\}=\{1,2,3\}$ such that $p_i=1$, $m_j=m_k$, $\e_j=\e_k$ and all the $\e_l$ are non-zero.
\end{definition}
Define, for $\lambda\in S_3$, 
\begin{equation*}
  \begin{split}
    & \eta_\lambda^\pm(\e_1,\e_2,\e_3,m_{1},m_{2},m_{3},\sigma_1,\sigma_2,\sigma_3,\varkappa_\infty,\Phi_{1},\Phi_{0})\\
    : = & (\pm\e_{\lambda(1)},\pm\e_{\lambda(2)},\pm\e_{\lambda(3)},m_{\lambda(1)},m_{\lambda(2)},m_{\lambda(3)},
    \sigma_{\lambda(1)},\sigma_{\lambda(2)},\sigma_{\lambda(3)},\varkappa_\infty,\Phi_{1},\Phi_{0}).
  \end{split}
\end{equation*}
\index{$\a$Aa@Notation!Maps!etalambdapm@$\eta_\lambda^\pm$}%
Note that $\eta_\lambda^\pm$ maps $\sfV$ to $\sfV$. 

\begin{prop}\label{prop:variableexandunique}
  Let $V\in \mfP^2_{\a_V}$ for some $\a_V\in (0,1)$ and
  \begin{equation}\label{eq:Pi def}
    \Pi :=(\e_1,\e_2,\e_3,m_{1},m_{2},m_{3},\sigma_1,\sigma_2,\sigma_3,\varkappa_\infty,\Phi_{1},\Phi_{0})\in\sfV.
  \end{equation}
  Then there is a unique solution to (\ref{eq:thetaprime})--(\ref{seq:phizphioev}), where $q$ is defined by
  (\ref{eq:altqdef}), such that
  $N_{k}\neq 0$ if and only if $\e_{k}\neq 0$, such that $N_{k}$, if non-zero, has the same sign as $\e_{k}$, and such that
  \begin{equation}\label{eq:desiredlimit}    
    \lim_{\tau\rightarrow-\infty}[\nu_{1}(\tau),\nu_{2}(\tau),\nu_{3}(\tau),s_{1}(\tau),s_{2}(\tau),s_3(\tau),\psi_{0}(\tau),\psi_{1}(\tau),\varkappa(\tau)]=0,
  \end{equation}
  where
  \begin{subequations}\label{seq:nuivarkappapsiidos}
    \begin{align}                        
      \nu_{i}(\tau) := & \mu_{i}(\tau)-6(\Sigma_i+1/3)\tau-m_{i},\label{eq:Niitonui}\\
      s_i(\tau) := & \Sigma_i(\tau)-\sigma_i,\label{eq:spm def}\\
      \varkappa(\tau) := & \ln\theta(\tau)+3\tau-\varkappa_{\infty},\label{eq:kappaidos}\\
      \psi_{0}(\tau) := & \phi_{0}(\tau)-\phi_{1}\tau-\Phi_{0}.\label{eq:phiidos}\\
      \psi_1(\tau) := & \phi_1(\tau)-\Phi_1,\label{eq:psi one def}      
    \end{align}
  \end{subequations}  
  Here, if $\e_{k}\neq 0$, $\mu_{k}$ is chosen so that $N_{k} =\e_{k}e^{\mu_{k}}$ (no summation). In particular, $\mu_{k}$ then satisfies
  (\ref{eq:mukder smoothness}). If $\e_{k}=0$, $\mu_{k}$ is defined to be the solution to (\ref{eq:mukder smoothness}) with the property that
  $\mu_{k}-6p_{k}\tau\rightarrow 0$ as $\tau\rightarrow-\infty$, where $p_k=\sigma_k+1/3$. Finally, the solution also satisfies the Hamiltonian
  constraint (\ref{eq:constraint}) and the relation (\ref{eq:qdef}).

  If $(\No,\Nt,\Nth,\Sigma_1,\Sigma_2,\Sigma_3,\theta,\phi_0,\phi_1)$ is the solution corresponding to $\Pi\in\sfV$ and $\lambda\in S_3$, then the solution
  corresponding to $\eta_\lambda^\pm(\Pi)$ is given by
  $(\pm N_{\lambda(1)},\pm N_{\lambda(2)},\pm N_{\lambda(3)},\Sigma_{\lambda(1)},\Sigma_{\lambda(2)},\Sigma_{\lambda(3)},\theta,\phi_0,\phi_1)$.
\end{prop}
\begin{remark}
  In order for the statement of the proposition to be meaningful, we need to prove that if $\e_{k}=0$, then there is a unique solution $\mu_{k}$ to
  (\ref{eq:mukder smoothness}) (where $q$ and $\Sigma_k$ are the functions associated with the solution to (\ref{eq:thetaprime})--(\ref{seq:phizphioev});
  cf. Remark~\ref{remark:Sigmapm vs Sigmai}) such that $\mu_{k}-6p_{k}\tau\rightarrow 0$ as
  $\tau\rightarrow-\infty$. Note, to this end, that if all the $p_i$ satisfy $p_i<1$, then the remaining assumptions (excluding the ones involving
  $\nu_{k}$ for $k$ such that $\e_{k}=0$) imply that the $N_{i}$ converge to zero exponentially. If there is an $i$ such that $p_i=1$, we can
  assume that $p_1=1$, so that Lemma~\ref{prop:limitcharsp} applies and implies that $\Nt=\Nth$ and that $\Sm=0$. This means that $\No$ and
  $\No\Nt$ converge to zero exponentially. Moreover, due to the asymptotics of $\phi_{0}$, $\phi_1$ and $\theta$; the fact that
  $|\Phi_1|\leq\sqrt{6}$, see (\ref{eq:Phi sp sm as constraint}); and the assumptions
  concerning $V$, it is clear that $\Omega$ decays to zero exponentially. Combining these observations with (\ref{eq:altqdef}), it is clear that
  $q-2$ converges to zero exponentially. Combining the above arguments with the equations for $\Sigma_{\pm}$, it follows that the $\Sigma_{\pm}$
  converge exponentially to their limits. Returning to (\ref{eq:mukder smoothness}) with this information in mind, the desired conclusion follows. 
\end{remark}
\begin{remark}\label{remark:exist Cauchy horizon data sing}
  Assume that $p_1=1$. Then if $\e_2\neq \e_3$ there is no corresponding solution. Moreover, if $\e_2=\e_3\neq 0$
  and $m_2\neq m_3$, there is no corresponding solution. 
\end{remark}
\begin{proof}
  In the following arguments, it is convenient to work with $\Sigma_\pm$ instead of $\Sigma_i$; cf. Remark~\ref{remark:Sigmapm vs Sigmai}. We
  therefore introduce $s_\pm=\Sigma_\pm-\sigma_\pm$, where the relation between $\Sigma_\pm$ ($\sigma_\pm$) and $\Sigma_i$ ($\sigma_i$) is given by
  (\ref{eq:SigmaiitoSpm}). With this notation, and notation as in (\ref{seq:nuivarkappapsiidos}),
  \begin{subequations}\label{seq:nuk prime etc}
    \begin{align}
      \nu_k' = & q-2+(2-q)f_k(0,\Sigma_+,\Sigma_-)\tau+3f_k(0,S_+,S_-)\tau,\label{eq:nuieq}\\
      s_{\pm}' = & -(2-q)\Sigma_\pm-3S_{\pm},\label{eq:spmeq}\\
      \varkappa' = & 2-q,\\
      \psi_0' = & [(2-q)\phi_1+9\theta^{-2}V'\circ\phi_0]\tau,\\
      \psi_1' = & -(2-q)\phi_1-9\theta^{-2}V'\circ\phi_0.
    \end{align}
  \end{subequations}
  Introducing $x:=(\nu_{1},\nu_{2},\nu_{3},s_{+},s_{-},\psi_{0},\psi_{1},\varkappa)$, the right hand side of (\ref{seq:nuk prime etc})
  defines a time dependent vector field on $\rn{8}$, say $\mX$, depending on the parameters $\Pi$. With this notation,
  (\ref{seq:nuk prime etc}) reads
  \begin{equation}\label{eq:x prime mX Pi}
    x'(\tau)=\mX[x(\tau),\tau;\Pi];
  \end{equation}
  note that since $N_k=\e_k e^{\mu_k}$ and (\ref{seq:nuivarkappapsiidos}) holds, the right hand side of (\ref{seq:nuk prime etc}) can be expressed
  in terms of $x$, $\tau$ and $\Pi$ (we here tacitly assume $q$ to be defined by (\ref{eq:altqdef})).  Next, let $\tau_{0}\leq 0$; $\eta_{i}:=6$ if
  $\e_{i}=0$; $\eta_{i}:=6p_{i}$ if $\e_{i}\neq 0$, where
  \begin{equation}\label{eq:pisitosigmapm}
    p_i:=f_i(2,\sigma_+,\sigma_-)/6=\sigma_i+1/3;
  \end{equation}and $\eta_{\phi}:=3(1-\a_V)$. If all the $p_i<1$, let
  \begin{equation}\label{eq:epsilondef}
    \e:=\min\{\eta_{1},\eta_{2},\eta_{3},\eta_{\phi}\}.
  \end{equation}
  If there is a $p_i=1$, assume $p_1=1$, so that $\sigma_+=-1$, $\sigma_-=0$ and $\Phi_1=0$. Let, moreover, $\e:=\min\{3,\eta_{\phi}\}$. Note that
  $\e>0$ due to the assumptions. Define $X_{\tau_{0}}$ and $d_{\tau_0}$ as in (\ref{eq:d tau zero def}) and the adjacent text, with $k=8$. Here and
  below, we use the notation
  \begin{equation}\label{eq:xadef}
    x_{a} = (\nu_{1,a},\nu_{2,a},\nu_{3,a},s_{+,a},s_{-,a},\psi_{0,a},\psi_{1,a},\varkappa_{a})
  \end{equation}
  and similarly for $x_b$. Next, we wish to verify that (\ref{eq:mX bd and lip bd}) is satisfied. Assume, to this end, that $x_a,x_b\in X_{\tau_1}$
  for some $\tau_1\leq\tau_0$ and note that
  \begin{subequations}\label{seq:thetaaphizaphioneadef}
    \begin{align}
      N_{i,a}(\tau) = & n_i\exp[6p_i\tau+\nu_{i,a}(\tau)+f_i(0,s_{+,a}(\tau),s_{-,a}(\tau))\tau],\\
      \theta_{a}(\tau) = & \theta_{\infty}\exp[-3\tau+\varkappa_{a}(\tau)],\\
      \phi_{1,a}(\tau) = & \psi_{1,a}(\tau)+\Phi_{1},\\
      \phi_{0,a}(\tau) = & \psi_{0,a}(\tau)+[\psi_{1,a}(\tau)+\Phi_{1}]\tau+\Phi_{0},    
    \end{align}
  \end{subequations}
  where $n_i:=\e_i e^{m_i}$ and $\theta_\infty:=e^{\varkappa_\infty}$. In case $\sigma_+=-1$, we assume $N_{2,a}=N_{3,a}$ and $\Sigma_{-,a}=0$, since this
  condition is necessary in case $(\Sp,\Sm)\rightarrow (-1,0)$; see Proposition~\ref{prop:limitcharsp}.  
  Replacing $\phi_{0}$ and $\theta$ in the second equality in (\ref{eq:PsiOmegadef}) with $\phi_{0,a}$ and $\theta_{a}$ respectively yields
  $\Omega_{a}$. Replacing $N_{i}$ and $\Omega$ in (\ref{eq:altqdef}) with $N_{i,a}$ and $\Omega_{a}$ respectively yields $q_{a}$. Moreover,
  replacing the $N_{i}$ appearing in the formulae for $S_{\pm}$ with $N_{i,a}$ yields $S_{\pm,a}$. Thus
  \[
  |S_{+,a}(\tau)|+|S_{-,a}(\tau)|\leq C(n_{1}^{2}+n_{2}^{2}+n_{3}^{2})e^{2\e\tau}
  \]
  for $\tau\leq\tau_{1}$, where $C$ is a numerical constant. Moreover, there is a constant $C$, depending only on $\e$, such that
  \[
  \Omega_{a}(\tau)+\theta_a(\tau)^{-2}|V'[\phi_{0,a}(\tau)]|\leq C c_{1}\theta_{\infty}^{-2}e^{\sqrt{6}|\Phi_{0}|}e^{2\eta_{\phi}\tau}
  \]
  for $\tau\leq\tau_{1}$, where we used the fact that $|\Phi_{1}|\leq \sqrt{6}$ and $c_1$ is the constant appearing in
  (\ref{eq:V k-derivatives exp estimate}). Introducing
  \begin{equation}\label{eq:K infty def}
    K_{\infty}:=n_{1}^{2}+n_{2}^{2}+n_{3}^{2}+c_{2}\theta_{\infty}^{-2}e^{\sqrt{6}|\Phi_{0}|},
  \end{equation}
  we conclude that there is a constant $C$, depending only on $\e$, such that
  \begin{equation}\label{eq:SpmaqamtwoVprbd}
    |S_{+,a}(\tau)|+|S_{-,a}(\tau)|+|q_{a}(\tau)-2|+\theta_a(\tau)^{-2}|V'[\phi_{0,a}(\tau)]|\leq CK_{\infty}e^{2\e\tau}
  \end{equation}
  for all $\tau\leq \tau_1$. Combining (\ref{eq:SpmaqamtwoVprbd}) with the fact that $|\Sigma_{\pm,a}(\tau)|\leq 2$ and
  $|\phi_{1,a}(\tau)|\leq\sqrt{6}+1$ for $\tau\leq \tau_{1}$, it follows that
  \begin{equation}\label{eq:mX xa Pi est}
    |\mX[x_a(\tau),\tau;\Pi]|\leq CK_{\infty}\ldr{\tau}e^{2\e\tau}
  \end{equation}
  for $\tau\leq \tau_{1}$, where $C$ only depends on $\e$. Thus the first estimate in (\ref{eq:mX bd and lip bd}) holds. 
  
  Next, we prove the second estimate in (\ref{eq:mX bd and lip bd}). Note, to this end, that (assuming one of $i,j$ equals $1$ in case
  $\sigma_+=-1$)
  \begin{equation*}
    \begin{split}
      N_{i,a}N_{j,a}-N_{i,b}N_{j,b}
      = & n_{i}n_{j}e^{2\e_{ij}\tau}\exp[\nu_{i,b}+\nu_{j,b}+f_i(0,s_{+,b},s_{-,b})\tau+f_j(0,s_{+,b},s_{-,b})\tau]\\
       & (\exp[\de\nu_{i}+\de\nu_{j}+f_i(0,\de s_{+},\de s_{-})\tau+f_j(0,\de s_{+},\de s_{-})\tau]-1),
    \end{split}
  \end{equation*}  
  where $\e_{ij}\geq\e$, $\de\nu_i:=\nu_{i,a}-\nu_{i,b}$ and $\de s_\pm:=s_{\pm,a}-s_{\pm,b}$. Due to the definition of $X_{\tau_{1}}$ and the
  fact that $\tau_1\leq 0$, it follows that $|\nu_{i,a}|\leq 1$, $|s_{+,b}(\tau)|\leq e^{\e\tau}$ etc., so that
  \begin{equation*}
    \begin{split}
      |N_{i,a}N_{j,a}-N_{i,b}N_{j,b}| \leq & C|n_{i}n_{j}|\ldr{\tau}e^{2\e_{ij}\tau}(|\de\nu_{i}|+|\de\nu_{j}|+|\de s_+|+|\de s_-|)\\
      \leq & C|n_{i}n_{j}|\ldr{\tau}e^{(2\e_{ij}+\e)\tau}d_{\tau_1}(x_{a},x_{b})
    \end{split}    
  \end{equation*}
  for $\tau\leq \tau_{1}$, where $C$ only depends on $\e$. Next,
  \begin{equation}\label{eq:theta a rt minus two Va diff}
    \left|\theta_a^{-2}V\circ\phi_{0,a}-\theta_b^{-2}V\circ\phi_{0,b}\right|
    \leq \theta_a^{-2}|V\circ\phi_{0,a}-V\circ\phi_{0,b}|+|V\circ\phi_{0,b}|\cdot|\theta_{a}^{-2}-\theta_{b}^{-2}|.
  \end{equation}
  Keeping in mind that
  \begin{equation}\label{eq:V xi minus V zeta}
    V(\xi)-V(\zeta)=\textstyle{\int}_{0}^{1}V'[s\xi+(1-s)\zeta]ds\cdot (\xi-\zeta)
  \end{equation}
  as well as the definitions of $\phi_{0,a}$, $\phi_{0,b}$, $\theta_{a}$ and $\theta_{b}$, the estimates for $V$ etc., it can be verified that
  \begin{equation}\label{eq:Vdiffest}
    \begin{split}
      \left|\theta_{a}^{-2}V\circ\phi_{0,a}-\theta_{b}^{-2}V\circ\phi_{0,b}\right|
      \leq & CK_{\infty}e^{2\e\tau}[|\delta\psi_{0}|+\ldr{\tau}|\delta\psi_{1}|+|\delta\varkappa|]\\
      \leq & CK_{\infty}\ldr{\tau}e^{3\e\tau}d_{\tau_1}(x_{a},x_{b})
    \end{split}
  \end{equation}
  for $\tau\leq\tau_{0}$, where $C$ only depends on $\e$; $\delta\psi_i:=\psi_{i,a}-\psi_{i,b}$; and $\delta\varkappa:=\varkappa_{a}-\varkappa_{b}$. There
  is a similar estimate with $V$ replaced by $V'$. However, this estimate necessitates control of two derivatives of $V$. It is here we use the fact that
  $V\in\mfP_{\a_V}^2$ and the fact that $c_2$ appears in the definition of $K_\infty$, see (\ref{eq:K infty def}), and not $c_1$. By arguments of this
  type, it can be verified that
  \[
  |\mX[x_{a}(\tau),\tau;\Pi]-\mX[x_{b}(\tau),\tau;\Pi]|\leq CK_{\infty}\ldr{\tau}^2e^{3\e\tau}d_{\tau_1}(x_{a},x_{b}).
  \]
  Thus the second estimate in (\ref{eq:mX bd and lip bd}) holds.

  \textbf{Existence and uniqueness.} That there is a unique solution $x\in C^{1}(I_{\max},\rn{k})$ to (\ref{eq:x prime mX Pi}) now follows from
  Lemma~\ref{lemma:abs exist and unique}. However, the uniqueness is in the class of solutions satisfying $e^{-\e\tau}|x(\tau)|\rightarrow 0$
  as $\tau\rightarrow -\infty$. Note that $x$ satisfies (\ref{seq:nuk prime etc}). Define $N_i$, $\phi_0$, $\phi_1$ and $\theta$ in analogy
  with (\ref{seq:thetaaphizaphioneadef}) and $\Sigma_{\pm}$ by $\Sigma_{\pm}:=s_{\pm}+\sigma_{\pm}$. Finally, the $S_{\pm}$ are defined by
  (\ref{seq:SpSmev}) and $q$ is defined by (\ref{eq:altqdef}).  It can then be verified that (\ref{eq:thetaprime})--(\ref{seq:phizphioev})
  are satisfied. Finally, note that, due to
  the fact that we replaced $q$ by (\ref{eq:altqdef}) in the equations, the identity (\ref{eq:fdiffeq}) holds. Since $2-q$ is integrable, the
  fact that the limit of $f$ is zero implies that $f$ vanishes identically. Thus (\ref{eq:constraint}) holds. Combining (\ref{eq:constraint})
  and (\ref{eq:altqdef}) yields (\ref{eq:qdef}). Finally $x(\tau)\rightarrow 0$ as $\tau\rightarrow-\infty$, so that (\ref{eq:desiredlimit})
  holds. In order to prove uniqueness for solutions $x$ such that $x(\tau)\rightarrow 0$ as $\tau\rightarrow-\infty$, we need to prove that
  convergence to zero implies that $e^{-\e\tau}|x(\tau)|\rightarrow 0$ as $\tau\rightarrow -\infty$. Assume that there is a solution to
  (\ref{eq:thetaprime})--(\ref{seq:phizphioev}) with the properties stated in the proposition. Then $x(\tau)\rightarrow 0$
  as $\tau\rightarrow-\infty$ due to (\ref{eq:desiredlimit}), where $x:(-\infty,\tau_{0}]\rightarrow\rn{8}$ is well defined for some choice of
  $\tau_0\in\rn{}$. Moreover, $x$ is smooth. Due to the assumptions, all the $N_{i}$ converge to zero exponentially (except if $\sigma_+=-1$,
  in which case $\No$ and $\No\Nt$ converge to zero exponentially). The same is true of
  $\Omega$ and $\theta^{-2}V'\circ\phi_0$. Due to the definition of $q$, it follows that $q-2$ converges to zero exponentially. Using these
  observations, it can be demonstrated (by an argument similar to the proof of (\ref{eq:mX xa Pi est})) that $e^{-\e\tau}|x(\tau)|\rightarrow 0$
  as $\tau\rightarrow -\infty$. Existence and uniqueness in the general setting follow.

  \textbf{Symmetry.} The final statement of the proposition follows by uniqueness and the symmetries of the equations. 

  In order to prove the statement in Remark~\ref{remark:exist Cauchy horizon data sing}, assume that there is a solution corresponding to data on
  the singularity with $\sigma_+=-1$ and $\sigma_-=0$. Then, due to Proposition~\ref{prop:limitcharsp}, $\Sm=0$ and $\Nt=\Nth$. This is clearly
  incompatible with $\e_2\neq \e_3$. Assume now, in addition, that $\e_2=\e_3\neq 0$. Then $\mu_2=\mu_3$ (since $\Nt=\Nth$). On the other hand,
  $\nu_2=\nu_3$; since $\Sm=0$ and $S_-=0$, it follows that $\nu_2'=\nu_3'$, so that $\nu_2=\nu_3$ since $\nu_2$ and $\nu_3$ both converge to zero
  as $\tau\rightarrow-\infty$. From (\ref{eq:Niitonui}) and the fact that $\Sm=0$, it then follows that $m_2=m_3$. 
\end{proof}

\chapter[Map to singularity]{The map taking regular initial data to data on the singularity}\label{chapter:map to sing}

For many Bianchi class A types, initial data give rise to data on the singularity, see Chapter~\ref{chapter:as to data on sing}. In the present
chapter, we take the first steps in the analysis of the properties of the corresponding map.

\section[Smoothness, map to singularity]{Smoothness of the map from regular initial data to the
  leading order asymptotics}\label{ssection:contregdatatoasdata}

In the end, we are interested in analysing the map taking isometry classes of regular initial data (or isometry classes of developments) to isometry
classes of initial data on the singularity; cf. Lemma~\ref{lemma:iso ri iso ids incl type}. In particular, we wish to prove that the corresponding map is
defined on an open subset and smooth. In practice, this involves analysing the behaviour of solutions to (\ref{eq:thetaprime}) and
(\ref{seq:NoSmprime})--(\ref{eq:constraint}). However, in what follows, it is convenient to use the
variables $\Sigma_i$ as opposed to $\Sigma_\pm$; see Remark~\ref{remark:Sigmapm vs Sigmai}. We begin by introducing some terminology concerning the
background solutions we perturb in the results below. 

\begin{definition}\label{def:admconvsol}
  Consider a solution $(\No,\Nt,\Nth,\Sigma_1,\Sigma_2,\Sigma_3,\theta,\phi_{0},\phi_{1})$ to (\ref{eq:thetaprime})--(\ref{eq:constraint}), cf.
  Remark~\ref{remark:Sigmapm vs Sigmai}, where
  $V\in C^{\infty}(\rn{})$. It is said to be an \textit{admissible convergent solution} if the following conditions are satisfied.
  First, the existence interval includes $(-\infty,T]$ for some $T\leq 0$ and there are constants $\sigma_{i}$, $i=1,2,3$, and $\Phi_{1}$ such that
  \begin{equation}\label{eq:SpSmphionelim}
    \lim_{\tau\rightarrow-\infty}[\Sigma_1(\tau),\Sigma_2(\tau),\Sigma_3(\tau),\phi_{1}(\tau)]=(\sigma_{1},\sigma_{2},\sigma_3,\Phi_{1}).
  \end{equation}
  Next, there is a constant $C_{\theta}$ such that $\ln\theta(\tau)\geq -3\tau-C_{\theta}$ for all $\tau\leq T$. Finally, one of the following conditions
  are satisfied:
  \begin{enumerate}[(i)]
  \item If $N_i\neq 0$, then $\sigma_i+1/3>0$. Moreover, there are constants $C_{\infty}$ and $\a_{\infty}>0$ such that
    $|N_{i}(\tau)|\leq C_{\infty}e^{\a_{\infty}\tau}$ for all $i\in\{1,2,3\}$ and all $\tau\leq T$.
  \item The solution is of LRS Bianchi type VIII or IX. Moreover, $\Phi_1=0$ and there are $\{i,j,k\}=\{1,2,3\}$ such that $\sigma_i=2/3$
    and $\sigma_j=\sigma_k=-1/3$.
  \end{enumerate}    
\end{definition}
\begin{remark}\label{remark:ex of admiss}
  If the conditions of one of the propositions in Chapter~\ref{chapter:as to data on sing} are satisfied, then the corresponding solution to
  (\ref{eq:thetaprime})--(\ref{eq:constraint}) is an admissible convergent solution. We leave the verification of this
  statement to the reader. 
\end{remark}
Before addressing the differentiability of the map from initial data to asymptotic data, it is convenient to change variables. Note, to this end, that
if $\mu_{k}$ are functions satisfying (\ref{eq:mukder smoothness}), then $N_{k}$ can be written $N_{k}=\e_{k}e^{\mu_{k}}$ (no summation), where $\e_{k}$
equals $-1$, $0$ or $1$, assuming the initial data for the $\mu_{k}$ are chosen appropriately. Introduce the variables
\begin{equation}\label{eq:nui def}
  \nu_k := \mu_k-6(\Sigma_k+1/3)\tau
\end{equation}
(the notation used in the present chapter differs from that used in Section~\ref{ssection:dataonsingexpnormvar}). Define, moreover,
\[
  \varkappa:=\ln\theta+3\tau,\ \ \ \psi_0:=\phi_0-\phi_1\tau
\]
(in what follows, we are only interested in asymptotic regimes in which $\theta>0$). Then 
\begin{subequations}\label{seq:nui etc eqs}
  \begin{align}
    \nu_k' = & q-2+6(2-q)\Sigma_k\tau+18S_k\tau,\label{eq:nuk der mts}\\    
    \Sigma_{k}' = & -(2-q)\Sigma_{k}-3S_{k},\\
    \varkappa' = & 2-q,\\
    \psi_0' = & [(2-q)\phi_1+9\theta^{-2}V'\circ\phi_{0}]\tau,\\
    \phi_1' = & -(2-q)\phi_1-9\theta^{-2}V'\circ\phi_{0}
  \end{align}
\end{subequations}
for $k=1,2,3$, where it is understood that $\phi_0=\psi_0+\phi_1\tau$, that $q$ is given by (\ref{eq:altqdef}), that $\theta=e^{-3\tau+\varkappa}$,
that $\Omega$ is defined by (\ref{seq:PsiPhiphiiqdef}) and that the $S_i$ are given by (\ref{eq:Si def}). Moreover, it is understood that
$N_{k}=\e_{k}e^{\mu_{k}}$ (no summation), where the $\mu_k$ are defined in terms of the $\nu_k$ and $\Sigma_k$ via (\ref{eq:nui def}).
For these equations, the constraint (\ref{eq:constraint}) takes the form
\begin{equation}\label{eq:constraint asympt var}
  \begin{split}
    1 = & \frac{3}{2}(\Sigma_1^2+\Sigma_2^2+\Sigma_3^2)+\frac{1}{6}\phi_1^2+3e^{6\tau-2\varkappa}V(\psi_0+\phi_1\tau)\\
    & +\frac{3}{4}[\e_1^2 e^{2\mu_1}+\e_2^2 e^{2\mu_2}+\e_3^2 e^{2\mu_3}-2(\e_1\e_2 e^{\mu_1+\mu_2}+\e_2\e_3 e^{\mu_2+\mu_3}+\e_1\e_3 e^{\mu_1+\mu_3})].
  \end{split}
\end{equation}
Given initial data at $\tau_0$ satisfying (\ref{eq:constraint asympt var}) with $\tau$ replaced by $\tau_0$, the system (\ref{seq:nui etc eqs}) is
equivalent to (\ref{eq:thetaprime})--(\ref{eq:constraint}) in the case of Bianchi types VIII and IX. For the other Bianchi types, there is a
redundancy: if $\e_k=0$, there is no need to keep track of $\nu_k$ since it decouples from the rest of the equations. On the other hand, if $V$ is
non-negative (the case we are most interested in), and the Bianchi type is different from IX, then the existence interval of solutions to
(\ref{eq:thetaprime})--(\ref{eq:constraint}) is $\ro$, see Lemma~\ref{lemma:globalexistenceresceq}. If $\e_k=0$, we can therefore define $\nu_k$ by
(\ref{eq:nuk der mts}) and the condition that $\nu_k(0)=0$. Here we are interested in solutions such
that the right hand sides in (\ref{seq:nui etc eqs}) all converge to zero exponentially. This means that we, with the initial data, can associate
asymptotic data, which are the limits of the $\nu_i$ etc. In order to, later on, obtain conclusions concerning the map from isometry classes of
initial data to the Einstein non-linear scalar field equations to isometry classes of initial data on the singularity, it is convenient to
specify initial data for $N_i$ instead of for $\nu_i$. If the initial datum at $\tau=0$ for $N_i$ is $n_i$, then we define
$\nu_i(0):=\ln |n_i|$ in case $n_i\neq 0$ and $\nu_i(0):=0$ otherwise. In order to be more specific, we introduce the following terminology.
\begin{definition}
  Fix $V\in C^{\infty}(\ro)$. Then $\sfP\subset\rn{9}$ denotes the set of
  \index{$\a$Aa@Notation!Symmetry reduced sets of regular initial data!sfP@$\sfP$}%
  \[
  (n_1,n_2,n_3,\sigma_1,\sigma_2,\sigma_3,\varkappa_0,\bpsi_0,\bpsi_1)\in\rn{9}
  \]
  such that $\sigma_1+\sigma_2+\sigma_3=0$ and 
  \begin{equation*}
    \begin{split}
      1 = & \frac{3}{2}(\sigma_1^2+\sigma_2^2+\sigma_3^2)+\frac{1}{6}\bpsi_1^2+3e^{-2\varkappa_0}V(\bpsi_0)
      +\frac{3}{4}[n_1^2+n_2^2+n_3^2-2(n_1n_2+n_2n_3+n_3n_1)],
    \end{split}
  \end{equation*}
  and such that if $\sigma_i=n_i=0$ for $i\in\{1,2,3\}$, then either $\bpsi_1\neq 0$ or $V'(\bpsi_0)\neq 0$. 
  Next, let $\sfP_{\mrI}$ denote the subset of $\sfP$ such that all the $n_i$ vanish; let $\sfP_{\mrII}$ denote
  the subset of $\sfP$ such that one of the $n_i$ is non-zero and the other two vanish; let $\sfP_{\mrVIz}$ denote the subset of $\sfP$ such that two
  of the $n_i$ are non-zero and have opposite signs, and the remaining $n_i$ vanishes; let $\sfP_{\mrVIIz}$
  \index{$\a$Aa@Notation!Symmetry reduced sets of regular initial data!sfPT@$\sfP_\mfT$}%
  denote the subset of $\sfP$ such that two
  of the $n_i$ are non-zero and have the same signs, and the remaining $n_i$ vanishes; let $\sfP_{\mrVIII}$ denote the subset of $\sfP$ such that all
  the $n_i$ are non-zero, but they do not all have the same sign; and let $\sfP_{\mrIX}$ denote the subset of $\sfP$ such that all the $n_i$ are
  non-zero, and all have the same sign.
\end{definition}
\begin{remark}
  Note that $\sfP$ is a smooth manifold.
\end{remark}
In analogy with (\ref{eq:psi sigma pm}), we introduce
\begin{equation}\label{eq:xi lambda pm}
  \begin{split}
    & \xi_{\lambda}^{\pm}(n_1,n_2,n_3,\sigma_1,\sigma_2,\sigma_3,\varkappa_0,\bpsi_0,\bpsi_1)\\
    := & (\pm n_{\lambda(1)},\pm n_{\lambda(2)},\pm n_{\lambda(3)},\sigma_{\lambda(1)},\sigma_{\lambda(2)},\sigma_{\lambda(3)},\varkappa_0,\bpsi_0,\bpsi_1)
  \end{split}  
\end{equation}
\index{$\a$Aa@Notation!Maps!xilambdapm@$\xi_{\lambda}^{\pm}$}%
for $\lambda\in S_3$. Note that $\xi_{\lambda}^{\pm}$ maps $\sfP$ to $\sfP$ and $\sfP_\mfT$ to $\sfP_\mfT$ for a Bianchi class A type $\mfT$.
However, with this definition, $\xi_\lambda^+(p)=\xi_\lambda^-(p)$ if $p\in\sfP_{\mrI}$. As in the case of $\psi_\lambda^\pm$, we therefore impose
the following restriction on the domains of definition:
\[
  \xi_\lambda^+:\sfP\rightarrow\sfP,\ \ \
  \xi_\lambda^-:\sfP-\sfP_{\mrI}\rightarrow\sfP-\sfP_{\mrI}.
\]
In order to relate $\sfP$ with the set $\sfB$ introduced in Definition~\ref{def:degenerate version of id}, it is convenient to define the map
$\sfF_\sfP:\sfP\rightarrow\sfB$:
\begin{equation}\label{eq:sfF sfP def}
  \begin{split}
    & \sfF_\sfP(n_1,n_2,n_3,\sigma_1,\sigma_2,\sigma_3,\varkappa_0,\bpsi_0,\bpsi_1)\\
    := & (e^{\varkappa_0}n_1,e^{\varkappa_0}n_2,e^{\varkappa_0}n_3,(\sigma_1+1/3)e^{\varkappa_0},(\sigma_2+1/3)e^{\varkappa_0},(\sigma_3+1/3)e^{\varkappa_0},
    \bpsi_0,e^{\varkappa_0}\bpsi_1/3).
  \end{split}
\end{equation}
\index{$\a$Aa@Notation!Maps!sfFP@$\sfF_\sfP$}%
\begin{lemma}\label{lemma:mfB mP rel}
  Let $\mfT$ be a Bianchi class A type. Then $\sfF_\sfP$ defines a diffeomorphism from $\sfP_\mfT$ to $\sfB_\mfT\cap\{\tr K>0\}$.
  Finally, if $\lambda\in S_3$, then
  \begin{equation}\label{eq:sfF sym comm}
    \psi_{\lambda}^{\pm}\circ\sfF_\sfP=\sfF_\sfP\circ\xi_\lambda^{\pm}.
  \end{equation}
\end{lemma}
\begin{proof}
  We leave the proof to the reader.
\end{proof}
In analogy with Definition~\ref{def:Gamma plus etc}, we introduce the following terminology. 
\begin{definition}\label{def:Lambda plus etc}
  Let $\Lambda$ denote the group consisting of the elements $\xi_\lambda^{\pm}$. Let $\Lambda^{\roev}$ ($\Lambda^{\roodd}$)
  \index{$\a$Aa@Notation!Symmetry groups!Lambdaev@$\Lambda^{\roev}$, $\Lambda^{\roodd}$}%
  consist of the set of $\xi_\lambda^\pm$ such that $\lambda$ is even (odd). Let $\Lambda^{\pm}$
  \index{$\a$Aa@Notation!Symmetry groups!Lambdapm@$\Lambda^{\pm}$}%
  be the subset of
  $\Lambda$ consisting of the maps $\xi_\lambda^\pm$ for $\lambda\in S_3$. Finally, let $\Lambda^{+,\roev}:=\Lambda^+\cap\Lambda^{\roev}$
  \index{$\a$Aa@Notation!Symmetry groups!Lambdapmev@$\Lambda^{\pm,\roev}$, $\Lambda^{\pm\roodd}$}%
  etc. 
\end{definition}
In analogy with Definition~\ref{def:BTs GTs}, we introduce the following terminology.
\begin{definition}\label{def:BTs LTs}
  Fix $V\in C^{\infty}(\ro)$. Let $\mfT$ be a Bianchi class A type. Define $\sfP_\mfT^{\iso}$ to be the subset of $\sfP_\mfT$ consisting of points invariant
  under $\Lambda^+$.
  Define $\sfP_{\mrI}^{\roper}=\varnothing$ and define, for $\mfT\neq\mrI$, $\sfP_{\mfT}^{\roper}$ to be the set of $p\in\sfP_{\mfT}$ such that there is a map
  $\xi_\lambda^-\in \Lambda^{-,\roodd}$ with $\xi_\lambda^-(p)=p$. Define
  $\sfP_{\mfT}^{\roLRS}$ to be the set of $p\in\sfP_{\mfT}-\sfP_{\mfT}^{\iso}$ such that there is a map $\xi_\lambda^+\in \Lambda^{+,\roodd}$ with $\xi_\lambda^+(p)=p$.
  Define
  \[
  \sfP_{\mfT}^{\rogen}:=\sfP_{\mfT}-\sfP_{\mfT}^{\iso}-\sfP_{\mfT}^{\roper}-\sfP_{\mfT}^{\roLRS}.
  \]
  Define $\Lambda_\mrI^{\iso}:=\{\roId\}$ and $\Lambda_{\mrIX}^{\iso}:=\{\roId,\xi^{-}_{\roId}\}$. Define $\Lambda_{\mrVIz}^{\roper}:=\Lambda^{\roev}$,
  $\Lambda_{\mrI}^{\roLRS}:=\Lambda^{+,\roev}$ and $\Lambda_{\mfT}^{\roLRS}:=\Lambda^{\roev}$ for $\mfT\notin \{\mrI,\mrVIz\}$. Finally, define
  $\Lambda_{\mrI}^{\rogen}:=\Lambda^+$ and $\Lambda_\mfT^{\rogen}:=\Lambda$ for $\mfT\neq\mrI$. 
\end{definition}
\begin{remark}
  Note that $\sfP_\mfT^{\iso}$ is empty unless $\mfT\in \{\mrI,\mrIX\}$; $\sfP_\mfT^{\roper}$ is empty unless $\mfT=\mrVIz$;
  and $\sfP_{\mrVIz}^{\roLRS}$ is empty.
\end{remark}
In analogy with Lemma~\ref{lemma:sfR mfT mfs}, the following statement holds. 
\begin{lemma}\label{lemma:sfQ mfT mfs}
  Fix $V\in C^{\infty}(\ro)$. Let $\mfT$ be a Bianchi class A type and $\mfs\in\{\iso,\roper,\roLRS,\rogen\}$ be such that $\sfP_{\mfT}^{\mfs}$ and
  $\Lambda_{\mfT}^{\mfs}$ are defined in Definition~\ref{def:BTs LTs} and non-empty. Then $\Lambda_{\mfT}^{\mfs}$ acts freely and properly discontinuously on
  $\sfP_{\mfT}^{\mfs}$ and the corresponding quotient is denoted $\sfN_{\mfT}^{\mfs}$.
  \index{$\a$Aa@Notation!Symmetry reduced sets of regular initial data!sfNTs@$\sfN_{\mfT}^{\mfs}$}%
\end{lemma}
\begin{proof}
  Since $\Lambda_\mfT^\mfs=\sfF_\sfP^{-1}\circ\Gamma_\mfT^\mfs\circ\sfF_{\sfP}$ and $\sfP_\mfT^\mfs=\sfF_\sfP^{-1}[\sfB_\mfT^\mfs\cap\{\tr K>0\}]$, the statement is an
  immediate consequence of Lemma~\ref{lemma:sfR mfT mfs}. 
\end{proof}
In some situations it is of interest to fix $\varkappa$; note that under the diffeomorphism between $\sfP_\mfT$ and $\sfB_\mfT$, see
Lemma~\ref{lemma:mfB mP rel}, fixing $\varkappa$
corresponds to fixing the initial mean curvature. We therefore, for a fixed $\varkappa_0\in\ro$, introduce the notation $\sfP_{\mfT}(\varkappa_0)$
\index{$\a$Aa@Notation!Symmetry reduced sets of regular initial data!sfPTvka@$\sfP_{\mfT}(\varkappa_0)$}%
and
$\sfP_{\mfT}^{\mfs}(\varkappa_0)$
\index{$\a$Aa@Notation!Symmetry reduced sets of regular initial data!sfPTsvka@$\sfP_{\mfT}^\mfs(\varkappa_0)$}%
for the subsets of $\sfP_{\mfT}$ and $\sfP_{\mfT}^{\mfs}$ with the initial datum for $\varkappa$ equalling $\varkappa_0$. Since we have eliminated trivial
initial data, these sets are manifolds. 

In the case of Bianchi type IX, we need to introduce terminology analogous to the one introduced in Definitions~\ref{definition:Bap Bp Bpp} and
\ref{def:plus ap pp nd}. 
We use the notation $\sfP_{\mrIX,+}^{\rond}$, $\sfP_{\mrIX,+}^{\mfs,\rond}$, $\sfP_{\mrIX,+}^{\rond}(\varkappa_0)$ and $\sfP_{\mrIX,+}^{\mfs,\rond}(\varkappa_0)$
\index{$\a$Aa@Notation!Symmetry reduced sets of regular initial data!sfPIXpnd@$\sfP_{\mrIX,+}^{\rond}$, $\sfP_{\mrIX,+}^{\mfs,\rond}$,
  $\sfP_{\mrIX,+}^{\rond}(\varkappa_0)$, $\sfP_{\mrIX,+}^{\mfs,\rond}(\varkappa_0)$}%
to denote the subsets of $\sfP_{\mrIX}$, $\sfP_{\mrIX}^{\mfs}$, $\sfP_{\mrIX}(\varkappa_0)$ and $\sfP_{\mrIX}^{\mfs}(\varkappa_0)$ respectively, such that the
corresponding solutions to (\ref{seq:nui etc eqs}) have the following properties: $\varkappa'(\tau)-3<0$ for all $\tau\leq 0$; and there is a
$\tau_0\leq 0$ such that
\[
\mfX[e^{-3\tau_0+\varkappa(\tau_0)},\psi_0(\tau_0)+\phi_1(\tau_0)\tau_0,e^{-3\tau_0+\varkappa(\tau_0)}\phi_1(\tau_0)/3]>0.
\]
Note here that initial data in $\sfP$ correspond to data at $\tau=0$ for (\ref{seq:nui etc eqs}). Next, we use the notation $\sfP_{\mrIX,\roap}^{\rond}$,
$\sfP_{\mrIX,\roap}^{\mfs,\rond}$, $\sfP_{\mrIX,\roap}^{\rond}(\varkappa_0)$ and $\sfP_{\mrIX,\roap}^{\mfs,\rond}(\varkappa_0)$
\index{$\a$Aa@Notation!Symmetry reduced sets of regular initial data!sfPIXapnd@$\sfP_{\mrIX,\roap}^{\rond}$, $\sfP_{\mrIX,\roap}^{\mfs,\rond}$,
  $\sfP_{\mrIX,\roap}^{\rond}(\varkappa_0)$, $\sfP_{\mrIX,\roap}^{\mfs,\rond}(\varkappa_0)$}%
to denote the subsets of $\sfP_{\mrIX}$,
$\sfP_{\mrIX}^{\mfs}$, $\sfP_{\mrIX}(\varkappa_0)$ and $\sfP_{\mrIX}^{\mfs}(\varkappa_0)$ respectively, such that the corresponding solutions to
(\ref{seq:nui etc eqs}) have the following properties: $\varkappa'(\tau)-3<0$ for all $\tau\leq 0$; 
\[
\lim_{\tau\rightarrow-\infty}\big(e^{6\tau-2\varkappa(\tau)}\mfX[e^{-3\tau+\varkappa(\tau)},\psi_0(\tau)+\phi_1(\tau)\tau,e^{-3\tau+\varkappa(\tau)}\phi_1(\tau)/3]\big)_-=0,
\]
where, for $x\in\ro$, $x_-=x$ if $x\leq 0$ and $x_-=0$ if $x>0$; and
\[
\lim_{\tau\rightarrow-\infty}\big[e^{6\tau-2\varkappa(\tau)}V(\psi_0(\tau)+\phi_1(\tau)\tau)\big]=0.
\]
Finally,
\begin{align*}
  \sfP_{\mrIX,\ropp}^{\rond} := & \sfP_{\mrIX,+}^{\rond}\cup \sfP_{\mrIX,\roap}^{\rond},\ \ \
  \sfP_{\mrIX,\ropp}^{\rond}(\varkappa_0):=\sfP_{\mrIX,+}^{\rond}(\varkappa_0)\cup \sfP_{\mrIX,\roap}^{\rond}(\varkappa_0),\\
  \sfP_{\mrIX,\ropp}^{\mfs,\rond} := & \sfP_{\mrIX,+}^{\mfs,\rond}\cup \sfP_{\mrIX,\roap}^{\mfs,\rond},\ \ \
  \sfP_{\mrIX,\ropp}^{\mfs,\rond}(\varkappa_0):=\sfP_{\mrIX,+}^{\mfs,\rond}(\varkappa_0)\cup \sfP_{\mrIX,\roap}^{\mfs,\rond}(\varkappa_0).
\end{align*}
\index{$\a$Aa@Notation!Symmetry reduced sets of regular initial data!sfPIXppnd@$\sfP_{\mrIX,\ropp}^{\rond}$, $\sfP_{\mrIX,\ropp}^{\mfs,\rond}$,
  $\sfP_{\mrIX,\ropp}^{\rond}(\varkappa_0)$, $\sfP_{\mrIX,\ropp}^{\mfs,\rond}(\varkappa_0)$}%
Next, we introduce terminology for the asymptotic information. 
\begin{definition}\label{def:ma mfT mfs}
  Let $V\in C^{\infty}(\ro)$, $\mfT$ be a Bianchi class A type and $\mfs$ be a corresponding symmetry type. Let $\sfA_{\mfT}$
  \index{$\a$Aa@Notation!Symmetry reduced sets of regular initial data!sfAT@$\sfA_{\mfT}$}%
  denote the subset of
  $\sfP_{\mfT}$ whose elements give rise to admissible convergent solutions  in the sense of Definition~\ref{def:admconvsol}, and let
  $\sfA_{\mfT}^{\mfs}:=\sfA_{\mfT}\cap \sfP_{\mfT}^{\mfs}$.
  \index{$\a$Aa@Notation!Symmetry reduced sets of regular initial data!sfATs@$\sfA_{\mfT}^\mfs$}%
  Define $\sfQ_{\mfT}$
  \index{$\a$Aa@Notation!Symmetry reduced sets of regular initial data!sfQT@$\sfQ_{\mfT}$}%
  to be the set of elements
  \[
  (\bn_1,\bn_2,\bn_3,\sigma_1,\sigma_2,\sigma_3,\varkappa_\infty,\Phi_0,\Phi_1)\in \rn{9}
  \]
  with the following properties. First, all the $\bn_i$ vanish if $\mfT=\mrI$; exactly one $\bn_i$ is non-zero if $\mfT=\mrII$; two $\bn_i$ are
  non-zero and have opposite signs and one $\bn_i$ vanishes if $\mfT=\mrVIz$; two $\bn_i$ are non-zero and have the same sign and one $\bn_i$
  vanishes if $\mfT=\mrVIIz$; all three $\bn_i$ are non-zero, but they do not all have the same sign if $\mfT=\mrVIII$; and all three $\bn_i$
  are non-zero and have the same sign if $\mfT=\mrIX$. Second, $\sigma_1+\sigma_2+\sigma_3=0$. Third,
  \begin{equation}\label{eq:as Ham con pi Phione}
    p_1^2+p_2^2+p_3^2+\tfrac{1}{9}\Phi_1^2=1,
  \end{equation}
  where $p_i:=\sigma_i+1/3$. Fourth, if $\{i,j,k\}=\{1,2,3\}$ and $p_i=1$, then $\bn_j=\bn_k$. Fifth, if $p_l<1$ for all $l$ and $\bn_i\neq 0$
  for some $i$, then $p_i>0$. 
\end{definition}
Clearly, $\xi_{\lambda}^{\pm}$ introduced in (\ref{eq:xi lambda pm}) leaves $\sfQ_\mfT$ invariant so that we can define $\sfQ_\mfT^\mfs$
\index{$\a$Aa@Notation!Symmetry reduced sets of regular initial data!sfQTs@$\sfQ_{\mfT}^\mfs$}%
in analogy with Definition~\ref{def:BTs LTs}. Moreover, the following analogy of Lemma~\ref{lemma:sfQ mfT mfs} holds. 
\begin{lemma}\label{lemma:sfU mfT mfs}
  Let $\mfT$ be a Bianchi class A type and $\mfs\in\{\iso,\roper,\roLRS,\rogen\}$ be such that $\sfQ_{\mfT}^{\mfs}$ and $\Lambda_{\mfT}^{\mfs}$ are
  defined in analogy with or as in Definition~\ref{def:BTs LTs}, respectively, and non-empty. Then $\Lambda_{\mfT}^{\mfs}$ acts freely and properly
  discontinuously on $\sfQ_{\mfT}^{\mfs}$ and the corresponding quotient is denoted $\sfW_{\mfT}^{\mfs}$.
  \index{$\a$Aa@Notation!Symmetry reduced sets of regular initial data!sfWTs@$\sfW_{\mfT}^\mfs$}%
\end{lemma}
Next, define
\begin{equation}\label{eq:sfFQ def}
  \begin{split}
    & \sfF_\sfQ(\bn_1,\bn_2,\bn_3,\sigma_1,\sigma_2,\sigma_3,\varkappa_\infty,\Phi_0,\Phi_1)\\
    := & (e^{2p_1\varkappa_\infty}\bn_1,e^{2p_2\varkappa_\infty}\bn_2,e^{2p_3\varkappa_\infty}\bn_3,p_1,p_2,p_3,\Phi_0+\Phi_1\varkappa_\infty/3,\Phi_1/3),
  \end{split}
\end{equation}
\index{$\a$Aa@Notation!Maps!sfFQ@$\sfF_\sfQ$}%
where $p_i=\sigma_i+1/3$. Then $\sfF_\sfQ$ maps $\sfQ_{\mfT}$ to $\sfS_{\mfT}$, see Definition~\ref{def:mfS def}. Moreover,
\begin{equation}\label{eq:sff sfQ comm}
  \psi_{\lambda}^{\pm}\circ\sfF_\sfQ=\sfF_\sfQ\circ\xi_\lambda^{\pm}.
\end{equation}
The reason for introducing $\sfF_\sfQ$ is that we, in the end, define $\bn_k$ as the limit of $\e_k e^{\nu_k}$, and then the $e^{2p_k\varkappa_\infty}\bn_k$
are the $o_k$ appearing in Definition~\ref{def:mfS def}. 

In some situations it is of interest to fix $\varkappa$ initially. This leads to the following terminology.
\begin{definition}\label{def:sfA fixed varkappa}
  Let $V\in C^{\infty}(\ro)$ and $\varkappa_0\in\ro$. Define $\sfA_{\mfT}^{\mfs}(\varkappa_0):=\sfA_{\mfT}^{\mfs}\cap \sfP_{\mfT}^{\mfs}(\varkappa_0)$.
  \index{$\a$Aa@Notation!Symmetry reduced sets of regular initial data!sfATsvak@$\sfA_{\mfT}^\mfs(\varkappa_0)$}%
  Similarly,
  for $\mft\in\{+,\roap,\ropp\}$, define $\sfA_{\mrIX,\mft}^{\mfs,\rond}:=\sfA_{\mrIX}^{\mfs}\cap \sfP_{\mrIX,\mft}^{\mfs,\rond}$
  \index{$\a$Aa@Notation!Symmetry reduced sets of regular initial data!sfAIXtsnd@$\sfA_{\mrIX,\mft}^{\mfs,\rond}$}%
  and
  $\sfA_{\mrIX,\mft}^{\mfs,\rond}(\varkappa_0):=\sfA_{\mrIX}^{\mfs}\cap \sfP_{\mrIX,\mft}^{\mfs,\rond}(\varkappa_0)$.
  \index{$\a$Aa@Notation!Symmetry reduced sets of regular initial data!sfAIXtsndvka@$\sfA_{\mrIX,\mft}^{\mfs,\rond}(\varkappa_0)$}%
\end{definition}
Next, we address the regularity of the map from regular initial data to asymptotic data. 
\begin{lemma}\label{lemma:reg of B}
  Let $0\leq V\in \mfP_{\a_V}^\infty$ for some $\a_{V}\in (0,1)$. Fix a Bianchi class A type $\mfT$ and $\mfs\in\{\iso,\roLRS,\roper,\rogen\}$. Then
  $\sfA_\mfT^\mfs$ is an open subset of $\sfP_\mfT^\mfs$. Fix
  \[
  \xi=(n_1,n_2,n_3,\sigma_1,\sigma_2,\sigma_3,\varkappa_0,\bpsi_0,\bpsi_1)\in\sfA_\mfT^\mfs.
  \]
  Let $\e_i=1$ if $n_i>0$, $\e_i=-1$ if $n_i<0$ and $\e_i=0$ if $n_i=0$. If $n_i\neq 0$, let $\nu_i(0)=\ln |n_i|$. If $n_i=0$, let $\nu_i(0)=0$.
  Finally, let $\Sigma_i(0)=\sigma_i$; $\varkappa(0)=\varkappa_0$; $\psi_0(0)=\bpsi_0$; and $\phi_1(0)=\bpsi_1$. Let $x$ be the corresponding solution
  to (\ref{seq:nui etc eqs}) and
  \begin{equation}\label{eq:x limit dos}
    (\nu_{1,\infty},\nu_{1,\infty},\nu_{1,\infty},\sigma_{1,\infty},\sigma_{2,\infty},\sigma_{3,\infty},\varkappa_\infty,\Phi_{0},\Phi_{1})
    =\lim_{\tau\rightarrow-\infty}x(\tau).
  \end{equation}
  Define
  \begin{equation}\label{eq:Psi mfT mfs def}
    \Psi_\mfT^\mfs(\xi):=(\e_1e^{\nu_{1,\infty}},\e_2e^{\nu_{2,\infty}},\e_3e^{\nu_{3,\infty}},\sigma_{1,\infty},\sigma_{2,\infty},\sigma_{3,\infty},
    \varkappa_\infty,\Phi_{0},\Phi_{1}). 
  \end{equation}
  Then $\Psi_\mfT^\mfs$ is defined on $\sfA_\mfT^\mfs$ and takes its values in $\sfQ_\mfT^{\mfs}$. Moreover, $\Psi_\mfT^\mfs$ is smooth and at each point
  $\xi\in \sfA_{\mfT}^\mfs$, $\Psi_{\mfT,*}^\mfs|_\xi$ is surjective. In particular, $\Psi_{\mfT}^\mfs$ is a local diffeomorphism.

  Next, let $\varkappa_0\in\ro$,
  $\xi_0\in\sfA_\mfT^\mfs(\varkappa_0)$ and assume that the $q$ corresponding to $\xi_0$ is $>-1$. Then there is an open neighbourhood of $\xi_0$ in
  $\sfA_\mfT^\mfs(\varkappa_0)$, say $U$, so that $\Psi_\mfT^\mfs$, restricted to $U$ and composed $\sfF_\sfQ$, see (\ref{eq:sfFQ def}),
  is a diffeomorphism onto its image.
\end{lemma}
\begin{remark}\label{remark:Psi T s Cm}
  If one is only interested in proving that $\Psi_\mfT^\mfs$ is $C^m$, it is sufficient to demand that $V\in\mfP_{\a_V}^{m+1}$. 
\end{remark}
\begin{proof}
  Fix an admissible convergent solution, say $S:=(\No,\Nt,\Nth,\Sigma_1,\Sigma_2,\Sigma_3,\theta,\phi_{0},\phi_{1})$. In the arguments to follow,
  it is convenient to replace the $\Sigma_i$ by $\Sigma_\pm$, see (\ref{eq:SigmaiitoSpm}), and we do so from now on. Moreover, in what follows,
  if condition (ii) of Definition~\ref{def:admconvsol} is satisfied, we assume that $\Sm=0$ and $N_2=N_3$. Consider (\ref{eq:constraint}).
  If condition (i) of Definition~\ref{def:admconvsol} is satisfied, then the assumptions imply that the polynomial in the $N_i$ on the left hand side
  of (\ref{eq:constraint}) can be estimated by $Ce^{2\a_{\infty}\tau}$ for all $\tau\leq T$, where $C$ only depends on $C_\infty$. This means that there is
  a constant $C$ such that $|\phi_1(\tau)|\leq \sqrt{6}+Ce^{2\a_{\infty}\tau}$ for all $\tau\leq T$, where $C$ only depends on $C_{\infty}$. Combining this
  observation with the fact that $\ln\theta(\tau)\geq -3\tau-C_{\theta}$, it is clear that there are constants $C_{\Phi}$ and $C$ such that
  \begin{equation}\label{eq:phizOmegaestbg}
    |\phi_{0}(\tau)|\leq \sqrt{6}|\tau|+C_{\Phi},\ \ \
    \Omega(\tau)\leq Ce^{6(1-\a_V)\tau}
  \end{equation}
  for all $\tau\leq T$, where $C$ only depends on $c_{0}$, $C_{\theta}$ and $C_{\Phi}$, and $c_0$ is the constant appearing in (\ref{eq:V k-derivatives exp estimate}).
  If condition (ii) of Definition~\ref{def:admconvsol}
  is satisfied, then $\Phi_1=0$. Thus (\ref{eq:phizOmegaestbg}) holds for all $\tau\leq T$. Next, note that if condition
  (i) of Definition~\ref{def:admconvsol} is satisfied, then combining (\ref{eq:phizOmegaestbg}) with the fact that the polynomial in the $N_i$ on the
  left hand side of (\ref{eq:constraint}) decays exponentially yields the conclusion that (\ref{eq:as Ham con pi Phione}) holds (or, equivalently,
  that $\sigma_+^2+\sigma_-^2+\Phi_1^2/6=1$, where $\sigma_\pm$ are the limits of $\Sigma_\pm$). If condition (ii) of
  Definition~\ref{def:admconvsol} is satisfied, then this equality holds automatically. Combining this observation with (\ref{eq:qdef}) yields the
  conclusion that $q\rightarrow 2$ as $\tau\rightarrow-\infty$, so that for any $\beta<6$, there is a $C_\beta$ such that $|\No(\tau)|\leq C_\beta e^{\b\tau}$
  and $|(\No\Nt)(\tau)|\leq C_\beta e^{\b\tau}$ for all $\tau\leq T$ (assuming (ii) holds). To summarise, regardless of whether condition (i) or (ii) holds,
  $q-2$, $\Sp-\sigma_{+}$ and $\Sm-\sigma_{-}$ converge to zero exponentially. In particular, if condition (ii) holds, there is a constant
  $C_\infty$ such that $|\No(\tau)|\leq C_\infty e^{6\tau}$ and $|(\No\Nt)(\tau)|\leq C_\infty e^{6\tau}$ for all $\tau\leq T$. 

  Fix an $\e>0$ and a $\tau_0\leq T$. Next, write (\ref{seq:nui etc eqs}) schematically as $x'(\tau)=\mX[x(\tau),\tau]$. Let $x$ denote the solution to
  (\ref{seq:nui etc eqs}) corresponding to the solution $S$ and let $\bx$ denote another solution to (\ref{seq:nui etc eqs}) arising from initial
  data in $\sfP$ (so that the constraint (\ref{eq:constraint}) is satisfied due to (\ref{eq:fdiffeq})). Assume the $\e_i$'s
  which are implicit in the definition of the right hand side of (\ref{seq:nui etc eqs}) to be the same for $x$ and $\bx$. Assume, moreover, $\bx$ to
  correspond to an LRS solution with $\bNt=\bNth$ and $\bSm=0$ in case $S$ satisfies condition (ii) of Definition~\ref{def:admconvsol}. We also use the
  notation
  \[
  x=(\nu_1,\nu_2,\nu_3,\Sp,\Sm,\varkappa,\psi_{0},\phi_{1}),\ \ \ \bx=(\bnu_1,\bnu_2,\bnu_3,\bSp,\bSm,\bvka,\bpsi_{0},\bphi_{1}).
  \]
  Assume that $|\bx(\tau_0)-x(\tau_0)|<\e$ and let
  \[
  \ma:=\{\tau\leq\tau_0 : |\bx(s)-x(s)|<2\e\ \forall\ s\in [\tau,\tau_0]\}.
  \]
  Clearly, $\ma$ is an open, connected and non-empty subset of $(-\infty,\tau_0]$. Thus $\ma=(\tau_-,\tau_0]$ for some $\tau_-\in [-\infty,\tau_0)$.
  Assume $\tau_->-\infty$ and estimate, for $\tau\in\ma$,
  \[
  |\bN_i|=|N_i|e^{\bnu_i-\nu_i}e^{f_i(0,\bSp-\Sp,\bSm-\Sm)\tau}\leq e^{-8\e\tau+2\e}|N_i|.
  \]
  If condition (i) of Definition~\ref{def:admconvsol} is satisfied and $N_i\neq 0$, it follows that $|\bN_i(\tau)|\leq C_\infty e^{(\a_\infty-8\e)\tau+2\e}$
  for all $\tau\in\ma$. Moreover, if condition (ii) of Definition~\ref{def:admconvsol} is satisfied, then
  \[
  |\bNo(\tau)|\leq C_\infty e^{(6-8\e)\tau+2\e},\ \ \
  |(\bNo\bNt)(\tau)|\leq C_\infty e^{(6-16\e)\tau+4\e}
  \]
  for all $\tau\in\ma$. Assuming $\e>0$ to be small enough, we can thus assume that there is an $\eta>0$, which is numerical if condition (ii) holds
  and which depends only on $\a_\infty$ in case condition (i) holds, and a constant $C$, depending only on $C_\infty$ such that
  \[
  |\bS_\pm(\tau)|+|2-\bq(\tau)-\bOmega(\tau)|\leq Ce^{2\eta\tau}
  \]
  for all $\tau\in\ma$; see (\ref{eq:altqdef}). Next, note that
  \[
  \ln\bth(\tau)\geq\ln\theta(\tau)-2\e\geq -3\tau-C_\theta-2\e
  \]
  for all $\tau\in\ma$. Thus
  \[
  \tfrac{V\circ\bphi_0(t)}{\bth^2(\tau)}+\tfrac{|V'\circ\bphi_0(t)|}{\bth^2(\tau)}\leq Ce^{6(1-\a_V)\tau}e^{\sqrt{6}\a_V|\bphi_0(\tau)-\phi_0(\tau)|}
  \leq Ce^{[6(1-\a_V)-2\sqrt{6}\a_V\e]\tau}
  \]
  for all $\tau\in\ma$, where $C$ only depends on $c_1$, $C_\theta$ and $C_{\Phi}$. Assuming $\e>0$ to be small enough (depending only on $\a_V$ and,
  in case condition (i) holds, $\eta_\infty$), we conclude that there is an $\eta>0$ with the same dependence such that
  \begin{equation}\label{eq:bSpm etc exp dec}
    |\bS_\pm(\tau)|+|2-\bq(\tau)|+\bOmega(\tau)+\tfrac{|V'\circ\bphi_0(t)|}{\bth^2(\tau)}\leq Ce^{2\eta\tau}
  \end{equation}
  for all $\tau\in\ma$, where $C$ only depends on $c_1$, $C_\theta$, $C_\infty$ and $C_{\Phi}$. Next, note that due to these estimates,
  (\ref{eq:constraint}) and (\ref{eq:altqdef}), it follows that there is a constant $C$, depending only on $c_1$, $C_\theta$, $C_\infty$ and $C_{\Phi}$,
  such that
  \begin{equation}\label{eq:bSp bSm bphi one bd}
    |\bSp(\tau)|+|\bSm(\tau)|+|\bphi_1(\tau)|\leq C
  \end{equation}
  for all $\tau\in\ma$. Combining (\ref{seq:nui etc eqs}), (\ref{eq:bSpm etc exp dec}) and (\ref{eq:bSp bSm bphi one bd}), it is clear that
  $|\bx'(\tau)|\leq C\ldr{\tau}e^{2\eta\tau}$ for all $\tau\in\ma$, so that $|\bx(\tau)-\bx(\tau_0)|\leq C\ldr{\tau_0}e^{2\eta\tau_0}$ for all
  $\tau\in\ma$. Note that this estimate holds for $x$ as well. This means that
  \begin{equation}\label{eq:bxtau minus xtau}
    |\bx(\tau)-x(\tau)|\leq \e+C_X\ldr{\tau_0}e^{2\eta\tau_0}
  \end{equation}
  for all $\tau\in\ma$. Here $C_X$ only depends on $c_1$, $C_\theta$, $C_\infty$, $C_{\Phi}$, $\a_V$ and, in case condition (i) holds, $\eta_\infty$.
  Moreover, this estimate is based on the assumption that $\e$ is small enough, the bound depending only on $\a_V$ and, in case condition (i) holds,
  $\eta_\infty$. Assuming $\tau_0$ to be close enough to $-\infty$, the bound depending only on $\e$, $c_1$, $C_\theta$, $C_\infty$, $C_{\Phi}$, $\a_V$ and, in
  case condition (i) holds, $\eta_\infty$, it is clear that it can be ensured that the right hand side of (\ref{eq:bxtau minus xtau}) is
  $\leq 3\e/2$. To summarise the above argument, we first fix a small enough $\e>0$, depending only on $\a_V$ and, in case condition (i) holds,
  $\eta_\infty$. Given this $\e>0$, we fix a $\tau_0$ close enough to $-\infty$, the bound depending only on $\e$, $c_1$, $C_\theta$, $C_\infty$, $C_{\Phi}$,
  $\a_V$ and, in case condition (i) holds, $\eta_\infty$. Assume that $|\bx(\tau_0)-x(\tau_0)|<\e$ and define $\ma$ as above. Then $\ma$ is of the form
  $\ma=(\tau_-,\tau_0]$ for some $\tau_-<\tau_0$. Moreover, the above argument yields the conclusion that if $\tau_->-\infty$, then
  $|\bx(\tau)-x(\tau)|\leq 3\e/2$ on $\ma$, so that there is a $\tau_1<\tau_-$ such that $\tau_1\in\ma$, in contradiction of the definition of $\tau_-$.
  Thus $\tau_-=-\infty$.
  
  By the above argument, we conclude that if condition (i) is satisfied, then there is an open neighbourhood of $x(\tau_0)$ (in the constraint
  hypersurface represented by (\ref{eq:constraint asympt var}), with $\tau$ replaced by $\tau_0$, and the same $\e_i$'s as in the case of $x$)
  such that the solutions to (\ref{seq:nui etc eqs}) arising from initial data in this neighbourhood give rise to admissible convergent solutions
  satisfying condition (i). Moreover, if condition (ii) is satisfied, then there is an open neighbourhood of $x(\tau_0)$ in the
  set of LRS Bianchi type VIII or IX initial data such that the corresponding solutions satisfy either condition (i) or (ii). In order to translate
  this observation to a statement concerning $\sfA_\mfT^\mfs$, note that, given $p\in\sfP_\mfT^\mfs$, there is an open neighbourhood of $p$ in
  $\sfP_\mfT^\mfs$, say $U$, such that the $\e_i$'s of the points in $U$ coincide with the $\e_i$'s of $p$. This means that the map from $U$ to initial
  data for (\ref{seq:nui etc eqs}) is a smooth map; cf. the statement of the lemma. This proves that $\sfA_\mfT^\mfs$ is an open subset of $\sfP_\mfT^\mfs$.
  Next, let us prove that the image of $\sfA_\mfT^\mfs$ under $\Psi_\mfT^\mfs$ is a subset of $\sfQ_\mfT$. That the Bianchi type is preserved is an
  immediate consequence of the definition of $\Psi_\mfT^\mfs$. That the limit of the sum of the $\Sigma_i$ is zero is an immediate consequence of the
  fact that the sum of the $\Sigma_i$ is zero. We have already verified that (\ref{eq:as Ham con pi Phione}) holds. If
  $\{i,j,k\}=\{1,2,3\}$ and $p_i=1$, then we can assume that $(i,j,k)=(1,2,3)$, so that $\sigma_1=2/3$ and $\sigma_2=\sigma_3=-1/3$. This means that
  $\sigma_+=-1$ and $\sigma_-=0$. Proposition~\ref{prop:limitcharsp} then implies that $\Nt=\Nth$ and $\Sm=0$. This means that $\bn_2=\bn_3$, as
  required by Definition~\ref{def:ma mfT mfs}. Next, assume, in order to obtain a contradiction, that all the $p_l<1$, that $\bn_1\neq 0$ and
  that $p_1\leq 0$. This means that $\sigma_+\geq 1/2$ so that $q-4\Sp$ converges to a limit which is $\leq 0$. However, this is not consistent with
  $\No$ converging to zero exponentially at a fixed rate. 
  Finally, that the image of $\Psi_\mfT^\mfs$ is at least as symmetric as the initial data is
  an immediate consequence of the definition. However, the image could potentially be more symmetric. That this is not possible follows from the
  fact that symmetries of the asymptotic data imply the corresponding symmetries of the solution; see Proposition~\ref{prop:variableexandunique}.
    
  Let $\xi_0\in\sfA_\mfT^\mfs$. By the above, there is an open neighbourhood of $\xi_0$, say $U\subset\sfP_\mfT^\mfs$, such that $U\subset\sfA_\mfT^\mfs$.
  Given $\xi\in U$, we denote by $x(\cdot;\xi)$ the solution $x$ to  (\ref{seq:nui etc eqs}) corresponding to initial data $x(0)$ obtained from
  $\xi$ as described in the statement of the lemma. When we write $\Sigma_k$ etc. in what follows, we view $\Sigma_k$ etc. as being defined on
  $(-\infty,0]\times U$. In the arguments to follow, we also prefer to phrase the equations in terms of $\Sigma_\pm$ as opposed to $\Sigma_i$;
  cf. Remark~\ref{remark:Sigmapm vs Sigmai}.  If condition (i) of Definition~\ref{def:admconvsol} is satisfied, the assumptions and the above
  arguments imply that, assuming $U$ to be small enough and $\tau_0$ to be close enough to $-\infty$, there are constants $\eta>0$, $C$ and $C_k$
  such that if $N_i\neq 0$ and $k\in\nn{}$, then 
  \begin{subequations}\label{seq:fi Ni theta asymptotics}
    \begin{align}
      e^{\mu_i} \leq & Ce^{\eta\tau},\label{eq:Ni est uniform}\\
      \Sp^2+\Sm^2+\phi_1^2 \leq & C,\label{eq:Sp Sm phi one est uniform}\\
      e^{-2\varkappa+6\tau}|V^{(k)}\circ\phi_0| \leq & C_ke^{\eta\tau}\label{eq:Vk uniform est}
    \end{align}
  \end{subequations}
  on $(-\infty,\tau_0]\times U$. In case condition (ii) of Definition~\ref{def:admconvsol} is satisfied, $\mfT\in\{\mrVIII,\mrIX\}$ and
  $\mfs=\roLRS$. This means that we can assume $\Sm=0$ and $\Nt=\Nth$ on $(-\infty,\tau_0]\times U$. Moreover, we can assume 
  (\ref{eq:Sp Sm phi one est uniform}) and (\ref{eq:Vk uniform est}) to still hold. However, (\ref{eq:Ni est uniform}) is replaced by 
  \begin{equation}\label{eq:no nont exp dec}
    |N_1|+|N_1N_2|\leq Ce^{\eta\tau}
  \end{equation}
  on $(-\infty,\tau_0]\times U$. In order to analyse the regularity in the initial data, let $\{E_i\}$ denote a frame of vector
  fields on $U$ and let $E_{\bfI}=E_{i_{1}}\cdots E_{i_k}$ if $\bfI=(i_{1},\dots,i_k)$, where $i_j\in\{1,2,3\}$ for $j\in\{1,\dots,k\}$.
    
  Due to (\ref{eq:altqdef}), (\ref{seq:fi Ni theta asymptotics}) and (\ref{eq:no nont exp dec}), the right hand sides of (\ref{seq:nui etc eqs})
  have uniform exponential decay, irrespective of whether condition (i) or condition (ii) holds. More
  specifically, the right hand sides are, in absolute value, bounded by $C\ldr{\tau}e^{\eta\tau}$ for $\tau\leq\tau_0$.
  Applying $E_i$ to (\ref{seq:nui etc eqs}) yields a linear system of equations for $E_i\nu_j$, $E_i\Sigma_\pm$ etc. However, each
  expression $E_i\nu_j$ etc. on the right hand side is multiplied by a uniformly exponentially decaying factor. Defining an appropriate energy
  (summing the squares of $E_i\nu_j$, $E_i\Sigma_\pm$ etc.) then yields the conclusion that the energy is bounded to the past. This argument
  can then be iterated to conclude that all the derivatives $E_\bfI\nu_i$ etc. are uniformly bounded to the past. Combining these
  observations with (\ref{seq:nui etc eqs}) then yields the conclusion that all the $E_\bfI\nu_i$ etc. converge uniformly and
  exponentially. This means that the limits of $\nu_i$ etc. are all smooth functions on $U$. In particular, the map $\Psi_\mfT^\mfs$ is smooth.
  
  Next, we wish to prove that the differential of $\Psi_{\mfT}^\mfs$ is always surjective. Recall, to this end, the constraint
  (\ref{eq:constraint asympt var}). In what follows, we consider the evolution from initial data at $\tau_0$ to asymptotic data. We therefore
  focus on $V_0:=\{x(\tau_0;\eta):\eta\in U\}$. Note that we can assume $\tau_0$ to be close enough to $-\infty$ that
  \begin{equation}\label{eq:lower bd essential variables}
    \Sp^2+\Sm^2+\tfrac{1}{6}\phi_1^2\geq \tfrac{1}{2}
  \end{equation}
  in $V_0$. Fix a $\zeta\in V_0$ and let $\d_{\nu_i}$, $i=1,2,3$, be vector fields on $V_0$ corresponding to differentiation in the $\nu_i$ directions (so
  that if we view $\Sp$ etc. as coordinates on $\rn{8}$, then $\d_{\nu_i}\Sp=0$ etc.). In general, the vector fields $\d_{\nu_i}$ are not tangential to
  the hypersurface defined by (\ref{eq:constraint asympt var}). In order to obtain a vector field which is tangential, let
  \begin{equation}\label{eq:Xr def}
    X_r:=\Sp\d_{\Sp}+\Sm\d_{\Sm}+\phi_1\d_{\phi_1}
  \end{equation}
  in case $\mfs=\rogen$. If $\mfs\in\{\roLRS,\roper\}$, we can assume $\Sm=0$ and then $X_r$ should be replaced by $\Sp\d_{\Sp}+\phi_1\d_{\phi_1}$.
  If $\mfs=\iso$, $\Sp=\Sm=0$ and $X_r$ should be replaced by $\phi_1\d_{\phi_1}$. In what follows, we assume $X_r$ to take the form (\ref{eq:Xr def})
  and leave it to the reader to carry out the corresponding modifications in case $\mfs\neq\rogen$. Note that
  \[
  X_r\left(\Sp^2+\Sm^2+\tfrac{1}{6}\phi_1^2\right)=2\left(\Sp^2+\Sm^2+\tfrac{1}{6}\phi_1^2\right)\geq 1,
  \]
  where we appealed to (\ref{eq:lower bd essential variables}). For some function $a_i$, apply $Y_i=a_i X_r+\d_{\nu_i}$ to the right hand
  side of (\ref{eq:constraint asympt var}) with $\tau$ replaced by $\tau_0$, say $F$. This yields
  \[
    Y_i(F) = 2a_i\left(\Sp^2+\Sm^2+\tfrac{1}{6}\phi_1^2\right)+a_i R_1+R_2,
  \]
  where $\{i,j,k\}=\{1,2,3\}$ and
  \begin{align*}
    R_1 := & 3\phi_1\tau_0 e^{6\tau_0-2\varkappa}V'(\psi_0+\phi_1\tau_0)+9\Sigma_1\tau_0\e_1^2 e^{2\mu_1}+9\Sigma_2\tau_0\e_2^2 e^{2\mu_2}
    +9\Sigma_3\tau_0\e_3^2 e^{2\mu_3}\\
    & -9(\Sigma_1+\Sigma_2)\tau_0 \e_1\e_2 e^{\mu_1+\mu_2}
    -9(\Sigma_2+\Sigma_3)\tau_0 \e_2\e_3 e^{\mu_2+\mu_3}-9(\Sigma_1+\Sigma_3)\tau_0 \e_1\e_3 e^{\mu_1+\mu_3},\\
    R_2 := & \tfrac{3}{2}(\e_i^2 e^{2\mu_i}-\e_i\e_j e^{\mu_i+\mu_j}-\e_i\e_k e^{\mu_i+\mu_k}).
  \end{align*}
  The goal is to choose $a_i$ so that $Y_i(F)=0$. However, assuming $\tau_0$ to be close enough to $-\infty$, we can assume that
  $|R_1|\leq 1/2$. This means that we can define $a_i$ by
  \[
  a_i:=-\frac{R_2}{2\left(\Sp^2+\Sm^2+\tfrac{1}{6}\phi_1^2\right)+R_1}.
  \]
  Then $Y_i$ is tangential to the constraint hypersurface at $\tau=\tau_0$, but typically not the constraint hypersurfaces at other times;
  note that $Y_i$ is independent of $\tau$. Moreover, $|a_i|\leq Ce^{\eta\tau_0}$, where we appealed to (\ref{eq:Ni est uniform}) and
  (\ref{eq:no nont exp dec}). From now on, we consider $Y_i$ to be a vector field on $V_0$ and we define
  $x:(-\infty,\tau_0]\times V_0\rightarrow\rn{8}$ by the condition that $x(\cdot;\zeta)$ is the solution to
  (\ref{seq:nui etc eqs}) with $x(\tau_0;\zeta)=\zeta$. Next, we consider the evolution equations for $Y_i x$. Applying $Y_i$ to
  (\ref{seq:nui etc eqs}), it can be estimated that
  \begin{equation}\label{eq:Yi x ev eqn}
    |(Y_i x)'|\leq C\ldr{\tau}^2e^{\eta\tau}|Y_i x|
  \end{equation}
  for $\tau\leq\tau_0$. Since $|(Y_i x)(\tau_0;\zeta)|$ can be assumed to be arbitrarily close to $1$ by assuming $\tau_0$ to be close enough to
  $-\infty$, (\ref{eq:Yi x ev eqn}) implies that $|(Y_i x)(\tau;\zeta)|\leq 2$ for $\tau\leq\tau_0$, assuming $\tau_0$ to be close enough to
  $-\infty$. Inserting this information into (\ref{eq:Yi x ev eqn}) implies that
  \[
  \big|\lim_{\tau\rightarrow-\infty}(Y_ix)(\tau;\zeta)-(Y_ix)(\tau_0;\zeta)\big|\leq C\ldr{\tau_0}^2e^{\eta\tau_0}.
  \]
  However, up to an error of the order of magnitude $O(e^{\eta\tau_0})$, $(Y_ix)(\tau_0;\zeta)$ is the unit vector in the direction of $\nu_i$. This
  means that if $\e_i\neq 0$, then $\d_{\bn_i}+O(\ldr{\tau_0}^2e^{\eta\tau_0})$ is in the image of $(\Psi_\mfT^\mfs)_*$ (where we use the standard
  Euclidean metric on $\rn{8}$ to measure the size of vector fields); here and below, we denote the coordinates on the target $\rn{8}$ by $\bn_i$,
  $\sigma_\pm$, $\varkappa_\infty$, $\Phi_0$ and $\Phi_1$, where $\bn_i=0$ in case $\e_i=0$. By a similar argument,
  $\d_{\varkappa_\infty}+O(\ldr{\tau_0}^2e^{\eta\tau_0})$ and $\d_{\Phi_0}+O(\ldr{\tau_0}^2e^{\eta\tau_0})$ are in the image of $(\Psi_\mfT^\mfs)_*$.
  Note that this proves surjectivity in the case that $\mfs=\iso$. 

  Next, assume $\mfs=\rogen$ and consider the vector fields
  \[
    W:=\Sm\d_{\Sp}-\Sp\d_{\Sm},\ \ \
    P_{\pm}:=\tfrac{1}{6}\phi_1\d_{\Sigma_{\pm}}-\Sigma_{\pm}\d_{\phi_1}.
  \]
  Unfortunately, $W(F)$ and $P_{\pm}(F)$ are in general non-zero. However, this can be corrected by adding a small multiple of $X_r$, as before.
  Arguments similar to the above then yield the conclusion that $w+O(\ldr{\tau_0}^2e^{\eta\tau_0})$ and $p_\pm+O(\ldr{\tau_0}^2e^{\eta\tau_0})$ are in
  the image of $(\Psi_\mfT^\mfs)_*$, where 
  \[
  w:=\sigma_-\d_{\sigma_+}-\sigma_+\d_{\sigma_-},\ \ \
  p_\pm:=\tfrac{1}{6}\Phi_1\d_{\sigma_{\pm}}-\sigma_{\pm}\d_{\Phi_1}.
  \]
  Note also that if $(\sigma_+,\sigma_-)=0$, then $|\Phi_1|=\sqrt{6}$ and $p_\pm$ span a two dimensional subset of the corresponding tangent space
  of $\sfQ_\mfT^\mfs$ (the same is true if $(\sigma_+,\sigma_-)$ is close to zero). If $(\sigma_+,\sigma_-)\neq 0$, then $w$ and
  \[
  \sigma_+p_++\sigma_- p_-=\tfrac{1}{6}\Phi_1(\sigma_+\d_{\sigma_{+}}+\sigma_-\d_{\sigma_-})-(\sigma_{+}^2+\sigma_-^2)\d_{\Phi_1}
  \]
  are orthogonal and both non-zero. For the above reasons, the map $(\Psi_\mfT^\mfs)_*$ is surjective at each point in $U$ of $\mfs=\rogen$. In case
  $\mfs=\iso$, we already know $(\Psi_\mfT^\mfs)_*$ to be surjective. In case $\mfs\in \{\roper,\roLRS\}$, a similar argument yields the same conclusion;  
  it is sufficient to focus $P_+$. In particular, $\Psi_\mfT^\mfs$ is a local diffeomorphism.

  Next, let $\varkappa_0\in\ro$, $\xi_0\in\sfA_\mfT^\mfs(\varkappa_0)$ and assume that $q(0;\xi_0)>-1$. Due to the above arguments, there is an open neighbourhood,
  say $U$, of $\xi_0$
  in $\sfA_\mfT^\mfs$ such that $q(0;\xi)>-1$ for $\xi\in U$ and $q$ converges uniformly to $2$ in $U$. This means that there is a $T_0\leq 0$ such that
  $q(\tau;\xi)\geq 1$ for all $\tau\leq T_0$ and $\xi\in U$. In particular,  $\theta'(\tau;\xi)\leq -2\theta(\tau;\xi)<0$ for all $\tau\leq T_0$
  and $\xi\in U$. Assume that $\tau_0\leq T_0$ and that $\theta(\tau_0;\xi_0)=\theta_0$. Then, due to the implicit function theorem, there is a
  neighbourhood $W$ in $\sfA_{\mfT}^{\mfs}(\varkappa_0)$ of $\xi_0$ and a smooth map $T:W\rightarrow\ro$ such that $\theta(T(\xi);\xi)=\theta_0$. Next,
  note that we can think of $\xi$ as initial data for (\ref{eq:thetaprime})--(\ref{eq:constraint}); $N_i(0):=n_i$,
  $\Sigma_i(0):=\sigma_i$, $\phi_0(0):=\bpsi_0$, $\phi_1(0):=\bpsi_1$, $\theta(0):=e^{\varkappa_0}$. We denote the corresponding solution by
  $X(\tau;\xi)$. Since the flow associated with (\ref{eq:thetaprime})--(\ref{eq:constraint}) is smooth, we obtain a smooth map $\Xi_{\roloc}$ defined by
  $\Xi_{\roloc}(\xi):=X[T(\xi);\xi]$ for $\xi\in W$. Taking the logaritm of the mean curvature, we can interpret $\Xi_{\roloc}(\xi)$ as an element of
  $\sfA_\mfT^\mfs(\ln\theta_0)$. By reversing the flow and restricting the domain, if necessary, it is clear that $\Xi_{\roloc}$ has a smooth inverse (recall
  that $\theta'(0;\xi_0)<0$), so that it is a diffeomorphism onto its image. Next, note that we can assume estimates such as
  (\ref{seq:fi Ni theta asymptotics}) and (\ref{eq:no nont exp dec}) to hold uniformly on $W$ to the past of $T_0$. Assume, from now on, that
  $\Xi_{\roloc}$ is a diffeomorphism from $W\subset \sfA_\mfT^\mfs(\varkappa_0)$ to $W_2\subset \sfA_\mfT^\mfs(\ln\theta_0)$. Let $\zeta\in W_2$ and
  $\xi:=\Xi_{\roloc}^{-1}(\zeta)$. We would like to start from $\zeta$ and evolve backwards in time. However, this leads to a problem: If we wish to
  focus on solutions to (\ref{seq:nui etc eqs}) corresponding to the initial data given in the statement of the lemma and then to calculate the
  limit (\ref{eq:x limit dos}), we should strictly speaking use $T(\xi)$ as
  the starting time. This causes complications when carrying out arguments as the one used above to prove that $(\Psi_\mfT^\mfs)_*$ is surjective.
  On the other hand, the map we are interested in is such that a change in initial time is immaterial. In order to justify this statement,
  consider two solutions to (\ref{eq:thetaprime})--(\ref{eq:constraint}) denoted by $X$ and $\bar{X}$ and with the property that
  $\bar{X}(\tau)=X(\tau+\tau_1)$; i.e., $\bar{N}_i(\tau)=N_i(\tau+\tau_1)$ etc. For $i$ such that $\e_i\neq 0$, it can then be verified that
  \[
  \bar{\nu}_i(\tau)=\nu_i(\tau+\tau_1)+6[\Sigma_i(\tau+\tau_1)+1/3]\tau_1.
  \]
  Similarly,
  \[
  \bar{\varkappa}(\tau)=\varkappa(\tau+\tau_1)-3\tau_1,\ \ \
  \bar{\psi}_0(\tau)=\psi_0(\tau+\tau_1)+\phi_1(\tau+\tau_1)\tau_1.
  \]
  Finally, $\bar{\Sigma}_i(\tau)=\Sigma_i(\tau+\tau_1)$ and $\bar{\phi}_1(\tau)=\phi_1(\tau+\tau_1)$. Using notation analogous to
  (\ref{eq:x limit dos}), it is then clear that $\bar{\sigma}_{i,\infty}=\sigma_{i,\infty}$, that $\bar{\Phi}_1=\Phi_1$ and that if $k$ is such
  that $N_k\neq 0$, then 
  \[
  \bar{\nu}_{k,\infty}=\nu_{k,\infty}+6(\sigma_{k,\infty}+1/3)\tau_1,\ \ \
  \bar{\varkappa}_\infty=\varkappa_\infty-3\tau_1,\ \ \
  \bar{\Phi}_0=\Phi_0+\Phi_1\tau_1.
  \]
  Given these observations, it can be verified that $\sfF_\sfQ\circ\Psi_\mfT^\mfs$ is independent of initial translations in $\tau$. We can thus
  assume the initial data $\zeta\in W_2$ to be specified at $\tau=\tau_0$ instead of at $\tau=T(\xi)$. In what follows, we ignore $\nu_i$ in case
  $\e_i=0$. Assume that
  \[
  \zeta=(n_1,n_2,n_3,\sigma_1,\sigma_2,\sigma_3,\varkappa_0,\bpsi_0,\bpsi_1).
  \]
  If $n_i\neq 0$, define $\e_i$ and $\mu_i$ by $n_i=\e_i e^{\hat{\mu}_i}$ and define
  \begin{align*}
    \nu_i(\tau_0) := & \hat{\mu}_i-6(\sigma_i+1/3)\tau_0,\ \ \Sigma_i(\tau_0) := \sigma_i,\ \ \varkappa(\tau_0) := \varkappa_0+3\tau_0,\\
    \psi_0(\tau_0) := & \bpsi_0-\bpsi_1\tau_0,\ \ \phi_1(\tau_0) :=  \bpsi_1.
  \end{align*}
  Let $x$ denote the solution to (\ref{seq:nui etc eqs}) corresponding to these initial data at $\tau=\tau_0$. Since the distance between
  $T(\xi)$ and $\tau_0$ can be assumed to be less than $1$, we can assume estimates such as (\ref{seq:fi Ni theta asymptotics}) and
  (\ref{eq:no nont exp dec}) to hold uniformly on $W_2$ to the past of $\tau_0$. This means that we can carry out arguments similar to
  those used to prove that $(\Psi_\mfT^\mfs)_*$ is surjective in order to prove that the pushforward of $\sfF_\sfQ\circ\Psi_\mfT^\mfs$, when
  restricted to $\sfA_{\mfT}^{\mfs}(\varkappa_0)$, is surjective. This map thus defines a local diffeomorphism from a neighbourhood $W_1$ of
  $\xi\in \sfA_{\mfT}^{\mfs}(\varkappa_0)$ to $\sfS_\mfT^\mfs$.  
\end{proof}
One consequence of the arguments given in the proof is that admissible convergent solutions to (\ref{eq:thetaprime})--(\ref{eq:constraint})
correspond to developments that induce data on the singularity.
\begin{lemma}\label{lemma:from adm to dos}
  Assume that $0\leq V\in \mfP_{\a_V}^1$ for some $\a_V\in (0,1)$. Consider a Bianchi class A non-linear scalar field development $\md[V](\mfI)$ as
  obtained in Proposition~\ref{prop:unique max dev}. Assume that there is a $t_0$ in the associated existence interval such that $\theta(t)>0$ for
  $t\leq t_0$ and such that $\theta$ is not integrable to the past. Assume, moreover, that introducing variables as in Section~\ref{ssection:whsuform},
  the corresponding solution to (\ref{eq:thetaprime})--(\ref{eq:constraint}) is an admissible convergent solution. Then
  $\md[V](\mfI)$ induces data on the singularity. 
\end{lemma}
\begin{proof}
  Due to the arguments given at the beginning of the proof of Lemma~\ref{lemma:reg of B}, $q-2$, $\Sp-\sigma_+$ and $\Sm-\sigma_-$ decay exponentially.
  Moreover, $\Omega$
  decays exponentially, and by an argument similar to the proof of this fact, $V'\circ\phi/\theta^2$ decays exponentially. This means that
  $\phi_1'$ and $(\ln\theta+3\tau)'$ decay exponentially. In particular, in case condition (i) of Definition~\ref{def:admconvsol} is satisfied,
  the estimates (\ref{seq:SpmphiprNilimmatterdom}) hold. In case condition (ii) of Definition~\ref{def:admconvsol} is satisfied, it can be verified
  that the conclusions of either Corollary~\ref{cor:LRS Bianchi type VIII} or Corollary~\ref{cor:LRS Bianchi type IX} are satisfied. For these reasons,
  the desired conclusion follows by arguments similar to the proofs of the results of Chapter~\ref{chapter:as to data on sing}.
\end{proof}
Next, it is of interest to make the following observation concerning the regularity of the map from regular data to data on the singularity.
\begin{lemma}\label{lemma:local repr map to sing}
  Let $0\leq V\in \mfP_{\a_V}^\infty$ for some $\a_{V}\in (0,1)$. Fix a Bianchi class A type $\mfT$ and $\mfs\in\{\iso,\roLRS,\roper,\rogen\}$
  satisfying one of the following conditions: $\mfs=\rogen$; $\mfs=\iso$ and $\mfT\in\{\mrI,\mrIX\}$; $\mfs=\roLRS$ and
  $\mfT\in \{\mrI,\mrII,\mrVIII,\mrIX\}$; and $(\mfT,\mfs)=(\mrVIz,\roper)$.
  Let $\varkappa_0\in\ro$ and assume that $e^{\varkappa_0}>[3\inf_sV(s)]^{1/2}$ if $\mfT\neq\mrIX$.
  Let, moreover, 
  $\xi_0\in\sfA_\mfT^\mfs(\varkappa_0)$ and assume that the $q$ corresponding to $\xi_0$ is $>-1$. Then there is an open neighbourhood of $\xi_0$ in
  $\sfA_\mfT^\mfs(\varkappa_0)$, say $U$, so that, considering $\sfF_\sfP(U)$ as an open subset of $\sfB_\mfT^\mfs(e^{\varkappa_0})$, the map
  \begin{equation}\label{eq:repr map from id to dos}
    \sfF_\sfP(U)\ni\eta\mapsto \sfF_\sfQ\circ\Psi_\mfT^{\mfs}\circ\sfF_\sfP^{-1}(\eta)
  \end{equation}
  is a diffeomorphism onto its image that decends to a map from $\sfF_\sfP(U)/\Gamma_\mfT^\mfs$ to $\sfU_\mfT^\mfs$ which is a local
  representation of the restriction of the map from Lemma~\ref{lemma:iso ri iso ids incl type} to the corresponding open subset of
  ${}^{\rosc}\mfB_{\mfT}^{\mfs}[V](e^{\varkappa_0})$. 
\end{lemma}
\begin{remark}
  When we speak of representations of the map from isometry classes of regular data to isometry classes of data on the singularity, we always have
  in mind that there is a natural parametrisation of ${}^{\rosc}\mfB_{\mfT}^{\mfs}[V](\vartheta_0)$ by $\sfR_\mfT^\mfs(\vartheta_0)$, see
  Remark~\ref{remark:par iso id fixed mc}; and that there is a natural parametrisation of ${}^{\rosc}\mfS_\mfT^\mfs$ by
  $\sfU_\mfT^\mfs$, see Lemma~\ref{lemma:sc mfB mfT mfs param sing}. 
\end{remark}
\begin{remark}
  The maps $\sfF_\sfP$ and $\sfF_\sfQ$ are defined in (\ref{eq:sfF sfP def}) and (\ref{eq:sfFQ def}) respectively. 
\end{remark}
\begin{remark}\label{remark:q big m one eq rond}
  In the case of isotropic Bianchi type I, the condition that $q>-1$ corresponds to the condition that the initial datum for the normal
  derivative of the scalar field is non-vanishing. 
\end{remark}
\begin{proof}
  Let $U$ be the open neighbourhood of $\xi_0$ in $\sfA_\mfT^\mfs(\varkappa_0)$ whose existence is guaranteed by Lemma~\ref{lemma:reg of B}. That the
  map (\ref{eq:repr map from id to dos}) is a diffeomorphism onto its image is an immediate consequence of Lemma~\ref{lemma:reg of B}. That it descends
  to a map from $\sfF_\sfP(U)/\Gamma_\mfT^\mfs$ to $\sfU_\mfT^\mfs$ is an immediate consequence of (\ref{eq:sfF sym comm}), (\ref{eq:sff sfQ comm}) and
  (\ref{eq:Psi mfT mfs def}). What remains to be verified is that it is a local representation of the restriction of the map from
  Lemma~\ref{lemma:iso ri iso ids incl type} to the corresponding open subset of ${}^{\rosc}\mfB_{\mfT}^{\mfs}[V](e^{\varkappa_0})$.
  Let $\eta\in \sfF_\sfP(U)\subset \sfB_{\mfT}^{\mfs}(e^{\varkappa_0})$. Then there are $\mfI\in {}^{\rosc}\mB_{\mfT}^{\mfs}[V](e^{\varkappa_0})$ such that $\eta$
  arises from $\mfI$; see Remarks~\ref{remark:eta BTs realised} and \ref{remark:x arising from mfI}. Let $\{e_i\}$ be the corresponding basis of $\mfg$.
  Let $\xi:=\sfF_\sfP^{-1}(\eta)$ and let $x$ be the solution to (\ref{seq:nui etc eqs}) corresponding to the initial data defined by $\xi$ at $\tau=0$ as
  in the statement of Lemma~\ref{lemma:reg of B}. By assumption, the limit on the right hand side of (\ref{eq:x limit dos}) exists, and we
  use the notation introduced on the left hand side of (\ref{eq:x limit dos}). Note also that the development associated with $\mfI$ is such that
  if we associate Wainwright-Hsu variables with this development as described in the proof of Proposition~\ref{prop:unique max dev} and in
  Section~\ref{ssection:whsuform} (using the frame $\{e_i\}$) and if we associate variables $\nu_k$, $\Sigma_k$, $\varkappa$, $\psi_0$ and $\phi_1$ with
  the Wainwright-Hsu variables as described in connection with (\ref{seq:nui etc eqs}), then this also produces the solution to (\ref{seq:nui etc eqs})
  corresponding to the initial data $\xi$ (at $\tau=0$). One particular consequence of the above observations is that the limit of $\mK$, say $\msK$,
  exists, and that $\msK e_i=p_i e_i$ (no summation), where $p_i=\sigma_{i,\infty}+1/3$. Next, the limit of $\theta^{-1}\phi_t$ is the same as the limit of
  $\phi_1/3$. If $\bPhi_1$ denotes the limit of $\theta^{-1}\phi_t$, we conclude that $\bPhi_1=\Phi_1/3$. Next, we turn to the limit of
  $\phi+\theta^{-1}\phi_t\ln\theta$, say $\bPhi_0$. Since
  \[
  \phi+\theta^{-1}\phi_t\ln\theta=\psi_0+\phi_1\varkappa/3\rightarrow \Phi_0+\Phi_1\varkappa_\infty/3,
  \]
  it follows that $\bPhi_0=\Phi_0+\Phi_1\varkappa_\infty/3$. Finally, we need to calculate the limit of the expansion normalised metric. Note, to
  this end, that the coefficients $a_i$ satisfy (\ref{eq:dai dtau}), where $\ell_i:=\Sigma_i+1/3$, and $a_i(0)=1$. This means that
  \[
  \bge(\theta^\mK e_i,\theta^\mK e_j)=a_i^{2}\exp\big(2(\varkappa-3\tau\big)\cdot(\Sigma_i+1/3))\de_{ij}
  \]
  (no summation). This means we have to prove that the following expressions have a limit:
  \begin{equation*}
    \begin{split}
      & 6\textstyle{\int}_{0}^\tau(\Sigma_i+1/3)ds+2\varkappa(\Sigma_i+1/3)-6\tau\Sigma_i-2\tau\\
      = & 6\textstyle{\int}_{0}^\tau(\Sigma_i-\sigma_{i,\infty})ds+2\varkappa(\Sigma_i+1/3)-6\tau(\Sigma_i-\sigma_{i,\infty}).
    \end{split}
  \end{equation*}
  However, since $\Sigma_i-\sigma_{i,\infty}$ and $\varkappa-\varkappa_\infty$ decay exponentially, it is clear that they do. In order to calculate
  the $o$ appearing in Definition~\ref{def:mfS def}, let
  \[
  e_{i}':=\theta^{-\ell_i}a_{i}^{-1}e_i
  \]
  (no summation). If the components of the commutator matrix associated with $a_{i}^{-1}e_i$ are $n_i(\tau)$, then the
  components of the commutator matrix associated with $e_i'$ are
  \[
  \hat{n}_k=\theta^{\ell_k-\ell_i-\ell_j}n_k=\theta^{1+\ell_k-\ell_i-\ell_j}(n_k/\theta)=\theta^{2\ell_k}\e_ke^{\mu_k},
  \]
  assuming $\{i,j,k\}=\{1,2,3\}$. Note that $o_k$ is the limit of $\hat{n}_k$. Compute
  \[
  \hat{n}_k=\e_k\exp\big(2\ell_k\cdot(\varkappa-3\tau)+\nu_k+6\ell_k\tau\big)\rightarrow\e_{k}e^{\nu_{k,\infty}+2p_k\varkappa_\infty}.
  \]
  This means that $o_k=\e_{k}e^{\nu_{k,\infty}+2p_k\varkappa_\infty}$. This means that the data at the singularity are represented by
  \[
  \sfF_\sfQ\circ\Psi_\mfT^\mfs\circ\sfF_\sfP^{-1}(\eta),
  \]
  as desired. 
\end{proof}

\section{Initial data with fixed mean curvature}\label{section:Id fix mean curv}
Let $0\leq V\in C^\infty(\ro)$, $\vartheta_0>\inf_{s\in\ro}[3V(s)]^{1/2}$ and $\vartheta_1>\vartheta_0$. Due to Lemma~\ref{lemma:BianchiAdevelopment}, we
know that if $(\mfT,\mfs)\neq (\mrI,\iso)$, $\mfT\neq\mrIX$ and $\mfI\in\mB_\mfT^\mfs[V](\vartheta_0)$, then the development $\md[V](\mfI)$, the mean
curvature of which is initially (say at $t=t_0$) equal to $\vartheta_0$, is such that there is a unique $t_1<t_0$ in the existence interval with the
property that the mean curvature of the $t=t_1$ hypersurface equals $\vartheta_1$. This defines a map
\begin{equation}\label{eq:PmfTmfs tz to}
  P_{\mfT,\mfs}^{\vartheta_0,\vartheta_1}:\mB_\mfT^\mfs[V](\vartheta_0)\rightarrow \mB_\mfT^\mfs[V](\vartheta_1).
\end{equation}
\index{$\a$Aa@Notation!Maps!PTsvathzvatho@$P_{\mfT,\mfs}^{\vartheta_0,\vartheta_1}$}%
Next, let $\mft\in \{+,\roap,\ropp\}$. Due to the requirements of Definition~\ref{def:plus ap pp nd}, it is possible to formulate a definition similar
to that of $P_{\mfT,\mfs}^{\vartheta_0,\vartheta_1}$ in order to obtain
\begin{equation}\label{eq:PmfTmfsplus tz to}
  P_{\mrIX,\mfs,\mft}^{\vartheta_0,\vartheta_1}:\mB_{\mrIX,\mft}^{\mfs,\rond}[V](\vartheta_0)\rightarrow \mB_{\mrIX,\mft}^{\mfs,\rond}[V](\vartheta_1).
\end{equation}
\index{$\a$Aa@Notation!Maps!PIXstsvathzvatho@$P_{\mrIX,\mfs,\mft}^{\vartheta_0,\vartheta_1}$}%
Next, we analyse the regularity of the corresponding maps between isometry classes of initial data. 
\begin{lemma}\label{lemma:PimfT well def}
  Let $V\in C^\infty(\ro)$ be non-negative. First, fix $(\mfT,\mfs)\neq (\mrI,\iso)$ such that $\mfT\neq\mrIX$ and assume that
  $\vartheta_0>\inf_{s\in\ro}[3V(s)]^{1/2}$ and that $\vartheta_1>\vartheta_0$. Then $P_{\mfT,\mfs}^{\vartheta_0,\vartheta_1}$
  induces a smooth map
  \begin{equation}\label{eq:PimfTmfs tz to}
    \Pi_{\mfT,\mfs}^{\vartheta_0,\vartheta_1}:{}^{\rosc}\mfB_\mfT^\mfs[V](\vartheta_0)\rightarrow {}^{\rosc}\mfB_\mfT^\mfs[V](\vartheta_1)
  \end{equation}
  \index{$\a$Aa@Notation!Maps!PiTsvathzvatho@$\Pi_{\mfT,\mfs}^{\vartheta_0,\vartheta_1}$}%
  which is a diffeomorphism onto its image.

  Second, fix $\mfs\in \{\iso,\roLRS,\rogen\}$ and assume that $\vartheta_0>0$ and that $\vartheta_1>\vartheta_0$. Then
  ${}^{\rosc}\mfB_{\mrIX,+}^{\mfs,\rond}[V](\vartheta_0)$ is an
  open subset of ${}^{\rosc}\mfB_{\mrIX}^{\mfs}[V](\vartheta_0)$ and $P_{\mrIX,\mfs,+}^{\vartheta_0,\vartheta_1}$ induces a smooth map
  \begin{equation}\label{eq:PimfTmfsplus tz to}
    \Pi_{\mrIX,\mfs,+}^{\vartheta_0,\vartheta_1}:{}^{\rosc}\mfB_{\mrIX,+}^{\mfs,\rond}[V](\vartheta_0)\rightarrow {}^{\rosc}\mfB_{\mrIX,+}^{\mfs,\rond}[V](\vartheta_1)
  \end{equation}
  \index{$\a$Aa@Notation!Maps!PiIXspvathzvatho@$\Pi_{\mrIX,\mfs,+}^{\vartheta_0,\vartheta_1}$}%
  which is a diffeomorphism onto its image. Finally, if, in addition to the above assumptions in the case that $\mfT=\mrIX$, $V\in\mfP_{\a_V}^{1}$, where
  $\a_V\in (0,1/3)$, and
  $\mft\in \{\roap,\ropp\}$, then ${}^{\rosc}\mfB_{\mrIX,\mft}^{\mfs,\rond}[V](\vartheta_0)$ is an open subset of ${}^{\rosc}\mfB_{\mrIX}^{\mfs}[V](\vartheta_0)$
  and $P_{\mrIX,\mfs,\mft}^{\vartheta_0,\vartheta_1}$ induces a smooth map
  \begin{equation}\label{eq:PimfTmfsmft tz to}
    \Pi_{\mrIX,\mfs,\mft}^{\vartheta_0,\vartheta_1}:{}^{\rosc}\mfB_{\mrIX,\mft}^{\mfs,\rond}[V](\vartheta_0)\rightarrow {}^{\rosc}\mfB_{\mrIX,\mft}^{\mfs,\rond}[V](\vartheta_1)
  \end{equation}
  \index{$\a$Aa@Notation!Maps!PiIXstvathzvatho@$\Pi_{\mrIX,\mfs,\mft}^{\vartheta_0,\vartheta_1}$}%
  which is a diffeomorphism onto its image.
\end{lemma}
\begin{proof}
  Consider (\ref{seq:EFEwrtvar})--(\ref{eq:phiddot}). Note that we can consider $\sfB$ introduced in Definition~\ref{def:degenerate version of id}
  to be a set of initial data for these equations. Note also that, given initial data in $\sfB$, the corresponding solutions to
  (\ref{seq:EFEwrtvar})--(\ref{eq:phiddot}) have the property that $n_{ij}=0$ for $i\neq j$ and $\sigma_{ij}=0$ for $i\neq j$; see the
  beginning of the proof of Proposition~\ref{prop:unique max dev}. Recall the notation $\sfB_\mfT^\mfs(\vartheta_0)$ introduced in
  Definition~\ref{def:B varthetaz}. 

  Let $\sfB_{\mrIX,+}^{\mfs,\rond}(\vartheta_0)$ denote the set of $\xi\in\sfB_{\mrIX}^\mfs(\vartheta_0)$ such that the solution to
  (\ref{seq:EFEwrtvar})--(\ref{eq:phiddot}) corresponding to initial data $\xi$ at $t=t_0$ has the following properties: $\theta_t(t)<0$ for all
  $t\leq t_0$; and there is a $t_1\leq t_0$ such that $X(t_1)>0$; see (\ref{eq:Xdef}). If $\xi\in \sfB_{\mrIX,+}^{\mfs,\rond}(\vartheta_0)$, the
  initial data are specified at $t_0$, $t_1\leq t_0$ and $X(t_1)>0$, then, due
  to the fact that solutions depend continuously on initial data, it is clear that there is a neighbourhood of $\xi$ in $\sfB_{\mrIX}^\mfs(\vartheta_0)$
  such that the corresponding solutions have the property that $\theta_t<0$ on $[t_1,t_0]$ and $X(t_1)>0$. Due to the proof of
  Lemma~\ref{lemma:BianchiAdevelopment}), it then follows that $X(t)>0$ and $\theta(t)>0$ for all $t\leq t_1$. Combining this observation with
  (\ref{eq:theta t ito X}), this means that $\theta_t<0$ for all $t\leq t_0$. To summarise: $\sfB_{\mrIX,+}^{\mfs,\rond}(\vartheta_0)$ is an open
  subset of $\sfB_{\mrIX}^\mfs(\vartheta_0)$. 
  
  Next, assume that $V\in\mfP_{\a_V}^{0}$, where $\a_V\in (0,1/3)$. Then
  ${}^{\rosc}\mfB_{\mrIX,\roap}^{\mfs,\rond}[V](\vartheta_0)={}^{\rosc}\mfB_{\mrIX,\ropp}^{\mfs,\rond}[V](\vartheta_0)$; see Corollary~\ref{cor:ap eq pp}. For this
  reason, it is sufficient to assume that $\mft=\roap$. Consider the set of $\xi\in\sfB_{\mrIX}^\mfs(\vartheta_0)$ such that the solution
  to (\ref{seq:EFEwrtvar})--(\ref{eq:phiddot}) corresponding to initial data $\xi$ at $t=t_0$ has the following properties: $\theta_t(t)<0$ for all
  $t\leq t_0$; $V\circ\phi(t)/\theta^2(t)\rightarrow 0$ as $t\downarrow t_-$; and the negative part of $X(t)/\theta^2(t)$ converges to zero as
  $t\downarrow t_-$. Denote this set by $\sfB_{\mrIX,\roap}^{\mfs,\rond}(\vartheta_0)$. Note that there are two possibilities, given
  $\xi\in\sfB_{\mrIX,\roap}^{\mfs,\rond}(\vartheta_0)$: Either the solution
  corresponding to $\xi$ is such that there is a $t_1\leq t_0$ with the property that $X(t_1)>0$ or there is no such $t_1$. If there is such a $t_1$,
  then there is an open neighbourhood of $\xi$ in $\sfB_{\mrIX}^\mfs(\vartheta_0)$ contained in $\sfB_{\mrIX,\roap}^{\mfs,\rond}(\vartheta_0)$; see the argument
  in the case that $\mft=+$ as well as Corollary~\ref{cor:ap eq pp} and its proof. We can therefore assume that $\xi$ is such that the corresponding
  solution has the property that
  $X(t)\leq 0$ for all $t\leq t_0$. This means that the conditions of Lemma~\ref{lemma:BianchiAdevelopment} are not satisfied, but that the conditions
  of Lemma~\ref{lemma:Bianchi IX remainder} are. Thus the last statement of Theorem~\ref{thm:dichotomy} applies. As noted in
  Remark~\ref{remark:ex of admiss}, this means that the corresponding solution to (\ref{eq:thetaprime})--(\ref{eq:constraint}) is an admissible
  convergent solution. Since the set of admissible solutions is open, see Lemma~\ref{lemma:reg of B} (note that the proof of openness only relies on
  $V$ belonging to $\mfP_{\a_V}^{1}$), there is an open neighbourhood of $\xi$ in $\sfB_{\mrIX}^\mfs(\vartheta_0)$ such that the corresponding solutions
  are admissible convergent solutions. Next, due to Lemma~\ref{lemma:reg of B} and Remark~\ref{remark:Psi T s Cm}, the map from initial data to
  asymptotic data is continuous. There is thus a neighbourhood of $\xi$ such that the corresponding solutions have $\Phi_1\neq 0$. Combining this
  observation with the definition of admissible convergent solutions, it follows that there is an open neighbourhood of $\xi$ such that the
  corresponding solutions have the property
  that $V\circ\phi/\theta^2$ and the negative part of $X/\theta^2$ converge to zero. What remains is to prove that the corresponding solutions also have
  the property that $\theta_t<0$ for all $t\leq t_0$. However, due to the proof of Lemma~\ref{lemma:reg of B}, it is clear that $q$ converges uniformly
  to $2$ (for initial data in a suitably small neighbourhood of $\xi$). From this observation and the continuity of the flow, it can be deduced that
  $\theta_t<0$, as desired. This finishes the proof of openness in the case that $\mft\in\{\roap,\ropp\}$.
  
  Due to the observations made below (\ref{eq:ham cons fixed theta}), we know that $\sfB_\mfT^\mfs(\vartheta_0)$ is a smooth submanifold of $\rn{8}$.
  Note also that an element $\xi\in \sfB_\mfT^\mfs(\vartheta_0)$ can be considered to be initial data for (\ref{seq:EFEwrtvar})--(\ref{eq:phiddot}),
  with the following caveat: the Weingarten matrix $K$ appearing in $\xi$ is divided into its trace $\theta$ and trace free part $\sigma$ in order
  to obtain initial data for (\ref{seq:EFEwrtvar})--(\ref{eq:phiddot}). Since this distinction corresponds to a trivial algebraic decomposition,
  we, in what follows, abuse notation and say that $\xi$ are initial data for (\ref{seq:EFEwrtvar})--(\ref{eq:phiddot}). Next, assume that
  $(\mfT,\mfs)\notin \{(\mrI,\iso),(\mrVIIz,\iso),(\mrVIIz,\roLRS)\}$ and $\mfT\neq\mrIX$. Then, due to Lemma~\ref{lemma:BianchiAdevelopment}, $\theta_t<0$
  for $t\leq t_0$ (assuming the initial data have been specified at $t_0$) and $\theta$ diverges to $\infty$ to the past. There is thus a map from
  $\sfB^\mfs_\mfT(\vartheta_0)$ to $\sfB^\mfs_\mfT(\vartheta_1)$,  defined as follows: given $\xi\in\sfB^\mfs_\mfT(\vartheta_0)$, let $x$ be the 
  solution to (\ref{seq:EFEwrtvar})--(\ref{eq:phiddot}) with $x(t_0)=\xi$. Then there is a unique $t_1<t_0$ such that the mean curvature of the solution
  equals $\vartheta_1$ at $t=t_1$, so that $x(t_1)\in\sfB^\mfs_\mfT(\vartheta_1)$ (note that the flow associated with (\ref{seq:EFEwrtvar})--(\ref{eq:phiddot})
  respects the symmetry type; see the proof of Proposition~\ref{prop:unique max dev}). This yields a map
  $\Xi:\sfB^\mfs_\mfT(\vartheta_0)\rightarrow\sfB^\mfs_\mfT(\vartheta_1)$ (which is a representation of the map introduced in (\ref{eq:PmfTmfs tz to}),
  keeping Remark~\ref{remark:eta BTs realised} in mind). Note that this map is injective due to the uniqueness of solutions to
  (\ref{seq:EFEwrtvar})--(\ref{eq:phiddot}) and the fact that $\theta_t(t)<0$ for all
  $t\leq t_0$.  Next, we prove that this map is a diffeomorphism onto its image. Let, to this end,
  $\xi_0 \in\sfB^\mfs_\mfT(\vartheta_0)$ and let $t_a<t_0$ be such that the solution $x$ to (\ref{seq:EFEwrtvar})--(\ref{eq:phiddot}) with $x(t_0)=\xi_0$
  has the property that the mean curvature at $t_a$ is $\vartheta_1$. There is an open neighbourhood $U\subset \sfB^\mfs_\mfT(\vartheta_0)$ of $\xi_0$ and
  an open interval $I$ containing $t_0$ and $t_a$ such that the flow associated with the equations (\ref{seq:EFEwrtvar})--(\ref{eq:phiddot}) is well defined
  and smooth on $I\times U$. Let $\theta(t;\xi)$ denote the mean curvature of the solution with initial data $\xi\in\sfB^\mfs_\mfT(\vartheta_0)$ at $t=t_0$,
  evaluated at $t$. Then $\theta(t_0;\xi)=\vartheta_0$ and $\theta(t_a;\xi_0)=\vartheta_1$. We wish to find the time $t$ such that $\theta(t;\xi)=\vartheta_1$.
  By the implicit function theorem, it is, locally, sufficient that the derivative of $\theta$ with respect to $t$ is non-zero. However, we
  know that $\theta_t(t;\xi)<0$. This means that there is an open neighbourhood $V$ of $\xi_0$ in
  $\sfB^\mfs_\mfT(\vartheta_0)$ and a smooth function $T:V\rightarrow\ro$ with $T(\xi_0)=t_a$ such that $\theta(T(\xi);\xi)=\vartheta_1$. This means
  that, locally, $\Xi$ is represented by $\Xi(\xi)=x(T(\xi);\xi)$, so that $\Xi$ is smooth. The proof that the local inverse is smooth is
  similar. To conclude, $\Xi$ is a diffeomorphism onto its image. Finally, due to the symmetries of the
  equations, it is clear that $\Xi$ descends to a map from $\sfR_\mfT^\mfs(\vartheta_0)$ to $\sfR_\mfT^\mfs(\vartheta_1)$, which is also a diffeomorphism
  onto its image; cf. Lemma~\ref{lemma:sfR mfT mfs}, Definition~\ref{def:B varthetaz} and Lemma~\ref{lemma:sfRmfTmfsvartheta sm mfd}. In order to see
  that this map is a representation of (\ref{eq:PimfTmfs tz to}), it is sufficient to appeal to Lemma~\ref{lemma:sc mfB mfT mfs param} and its proof,
  Remark~\ref{remark:eta BTs realised} and Lemma~\ref{lemma:isometryinducesequalnuK}. The arguments in the remaining two cases are essentially
  identical. 
\end{proof}
In the case of isotropic Bianchi type I, we need to impose additional conditions on the potential. Assume that $V\in\mfP_{\ropar}$ and that
$\vartheta_-:=[3v_{\max}(V)]^{1/2}$. Fix $\vartheta_-<\vartheta_0<\vartheta_1$. Let $\mfI\in {}^{\rosc}\mfB_{\mrI}^{\iso}[V]$ with initial mean curvature
$\vartheta_0$. Due to Lemma~\ref{lemma:bth large enough}, the range of the mean curvature in $\md[V](\mfI)$ includes $(\vartheta_-,\infty)$.
Moreover, due to the same lemma, if the initial data induced by the development at $t_0$ is $\mfI$, there is a unique $t_1<t_0$ such that
the initial data induced at the $t=t_1$ hypersurface has mean curvature $\vartheta_1$. This naturally leads to a map as in (\ref{eq:PmfTmfs tz to}).
Moreover, this map induces a map
\[
{}^{\ropre}\Pi_{\mrI,\iso}^{\vartheta_0,\vartheta_1}:{}^{\rosc}\mfB_{\mrI}^{\iso}[V](\vartheta_0)\rightarrow {}^{\rosc}\mfB_{\mrI}^{\iso}[V](\vartheta_1)
\]
which is a bijection. Let, for $\vartheta>\vartheta_-$, ${}^{\rosc}\mfB_{\mrI,\rond}^{\iso}[V](\vartheta)$ be the subset of 
${}^{\rosc}\mfB_{\mrI}^{\iso}[V](\vartheta)$ with $\bphi_1\neq 0$. The main advantage of this set is that if the isometry class of $\mfI$ belongs to
it, then the initial value of $\theta_t$ is strictly negative; see (\ref{eq:thetad}). Let
\[
U_{\mrI,\iso}^{\vartheta_0,\vartheta_1}:=({}^{\ropre}\Pi_{\mrI,\iso}^{\vartheta_0,\vartheta_1})^{-1}\big({}^{\rosc}\mfB_{\mrI,\rond}^{\iso}[V](\vartheta_1)\big)
\cap {}^{\rosc}\mfB_{\mrI,\rond}^{\iso}[V](\vartheta_0).
\]
Then, by arguments similar to the proof of Lemma~\ref{lemma:PimfT well def}, $U_{\mrI,\iso}^{\vartheta_0,\vartheta_1}$ is an open set and
\[
\Pi_{\mrI,\iso}^{\vartheta_0,\vartheta_1}:={}^{\ropre}\Pi_{\mrI,\iso}^{\vartheta_0,\vartheta_1}|_{U_{\mrI,\iso}^{\vartheta_0,\vartheta_1}}
\]
defines a smooth map
\begin{equation}\label{eq:Pi mrI iso vto vtt}
  \Pi_{\mrI,\iso}^{\vartheta_0,\vartheta_1}:U_{\mrI,\iso}^{\vartheta_0,\vartheta_1}\rightarrow {}^{\rosc}\mfB_{\mrI,\rond}^{\iso}[V](\vartheta_1)
\end{equation}
which is a diffeomorphism onto its image. Next, we prove Lemma~\ref{lemma:dev smo str nisoI nIX}.
\begin{proof}[Proof of Lemma~\ref{lemma:dev smo str nisoI nIX}]
  That the map $\iota_\vartheta$ is well defined is an immediate consequence of Proposition~\ref{prop:iso id to iso sol}. In order to prove that it is
  injective, assume that $\mfI_\a\in {}^{\rosc}\mB_\mfT^\mfs[V](\vartheta)$, $\a\in\{0,1\}$, and that $\md[V](\mfI_0)$ and $\md[V](\mfI_1)$ are isometric.
  Due to Lemma~\ref{lemma:BianchiAdevelopment}, we know that the existence interval of both developments is $(0,\infty)$. Moreover, the time coordinate
  $t$ of a point in one of the developments can invariantly be characterised as the maximal length of a past directed timelike curve. Since
  this characterisation is preserved by isometries, it is clear that the constant-$t$ hypersurfaces are preserved by the isometry. This means
  that the initial data induced on the constant-$t$ hypersurfaces of $\md[V](\mfI_0)$ and $\md[V](\mfI_1)$ are isometric. In particular, the
  mean curvature of $G_0\times\{t\}$ in $\md[V](\mfI_0)$ equals the mean curvature of $G_1\times\{t\}$ in $\md[V](\mfI_1)$. It is therefore
  meaningful to speak of $\theta(t)$, the mean curvature at the constant-$t$ hypersurface, irrespective of the development. Since
  $\theta_t<0$, there is a unique $t_0$ such that $\theta(t_0)=\vartheta$ (the existence is guaranteed by construction). Since the isometry between
  the developments preserves the hypersurfaces with $t=t_0$, it induces an isometry from $\mfI_0$ to $\mfI_1$. This proves injectivity of
  $\iota_\vartheta$. By the definition of ${}^{\rosc}\mfD_\mfT^\mfs[V]$, it is also clear that it is the union of the images of the $\iota_\vartheta$. 

  Next, we need to define a smooth structure on ${}^{\rosc}\mfD_\mfT^\mfs[V]$. Let $\vartheta_-:=\inf_{s\in\ro}[3V(s)]^{1/2}$ and define
  \begin{equation}\label{eq:Sol def}
    {}^{\rosc}\mathfrak{Sol}_\mfT^\mfs[V]:=\bigsqcup_{\vartheta>\vartheta_-}{}^{\rosc}\mfB_\mfT^\mfs[V](\vartheta)/\sim,
  \end{equation}
  where $x_1\sim x_2$ for $x_i\in \mfB_\mfT^\mfs[V](\vartheta_i)$ if $x_2=\Pi_{\mfT,\mfs}^{\vartheta_1,\vartheta_2}(x_1)$, in case $\vartheta_2>\vartheta_1$;
  $x_1=\Pi_{\mfT,\mfs}^{\vartheta_2,\vartheta_1}(x_2)$, in case $\vartheta_1>\vartheta_2$; and $x_1=x_2$, in case $\vartheta_1=\vartheta_2$. In what follows,
  we denote the map from the disjoint union
  to the quotient by $\pi$. In order to prove that there is a bijection from ${}^{\rosc}\mfD_\mfT^\mfs[V]$ to ${}^{\rosc}\mathfrak{Sol}_\mfT^\mfs[V]$,
  note that by the above observations concerning $\iota_\vartheta$, there is a surjective map
  \[
  \iota:\bigsqcup_{\vartheta>\vartheta_-}{}^{\rosc}\mfB_\mfT^\mfs[V](\vartheta)\rightarrow {}^{\rosc}\mfD_\mfT^\mfs[V]
  \]
  defined by taking $x\in {}^{\rosc}\mfB_\mfT^\mfs[V](\vartheta)$ to $\iota_\vartheta(x)$. Assume now that $\iota(x_0)=\iota(x_1)$ for
  $x_\a\in {}^{\rosc}\mfB_\mfT^\mfs[V](\vartheta_\a)$, where $\vartheta_-<\vartheta_0\leq \vartheta_1$. If $\vartheta_0=\vartheta_1$, we know
  that $x_0=x_1$, since $\iota_\vartheta$ is injective. We can therefore assume that $\vartheta_0<\vartheta_1$. Let $\mfI_\a$, $\a\in\{0,1\}$,
  be such that $\iota(x_\a)$ is the isometry class of $\md[V](\mfI_\a)$. Then, by assumption, $\md[V](\mfI_0)$ and $\md[V](\mfI_1)$ are
  isometric. By an argument similar to the proof of the injectivity of $\iota_\vartheta$, it is clear that there is a $t_1>0$ such that
  the initial data induced on $G_\a\times\{t_1\}$ by $\md[V](\mfI_\a)$ are isometric; the mean curvature of the $t=t_1$ hypersurface is
  $\vartheta_1$; and the isometry class of the initial data induced on $G_1\times\{t_1\}$ by $\md[V](\mfI_1)$ is $x_1$. Moreover, there is
  a $t_0>t_1$ such that the mean curvature of the $t=t_0$ hypersurface in $\md[V](\mfI_0)$ is $\vartheta_0$ and the isometry class of the
  initial data induced on the $t=t_0$ hypersurface in $\md[V](\mfI_0)$ is $x_0$. This means that $x_1=\Pi_{\mfT,\mfs}^{\vartheta_0,\vartheta_1}(x_0)$.
  This proves that $\iota$ induces a bijection, say $\biota:{}^{\rosc}\mathfrak{Sol}_\mfT^\mfs[V]\rightarrow {}^{\rosc}\mfD_\mfT^\mfs[V]$.

  To define a topology on ${}^{\rosc}\mathfrak{Sol}_\mfT^\mfs[V]$, we use the fact that ${}^{\rosc}\mfB_\mfT^\mfs[V](\vartheta)$ are smooth manifolds;
  see Lemma~\ref{lemma:sfRmfTmfsvartheta sm mfd} and Remark~\ref{remark:par iso id fixed mc}. This defines a topology on the disjoint union
  of the ${}^{\rosc}\mfB_\mfT^\mfs[V](\vartheta)$ for $\vartheta>\vartheta_-$. We endow ${}^{\rosc}\mathfrak{Sol}_\mfT^\mfs[V]$ with the quotient
  space topology. Assume now that $x_\a\in {}^{\rosc}\mathfrak{Sol}_\mfT^\mfs[V]$, $\a\in\{0,1\}$, are such that $x_0\neq x_1$. We can, without
  loss of generality, assume that $x_\a$ is the projection of $y_\a\in {}^{\rosc}\mfB_\mfT^\mfs[V](\vartheta_\a)$, where
  $\vartheta_-<\vartheta_0\leq \vartheta_1$.
  If $\vartheta_0=\vartheta_1$, then we can choose open neighbourhoods $U_\a$ of $y_\a$ in ${}^{\rosc}\mfB_\mfT^\mfs[V](\vartheta_1)$ which are
  disjoint. If $\vartheta_0<\vartheta_1$, we can choose open neighbourhoods $U_\a$ of $\Pi_{\mfT,\mfs}^{\vartheta_0,\vartheta_1}(y_0)$ and
  $y_1$ in ${}^{\rosc}\mfB_\mfT^\mfs[V](\vartheta_1)$ which are disjoint. Let $V_\a:=\pi(U_\a)$. We need to prove that the $V_\a$ are open and disjoint.
  Note, to this end, that
  \begin{equation}\label{eq:pi inv Va}
    \pi^{-1}(V_\a)=\bigsqcup_{\vartheta_-<\vartheta<\vartheta_1}(\Pi_{\mfT,\mfs}^{\vartheta,\vartheta_1})^{-1}(U_\a)\bigsqcup U_\a
    \bigsqcup_{\vartheta_1<\vartheta}\Pi_{\mfT,\mfs}^{\vartheta_1,\vartheta}(U_\a),
  \end{equation}
  which is clearly an open set. This means that $V_\a$ is open. From this formula, it is also clear that since the $U_\a$ are disjoint, the $V_\a$
  are disjoint. This proves that ${}^{\rosc}\mathfrak{Sol}_\mfT^\mfs[V]$ is Hausdorff. Next, let $Q$ be the set of rational $\vartheta>\vartheta_-$ and
  for for each $\vartheta\in Q$, let $B_\vartheta$ be a countable basis for the topology of ${}^{\rosc}\mfB_\mfT^\mfs[V](\vartheta)$. Let $B$ be the set
  consisting of the projections of the sets in $B_\vartheta$ for $\vartheta\in Q$ to ${}^{\rosc}\mathfrak{Sol}_\mfT^\mfs[V]$. We leave it as an exercise
  to the reader to verify that $B$ is a basis for the topology.

  Next, let $\pi_\vartheta$ denote the restriction of $\pi$ to ${{}^{\rosc}\mfB_\mfT^\mfs[V](\vartheta)}$. In order to prove that this map is a homeomorphism
  onto its image, note that it can be written as a composition of the inverse of $\biota$ with $\iota_\vartheta$. Since $\biota$ is a bijection
  and $\iota_\vartheta$ is injective, it is clear that $\pi_\vartheta$ is a bijection onto its image. Due to (\ref{eq:pi inv Va}), it is clear that
  $\pi_\vartheta$ maps open sets to open sets. That the inverse image of an open set under $\pi_\vartheta$ is an open set is an immediate consequence of
  the definition of the quotient space topology. Thus $\pi_\vartheta$ is a homeomophism onto its image, and the image is an open set. 
  
  Next, we need to define a smooth structure. Fix $\vartheta<\vartheta_0\leq\vartheta_1$ and local charts $(U_\a,\varphi_\a)$ in
  ${}^{\rosc}\mfB_\mfT^\mfs[V](\vartheta_\a)$, $\a\in\{0,1\}$; here $\varphi_\a:U_\a\rightarrow\rn{n}$.
  Let $V_\a=\pi_{\vartheta_\a}(U_\a)$ and define
  \begin{equation}\label{eq:phia def}
    \phi_\a:=\varphi_\a\circ(\pi_{\vartheta_\a})^{-1}|_{V_\a}.
  \end{equation}
  Note that $\phi_\a$ is a homeomorphism from $V_\a$ to $\varphi_\a(U_\a)$. Note also that the domains of maps of this form cover
  ${}^{\rosc}\mathfrak{Sol}_\mfT^\mfs[V]$. Next, note that 
  \[
  \phi_1\circ\phi_0^{-1}=\varphi_1\circ\Pi_{\mfT,\mfs}^{\vartheta_0,\vartheta_1}\circ\varphi_0^{-1}.
  \]
  This means that $\phi_1\circ\phi_{0}^{-1}$ and $\phi_0\circ\phi_1^{-1}$ are smooth. Thus coordinate charts of the form $(V_\a,\phi_\a)$ define a
  smooth atlas. This atlas can be extended to a maximal atlas, endowing ${}^{\rosc}\mathfrak{Sol}_\mfT^\mfs[V]$ with a smooth structure. Combining
  this smooth structure with the bijection $\biota$ endows ${}^{\rosc}\mfD_\mfT^\mfs[V]$ with a (topology and) smooth structure. Finally, we wish to
  prove that the image of $\iota_\vartheta$ is open and that $\iota_\vartheta$ is a diffeomorphism onto its image (we have already noted that the union
  of the images of the $\iota_\vartheta$ is all of ${}^{\rosc}\mfD_\mfT^\mfs[V]$). Since $\iota_\vartheta=\biota\circ\pi_\vartheta$, it is clear that the
  range of $\iota_\vartheta$ is an open set and that $\iota_\vartheta$ is a homeomorphism onto its image. Next, let
  $p\in {}^{\rosc}\mfB_\mfT^\mfs[V](\vartheta)$. Let $(U,\varphi)$ be local coordinates on ${}^{\rosc}\mfB_\mfT^\mfs[V](\vartheta)$ with $p\in U$.
  Define $W:=\pi_{\vartheta}(U)$ and $\phi:W\rightarrow\rn{n}$ in analogy with (\ref{eq:phia def}). Let $\psi$ be the composition of $\phi$ with
  the inverse of $\biota$. Then $(\iota_\vartheta(U),\psi)$ are local coordinates on ${}^{\rosc}\mfD_\mfT^\mfs[V]$. Calculate
  \[
  \psi\circ\iota_\vartheta\circ\varphi^{-1}=\varphi\circ\pi_\vartheta^{-1}\circ(\biota)^{-1}\circ\iota_\vartheta\circ\varphi^{-1}
  =\varphi\circ(\biota\circ\pi_\vartheta)^{-1}\circ\iota_\vartheta\circ\varphi^{-1}=\roId. 
  \]
  This proves that $\iota_\vartheta$ is a local diffeomorphism. The lemma follows. 
\end{proof}
Next, we prove Lemmas~\ref{lemma:dev smo str IX} and \ref{lemma:dev smo str isoI}. 
\begin{proof}[Proof of Lemma~\ref{lemma:dev smo str IX}]
  The proof is similar to the proof of Lemma~\ref{lemma:dev smo str nisoI nIX}. The details are left to the reader. 
\end{proof}
\begin{proof}[Proof of Lemma~\ref{lemma:dev smo str isoI}]
  Again, the proof is similar to the proof of Lemma~\ref{lemma:dev smo str nisoI nIX}. However, there are some differences. First, in the
  proof of Lemma~\ref{lemma:dev smo str nisoI nIX}, we use the fact that the existence interval of all developments is $(0,\infty)$. In the
  present setting, it could be $\ro$. However, that the map $\iota_\vartheta$ is well defined is an immediate consequence of
  Proposition~\ref{prop:iso id to iso sol} and Proposition~\ref{lemma:bth large enough}. In order to prove that it is injective, assume that
  $\mfI_\a\in {}^{\rosc}\mB_{\mrI,\rond}^{\iso}[V](\vartheta)$, $\a\in\{0,1\}$, where $\vartheta\in I_V$, and that $\md[V](\mfI_0)$ and $\md[V](\mfI_1)$
  are isometric. Denote the isometry by $\psi$. Let $G_{\a,t}:=G_\a\times\{t\}$ denote the
  constant-$t$ hypersurfaces in $\md[V](\mfI_\a)=(M_\a,g_\a,\phi_\a)$, where $M_\a=G_\a\times J_\a$. Due to the proof of
  Lemma~\ref{lemma:bth large enough}, we know that the constant-$t$ hypersurfaces on which the time derivative of the scalar field are zero are
  isolated. Fix a $t_0\in J_0$ such that $(\d_t\phi_0)(t_0)\neq 0$. Let $p_0\in G_0$. Note that $\d_t|_{(p_0,t_0)}$ can be characterised as the
  unique future directed unit vector parallel with the spacetime gradient of $\phi_0$ at $(p_0,t_0)$. Since all these properties are preserved by
  $\psi$, it follows that $\psi_*\d_t|_{(p_0,t_0)}=\d_t|_{\psi(p_0,t_0)}$. Since this holds on an open and dense subset of $M_0$, it is clear that
  $\psi$ maps $\d_t$ to $\d_t$. Since $\psi$ is an isometry, curves in $G_{0,t}$ are mapped to curves perpendicular to $\d_t$. This means that
  their image is contained in a constant-$t$ hypersurface. Due to the connectedness of $G_0$, this means that for each $t$, $G_{0,t}$ is mapped
  into a constant-$t$ hypersurface. Since the same argument can be applied to $\psi^{-1}$, it follows that $\psi$ preserves the constant-$t$
  hypersurfaces. Since $\psi$ is an isometry and since the mean curvatures are strictly decreasing, see Lemma~\ref{lemma:bth large enough}, it
  also follows that $\psi$ maps constant-$t$ hypersurfaces of a mean curvature belonging to $I_V$ to a constant-$t$ hypersurface with the same
  mean curvature. This applies, in particular, to the case that the mean curvature is $\vartheta$, so that $\psi$ induces an isometry between
  $\mfI_0$ and $\mfI_1$. This proves injectivity. That the union of the images of the $\iota_\vartheta$ is 
  ${}^{\rosc}\mfD_{\mrI,\roc}^{\iso}[V]$ is an immediate consequence of the definition and Lemma~\ref{lemma:bth large enough}.

  In order to define the smooth structure, define
  \[
    {}^{\rosc}\mathfrak{Sol}_{\mrI}^{\iso}[V]:=\bigsqcup_{\vartheta>\vartheta_-}{}^{\rosc}\mfB_{\mrI,\rond}^{\iso}[V](\vartheta)/\sim,
  \]    
  where $\vartheta_-:=[3v_{\max}(V)]^{1/2}$. The equivalence relation $\sim$ is defined in analogy with the equivalence relation introduced in
  connection with (\ref{eq:Sol def}). The only difference is that we use the maps introduced in (\ref{eq:Pi mrI iso vto vtt}) instead of
  the ones given in (\ref{eq:PimfTmfs tz to}). We leave it to the reader to verify that the remaining steps of the proof of
  Lemma~\ref{lemma:dev smo str nisoI nIX} can be modified to apply in the present setting as well. The lemma follows. 
\end{proof}
Under the assumptions of Lemma~\ref{lemma:parametr by single const th surf}, the set of isometry classes of initial data with one fixed mean
curvature parametrises the set of isometry classes of developments. The situation in case $(\mfT,\mfs)=(\mrI,\iso)$ is somewhat more complicated.
Nevertheless, Lemma~\ref{lemma:iso BI dev homeom one fix CMC} holds.  
\begin{proof}[Proof of Lemma~\ref{lemma:iso BI dev homeom one fix CMC}]
  Let $x\in {}^{\rosc}\mfB_{\mrI}^{\iso}[V](\vartheta)$. Due to Lemma~\ref{lemma:bth large enough}, developments corresponding to initial data in $x$ have
  crushing singularities. Combining this observation with Proposition~\ref{prop:iso id to iso sol}, it follows that $x$ gives rise to a unique element
  of ${}^{\rosc}\mfD_{\mrI,\roc}^{\iso}[V]$. That the corresponding map, say $\varphi$, is injective follows by an argument similar to the beginning of the
  proof of Lemma~\ref{lemma:dev smo str isoI}. Next, given an isotropic and simply connected Bianchi type I development with a crushing singularity,
  it has to induce an element in ${}^{\rosc}\mfB_{\mrI}^{\iso}[V](\vartheta)$ due to Lemma~\ref{lemma:bth large enough}. Thus $\varphi$ is a bijection.
  In order to prove that $\varphi$ is a homeomorphism, let $x$ be as above and $y=\md[V](\mfI)=(M,g,\phi)$, where $\mfI\in x$ and $M=G\times J$.
  Assume $t_0\in J$ to be such that the initial data induced by $y$ on $G_{t_0}$ are $\mfI$. Let $t_1<t_0$ with
  $t_1\in J$ be such that the isometry class of initial data induced by $y$ at $t_1$ belongs to ${}^{\rosc}\mfB_{\mrI,\rond}^{\iso}[V](\vartheta_1)$ for
  some $\vartheta_1>\vartheta$. That there is such a $t_1$ follows from Lemma~\ref{lemma:bth large enough}.
  By an argument similar to the proof of Lemma~\ref{lemma:PimfT well def}, there is then an open neighbourhood, say
  $U$, of $x$ such that the time at which the isometry class of solutions, corresponding to initial data $\xi\in U$ at $t_0$, induces data at
  ${}^{\rosc}\mfB_{\mrI,\rond}^{\iso}[V](\vartheta_1)$, say $T(\xi)$, depends smoothly on $\xi$. If $y(\cdot;\eta)$ is a development corresponding to
  initial data $\eta\in\xi$ at $t_0$, it is then clear that the map taking $\xi\in U$ to the isometry class of initial data induced by  $y(\cdot;\eta)$
  at $T(\xi)$ is smooth. For this reason, the map $\varphi$ is smooth. The continuity of the inverse can be demonstrated by the strict monotonicity
  of $\theta$ along a solution curve and the smoothness of the flow associated with the equations. We leave the details to the reader. 
\end{proof}

\chapter{Geometric results in the direction of the singularity}\label{chapter:geo res dir of sing}

\section{Geometric data on the singularity}\label{ssection:dataonsingularity}
Our next goal is to prove Theorem~\ref{thm:dataonsingtosolution}. However, it is convenient to first characterise the symmetry conditions
in Definition~\ref{def:isotropic} in terms of a basis for the Lie algebra $\mfg$. Let $(G,\msH,\msK,\Phi_{0},\Phi_{1})\in \mS$; see
Definition~\ref{def:sets of id on singularity}. Let $\{e_{i}\}$ and $\{e_{i}'\}$ be two bases of $\mfg$ satisfying
\begin{equation}\label{eq:canonicalbasisBianchiA}
  \msH(e_{i},e_{j})=\delta_{ij},\ \ \
  \msK e_{i}=p_{i}e_{i},\ \ \
  [e_{i},e_{j}]=\e_{ijk}o_{k}e_{k}
\end{equation}
(no summation on $i$ in the second equality and no summation on $k$ in the third equality)     
and (\ref{eq:canonicalbasisBianchiA}) with $e_i$, $p_i$ and $o_i$ replaced by $e_i'$, $p_i'$ and $o_i'$ respectively; cf. Lemma~\ref{lemma:x arises from Iinf}. Denote the
corresponding commutator matrices (see Definition~\ref{def:comm and Wein matr}) by $o=\mathrm{diag}(o_1,o_2,o_3)$ and $o'=\mathrm{diag}(o_1',o_2',o_3')$ respectively.
Then there is an orthogonal matrix $A$ such that $e_i=A_i^{\phantom{i}j}e_j'$. Moreover, due to \cite[Lemma~19.6, p.~207]{RinCauchy}, 
\begin{equation}\label{eq:o to o prime}
  o'=(\det A)^{-1}A^t oA.
\end{equation}
In addition, $K'=A^tKA$, where $K=\mathrm{diag}(p_1,p_2,p_3)$ and $K'=\mathrm{diag}(p_1',p_2',p_3')$. 

\begin{lemma}\label{lemma:frame isotropic setting}
  Let $\mfI_\infty=(G,\msH,\msK,\Phi_{0},\Phi_{1})\in\mS$.
  Then $\mfI_\infty$ is isotropic if and only if there is a basis $\{e_i\}$ of $\mfg$ such that (\ref{eq:canonicalbasisBianchiA}) holds with all the
  $p_i$ equal to $1/3$ and all the $o_i$ equal and $\geq 0$. 
\end{lemma}
\begin{remark}
  In the case of isotropic Bianchi type I and IX data, applying an isometry $T:\mfg\rightarrow\mfg$ with respect to $\msH$ to a frame $\{e_i\}$ as in the
  statement of the lemma yields a frame with the same properties, except that $o$ is transformed to $-o$ if $T$ is orientation reversing. 
\end{remark}
\begin{proof}
  The statement is an immediate consequence of Lemma~\ref{lemma:psi sigma plus isotropic sing}, up to the requirement that $o_i\geq 0$. However, interchanging
  $e_1$ with $-e_1$, if necessary, this condition can also be ensured; see (\ref{eq:o to o prime}). 
\end{proof}

\begin{lemma}\label{lemma:frame LRS case}
  Let $\mfI_\infty=(G,\msH,\msK,\Phi_{0},\Phi_{1})\in\mS$. Then $\mfI_\infty$ are LRS if and only if there is a frame $\{e_i\}$ of
  $\mfg$ satisfying the conditions of (\ref{eq:canonicalbasisBianchiA}) with $p_2=p_3$, $o_2=o_3$, signs as in Table~\ref{table:bianchiA}
  and either $p_1\neq p_2$ or $o_1\neq o_2$. 

  Assume $\mfI_\infty$ to be LRS. Let $\{e_i'\}$ be a second basis of $\mfg$ satisfying the conditions of
  (\ref{eq:canonicalbasisBianchiA}) with $p_i$ and $o_i$ replaced by $p_i'$ and $o_i'$ respectively, $p_2'=p_3'$, $o_2'=o_3'$ and signs as in
  Table~\ref{table:bianchiA}. Then, if $e_i=A_i^{\phantom{i}j}e_j'$, $A=\mathrm{diag}(\pm 1,A_2)$, where $A_2$ is an orthogonal
  $2\times 2$-matrix. Moreover, for all Bianchi types except I, $\det A$ has to equal $1$. 
\end{lemma}
\begin{proof}
  The first statement of the lemma is an immediate consequence of Lemma~\ref{lemma:sfB char per symm VIIz sing} (up to relabeling of the frame and
  interchanging $e_1$ with $-e_1$, if necessary).  

  Next, let $\{e_i\}$ and $\{e_i'\}$ be frames as in the statement of the theorem. If $p_1\neq p_2$, then $p_1'$ has to differ from $p_2'$, and $p_1=p_1'$.
  This means that $e_1$ and $e_1'$ are unit eigenvectors corresponding to the same distinct eigenvalue, so that $e_1'=\pm e_1$. The argument in case $o_1\neq o_2$
  is similar due to (\ref{eq:o to o prime}). Thus $A=\mathrm{diag}(\pm 1,A_2)$, where $A_2$ is an orthogonal $2\times 2$-matrix. In order for the
  signs in Table~\ref{table:bianchiA} to be respected, $\det A$ has to equal $1$, except for Bianchi type I. 
\end{proof}

\begin{lemma}\label{lemma:frame general case}
  Let $\mfI_\infty=(G,\msH,\msK,\Phi_{0},\Phi_{1})\in\mS$. Assume that $\mfI_\infty$ are neither isotropic nor LRS. Then there is a frame $\{e_i\}$ of
  $\mfg$ satisfying the conditions of (\ref{eq:canonicalbasisBianchiA}), with signs as in Table~\ref{table:bianchiA}. Moreover, the following Bianchi
  type specific conditions hold:
  \begin{itemize}
  \item In Bianchi type I, $p_1>p_2>p_3$.
  \item In Bianchi type II, $p_2>p_3$.
  \item In Bianchi type VI${}_0$, $p_2\geq p_3$ and $o_2\geq -o_3$ in case $p_2=p_3$.
  \item In Bianchi type VII${}_0$ and VIII, $p_2\geq p_3$ and $o_2>o_3$ in case $p_2=p_3$.
  \item In Bianchi type IX, $p_1>p_2>p_3$ if the $p_i$ are all distinct; $p_1\neq p_2=p_3$ and $o_2>o_3$ if two $p_i$'s coincide
    and are different from the third; and $o_1>o_2>o_3$ if $p_1=p_2=p_3$.
  \end{itemize}
  If $\{e_i'\}$ is a frame of $\mfg$ satisfying the same conditions with $p_i$ and $o_i$ replaced by $p_i'$ and $o_i'$ respectively and
  $e_i=A_i^{\phantom{i}j}e_j'$, then $A$ is a diagonal element of $\mathrm{O}(3)$ in the case of Bianchi type I and a diagonal element of
  $\mathrm{SO}(3)$ for the other Bianchi types, with one exception: for Bianchi type VI${}_0$ data such that $p_2=p_3$ and $o_2=-o_3$,
  replacing $e_2$ with $e_3$ is also allowed.
\end{lemma}
\begin{proof}
  We first verify that $\{e_i\}$ can be chosen so that the Bianchi specific conditions are satisfied. In the case of Bianchi type I, we can
  assume the $p_i$ to all be distinct, since the data would otherwise be isotropic or LRS. Moreover, the condition $p_1>p_2>p_3$ can be achieved
  by permuting
  the basis vectors. In the case of Bianchi II, $o_2=o_3=0$. We can thus assume $p_2$ and $p_3$ to be distinct, since the data would otherwise
  be LRS. If $p_2>p_3$, nothing needs to be done. If $p_2<p_3$, we permute $e_2$ and $e_3$ while, at the same time, mapping $e_1$ to $-e_1$. This
  yields the desired frame. The argument in the case of Bianchi types VI${}_0$, VII${}_0$ and VIII is similar. Consider Bianchi type IX. If all
  the $p_i$ are distinct, first permuting the frame vectors appropriately and then flipping the sign of one vector, if necessary, yields
  $p_1>p_2>p_3$. If two $p_i$ are equal and different from the third, we can similarly achieve $p_1\neq p_2=p_3$. Ensuring $o_2>o_3$ is then
  achieved by an argument similar to the one in the case of Bianchi type II. If all the $p_i$ coincide, then all the $o_i$ are distinct (since
  the data would otherwise be isotropic or LRS). We can then assume $o_1>o_2>o_3$ by an argument similar to the case that all the $p_i$ are distinct.

  Assume now that we have a second frame $\{e_i'\}$ with the same properties. In the case of Bianchi type I, since the $p_i$ and $p_i'$ are
  distinct and ordered, we have to have $p_i=p_i'$ for all $i$, so that $e_i=\pm e_i'$. This means that $A$ is a diagonal element of
  $\mathrm{O}(3)$. In the case of Bianchi type II, since $o_1>0$ and $o_2=o_3=0$ (and similarly for $o_i'$), $A=\mathrm{diag}(\pm 1,A_2)$,
  where $A_2$ is an orthogonal $2\times 2$-matrix. Since $p_2>p_3$ (and similarly for the $p_i'$), we, additionally have to have $p_B=p_B'$
  for $B=2,3$, $e_2'=\pm e_2$ and $e_3'=\pm e_3$. Thus $A$ is a diagonal element of $\mathrm{O}(3)$. However, since $o_1$ and $o_1'$ are both
  strictly positive, we have to have $A\in\mathrm{SO}(3)$. The arguments in the remaining cases are similar. The lemma follows. 
\end{proof}

Next, we prove Theorem~\ref{thm:dataonsingtosolution} and Remark~\ref{remark:form of dev of dos}. 

\begin{proof}[Proof of Theorem~\ref{thm:dataonsingtosolution} and Remark~\ref{remark:form of dev of dos}]
  Let $p_i$ denote the eigenvalues of $\msK$ and choose $\sigma_{+}$, $\sigma_{-}$ and $\sigma_i$ so that (\ref{eq:pisitosigmapm}) holds. Then
  \begin{equation}\label{eq:sigmapmPhioneascon}
    \sigma_{+}^{2}+\sigma_{-}^{2}+\bPhi_{1}^{2}/6=1,
  \end{equation}
  where the relation between the $\Phi_{1}$ appearing in the statement of the theorem and $\bPhi_{1}$ is $\bPhi_{1}:=3\Phi_{1}$; this relation
  holds throughout the present proof. Next, let $\{e_i\}$ be a frame satisfying the conditions of (\ref{eq:canonicalbasisBianchiA}). 
  Assume the frame to be as described in one of Lemmas~\ref{lemma:frame isotropic setting}; \ref{lemma:frame LRS case}; and \ref{lemma:frame general case}
  according to whether the data are isotropic; LRS; or neither isotropic nor LRS, respectively. Note, in particular,
  that the $o_i$ have signs as in Table~\ref{table:bianchiA}.

  \textbf{Uniqueness.} Assume that there is a development $(M,g,\phi)=\md[V](\mfI)$ as in the statement of the theorem. In particular $M:=G\times J$,
  where $J$ is an open interval, and $g$ can be written
  \begin{equation}\label{eq:g unique d on sing}
    g=-dt\otimes dt+\textstyle{\sum}_{i=1}^{3}a_{i}^{2}(t)\eta^{i}\otimes\eta^{i},
  \end{equation}
  where $a_i\in C^{\infty}(J,(0,\infty))$, $i=1,2,3$, $\{\eta^{i}\}$ is the dual basis of a basis $\{\fbar_{i}\}$ of the Lie algebra
  $\mathfrak{g}$ of $G$ with associated commutator matrix $\bn=\rodiag(\bn_1,\bn_2,\bn_3)$, and the $\bn_i$ have signs as indicated in
  Table~\ref{table:bianchiA}. This statement follows from the proof of Proposition~\ref{prop:unique max dev}; see, in particular, 
  (\ref{eq:for the ai}) and the adjacent text. We also assume that $\msH(\fbar_i,\fbar_i)=1$ (no summation) for all $i$. 
  Note that replacing $\fbar_i$ with $-\fbar_i$ leaves the metric invariant. However, with the exception of Bianchi types I and VI${}_0$, we have to
  flip two basis vectors in this way in order for the sign convention in Table~\ref{table:bianchiA} to be respected; in the case of Bianchi type
  VI${}_0$, interchanging $\fbar_2$ and $\fbar_3$ respects the signs in Table~\ref{table:bianchiA}.  By assumption, $(M,g,\phi)$ has a
  crushing singularity at the left endpoint of $J$, say $t_-$. Fix $t_0>t_-$ such that $\theta(t)>0$ on $(t_-,t_{0})$. On this
  interval, the expansion normalised Weingarten map $\mK$ is well defined, and the $\fbar_{k}$ are eigenvectors. 

  Let $\theta_{i}:=\dot{a}_{i}/a_{i}$ and $\ell_i:=\theta_i/\theta$, so that $\mK \fbar_{i}=\ell_i\fbar_{i}$ (no summation). Combining this observation with
  the second equality in (\ref{eq:canonicalbasisBianchiA}) and the fact that $\mK$ converges to $\msK$ yields the conclusion that $\ell_i$ converges
  to an element in $\{p_1,p_2,p_3\}$ as $t\downarrow t_-$, say $\bp_i$. Moreover, $\fbar_i$ is an eigenvector of $\msK$ corresponding to $\bp_i$. Note
  that $\{\bp_1,\bp_2,\bp_3\}=\{p_1,p_2,p_3\}$. Next, note that, by assumption,
  \begin{equation}\label{eq:hgdef limit}
    \mH:=\bge(\theta^{\mK}\cdot,\theta^\mK\cdot)\rightarrow\msH.
  \end{equation}
  This means, in particular, that $\theta^{\ell_i}a_i$ converges to $1$ (recall that the $\fbar_i$ are unit vectors with respect to $\msH$). Let
  $\hf_i:=a_i^{-1}\theta^{-\ell_i}\fbar_i$ (no summation). Then $\{\hf_i\}$ is an orthonormal frame with respect to $\mH$ introduced in
  (\ref{eq:hgdef limit}). Since $\mH$ converges to $\msH$ and $a_i\theta^{\ell_i}$ converges to $1$, it follows that $\{\fbar_i\}$ is an
  orthonormal frame with respect to $\msH$. This means that $\{\fbar_i\}$ satisfies all the conditions of (\ref{eq:canonicalbasisBianchiA}),
  but with $o_i$ and $p_i$ replaced by $\bn_i$ and $\bp_i$ respectively. Moreover, the $\bn_i$ have signs as in Table~\ref{table:bianchiA}.
  If we let $\bn:=\mathrm{diag}(\bn_1,\bn_2,\bn_3)$, then there is an $A\in \mathrm{O}(3)$ such that $e_i=A_{i}^{\phantom{i}j}\fbar_j$ and
  (\ref{eq:o to o prime}) hold with $o'$ replaced by $\bn$. This means, in particular, that unless $G$ is of Bianchi type I or VI${}_0$,
  $A\in\mathrm{SO}(3)$, since the $o_i$ and the $\bn_i$ have signs as in Table~\ref{table:bianchiA}.

  If $G$ is of Bianchi type I, the transformation law does not say anything, since $o=\bn=0$ in that case. However, we can
  permute the frame so that $\bp_i=p_i$. In the case of Bianchi type IX, we can similarly permute the frame (and, in the end, replace $\fbar_1$ by
  $-\fbar_1$, if necessary) in order to ensure $\bp_i=p_i$, while respecting $\bn_i>0$. For the remaining Bianchi types, it is clear that
  $e_1$ is distinguished by the requirement that the $\bn_i$ and $o_i$ have signs as in Table~\ref{table:bianchiA}. This means that $\fbar_1=\pm e_1$
  and $\bp_1=p_1$. Finally, interchanging $\fbar_2$ and $\fbar_3$, if necessary, and changing $\fbar_1$ to $-\fbar_1$, if necessary, in the case of
  Bianchi types II-VIII, it can be ensured that $\bp_B=p_B$ for $B=2,3$. Finally, if $p_2=p_3$, we can, additionally, ensure that $\bn_2\geq\bn_3$
  for Bianchi types II, VII${}_0$ and VIII. To summarise, we can assume that $\{\fbar_i\}$ satisfies all the conditions
  of (\ref{eq:canonicalbasisBianchiA}), but with $o_i$ and $p_i$ replaced by $\bn_i$ and $\bp_i$ respectively. Moreover, we can assume that $\bp_i=p_i$.
  Finally, we can assume the $\bn_i$ have signs as in Table~\ref{table:bianchiA}.

  In the case of Bianchi type VI${}_0$, we have to have $\fbar_1=\pm e_1$ by the above argument. If $p_2\neq p_3$, then we have to have $\fbar_2=\pm e_2$
  and $\fbar_3=\pm e_3$. This means that $o_i=\bn_i$ for all $i$. If $p_2=p_3$, we can interchange $\fbar_2$ and $\fbar_3$, if necessary, to ensure that
  $\bn_2\geq -\bn_3$. If $\bn_2>-\bn_3$, we then have to have $\fbar_i=\pm e_i$ for all $i$. Thus $o_i=\bn_i$ for all $i$. If $\bn_2=-\bn_3$, then we also
  have $o_i=\bn_i$ for all $i$. Moreover, by interchanging $\fbar_2$ and $\fbar_3$, if necessary, we can assume that $\fbar_i=\pm e_i$. To summarise, in
  the case of Bianchi types I and VI${}_0$, we can always assume $o_i=\bn_i$ for all $i$. Moreover, in the case of Bianchi type VI${}_0$, we can assume
  that $\fbar_i=\pm e_i$ for all $i$.   

  \textit{Isotropic setting.} In case the data on the singularity are isotropic, then $\msK=\mathrm{Id}/3$, so that $p_i=1/3$, irrespective of the
  choice of frame. Moreover, the $o_i$ are all equal; see Lemma~\ref{lemma:frame isotropic setting}. Due to the transformation law
  (\ref{eq:o to o prime}) and the requirement that the $o_i$ and the $\bn_i$ both have signs as in Table~\ref{table:bianchiA}, it follows that
  $o_i=\bn_i$ for all $i$. To summarise, $\bp_i=p_i$ and $\bn_i=o_i$ for all $i$.

  \textit{LRS setting.} Assume the data to be LRS. Then $\{e_i\}$ is such that $o_2=o_3$ and $p_2=p_3$, but either $p_1\neq p_2$ or $o_1\neq o_2$; see
  Lemma~\ref{lemma:frame LRS case}. Since $\fbar_i$ is an eigenvector of $\msK$ corresponding to the eigenvalue $p_i$, it is clear that if $p_1\neq p_2$,
  then $\fbar_1=\pm e_1$. Similarly, if $o_1\neq o_2$ and $\mfT\neq\mrIX$, then $\fbar_1=\pm e_1$ due to the
  transformation law (\ref{eq:o to o prime}) and the fact that the $o_i$ and the $\bn_i$ have signs as in Table~\ref{table:bianchiA}. Thus
  $\bn_1=o_1$. Again, due to the transformation law (\ref{eq:o to o prime}) and the fact that the $o_i$ and the $\bn_i$ have signs as in
  Table~\ref{table:bianchiA}, we also conclude that $\bn_2=\bn_3=o_2$. If $\mfT=\mrIX$ and $p_1=p_2=p_3$, then we can permute the frame (and replace
  $\fbar_1$ by $-\fbar_1$, if necessary) in order to ensure that $\bn_i=o_i$. To summarise, $\bp_i=p_i$ and $\bn_i=o_i$. Due to
  Lemma~\ref{lemma:frame LRS case} we conclude that $A$ must take the form $A=\mathrm{diag}(\pm 1,A_2)$, where $A_2$ is an orthogonal matrix. Moreover,
  $\det A=1$ for all Bianchi types except I.

  \textit{Data that are neither isotropic nor LRS.} In the case of Bianchi types I and VI${}_0$, we already know that $\bp_i=p_i$ and that
  $\bn_i=o_i$. Moreover, $\fbar_i=\pm e_i$ (in the case of Bianchi type VI${}_0$, we already know this to be true, and in the case of Bianchi
  type I it follows from the fact that $p_1>p_2>p_3$). 
  In the case of Bianchi types II, VII${}_0$ and VIII, we know that $\fbar_1=\pm e_1$. If $p_2\neq p_3$, we have to have $\fbar_B=\pm e_B$
  for $B\in \{2,3\}$. If $p_2=p_3$, we know that $o_2>o_3$ and $\bn_2>\bn_3$, so that $\fbar_B=\pm e_B$ for $B\in \{2,3\}$. We thus have $\bp_i=p_i$,
  $\bn_i=o_i$ and $\fbar_i=\pm e_i$ for these Bianchi types. Consider Bianchi type IX. If all the $p_i$ are distinct, then $\fbar_i=\pm e_i$ for all $i$.
  This means that $\bn_i=o_i$ and $\bp_i=p_i$ for all $i$. Assume that two of the $p_i$ are equal and different from the third. Then $p_1\neq p_2=p_3$,
  $\fbar_1=\pm e_1$ and $\bn_1=o_1$. Since $o_2> o_3$, we can exchange $\fbar_2$ and $\fbar_3$ and replace $\fbar_1$ with $-\fbar_1$ (if necessary) so
  that $\bn_2>\bn_3$. This means that $\fbar_A=\pm e_A$ for $A\in \{2,3\}$. To summarise, $\bp_i=p_i$, $\bn_i=n_i$ and $\fbar_i=\pm e_i$. Finally, if
  all the $p_i$ are equal, permuting the elements of the frame $\{\fbar_i\}$ and interchanging $\fbar_1$ with $-\fbar_1$, if necessary, it can be
  ensured that $\bn_1>\bn_2>\bn_3$. This means that $\fbar_i=\pm e_i$, so that $\bn_i=o_i$. To summarise the case that the data are neither isotropic
  nor LRS, we can (by permuting the elements of $\{\fbar_i\}$) ensure that $\fbar_i=\pm e_i$, that $p_i=\bp_i$ and that $\bn_i=o_i$. 
  
  To summarise: in what follows, we can always assume that $\bn_i=o_i$ and $\bp_i=p_i$ for all $i$.
  
  \textit{Asymptotics.} Defining $n_i$ by (\ref{eq:chco}), 
  \[
  \frac{n_1}{\theta}=\frac{a_1\theta^{\ell_1}}{a_2\theta^{\ell_2}a_3\theta^{\ell_3}}\bn_1 \theta^{\ell_2+\ell_3-\ell_1-1},
  \]
  and similarly for $n_2$ and $n_3$. Since the $\ell_i$ sum to one, the last factor equals $\theta^{-2\ell_1}$. In case $p_1>0$, it follows that
  $n_1/\theta$ converges to zero; recall that $\theta^{\ell_i}a_i$ converges to $1$. In particular, if all the $p_i$ are strictly positive, all
  the $n_i/\theta$ converge to zero. Assume that $p_3\leq 0$. Due to Definition~\ref{def:ndvacidonbbssh}, there are two possibilities.

  \textit{The first possibility} is that the eigenvalues of $\msK$ are all $<1$. Then, since the eigenvalues sum up to one, it is clear that
  $p_1,p_2>0$. This means that $e_3$ is uniquely determined (up to a constant multiple). Since
  \begin{equation}\label{eq:f one eigenvector}
    \msK \fbar_3=\lim_{t\downarrow t_-}\mK \fbar_3=\lim_{t\downarrow t_-}\ell_3 \fbar_3=p_3\fbar_3,
  \end{equation}
  it is clear that $\fbar_3=\pm e_3$. Next, by assumption, the vector subspace of $\mfg$ spanned by $e_1$ and $e_2$ is a subalgebra. This means that
  $e_3$ is perpendicular to $[e_1,e_2]$, so that $\fbar_3$ is perpendicular to $[e_1,e_2]$. Since $\fbar_3$ is also perpendicular to $\fbar_1$ and
  $\fbar_2$, it follows that there are constants $\kappa_{A}^B$, $A,B\in \{1,2\}$ such that $\fbar_A=\kappa_A^B e_B$. This means that $[\fbar_1,\fbar_2]$
  is a constant multiple of $[e_1,e_2]$, so that $\fbar_3$ is perpendicular to $[\fbar_1,\fbar_2]=\bn_3 \fbar_3$. Thus $\bn_3=0$ and
  \begin{equation}\label{eq:ni by theta to zero}
    \lim_{t\downarrow t_-}[n_i(t)/\theta(t)]=0
  \end{equation}
  for all $i$, irrespective of whether $p_3\leq 0$ or not. Combining this observation with (\ref{eq:thetad}) and the assumptions, which imply
  \begin{equation}\label{eq:sigma sq pt sq limit}
    \frac{3}{2}\theta^{-2}\sigma^{ij}\sigma_{ij}+\frac{3}{2}\theta^{-2}\phi_t^2\rightarrow 1,
  \end{equation}
  it follows that $\theta_t/\theta^2\rightarrow -1$. This means, in particular, that $\theta$ blows up in finite time to the past, so that we can
  assume $t_-=0$. Moreover, it means that $\theta$ is not integrable to the past. Introducing a time coordinate $\tau$ according to (\ref{eq:dtdtau}),
  it is thus clear that the range of $\tau$ is of the form $(-\infty,\tau_0)$. From now on, a prime denotes differentiation with respect to $\tau$.

  \textit{The second possibility} is that $1$ is an eigenvalue of $\msK$. We then assume $\msK e_1=e_1$, so that $n_1/\theta$ converges to zero as
  above. Moreover, $n_1n_i/\theta^2$ converge to zero for $i=2,3$. Next, note that, by the assumptions, (\ref{eq:sigma sq pt sq limit}) holds. 
  Combining this observation with (\ref{eq:hamconfin}) and the fact that $n_1n_i/\theta^2\rightarrow 0$
  for all $i$, it follows that $V\circ\phi/\theta^2\rightarrow 0$ and $(n_2-n_3)/\theta\rightarrow 0$ (\textit{note that this is the only point at
    which we use the assumption that} $V\geq 0$). Combining these observations with
  (\ref{eq:thetad}) and
  (\ref{eq:sigma sq pt sq limit}), it is clear that $\theta_t/\theta^2\rightarrow -1$. Thus $t_-=0$ and the range of $\tau$ is of the form
  $(-\infty,\tau_0)$, as before.

  Next, define $\Sigma_{i}$, $\Sigma_{\pm}$ and $N_i$ by
  \begin{equation}\label{eq:Sigmai Spm def}
    \Sigma_{i}:=\theta_{i}/\theta-1/3,\ \ \
    \Sigma_{+}:=\tfrac{3}{2}(\Sigma_{2}+\Sigma_{3})=-\tfrac{3}{2}\Sigma_{1},\ \ \
    \Sigma_{-}:=\tfrac{\sqrt{3}}{2}(\Sigma_{2}-\Sigma_{3}),\ \ \
    N_{i}:=\tfrac{n_{i}}{\theta}.
  \end{equation}
  Finally, define $\phi_{1}=\phi'$ and $\phi_{0}=\phi$. Then the arguments presented in Subsections~\ref{ssection:EquationsBclassAdev} and
  \ref{ssection:whsuform} yield the conclusion that (\ref{eq:thetaprime})--(\ref{eq:constraint}) hold. In case $1$ is an eigenvalue of $\msK$,
  we can assume that $(\Sigma_+,\Sigma_-)$ converges to $(-1,0)$. Moreover, in this case, we assume $V\geq 0$. Note also that, due to
  (\ref{eq:Xaltversion}) and the above observations, $X$ converges to $1$ in this setting. Thus either the conditions of
  Lemma~\ref{lemma:tautimecoord} or the conditions of Lemma~\ref{lemma:tautimecoord BIX} are satisfied. Thus Proposition~\ref{prop:limitcharsp}
  applies, so that $\Sm=0$ and
  $\Nt=\Nth$. This means that $\sigma_{22}=\sigma_{33}$, so that $a_2/a_3$ is constant due to the first equality in (\ref{eq:for the ai}).
  On the other hand, we also know that $\theta^{\ell_A}a_A\rightarrow 1$ for $A\in\{2,3\}$ and that $\ell_2=\ell_3$. This means that $a_2/a_3$
  converges to $1$. Thus $a_2=a_3$. Combining this observation with
  (\ref{eq:chco}) and the fact that $\Nt=\Nth$, it follows that $\bn_2=\bn_3$. This means that $o_2=o_3$ and $p_2=p_3=0$, so that $\mfI_\infty$
  are LRS. Note that this proves the necessity of the last requirement in Definition~\ref{def:ndvacidonbbssh} in case $1$ is an eigenvalue of
  $\msK$. In particular, the last statement of Theorem~\ref{thm:dataonsingtosolution} follows.

  If $1$ is an eigenvalue of $\msK$, then we can, by the above, assume that $p_1=1$, $p_2=p_3=0$, $a_2=a_3$, $\Nt=\Nth$ and $\Sm=0$, and we do so in
  what follows. If the eigenvalues of $\msK$ are all $<1$, then, by the argument presented below (\ref{eq:f one eigenvector}), it follows that if
  $p_i\leq 0$, then $N_i=0$. Moreover, if $p_3\leq 0$, then $p_1,p_2>0$. If the eigenvalues of $\msK$ are $<1$ and there is a $p_i\leq 0$, we, from
  now on, assume that it is $p_3$. In this case, we can thus always assume $p_1,p_2>0$, so that $f_i(2,\sigma_+,\sigma_-)>0$ for $i=1,2$; see
  (\ref{eq:pisitosigmapm}). If $p_{3}>0$, we can, additionally, assume that $f_3(2,\sigma_{+},\sigma_-)>0$; see (\ref{eq:pisitosigmapm}). Keeping
  (\ref{eq:PsiOmegadef}) in mind, the Hamiltonian constraint (\ref{eq:constraint}) can be written
  \begin{equation}\label{eq:HamConphioVexp}
    \Sp^2+\Sm^2+\tfrac{1}{6}\phi_{1}^{2}+\tfrac{3V\circ\phi_{0}}{\theta^{2}}+\tfrac{3}{4}[\No^2+\Nt^2+\Nth^2-2(\No\Nt+\Nt\Nth+\No\Nth)]=1.
  \end{equation}
  By assumption $\Sigma_{\pm}$ converges to $\sigma_{\pm}$ and $\phi_{1}$ converges to $\bPhi_{1}$. Moreover, the sum of the first three terms on
  the left hand side converges to $1$. Next,  either all the $N_i$ tend to zero as $\tau\rightarrow-\infty$, or $\Nt=\Nth$, $\Sm=0$, $\No$ tends
  to zero and $\No\Nt$ tends to zero as $\tau\rightarrow-\infty$. Combining these observations with (\ref{eq:HamConphioVexp}), it follows that
  $\Omega$ tends to zero. Moreover, keeping the definition of $q$, i.e. (\ref{eq:qdef}), in mind, it follows that
  $q$ converges to $2$. Combining this fact with the above observations and (\ref{eq:nip}) yields the conclusion that if all the eigenvalues of
  $\msK$ are $<1$, then the $N_{i}$ converge to zero exponentially. Moreover, if one of the eigenvalues of $\msK$ equals $1$, then $\No$ and
  $\No\Nt$ converge to zero exponentially. Next, note that since
  $|\phi'|\leq \sqrt{6}+\dots$ (where the dots represent exponentially decaying terms), see (\ref{eq:HamConphioVexp}), there exists, for every $\eta>0$, a
  constant $C_{\eta}$ such that
  \begin{align*}
    |\phi(\tau)| \leq & \sqrt{6}|\tau|+C_{\eta},\ \ \
    \ln\theta(\tau) \geq -(3-\eta)\tau-C_{\eta}
  \end{align*}
  for all $\tau\leq 0$, where we appealed to (\ref{eq:thetaprime}).
  Combining this observation with the assumption that $V\in\mfP^2_{\a_V}$, it follows that $\Omega$ decays exponentially. Note that 
  polynomials in the $N_{i}$ appearing in (\ref{seq:SpSmev}) and (\ref{eq:altqdef}) decay exponentially. Since $\Omega$ converges to zero exponentially,
  (\ref{eq:altqdef}) implies that $q-2$ converges to zero exponentially. Since $S_\pm$ also converge to zero exponentially, we also conclude that
  $\Sigma_\pm-\sigma_\pm$ decay exponentially. Moreover,
  \begin{equation}\label{eq:thetaasfoftau}
    \theta(\tau)=\exp\left(\textstyle{\int}_{\tau}^{\tau_{0}}[1+q(s)]ds\right)\theta(\tau_{0})=C_{0}e^{-3\tau}[1+O(e^{\e\tau})]
  \end{equation}
  for some constants $C_{0},\e>0$, where $\tau_{0}$ is an element of the existence interval of the solution. At this point, we can fix the translation
  ambiguity in the definition of $\tau$ in order to ensure that
  \begin{equation}\label{eq:normthetaas}
    \theta(\tau)=e^{-3\tau}[1+O(e^{\e\tau})].
  \end{equation}
  Since $t\downarrow 0$ corresponds to $\tau\rightarrow-\infty$, we conclude that
  \begin{equation}\label{eq:ttauas}
    t(\tau)=\textstyle{\int}_{-\infty}^{\tau}3e^{3s}[1+O(e^{\e s})]ds=e^{3\tau}[1+O(e^{\e\tau})].
  \end{equation}
  In particular,
  \begin{equation}\label{eq:tthetalimit}
    t\theta=1+O(e^{\e\tau}).
  \end{equation}
  Next, let $m_{i}=\ln |\bn_{i}|$ if $\bn_{i}\neq 0$ and $m_{i}=0$ otherwise. Let $\e_i=\mathrm{sign}\bn_i$ if $\bn_i\neq 0$ and $\e_i=0$ otherwise.
  Note, for future reference, that $m_i$ and $\e_i$ are determined by the data on the singularity and the choice of frame $\{e_i\}$. In particular,
  they do not depend on the development (\ref{eq:g unique d on sing}). Define $\mu_{i}$ by
  \begin{equation}\label{eq:mukdef}
    \mu_{i}:=\ln\frac{a_{i}}{\theta a_{j}a_{k}}+m_{i},
  \end{equation}
  where $\{i,j,k\}=\{1,2,3\}$. Then $N_i=\e_i e^{\mu_i}$. Moreover,
  \[
    \mu_i'=q+6\Sigma_i.
  \]
  Next, compute
  \begin{equation}\label{eq:mui as exp}
    \mu_{i}=\ln\frac{\theta^{\ell_i}a_{i}}{\theta^{\ell_j}a_{j}\theta^{\ell_k}a_{k}}-2\ell_i\ln\theta+m_{i}
    =-\frac{1}{3}(2+6\Sigma_i)\ln\theta+\ln\frac{\theta^{\ell_i}a_{i}}{\theta^{\ell_j}a_{j}\theta^{\ell_k}a_{k}}+m_{i};
  \end{equation}
  note that $\ell_i=\Sigma_i+1/3$ due to (\ref{eq:Sigmai Spm def}) and that $\ell_1+\ell_2+\ell_3=1$. Due to (\ref{eq:normthetaas}), $\ln\theta+3\tau$
  decays exponentially. Finally, $\theta^{\ell_i}a_{i}$ converges to $1$. Combining these observations with (\ref{eq:mui as exp}) yields
  \begin{equation}\label{eq:mui asympt unique}
    \mu_i(\tau)=6(\Sigma_i(\tau)+1/3)\tau+m_i+\dots,
  \end{equation}
  where the dots signify terms that converge to zero as $\tau\rightarrow-\infty$. In particular, defining $\nu_i$ by (\ref{eq:Niitonui}), it follows that
  $\nu_i(\tau)\rightarrow 0$ as $\tau\rightarrow-\infty$. Next, $s_i:=\Sigma_i-\sigma_i\rightarrow 0$ as $\tau\rightarrow-\infty$.
  Defining $\varkappa$ by (\ref{eq:kappaidos}) with $\varkappa_\infty=0$, it is clear that $\varkappa(\tau)\rightarrow 0$ as $\tau\rightarrow-\infty$.
  Next, note that $\phi_1$ is bounded due to (\ref{eq:HamConphioVexp}) and the fact that the polynomial in the $N_i$ converges to zero exponentially.
  Since $2-q$ converges
  to zero exponentially, this means that the first term on the right hand side of (\ref{eq:phioneev}) converges to zero exponentially. By an
  argument similar to the proof of the fact that $\Omega$ converges to zero exponentially, it is clear that the second term on the right hand side
  of (\ref{eq:phioneev}) converges to zero exponentially. This means that $\phi_1$ converges to its limit $\bPhi_1$ exponentially. In particular,
  $\psi_1\rightarrow 0$ as $\tau\rightarrow-\infty$, where $\psi_1$ is defined by (\ref{eq:psi one def}) with $\Phi_1$ replaced by $\bPhi_1$. Moreover,
  integrating the exponential convergence of $\phi_1$ to $\bPhi_1$ yields the conclusion that $\phi_0-\bPhi_1\tau-\bPhi_0$ converges to zero
  exponentially for some $\bPhi_0$. However, since $\bPhi_1=3\Phi_1$ and $\theta+3\tau$ decays exponentially, it follows that
  $\phi_0+\Phi_1\ln\theta-\bPhi_0$ converges to zero. Since $\phi_t/\theta-\Phi_1$ decays exponentially, this means that $\bPhi_0=\Phi_0$. Moreover,
  $\psi_0$ defined by (\ref{eq:phiidos}) converges to zero. Thus Proposition~\ref{prop:variableexandunique} applies, so that $\nu_i$, $s_\pm$,
  $\varkappa$, $\psi_0$ and $\psi_1$ are thus uniquely determined. Note also, for future
  reference, that $\Pi$ appearing in (\ref{eq:Pi def}) is determined by the data on the singularity and the choice of frame $\{e_i\}$. In particular,
  it does not depend on the development converging to the data on the singularity. Next, note that since $\varkappa$ is uniquely determined, $\theta$
  is uniquely determined as a function of $\tau$. Combining this observation with (\ref{eq:dtdtau}) and the fact that $\tau\rightarrow-\infty$
  corresponds to $t\rightarrow 0$, it is clear that $t$ is uniquely determined as a function of $\tau$. Next, since the $s_i$ (or, equivalently, the
  $s_\pm$) are uniquely determined, it follows that the $\Sigma_i$ are uniquely determined. Combining this observation with the fact that the
  $\nu_i$ are uniquely determined, it follows that the $\mu_i$ are uniquely determined. Due to (\ref{eq:mukdef}), it follows that the $a_i$ are
  uniquely determined. Thus
  \begin{equation}\label{eq:gdiagonalSH}
    g=-dt\otimes dt+\textstyle{\sum}_{i}a_{i}^{2}(t)\eta^{i}\otimes \eta^{i},
  \end{equation}
  where $\{\eta^i\}$ is the frame dual to $\{\fbar_i\}$ and the $a_i$ are uniquely determined by $\Pi$ appearing in (\ref{eq:Pi def}), which, in its turn,
  is determined by the data on the singularity and the choice of frame $\{e_i\}$. Similarly, the scalar field is uniquely determined by $\Pi$. 
    
  Next, we wish to prove that if the data at the singularity are isotropic or LRS, then so is the development. Assume, to begin with, that the
  data on the singularity are isotropic. Then $\sigma_i=0$, all the $\e_i$ are equal and all the $m_i$ are equal. Thus all the $N_i$ equal and
  all the $\Sigma_i$ equal (since $\eta_\lambda(\Pi)=\Pi$ for any $\lambda\in S_3$, using the notation of Proposition~\ref{prop:variableexandunique}, the
  $N_i$ and $\Sigma_i$ are invariant under permutations). Since the $\Sigma_i$ sum up to $0$, this means that all the $\Sigma_i$ vanish. Next, due to
  Proposition~\ref{prop:variableexandunique}, all the $\nu_i$ equal. Thus, all the $\mu_i-m_i$ are identically equal for the entire solution.
  Recalling (\ref{eq:mukdef}), this means that $a_i/(a_ja_k)$, where $\{i,j,k\}=\{1,2,3\}$, is independent of $i$. This means that $a_1=a_2=a_3$.
  Letting $a(t):=a_1(t)$, it follows that
  \begin{equation}\label{eq:g iso as form}
    g=-dt\otimes dt+a^2(t)\msH.
  \end{equation}
  This proves uniqueness in the isotropic setting, since $a$ is uniquely determined by the data on the singularity. In addition, the above observations
  justify the first statement of Remark~\ref{remark:form of dev of dos}.
  
  In the LRS setting, $\bn_2=\bn_3$ (so that $\e_2=\e_3$, $m_2=m_3$) and $\sigma_2=\sigma_3$. By the uniqueness statement of
  Proposition~\ref{prop:variableexandunique}, it follows that $\nu_2=\nu_3$ and $\Sigma_2=\Sigma_3$ (so that $\Sm=0$). This means, in particular, that
  $\mu_2-m_2=\mu_3-m_3$, so that $a_2=a_3$. Thus the metric takes the form
  \begin{equation}\label{eq:g form LRS}
    g=-dt\otimes dt+a_1^2(t)\eta^1\otimes \eta^1+a_2^2(t)[\eta^2\otimes \eta^2+\eta^3\otimes \eta^3].
  \end{equation}
  Since $e_i=A_i^{\phantom{i}j}\fbar_j$, where $A=\mathrm{diag}(\pm 1,A_2)$ and $A_2$ is an orthogonal $2\times 2$-matrix, it follows that
  \begin{equation}\label{eq:g form LRS xi}
    g=-dt\otimes dt+a_1^2(t)\xi^1\otimes \xi^1+a_2^2(t)[\xi^2\otimes \xi^2+\xi^3\otimes \xi^3],
  \end{equation} 
  where $\{\xi^i\}$ is the basis dual to $\{e_i\}$ and the $a_i$ are uniquely determined by the data on the singularity. This proves uniqueness.  
  In addition, the above observations justify the second statement of Remark~\ref{remark:form of dev of dos}.

  Next, consider the case that $G$ is of Bianchi type VI${}_0$. Assume, moreover, that $\bn_2=-\bn_3$ and $p_2=p_3$. Then, due to
  Proposition~\ref{prop:variableexandunique}, we know that $\nu_2=\nu_3$ and that $\Sigma_2=\Sigma_3$; note that $\eta_\lambda^-(\Pi)=\Pi$, if
  $\lambda\in S_3$ is such that $\lambda(1)=1$, $\lambda(2)=3$ and $\lambda(3)=2$. This means that $\mu_2-m_2=\mu_3-m_3$. Due to (\ref{eq:mukdef}),
  this means that $a_2=a_3$. Since we, as noted above, can assume $\fbar_i=\pm e_i$, it follows that the metric takes the form
  (\ref{eq:g form LRS xi}), where the $a_i$ are uniquely determined by the data on the singularity. This proves uniqueness. In addition, the above
  observations justify the third statement of Remark~\ref{remark:form of dev of dos}. This proves uniqueness. 

  In the remaining cases, the $a_i$ are uniquely determined by the data on the singularity. Moreover, $\fbar_i=\pm e_i$. This means that the
  metric is uniquely determined. 
  
  \textbf{Existence.}
  It remains to prove existence. Fix $\{e_{i}\}$, $p_{i}$, $o_{i}$, $\sigma_i$ and $\sigma_{\pm}$ as at the beginning of the proof. Then, with $\Phi_{0}$ and
  $\Phi_{1}$ as in the statement of the theorem and $\bPhi_1$ defined as at the beginning of the proof, $\sigma_{+}^{2}+\sigma_{-}^{2}+\bPhi_{1}^{2}/6=1$.
  Next, let $\e_{i}$ equal zero if $o_{i}=0$ and
  $\e_{i}=o_{i}/|o_{i}|$ if $o_{i}\neq 0$. Let $\varkappa_{\infty}=0$. Finally, let $m_{k}=\ln |o_{k}|$ if $o_{k}\neq 0$ and $m_{k}=0$ otherwise. Note
  that there are two possibilities: either $p_i>0$ whenever $\e_{i}\neq 0$; or $\sigma_+=-1$, $\sigma_-=0$, $\Phi_1=0$, $m_2=m_3$, $\e_2=\e_3$
  and all the $\e_i$ are non-vanishing. The conditions of Proposition~\ref{prop:variableexandunique} are thus satisfied. Applying
  this proposition yields a unique solution to (\ref{eq:thetaprime})--(\ref{seq:phizphioev}) such that the statements of
  Proposition~\ref{prop:variableexandunique} hold; note, in particular, that $\Nt=\Nth$ and $\Sm=0$ in case $\sigma_+=-1$ and all the $\e_i$ are
  non-zero. Define $\theta_{i}:=\theta(\Sigma_{i}+1/3)$, define the time coordinate $t$ by (\ref{eq:dtdtau}) and the requirement that
  $\tau\rightarrow-\infty$ corresponds to $t\downarrow 0$. Define $a_{i}$, up to a constant, by $\d_{t}\ln a_{i}=\theta_{i}$. Due to
  Proposition~\ref{prop:variableexandunique}, we know that $\theta_{i}/\theta$ converges to $p_{i}$. Due to the assumptions, the polynomial in the
  $N_{i}$ appearing on the right hand side of (\ref{eq:altqdef}) converges to zero exponentially. Moreover, $\Omega$ also converges to zero
  exponentially. Due to (\ref{eq:altqdef}), this means that $q-2$ converges
  to zero exponentially. Combining this observation with (\ref{eq:thetaprime}) and the fact that $\varkappa_{\infty}=0$, it follows that
  (\ref{eq:normthetaas}) holds. Since $t\downarrow 0$ corresponds to $\tau\rightarrow-\infty$, we conclude that (\ref{eq:ttauas}) and
  (\ref{eq:tthetalimit}) hold. Next, since $q-2$ and the $S_\pm$ converge to zero exponentially, it follows that $\Sigma_{i}$ converges exponentially
  to its limit, so that $\theta_{i}/\theta$ converges exponentially to $p_{i}$. Since $\d_{\tau}\ln a_{i}=3\theta_{i}/\theta$, we conclude that
  $\ln a_{i}=3p_{i}\tau+\alpha_{i}+O(e^{\e\tau})$ for some constants $\alpha_{i}$ and $\e>0$. Fix the ambiguity in the definition of $a_{i}$ so that
  $\ln a_{i}=3p_{i}\tau+O(e^{\e\tau})$. This means that $\ln a_{i}=p_{i}\ln t+O(t^{\e})$. At this stage, we can define $\phi$ by $\phi=\phi_{0}$ and
  $g$ by (\ref{eq:gdiagonalSH}) with $\eta^i$ replaced by $\xi^i$, where $\{\xi^i\}$ is the frame dual to $\{e_i\}$. By the assumed asymptotics, it
  follows that $\theta\rightarrow\infty$ as $t\downarrow 0$. The above observations also yield the conclusion that $(M,g,\phi)$ induces the desired
  data on the singularity.  By appealing to the arguments appearing in Subsections~\ref{ssection:EquationsBclassAdev} and \ref{ssection:whsuform}, it
  can be verified that the Einstein non-linear scalar field equations are satisfied; the details are left to the reader. Finally, it remains to prove
  that the development arises from regular initial data in the same symmetry class as the data on the singularity.

  Assume the data on the singularity to be isotropic. Then they are of Bianchi type I or IX and the development is of the same type. By the argument
  presented in the proof of uniqueness, the development is isotropic. Moreover, if the development is isotropic, then the data on the singularity have
  to be isotropic. If $\mfI_\infty\in\mS_{\mrVIz}^{\roper}$, the development arises from initial data in $\mB_{\mrVIz}^{\roper}[V]$ due to the argument presented
  in the proof of uniqueness. Next, assume the data on the singularity to be LRS. Then the development is LRS or isotropic due to the argument presented
  in the proof of uniqueness. It can only be isotropic in the case of Bianchi types I and IX (note that we have excluded isotropic and LRS initial data
  on the singularity in the case of Bianchi type VII${}_0$). Moreover, if it is isotropic, then the initial data on induced on the singularity are
  isotropic, which they are not. Finally, if $\mfs=\rogen$, then the development has to be of the same type, since it would otherwise induce data on
  the singularity with $\mfs\neq\rogen$. The theorem follows. 
\end{proof}

\section[Regularity, map to singularity]{Regularity of the map from developments to data on the singularity}\label{section:Reg map from dev to dos}

Next, we prove Proposition~\ref{prop:dev to data on sing nIX}.

\begin{proof}[Proof of Proposition~\ref{prop:dev to data on sing nIX}]
  Fix $\vartheta\in I_V$, let $\mfI\in {}^{\rosc}\mB_{\mfT}^{\mfs}[V](\vartheta)$ and assume $\md[V](\mfI)$ to induce data on the singularity. Since
  the conditions of Lemma~\ref{lemma:tautimecoord} are satisfied, the development corresponds to a solution to
  (\ref{seq:EFEwrtvar})--(\ref{eq:phiddot}) (see the proof of Proposition~\ref{prop:unique max dev}) which induces a global solution $S$ to
  (\ref{eq:thetaprime})--(\ref{eq:constraint}). Moreover, due to the uniqueness part of the proof of Theorem~\ref{thm:dataonsingtosolution}, it
  follows that $S$ is an admissible convergent solution to (\ref{eq:thetaprime})--(\ref{eq:constraint}); see
  Definition~\ref{def:admconvsol}. In order to justify this statement, note that $\phi_1$ and the $\Sigma_i$ converge due to the assumption that
  the development induces data on the singularity. In addition, (\ref{eq:normthetaas})
  holds, so that $\theta(\tau)\geq-3\tau-C_\theta$ for $\tau\leq 0$.  Next, assume that all the eigenvalues
  of $\msK$ are $<1$. Then, due to the proof of Theorem~\ref{thm:dataonsingtosolution}, all the $N_i$ converge to zero exponentially.
  In fact, if $\bar{\sigma}_i$ is the limit of $\Sigma_i$, then $\bar{\sigma}_i+1/3>0$ in case $N_i\neq 0$; see, again, the proof of
  Theorem~\ref{thm:dataonsingtosolution}. Finally, assume that there is an eigenvalue of $\msK$ which equals $1$. Then, due to the
  proof of Theorem~\ref{thm:dataonsingtosolution}, we can assume that $\Nt=\Nth$, that $\Sm=0$ and that $\Sp\rightarrow -1$. Note
  also that the limit of $\phi_1$ has to equal zero, since $\msK$ could otherwise not have an eigenvalue equal to $1$. If $\mfT=\mrVIII$,
  condition (ii) of Definition~\ref{def:admconvsol} is thus satisfied. Since Bianchi type I
  solutions clearly satisfy (i) of Definition~\ref{def:admconvsol}; since we have excluded isotropic and LRS Bianchi type VII${}_0$;
  and since Bianchi type VI${}_0$ is not consistent with local rotational symmetry, what remains to be considered is LRS Bianchi type
  II with $\No>0$. Since $\Sp$ converges to $-1$, $\No$ converges to zero exponentially and $\bar{\sigma}_1+1/3>0$. Thus condition
  (i) of Definition~\ref{def:admconvsol} is satisfied.
  
  Next, note that, due to Lemma~\ref{lemma:sfR mfT mfs}, Definition~\ref{def:B varthetaz} and Remark~\ref{remark:par iso id fixed mc}, we can identify
  ${}^{\rosc}\mfB_{\mfT}^{\mfs}[V](\vartheta)$ with the quotient $\sfR_\mfT^\mfs(\vartheta)=\sfB_{\mfT}^{\mfs}(\vartheta)/\Gamma_{\mfT}^\mfs$.
  Moreover, due to Lemma~\ref{lemma:sfRmfTmfsvartheta sm mfd}, $\sfR_\mfT^\mfs(\vartheta)$ is a smooth manifold. Next, $\sfF_\sfP^{-1}$, introduced in
  (\ref{eq:sfF sfP def}) defines a diffeomorphism from $\sfB_{\mfT}^{\mfs}(\vartheta)$ to $\sfP_\mfT^\mfs(\ln\vartheta)$;
  cf. Lemma~\ref{lemma:mfB mP rel}; Definitions~\ref{def:BTs GTs}, \ref{def:B varthetaz} and \ref{def:BTs LTs}; and the proof of Lemma~\ref{lemma:sfQ mfT mfs}.
  In addition, note that, given $\mfI\in {}^{\rosc}\mB_{\mfT}^{\mfs}[V](\vartheta)$, there is an
  $\eta\in \sfB_{\mfT}^{\mfs}(\vartheta)$ arising from $\mfI$ and that this correspondence gives rise to the identification between
  ${}^{\rosc}\mfB_{\mfT}^{\mfs}[V](\vartheta)$ and $\sfR_\mfT^\mfs(\vartheta)$; cf. the proof of Lemma~\ref{lemma:sc mfB mfT mfs param} and
  Remark~\ref{remark:eta BTs realised}. Given that $\mfI_0$ and $\eta_0$ are related in this way, let $\xi_0:=\sfF_\sfP^{-1}(\eta_0)$. Then the first
  part of the proof yields the conclusion that $\xi_0\in \sfA_\mfT^\mfs(\ln\vartheta)$ in case $\mD[V](\mfI_0)$ induces data on the singularity; cf.
  Definitions~\ref{def:ma mfT mfs} and \ref{def:sfA fixed varkappa}. Due to Lemma~\ref{lemma:reg of B}, there is thus an open neighbourhood of
  $\xi_0$ in $\sfP_\mfT^\mfs(\ln\vartheta)$, say $U$, contained in $\sfA_\mfT^\mfs(\ln\vartheta)$. Due to Lemma~\ref{lemma:from adm to dos}, it follows
  that if $\mfI\in{}^{\rosc}\mB_{\mfT}^{\mfs}[V](\vartheta)$ and $\eta\in\sfF_\sfP(U)$ arises from $\mfI$, then $\mD[V](\mfI)$ 
  induces data on the singularity. The subset of $\sfB_{\mfT}^{\mfs}(\vartheta)$ that gives rise to data on the singularity, say $W$, is
  therefore open. Since the corresponding subset of ${}^{\rosc}\mfB_{\mfT}^{\mfs}[V](\vartheta)$ is a quotient of $W$, we conclude that
  ${}^{\rosc}_{\roqu}\mfB_{\mfT}^{\mfs}[V](\vartheta)$ is an open subset of ${}^{\rosc}\mfB_{\mfT}^{\mfs}[V](\vartheta)$.
  
  That $\mfR_{\mfT,\mfs}^{\vartheta}$ is a local diffeomorphism follows from Lemma~\ref{lemma:local repr map to sing}; due to
  Lemma~\ref{lemma:BianchiAdevelopment}, $\theta_t<0$ for all $t$ in the developments of interest here, so that the condition $q>-1$ is always
  fulfilled. In order to prove that it is a diffeomorphism onto its image, it is sufficient to prove injectivity. Assume, to this end, that
  $\mfR_{\mfT,\mfs}^{\vartheta}(x_1)=\mfR_{\mfT,\mfs}^{\vartheta}(x_2)$. This means
  that if $[\mfI_i]=x_i$, $i=1,2$, then the data induced on the singularity by $\md[V](\mfI_i)$ are isometric. Denote the isometry between
  the data on the singularity by $\chi$. Due to Proposition~\ref{prop:iso idos to iso dev}, $\chi\times\roId$ is then an isometry between
  the developments $\md[V](\mfI_i)$. In particular, $\chi$ induces an isometry between $\mfI_i$, $i=1,2$, since the mean curvatures of $\mfI_i$
  coincide and $\theta_t<0$ for all $t$. This means that $x_1=x_2$ and the global injectivity of $\mfR_{\mfT,\mfs}^{\vartheta}$ follows.

  Next, assume that (\ref{eq:iota vthonexo eq iota vthtwoextwo}) holds with $\vartheta_2\geq\vartheta_1$. Then, since $\theta\rightarrow\infty$ to
  the past and $\theta_t<0$ for the developments under consideration here, see Lemma~\ref{lemma:BianchiAdevelopment}, a development corresponding
  to the isometry class of developments induced by $x_1$ has a hypersurface of the form $G_1\times\{t\}$ with mean curvature $\vartheta_2$. Call the
  corresponding isometry class of initial data $x_3$. Then $\iota_{\vartheta_2}(x_3)=\iota_{\vartheta_2}(x_2)$, so that $x_2=x_3$ due to
  Lemma~\ref{lemma:dev smo str nisoI nIX}. Thus $\mfR_{\mfT,\mfs}^{\vartheta_1}(x_1)=\mfR_{\mfT,\mfs}^{\vartheta_2}(x_2)$. The surjectivity of
  $\mfR_{\mfT,\mfs}$ follows from Theorem~\ref{thm:dataonsingtosolution}. That ${}^{\rosc}_{\roqu}\mfD_\mfT^\mfs[V]$ is an open subset of ${}^{\rosc}\mfD_\mfT^\mfs[V]$
  and that $\mfR_{\mfT,\mfs}$ is a local diffeomorphism follows from Lemma~\ref{lemma:dev smo str nisoI nIX}. In order to prove injectivity, assume
  that $\mfR_{\mfT,\mfs}(y_1)=\mfR_{\mfT,\mfs}(y_2)$. We can assume that $y_i=\iota_{\vartheta_i}(x_i)$ and that $\vartheta_2\geq \vartheta_1$. Then
  $\mfR_{\mfT,\mfs}^{\vartheta_1}(x_1)=\mfR_{\mfT,\mfs}^{\vartheta_2}(x_2)$. By arguments similar to the above, this equality can be used to conclude that
  $y_1=y_2$. The proposition follows. 
\end{proof}

Next, we prove Proposition~\ref{prop:dev to dos BIX}. 
\begin{proof}[Proof of Proposition~\ref{prop:dev to dos BIX}]
  The proof is similar to the proof of Proposition~\ref{prop:dev to data on sing nIX}. However, there are a few slightly subtle differences. For that
  reason, we provide a complete argument. Fix $\vartheta>0$. Due to Lemma~\ref{lemma:sfR mfT mfs}, Definition~\ref{def:B varthetaz} and
  Remark~\ref{remark:par iso id fixed mc}, we can identify ${}^{\rosc}\mfB_{\mrIX}^{\mfs}[V](\vartheta)$ with the quotient
  $\sfR_{\mrIX}^\mfs(\vartheta)=\sfB_{\mrIX}^{\mfs}(\vartheta)/\Gamma_{\mrIX}^\mfs$. Moreover, due to Lemma~\ref{lemma:sfRmfTmfsvartheta sm mfd},
  $\sfR_{\mrIX}^\mfs(\vartheta)$ is a smooth manifold. Denote the subset of $\sfB_{\mrIX}^{\mfs}(\vartheta)$ which is projected to
  ${}^{\rosc}\mfB_{\mrIX,\mft}^{\mfs,\rond}[V](\vartheta)$ via this identification by $\sfB_{\mrIX,\mft}^{\mfs,\rond}(\vartheta)$. Next, $\sfF_\sfP$
  induces a diffeomorphism from $\sfP_{\mrIX}^\mfs(\ln \vartheta)$ to $\sfB_{\mrIX}^{\mfs}(\vartheta)$; cf. Lemma~\ref{lemma:mfB mP rel};
  Definitions~\ref{def:BTs GTs}, \ref{def:B varthetaz} and \ref{def:BTs LTs}; and the proof of Lemma~\ref{lemma:sfQ mfT mfs}.
  Let
  \[
    \sfP_{\mrIX,\mft}^{\mfs,\rond}(\ln \vartheta):=\sfF_\sfP^{-1}[\sfB_{\mrIX,\mft}^{\mfs,\rond}(\vartheta)].
  \]
  Due to Lemmas~\ref{lemma:BianchiAdevelopment} and \ref{lemma:Bianchi IX remainder} and Definition~\ref{def:plus ap pp nd},
  initial data in $\sfB_{\mrIX,\mft}^{\mfs,\rond}(\vartheta)$ correspond to developments such that the left end point of the existence interval
  is $0$; $\theta(t)\rightarrow\infty$ as $t\downarrow 0$; there is a $t_0>0$ with the property that $\theta(t)>0$ and $\theta_t(t)<0$ for
  $t\leq t_0$; and $\theta$ is not integrable on $(0,t_0)$. Moreover, $t_0$ can be assumed to be the initial time at which the initial data
  are specified. Next, since $\theta$ is not integrable on $(0,t_0)$, it is possible to introduce $\tau$ by the requirements that $\tau(t_0)=0$
  and that (\ref{eq:dtdtau}) hold. The range of $\tau$ then includes $(-\infty,0]$. Let
  $\mfI_0\in {}^{\rosc}\mB_{\mrIX,\mft}^{\mfs,\rond}[V](\vartheta)$ and $\eta_0\in \sfB_{\mrIX,\mft}^{\mfs,\rond}(\vartheta)$ be such that it arises from $\mfI_0$;
  cf. Remark~\ref{remark:eta BTs realised}. Note that this correspondence gives rise to the identification between
  ${}^{\rosc}\mfB_{\mrIX}^{\mfs}[V](\vartheta)$ and $\sfR_{\mrIX}^\mfs(\vartheta)$; cf. the proof of Lemma~\ref{lemma:sc mfB mfT mfs param}. Next, let
  $\xi_0:=\sfF_\sfP^{-1}(\eta_0)$. Assume $\mfI_0$ to induce data on the singularity. Then, due to the uniqueness part of the proof of
  Theorem~\ref{thm:dataonsingtosolution}, it follows that $\md[V](\mfI_0)$ corresponds to an admissible convergent solution to
  (\ref{eq:thetaprime})--(\ref{eq:constraint}); see Definition~\ref{def:admconvsol}. In order to justify this statement, note first that
  the existence interval (in $\tau$-time) is unbounded to the past. Moreover, $\phi_1$ and the $\Sigma_i$
  converge due to the assumption that the development induces data on the singularity. In addition, (\ref{eq:normthetaas}) holds,
  so that $\theta(\tau)\geq-3\tau-C_\theta$ for $\tau\leq T$.  Next, assume that all the eigenvalues
  of $\msK$ are $<1$. Then, due to the proof of Theorem~\ref{thm:dataonsingtosolution}, all the $N_i$ converge to zero exponentially.
  In fact, if $\bar{\sigma}_i$ is the limit of $\Sigma_i$, then $\bar{\sigma}_i+1/3>0$ in case $N_i\neq 0$; see, again, the proof of
  Theorem~\ref{thm:dataonsingtosolution}. Finally, assume that there is an eigenvalue of $\msK$ which equals $1$. Then, due to the
  proof of Theorem~\ref{thm:dataonsingtosolution}, we can assume that $\Nt=\Nth$, that $\Sm=0$ and that $\Sp\rightarrow -1$. Note also
  that the limit of $\phi_1$ has to equal zero, since $\msK$ could otherwise not have an eigenvalue equal to $1$. This means that (ii) of
  Definition~\ref{def:admconvsol} is satisfied. 

  Due to Lemma~\ref{lemma:reg of B}, there is an open neighbourhood of $\xi_0$ in $\sfP_{\mrIX}^\mfs(\ln\vartheta)$, say $U$, which leads to
  admissible solutions. Since ${}^{\rosc}\mfB_{\mrIX,\mft}^{\mfs,\rond}[V](\vartheta)$ is an open subset of ${}^{\rosc}\mfB_{\mrIX}^{\mfs}[V](\vartheta)$,
  see Lemma~\ref{lemma:PimfT well def}, we can assume that $U$ is contained in $\sfP_{\mrIX,\mft}^{\mfs,\rond}(\ln\vartheta)$.   
  Due to Lemma~\ref{lemma:from adm to dos}, it follows that if $\mfI\in {}^{\rosc}\mB_{\mrIX}^{\mfs}[V](\vartheta)$ and $\eta\in\sfF_\sfP(U)$
  arises from $\mfI$, then $\mD[V](\mfI)$ induces data on the singularity. The set of elements in $\sfB_{\mrIX,\mft}^{\mfs,\rond}(\vartheta)$ that
  give rise to data on the singularity, say $W$, is therefore an open subset of $\sfB_{\mrIX}^{\mfs}(\vartheta)$. Since the corresponding subset of
  ${}^{\rosc}\mfB_{\mrIX,\mft}^{\mfs,\rond}[V](\vartheta)$ is a quotient of $W$, we conclude that
  ${}^{\rosc}_{\roqu}\mfB_{\mrIX,\mft}^{\mfs,\rond}[V](\vartheta)$ is an open subset of ${}^{\rosc}\mfB_{\mrIX}^{\mfs}[V](\vartheta)$.
  
  That $\mfR_{\mrIX,\mfs,\mft}^{\vartheta}$ is a local diffeomorphism follows from Lemma~\ref{lemma:local repr map to sing}. In order to prove that
  it is a diffeomorphism onto its image, it is sufficient to prove injectivity. Assume, to this end, that
  $\mfR_{\mrIX,\mfs,\mft}^{\vartheta}(x_1)=\mfR_{\mrIX,\mfs,\mft}^{\vartheta}(x_2)$. This means
  that if $[\mfI_i]=x_i$, $i=1,2$, then the data induced on the singularity by $\md[V](\mfI_i)$ are isometric. Denote the isometry between
  the data on the singularity by $\chi$. Due to Proposition~\ref{prop:iso idos to iso dev}, $\chi\times\roId$ is then an isometry between
  the developments $\md[V](\mfI_i)$. In particular, $\chi$ induces an isometry between $\mfI_i$, $i=1,2$, since the mean curvatures of
  $\mfI_i$ coincide and $\theta_t<0$ for $t\leq t_0$. This means that $x_1=x_2$ and the global injectivity of $\mfR_{\mrIX,\mfs,\mft}^{\vartheta}$ follows.

  Next, assume that (\ref{eq:iota vthonexo eq iota vthtwoextwo IX}) holds with $\vartheta_2\geq\vartheta_1$. Then, since $\theta\rightarrow\infty$ to
  the past and $\theta_t<0$ for $t\leq t_0$ for the developments under consideration here, see Definition~\ref{def:plus ap pp nd}, a development
  corresponding to the isometry class of developments induced by $x_1$ has a hypersurface of the form $G_1\times\{t\}$ with mean curvature $\vartheta_2$
  and $t\leq t_0$. Call the corresponding isometry class of initial data $x_3$. Then $\iota_{\vartheta_2}(x_3)=\iota_{\vartheta_2}(x_2)$, so that $x_2=x_3$
  due to Lemma~\ref{lemma:dev smo str IX}. Thus $\mfR_{\mrIX,\mfs,\mft}^{\vartheta_1}(x_1)=\mfR_{\mrIX,\mfs,\mft}^{\vartheta_2}(x_2)$. In order to prove surjectivity
  in case $\mft\in\{\roap,\ropp\}$, note that, given data on the singularity, there is a corresponding development due to
  Theorem~\ref{thm:dataonsingtosolution}. Due to the proof of Theorem~\ref{thm:dataonsingtosolution}, we know that $\theta$ tends to infinity,
  that $V/\theta^2$ converges to zero and that $\bS/\theta^2$ converges to zero in the direction of the singularity. Moreover, we know that $q$
  converges to $2$, so that there is a $t_1$ with the property that $\theta_t(t)<0$ for $t\leq t_1$. Combining these observations with
  Definitions~\ref{definition:Bap Bp Bpp} and \ref{def:plus ap pp nd}, it follows that the data are in the image of $\mfR_{\mrIX,\mfs,\mft}$ for
  $\mft\in\{\roap,\ropp\}$. That ${}^{\rosc}_{\roqu}\mfD_{\mrIX,\mft}^{\mfs}[V]$ is an open subset of ${}^{\rosc}\mfD_{\mrIX,\mft}^\mfs[V]$
  and that $\mfR_{\mfT,\mfs,\mft}$ is a local diffeomorphism follows from Lemma~\ref{lemma:dev smo str IX}. In order to prove injectivity, assume
  that $\mfR_{\mrIX,\mfs,\mft}(y_1)=\mfR_{\mrIX,\mfs,\mft}(y_2)$. We can assume that $y_i=\iota_{\vartheta_i}(x_i)$ and that $\vartheta_2\geq \vartheta_1$. Then
  $\mfR_{\mrIX,\mfs,\mft}^{\vartheta_1}(x_1)=\mfR_{\mrIX,\mfs,\mft}^{\vartheta_2}(x_2)$. By arguments similar to the above, this equality can be used to conclude that
  $y_1=y_2$. The proposition follows. 
\end{proof}
Finally, we prove Proposition~\ref{prop:dev to dos iso BI}.
\begin{proof}[Proof of Proposition~\ref{prop:dev to dos iso BI}]
  The proof is again similar to the proof of Proposition~\ref{prop:dev to data on sing nIX}. However, there are, again, a few slightly subtle
  differences. For that reason, we provide a complete argument here. Fix $\vartheta>[3v_{\max}(V)]^{1/2}$. Due to Lemma~\ref{lemma:sfR mfT mfs},
  Definition~\ref{def:B varthetaz} and Remark~\ref{remark:par iso id fixed mc}, we can identify ${}^{\rosc}\mfB_{\mrI,\rond}^{\iso}[V](\vartheta)$ with
  the trivial quotient $\sfR_{\mrI,\rond}^{\iso}(\vartheta)=:\sfB_{\mrI,\rond}^{\iso}(\vartheta)$; recall that
  $\Gamma_{\mrI}^{\iso}=\{\roId\}$,  see Definition~\ref{def:BTs GTs}. Next note that the inverse of $\sfF_\sfP$, introduced in
  (\ref{eq:sfF sfP def}) defines a diffeomorphism from $\sfB_{\mrI,\rond}^{\iso}(\vartheta)$ to
  \[
    \sfP_{\mrI,\rond}^{\iso}(\ln\vartheta):=\sfF_\sfP^{-1}[\sfB_{\mrI,\rond}^{\iso}(\vartheta)];
  \]
  cf. Lemma~\ref{lemma:mfB mP rel}; Definitions~\ref{def:BTs GTs}, \ref{def:B varthetaz} and \ref{def:BTs LTs}; and the proof of
  Lemma~\ref{lemma:sfQ mfT mfs}. Next, assume that $\mfI_0\in {}^{\rosc}\mB_{\mrI,\rond}^{\iso}[V](\vartheta)$ are such that $\mD[V](\mfI_0)$ induces
  data on the singularity. This means that the mean curvature $\theta$ diverges to $\infty$ to the past. Moreover, the expansion normalised Weingarten
  map $\mK$ converges to $\roId/3$, so that, by the asymptotic constraint, $\phi_t^2/\theta^2$ has to converge to $2/3$. Combining
  this observation with (\ref{eq:thetad}), it follows that 
  $\theta_t/\theta^2\rightarrow -1$. From the latter observation, it follows that $\theta$ blows up in finite time to the past and
  that $\theta$ is not integrable to the past. Combining these observations with Lemma~\ref{lemma:BianchiAdevelopment}, it follows that
  $\theta$ is neither integrable to the future nor to the past. It is therefore
  possible to introduce $\tau$ by the requirement that if $t_0$ corresponds to the initial hypersurface, then $\tau(t_0)=0$ and $\tau$ satisfies
  (\ref{eq:dtdtau}). The range of $\tau$ then includes $(-\infty,0]$. Let $\eta_0\in \sfB_{\mrI,\rond}^{\iso}(\vartheta)$ be such that it arises from $\mfI_0$;
  cf. Remark~\ref{remark:eta BTs realised}. Finally, let $\xi_0:=\sfF_\sfP^{-1}(\eta_0)$. Due to the uniqueness part of the proof of
  Theorem~\ref{thm:dataonsingtosolution}, it follows that $\md[V](\mfI_0)$ corresponds to an admissible convergent solution to
  (\ref{eq:thetaprime})--(\ref{eq:constraint}); see Definition~\ref{def:admconvsol}. In order to justify this statement, note first that
  the existence interval (in $\tau$-time) is unbounded to the past. Moreover, $\phi_1$ and the $\Sigma_i$
  converge due to the assumption that the development induces data on the singularity. In addition, (\ref{eq:normthetaas}) holds,
  so that $\theta(\tau)\geq-3\tau-C_\theta$ for $\tau\leq T$. Finally, since all the $N_i$ vanish, it is clear that condition (i) of
  Definition~\ref{def:admconvsol} is satisfied. 

  Due to Lemma~\ref{lemma:reg of B}, there is an open neighbourhood of $\xi_0$ in $\sfP_{\mrI,\rond}^{\iso}(\ln\vartheta)$, say $U$, which leads to
  admissible solutions. Moreover, the solutions are such that $\phi_t^2/\theta^2$ converges to $2/3$. As above, this means that $\theta$
  diverges to $\infty$ to the past (in finite time) and that $\theta$ is not integrable to the past. Combining these observations with
  Lemma~\ref{lemma:from adm to dos}, it follows that if $\mfI\in{}^{\rosc}\mB_{\mrI,\rond}^{\iso}[V](\vartheta)$ and $\eta\in\sfF_\sfP(U)$ arises from
  $\mfI$, then $\mD[V](\mfI)$ induces data on the singularity. The subset of $\sfB_{\mrI,\rond}^{\iso}(\vartheta)$ that gives rise to data on the
  singularity, say $W$, is therefore open, so that ${}^{\rosc}_{\roqu}\mfB_{\mrI,\rond}^{\iso}[V](\vartheta)$ is an open subset of
  ${}^{\rosc}\mfB_{\mrI,\rond}^{\iso}[V](\vartheta)$.
  
  That $\mfR_{\mrI,\iso}^{\vartheta}$ is a local diffeomorphism follows from Lemma~\ref{lemma:local repr map to sing}, keeping
  Remark~\ref{remark:q big m one eq rond} in mind. In order to prove that
  it is a diffeomorphism onto its image, it is sufficient to prove injectivity. Assume, to this end, that
  $\mfR_{\mrI,\iso}^{\vartheta}(x_1)=\mfR_{\mrI,\iso}^{\vartheta}(x_2)$. This means
  that if $[\mfI_i]=x_i$, $i=1,2$, then the data induced on the singularity by $\md[V](\mfI_i)$ are isometric. Denote the isometry between
  the data on the singularity by $\chi$. Due to Proposition~\ref{prop:iso idos to iso dev}, $\chi\times\roId$ is then an isometry between
  the developments $\md[V](\mfI_i)$. In particular, $\chi$ induces an isometry between $\mfI_i$, $i=1,2$ (in order to obtain this conclusion,
  we use the monotonicity of the mean curvature obtained in Lemma~\ref{lemma:bth large enough}). This means that $x_1=x_2$ and
  the global injectivity of $\mfR_{\mrI,\iso}^{\vartheta}$ follows.

  Next, assume that (\ref{eq:iota vthonexo eq iota vthtwoextwo I iso}) holds. Due to Lemma~\ref{lemma:bth large enough}, we know that
  $\theta_t\leq 0$ and that the zeros of $\theta_t$ are isolated. This means that there is a $\vartheta_3>\vartheta_i$, $i=1,2$, such
  that developments corresponding to the isometry classes of developments induced by the $x_i$ have hypersurfaces of the form $G_i\times\{t\}$
  with mean curvature $\vartheta_3$ and $\theta_t<0$ (which is equivalent to the normal derivative of the scalar field being non-zero). Call
  the corresponding isometry classes of initial data $x_{3,i}$. Then $\iota_{\vartheta_3}(x_{3,1})=\iota_{\vartheta_3}(x_{3,2})$, so that
  $x_{3,2}=x_{3,1}$ due to Lemma~\ref{lemma:dev smo str isoI}. Thus $\mfR_{\mrI,\iso}^{\vartheta_1}(x_1)=\mfR_{\mrI,\iso}^{\vartheta_2}(x_2)$.
  The surjectivity of $\mfR_{\mrI,\iso}$ follows from Theorem~\ref{thm:dataonsingtosolution}. That ${}^{\rosc}_{\roqu}\mfD_{\mrI,\roc}^{\iso}[V]$ is an open subset
  of ${}^{\rosc}\mfD_{\mrI,\roc}^{\iso}[V]$ and that $\mfR_{\mrI,\iso}$ is a local diffeomorphism follows from Lemma~\ref{lemma:dev smo str isoI}. In order
  to prove injectivity, assume that $\mfR_{\mrI,\iso}(y_1)=\mfR_{\mrI,\iso}(y_2)$. We can assume that $y_i=\iota_{\vartheta_i}(x_i)$ and that
  $\vartheta_2\geq \vartheta_1$. Then $\mfR_{\mrI,\iso}^{\vartheta_1}(x_1)=\mfR_{\mrI,\iso}^{\vartheta_2}(x_2)$. By arguments similar to the above, this equality
  can be used to conclude that $y_1=y_2$. The proposition follows. 
\end{proof}

\chapter{The $k=0$ and $k=-1$ FLRW solutions}\label{chapter:k eq zero and minus one cases}

In the present chapter, we discuss spatially homogeneous and isotropic solutions with vanishing or negative spatial curvature. 

\section{The isotropic Bianchi type I setting}

As is clear from Remark~\ref{remark:Exceptional Bianchi type I}, isotropic Bianchi type I solutions need not have a crushing singularity. However, in the present section,
we provide conditions on $V$ ensuring that this behaviour is non-generic. We also derive asymptotics that can be combined with \cite{GPR} in order to
yield stability results. Due to Remark~\ref{remark:Exceptional Bianchi type I}, it is clear that $\theta^{2}(t)=3V\circ\phi(t)$ is a problematic condition. The next
lemma illustrates that if $\theta$ is slightly bigger than this, a crushing singularity occurs. In the following results, we use the terminology introduced in
Definition~\ref{def:M plus}. 
\begin{lemma}\label{lemma: Bianchi I isotropic}
  Let $V\in C^{\infty}(\rn{})$ be non-negative, and let $\theta\in C^{\infty}(J,(0,\infty))$ and $\phi\in C^{\infty}(J,\rn{})$ be a solution to
  (\ref{seq:EFEwrtvar})--(\ref{eq:phiddot}) corresponding to initial data in $B^{\iso}_{\mrI,+}$, where $J$ is the maximal existence interval.
  Assume that there is a $t_{0}\in J$ and a $\delta>0$ such that $\theta^{2}(t)\geq 3(1+\delta)V\circ\phi(t)$ for all $t\leq t_0$. Then, after a translation
  in $t$, if necessary,
  $J=(0,\infty)$, $\theta(t)\rightarrow \infty$ as $t\downarrow 0$, $\theta\notin L^{1}(0,1)$ and $\theta\notin L^{1}(1,\infty)$. In particular, introducing
  $\tau$ according to (\ref{eq:dtdtau}), the range of $\tau$ is $\rn{}$.  

  Assume, in addition to the above, that $V\in\mfP_{\a_V}^0$ for some $\alpha_V\in (0,1)$. Let, moreover, $a_1=a_2=a_3=:a$, where $a_i$ is defined by
  (\ref{eq:for the ai}) and $\sigma_{ii}=0$ in the present setting. Then, if $\alpha_V<\delta/(1+\delta)$, there are constants $\Phi_0$,
  $\Phi_1=\pm (2/3)^{1/2}$, $\alpha>0$ and $C$ such that 
  \begin{equation}\label{eq:qms}
    |\theta^{1/3} a-\alpha|+|\phi+\theta^{-1}\phi_t\ln\theta-\Phi_0|+|\theta^{-1}\phi_{t}-\Phi_1|\leq C\ldr{\ln\theta}\theta^{-2(1-\alpha_V)}
  \end{equation}
  for all $t\leq t_0$. 
\end{lemma}
\begin{remark}\label{remark:data induced on sing BI iso}
  Given that all the assumptions of the lemma are satisfied, the spacetime metric on the universal covering space can be written
  \[
    g=-dt\otimes dt+a^2(t)\bge_0,
  \]
  where $\bge_0$ is the standard flat metric on $\rn{3}$. Letting $\msH:=\a^2\bge_0$, $\msK=\mathrm{Id}/3$ and $\bge=a^2(t)\bge_0$, it then follows from
  the conclusions of the lemma that $(\rn{3},\msH,\msK,\Phi_0,\Phi_1)$ satisfy the conditions of Definition~\ref{def:ndvacidonbbssh} and that,
  for all $v,w\in T_p\rn{3}$ and all $p\in\rn{3}$,
  \[
    \lim_{t\downarrow 0}\bge(\theta^{\mK}v,\theta^{\mK}w)=\msH(v,w),\ \ \
    \mK=\msK,\ \ \
    \lim_{t\downarrow 0}(\theta^{-1}\phi_t)=\Phi_1,\ \ \
    \lim_{t\downarrow 0}(\phi+\theta^{-1}\phi_t\ln\theta)=\Phi_0,\ \ \
  \]
  where $\mK=\mathrm{Id}/3$. In other words, the development induces data on the singularity. 
\end{remark}
\begin{remark}
  In the isotropic setting, all the eigenvalues of the expansion normalised Weingarten map equal $1/3$. Imposing suitable conditions on the potential, the
  estimate (\ref{eq:qms}) can thus be used to prove that the development constructed in Proposition~\ref{prop:unique max dev} (after taking a quotient by a
  co-compact free and properly discontinuous subgroup of the isometry group of $(\rn{3},\bge_0)$) is a quiescent model solution in the sense of
  \cite[Definition~46]{GPR}. Moreover, for quiescent model solutions, stable big bang formation is demonstrated in \cite[Theorem~49]{GPR}.
\end{remark}
\begin{remark}\label{remark:V divided by theta sq to zero}
  Assuming $J=(t_-,\infty)$, that $\theta>0$ and that $V\circ\phi(t)/\theta(t)^{2}\rightarrow 0$ as $t\downarrow t_-$, any $\alpha_V<1$ is sufficient.
  In fact, if $\alpha_V<1$, choose $\delta$ large enough that $\alpha_V<\delta/(1+\delta)$. Choosing $t_0$ close enough to $t_-$, we then have
  $\theta^{2}(t)\geq 3(1+\delta)V\circ\phi(t)$ for all $t\leq t_0$.
\end{remark}
\begin{proof}
  By Lemma~\ref{lemma:BianchiAdevelopment}, $J=(t_-,\infty)$. Moreover, combining (\ref{eq:hamconfin}) with the assumptions yields
  \[
    \tfrac{3}{2}\phi_t^2+3V\circ\phi=\theta^{2}\geq 3(1+\delta)V\circ\phi
  \]
  on $J_-:=(t_-,t_0]$, so that $\phi_t^{2}\geq 2\delta V\circ\phi$ on $J_-$. Combining this observation with (\ref{eq:hamconfin}) again yields
  \[
    \theta^{2}=\tfrac{3}{2}\phi_t^2+3V\circ\phi\leq \tfrac{3}{2}(1+\delta^{-1})\phi_t^2.
  \]
  Due to this estimate and (\ref{eq:thetad}),
  \begin{equation}\label{eq:thetat estimate}
    \theta_t\leq -\tfrac{\delta}{1+\delta}\theta^{2}
  \end{equation}
  on $J_-$. Since $\theta>0$ on $J$, this means that $\theta$ explodes in finite time to the past. Since we already know that solutions exist globally
  to the future, see Lemma~\ref{lemma:BianchiAdevelopment}, we can assume that $J=(0,\infty)$. The statements concerning the lack of integrability
  of $\theta$ follow from Lemma~\ref{lemma:BianchiAdevelopment}.

  Next, note that $|\phi_\tau|\leq \sqrt{6}$ due to (\ref{eq:hamconfin}) and the non-negativity of $V$. Combining this estimate with the fact that
  $V\in\mfP_{\a_V}^0$ yields $|V(\phi(\tau))|\leq Ce^{6\alpha_V|\tau|}$ for all $\tau\leq 0$. On the other hand, due to (\ref{eq:thetat estimate}), there
  is a constant $c_a>0$ such that
  \[
    \theta(\tau)\geq c_a \exp\big(-\tfrac{3\delta}{1+\delta}\tau\big)
  \]
  for all $\tau\leq 0$. Combining the last two estimates yields
  \begin{equation}\label{eq:exp norm V est}
    \tfrac{|V(\phi(\tau))|}{\theta(\tau)^2}\leq C\exp\big[6\big(\tfrac{\delta}{1+\delta}-\alpha_V\big)\tau\big]
  \end{equation}
  for all $\tau\leq 0$. By assumption, the right hand side decays to zero exponentially as $\tau\rightarrow -\infty$. On the other hand, combining
  (\ref{eq:thetad}) and (\ref{eq:hamconfin}) yields
  \begin{equation}\label{eq: theta tau isotropic}
    \theta_\tau=-\left(3-\tfrac{9V\circ\phi}{\theta^{2}}\right)\theta.
  \end{equation}
  This means that $e^{3\tau}\theta(\tau)$ converges exponentially to a non-zero constant $\theta_\infty$ as $\tau\rightarrow -\infty$. Returning to the above arguments
  with this information at hand yields the following improvement of (\ref{eq:exp norm V est}):
  \begin{equation}\label{eq:exp norm V est improved}
    \tfrac{|V(\phi(\tau))|}{\theta(\tau)^2}\leq C\exp\left[6\left(1-\alpha_V\right)\tau\right]
  \end{equation}
  for all $\tau\leq 0$. Returning to (\ref{eq: theta tau isotropic}) then yields
  \begin{equation}\label{eq:theta asymp isotropic}
    |e^{3\tau}\theta(\tau)-\theta_\infty|\leq C\exp\left[6\left(1-\alpha_V\right)\tau\right]
  \end{equation}
  for all $\tau\leq 0$. Next, note that $a_\tau/a=1$. Thus $a(\tau)=a(0)e^{\tau}$. Thus
  \[
    |(\theta^{1/3}a)(\tau)-\theta_\infty^{1/3} a(0)|\leq C\exp\left[6\left(1-\alpha_V\right)\tau\right]\leq C[\theta(\tau)]^{-2(1-\alpha_V)}
  \]
  for all $\tau\leq 0$. Next, note that due to (\ref{eq:hamconfin}),
  \[
    |\phi_\tau^{2}(\tau)-6|\leq C\exp\left[6\left(1-\alpha_V\right)\tau\right]
  \]
  for all $\tau\leq 0$. Thus there is a $\Phi_{1}=\pm \sqrt{2/3}$ such that
  \[
    |\theta^{-1}\phi_{t}-\Phi_1|\leq C\theta^{-2(1-\alpha_V)}
  \]
  for all $\tau\leq 0$. Integrating this estimate yields a real number $\Phi_0$ such that
  \[
    |\phi+\Phi_1\ln\theta-\Phi_0|\leq C\theta^{-2(1-\alpha_V)}
  \]
  for all $\tau\leq 0$. The lemma follows. 
\end{proof}

Next, we discuss the case of a bounded potential.
\begin{cor}\label{cor: Bianchi I isotropic V bd}
  Let $V\in C^{\infty}(\rn{})$ be non-negative and bounded, and let $\theta\in C^{\infty}(J,(0,\infty))$ and $\phi\in C^{\infty}(J,\rn{})$ be a solution to
  (\ref{seq:EFEwrtvar})--(\ref{eq:phiddot}) corresponding to initial data in $B^{\iso}_{\mrI,+}$, where $J$ is the maximal existence interval. Then
  there are two possibilities. Either all the conclusions of Lemma~\ref{lemma: Bianchi I isotropic} hold; or $J=\ro$, $\theta(t)\leq [3\sup_sV(s)]^{1/2}$
  for all $t$ and there is a $\theta_\infty\in (0,\infty)$ such that $\theta(t)\rightarrow\theta_\infty$ as $t\rightarrow-\infty$
\end{cor}
\begin{proof}
  Assume that there is a $t_0\in J$ such that $\theta(t_0)>[3\sup_sV(s)]^{1/2}$. Since $\theta$ is increasing to the past, see
  (\ref{eq:thetad}), this means that $\theta(t)\geq\theta(t_0)>[3\sup_sV(s)]^{1/2}$ for all $t\leq t_0$. Combining this observation with
  (\ref{eq:hamconfin}) yields
  \[
  \tfrac{3}{2}\phi_t^2(t)=\theta^2(t)-3V\circ\phi(t)\geq \theta^2(t_0)-3\sup_sV(s)=:\alpha
  \]
  for all $t\leq t_0$, where the last equality defines the constant $\alpha>0$. Combining this observation with (\ref{eq:thetad}) yields
  the conclusion that $\theta_t(t)\leq -\a$ for all $t\leq t_0$. Assuming that the solution exists for all $t\leq t_0$, we conclude that
  $\theta(t)\rightarrow\infty$ as $t\rightarrow-\infty$. Since $V$ is bounded, this means that $V\circ\phi(t)/\theta^2(t)\rightarrow 0$
  as $t\rightarrow-\infty$. Combining this observation with Lemma~\ref{lemma: Bianchi I isotropic}, it follows that all the conclusions
  of Lemma~\ref{lemma: Bianchi I isotropic} hold; in particular, the solution does not exist for all $t\leq t_0$. We can thus assume that
  the solution does not exhibit past global existence. Thus $\theta$ blows up in finite time to the past; cf.
  Remark~\ref{remark:improved dichotomy}. This means that all the conclusion of Lemma~\ref{lemma: Bianchi I isotropic} hold.

  Next, assume that $\theta(t)\leq[3\sup_sV(s)]^{1/2}$ for all $t\in J$. Then $\theta$ converges to a finite number $\theta_\infty$ to the past
  due to (\ref{eq:thetad}). 
\end{proof}

By Lemma~\ref{lemma: Bianchi I isotropic}, the condition that $\theta^{2}(t)\geq 3(1+\delta)V\circ\phi(t)$ for all $t\leq t_0$ is sufficient
to guarantee the desired asymptotics, assuming a suitable bound on $V$. The advantage of this result is that it illustrates that while the equality
$\theta^{2}(t)=3V\circ\phi(t)$ is problematic (in that it is compatible with solutions that do not have a crushing singularity), a slightly bigger
mean curvature leads to a crushing singularity. The disadvantage is that the result is based on assumptions concerning the solution for all $t\leq t_0$.
For this reason, we next prove a result which is only based on assumptions concerning the initial data.

\begin{lemma}\label{lemma:open set yielding data on sing BI iso}
  Let $0\leq V\in\mfP_{\a_V}^0$ for some $\a_V\in (0,1)$. Fix a $\Phi_{1}=\pm \sqrt{2/3}$ and a $\Phi_0$. Then there is a constant $c_1>1$ depending
  continuously and exclusively on the constant $c_0$ (appearing in Definition~\ref{def:mfP a inf}), $\a_V$ and an upper bound on $|\Phi_0|$, such that
  the following holds: If
  \[
    \theta(t_0)> c_1\ \ and\ \
    \phi(t_0)=-\Phi_1 \ln\theta(t_0)+\Phi_0,
  \]
  then $\theta^{2}(t_0)-3V\circ\phi(t_0)>0$, so that 
  \begin{equation}\label{eq:phit tz iso}
    \phi_t(t_0):=\pm\big(\tfrac{2}{3}\big)^{1/2}\left(\theta^{2}(t_0)-3V\circ\phi(t_0)\right)^{1/2}
  \end{equation}
  is well defined, where the sign is chosen so that $\phi_t(t_0)$ and $\Phi_1$ have the same sign. Moreover, if $\theta\in C^{\infty}(J,(0,\infty))$ and
  $\phi\in C^{\infty}(J,\rn{})$ constitute the maximal isotropic Bianchi type I solution to (\ref{seq:EFEwrtvar})--(\ref{eq:phiddot}) corresponding to
  the initial data $\theta(t_0)$, $\phi(t_0)$ and $\phi_t(t_0)$, then 
  \begin{equation}\label{eq:limit V divided by theta squared}
    \lim_{t\downarrow t_-}\frac{V\circ\phi(t)}{\theta^{2}(t)}=0,
  \end{equation}
  where $t_-$ is the left endpoint of $J$. Finally, all the conclusions of Lemma~\ref{lemma: Bianchi I isotropic} hold. 
\end{lemma}
\begin{remark}
  The set of initial data satisfying the conditions is an open subset of $B^{\iso}_{\mrI,+}$.
\end{remark}
\begin{proof}
  Estimate, appealing to (\ref{eq:V k-derivatives exp estimate}), 
  \[
    \tfrac{V\circ\phi(t_0)}{\theta^{2}(t_0)}\leq c_0\exp\big(-2(1-\a_V)\ln\theta(t_0)+\sqrt{6}\a_V |\Phi_0|\big).
  \]
  Assuming $c_1$ to be large enough, the bound depending continuously and exclusively on $c_0$, $\a_V$ and an upper bound on $|\Phi_0|$, we can
  assume the right hand side to be as small as we wish.

  Next, we change time coordinate according to (\ref{eq:dtdtau}). Due to Lemma~\ref{lemma:BianchiAdevelopment}, the existence time in $\tau$-time
  is $\ro$. Define $\tau$ so that $\tau(t_0)=0$. Expressing everything with respect to $\tau$-time from now on (so that we write $\theta(0)$ instead
  of $\theta(t_0)$, for example), define
  \[
    \ma:=\{\tau\leq 0:\ln\theta(s)\geq -3s+\ln\theta(0)-1\ \forall s\in [\tau,0]\}.
  \]
  Clearly, $0\in\ma$. It is also clear that $\ma$ is closed and connected. It remains to prove that $\ma$ is open. Let $\tau_1\in\ma$.
  Since $|\phi_\tau|\leq \sqrt{6}$,
  \[
    |\phi(\tau)|\leq \sqrt{6}|\tau|+(2/3)^{1/2}\ln\theta(0)+|\Phi_0|
  \]
  for all $\tau\leq 0$. Combining this estimate with (\ref{eq:V k-derivatives exp estimate}) and the assumption that $\tau_1\in\ma$ yields
  \begin{equation}\label{eq:V estimate boot}
    \begin{split}
      \tfrac{V\circ\phi(\tau)}{\theta^{2}(\tau)} \leq & c_0\exp\left(6\tau-2\ln\theta(0)+2-6\a_V\tau+2\a_V\ln\theta(0)+\sqrt{6}\a_V|\Phi_0|\right)\\
      \leq & Ce^{6(1-\a_V)\tau}[\theta(0)]^{-2(1-\a_V)}
    \end{split}
  \end{equation}
  for all $\tau\in [\tau_1,0]$, where $C$ depends continuously and exclusively on $c_0$ and an upper bound on $|\Phi_0|$. Next, note that, due to
  (\ref{eq:thetad}) and (\ref{eq:hamconfin}), 
  \[
    \d_\tau \ln\theta=-\tfrac{9}{2\theta^2}\phi_t^2=-\tfrac{1}{2}\phi_\tau^2=-3+\tfrac{9V\circ\phi(\tau)}{\theta^{2}(\tau)}.
  \]
  Integrating this equality from $\tau_1$ to $0$, keeping (\ref{eq:V estimate boot}) in mind, yields
  \[
    \ln\theta(0)-\ln\theta(\tau_1)\leq 3\tau_1+C[\theta(0)]^{-2(1-\a_V)}.
  \]
  Assuming $c_1$ to be large enough, the bound depending continuously and exclusively on $c_0$, $\a_V$ and an upper bound on $|\Phi_0|$, it follows
  that $\ln\theta(\tau_1)\geq -3\tau_1+\ln\theta(0)-1/2$. This means that $\ma$ is open, so that $\ma=(-\infty,0]$.
  Next, fix $\delta>0$ large enough that $\a_V<\de/(1+\de)$. Then, due to (\ref{eq:V estimate boot}), there is a $t_0\in J$
  such that $\theta^2(t)\geq 3(1+\de)V\circ\phi(t)$ for all $t\leq t_0$. This means that all the conclusions of Lemma~\ref{lemma: Bianchi I isotropic}
  hold. The lemma follows. 
\end{proof}
Assume that $0\leq V\in\mfP_{\a_V}^0$ for some $\alpha_V\in (0,1)$. Then there are open conditions on initial data in $B^{\iso}_{\mrI,+}$ ensuring that the
conclusions of Lemmas~\ref{lemma: Bianchi I isotropic}
and \ref{lemma:open set yielding data on sing BI iso} hold. In particular, the corresponding solutions induce data on the singularity in the
sense of Definition~\ref{def:ndvacidonbbssh}; see Remark~\ref{remark:data induced on sing BI iso}. However, in order to obtain a more complete picture
of the dynamics, we here impose stronger conditions on the potential. To begin with, we characterise the case that the mean curvature remains bounded
to the past.

\begin{prop}\label{prop:iso theta bounded}
  Assume $V\in C^{\infty}(\rn{})$ to be non-negative and such that if $V'(s)=0$, then $V''(s)\neq 0$. Assume, moreover, that 
  $V(s)\rightarrow\infty$ as $|s|\rightarrow\infty$. Let $\theta\in C^{\infty}(J,(0,\infty))$ and $\phi\in C^{\infty}(J,\rn{})$ be a solution to
  (\ref{seq:EFEwrtvar})--(\ref{eq:phiddot}) corresponding to initial data in $B^{\iso}_{\mrI,+}$, where $J$ is the maximal existence
  interval. If $\theta$ is bounded, then $J=\rn{}$ and there are $\theta_\infty,\phi_\infty\in\rn{}$ such that
    \[
      \lim_{t\rightarrow-\infty}\theta(t)=\theta_\infty,\ \ \
      \lim_{t\rightarrow-\infty}\phi(t)=\phi_\infty,\ \ \
      \lim_{t\rightarrow-\infty}\phi_t(t)=0.
    \]
    Moreover, $V'(\phi_\infty)=0$, so that $(\theta_\infty,\phi_\infty,0)$ is a fixed point of the system. In addition, $V(\phi_\infty)>0$,
    $V''(\phi_\infty)<0$ and there are two solutions (up to time translations) in the set $B^{\iso}_{\mrI,+}$ converging to the fixed point, 
    each of which is a submanifold of $B^{\iso}_{\mrI,+}$.  
\end{prop}
\begin{remark}\label{remark:non crushing non generic}
  Since $V'$ has at most countably many zeros, the set of initial data corresponding to solutions with bounded mean curvature is contained in a
  countable union of codimension one submanifolds. In other words, the developments that do not have a crushing singularity correspond to a
  Baire and Lebesgue non-generic set of initial data. 
\end{remark}
\begin{proof}
  Since $\theta$ is bounded to the past, the solution exists globally to the past; cf. Remark~\ref{remark:improved dichotomy}. Since
  $V(s)\rightarrow\infty$ as $|s|\rightarrow\infty$, the equality (\ref{eq:hamconfin}) implies that $\phi$ and $\phi_t$ are bounded to the
  past. Moreover, $\phi_t^2$ is integrable to the past due to (\ref{eq:thetad}) and the fact that $\theta$ is bounded to the past. Due to the
  boundedness of $\theta$, $\phi_t$ and $\phi$, we conclude that $\phi_{tt}$ is bounded to the past, see
  (\ref{eq:phiddot}). Combining this observation with the integrability of $\phi_t^2$ to the past leads to the conclusion that $\phi_t$
  converges to zero. Since $\theta$ converges to a limit, say $\theta_\infty>0$ to the past, it follows that
  \begin{equation}\label{eq:Vphilimit}
    \lim_{t\rightarrow-\infty}V[\phi(t)]=\theta_\infty^{2}/3.
  \end{equation}
  Combining this observation with the fact that $\phi$ is bounded and the fact that $V''\neq 0$ when $V'=0$, it follows that $\phi$ converges
  to a limit, say $\phi_\infty$. Thus, due to (\ref{eq:phiddot}), $\phi_{tt}$ converges to $V'(\phi_\infty)$. In order for this to be consistent,
  $V'(\phi_\infty)$ has to vanish. In other words, the solution converges to a fixed point of the system. Note also that $V(\phi_\infty)>0$ due to
  (\ref{eq:Vphilimit}) and the fact that $\theta_\infty>0$. Rewriting (\ref{seq:EFEwrtvar})--(\ref{eq:phiddot}), close to the fixed point, as an
  unconstrained second order equation for only $\phi$, we conclude that if $V''(\phi_\infty)>0$, then both eigenvalues of the linearisation of the
  corresponding first order system have negative real parts. This means that no solutions can converge to this fixed point to the past, except for
  the fixed point itself. However, this possibility has been excluded by the definition of $B^{\iso}_{\mrI,+}$. 
  In case $V''(\phi_\infty)<0$, the linearisation has one negative and one positive eigenvalue. This means that there is a
  one dimensional stable manifold (going into the past; i.e., for the time reversed flow). In order for the solution to converge to the fixed point,
  it therefore, for $t$ close enough to $-\infty$, has to belong to the stable manifold. Since we have excluded the possibility that the solution
  coincides with the fixed point, there are thus two possibilities: the solution converges to the fixed point along $\pm v_+$, where $v_+$ is an
  eigenvector corresponding to the positive eigenvalue. Each of these possibilities corresponds to a unique (up to time translation) solution to the
  equations. Finally, we wish to prove that these solutions constitute submanifolds of the state space $B^{\iso}_{\mrI,+}$. In order to prove this, note
  first that if $\theta_t(t_0)=0$, then
  $\phi_t(t_0)=0$. If, in addition, $V'(\phi(t_0))=0$, then $(\theta(t_0),\phi(t_0),\phi_t(t_0))$ is a fixed point, contradicting the assumption that
  the solution is non-constant. In other words, if $\theta_t(t_0)=0$, then we have to have $\phi_{tt}(t_0)\neq 0$, so that $\theta_t<0$ in a punctured
  neighbourhood of $t_0$. Next, we know that the solution curve is a smooth immersion. In order to prove that it is an embedding, note that the map taking
  $t\in\rn{}$ to $(\theta(t),\phi(t),\phi_t(t))$ is injective, since $\theta$ is strictly monotonically decreasing (it is of course also surjective onto its
  image and continuous). In order to prove that the inverse is continuous, let $(\theta_k,\phi_{0,k},\phi_{1,k})$ be elements of the image convering
  to $(\theta_*,\phi_{0,*},\phi_{1,*})$ in the image. Then, since $\theta$ is strictly monotonic, the $t_k$ corresponding to
  $(\theta_k,\phi_{0,k},\phi_{1,k})$ have to converge to a limit, say $t_*$. Moreover, $t_*$ has to correspond to $(\theta_*,\phi_{0,*},\phi_{1,*})$, again
  due to the monotonicity of $\theta$. Thus the solution is a smooth embedding of $\rn{}$ into $B^{\iso}_{\mrI,+}$. Thus the solution is a submanifold of
  $B^{\iso}_{\mrI,+}$.  
\end{proof}

Next, we consider the case that $\theta$ is unbounded. 

\begin{prop}\label{prop:iso theta unbounded}
  Assume $0\leq V\in\mfP_{\a_V}^0$ for some $\a_V\in (0,1)$. Let $\theta\in C^{\infty}(J,(0,\infty))$ and $\phi\in C^{\infty}(J,\rn{})$ be a solution to
  (\ref{seq:EFEwrtvar})--(\ref{eq:phiddot}) corresponding to initial data in $B^{\iso}_{\mrI,+}$, where $J=(t_-,\infty)$ is the maximal existence interval.
  If $\theta$ is unbounded, then $|\phi(t)|\rightarrow\infty$ as $t\downarrow t_-$.
\end{prop}
\begin{proof}
  Due to Lemma~\ref{lemma:BianchiAdevelopment}, we know that $\theta$ is neither integrable to the future nor to the past. Thus, by introducing $\tau$
  according to (\ref{eq:dtdtau}), the range of $\tau$ is $\rn{}$.
  In order to prove that $|\phi|$ has to
  converge to $\infty$, assume that it does not. This means that there is a sequence $\{\tau_k\}$ with $\tau_k\rightarrow -\infty$ such that
  $\phi(\tau_k)$ is bounded. In fact, we can assume $\phi(\tau_k)$ to converge to, say, $\phi_\infty$. Fix $\e>0$. Given $k$, let 
  \[
  \ma_k:=\{\tau\leq\tau_k:\phi_\tau^2(s)/2\geq 3-\e\ \forall s\in [\tau,\tau_k]\}.
  \]
  Note that $\ma_k$ is closed, connected and, if $k$ is large enough, non-empty. The last statement is due to the fact that
  \begin{equation}\label{eq:exp norm constr}
    3=\tfrac{1}{2}\phi_\tau^2+\tfrac{9V\circ\phi}{\theta^2},
  \end{equation}
  see (\ref{eq:hamconfin}), the fact that $V\circ\phi(\tau_k)\rightarrow V(\phi_\infty)$ and the fact that $\theta(\tau_k)\rightarrow\infty$. Next, we
  wish to prove that if $k$ is large enough, then $\ma_k$ is an open subset of $(-\infty,\tau_k]$. Let, to this end, $\tau\in\ma_k$. Then
  \[
  \d_\tau\theta(s)=-\tfrac{1}{2}\phi_\tau^2(s)\theta(s)\leq -(3-\e)\theta(s)
  \]
  for all $s\in [\tau,\tau_k]$, so that $\theta(s)\geq \theta(\tau_k)e^{(3-\e)(\tau_k-s)}$ for all $s\in [\tau,\tau_k]$. On the other hand,
  \[
  |\phi(s)|\leq \sqrt{6}(\tau_k-s)+|\phi(\tau_k)|
  \]
  for all $s\in [\tau,\tau_k]$. Combining the above observations
  \begin{equation}\label{eq:nine V by theta sq}
    \tfrac{9V\circ\phi(s)}{\theta^2(s)}\leq \tfrac{9c_0 e^{\sqrt{6}\a_V |\phi(\tau_k)|}}{\theta(\tau_k)^2}e^{-6(1-\e/3-\a_V)(\tau_k-s)}
  \end{equation}
  for all $s\in [\tau,\tau_k]$. Assume $\e>0$ to be small enough that $1-\e/3-\a_V>0$. Assume, moreover, $k$ to be large enough that
  the first factor on the left hand side of (\ref{eq:nine V by theta sq}) is bounded from above by $\e/2$. Then $s\in\ma_k$ implies that
  $\phi_\tau^2(s)/2\geq 3-\e/2$ and $\ma_k$ contains a neighbourhood of $s$. In other words, $\ma_k$ is open, so that $\ma_k=(-\infty,\tau_k]$.
  This means that $V\circ\phi/\theta^2$ decays exponentially. This means that Lemma~\ref{lemma: Bianchi I isotropic} and
  Remark~\ref{remark:V divided by theta sq to zero} apply. In particular, $|\phi(t)|\rightarrow\infty$ as $t\downarrow t_-$.
\end{proof}
It is natural to conjecture that generic solutions have asymptotics corresponding to the conclusions of Lemma~\ref{lemma: Bianchi I isotropic}.
In the case of an exponential potential, this follows from Proposition~\ref{eq:gen isotropic Bianchi I exp pot} below. Proving it for
general potentials is more difficult. However, we prove one result of this nature in Theorem~\ref{thm:asympt as exp pot} below. This theorem
rests on the following proposition. 
\begin{prop}\label{prop:phi exc limit}
  Assume $0\leq V\in C^{\infty}(\rn{})$ and that there is an $M\geq 1$ and an $\a_V\in (0,1)$ such that $V(s)>0$ for all $|s|\geq M$;
  \begin{equation}\label{eq:VprbyVbd}
    |(\ln V)'(s)|\leq \sqrt{6}\a_V
  \end{equation}
  for all $|s|\geq M$; and
  \begin{equation}\label{eq:as exp or pol}
    \lim_{s\rightarrow\pm\infty}(\ln V)''(s)=0.
  \end{equation}
  Let $\theta\in C^{\infty}(J,(0,\infty))$ and $\phi\in C^{\infty}(J,\rn{})$ be a solution to (\ref{seq:EFEwrtvar})--(\ref{eq:phiddot}) corresponding
  to initial data in $B^{\iso}_{\mrI,+}$, where $J$ is the maximal existence interval. Assume that $\theta$ is unbounded to the past. Then either 
  all the conclusions of Lemma~\ref{lemma: Bianchi I isotropic} hold, or
  \begin{equation}\label{eq:phi exc limit}
    \lim_{\tau\rightarrow-\infty}\left[\phi_\tau(\tau)+(\ln V)'[\phi(\tau)]\right]=0.
  \end{equation}
\end{prop}
\begin{remark}
  If $V$ is asymptotically a polynomial or an exponential function, then (\ref{eq:as exp or pol}) is satisfied. 
\end{remark}
\begin{remark}
  It is reasonable to expect (\ref{eq:phi exc limit}) to correspond to an exceptional situation.
\end{remark}
\begin{remark}\label{remark:bounded in one direction}
  If we replace the assumptions that $V(s)>0$ and that (\ref{eq:VprbyVbd}) hold for $|s|\geq M$ with the assumptions that these estimates hold only for
  $s\geq M$ and the assumption that $V$ is bounded for $s\leq 0$; and if we replace the assumption that (\ref{eq:as exp or pol}) holds with the assumption
  that it only holds as $s\rightarrow\infty$, then the same conclusions hold. Finally, if $\bar{V}(s)=V(-s)$ satisfies these modified conditions, then the
  same conclusions hold. 
\end{remark}
\begin{proof}
  Due to (\ref{eq:VprbyVbd}), $V\in\mfP_{\a_V}^0$. Due to Lemma~\ref{lemma:BianchiAdevelopment}, we know that $\theta$ is neither
  integrable to the future nor to the past. Thus, introducing $\tau$ according to (\ref{eq:dtdtau}), the range of $\tau$ is $\rn{}$. Moreover,
  $|\phi(\tau)|\rightarrow\infty$ as $\tau\rightarrow-\infty$ due to Proposition~\ref{prop:iso theta unbounded}. Note that
  \begin{equation}\label{eq:phi tau tau iso}
    \phi_{\tau\tau}=-\tfrac{9V\circ\phi}{\theta^2}\phi_\tau-\tfrac{9V'\circ\phi}{\theta^2}
    =-\tfrac{9V\circ\phi}{\theta^2}\big(\phi_\tau+\tfrac{V'\circ\phi}{V\circ\phi}\big)
    =-\tfrac{9V\circ\phi}{\theta^2}\big(\phi_\tau+(\ln V)'\circ\phi\big).
  \end{equation}
  This means that
  \begin{equation}\label{eq:dtau phitau minus vproverv}
    \begin{split}
      \d_{\tau}\big(\phi_\tau+(\ln V)'\circ\phi\big)
      = & -\tfrac{9V\circ\phi}{\theta^2}\big(\phi_\tau+(\ln V)'\circ\phi\big)
      +(\ln V)''\circ\phi\cdot\phi_\tau\\
      = & -\big(\tfrac{9V\circ\phi}{\theta^2}-(\ln V)''\circ\phi\big)
      \big(\phi_\tau+(\ln V)'\circ\phi\big)\\
      & -(\ln V)''\circ\phi\cdot(\ln V)'\circ\phi.
    \end{split}
  \end{equation}  
  Let $\e>0$ and $\b_V:=2(1-\a_V^2)$. Assuming $\e$ to be small enough, the bound depending only on $\a_V$; $\tau$ to be such that $|\phi(\tau)|\geq M$;
  and the estimate $|\phi_\tau(\tau)+(\ln V)'\circ\phi(\tau)|\leq\e$ to hold, it follows that $9V\circ\phi(\tau)/\theta^2(\tau)\geq\b_V$; cf.
  (\ref{eq:exp norm constr}) and (\ref{eq:VprbyVbd}). Let $T$ be close enough to $-\infty$ that
  \[
    \big|(\ln V)''\circ\phi(\tau)\big|\leq \tfrac{\b_V}{2},\ \ \
    \big|(\ln V)''\circ\phi(\tau)\cdot(\ln V)'\circ\phi(\tau)\big|\leq \tfrac{\e\b_V}{4}
  \]
  and $|\phi(\tau)|\geq M$
  for all $\tau\leq T$. This means that if $|\phi_\tau+(\ln V)'\circ\phi|=\e$, then the absolute value of the first term on the far
  right hand side of (\ref{eq:dtau phitau minus vproverv}) is bounded from below by $\e\b_V/2$. Moreover, the absolute value of the second
  term on the far right hand side is bounded from above by $\e\b_V/4$. Since the first term has the opposite sign of
  $\phi_\tau+(\ln V)'\circ\phi$, it is clear that if $|\phi_\tau(\tau)+(\ln V)'\circ\phi(\tau)|=\e$ for some $\tau\leq T$, then
  $|\phi_\tau(\tau)+(\ln V)'\circ\phi(\tau)|\geq\e$ for all $\tau\leq T$. To conclude, either (\ref{eq:phi exc limit}) holds or there is
  an $\e>0$ and a $T$ such that $|\phi_\tau(\tau)+(\ln V)'\circ\phi(\tau)|\geq\e$ for all $\tau\leq T$. In the latter case, it follows from
  (\ref{eq:phi tau tau iso}) that $|\phi_{\tau\tau}|\geq 9\e V\circ\phi/\theta^2$ for all $\tau\leq T$. Since we can assume $V(\phi(\tau))>0$
  for $\tau\leq T$, this means that $\phi_\tau$ is strictly monotonic for $\tau\leq T$. This means that $\phi_\tau$ converges to a limit. If this
  limit is different from $\pm\sqrt{6}$, it follows that $9V\circ\phi/\theta^2$ converges to a non-zero limit; see (\ref{eq:exp norm constr}). Since
  $|\phi_{\tau\tau}|\geq 9\e V\circ\phi/\theta^2$ for all $\tau\leq T$, this means that $|\phi_\tau|$ tends to infinity, a conclusion which is
  inconsistent with (\ref{eq:exp norm constr}). Thus $V\circ\phi(\tau)/\theta^{2}(\tau)$ converges to zero as $\tau\rightarrow-\infty$ and all the
  conclusions of Lemma~\ref{lemma: Bianchi I isotropic} hold. The proposition follows.

  In order to prove Remark~\ref{remark:bounded in one direction}, note that we know that $|\phi(\tau)|\rightarrow\infty$ as $\tau\rightarrow-\infty$.
  If $\phi(\tau)\rightarrow-\infty$ as $\tau\rightarrow-\infty$, then $V\circ\phi$ is bounded to the past, so that
  $V\circ\phi(\tau)/\theta^2(\tau)\rightarrow 0$ as $\tau\rightarrow-\infty$. Thus all the conclusions of Lemma~\ref{lemma: Bianchi I isotropic} hold.
  If $\phi(\tau)\rightarrow\infty$ as $\tau\rightarrow\infty$, then the first part of the proof applies (alternately, one can modify the potential $V(s)$
  for $s\leq 0$ so that the conditions of the proposition apply). 
\end{proof}
As mentioned, we expect (\ref{eq:phi exc limit}) to correspond to an exceptional situation. Next, we verify that this is the case when the
potential is of the form $V(s)=V_0e^{\lambda s}$.
\begin{prop}\label{eq:gen isotropic Bianchi I exp pot}
  Let $0<V_0\in\rn{}$ and $0<\lambda<\sqrt{6}$. Define $V(s):=V_0e^{\lambda s}$. Let $\theta\in C^{\infty}(J,(0,\infty))$ and $\phi\in C^{\infty}(J,\rn{})$
  be a solution to (\ref{seq:EFEwrtvar})--(\ref{eq:phiddot}) corresponding to initial data in $B^{\iso}_{\mrI,+}$, where $J$ is the maximal existence
  interval. Then  $J$ can be assumed to equal $(0,\infty)$ and introducing $\tau$ as in (\ref{eq:dtdtau}), the range of $\tau$ is $\rn{}$. Finally,
  there are two possibilities concerning the asymptotics as $\tau\rightarrow-\infty$. Either
  $\phi_\tau=-\lambda$ for the entire solution or all the conclusions of Lemma~\ref{lemma: Bianchi I isotropic} apply. Moreover, there is a
  unique solution satisfying $\phi_\tau=-\lambda$. 
\end{prop}
\begin{remark}
  Since $B^{\iso}_{\mrI,+}$ is two dimensional, the condition $\phi_\tau=-\lambda$ is non-generic. 
\end{remark}
\begin{proof}
  Due to Lemma~\ref{lemma:BianchiAdevelopment}, $\theta(t)>0$ for all $t$ and $\theta$ is neither integrable to the future nor to the past. By
  introducing $\tau$ according to (\ref{eq:dtdtau}), the range of $\tau$ is thus $\rn{}$. Next, combining (\ref{eq:hamconfin}) and
  (\ref{eq:phi tau tau iso}) yields
  \begin{equation}\label{eq:phitau minus lambda}
    \d_\tau\phi_\tau=-\tfrac{9V\circ\phi}{\theta^2}(\phi_\tau+\lambda)=\tfrac{1}{2}(\phi_\tau-\sqrt{6})(\phi_\tau+\sqrt{6})(\phi_\tau+\lambda).
  \end{equation}
  In particular, $\phi_\tau$ converges to a limit belonging to $\{\pm\sqrt{6},-\lambda\}$. Due to (\ref{eq:thetad}), there is thus a constant $\b>0$ and
  a $T\in J$ such that $\theta_t\leq -\beta\theta^2$ for all $t\in J$ with $t\leq T$. This means that $\theta$ blows up in finite time to the past, so
  that $J$ can be assumed to equal
  $(0,\infty)$. Next, note that if $\phi_\tau$ converges to $\pm\sqrt{6}$, then all the conclusions of Lemma~\ref{lemma: Bianchi I isotropic} apply. What
  remains to be considered is thus that $\phi_\tau$ asymptotes to $-\lambda$. However, if $\phi_\tau\in (-\sqrt{6},-\lambda)$, then $\phi_{\tau\tau}>0$
  (meaning that $\phi_\tau$ decreases to the past), and if $\phi_\tau\in (-\lambda,\sqrt{6})$, then $\phi_{\tau\tau}<0$ (meaning the $\phi_\tau$ increases to
  the past). In order for $\phi_\tau$ to asymptote to $-\lambda$ to the past, we thus need to have
  $\phi_\tau=-\lambda$ for the entire solution. Then $\phi(\tau)=-\lambda\tau+\phi_0$ for some constant $\phi_0$. However, by
  a time translation, we can assume $\phi_0=0$. This means that $\ln\theta(\tau)=-\lambda^2\tau/2+\ln\theta_0$ for some constant $\theta_0>0$.
  Finally, in order for the constraint to be satisfied, we have to have $3=\lambda^2/2+9\theta_0^{-2}V_0$. 
  This means that there is, up to time translation, a unique solution with $\phi_\tau=-\lambda$. The proposition follows. 
\end{proof}

\begin{thm}\label{thm:asympt as exp pot}
  Assume $0\leq V\in C^{\infty}(\rn{})$ and that there are constants $C_V$ and $M$ such that $V(s)>0$ and
  \begin{equation}\label{eq:ln V biss est}
    \left|\left(\ln V\right)''(s)\right|\leq C_V\ldr{s}^{-2}
  \end{equation}
  for all $|s|\geq M$. This means that $(\ln V)'(s)$ converges to limits as $s\rightarrow\pm \infty$. Call the limits $\lambda_\pm$ and assume
  that $-\sqrt{6}<\lambda_-<0$ and that $0<\lambda_+<\sqrt{6}$. Let $\theta\in C^{\infty}(J,(0,\infty))$ and $\phi\in C^{\infty}(J,\rn{})$ be a
  solution to (\ref{seq:EFEwrtvar})--(\ref{eq:phiddot}) corresponding to initial data in $B^{\iso}_{\mrI,+}$, where $J=(t_-,t_+)$ is the maximal
  existence interval. Assuming that $\theta$ is unbounded, there are the following, mutually exclusive, cases:
  \begin{enumerate}[(i)]
  \item The solution is such that (\ref{eq:phi exc limit}) holds and $\phi(t)\rightarrow\infty$ as $t\rightarrow t_-$. Up to
    time translation, there is exactly one such solution, and its image is a smooth submanifold of $B^{\iso}_{\mrI,+}$.
  \item The solution is such that (\ref{eq:phi exc limit}) holds and $\phi(t)\rightarrow-\infty$ as $t\rightarrow t_-$. Up to
    time translation, there is exactly one such solution, and its image is a smooth submanifold of $B^{\iso}_{\mrI,+}$.
  \item All the conclusions of Lemma~\ref{lemma: Bianchi I isotropic} apply.
  \end{enumerate}
  In particular, excluding two exceptional solutions (which are also codimension one submanifolds), initial data in $B^{\iso}_{\mrI,+}$, whose developments
  have crushing singularities, give rise to solutions such that all the conclusions of Lemma~\ref{lemma: Bianchi I isotropic} apply.
\end{thm}
\begin{remark}\label{remark: only one direction V inf}
  If we replace the assumptions that $V(s)>0$ and that (\ref{eq:ln V biss est}) hold for $|s|\geq M$ with the assumption that these estimates hold only for
  $s\geq M$ and the assumption that $V$ is bounded for $s\leq 0$; and if we remove the assumption that $(\ln V)'(s)\rightarrow\lambda_-$ as
  $s\rightarrow -\infty$, then either \textit{(i)} or \textit{(iii)} have to hold. Finally, if $\bar{V}(s)=V(-s)$ satisfies these modified conditions,
  then either \textit{(ii)} or \textit{(iii)} have to hold.
\end{remark}
\begin{proof}
  Increasing $M$ appearing in the statement of the theorem, if necessary, there is an $\a_V\in (0,1)$ such that (\ref{eq:VprbyVbd}) and
  (\ref{eq:as exp or pol}) hold. This means that the conditions of Propositions~\ref{prop:iso theta unbounded} and \ref{prop:phi exc limit} are satisfied.
  We can therefore assume that (\ref{eq:phi exc limit}) holds and that $|\phi(\tau)|\rightarrow\infty$ as $\tau\rightarrow-\infty$. Due to
  (\ref{eq:phi exc limit}), $\phi_\tau(\tau)$ converges
  to either $-\lambda_+$ or to $-\lambda_-$ as $\tau\rightarrow-\infty$. In the first case, $\phi(\tau)$ tends to $\infty$ (essentially linearly),
  and in the second case, it tends to $-\infty$ (essentially linearly) as $\tau\rightarrow-\infty$. In particular, there is a time $T$ and a constant
  $C\geq 1$ such that $|\phi(\tau)|\geq M$ and 
  \[
  C^{-1}\ldr{\tau}\leq\ldr{\phi(\tau)}\leq C\ldr{\tau}
  \]
  for all $\tau\leq T$. Combining this estimate with (\ref{eq:ln V biss est}) yields
  \[
  \left|(\ln V)''\circ\phi(\tau)\right|\leq C\ldr{\tau}^{-2}
  \]
  for all $\tau\leq T$. Next, since $\phi_\tau(\tau)$ converges to $-\lambda_\pm$, it follows that $9V\circ\phi(\tau)/\theta^2(\tau)$ converges
  to $3-\lambda_\pm^2/2$; see (\ref{eq:exp norm constr}). Moreover, due to (\ref{eq:VprbyVbd}), $|\lambda_\pm|\leq \sqrt{6}\a_V$. This means
  that (\ref{eq:dtau phitau minus vproverv}) can be written
  \[
  \psi_\tau=-a\psi+b,
  \]
  where $a$ is a function converging to $3-\lambda_\pm^2/2>0$ as $\tau\rightarrow-\infty$ and $|b(\tau)|\leq C_b\ldr{\tau}^{-2}$ for all $\tau\leq T$
  and some constant $C_b$. Moreover, $\psi=\phi_\tau+(\ln V)'\circ\phi$. Since we know $\psi$ to converge to zero and since we can assume
  $a(\tau)\geq\a$ for all $\tau\leq T$, where $0<\a\in\rn{}$, this means that $|\psi(\tau)|\leq \a^{-1}C_b\ldr{\tau}^{-2}$ for all $\tau\leq T$.

  \textbf{A model for the asymptotics.} Next, we introduce a model for the asymptotic behaviour:
  \begin{equation}\label{eq:model solution iso}
    \xi_\tau+(\ln V)'\circ\xi=0.
  \end{equation}
  There is a unique (up to time translation) solution to this equation diverging to $-\infty$ as $\tau\rightarrow-\infty$, and there is a unique
  (up to time translation) solution diverging to $\infty$ as $\tau\rightarrow-\infty$. In order to justify the first statement, let $\xi_0$ be close
  enough to $-\infty$ that $(\ln V)'(s)\leq\lambda_-/2$ for all $s\leq\xi_0$. Consider the solution to (\ref{eq:model solution iso}) with $\xi(0)=\xi_0$.
  Since $\xi_\tau(\tau)\geq -\lambda_-/2>0$ if $\xi(\tau)\leq\xi_0$, it is clear that $\xi$ decreases to the past, so that $\xi(\tau)\leq\xi_0$ for all
  $\tau\leq 0$. Since $(\ln V)'(s)$ converges to $\lambda_-$ as $s\rightarrow-\infty$, it is clear that $\xi_\tau$ is bounded for $\tau\leq 0$. This
  means that $\xi$ exists globally to the past. This yields existence of the desired solution. In order to prove uniqueness (up to time translation),
  let $\xi_1$ and $\xi_2$ be two solutions. Then there are $M_i\in\rn{}$, $i=1,2$, such that the range of $\xi_i$ contains $(-\infty,M_i]$. Let
  $M_3\leq \min\{M_1,M_2\}$. Then there are $\tau_i$, $i=1,2$, in the existence intervals of $\xi_i$ such that $\xi_i(\tau_i)=M_3$. By uniqueness of
  solutions to (\ref{eq:model solution iso}), it follows that $\xi_{1}(\tau)=\xi_2(\tau+\tau_2-\tau_1)$. In other words, we have uniqueness up to
  time translation. The proof of the statement that there is a unique (up to time translation) solution diverging to $\infty$ as
  $\tau\rightarrow-\infty$ is similar.
  
  Fix a solution $\xi$ to (\ref{eq:model solution iso}) converging to $-\infty$ as $\tau\rightarrow-\infty$. Let $\max\{1,M\}\leq n_0\in\nn{}$ be such that
  $-n$ belongs to the range of $\xi$ for $n\geq n_0$. Assume, moreover, that $(\ln V)'(s)\leq\lambda_-/2$ for all $s\leq -n_0$ and note that $V(s)>0$ for
  $s\leq -n_0$. Next, let $\b_V:=(\a_V+1)/2$ and assume, in addition to the above, $n_0$ to be such that
  \begin{equation}\label{eq:nz restriction}
    C_V\ldr{n_0}^{-2}\leq 2(1-\b_V^2)(1-\b_V)<1-\b_V^2.
  \end{equation}
  Let $\tau_n$ be such that $\xi(\tau_n)=-n$. Then $\tau_n>\tau_{n+1}$ and $\tau_n\rightarrow-\infty$ as $n\rightarrow\infty$. By a time translation,
  it can be arranged that $\tau_{n_0}=2n_0/\lambda_-$. Then, for $\tau\leq\tau_{n_0}$,
  \[
  \xi(\tau_{n_{0}})-\xi(\tau)=-\textstyle{\int}_\tau^{\tau_{n_{0}}}(\ln V)'\circ\xi(s)ds\geq-\frac{\lambda_-}{2}(\tau_{n_{0}}-\tau)
  =-n_0+\frac{\lambda_-}{2}\tau.
  \]
  Since $\xi(\tau_{n_0})=-n_0$, this means that
  \begin{equation}\label{eq:xi upp bd}
    \xi(\tau)\leq -\lambda_-\tau/2
  \end{equation}
  for all $\tau\leq \tau_{n_0}$. In what follows, we use $\xi$ as a scaffold for constructing a solution as in \textit{(ii)}.
  
  \textbf{Existence of solutions.} We construct a solution as in \textit{(ii)} by taking a limit of a sequence of solutions. Moreover, the sequence is
  constructed by specifying the initial data at times closer and closer to $-\infty$ to equal those of the model solution (scaffold) constructed previously.
  More specifically, for $n\geq n_0$, let $\phi_n(\tau_n):=-n$ and $(\d_\tau\phi_n)(\tau_n):=\xi_\tau(\tau_n)$. Finally, define $\theta_n(\tau_n)>0$ by the
  requirement that (\ref{eq:exp norm constr}) hold. That this is possible follows from the fact that $(\d_\tau\phi_n)(\tau_n)^2\leq 6\a_V^2<6$. Let
  $(\phi_n,\theta_n)$ be the solution to (\ref{seq:EFEwrtvar})--(\ref{eq:phiddot}) corresponding to these data. Since $\theta$ is neither integrable to the
  future or to the past, see Lemma~\ref{lemma:BianchiAdevelopment}, the existence interval of the solution in $\tau$-time is $\rn{}$. Next, let
  \[
  \ma_n:=\{\tau\geq\tau_n:\phi_n(s)\leq-n_0,\ |(\d_\tau\phi_n)(s)|\leq \sqrt{6}\b_V\ \forall s\in [\tau_n,\tau]\}.
  \]
  Note that $\ma_n$ is closed, connected and non-emtpy ($\tau_n\in\ma_n$).
  Assume now that $\phi_n(\tau_*)<-n_0$ and that $\tau_*\in\ma_n$. If $|(\d_\tau\phi_n)(\tau_*)|<\sqrt{6}\b_V$, it is clear that $\ma_n$ contains an
  open neighbourhood of $\tau_*$ in $[\tau_n,\infty)$. Let us therefore assume that $|(\d_\tau\phi_n)(\tau_*)|=\sqrt{6}\b_V$. Let us first prove that
  $(\d_\tau\phi_n)(\tau_*)>0$. If not, $\d_\tau\phi_n$ has to have passed through $0$ in the interval $[\tau_n,\tau_*]$, since
  $(\d_\tau\phi_n)(\tau_n)=\xi_\tau(\tau_n)=-(\ln V)'(-n)>0$. Say that $(\d_\tau\phi_n)(\tau_a)=0$ for $\tau_n\leq\tau_a\leq\tau_*$ (so that
  $\tau_a\in\ma_n$) and that $\tau_a$ is the first element of $[\tau_n,\tau_*]$ with this property. Then (\ref{eq:exp norm constr}) and
  (\ref{eq:phi tau tau iso}) yield
  \[
  (\d_\tau^2\phi_n)(\tau_a)=-3(\ln V)'[\phi_n(\tau_a)]>0.
  \]
  On the other hand, since $(\d_\tau\phi_n)(\tau_a)=0$ and $(\d_\tau\phi_n)(s)\geq 0$ for $s\in [\tau_n,\tau_a]$, we have to have
  $(\d_\tau^2\phi_n)(\tau_a)\leq 0$. This yields a contradiction. To conclude, $\d_\tau\phi_n>0$ in $\ma_n$. This means that
  $(\d_\tau\phi_n)(\tau_*)=\sqrt{6}\b_V$. In particular,
  \[
  \tfrac{9V\circ\phi_n(\tau_*)}{\theta^2_n(\tau_*)}=3(1-\b_V^2),
  \]
  see (\ref{eq:exp norm constr}), and, due to (\ref{eq:VprbyVbd}), 
  \[
  \d_\tau\phi_n(\tau_*)+(\ln V)'[\phi_n(\tau_*)]\geq\sqrt{6}(\b_V-\a_V)=\sqrt{6}(1-\a_V)/2.
  \]
  Combining these observations with (\ref{eq:phi tau tau iso}) yields the conclusion that $\d_\tau^2\phi_n(\tau_*)<0$. Since $\tau_*>\tau_n$,
  $\d_\tau\phi_n(\tau)\leq\sqrt{6}\b_V$ for $\tau\in [\tau_n,\tau_*]$ and $\d_\tau\phi_n(\tau_*)=\sqrt{6}\b_V$, we, on the other hand, have
  $\d_\tau^2\phi_n(\tau_*)\geq 0$. This contradiction proves that $|(\d_\tau\phi_n)(\tau_*)|<\sqrt{6}\b_V$, so that $\ma_n$ contains an open
  neighbourhood of $\tau_*$. In particular, $\phi_n$ has to reach $-n_0$ before $|\d_\tau\phi_n|$ reaches $\sqrt{6}\b_V$. To conclude,
  $\phi_n$ is strictly increasing from $-n$ to $-n_0$ as $\tau$ increases. Moreover, $0<\d_\tau\phi_n<\sqrt{6}\b_V$ in the corresponding interval.
  Next, consider
  \[
  h_n:=\phi_n^2\left(\d_\tau\phi_n+(\ln V)'\circ\phi_n\right).
  \]
  Due to (\ref{eq:dtau phitau minus vproverv}),
  \[
  \d_\tau h_n= -\big(\tfrac{9V\circ\phi_n}{\theta_n^2}-2\tfrac{\d_\tau\phi_n}{\phi_n}-(\ln V)''\circ\phi_n\big)h_n
     -\phi_n^2\cdot(\ln V)''\circ\phi_n\cdot(\ln V)'\circ\phi_n.
  \]
  Note that $\phi_n\leq -n_0$ and that $\d_\tau\phi_n>0$ in $\ma_n$. Thus $\d_\tau\phi_n/\phi_n<0$ on $\ma_n$. Moreover, due to
  (\ref{eq:exp norm constr}), 
  \[
    \tfrac{9V\circ\phi_n(\tau)}{\theta^2_n(\tau)}-(\ln V)''\circ\phi_n(\tau)\geq 3(1-\b_V^2)-C_V\ldr{n_0}^{-2}>2(1-\b_V^2)
  \]
  for all $\tau\in \ma_n$, where we appealed to (\ref{eq:ln V biss est}) and (\ref{eq:nz restriction}) in the last two steps. Next,
  \[
  \left|\phi_n^2\cdot (\ln V)''\circ\phi_n\cdot (\ln V)'\circ\phi_n\right|\leq \sqrt{6}\a_VC_V.
  \]
  Moreover, $h_n(\tau_n)=0$ by definition. Combining these estimates, it is clear that
  \[
  |\phi_n^2(\d_\tau\phi_n+(\ln V)'\circ\phi_n)|=|h_n|\leq \tfrac{\sqrt{6}C_V}{2(1-\b_V^2)}
  \]
  in $\ma_n$. Denoting the right hand side by $C_\phi/2$, we conclude that
  \begin{equation}\label{eq:phin almost sol}
    \left|\d_\tau\phi_n+(\ln V)'\circ\phi_n\right|\leq C_\phi\ldr{\phi_n}^{-2}
  \end{equation}
  in $\ma_n$ (since $n_0\geq 1$). Next, let $\zeta_n:=\phi_n-\xi$ and compute
  \begin{equation}\label{eq:dtau zeta n}
    \d_\tau\zeta_n=\d_\tau\phi_n+(\ln V)'\circ\phi_n+(\ln V)'\circ\xi-(\ln V)'\circ\phi_n.
  \end{equation}
  The sum of the first two terms on the right hand side can be estimated by appealing to (\ref{eq:phin almost sol}). Consider
  \begin{equation}\label{eq:V prime over V diff}
    (\ln V)'\circ\xi-(\ln V)'\circ\phi_n=-\textstyle{\int}_0^1\left(\ln V\right)''(s\xi+(1-s)\phi_n)ds\cdot\zeta_n.
  \end{equation}
  Combining the above observations, there are functions $f_n$ and $g_n$ such that
  \begin{equation}\label{eq:zeta n evolution}
    \d_\tau\zeta_n=f_n\zeta_n+g_n.
  \end{equation}
  Let $\mB_n$ be the set of $\tau\in\ma_n\cap (-\infty,\tau_{n_0}]$ such that $|\zeta_n(s)|\leq 2$ for all $s\in[\tau_n,\tau]$. Then there
  is a constant $C_f'$, depending only on $C_V$ and $\a_V$, and a constant $C_f$, depending only on $C_V$, $\a_V$ and $\lambda_-$, such that 
  \[
  |f_n(\tau)|+|g_n(\tau)|\leq C_f'\ldr{\xi(\tau)}^{-2}\leq C_f'\ldr{-\lambda_-\tau/2}^{-2}\leq C_f\ldr{\tau}^{-2}
  \]
  for all $\tau\in\mB_n$, where we appealed to (\ref{eq:xi upp bd}) in the second to last step. Estimate
  \begin{equation}\label{eq:dtau zeta n energy}
    \d_{\tau}(\zeta_n^2+1)=2\zeta_n\d_\tau\zeta_n=2f_n\zeta_n^2+2g_n\zeta_n\leq 2C_f\ldr{\tau}^{-2}(\zeta_n^2+1)
  \end{equation}
  for all $\tau\in\mB_n$. Since $\zeta_n(\tau_n)=0$, this means that
  \[
  \zeta_n(\tau)^2+1\leq \exp\left(\textstyle{\int}_{\tau_n}^\tau 2C_f\ldr{s}^{-2}ds\right).
  \]
  Assuming $n_0$ to be large enough, the bound depending only on $C_V$, $\a_V$ and $\lambda_-$, we can assume the right hand side to be bounded
  from above by $2$ for all $\tau\leq\tau_{n_0}$. By a bootstrap argument, it can therefore be verified that $\mB_n=\ma_n\cap (-\infty,\tau_{n_0}]$.
  Moreover, $|\zeta_n(\tau)|\leq 1$ for all $\tau\in\mB_n$. 
  Recall that, by definition, $\xi(\tau_{n_0})=-n_0$ and that if $\sigma_n$ is the right end point of $\ma_n$, then $\phi_n(\sigma_n)=-n_0$.
  If $\tau_{n_0}\leq\sigma_n$, then $\mB_n=[\tau_n,\tau_{n_0}]$. If $\sigma_n\leq\tau_{n_0}$, then
  \[
  -1-n_0= -1+\phi_n(\sigma_n)\leq\xi(\sigma_n)\leq\xi(\tau_{n_0})=-n_0.
  \]
  Thus
  \[
  -\tfrac{\lambda_-}{2}(\tau_{n_0}-\sigma_n)\leq \textstyle{\int}_{\sigma_n}^{\tau_{n_0}}\xi_\tau(s)ds=\xi(\tau_{n_0})-\xi(\sigma_n)\leq 1,
  \]
  so that $u_{n_0}\leq\sigma_n$, where $u_{n_0}:=\tau_{n_0}+2/\lambda_-$. This means that $\mC_n:=[\tau_n,\tau_{n_0}+2/\lambda_-]\subset\mB_n$. For
  $k\geq n$, $\phi_k$ is uniformly bounded on $\mC_n$; $|\phi_k(\tau)|\leq 1+|\xi(\tau)|$. Moreover, since $|\d_\tau\phi_k|\leq\sqrt{6}\b_V$
  on $\mC_n$, it is clear that $|\d_\tau\phi_k|$, $\theta_k$ and $1/\theta_k$ are uniformly bounded on $\mC_n$; see
  (\ref{eq:exp norm constr}). Taking a subsequence, say $k_l$, it is clear that $\theta_{k_l}(u_{n_0})$, $\phi_{k_l}(u_{n_0})$ and
  $(\d_\tau\phi_{k_l})(u_{n_0})$ converge as $l\rightarrow\infty$. Denote the limits by $\theta_{0,*}$, $\phi_{0,*}$ and $\phi_{1,*}$. The limits satisfy
  the constraint equation (\ref{eq:exp norm constr}). Let $(\theta_*,\phi_*)$ be the solution to (\ref{seq:EFEwrtvar})--(\ref{eq:phiddot})
  corresponding to the initial data $\theta_{0,*}$, $\phi_{0,*}$ and $\phi_{1,*}$ at $u_{n_0}$. Due to continuous dependence on initial data,
  $(\theta_{k_l},\phi_{k_l},\d_\tau\phi_{k_l})$ converges uniformly to $(\theta_*,\phi_*,\d_\tau\phi_*)$ on any compact subset of $(-\infty,u_{n_0}]$.
  Combining this observation with (\ref{eq:phin almost sol}), it is clear that
  \[
    \left|\d_\tau\phi_*(\tau)+(\ln V)'\circ\phi_*(\tau)\right|\leq C_\phi\ldr{\phi_*(\tau)}^{-2}
  \]
  for all $\tau\leq u_{n_0}$. Note also that $|\phi_*(\tau)-\xi(\tau)|\leq 1$ for all $\tau\leq u_{n_0}$. 

  \textbf{Uniqueness of solutions.}
  Next, we wish to prove that $\phi_*$ is the unique (up to time translation) solution such that $\phi_*(\tau)\rightarrow-\infty$ as $\tau\rightarrow-\infty$
  and such that (\ref{eq:phi exc limit}) holds. Assume, to this end, that there are two such solutions, say $(\theta_a,\phi_a)$
  and $(\theta_b,\phi_b)$. By the arguments presented at the beginning of the proof, there are constants $C_a$ and $T_a$ such that
  \begin{equation}\label{eq:phi a asympt}
  \left|\d_\tau\phi_a(\tau)+(\ln V)'\circ\phi_a(\tau)\right|\leq C_a\ldr{\phi_a(\tau)}^{-2}
  \end{equation}
  for all $\tau\leq T_a$. An analogous statement holds for $\phi_b$. Since $\phi_a$ converges to $-\infty$, we can, without loss of
  generality, assume that $\phi_a(\tau)\leq -n_0$ and $\d_\tau\phi_a(\tau)\geq-\lambda_-/2$ for all $\tau\leq T_a$, and similarly for $\phi_b$.
  Thus, $\phi_a$ and $\phi_b$ can be assumed to be strictly decreasing to the past. Let
  \[
  \Phi_0\in \phi_a((-\infty,T_a])\cap \phi_b((-\infty,T_b])\cap \xi((-\infty,\tau_{n_0}]).
  \]
  Let $\tau_\Phi\in (-\infty,\tau_{n_0}]$ be such that $\xi(\tau_\Phi)=\Phi_0$ and modify $\phi_a$ and $\phi_b$ by a time translation so that 
  $\phi_a(\tau_\Phi)=\phi_b(\tau_\Phi)=\Phi_0$. Let $T:=\min\{T_a,T_b,\tau_{n_0}\}$. Introduce $\zeta_a:=\phi_a-\xi$
  and define $\zeta_b$ analogously. Then $\zeta_a$ and $\zeta_b$ satisfy (\ref{eq:dtau zeta n}) with $n$ replaced by $a$ and $b$ respectively.
  This means that they also satisfy (\ref{eq:zeta n evolution}) with $n$ replaced by $a$ and $b$ for suitable choices of $f_a$, $g_a$, $f_b$
  and $g_b$. Define $\mC_a$ to be the set of $\tau\leq T$ such that for all $s$ between $\tau$ and $\tau_\Phi$, $|\zeta_a(s)|\leq 1$. Note that
  $\mC_a$ is closed, connected and non-empty. Due to (\ref{eq:ln V biss est}), (\ref{eq:xi upp bd}), (\ref{eq:phi a asympt}),
  \[
  f_a(\tau):=-\textstyle{\int}_0^1(\ln V)''(s\xi(\tau)+(1-s)\phi_a(\tau))ds,
  \]
  and the fact that $|\zeta_a(\tau)|\leq 1$ for $\tau\in\mC_a$,
  \[
  |f_a(\tau)|+|g_a(\tau)|\leq C\ldr{\tau}^{-2}
  \]
  for $\tau\in \mC_a$, where $C$ only depends on $C_V$, $\lambda_-$ and $C_a$. In analogy with (\ref{eq:dtau zeta n energy}),
  \[
  |\d_\tau(\zeta_a^2+1)(\tau)|\leq 2C\ldr{\tau}^{-2}(\zeta_a^2+1)(\tau)
  \]
  for all $\tau\in\mC_a$. Integrating this estimate yields
  \begin{equation}\label{eq:zeta a energy esti}
    \zeta_a^2(\tau)+1\leq \exp\big(-\tfrac{2C}{\max\{\tau,\tau_\Phi\}}\big)
  \end{equation}
  for $\tau\in\mC_a$, where we used the fact that $\zeta_a(\tau_\Phi)=0$. Assuming $T$ to be close enough to $-\infty$ that $-2C/T<\ln 2$,
  a bootstrap argument can be used to conclude that $|\zeta_a(\tau)|\leq 1$ for $\tau\leq T$. An identical argument yields the same
  conclusion for $\zeta_b$. Moreover, (\ref{eq:zeta a energy esti}) and a similar estimate for $\zeta_b$ hold for $\tau\leq T$. This
  means, in particular, that there is a constant depending only on $C_V$, $\lambda_-$, $C_a$ and $C_b$ such that
  \[
  |\zeta_a(\tau)|+|\zeta_b(\tau)|\leq C\ldr{\tau_\Phi}^{-1/2}
  \]
  for all $\tau\leq\tau_\Phi$, where we used the fact that $|e^x-1|\leq |x|e^{|x|}$. This means, in particular, that if $\varphi:=\phi_a-\phi_b$, then
  $|\varphi(\tau)|\leq 2$ for all $\tau\leq T$ and $|\varphi(\tau)|\leq C\ldr{\tau_\Phi}^{-1/2}$ for all $\tau\leq \tau_\Phi$. Let $\{\Phi_l\}$ be a
  decreasing sequence
  such that $\Phi_1\leq \Phi_0-1$ and $\Phi_l\rightarrow-\infty$. Let $\tau_l$ be such that $\xi(\tau_l)=\Phi_l$. Since $\phi_a(\tau_\Phi)=\xi(\tau_\Phi)=\Phi_0$ and
  $\d_\tau\phi_a(\tau)\geq-\lambda_-/2$ for $\tau\leq T$, there is a $\tau_{a,l}\leq\tau_\Phi$ such that $\phi_a(\tau_{a,l})=\Phi_l$. Moreover,
  \[
  -\frac{\lambda_-}{2}|\tau_l-\tau_{a,l}|\leq |\phi_a(\tau_{a,l})-\phi_a(\tau_l)|=|\xi(\tau_l)-\phi_a(\tau_l)|\leq C\ldr{\tau_\Phi}^{-1/2}.
  \]
  Assuming $T$ to be close enough to $-\infty$, the bound depending only on $C_V$, $\lambda_-$ and $C_a$, we can assume $|\tau_l-\tau_{a,l}|\leq 1/2$.
  We can define $\tau_{b,l}$ similarly and prove that $|\tau_l-\tau_{b,l}|\leq 1/2$. Define
  \[
  \phi_{a,l}(\tau):=\phi_a(\tau+\tau_{a,l}-\tau_l),\ \ \
  \phi_{b,l}(\tau):=\phi_b(\tau+\tau_{b,l}-\tau_l).
  \]
  By taking subsequences, if necessary, we can assume $\tau_{a,l}-\tau_l$ and $\tau_{b,l}-\tau_l$ to converge to limits, say $\de_a$ and $\de_b$
  respectively. Define $\phi_{a,*}(\tau):=\phi_a(\tau+\de_a)$ and $\phi_{b,*}(\tau):=\phi_b(\tau+\de_b)$. Note that
  \begin{equation*}
    \begin{split}
      |\phi_{a,l}(\tau)-\phi_{b,l}(\tau)| \leq & |\phi_a(\tau+\tau_{a,l}-\tau_l)-\xi(\tau+\tau_{a,l}-\tau_l)|
      +|\xi(\tau+\tau_{a,l}-\tau_l)-\xi(\tau+\tau_{b,l}-\tau_l)|\\
      & +|\xi(\tau+\tau_{b,l}-\tau_l)-\phi_b(\tau+\tau_{b,l}-\tau_l)|\leq 2+\sqrt{6}\a_V
    \end{split}
  \end{equation*}
  for all $\tau\leq T-1$. In particular, there is a numerical bound on the difference. On the other hand, $\phi_{a,l}(\tau_l)=\phi_{b,l}(\tau_l)$.
  Note also that (\ref{eq:xi upp bd}), (\ref{eq:phi a asympt}), $|\zeta_a(\tau)|\leq 1$ for $\tau\leq T$ and $|\tau_l-\tau_{a,l}|\leq 1/2$ yield the
  conclusion that
  \begin{equation}\label{eq:dtau phia plus V prime over V}
    \left|\d_\tau\phi_{a,l}(\tau)+(\ln V)'\circ\phi_{a,l}(\tau)\right|\leq C\ldr{\tau}^{-2}
  \end{equation}
  for all $\tau\leq T-1$, where $C$ only depends on $C_a$ and $\lambda_-$. There is a similar estimate for $\phi_{b,l}$.
  Since $\phi_{a,l}(\tau_l)=\phi_{b,l}(\tau_l)$, this means that 
  \begin{equation}\label{eq:dtau varphi l ini}
    |(\d_\tau\phi_{a,l})(\tau_l)-(\d_\tau\phi_{b,l})(\tau_l)|\leq C\ldr{\tau_l}^{-2},
  \end{equation}
  where $C$ only depends on $C_a$, $C_b$ and $\lambda_-$. Next, we wish to compare $\phi_{a,l}$ and $\phi_{b,l}$ on $[\tau_l,T-1]$. 
  Note, to this end, that combining (\ref{eq:exp norm constr}) and (\ref{eq:phi tau tau iso}) yields
  \begin{equation}\label{eq:phi tau tau autonomous eq}
    \phi_{\tau\tau}=-\left(3-\phi_\tau^2/2\right)\left(\phi_\tau+(\ln V)'\circ\phi\right). 
  \end{equation}
  Since $\phi_{a,l}$ and $\phi_{b,l}$ both satisfy this equation, it can be deduced that $\varphi_l:=\phi_{a,l}-\phi_{b,l}$ satisfies
  \begin{equation}\label{eq:varphi l tau tau}
    \begin{split}
      \d_{\tau}^{2}\varphi_l = & (\d_\tau\phi_{a,l}/2+\d_\tau\phi_{b,l}/2)
      \left(\d_\tau\phi_{a,l}+(\ln V)'\circ\phi_{a,l}\right)\d_\tau\varphi_l\\
      & -\left(3-(\d_\tau\phi_{b,l})^2/2\right)
      \left[\d_\tau\varphi_l+(\ln V)'\circ\phi_{a,l}-(\ln V)'\circ\phi_{b,l}\right].
    \end{split}
  \end{equation}
  We can assume $T$ to be such that $|\d_\tau\phi_a(\tau)|\leq\sqrt{6}\b_V$ for $\tau\leq T$, and similarly for $\phi_b$.
  Finally, due to (\ref{eq:ln V biss est}), (\ref{eq:xi upp bd}), $|\zeta_a(\tau)|\leq 1$ for $\tau\leq T$,
  $|\zeta_b(\tau)|\leq 1$ for $\tau\leq T$, $|\tau_l-\tau_{a,l}|\leq 1/2$ and $|\tau_l-\tau_{b,l}|\leq 1/2$, 
  \begin{equation}\label{eq:V prime over V diff phi a l}
    \left|(\ln V)'\circ\phi_{a,l}(\tau)-(\ln V)'\circ\phi_{b,l}(\tau)\right|
    \leq C\ldr{\tau}^{-2}|\varphi_l(\tau)|
  \end{equation}
  for all $\tau\leq T-1$, where $C$ only depends on $C_V$ and $\lambda_-$ and we rewrite the difference in analogy with
  (\ref{eq:V prime over V diff}). Introduce
  \[
  \mfE_l:=\frac{1}{2}\ldr{\tau}^{2\a}(\d_\tau\varphi_l)^2+\frac{1}{2}\ldr{\tau}^{2\b}\varphi_l^2,
  \]
  where $\a,\b\geq 0$ are constants that remain to be determined. Note that since $\tau\leq 0$, $\d_\tau\ldr{\tau}^{2\a}\leq 0$ for $\a\geq 0$.
  Due to (\ref{eq:dtau phia plus V prime over V}); $|\d_\tau\phi_{a,l}(\tau)|\leq\sqrt{6}\b_V$ and $|\d_\tau\phi_{b,l}(\tau)|\leq\sqrt{6}\b_V$
  for $\tau\leq T-1$ (so that $3-(\d_\tau\phi_{b,l})^2/2>0$); (\ref{eq:varphi l tau tau}); and (\ref{eq:V prime over V diff phi a l})  
  it can be deduced that
  \begin{equation}\label{eq:dtau mfEl}
    \d_\tau\mfE_l(\tau)\leq -\g_V\ldr{\tau}^{2\a}(\d_\tau\varphi_l)^2
    +C\ldr{\tau}^{-2}\mfE_l(\tau)+C\ldr{\tau}^{\b-\a}\mfE_l(\tau)+C\ldr{\tau}^{2\a-2}|\d_\tau\varphi_l(\tau)||\varphi_l(\tau)|
  \end{equation}
  for $\tau\leq T-1$, where $C$ only depends on $C_a$, $C_b$, $C_V$ and $\lambda_-$. Here $\g_V:=3(1-\b_V^2)>0$.
  In order to estimate the last term on the right hand side of (\ref{eq:dtau mfEl}), note that
  \begin{equation*}
    \begin{split}
      & \sqrt{\g_V}\ldr{\tau}^{\a}|\d_\tau\varphi_l(\tau)|\cdot \g_V^{-1/2}C\ldr{\tau}^{\a-\b-2}\ldr{\tau}^{\b}|\varphi_l(\tau)|\\
      \leq & \frac{1}{2}\g_V\ldr{\tau}^{2\a}(\d_\tau\varphi_l)^2+\frac{1}{2}C^2\g_V^{-1}\ldr{\tau}^{2\a-2\b-4}\ldr{\tau}^{2\b}\varphi_l(\tau)^2.
    \end{split}
  \end{equation*}
  Combining the last two estimates yields
  \[
  \d_\tau\mfE_l(\tau)\leq C\ldr{\tau}^{-2}\mfE_l(\tau)+C\ldr{\tau}^{\b-\a}\mfE_l(\tau)+C\ldr{\tau}^{2\a-2\b-4}\mfE_l
  \]
  for $\tau\leq T-1$, where $C$ only depends on $C_a$, $C_b$, $C_V$, $\a_V$ and $\lambda_-$.
  In order for $\mfE_l$ not to change too much, we want
  to have $\b-\a<-1$ and $2\a-2\b-4<-1$. Since $\varphi_l(\tau_l)=0$ and (\ref{eq:dtau varphi l ini}) holds, we also know that
  $\mfE_l(\tau_l)\leq C\ldr{\tau_l}^{2\a-4}$, where $C$ only depends on $C_a$, $C_b$ and $\lambda_-$. In order for this to tend to zero as
  $l\rightarrow\infty$, we need to have $\a<2$. However, $\b=3/4$ and $\a=15/8$ satisfy all the above requirements. With these choices,
  \[
  \mfE_l(\tau_l)\leq C\ldr{\tau_l}^{-1/4},\ \ \
  \d_\tau\mfE_l(\tau)\leq C\ldr{\tau}^{-9/8}\mfE_l(\tau)
  \]
  for all $\tau\leq T-1$. This means that $\mfE_l(\tau)\leq C\ldr{\tau_l}^{-1/4}$ for all $\tau\in [\tau_l,T-1]$. Thus
  \[
  \lim_{l\rightarrow\infty}[\phi_{a,l}(T-1),(\d_\tau\phi_{a,l})(T-1)]=\lim_{l\rightarrow\infty}[\phi_{b,l}(T-1),(\d_\tau\phi_{b,l})(T-1)].
  \]
  This means that $\phi_{a,*}(T-1)=\phi_{b,*}(T-1)$ and similarly for the first derivative. Since $\phi_{a,*}$ and $\phi_{b,*}$ both satisfy
  (\ref{eq:phi tau tau autonomous eq}), it follows that they are equal. Thus $\phi_a$ is a time translation of $\phi_b$. That the
  image of the corresponding solution is a smooth submanifold follows by an argument identical to the one presented in the proof
  of Proposition~\ref{prop:iso theta bounded}. This finishes the proof of \textit{(ii)}. In order to prove \textit{(i)}, it is sufficient to
  note that introducing $\bar{\phi}:=-\phi$ and $\bar{V}(s):=V(-s)$, case \textit{(i)} can be reduced to case \textit{(ii)}.

  Next, let us prove the statements in Remark~\ref{remark: only one direction V inf}. Proposition~\ref{prop:iso theta unbounded} still applies
  and leads to the conclusion that $|\phi(\tau)|\rightarrow\infty$ as $\tau\rightarrow-\infty$. If $\phi(\tau)\rightarrow-\infty$, then, since
  $V(s)$ is bounded for $s\leq 0$, we conclude that $V\circ\phi(\tau)/\theta^2(\tau)\rightarrow 0$, so that all the conclusions of
  Lemma~\ref{lemma: Bianchi I isotropic} apply. If $\phi(\tau)\rightarrow\infty$, then we  can reduce the argument to the above (alternately, we
  can modify the potential in regimes that are irrelevant for the dynamics in such a way that the exact statement of the theorem can be applied). 
\end{proof}

\section[The $k=-1$ case]{The spatially homogeneous and isotropic setting; the case of negative spatial curvature}\label{section: k negative case}

Let $(\bM,\bge_{-})$ be a complete hyperbolic $3$-manifold with constant scalar curvature $-6$. This means that $\mathrm{Ric}[\bge_{-}]=-2\bge_{-}$.
Let $J$ be an open interval, $a\in C^\infty(J,(0,\infty))$, $\phi\in C^\infty(J)$ and consider a metric of the form
\begin{equation}\label{eq:hyperbolic version of metric}
  g:=-dt\otimes dt+a^{2}(t)\bge_{-}
\end{equation}
on $M:=\bM\times J$. Finally, let $V\in C^{\infty}(\ro)$. The Einstein non-linear scalar field equations for $(M,g,\phi)$ read
\begin{subequations}\label{seq:ESF hyperbolic setting}
  \begin{align}
    \tfrac{a_{tt}}{a} = & -\tfrac{1}{3}\phi_{t}^{2}+\tfrac{1}{3}V\circ\phi,\label{eq:att k minus one}\\
    \phi_{tt} = & -\theta\phi_{t}-V'\circ\phi,\label{eq:sca field eq k minus one}\\
    \tfrac{2}{3}\theta^{2} = & \phi_{t}^{2}+\tfrac{6}{a^{2}}+2V\circ\phi,\label{eq:hamcon k minus one}
  \end{align}  
\end{subequations}
where $\theta:=3a_{t}/a$ denotes the mean curvature; see \cite[(220), p.~214]{RinPL}. Given initial data
as in Definition~\ref{def:id k minus one}, there is a unique corresponding maximal development in the sense of Definition~\ref{def:k eq minus one dev}. In
fact, the following holds.

\begin{lemma}\label{lemma:ex un k minus one}
  Let $V\in C^{\infty}(\ro)$ and $\mfI=(\bM,\bge,\bk,\bphi_0,\bphi_1)\in\mN[V]$. Then there is a unique (up to time translation) spatially locally homogeneous
  and isotropic non-linear scalar field development of $\mfI$, say $(M,g,\phi)$ with $M=\bM\times J$ and $J=(t_-,t_+)$, which is maximal in the sense that
  either $t_\pm=\pm\infty$ or $|\theta|+|\phi|$ is unbounded as $t\rightarrow t_\pm$, where $\theta$ denotes the mean curvature of the constant-$t$
  hypersurfaces. Moreover, this development is globally hyperbolic. 

  If, in addition to the above, $V\geq 0$, $J$ can be assumed to equal $(0,\infty)$. Moreover, there is a unique $t_0\in J$ such that if $g$ is given by
  (\ref{eq:hyperbolic version of metric}), where $\bge_-$ is the constant multiple of $\bge$ with scalar curvature $-6$, then the initial data induced on
  $\bM_{t_0}:=\bM\times\{t_0\}$ by $(M,g,\phi)$ and pulled back by $\iota_{t_0}$ equal $\mfI$, where $\iota_{t_0}(p):=(p,t_0)$. Moreover, this solution
  has a crushing singularity at $t=0$ and is such that $\theta\notin L^{1}(0,1)$ and $\theta\notin L^{1}(1,\infty)$.
\end{lemma}
\begin{proof}
  Let $\a>0$ and $\b\geq 0$ be defined by the conditions that $\bge=\a^2\bge_-$ and $\bk=\a\b\bge_-$, where $\bge_-$ is a metric satisfying
  $\mathrm{Ric}[\bge_{-}]=-2\bge_{-}$. Let $(a,\phi)$ be the solution to (\ref{eq:att k minus one}) and (\ref{eq:sca field eq k minus one}) with initial data
  $a(0)=\a$, $a_t(0)=\b$, $\phi(0)=\bphi_0$ and $\phi_t(0)=\bphi_1$. Let $J=(t_-,t_+)$ be the maximal existence interval on which $a>0$. Define
  \begin{equation}\label{eq:Hdef k minus one}
    H:=\tfrac{2}{3}\theta^{2}-\phi_{t}^{2}-\tfrac{6}{a^{2}}-2V\circ\phi,
  \end{equation}
  where $\theta=3a_t/a$. It can then be verified that $H_t=-2\theta H/3$. Moreover, since $\mfI$ satisfies (\ref{eq:ham con id k minus one}), we know that
  (\ref{eq:hamcon k minus one}) holds with $\theta$, $a$, $\phi$ and $\phi_t$ replaced by $3\b/\a$, $\a$, $\bphi_0$ and $\bphi_1$ respectively, so that
  $H(0)=0$.
  Combining these two observations yields the conclusion that $H$ vanishes identically. Thus (\ref{eq:hamcon k minus one}) holds for all $t\in J$. To
  conclude, $(M,g,\phi)$, where $g$ is given by (\ref{eq:gSH bge minus}), is a solution to the Einstein non-linear scalar field equations. The uniqueness
  statement follows from the uniqueness (up to time translation) of solutions to  (\ref{eq:att k minus one}) and (\ref{eq:sca field eq k minus one}), given
  initial data, and the fact that the initial data for these equations are determined by $\mfI$. The fact that $(M,g)$ is a globally hyperbolic manifold
  follows from the fact that $(\bM,\bge_-)$ is a complete manifold and an argument similar to the one presented in the proof of
  \cite[Proposition~20.3, p.~215]{RinCauchy}. In order to prove the characterisation of $t_\pm$, assume that $t_+<\infty$ and that $\phi$ and $\theta$ are
  bounded to the future. Since $a_t/a=\theta/3$, it follows that $\ln a$ remains bounded to the future. Moreover, $\phi_t$ remains bounded to the future due
  to (\ref{eq:sca field eq k minus one}). To conclude, $\phi$, $\phi_t$, $\ln a$ and $\theta$ remain bounded to the future. This means that the solution can
  be extended to the future, contradicting the assumption that $J$ is the maximal existence interval. The argument in the opposite time direction is the same. 

  Assume now that, in addition to the above, $V\geq 0$. Then, by an argument similar to, but simpler than, the proof of
  Lemma~\ref{lemma:BianchiAdevelopment}, the existence interval of solutions can be assumed to be of the form $J=(t_-,\infty)$ where $t_->-\infty$. Moreover,
  $\theta(t)\rightarrow\infty$ as $t\downarrow t_-$, $\theta\notin L^{1}(t_-,t_-+1)$ and $\theta\notin L^{1}(t_-+1,\infty)$. By a translation in time, there is
  a unique $t_0>0$ such that the existence interval is $(0,\infty)$ and
  the original initial data are induced at $t_0$. This demonstrates existence. In order to prove uniqueness,
  assume that $a_i$, $\phi_i$ and $t_{0,i}$, $i=1,2$, satisfy the conditions of the lemma. Then, due to the arguments provided in connection with
  (\ref{seq:ESF hyperbolic setting}), we know that $a_i$ and $\phi_i$ both solve (\ref{seq:ESF hyperbolic setting}). Due to the requirement that
  the solutions should induce the correct initial data at $\bM_{t_{0,i}}$, we conclude that $a_{1}(t_{0,1})=a_2(t_{0,2})$ and
  $\phi_{1}(t_{0,1})=\phi_2(t_{0,2})$, and similarly for the time derivatives. Due to uniqueness of solutions of systems of ODE's, we conclude that
  the solution $(a_1,\phi_1)$ is simply a time translation of $(a_2,\phi_2)$. However, since both solutions have a crushing singularity at
  $t=0$, the time translation has to be trivial. The lemma follows. 
\end{proof}
Assume that all the conditions of Lemma~\ref{lemma:ex un k minus one} are fulfilled. In the special case of trivial initial data as in
Definition~\ref{def:id k minus one}, we conclude from (\ref{eq:sca field eq k minus one}) that $\phi\equiv\bphi_0$. Let
$\Lambda:=V(\bphi_0)$ and $H:=(\Lambda/3)^{1/2}$. Then (\ref{eq:att k minus one}) implies that
$a(t)=\a\sinh(Ht)$ for some constant $\a>0$ in case $\Lambda>0$ and that $a(t)=\a t$ for some constant $\a>0$ in case $\Lambda=0$. Combining this observation
with the requirement that (\ref{eq:hamcon k minus one}) hold yields the conclusion that $a(t)=a_\Lambda(t)$, where $a_\Lambda$ is defined by
(\ref{eq:a of t expl formula}). The statements of Remark~\ref{remark:Milne and Milne Lambda solns} follow. 

Given that all the assumptions of Lemma~\ref{lemma:ex un k minus one} are fulfilled, let $\tau$ be a time coordinate satisfying (\ref{eq:dtdtau}).
Then the range of $\tau$ equals $\ro$; cf. the last conclusions of Lemma~\ref{lemma:ex un k minus one}. Moreover, 
\begin{equation}\label{eq:theta prime etc}
  \theta'=-(1+q)\theta,\ \ \
  a(\tau)=a(0)e^{\tau},
\end{equation}
where $q:=\phi_{\tau}^{2}/3-3V\circ\phi/\theta^{2}$ and a prime denotes differentiation with respect to $\tau$. Next, we derive the asymptotics of solutions. 
\begin{prop}\label{prop:two outcomes as k minus one}
  Assume that all the conditions of Lemma~\ref{lemma:ex un k minus one} are fulfilled. Assume, in addition, that $V\in \mfP_{\a_V}^1$ for some
  $\alpha_V\in (0,1/3)$. Let $(a,\phi)$ be the solution obtained in Lemma~\ref{lemma:ex un k minus one}, defined on $J=(0,\infty)$. Then there are
  two possibilities. The first is that there is a $\phi_\infty\in\rn{}$ and a constant $C$ such that
  \begin{equation}\label{eq:phi asymptotics vacuum hyp t ver}
    |\phi_t(t)/\theta(t)|+|\phi(t)-\phi_\infty|\leq Ct^2
  \end{equation}
  holds for all $t\leq 1$. Then the metric asymptotes to a solution to the Einstein vacuum equations with a cosmological constant $\Lambda:=V(\phi_\infty)$
  in the direction of the singularity, cf. Remark~\ref{remark:Milne and Milne Lambda solns}, in the sense that there is a constant $C$ such that  
  \begin{subequations}\label{seq:a theta as k minus one dos}
    \begin{align}
      |a(t)-a_\Lambda(t)|\leq & Ct^5,\label{eq:a as vacuum setting}\\
      \left|\theta(t)-3\tfrac{\dot{a}_\Lambda(t)}{a_\Lambda(t)}\right|\leq & Ct^3\label{eq:theta as vacuum setting k minus one}
    \end{align}
  \end{subequations}
  for all $t\leq 1$. 
  
  The second possibility is that there is a $\Phi_{1}=\pm\sqrt{2/3}$, a $\Phi_0$, an $a_0>0$ and a constant $C$ such that
  \begin{equation}\label{eq:a phi as k minus one dos}
    |\theta^{p}a-a_0|+|\phi_t/\theta-\Phi_1|+|\phi+\Phi_1\ln\theta-\Phi_0|\leq C\theta^{-4/3}
  \end{equation}
  for all $t\leq 1$, where $p=1/3$. In this case, there is also a constant $C$ such that
  \begin{equation}\label{eq:theta as k minus one dos}
    |t\theta(t)-1|\leq Ct^{4/3}
  \end{equation}
  for all $t\leq 1$. 
\end{prop}
\begin{remark}
  The estimates (\ref{seq:a theta as k minus one dos}) can also be written
  \[
    \left|a(t)-t-\tfrac{\Lambda}{18}t^3\right|\leq Ct^4,\ \ \
    \left|\theta(t)-\tfrac{3}{t}-\tfrac{\Lambda}{3}t\right|\leq Ct^2
  \]
  for all $t\leq 1$.
\end{remark}
\begin{remark}
  When (\ref{eq:a phi as k minus one dos}) holds, the development induces data on the singularity in the sense of
  Definition~\ref{def:ind data on sing k negative}. Moreover, due to (\ref{eq:theta as k minus one dos}), the $\theta^{-4/3}$ appearing on
  the right hand side of (\ref{eq:a phi as k minus one dos}) can be replaced by $t^{4/3}$.
\end{remark}
\begin{remark}
  In the case of the Einstein scalar field equations, the only way for $\phi_t/\theta$ to converge to zero to the past is if $\phi_t$ is
  identically zero; due to (\ref{seq:ESF hyperbolic setting}), $\d_t(\phi_t/\theta)=-6\theta^{-1}a^{-2}\cdot(\phi_t/\theta)$. This means that
  the scalar field is constant and the relevant solution is the Milne model. 
\end{remark}
\begin{proof}
  Introduce a time coordinate $\tau$ according to (\ref{eq:dtdtau}) and note that (\ref{seq:ESF hyperbolic setting}) and (\ref{eq:theta prime etc}) yield
  $\theta(\tau)\geq c_\theta e^{-\tau}$ for some constant $c_\theta>0$.
  Define $\Psi$ and $\Omega$ as in (\ref{seq:PsiPhiphiiqdef}). Then $\Psi$ satisfies (\ref{eq:Psieq}). At this stage, we can revisit the arguments
  presented in the proof of Theorem~\ref{thm:dichotomy}. This means that we have two possibilities. The first one is that
  (\ref{eq:Omegaphiprstrongdec}) holds. Thus $q$ decays exponentially, so that there is a $\theta_\infty>0$ such that
  $e^{\tau}\theta(\tau)\rightarrow \theta_\infty$ as $\tau\rightarrow-\infty$. This means that $|\phi_\tau(\tau)|\leq Ce^{2\tau}$, so that there is a
  $\phi_\infty$ such that
  \begin{equation}\label{eq:phi asymptotics vacuum hyp}
    |\phi_\tau(\tau)|+|\phi(\tau)-\phi_\infty|\leq Ce^{2\tau}
  \end{equation}
  holds for all $\tau\leq 0$. Combining this information with (\ref{seq:ESF hyperbolic setting}) and (\ref{eq:theta prime etc}) yields the
  conclusion that
  \begin{equation}\label{eq:theta as vac k minus one}
    \left|\theta^{2}(\tau)-9/a^{2}(\tau)-3V(\phi_\infty)\right|\leq Ce^{2\tau}
  \end{equation}
  for all $\tau\leq 0$. This yields
  \begin{equation}\label{eq:asymptotic metric vacuum hyp}
    g=-9[\theta(\tau)]^{-2}d\tau\otimes d\tau+a^{2}(\tau)\bge_{-}=-3\alpha e^{2\tau}[1+\a\Lambda e^{2\tau}+\psi(\tau)]^{-1}d\tau\otimes d\tau+3\alpha e^{2\tau}\bge_{-},
  \end{equation}
  where $\alpha=a^{2}(0)/3$ and $\psi\in C^{\infty}(\rn{})$ is such that $|\psi(\tau)|\leq Ce^{4\tau}$ for some constant $C$ and all $\tau\leq 0$.
  Considering proper time as a function of $\tau$, it follows that
  \[
    \tfrac{dt}{d\tau}=\sqrt{3\alpha} e^{\tau}[1+\a\Lambda e^{2\tau}+\psi(\tau)]^{-1/2}
    =\sqrt{3\alpha} e^{\tau}-\tfrac{1}{2}\sqrt{3}\a^{3/2}\Lambda e^{3\tau}+O(e^{5\tau}).
  \]
  Since $t(\tau)\rightarrow 0$ as $\tau\rightarrow-\infty$, this means that
  \begin{equation}\label{eq:t of tau vacuum k minus one}
    t(\tau)=\sqrt{3\alpha} e^{\tau}-\tfrac{1}{6}\sqrt{3}\a^{3/2}\Lambda e^{3\tau}+O(e^{5\tau}).
  \end{equation}
  This relation implies that
  \[
    a(t)=\sqrt{3\a}e^{\tau(t)}=t+\tfrac{\Lambda}{18}t^3+O(t^5)
  \]
  so that (\ref{eq:a as vacuum setting}) holds. Combining this with (\ref{eq:theta as vac k minus one}) yields 
  \[
    \left|\theta(t)-\tfrac{3}{t}-\tfrac{\Lambda}{3}t\right|\leq Ct^3
  \]
  for all $t\leq 1$. Thus (\ref{eq:theta as vacuum setting k minus one}) holds. Finally, combining (\ref{eq:phi asymptotics vacuum hyp}) and
  (\ref{eq:t of tau vacuum k minus one}) yields (\ref{eq:phi asymptotics vacuum hyp t ver}).

  The second possibility is that $F_0$, introduced in the proof of Theorem~\ref{thm:dichotomy}, is non-zero. This means that $2-q$ is integrable to
  the past; see (\ref{eq:qmtint}). There is thus a $\theta_\infty>0$ such that $e^{3\tau}\theta(\tau)\rightarrow\theta_\infty$ as
  $\tau\rightarrow-\infty$, so that
  \begin{equation}\label{eq:Omega as matter hyp}
    \Omega(\tau)\leq Ce^{6(1-\a_V)\tau}
  \end{equation}
  for all $\tau\leq 0$. Combining this with (\ref{seq:ESF hyperbolic setting}) yields, recalling $6(1-\a_V)>4$ (since $\a_V<1/3$),
  \begin{equation}\label{eq:phitau asymp matter hyp}
    \big||\phi_\tau(\tau)|-\sqrt{6}\big|\leq Ce^{4\tau}
  \end{equation}
  for all $\tau\leq 0$, so that
  \begin{equation}\label{eq:phi phitau as hyp}
    |\phi_\tau(\tau)-\Phi_1|+|\phi(\tau)-\Phi_1\tau-\Phi_0|\leq Ce^{4\tau}
  \end{equation}
  holds for all $\tau\leq 0$. Due to the definition of $q$, (\ref{eq:Omega as matter hyp}) and
  (\ref{eq:phitau asymp matter hyp}) imply that $|q(\tau)-2|$ is bounded by the right hand side of (\ref{eq:phitau asymp matter hyp}).
  Combining this estimate with $\theta'=-(1+q)\theta$ yields a $\theta_\infty>0$ and a constant $C$ such that 
  \begin{equation}\label{eq:theta q as hyp}
    |\ln\theta(\tau)+3\tau-\ln\theta_\infty|+|q(\tau)-2|\leq Ce^{4\tau}
  \end{equation}
  holds for all $\tau\leq 0$, where $q$ is defined in connection with (\ref{eq:theta prime etc}). Since $a(\tau)=a(0)e^{\tau}$, the estimates
  (\ref{eq:phi phitau as hyp}) and (\ref{eq:theta q as hyp}) yield (\ref{eq:a phi as k minus one dos}). Due to (\ref{eq:dtdtau}) and
  (\ref{eq:theta q as hyp}), one can finally verify that (\ref{eq:theta as k minus one dos}) holds. 
\end{proof}

\subsection{Specifying the asymptotics}\label{ssection:spec as Milne k minus one}  

In Subsection~\ref{subsection:generic behaviour k minus one} below, we prove that the first outcome in Proposition~\ref{prop:two outcomes as k minus one}
is non-generic, and in Remark~\ref{remark:Milne and Milne Lambda solns}, we prove that it is possible. Our next goal is to prove that we can parametrise
solutions with the property that $\phi(t)\rightarrow\phi_\infty$ as $t\downarrow 0$ by the limit. We begin by reformulating the relevant equations. 

\textbf{Equations.} Assume that all the conditions of Lemma~\ref{lemma:ex un k minus one} are satisfied. Then the equations for $\phi$ and $\theta$,
i.e. (\ref{eq:sca field eq k minus one}) and the first equation in (\ref{eq:theta prime etc}), can be written
\begin{subequations}\label{seq:phi vartheta version hyp}
  \begin{align}
    \phi_{\tau\tau} = & (q-2)\phi_\tau-9\vartheta^2 V'\circ\phi,\label{eq:phi vartheta ver}\\
    \vartheta' = & (1+q)\vartheta,\label{eq:vartheta vartheta ver}
  \end{align}
\end{subequations} 
where $\vartheta=1/\theta$ and $q:=\phi_{\tau}^{2}/3-3\vartheta^2V\circ\phi$. Here the initial data have to be such that
\begin{equation}\label{eq:vartheta ini cond}
  \tfrac{1}{2}\phi_\tau^2+9\vartheta^2 V\circ\phi<3;
\end{equation}
this condition is required to ensure that there is an $a>0$ such that (\ref{eq:hamcon k minus one}) holds. Given $\phi_\infty\in\ro$,
(\ref{seq:phi vartheta version hyp}) can be written
\begin{subequations}\label{seq:phi vartheta version hyp dos}
  \begin{align}
    \varphi_0' = & e^{-2\tau}\varphi_1,\label{eq:varphiz prime}\\
    \varphi_1' = & q\varphi_1-9e^{2\tau}\vartheta^2V'\circ\phi_0,\label{eq:phi vartheta ver dos}\\
    r' = & q,\label{eq:vartheta vartheta ver dos}
  \end{align}
\end{subequations}
where
\[
  \phi_0:=\phi,\ \ \
  \varphi_0:=\phi_0-\phi_\infty,\ \ \
  \varphi_1:=e^{2\tau}\phi_\tau,\ \ \
  r:=\ln\vartheta-\tau.
\]
In particular, in (\ref{seq:phi vartheta version hyp dos}), it is understood that 
\begin{equation}\label{eq:phiz varth q def}
  \phi_0=\varphi_0+\phi_\infty,\ \ \
  \vartheta=e^{r+\tau},\ \ \
  q=e^{-4\tau}\varphi_1^2/3-3\vartheta^2V\circ\phi_0.
\end{equation}
Next, we prove Proposition~\ref{prop:Milne lambda as k minus one}. 
\begin{proof}[Proof of Proposition~\ref{prop:Milne lambda as k minus one}]
  We begin by proving existence.

  \textbf{Existence.} Fix $\tau_0\leq 0$ and let $\mC_{\tau_0}$ denote the set of $\varphi_{0},\varphi_{1},r\in C^0(-\infty,\tau_0]$, $i=1,2$, such that 
  \begin{equation}\label{eq:varphi etc bds mC tauz}
    |\varphi_{0}(\tau)|\leq e^{\tau/2},\ \ \
    |\varphi_{1}(\tau)|\leq e^{3\tau},\ \ \
    |r(\tau)|\leq e^{\tau}
  \end{equation}
  for all $\tau\leq \tau_0$. Given $x_i=(\varphi_{0,i},\varphi_{1,i},r_i)\in \mC_{\tau_0}$, $i=1,2$, define
  \begin{equation*}
    \begin{split}
      d_{\tau_0}(x_1,x_2) := & \textstyle{\sup}_{\tau\leq \tau_0}[e^{-\tau/2}|\varphi_{0,1}(\tau)-\varphi_{0,2}(\tau)|]
      +\sup_{\tau\leq \tau_0}[e^{-3\tau}|\varphi_{1,1}(\tau)-\varphi_{1,2}(\tau)|]\\
      & +\textstyle{\sup}_{\tau\leq \tau_0}[e^{-\tau}|r_{1}(\tau)-r_{2}(\tau)|].
    \end{split}
  \end{equation*}
  Note that $(\mC_{\tau_0},d_{\tau_0})$ is a complete metric space.  Given $x_a=(\varphi_{0,a},\varphi_{1,a},r_a)\in \mC_{\tau_0}$, define
  $\varphi_{0,b}$, $\varphi_{1,b}$, $r_b\in C^0(-\infty,\tau_0]$ by
  \begin{subequations}\label{seq:varphi b def}
    \begin{align}
      \varphi_{0,b}(\tau) = & \textstyle{\int}_{-\infty}^{\tau}e^{-2s}\varphi_{1,a}(s)ds,\label{eq:varphizbdef}\\
      \varphi_{1,b}(\tau) = & \textstyle{\int}_{-\infty}^{\tau}[q_a(s)\varphi_{1,a}(s)-9e^{2s}\vartheta_a^2(s)V'\circ\phi_{0,a}(s)]ds,\label{eq:varphioneb def}\\
      r_b(\tau) = & \textstyle{\int}_{-\infty}^{\tau}q_a(s)ds.\label{eq:rb def}
    \end{align}
  \end{subequations}
  Here $\phi_{0,a}$, $\vartheta_a$ and $q_a$ are defined by (\ref{eq:phiz varth q def}) in which $\varphi_0$, $r$ and $\varphi_1$ are replaced by
  $\varphi_{0,a}$, $r_a$ and $\varphi_{1,a}$ respectively. In what follows, we denote $(\varphi_{0,b},\varphi_{1,b},r_b)$ by $x_b$, and we denote the
  map taking $x_a\in \mC_{\tau_0}$ to $x_b$ by $F$; i.e., $x_b=F(x_a)$. Since $|\varphi_{1,a}(\tau)|\leq e^{3\tau}$, (\ref{eq:varphizbdef}) implies that
  \begin{equation}\label{eq:varphizb}
    |\varphi_{0,b}(\tau)|\leq e^{\tau}
  \end{equation}
  for all $\tau\leq\tau_0$. Let
  \[
    M_V:=\textstyle{\sup}_{s\in [-1,1]}|V(\phi_\infty+s)|+\sup_{s\in [-1,1]}|V'(\phi_\infty+s)|+\sup_{s\in [-1,1]}|V''(\phi_\infty+s)|.
  \]
  We can then estimate that the integrand in (\ref{eq:varphioneb def}) is bounded in absolute value by $4C_\varphi e^{4s}$, where $C_\varphi$ only depends
  on $M_V$. Similarly, $|q_a(\tau)|\leq 2C_qe^{2\tau}$, where $C_q$ only depends on $M_V$. This means that
  \begin{equation}\label{eq:varphioneb rbb}
    |\varphi_{1,b}(\tau)|\leq C_\varphi e^{4\tau},\ \ \
    |r_b(\tau)|\leq C_qe^{2\tau}
  \end{equation}
  for all $\tau\leq\tau_0$. From (\ref{eq:varphizb}) and (\ref{eq:varphioneb rbb}), it follows that there is a $T_0\leq 0$, depending only on $M_V$, such
  that if $\tau_0\leq T_0$, then $x_b\in \mC_{\tau_0}$. For $\tau_0\leq T_0$, we thus conclude that $F$ is a map from $\mC_{\tau_0}$ to itself. We wish to prove
  that this map is a contraction for a suitable choice of $\tau_0$.
  
  Let $x_a,x_c\in \mC_{\tau_0}$, $x_b:=F(x_a)$ and $x_d:=F(x_c)$. Then, using notation analogous to the above, 
  \begin{equation}\label{eq:varphi zero bd diff}
    |\varphi_{0,b}(\tau)-\varphi_{0,d}(\tau)|\leq \textstyle{\int}_{-\infty}^\tau e^{-2s}|\varphi_{1,a}(s)-\varphi_{1,c}(s)|ds\leq e^{\tau}d_{\tau_0}(x_a,x_c)
  \end{equation}
  for all $\tau\leq\tau_0$. Before proceeding to the other components, estimate, for $k\in\{0,1\}$,
  \begin{equation*}
    \begin{split}
      |V^{(k)}[\phi_{0,a}(\tau)]-V^{(k)}[\phi_{0,c}(\tau)]| = & \left|\textstyle{\int}_{0}^{1}V^{(k+1)}[s\phi_{0,a}(\tau)+(1-s)\phi_{0,c}(\tau)]ds\cdot
        [\varphi_{0,a}(\tau)-\varphi_{0,c}(\tau)]\right|\\
      \leq & M_V|\varphi_{0,a}(\tau)-\varphi_{0,c}(\tau)|\leq M_Ve^{\tau/2}d_{\tau_0}(x_a,x_c)
    \end{split}
  \end{equation*}
  for $\tau\leq\tau_0$. We also estimate
  \begin{equation*}
    \begin{split}
      |\vartheta_a^2(\tau)-\vartheta_c^2(\tau)| \leq & [\vartheta_a(\tau)+\vartheta_c(\tau)]\vartheta_c(\tau)|e^{r_a(\tau)-r_c(\tau)}-1|\\
      \leq & 2e^4 e^{2\tau}|r_a(\tau)-r_c(\tau)|\leq 2e^4 e^{3\tau}d_{\tau_0}(x_a,x_c)
    \end{split}
  \end{equation*}
  for $\tau\leq\tau_0$. Due to estimates of this type, it follows that
  \begin{equation}\label{eq:qa minus qc est}
    |q_a(\tau)-q_c(\tau)|\leq 2K_qe^{2\tau}d_{\tau_0}(x_a,x_c)
  \end{equation}
  for all $\tau\leq\tau_0$, where $K_q$ only depends on $M_V$. Similarly
  \begin{equation}\label{eq:varphi one diff est}
    \begin{split}
      & \left|[q_a(\tau)\varphi_{1,a}(\tau)-9e^{2\tau}\vartheta_a^2(\tau)V'\circ\phi_{0,a}(\tau)]
        -[q_c(\tau)\varphi_{1,c}(\tau)-9e^{2\tau}\vartheta_c^2(\tau)V'\circ\phi_{0,c}(\tau)]\right|\\
      \leq & 4K_\varphi e^{4\tau}d_{\tau_0}(x_a,x_c)
    \end{split}
  \end{equation}
  for all $\tau\leq\tau_0$, where $K_\varphi$ only depends on $M_V$. Combining (\ref{eq:rb def}) with (\ref{eq:qa minus qc est}) yields
  \begin{equation}\label{eq:r bd diff}
    |r_b(\tau)-r_d(\tau)|\leq \textstyle{\int}_{-\infty}^{\tau}|q_a(s)-q_c(s)|ds\leq K_qe^{2\tau}d_{\tau_0}(x_a,x_c)
  \end{equation}
  for all $\tau\leq\tau_0$. Similarly, combining (\ref{eq:varphioneb def}) and (\ref{eq:varphi one diff est}) yields
  \begin{equation}\label{eq:varphi one bd diff}
    |\varphi_{1,b}(\tau)-\varphi_{1,d}(\tau)|\leq K_\varphi e^{4\tau}d_{\tau_0}(x_a,x_c)
  \end{equation}
  for all $\tau\leq\tau_0$. Combining (\ref{eq:varphi zero bd diff}), (\ref{eq:r bd diff}) and (\ref{eq:varphi one bd diff}) yields
  \[
    d_{\tau_0}[F(x_a),F(x_c)]=d_{\tau_0}(x_b,x_d)\leq Ke^{\tau_0/2}d_{\tau_0}(x_a,x_c).
  \]
  There is thus a $T_0\leq 0$ such that for $\tau_0\leq T_0$, $F$ maps $\mC_{\tau_0}$ to $\mC_{\tau_0}$ and, for $x_a,x_c\in\mC_{\tau_0}$,
  \[
    d_{\tau_0}[F(x_a),F(x_c)]\leq\tfrac{1}{2}d_{\tau_0}(x_a,x_c).
  \]
  By the Banach fixed point theorem, we conclude that $F$ has a unique fixed point in $\mC_{\tau_0}$. This means that there are unique continuous
  functions $\varphi_0$, $\varphi_1$ and $r$ for $\tau\leq\tau_0$, satisfying the bounds (\ref{eq:varphi etc bds mC tauz}) and the equations
  \begin{subequations}\label{seq:varphi eqs}
    \begin{align}
      \varphi_{0}(\tau) = & \textstyle{\int}_{-\infty}^{\tau}e^{-2s}\varphi_{1}(s)ds,\label{eq:varphiz fin}\\
      \varphi_{1}(\tau) = & \textstyle{\int}_{-\infty}^{\tau}[q(s)\varphi_{1}(s)-9e^{2s}\vartheta^2(s)V'\circ\phi_{0}(s)]ds,\label{eq:varphione fin}\\
      r(\tau) = & \textstyle{\int}_{-\infty}^{\tau}q(s)ds.\label{eq:r fin}
    \end{align}
  \end{subequations}
  From these equalities, it follows that the components of the fixed point are smooth for $\tau\leq\tau_0$. Define $\theta(\tau)=e^{-r(\tau)-\tau}$ and
  $\phi(\tau)=\varphi_0(\tau)+\phi_\infty$. Then, due to (\ref{eq:varphi etc bds mC tauz}) and (\ref{eq:varphiz fin}),
  \[
    |\phi(\tau)-\phi_\infty|\leq e^{\tau/2},\ \ \
    |e^\tau\theta(\tau)-1|\leq Ce^\tau
  \]
  for $\tau\leq\tau_0$. Due to (\ref{seq:varphi eqs}), it can also be verified that
  \begin{equation}\label{eq:phi tau eq theta tau eq}
    \phi_{\tau\tau}=(q-2)\phi_\tau-9\tfrac{V'\circ\phi}{\theta^2},\ \ \
    \theta_\tau=-(1+q)\theta.
  \end{equation}
  Moreover,
  \begin{equation}\label{eq:q in terms of phi tau etc}
    q=\tfrac{1}{3}\phi_\tau^2-\tfrac{3V\circ\phi}{\theta^2}.
  \end{equation}
  Next, introduce a time coordinate $t$ by
  \begin{equation}\label{eq:t ito tau def}
    t(\tau)=\textstyle{\int}_{-\infty}^{\tau}\tfrac{3}{\theta(s)}ds.
  \end{equation}
  Combining (\ref{eq:phi tau eq theta tau eq}), (\ref{eq:q in terms of phi tau etc}) and (\ref{eq:t ito tau def}) yields
  \begin{equation}\label{eq:theta t phi tt k negative}
    \theta_t=-\tfrac{1}{3}\theta^2-\phi_t^2+V\circ\phi,\ \ \
    \phi_{tt}=-\theta\phi_t-V'\circ\phi.
  \end{equation}
  Next, define $a$ by the condition that
  \begin{equation}\label{eq: at k negative}
    a_t=\tfrac{1}{3}\theta a,
  \end{equation}
  the condition that $a$ be strictly positive and the condition that
  \begin{equation}\label{eq:ham con k negative}
    \tfrac{2}{3}\theta^{2}=\phi_{t}^{2}+\tfrac{6}{a^{2}}+2V\circ\phi
  \end{equation}
  hold at some $t=t_0$. Given this definition of $a$, define $H$ by (\ref{eq:Hdef k minus one}). It can then be verified that $H_t=-2\theta H/3$.
  Since $H(t_0)=0$, we conclude that $H$ vanishes identically, so that (\ref{eq:ham con k negative}) holds for all $t$. Finally, combining
  (\ref{eq:theta t phi tt k negative}), (\ref{eq: at k negative}) and (\ref{eq:ham con k negative}), it follows that the metric $g$ defined by
  (\ref{eq:hyperbolic version of metric}) and the scalar field $\phi$ satisfy the Einstein scalar field equations; see (\ref{seq:ESF hyperbolic setting}).

  \textbf{Uniqueness.} In order to prove uniqueness, assume $(a,\phi)$ to satisfy the conditions of the proposition. In particular,
  $\phi(t)\rightarrow\phi_\infty$ as $t\downarrow 0$. Let $S$ be the image of $(0,1)$ under $\phi$. Then $S$ is bounded and its closure is compact.
  We can modify $V$, if necessary, outside an open neighbourhood containing the closure of $S$ so that $V\in\mfP_{\a_V}^1$ for some $\a_V\in (0,1/3)$.
  As far as the asymptotics in the direction of the singularity are concerned, we can therefore, without loss of generality, assume that the conditions
  of Proposition~\ref{prop:two outcomes as k minus one} are satisfied. This means that the conclusions of this proposition concerning the first
  possibility in Proposition~\ref{prop:two outcomes as k minus one} (and the corresponding conclusions in the proof) hold. In particular, we can
  assume (\ref{eq:phi asymptotics vacuum hyp}) to hold. This means that $|\varphi_0(\tau)|\leq Ce^{2\tau}$ and $|\varphi_1(\tau)|\leq Ce^{4\tau}$
  for all $\tau\leq\tau_0$ and some suitably chosen $\tau_0$. By replacing the time coordinate $\tau$ by $\tau+\tau_1$ for some $\tau_1\in\ro$,
  if necessary, the estimate (\ref{eq:theta as vac k minus one}) implies that $e^{\tau}\theta(\tau)=1+O(e^{2\tau})$. As a consequence,
  $|r(\tau)|\leq Ce^{2\tau}$. In particular, for $\tau_0$ close enough to $-\infty$, $x:=(\varphi_0,\varphi_1,r)$ belongs to $\mC_{\tau_0}$. Since
  $\varphi_0$, $\varphi_1$ and $r$ satisfy (\ref{seq:phi vartheta version hyp dos}), it is clear that $x$ is a fixed point of the map $F$.
  By uniqueness of fixed points to $F$, we conclude that $x$ is the unique fixed point we already constructed. Since the original solution can
  be uniquely recovered from $x$ by the argument described at the end of the proof of existence, we conclude that the $(a,\phi)$ constructed
  in the existence part are unique.

  \textbf{Asymptotics.} The statements concerning the asymptotics follow from Proposition~\ref{prop:two outcomes as k minus one}. 
\end{proof}

\subsection{Initial data on the singularity}\label{ssection:idos k minus one}

Next, we prove Proposition~\ref{prop:dos ind dev k minus one}.

\begin{proof}[Proof of Proposition~\ref{prop:dos ind dev k minus one}]
  Let $(\bM,\msH,\msK,\bPhi_{1},\bPhi_{0})=\mfI_\infty$. By definition, $(\bM,\msH)$ is a complete hyperbolic $3$-manifold. In particular, there is a
  unique hyperbolic metric, say $\bge_-$, with scalar curvature $-6$ on $\bM$ such that $\msH=a_\infty^2\bge_-$ for some $a_\infty\in (0,\infty)$.
  In what follows, it is convenient to use the notation $\nu_0:=\ln a_\infty$. We also use the notation $\Phi_1:=3\bPhi_1$ and $\Phi_0:=\bPhi_0$. 

  \textbf{Existence.} We begin by proving existence of the desired development. Moreover, we do so by appealing to the abstract framework developed
  in Section~\ref{section:abstract setting dos}. However, this necessitates reformulating the equations. Note, to this end, that the Hamiltonian
  constraint (\ref{eq:hamcon k minus one}) can be reformulated to
  \[
    1=\tfrac{1}{6}\phi_\tau^2+A^2+\tfrac{3V\circ\phi}{\theta^2},
  \]
  where $A:=3/(\theta a)$. This means that we can rewrite
  \begin{equation}\label{eq:q diff ver k minus one}
    q=\tfrac{1}{3}\phi_\tau^2-\tfrac{3V\circ\phi}{\theta^2}=2-2A^2-\tfrac{9V\circ\phi}{\theta^2}.
  \end{equation}
  \textit{Variables.} The basic variables in the existence proof are $\nu$, $\varkappa$, $\psi_0$ and $\psi_1$. Given these variables and constants
  $\nu_0\in\ro$, $\Phi_0\in\ro$ and $\Phi_1\in\{-\sqrt{6},\sqrt{6}\}$ as above, we \textit{define} the following quantities:
  \begin{subequations}\label{seq:bth etc def}
    \begin{align}
      \bth := & e^{-3\tau+\varkappa},\label{eq:bth def}\\
      \bA := & 3e^{2\tau-\nu_0-\nu},\\
      \phi_1 := & \psi_1+\Phi_1,\\
      \phi_0 := & (\Phi_1+\psi_1)\tau+\Phi_0+\psi_0,\label{eq:phi zero def k minus one}\\
      \bq := & 2-2\bA^2-\tfrac{9V\circ\phi_0}{\bth^2}.\label{eq:bq def k minus one}
    \end{align}
  \end{subequations}
  Given these definitions, consider the following system of equations
  \begin{subequations}\label{seq:varkappa psiz psio prime etc k minus one}
    \begin{align}
      \nu' = & 2-\bq,\\
      \varkappa' = & 2-\bq,\\
      \psi_0' = & [(2-\bq)\phi_1+9\bth^{-2}V'\circ\phi_0]\tau,\\
      \psi_1' = & -(2-\bq)\phi_1-9\bth^{-2}V'\circ\phi_0.
    \end{align}
  \end{subequations}
  Our goal is to prove that there is a unique solution to (\ref{seq:varkappa psiz psio prime etc k minus one}) such that $x:=(\nu,\varkappa,\psi_0,\psi_1)$
  converges to zero exponentially. Given $\Pi:=(\nu_0,\Phi_0,\Phi_1)$, let, to this end, $\mX[x(\tau),\tau;\Pi]$ be defined by the right hand side of
  (\ref{seq:varkappa psiz psio prime etc k minus one}). Note also that, comparing with the framework developed Section~\ref{section:abstract setting dos},
  $k=4$ in our case. For a given $\e>0$ and a $\tau_0\leq 0$ we define $X_{\tau_0}$ and $d_{\tau_0}$ as at the beginning of Section~\ref{section:abstract setting dos}.
  We wish to prove that there is a choice of $\e>0$ such that (\ref{eq:mX bd and lip bd}) is satisfied.

  \textit{Existence, abstract setting.}  
  Let $x_a,x_b\in X_{\tau_0}$, $(\nu_a,\varkappa_a,\psi_{0,a},\psi_{1,a})=x_a$ and similarly for $x_b$. We also assume
  that $\bq_a$, $\bA_a$, $\phi_{1,a}$ etc. are obtained by inserting $\nu_a$ etc. in (\ref{seq:bth etc def}) and similarly for $\bq_b$ etc. In order to
  estimate $\mX[x_a(\tau),\tau;\Pi]$, note that 
  \begin{equation}\label{eq:qa minus two etc est}
    |\bq_a-2|+\tfrac{9|V'\circ\phi_{0,a}|}{\bth_a^2}
    \leq 2\bA_a^2+\tfrac{9|V\circ\phi_{0,a}|}{\bth_a^2}+\tfrac{9|V'\circ\phi_{0,a}|}{\bth_a^2}\leq C_Ae^{4\tau}+C_\Omega e^{6(1-\a_V)\tau}
  \end{equation}
  for all $\tau\leq\tau_0$, where $C_A$ only depends on $\nu_0$ and $C_\Omega$ only depends on $c_1$, $\e$ and $\Phi_0$; here $c_1$ is the
  constant appearing in (\ref{eq:V k-derivatives exp estimate}). From (\ref{eq:qa minus two etc est}), we deduce that
  \[
    |\mX[x_a(\tau),\tau;\Pi]|\leq C_{\mX,0}\ldr{\tau}e^{\b_V\tau}
  \]
  for all $\tau\leq \tau_0$, where $\b_V:=\min\{4,6(1-\a_V)\}$ and $C_{\mX,0}$ only depends on $\Pi$, $\e$ and $c_1$. This means that the first estimate in
  (\ref{eq:mX bd and lip bd}) is satisfied for any $\e<2\b_V/3$. Next, note that
  \[
    |\bA_a^2-\bA_b^2|\leq \bA_b(\bA_a+\bA_b)\left|\tfrac{\bA_a}{\bA_b}-1\right|\leq Ce^{4\tau}|e^{\nu_b-\nu_a}-1|\leq Ce^{4\tau}|\nu_b-\nu_a|,
  \]
  for $\tau\leq \tau_0$, where $C$ changes from line to line and only depends on $\nu_0$. Similarly,
  \[
    |\bth_a^{-2}-\bth_b^{-2}|\leq Ce^{6\tau}|\varkappa_a-\varkappa_b|
  \]
  for all $\tau\leq \tau_0$, where $C$ is a numerical constant. Finally, by equalities such as (\ref{eq:V xi minus V zeta}), if $l\in \{0,1\}$, then
  \[
    |V^{(l)}\circ\phi_{0,a}-V^{(l)}\circ\phi_{0,b}|\leq Ce^{-6\a_V\tau}[|\psi_{0,a}-\psi_{0,b}|+\ldr{\tau}|\psi_{1,a}-\psi_{1,b}|]
  \]
  for all $\tau\leq \tau_0$, where $C$ only depends on $c_2$, $\e$ and $\Phi_0$. Combining these observations with estimates such as
  (\ref{eq:theta a rt minus two Va diff}) yields the conclusion that
  \[
    |\mX[x_a(\tau),\tau;\Pi]-\mX[x_b(\tau),\tau;\Pi]|\leq C_{\mX,0}\ldr{\tau}^2e^{\b_V\tau}|x_a-x_b|
  \]
  for all $\tau\leq \tau_0$, where $C$ only depends on $c_2$, $\e$ and $\Pi$. This means the second estimate in (\ref{eq:mX bd and lip bd}) is satisfied for
  any $\e<\b_V$. To conclude, letting, e.g., $\e:=\b_V/2$, Lemma~\ref{lemma:abs exist and unique} applies and demonstrates that there is a unique solution to 
  (\ref{seq:varkappa psiz psio prime etc k minus one}) such that $e^{-\e\tau}|x(\tau)|\rightarrow 0$ as $\tau\rightarrow-\infty$.

  \textit{Existence, geometric setting.}  
  Given the unique solution to (\ref{seq:varkappa psiz psio prime etc k minus one}) constructed above, define $\bth$ by (\ref{eq:bth def}) and a new time
  coordinate $t$ by 
  \begin{equation}\label{eq:t of tau as k minus one dos}
    t(\tau):=\textstyle{\int}_{-\infty}^\tau3e^{3s-\varkappa(s)}ds=e^{3\tau}+O(e^{(3+\e)\tau}).
  \end{equation}
  Next, define $a$ by $a:=e^{\tau+\nu_0}e^{-\varkappa+\nu}$. It can then be verified that $\bth=3a_t/a$, so that we can interpret $\bth$ as the mean curvature
  of a metric of the form (\ref{eq:gSH k minus one}). For this reason, we, from now on, write $\theta=\bth$. These definitions yield
  \[
    \tfrac{3}{\theta a}=\tfrac{3}{\bth a}=3e^{3\tau-\varkappa}e^{-\tau-\nu_0}e^{\varkappa-\nu}=3e^{2\tau-\nu-\nu_0}=\bA.
  \]
  For this reason, we, from now on, write $\bA=A=3/(\theta a)$. Defining $\phi:=\phi_0$, see (\ref{eq:phi zero def k minus one}), we conclude that
  $\bq$ introduced in (\ref{eq:bq def k minus one}) coincides with the $q$ defined by the far right hand side of (\ref{eq:q diff ver k minus one}).
  For this reason, we, from now on, equate $\bq$ and $q$. Next, compute
  \begin{equation}\label{eq:dot theta k minus one}
    \theta_t=\theta'\tfrac{\theta}{3}=(-3+\varkappa')\theta\tfrac{\theta}{3}=-(1+q)\tfrac{\theta^2}{3}.
  \end{equation}
  Since $a_t=\theta a/3$, this equality implies that
  \begin{equation}\label{eq:ddot a k minus one}
    a_{tt}=\tfrac{1}{3}\theta_ta+\tfrac{1}{3}\theta a_t=-\tfrac{1}{9}q\theta^2 a.
  \end{equation}
  Next, note that $\phi'=\phi_1$, so that $\phi_t=\theta\phi_1/3$. Compute
  \[
    \d_t\phi_1=\tfrac{\theta}{3}\psi_1'=-(2-q)\phi_t-3\theta^{-1}V'\circ\phi.
  \]
  Combining this computation with (\ref{eq:dot theta k minus one}) yields
  \begin{equation}\label{eq:ddot phi deriv eqs k minus one}
    \phi_{tt}=\tfrac{1}{3}\theta_t\phi_1+\tfrac{1}{3}\theta\d_t\phi_1=-\theta\phi_t-V'\circ\phi.
  \end{equation}
  Next, consider $H$ defined by (\ref{eq:Hdef k minus one}). Given the equations that are satisfied here, it can be computed that
  $H_t=-2\theta H$. Introducing $\mH:=H/\theta^2$, it can then be calculated that
  \[
    \mH'=2(q-2)\mH.
  \]
  Since $q-2$ is integrable to the past, this means that if $\mH$ is ever non-zero, $\mH(\tau)$ converges to a non-zero number as
  $\tau\rightarrow-\infty$. On the other hand, we know that $\mH$ converges to $2/3-\Phi_1^2/9=0$ as $\tau\rightarrow-\infty$. This
  means that $\mH\equiv 0$, so that $H\equiv 0$ and (\ref{eq:hamcon k minus one}) holds. Due to (\ref{eq:ddot phi deriv eqs k minus one}),
  we also know that (\ref{eq:sca field eq k minus one}) holds. Finally, combining (\ref{eq:hamcon k minus one}) and
  (\ref{eq:ddot a k minus one}) yields (\ref{eq:att k minus one}). To conclude $a$ and $\phi$ satisfy (\ref{seq:ESF hyperbolic setting}).
  This means that we have a solution to the Einstein non-linear scalar field equations. It remains to be verified that this solution has
  the desired asymptotics. However, to begin with,  
  \[
    \theta^{-1}\phi_t=\tfrac{1}{3}\phi_1=\tfrac{1}{3}\Phi_1+O(t^{\e/3}),
  \]
  where we used that $\phi_1=\Phi_1+\psi_1$, that $\psi_1=O(e^{\e\tau})$ and that (\ref{eq:t of tau as k minus one dos}) holds. Next, note
  that $e^{3\tau}\theta(\tau)=e^{\varkappa}=1+O(e^{\e\tau})$, so that $t\theta=1+O(t^{\e/3})$. By the definition of $a$, it follows that
  \[
    \theta^{1/3}a=e^{\nu_0}[1+O(t^{\e/3})].
  \]
  Finally, similar arguments yield
  \[
    \phi+\theta^{-1}\phi_t\ln\theta=\Phi_0+O(t^{\e/3}).
  \]
  This concludes the proof of existence. 

  \textbf{Uniqueness.} 
  Next, let us assume that $(a,\phi)$ is a solution inducing the given data on the singularity. By assumption, the corresponding spacetime has
  a crushing singularity, so that $\theta(t)\rightarrow\infty$ as $t\downarrow t_-$. From now on, we fix a $t_0$ in the existence interval of the
  solution such that $\theta(t)\geq 1$ for all $t\leq t_0$. Moreover, for $l\in\{0,1,2\}$, we estimate
  \begin{equation}\label{eq:Vl db th sq bd k minus one}
    \tfrac{|V^{(l)}\circ\phi(t)|}{\theta^2(t)}\leq C\exp\left[-2(1-\a_V)\ln\theta\right]
  \end{equation}
  for all $t\leq t_0$ and some constant $C$, depending only on $c_2$ and $\Phi_0$. In order to obtain this estimate, we used the fact that
  $\theta^{-1}|\phi_t|\leq (2/3)^{1/2}$; cf. (\ref{eq:hamcon k minus one}). Next, note that $a=\theta^{1/3}a\cdot\theta^{-1/3}$,
  so that $a(t)\rightarrow 0$ as $t\downarrow t_-$. Since $a_t=\theta a/3$, this means that $\theta$ is not integrable to the past. Next, 
  \begin{equation}\label{eq:theta t dos k minus one}
    \theta_t=-\theta^2+\tfrac{6}{a^2}+3V\circ\phi.
  \end{equation}
  Due to this equality and previous observations, it is clear that $\theta$ blows up in finite time to the past. This means that we can assume
  $t_-=0$. Introducing a time coordinate satisfying (\ref{eq:dtdtau}), it is clear that $\tau\rightarrow-\infty$ to the past. Combining this change
  of time coordinate with (\ref{eq:theta t dos k minus one}) yields
  \begin{equation}\label{eq:ln theta prime k minus one}
    (\ln\theta)'=-3+\tfrac{18}{a^2\theta^2}+\tfrac{9V\circ\phi}{\theta^2}.
  \end{equation}
  Since the last two terms on the right hand side converge to zero, it is clear from this equality that $\theta$ tends to infinity exponentially in
  $\tau$. Combining this observation with (\ref{eq:Vl db th sq bd k minus one}) and the fact that $\theta^{1/3}a$ converges to a strictly positive limit,
  it follows that the last two terms on the right hand side of (\ref{eq:ln theta prime k minus one}) converge to zero exponentially. This means that
  $\ln\theta+3\tau$ converges. This information can be used to deduce the existence of $\theta_\infty\in (0,\infty)$ and a constant $C$ such that
  \[
    |\ln\theta(\tau)+3\tau-\ln\theta_\infty|\leq Ce^{\b_V\tau}
  \]
  for all $\tau\leq \tau_0$, where $\b_V:=\min\{4,6(1-\a_V)\}$. By a translation in $\tau$, we can assume $\theta_\infty=1$. This means that
  \begin{equation}\label{eq:ln theta as k minus one dos}
    |\varkappa(\tau)|\leq Ce^{\b_V\tau}
  \end{equation}
  for all $\tau\leq\tau_0$, where $\varkappa(\tau):=\ln\theta(\tau)+3\tau$. Next, let $\nu:=\ln(\theta a)+2\tau-\nu_0$. Since we assume that
  $\theta^{1/3}a\rightarrow a_\infty=e^{\nu_0}$ as $\tau\rightarrow-\infty$, and since (\ref{eq:ln theta as k minus one dos}) holds,
  we conclude that $\nu(\tau)\rightarrow 0$ as $\tau\rightarrow-\infty$. Since $\nu'=(\ln\theta)'+3$, (\ref{eq:ln theta prime k minus one}) yields
  $|\nu(\tau)|\leq Ce^{\b_V\tau}$ for all $\tau\leq\tau_0$. Next, define $\psi_1:=\phi'-\Phi_1$. Then $\psi_1(\tau)\rightarrow 0$ as
  $\tau\rightarrow-\infty$ by assumption. On the other hand, it can be computed that
  \[
    \psi_1'=-\left(\tfrac{18}{\theta^2a^2}+\tfrac{9V\circ\phi}{\theta^2}\right)\phi'-\tfrac{9V'\circ\phi}{\theta^2}.
  \]
  By arguments similar to the above, we conclude that $|\psi_1(\tau)|\leq Ce^{\b_V\tau}$
  for all $\tau\leq\tau_0$. Next, we know that $\phi+\theta^{-1}\phi_t\ln\theta$ converges to $\Phi_0$ in the direction of the singularity.
  Due to (\ref{eq:ln theta as k minus one dos}), we conclude that $\psi_0:=\phi-\phi'\tau-\Phi_0$ converges to zero. On the other hand,
  $\psi_0'=-\phi''\tau=-\psi_1'\tau$. Thus $|\psi_0(\tau)|\leq C\ldr{\tau}e^{\b_V\tau}$ for all $\tau\leq\tau_0$. To summarise, $\nu$, $\varkappa$,
  $\psi_0$ and $\psi_1$ satisfy (\ref{seq:varkappa psiz psio prime etc k minus one}). Moreover, for any $\e<\b_V/2$, it is clear that
  $x:=(\nu,\varkappa,\psi_0,\psi_1)$ satisfies $e^{-\e\tau}|x(\tau)|\rightarrow 0$ as $\tau\rightarrow-\infty$.  This means that $x$ is the unique
  solution we already constructed. The estimates recorded in Remark~\ref{remark:as dos k minus one} follow by arguments similar to the above. 
\end{proof}

\subsection{Generic behaviour}\label{subsection:generic behaviour k minus one}

Finally, we are in a position to prove Proposition~\ref{prop:generic k minus one intro}.
\begin{proof}[Proof of Proposition~\ref{prop:generic k minus one intro}]
  In order to prove the statement concerning non-genericity, recall (\ref{seq:phi vartheta version hyp}). 
  Note that the maximal existence interval of solutions to (\ref{seq:phi vartheta version hyp}) whose initial data satisfy
  (\ref{eq:vartheta ini cond}) is always $\rn{}$ (assuming $V$ to be non-negative). This is also true if $\vartheta$ vanishes.
  Moreover, for any $\phi_\infty\in \rn{}$, $\phi_\tau=0$, $\phi=\phi_\infty$ and $\vartheta=0$ is a fixed point of
  (\ref{seq:phi vartheta version hyp}). Moreover, the corresponding fixed points satisfy (\ref{eq:vartheta ini cond}).
  Rewriting (\ref{seq:phi vartheta version hyp}) as a first order system with variables $\phi_0:=\phi$, $\phi_1:=\phi_\tau$
  and $\vartheta$, each fixed point has one negative, one positive and one vanishing eigenvalue. At this stage, we can apply
  \cite[Theorem~B.3]{Radermacher} (which is a special case of \cite[Theorem 9.1 (i) and (iv)]{Fenichel}). In the application
  of this theorem $M$ is the open subset of $\rn{3}$ defined by
  (\ref{eq:vartheta ini cond}); $Y$ is the smooth vector field determined by rewriting (\ref{seq:phi vartheta version hyp}) as
  a first order system for $\phi_0$, $\phi_1$ and $\vartheta$; $\me$ is the subset of $M$ defined by $\phi_0\in\rn{}$, $\phi_1=0$
  and $\vartheta=0$; and $K=K_n:=[-n,n]\times\{0\}\times\{0\}$, $n\in\nn{}$. Then $K$ is a compact subset of $\me$. Let
  $r\in\nn{}$. Then, due to \cite[Theorem~B.3]{Radermacher}, there is a $C^r$ center unstable manifold $\mC^u_n$ near $K$ and a
  neighbourhood $U_n$ of $K$ such that $A^{-}(U_n)\subset\mC^u_n$. This means the following. First, $K\subset\mC^u_n$; $\mC^u_n$ is locally
  invariant under the flow of $Y$; and for each $m\in K$, $\mC^u_n$ is tangent to the direct sum of the eigenspaces of the linearization
  of the equation at $m$ corresponding to the positive and the vanishing eigenvalue. Note, in particular, that $\mC^u_n$ is two
  dimensional. Second
  \[
  m\in A^-(U_n) \Longleftrightarrow \overline{\{\Phi_\tau(m): \tau\in (-\infty,0]\}}\subset U_n,
  \]
  where $\Phi$ is the flow associated with $Y$. Note, in particular, that any solution converging to a point in $K_n$ sooner or later
  has to belong to $\mC^u_n$ in the past direction. In other words, the solutions that converge to a point in $K_n$ have to, in finite
  time, flow backwards to a point in $\mC^u_n$. This means that taking the union of $\Phi_{\tau}(\mC^u_n)$ for $\tau\in\mathbb{Q}$ and
  $n\in\nn{}$, we obtain a set that contains all the solutions converging to a point in $\me$ in the past. Moreover, this set is a
  countable union of codimension one submanifolds. 
\end{proof}

\chapter{Future asymptotics}\label{section:futureasymptotics}

The future asymptotics of spatially homogeneous solutions to the Einstein non-linear scalar field equations have been studied in many
articles; see, e.g., \cite{mas,reninflation,stab}. Here we state results in the Bianchi class A and $k=-1$ FLRW settings. 

\section{The Bianchi class A setting}

We begin by using modifications of the arguments of \cite{stab} to derive the following conclusions.

\begin{prop}\label{prop:futureasstepone}
  Assume $V\in C^{\infty}(\rn{})$ to have a strictly positive lower bound; $V(s)\rightarrow\infty$ as $|s|\rightarrow\infty$; and $V'$ to only
  have simple zeroes. Fix a Bianchi class A type $\mfT\neq\mrIX$ and an $\mfs\in\{\iso,\roLRS,\roper,\rogen\}$. Let ${}^{\rosc}\mfG_\mfT^\mfs[V]$
  be the subset of ${}^{\rosc}\mfB_\mfT^\mfs[V]$ with the following property: If $\xi\in{}^{\rosc}\mfG_\mfT^\mfs[V]$ is represented by an
  $\eta\in \sfB_\mfT^\mfs$ and $\sigma_{ij}$, $n_{ij}$, $\theta$ and $\phi$ denote the solution to (\ref{eq:silmd})--(\ref{eq:phiddot})
  with initial data determined by $\eta$, then the $n_{ij}$ converge to zero as $t\rightarrow\infty$; the $\sigma_{ij}$ converge
  to zero; $\phi_{t}$ converges to zero; $\phi$ converges to a limit, say $s_{0}$; $\theta$ converges to $[3V(s_{0})]^{1/2}$; $V'(s_{0})=0$; 
  and $V''(s_{0})>0$. Then ${}^{\rosc}\mfG_\mfT^\mfs[V]$, considered as a subset of ${}^{\rosc}\mfB_\mfT^\mfs[V]\cap \{\theta>0\}$, is open, dense
  and has full measure.
\end{prop}
\begin{remark}
  Due to Lemmas~\ref{lemma:sfR mfT mfs} and \ref{lemma:sc mfB mfT mfs param}, ${}^{\rosc}\mfB_\mfT^\mfs[V]$ can be represented by a
  quotient of $\sfB_\mfT^\mfs$ by a finite group $\Gamma_\mfT^\mfs$. In this sense, elements of the former set are represented by a finite number
  of elements of $\sfB_\mfT^\mfs$. Moreover, given an element $(\bn,K,\bphi_0,\bphi_1)\in \sfB_\mfT^\mfs$, the corresponding initial data to
  (\ref{eq:silmd})--(\ref{eq:phiddot}) are obtained as follows: $\sigma(0)$ is the trace free part of $K$; $\theta(0)=\rotr K$;
  $n(0)=\bn$; $\phi(0)=\bphi_0$ and $\phi_t(0)=\bphi_1$. Finally, since the elements of $\Gamma_\mfT^\mfs$ map solutions to solutions and since they
  only permute the diagonal elements of $n$ and $\sigma$ and, possibly, change the signs of all the elements of $n$ at the same time, it
  is clear that the elements of $\Gamma_\mfT^\mfs$ map solutions to (\ref{eq:silmd})--(\ref{eq:phiddot}) with the properties stated in the proposition
  to solutions with the same properties. For this reason, the choice of the representative $\eta\in \sfB_\mfT^\mfs$ for an element
  $\xi\in{}^{\rosc}\mfG_\mfT^\mfs[V]$ does not matter. 
\end{remark}
\begin{proof}
  Let $\sfG_\mfT^\mfs$ be the subset of $\sfB_\mfT^\mfs$ defined in analogy with ${}^{\rosc}\mfG_\mfT^\mfs[V]$. Then, since $\sfB_\mfT^\mfs$ is a finite
  covering space of ${}^{\rosc}\mfB_\mfT^\mfs[V]$, it is sufficient to prove that $\sfG_\mfT^\mfs$ is open, dense and of full measure. 

  Consider a solution corresponding to initial data in $\sfB_\mfT^\mfs$. Due to \cite[Corollary~26.4, p.~450]{stab}, we then know that there is a
  $0\leq V_{1}\in\rn{}$ such that $\theta(t)\rightarrow\sqrt{3V_{1}}$ and
  $\phi_{t}^{2}(t)+2V\circ\phi(t)\rightarrow 2V_{1}$ as $t\rightarrow\infty$. Moreover, $\sigma_{ij}(t)\rightarrow 0$ as $t\rightarrow\infty$.
  Since $V$ has a strictly positive lower bound, $V_{1}>0$ and (\ref{eq:nijd}) yields the conclusion that $n_{ij}(t)$ converges to zero
  exponentially. Due to the arguments presented in the proof of \cite[Corollary~26.7, p.~452]{stab}, we also know that $\phi_{t}(t)\rightarrow 0$
  as $t\rightarrow\infty$ and that $\phi(t)$ converges to a limit. Since $V(s)\rightarrow\infty$ as $|s|\rightarrow\infty$, we know that this
  limit cannot be infinite. Thus there is an $s_{0}\in \rn{}$ such that $\phi(t)\rightarrow s_{0}$ as $t\rightarrow\infty$. Moreover, $V'(s_{0})=0$.
  Note that, by assumption, $V''(s_{0})\neq 0$ and that
  \[
  \sigma_{ij}=0,\ \ \ n_{ij}=0, \ \ \ \phi_{t}=0, \ \ \ \phi=s_{0},\ \ \ \theta=\sqrt{3V(s_{0})}
  \]
  is a fixed point of (\ref{seq:EFEwrtvar})--(\ref{eq:phiddot}). In order to find out what type of fixed point this is, it is useful to reformulate
  the equations as an unconstrained system. Recall, to this end, that $n$ and $\sigma$ are diagonal. The two constraints to be eliminated are then
  (\ref{eq:hamconfin}) and the fact that the trace of $\sigma$ vanishes. To eliminate the
  second one, we remove the equation for $\sigma_{11}$ and define $\sigma_{11}:=-\sigma_{22}-\sigma_{33}$. In order to remove the first one, we
  define $\theta$ to be the positive square root of the right hand side of (\ref{eq:hamconfin}) (note that the right hand side is strictly positive
  due to the fact that $V$ has a strictly positive lower bound and due to the fact that we have excluded Bianchi type IX). If the Bianchi type is such
  that a certain $n_{ii}$ (no summation) vanishes, we also remove the equation for that variable. However, we restrict the present discussion to the case
  that all the $n_{ii}$ (no summation) are non-vanishing; the modifications necessary to cover the other cases simply consist of omitting the corresponding
  $n_{ii}$'s from the variables. The remaining variables are then
  \[
  x=(\sigma_{22},\sigma_{33},n_{11},n_{22},n_{33},\phi_{0},\phi_{1}),
  \]
  where $\phi_{0}:=\phi$ and $\phi_{1}:=\phi_{t}$, and the relevant equations are (\ref{eq:silmd}) and (\ref{seq:slmblmdef}) for $\{l,m\}=\{2\}$ and
  $\{l,m\}=\{3\}$; the diagonal components of (\ref{eq:nijd}); and (\ref{eq:phiddot}).  Assume that we have a solution to this unconstrained system of
  equations. We then wish to prove that it corresponds to a solution to the original system. Note, to this end, that
  \[
  \dot{\sigma}_{11}=-\dot{\sigma}_{22}-\dot{\sigma}_{33}=\theta\sigma_{22}+\theta\sigma_{33}+s_{22}+s_{33}=-\theta\sigma_{11}-s_{11},
  \]
  where we used the fact that the trace of the matrix with components $s_{ij}$ vanishes. This means that all the equations in
  (\ref{seq:EFEwrtvar})--(\ref{eq:phiddot}) are satisfied, except possibly (\ref{eq:thetad}). However, at this stage, (\ref{eq:thetad}) can be
  verified by differentiating (\ref{eq:hamconfin}) and appealing to the equations we know to hold. In short, we can reformulate
  (\ref{seq:EFEwrtvar})--(\ref{eq:phiddot}) to an unconstrained system, say $\dot{x}=F\circ x$. We wish to linearise this system around
  $x_{0}=(0,0,0,0,0,s_{0},0)$. However, denoting the Jacobi matrix of $F$ at $x_{0}$ by $F'(x_{0})$, then
  \[
  F'(x_{0})=\mathrm{diag}(D_{1},D_{2},R),
  \]
  where $D_{1}:=-3H\mathrm{Id}_{2}$, $D_{2}:=-H\mathrm{Id}_{3}$ and
  \[
  R:=\left(\begin{array}{cc} 0 & 1 \\ -9H^{2}\zeta/4 & -3H\end{array}\right),
  \]
  $H:=\sqrt{V(s_{0})/3}$, $\zeta=4V''(s_{0})/(3H)^{2}$ and $\mathrm{Id}_{n}$ denotes the $n\times n$ dimensional identity matrix. Clearly, the eigenvalues
  of $F'(x_{0})$ are $-3H$, $-H$ and
  \[
  \lambda_{\pm}:=\tfrac{3H}{2}[-1\pm(1-\zeta)^{1/2}].
  \]
  In particular, if $V''(s_{0})>0$, all the eigenvalues have strictly negative real parts, and the fixed point is stable. If $V''(s_{0})<0$, then all the
  eigenvalues but one are strictly negative. However, the remaining eigenvalue is strictly positive, and the fixed point is unstable. Let
  $\sfG_{\mfT}^\mfs$ denote the subset of $\sfB_\mfT^\mfs$ such that the corresponding solutions converge to one of the fixed points
  with $V''>0$. Then, by the above, $\sfG_{\mfT}^\mfs$ is open. We wish to prove that it also has full measure in the state space (and therefore is also
  dense). Since the $\phi$-coordinates of the fixed points are isolated, there are at most countably many zeroes of $V'$ such that $V''<0$. Denote
  the set of such zeroes by $\mZ_{-}$. Each $s\in \mZ_{-}$ corresponds to an unstable hyperbolic fixed point of the above system. There is thus an
  open neighbourhood $U_{s}$ of the fixed point and 
  a $C^{1}$-submanifold of $U_{s}$, say $M_{s}$, such that solutions that do not leave $U_{s}$ to the future have to belong to $M_{s}$. Let $\{W_{i}\}$,
  $i\in\nn{}$, be a countable basis of the topology of $\rn{7}$ and $V_{i}:=W_{i}\cap M_{s}$. For each $i\in\nn{}$, let $Q_{i}$ be the set of
  $r\in\mathbb{Q}$ such that the time $r$ flow of the vector field of the unconstrained system is defined on all of $W_{i}$. Finally, define
  \[
    \mV_{s}:=\cup_{i\in\nn{}}\cup_{r\in Q_{i}}\Phi_{r}(V_{i}),\ \ \
    \mV:=\cup_{s\in\mZ_{-}}\mV_{s},
  \]
  where $\Phi_{T}(\xi)$ denotes the solution of the unconstrained system with initial data $\xi$ at $t=0$, evaluated at $t=T$. Since $\Phi_{T}$ is a
  diffeomorphism when it is defined, $\mV_{s}$ and $\mV$ are countable unions of codimension $1$ submanifolds. Next, assume $\xi$ to be initial data
  such that if $x$ is the solution determined by $x(0)=\xi$, then $x$ converges to the fixed point corresponding to $s\in\mZ_{-}$ to the future. Then
  there is a $T$ such that
  $x(t)\in U_{s}$ for all $t\geq T-1$. In particular, this means that $x(t)\in M_{s}$ for all $t\geq T-1$. In particular, for $0<\epsilon<1$, $x(t)\in M_{s}$
  for $t\in [T-\e,T+\e]$. Due to the fact that the domain of the flow associated with the unconstrained system is defined on an open subset of
  $\rn{}\times \rn{7}$, it follows that there is an $\e>0$ and an open neighbourhood $X_{T}$ of $x(T)$ in $\rn{7}$ such that the flow is defined
  on $I_{\e}\times X_{T}$, where $I_{\e}:=[-T-\e,-T+\e]$. Let $W_{i}$ be such that $x(T)\in W_{i}$ and $W_{i}\subseteq X_{T}$. Next, let
  $0<\de<\e$ be such that the image of $[T-\de,T+\de]$ under $x$ belongs to $W_{i}$ and let $u\in\mathbb{Q}\cap I_{\de}$. Then the time $-u$ flow
  of $W_{i}$ is defined on all of $W_{i}$. This means that the time $-u$ flow of $x(u)$ belongs to $\mV_{s}$; i.e., $\xi\in\mV_{s}$.
  To conclude: $\mV_{s}$ consists of all the elements of the phase space such that the corresponding solutions converge to the fixed point corresponding
  to $s\in\mZ_{-}$ to the future. Moreover, $\mV$ consists of all the elements of the phase space such that the corresponding solutions converge to an
  unstable fixed point to the future. Since $\mV$ has zero measure, the complement, i.e. $\sfG_\mfT^\mfs$, is dense. Since we know $\sfG_\mfT^\mfs$ to be
  open, we conclude that it is open, dense and has full measure. The proposition follows. 
\end{proof}

\begin{prop}\label{prop:futureas}
  Assume the conditions of Proposition~\ref{prop:futureasstepone}
  to be satisfied. Let, using the corresponding notation, $\xi\in {}^{\rosc}\mfG_\mfT^\mfs[V]$;
  $\mfI\in \mB_\mfT^\mfs[V]$ be such that $[\mfI]=\xi$; and $(M,g,\phi)=\md[V](\mfI)$. Then $M$ is of the form $M=G\times J$, where
  $J=(t_-,\infty)$ and $g$ is of the form (\ref{eq:gSH}). Moreover, there is a left invariant Riemannian metric $\bge_\infty$ on $G$; an $s_0\in\ro$
  with the property that $V'(s_0)=0$ and $V''(s_0)>0$; and constants $C>0$ and $b>0$ such that
  \begin{subequations}
    \begin{align}
      |\phi(t)-s_{0}|+|\phi_{t}(t)| \leq & Ce^{-bt},\\
      |\bK(t)-H\roId|_{\bge_\infty}+|e^{-2Ht}\bge(t)-\bge_\infty|_{\bge_\infty} \leq & Ce^{-bt}\label{eq:K m HId eHt bge inf}
    \end{align}
  \end{subequations}  
  for all $t\geq t_{0}$, where $H:=\sqrt{V(s_{0})/3}$, $t_{0}\in J$ and $\bge(t)$ and $\bK(t)$ are the metric and Weingarten map induced on
  $M_t=G\times\{t\}$ respectively.
\end{prop}
\begin{remark}
  There is a basis $\{e_i'\}$ of $\mfg$ such that the metric $g$ of the proposition can be written
  \begin{equation}\label{eq:g ai xi form}
    g=-dt\otimes dt+\textstyle{\sum}_ia_i^2(t)\xi^i\otimes \xi^i,
  \end{equation}
  where $\{\xi^i\}$ is the dual basis to $\{e_i'\}$. Using this form of the metric, (\ref{eq:K m HId eHt bge inf}) can be written
  \[
  |\d_t a_{i}(t)/a_{i}(t)-H|+|e^{-Ht}a_{i}(t)-\a_{i}| \leq Ce^{-bt}
  \]
  for all $t\geq t_0$ and some constants $\a_i>0$. 
\end{remark}
\begin{proof}
  That the metric $g$ takes the form (\ref{eq:g ai xi form}) follows from the proof of Proposition~\ref{prop:unique max dev}, see
  Section~\ref{ssection:ex of dev step one}. Moreover, the diagonal components of $\sigma$ are given by $\d_ta_i/a_i-\theta$.   
  Combining the knowledge obtained from Proposition~\ref{prop:futureasstepone} with the arguments presented at the beginning of
  \cite[Proof of Lemma~27.1, pp.~462--464]{stab}, it follows that $\phi_{t}$ and $\phi-s_{0}$ converge to zero exponentially. This means that
  $V\circ\phi-V(s_{0})$ converges exponentially to zero, so that the last two terms on the right hand side of (\ref{eq:Xdef}) converge exponentially
  to their limits. Since the left hand side converges exponentially to zero, see (\ref{eq:dotXest}), it follows that $\theta$ converges exponentially to
  $3H$, where $H=\sqrt{V(s_{0})/3}$. Next, combining (\ref{eq:Xaltversion}) with the fact that $X$ converges to zero exponentially, it is clear that
  $\sigma_{ij}$ converges to zero exponentially. This means that
  \[
  \dot{a}_{i}/a_{i}-H=\sigma_{i}+(\theta-3H)/3
  \]
  converges to zero exponentially. This means that there are constants $\a_{i}>0$ such that $e^{-Ht}a_{i}(t)-\a_{i}$ decays exponentially.
  The proposition follows. 
\end{proof}

\section{The $k=-1$ FLRW setting}\label{section:Future as k minus one setting}

Next, we prove Proposition~\ref{prop:generic k minus one intro exp}.
\begin{proof}[Proof of Proposition~\ref{prop:generic k minus one intro exp}]
  Note that \cite[Corollary~26.4, p.~450]{stab} and \cite[Corollary~26.7, p.~452]{stab} apply in this case as well. We can thus argue as at the
  beginning of the proof of Proposition~\ref{prop:futureasstepone} in order to conclude that $\phi_t(t)\rightarrow 0$ as $t\rightarrow \infty$,
  that there is an $s_0\in \ro$ such that $\phi(t)\rightarrow s_0$ and that $V'(s_0)=0$. Moreover, we know that $a$ tends to infinity exponentially.
  Next, consider the unconstrained system (\ref{seq:phi vartheta version hyp}), where the initial data are required to satisfy
  (\ref{eq:vartheta ini cond}). Note that
  \begin{equation}\label{eq:fixed point k minus one}
    \phi_\tau=0,\ \ \ \phi=s_0,\ \ \ \vartheta=[3V(s_0)]^{-1/2}
  \end{equation}
  constitutes a fixed point for (\ref{seq:phi vartheta version hyp}). Moreover, if we formulate (\ref{seq:phi vartheta version hyp}) as a first
  order system, the linearised system around the fixed point (\ref{eq:fixed point k minus one}) has three eigenvalues with negative real part
  if $V''(s_0)>0$ and two negative eigenvalues and one positive eigenvalue in case $V''(s_0)>0$. By an argument similar to the proof of
  Proposition~\ref{prop:futureasstepone}, the set of initial data corresponding to solutions converging to a fixed point with $V''(s_0)<0$
  is contained in a countable union of one dimensional submanifolds. The proof of exponential convergence is quite similar to the
  argument presented in the proof of Proposition~\ref{prop:futureas}. We leave the details to the reader. 
\end{proof}

\appendix

\chapter[Parametrisations]{Parametrising solutions, regular data and data on the singularity}\label{section:isometryclasses}

The purpose of this appendix is to prove the statements concerning the parametrisation of solutions, regular data and data on the singularity we need in the
main body of the text. As noted in Lemma~\ref{lemma:x arising from mfI}, regular initial data give rise to an element of $\sfB$. We begin
by demonstrating that this element determines the data up to a Lie group isomorphism, assuming the Lie group to be simply connected.

\begin{lemma}\label{lemma:nuKsameisometric}
  Let $V\in C^\infty(\ro)$ and $\mfI_\a=(G_{\a},\bge_{\a},\bk_{\a},\bphi_{\a,0},\bphi_{\a,1})\in{}^{\rosc}\mB[V]$, $\a\in\{0,1\}$; see Definition~\ref{def:BV}. Let
  $\{e_{\a,i}\}$ be an orthonormal basis of $\mfg_{\a}$ (the Lie algebra of $G_{\a}$) with
  respect to $\bge_{\a}$, $\a\in\{0,1\}$, such that if $\nu_{\a}$ and $K_{\a}$ are the corresponding commutator and Weingarten matrices respectively, then
  $x_{\a}:=(\nu_{\a},K_{\a},\bphi_{\a,0},\bphi_{\a,1})\in\sfB$, where $\sfB$ is introduced in Definition~\ref{def:degenerate version of id}. Then, if
  $x_{0}=\psi_\sigma^\pm(x_{1})$ for some $\sigma\in S_3$, see (\ref{eq:psi sigma pm}), there is a Lie group isomorphism $\psi:G_{0}\rightarrow G_{1}$ such
  that $\psi^{*}\bge_{1}=\bge_{0}$, $\psi^{*}\bk_{1}=\bk_{0}$, $\psi^{*}\bphi_{1,0}=\bphi_{0,0}$ and $\psi^{*}\bphi_{1,1}=\bphi_{0,1}$. 
\end{lemma}
\begin{remark}
  That bases $\{e_{\a,i}\}$ with the properties stated in the lemma exist follows fromLemma~\ref{lemma:x arising from mfI}.
\end{remark}
\begin{proof}
  Let $\nu_{\a,i}$ denote the diagonal components of $\nu_{\a}$. 
  Define $\Psi:\mfg_0\rightarrow\mfg_1$ to be the linear map such that $\Psi(e_{0,i})=\pm \sgn\sigma\cdot e_{1,\sigma(i)}$.
  Then, if $\{i,j,k\}$ is an even permutation of $\{1,2,3\}$,
  \begin{equation*}
    \begin{split}
      [\Psi(e_{0,i}),\Psi(e_{0,j})] = & [e_{1,\sigma(i)},e_{1,\sigma(j)}]=\e_{\sigma(i)\sigma(j)\sigma(k)}\nu_{1,\sigma(k)}e_{1,\sigma(k)}\\
      = & \pm\sgn\sigma\cdot \nu_{0,k}e_{1,\sigma(k)}=\e_{ijk}\nu_{0,k}\Psi(e_{0,k})=\Psi([e_{0,i},e_{0,j}]). 
    \end{split}
  \end{equation*}
  Appealing to, e.g., \cite[Theorem~20.19, p.~530]{Lee}, it follows that $\Psi$ arises from a Lie group isomorphism $\psi$. In particular $\psi_{*}=\Psi$.
  By the definition of $\Psi$ and the conditions of the lemma, it is clear that $\psi$ is an isometry of initial data. 
\end{proof}
A similar argument yields a similar conclusion in the case of initial data on the singularity.
\begin{lemma}\label{lemma:nuKsameisometric sing}
  Let $(G_{\a},\msH_{\a},\msK_{\a},\Phi_{\a,0},\Phi_{\a,1})\in{}^{\rosc}\mS$, $\a\in\{0,1\}$; see Definition~\ref{def:sets of id on singularity}. Let
  $\{e_{\a,i}\}$ be an orthonormal basis of $\mfg_{\a}$ (the Lie algebra of $G_{\a}$), with respect to $\msH_{\a}$, $\a\in\{0,1\}$, such that if
  $o_{\a}$ and $P_{\a}$ are defined as in Lemma~\ref{lemma:x arises from Iinf}, then
  $x_{\a}:=(o_{\a},P_{\a},\bphi_{\a,0},\bphi_{\a,1})\in\sfS$, where $\sfS$ is introduced in Definition~\ref{def:mfS def}. Then, if
  $x_{0}=\psi_\sigma^\pm(x_{1})$ for some $\sigma\in S_3$, there is a Lie group isomorphism $\psi:G_{0}\rightarrow G_{1}$ such that
  $\psi^{*}\msH_{1}=\msH_{0}$, $\msK_{1}\psi_*=\psi_*\msK_{0}$, $\psi^{*}\Phi_{1,0}=\Phi_{0,0}$ and $\psi^{*}\Phi_{1,1}=\Phi_{0,1}$. 
\end{lemma}
We also need to prove that all the cases admitted by Definition~\ref{def:degenerate version of id} are realised by initial data. 
\begin{lemma}\label{lemma:realcond}
  Let $V\in C^\infty(\ro)$ and $(\nu,K,\bphi_{0},\bphi_{1})\in\sfB$; see Definition~\ref{def:degenerate version of id}. Then there are non-trivial simply
  connected Bianchi class A initial data for the Einstein non-linear scalar field equations, say $(G,\bge,\bk,\bphi_{0},\bphi_{1})$, and an orthonormal
  frame $\{e_{i}\}$ of $\mfg$ with respect to $\bge$ such that if $\bnu$ and $\bK$ are the commutator and Weingarten matrices associated with $\{e_{i}\}$,
  then $\bnu=\nu$ and $\bK=K$. 
\end{lemma}
\begin{proof}
  To begin with, let us give examples of Lie groups that belong to the different Bianchi class A types. The Lie group $\rn{3}$, with group operation
  given by addition, is of Bianchi type I. The Heisenberg group is of Bianchi type II; see, e.g., \cite[Section~23.4, p.~247]{RinCauchy}. The
  Lie group Sol is of Bianchi type VI${}_{0}$; see, e.g., \cite[Section~23.3, p.~246]{RinCauchy}. An example of a Lie group of Bianchi type VII${}_{0}$
  is given in the proof of \cite[Lemma~23.3, p.~245]{RinCauchy}. The groups $\mathrm{Sl}(2,\rn{})$ and $\mathrm{SU}(2)$ are of Bianchi type VIII and
  IX respectively. Let $x:=(\nu,K,\bphi_{0},\bphi_{1})\in\sfB$. Fix a Lie group $G$, with Lie algebra
  $\mfg$, of the appropriate type. Then, due to
  \cite[Lemma~19.8, p.~208]{RinCauchy} and the arguments used to prove it, it follows that there is a basis $\{e_{i}\}$ of $\mfg$ such that
  the associated commutator matrix is $\nu$. At this point, we can define a metric $\bge$ by declaring $\{e_{i}\}$ to be orthonormal. Moreover, we
  can define $\bk$ by requiring that the Weingarten matrix associated with $\{e_{i}\}$ equals $K$. Since (\ref{seq:compverconstraints}) is satisfied,
  we conclude that (\ref{seq:constraintsBA}) holds. Thus $(G,\bge,\bk,\bphi_{0},\bphi_{1})$ are Bianchi class A initial data for the
  Einstein non-linear scalar field equations and $\{e_{i}\}$ is a frame with the desired properties. 
\end{proof}
There is a similar result in the case of data on the singularity. 
\begin{lemma}\label{lemma:realcond sing}
  Let $(o,P,\Phi_{0},\Phi_{1})\in\sfS$; see Definition~\ref{def:mfS def}. Then there are simply connected quiescent Bianchi class A initial data on the
  singularity for the Einstein non-linear scalar field equations $(G,\msH,\msK,\Phi_{0},\Phi_{1})$, and an orthonormal frame $\{e_{i}\}$ of $\mfg$ with
  respect to $\msH$ such that if $\bar{o}$ is the commutator matrix associated with $\{e_{i}\}$ and $\bP$ is the matrix with components
  $\bP_{ij}=\msH(\msK e_i,e_j)$, then $\bar{o}=o$ and $\bar{P}=P$. 
\end{lemma}
Before we proceed, it is useful to note that simply connected isotropic and LRS Bianchi type VII${}_0$ initial data are of Bianchi type I.
\begin{lemma}\label{lemma:LRS BVIIz is BI}
  Let $V\in C^\infty(\ro)$ and $\mfI=(G,\bge,\bk,\bphi_0,\bphi_1)$ be simply connected Bianchi type VII${}_0$ initial data for the Einstein non-linear scalar
  field equations which are either isotropic or LRS. Then $\mfI$ are isometric to initial data of Bianchi type I. 
\end{lemma}
\begin{proof}
  Assume first $\mfI$ to be isotropic. Then, due to the arguments given in the proof of Lemma~\ref{lemma:psi sigma plus isotropic}, there is an orthonormal
  frame $\{e_i\}$ of $\mfg$ with respect to $\bge$ such that the corresponding commutator and Weingarten matrices are diagonal and such that if $\nu_i$ are
  the diagonal components of the commutator matrix and $k_i$ are the diagonal components of the Weingarten matrix, then $\nu_1=0$; $\nu_2=\nu_3$; and all the
  $k_i$ are equal. If $\mfI$ are LRS, the arguments given in the proof of Lemma~\ref{lemma:sfB char per symm VIIz} yield the same conclusion, with one
  modification; we still have $k_2=k_3$, but $k_1\neq k_2$. 
  
  Next, due to \cite[Section~23.2, pp.~245--246]{RinCauchy}, the set $\rn{3}$ with multiplication given by \cite[(23.11), p.~245]{RinCauchy} is a simply
  connected Lie group of Bianchi type VII${}_0$, say $H$. Moreover, given $\nu_{2}\neq 0$, the frame
  \begin{equation}\label{eq:explframe}
    e_{1}'=\nu_{2}\d_{x},\ \ \
    e_{2}'=-\cos(x)\d_{y}+\sin(x)\d_{z},\ \ \
    e_{3}'=\sin(x)\d_{y}+\cos(x)\d_{z}
  \end{equation}
  is left invariant and such that the associated commutator matrix satisfies $\nu'=\rodiag(0,\nu_{2},\nu_{2})$. Defining $\psi:\mfh\rightarrow\mfg$ to be the
  linear map satisfying $\psi(e_i')=e_i$, it follows that $\psi$ is a Lie algebra isomorphism. Combining this observation with \cite[Theorem~20.19, p.~530]{Lee}
  and the fact that $G$ and $H$ are simply connected, it follows that $\psi=\Psi_*$ for some Lie group isomorphism $\Psi$. Pulling back the initial data using $\Psi$,
  we can assume $G=H$ and that $e_i=e_i'$. The dual frame to (\ref{eq:explframe}) is given by
  \[
  \xi^{1}=\nu_{2}^{-1}dx,\ \ \
  \xi^{2}=-\cos(x)dy+\sin(x)dz,\ \ \
  \xi^{3}=\sin(x)dy+\cos(x)dz.
  \]
  This means that
  \[
  \bge=\xi^{1}\otimes\xi^{1}+\xi^{2}\otimes\xi^{2}+\xi^{3}\otimes\xi^{3}=\nu_{2}^{-2}dx\otimes dx+dy\otimes dy+dz\otimes dz;
  \]
  i.e., that the metric is Euclidean. Next, note that since $k_{2}=k_{3}$, then
  \[
  \bk=k_{1}\xi^{1}\otimes\xi^{1}+k_{2}\xi^{2}\otimes\xi^{2}+k_{2}\xi^{3}\otimes\xi^{3}=k_{1}\nu_{2}^{-2}dx\otimes dx+k_{2}dy\otimes dy+k_{2}dz\otimes dz.
  \]
  Due to the above, it is clear that the initial data are in fact of Bianchi type I.
\end{proof}
There is a similar result for data on the singularity.
\begin{lemma}\label{lemma:LRS VIIz is I sing}
  Let $\mfI$ be simply connected quiescent Bianchi type VII${}_0$ initial data on the singularity for the Einstein non-linear scalar field equations
  which are either isotropic or LRS. Then $\mfI$ are isometric to initial data on the singularity of Bianchi type I. 
\end{lemma}
\begin{proof}
  The argument concerning the metric is identical to the argument in the proof of Lemma~\ref{lemma:LRS BVIIz is BI}. Turning to the $(1,1)$-tensor field,
  say $\msK$:
  \[
    \msK=P_1e_1\otimes\xi^1+P_2e_2\otimes\xi^2+P_2e_3\otimes\xi^3=P_1\d_x\otimes dx+P_2\d_y\otimes dy+P_2\d_z\otimes dz.
  \]
  The lemma follows. 
\end{proof}
Due to the last two lemmas, we exclude isotropic and LRS Bianchi type VII${}_0$ data in what follows. Next, we need to prove that isometric data,
when represented by elements of $\sfB$, see Definition~\ref{def:degenerate version of id} and Lemma~\ref{lemma:x arising from mfI}, are related by
a map of the form $\psi_\sigma^\pm$. 
\begin{lemma}\label{lemma:isometryinducesequalnuK}
  Let $V\in C^{\infty}(\ro)$ and $\mfI_\a:=(G_{\a},\bge_{\a},\bk_{\a},\bphi_{\a,0},\bphi_{\a,1})\in\mB[V]$, $\a=0,1$. Let $\{e_{\a,i}\}$ be an orthonormal basis of
  $\mfg_{\a}$ (the Lie
  algebra of $G_{\a}$) with respect to $\bge_{\a}$, $\a=0,1$, such that if $\nu_{\a}$ and $K_{\a}$ are the corresponding commutator and Weingarten matrices
  respectively, then $x_{\a}:=(\nu_{\a},K_{\a},\bphi_{\a,0},\bphi_{\a,1})\in\sfB$. Assume that there is an isometry $\psi$ between $\mfI_0$ and $\mfI_1$. Then, if
  $x_0\in\sfB_\mfT^\mfs$, it follows that $x_1\in\sfB_\mfT^\mfs$. Moreover, there is a $\psi_\sigma^\pm\in\Gamma_\mfT^\mfs$, see Definition~\ref{def:BTs GTs},
  such that $x_1=\psi_\sigma^\pm(x_0)$. 
\end{lemma}
\begin{remark}\label{remark:isometric data are of same type}
  One particular consequence of the proof is that if $\mfI_\a\in\mB[V]$, $\a=0,1$, are isometric and $\mfI_0\in\mB_\mfT^\mfs[V]$, then
  $\mfI_1\in\mB_\mfT^\mfs[V]$. 
\end{remark}
\begin{remark}\label{remark:frame to frame}
  In case $\mfI_0$ is neither isotropic nor LRS, there is a $\sigma\in S_3$ and $\e_i\in\{-1,1\}$, $i=1,2,3$, such that
  $\psi_{*\rod}e_{0,i}=\e_i e_{1,\sigma(i)}$; see the proof of the lemma. Here $(\psi_{*\rod}X)_{x}=\psi_{*}X|_{\psi^{-1}(x)}$. 
\end{remark}
\begin{proof}
  That there are frames as in the statement of the lemma follows from Lemma~\ref{lemma:x arising from mfI}. From now on, we assume such frames to have
  been chosen, and we consistently use the following notation:
  \begin{equation}\label{eq:nu K R term}
    \nu_{\a}=\rodiag(\nu_{\a,1},\nu_{\a,2},\nu_{\a,3});\ \ K_{\a}=\rodiag(K_{\a,1},K_{\a,2},K_{\a,3});\ \
    R_{\a}=\rodiag(R_{\a,1},R_{\a,2},R_{\a,3}).
  \end{equation}
  Here $R_{\a}$ collects the diagonal components of the Ricci tensor of $\bge_{\a}$ with respect to
  $\{e_{\a,i}\}$ (the off diagonal components vanish) and is called the \textit{Ricci matrix} associated with $\nu_{\a}$. Due to
  \cite[Lemma~19.11, p.~209]{RinCauchy}, if $\{i,j,k\}=\{1,2,3\}$, 
  \begin{equation}\label{eq:Riccidiagitocm app}
    R_{\a,i}:=-\tfrac{1}{2}(\nu_{\a,j}-\nu_{\a,k})^{2}+\tfrac{1}{2}\nu_{\a,i}^{2}
  \end{equation}
  (no summation on any index). Note also that, due to the existence of the isometry, $\bphi_{0,0}=\bphi_{1,0}$ and $\bphi_{0,1}=\bphi_{1,1}$. Consequently, we
  only need to focus on $\nu_\a$ and $K_\a$. It is of interest to sort out the consequences of the components of the Ricci tensor being equal. We have
  \begin{equation}\label{eq:two Ra equal}
    R_{\a,2}=R_{\a,3}\ \text{iff}\ \nu_{\a,2}=\nu_{\a,3}\ \text{or}\ \nu_{\a,1}=\nu_{\a,2}+\nu_{\a,3}.
  \end{equation}
  In particular, if, for a fixed $\a$, all the $R_{\a,i}$ are equal, there are two possibilities. Either all the $\nu_{\a,i}$ are equal or one vanishes and the
  remaining two are equal.   

  \textit{Assume $\mfI_0$ to be isotropic.} Then $\mfI_1$ are also isotropic. Since we have excluded isotropic Bianchi type VII${}_0$ initial data, this is
  only possible in the case of Bianchi types I and IX. For topological reasons, we conclude that $G_\a$ are either both of type I or both of type IX. Since
  the eigenvalues of the Weingarten maps have to conicide, it follows that $K_0=K_1$. Moreover, since all the $\nu_{0,i}$ have to equal, all the $\nu_{1,i}$
  have to equal and the scalar curvatures of $\bge_\a$ are $3\nu_{\a,1}^2/2$, it is clear that $\nu_0=\pm\nu_1$. There is thus a
  $\psi_\sigma^\pm\in\Gamma_\mfT^{\iso}$ such that $x_1=\psi_\sigma^\pm(x_0)$; cf. Definition~\ref{def:BTs GTs}. One immediate consequence of the above
  observations is that if $\mfI_0\in \mB_\mfT^{\iso}[V]$, then $\mfI_1\in \mB_\mfT^{\iso}[V]$; cf. Remark~\ref{remark:isometric data are of same type}.
  Moreover, if $\mfI_0\in \mB_\mfT^{\iso}[V]$, then $x_0,x_1\in\sfB_\mfT^{\iso}$. 

  \textit{Assume that $\mfI_0$ are LRS.} Assume, to begin with, that $\{e_{0,i}\}$ is an orthonormal frame of the type assumed to exist in
  Definition~\ref{def:LRS}. It can then be deduced that $\nu_{0,2}=\nu_{0,3}$ and that  $K_{0,2}=K_{0,3}$. On the other hand, either $\nu_{0,1}\neq\nu_{0,2}$ or
  $K_{0,1}\neq K_{0,2}$ (since the data are not isotropic). Then $R_{0,2}=R_{0,3}$. Note that we can assume that either $R_{0,1}\neq R_{0,2}$ or $K_{0,1}\neq K_{0,2}$
  (since $\mfI_0$ would otherwise be isotropic). Next, we can assume that $K_{1,2}=K_{1,3}=K_{0,2}$ and that $R_{1,2}=R_{1,3}=R_{0,2}$ by permuting the basis
  $\{e_{1,i}\}$. In fact, we can assume the permutation to be even, since we only need to align the basis vector corresponding to the element of $K_1$ or
  $R_1$ which is distinct. Thus the change of basis corresponds to replacing $x_1$ by $\psi_{\sigma}^+(x_1)$, where $\sigma$ is even. Combining these
  observations implies that $\psi_{*\rod}e_{0,1}=\xi e_{1,1}$, where $|\xi|=1$ and $(\psi_{*\rod}X)_{x}=\psi_{*}X|_{\psi^{-1}(x)}$. Thus
  \begin{equation}\label{eq:ezottwoeootwo LRS}
    \psi_{*\rod}e_{0,2}=\a e_{1,2}+\b e_{1,3},\ \ \
    \psi_{*\rod}e_{0,3}=\g e_{1,2}+\de e_{1,3}.
  \end{equation}
  Moreover, since $\psi$ is an isometry,
  \begin{equation}\label{eq:abgdecondiso LRS}
    \a^{2}+\b^{2}=1, \ \ \
    \g^{2}+\de^{2}=1,\ \ \
    \a\g+\b\de=0.
  \end{equation}
  Due to the fact that the Levi-Civita connection is torsion free and the fact that isometries preserve the Levi-Civita connection (see, e.g.,
  \cite[Proposition~59, p.~90]{oneill}), we conclude that 
  \begin{equation}\label{eq:BInotBVIIz LRS}
    \begin{split}
      \nu_{0,3}(\g e_{1,2}+\de e_{1,3}) = & \psi_{*\rod}([e_{0,1},e_{0,2}])=[\psi_{*\rod}e_{0,1},\psi_{*\rod}e_{0,2}]=[\xi e_{1,1},\a e_{1,2}+\b e_{1,3}]\\
      = & \xi(e_{1,1}(\a)-\nu_{1,2}\b)e_{1,2}+\xi(e_{1,1}(\b)+\nu_{1,3}\a)e_{1,3}.
    \end{split}
  \end{equation}
  This means that
  \begin{subequations}\label{seq:eoo a b}
    \begin{align}
      e_{1,1}(\a) = & \nu_{1,2}\b+\xi\nu_{0,3}\g,\\
      e_{1,1}(\b) = & -\nu_{1,3}\a+\xi\nu_{0,3}\de.
    \end{align}
  \end{subequations}
  Applying $\psi_{*\rod}$ to $[e_{0,1},e_{0,3}]$ similarly yields
  \begin{subequations}\label{seq:eoo g de}
    \begin{align}
      e_{1,1}(\g) = & \nu_{1,2}\de-\xi\nu_{0,2}\a,\\
      e_{1,1}(\de) = & -\nu_{1,3}\g-\xi\nu_{0,2}\b.
    \end{align}
  \end{subequations}  
  Applying $e_{1,1}$ to the first equality in (\ref{eq:abgdecondiso LRS}) and appealing to (\ref{eq:abgdecondiso LRS}) and (\ref{seq:eoo a b}) yields
  \begin{equation}\label{eq:either or ab etc}
    0=\a e_{1,1}(\a)+\b e_{1,1}(\b)=\a\b(\nu_{1,2}-\nu_{1,3}).
  \end{equation}
  There are thus two possibilities. Either $\a\b=0$ or $\nu_{1,2}=\nu_{1,3}$. Assume, for the moment, that $\a\b=0$. Combining this equality with
  (\ref{eq:abgdecondiso LRS}) yields two possibilities. Either $|\a|=1$, $\b=0$, $\g=0$ and $|\de|=1$; or $\a=0$, $|\b|=1$, $|\g|=1$ and $\de=0$.
  In the first case, (\ref{seq:eoo a b}) and (\ref{seq:eoo g de}) imply that $\nu_{1,i}=\a\de\xi\nu_{0,i}$ for $i=\{2,3\}$. In particular, this means
  that $\nu_{1,2}=\nu_{1,3}$. Assume now that $\a=0$, $|\b|=1$, $|\g|=1$ and $\de=0$. In this case, a similar argument yields
  $\nu_{1,2}=-\b\g\xi\nu_{0,3}$ and $\nu_{1,3}=-\b\g\xi\nu_{0,2}$. To summarise, we can always assume that $\nu_{1,2}=\nu_{1,3}$. Thus $\mfI_1$
  are also LRS. Since we have excluded LRS Bianchi type VII${}_0$, we know that if $\nu_{0,2}\neq 0$, then $\nu_{0,1}\neq 0$. Similarly, if $\nu_{1,2}\neq 0$,
  then $\nu_{1,1}\neq 0$. Assume now that $\mfI_0$ is of Bianchi type I. Then $\mfI_1$ is either of Bianchi type I or Bianchi type VII${}_0$ (since these
  are the only Bianchi types compatible with vanishing Ricci curvature). Since we know $\mfI_1$ to be LRS, $\mfI_1$ can, on the other hand, not be of
  Bianchi type VII${}_0$. This means that $\mfI_1$ is of Bianchi type I. Since the eigenvalues of the Weingarten map have to coincide, there is a
  $\psi_\sigma^+$ such that $x_1=\psi_\sigma^+(x_0)$. Moreover, since we can always interchange $K_{0,2}$ and $K_{0,3}$, we can assume $\sigma$ to be even. This
  means that we can assume $\psi_\sigma^+\in \Gamma_{\mrI}^{\roLRS}$. Assume now that $\nu_{0,1}\neq 0$. Since $R_{0,1}=R_{1,1}$, (\ref{eq:Riccidiagitocm app}) and
  previous deductions imply that $\nu_{0,1}=\eta\nu_{1,1}$, where $|\eta|=1$. Combining this observation with $R_{0,2}=R_{1,2}$ and previous deductions imply
  that $\nu_{0,1}\nu_{0,3}=\nu_{1,1}\nu_{1,3}$. Since $\nu_{0,1}\neq 0$, it follows that $\nu_{0,3}=\eta\nu_{1,3}$. Thus $\nu_1=\eta\nu_0$. Since $K_{0,i}=K_{1,i}$
  for $i=2,3$ and the eigenvalues of the Weingarten matrices have to coincide, we also have $K_0=K_1$. This means that $\mfI_0$ and $\mfI_1$ have the same
  Bianchi type, say $\mfT\neq\mrI$, and that there is a $\sigma\in S_3$ such that $x_1=\psi_\sigma^{\pm}(x_0)$. Due to the fact that the data are LRS, we can,
  if necessary, add a transposition to $\sigma$ in order to ensure that $\sigma$ is even. This means that $\psi_\sigma^\pm\in\Gamma_\mfT^{\roLRS}$. If
  $\{e_{0,i}\}$ is not an orthonormal frame of the type assumed to exist in Definition~\ref{def:LRS}, we can apply the above argument twice; from an
  orthonormal frame of the type assumed to exist in Definition~\ref{def:LRS}, say $\{f_{0,i}\}$, to $\{e_{0,i}\}$ (with $\psi$ replaced by $\roId$) and from
  $\{f_{0,i}\}$ to $\{e_{1,i}\}$. Combining the corresponding maps then yields the desired conclusion. One immediate consequence of the above observations is
  that if $\mfI_0\in \mB_\mfT^{\roLRS}[V]$, then $\mfI_1\in \mB_\mfT^{\roLRS}[V]$; cf. Remark~\ref{remark:isometric data are of same type}. Moreover, if
  $\mfI_0\in \mB_\mfT^{\roLRS}[V]$, then $x_0,x_1\in\sfB_\mfT^{\roLRS}$. 

  In what follows, we assume that neither $\mfI_0$ nor $\mfI_1$ are LRS (or isotropic). 
  
  \textit{Remaining cases.} If $\mfI_0$ is of Bianchi type I but neither isotropic nor LRS, then the $K_{0,i}$ are all distinct. Thus the conclusion of
  Remark~\ref{remark:frame to frame} holds. Since $\psi_{*\rod}$ preserves the Lie bracket, this means that $\nu_1=0$, so that $\mfI_1$ is of Bianchi type
  I. Moreover, it is clear that there is a $\sigma\in S_3$ such that $x_1=\psi_\sigma^+(x_0)$. In what follows, we therefore assume that $\mfI_0$ is not of
  Bianchi type I. If all the $K_{0,i}$ are equal and all the $R_{0,i}$ are equal, we are in the isotropic setting, which has already been handled. By
  permuting the elements of the frame $\{e_{0,i}\}$, we can therefore assume that either $K_{0,1}\notin \{K_{0,2},K_{0,3}\}$ or
  $R_{0,1}\notin \{R_{0,2},R_{0,3}\}$. Similarly, permuting the frame $\{e_{1,i}\}$, we can assume that $\psi_{*\rod}e_{0,1}=\e_1 e_{1,1}$, where $|\e_1|=1$. This
  means that $K_{0,1}=K_{1,1}$. If $K_{0,2}\neq K_{0,3}$, we can assume that $\psi_{*\rod}e_{0,A}=\e_A e_{1,A}$ for $A\in\{2,3\}$, where $|\e_A|=1$; i.e. that the
  conclusion of Remark~\ref{remark:frame to frame} holds. This means that $x_1=\psi_\sigma^\pm(x_0)$ for some $\sigma\in S_3$. Assume now that
  $K_{0,2}=K_{0,3}$. Then $K_{1,2}=K_{1,3}$, so that $\nu_{1,2}\neq \nu_{1,3}$, since $\mfI_1$ would otherwise be LRS. The argument leading up to
  (\ref{eq:either or ab etc}) still applies
  and implies that $\a\b=0$, since $\nu_{1,2}\neq \nu_{1,3}$. We conclude that the conclusion of Remark~\ref{remark:frame to frame} holds. Thus
  $x_1=\psi_\sigma^\pm(x_0)$ for some $\sigma\in S_3$. This observation covers all the remaining cases except for permutation symmetric Bianchi type VI${}_0$
  initial data. However, just as in the LRS setting, that case can be handled by adding a transposition, if necessary. One immediate consequence of the
  above observations is that if $\mfs\in\{\roper,\rogen\}$ and $\mfI_0\in \mB_\mfT^{\mfs}[V]$, then $\mfI_1\in \mB_\mfT^{\mfs}[V]$; cf.
  Remark~\ref{remark:isometric data are of same type}. Moreover, if $\mfI_0\in \mB_\mfT^{\mfs}[V]$, then $x_0,x_1\in\sfB_\mfT^{\mfs}$. The lemma follows. 
\end{proof}
With this information at hand, we are in a position to prove Proposition~\ref{prop:iso id to iso sol}. 
\begin{proof}[Proof of Proposition~\ref{prop:iso id to iso sol}]
  Assume that the initial data are $\mfI_\a=(G_\a,\bah_\a,\bk_\a,\bphi_{\a,0},\bphi_{\a,1})$, $\a\in\{0,1\}$. Fix frames $\{e_{\a,i}\}$ as in the
  statement of Lemma~\ref{lemma:isometryinducesequalnuK}. Due to Lemma~\ref{lemma:isometryinducesequalnuK}, we know that $x_1=\psi_\sigma^\pm(x_0)$,
  using the terminology of Lemma~\ref{lemma:isometryinducesequalnuK}. Relabelling the frame $\{e_{1,i}\}$  and replacing $e_{1,1}$
  by $-e_{1,1}$, if necessary, we can ensure that $x_1=x_0$. Fix $t_0$ as in the statement of the proposition and compute
  $a_{\a,i}$ as described in the proof of Proposition~\ref{prop:unique max dev}; see (\ref{eq:for the ai}). Then, since the initial data and the
  initial time are the same, it follows that the existence intervals of the two developments coincide, that the scalar fields of the two developments
  are the same and that $a_{0,i}(t)=a_{1,i}(t)$ for all $t$. Moreover, the metrics take the form
  \begin{equation}\label{eq:g alpha format dev}
    g_{\a}=-dt\otimes dt+\textstyle{\sum}_{i}a_i^2(t)\xi_\a^i\otimes\xi_\a^i,
  \end{equation}
  where $a_i:=a_{0,i}=a_{1,i}$ and $\{\xi_\a^i\}$ is the frame dual to $\{e_{\a,i}\}$.

  If both of the data sets are isotropic, the $a_i$ coincide, see (\ref{eq:for the ai}), so that the developments take the form
  \[
    g_{\a}=-dt\otimes dt+a_1^2(t)\bah_\a.
  \]
  Since $\chi$ is an isometry from $\mfI_0$ to $\mfI_1$, it follows that $\chi\times\roId$ is an isometry between the developments. 

  Next, assume that the data are LRS but not isotropic. Then, using the notation introduced in (\ref{eq:nu K R term}), we can assume that
  \[
    \nu_{0,2}=\nu_{0,3}=\nu_{1,2}=\nu_{1,3},\ \ \
    K_{0,2}=K_{0,3}=K_{1,2}=K_{1,3}.
  \]
  Note that this implies that $a_{2}=a_3$; cf. (\ref{eq:LRS spacetime metric}) and the adjacent text. This means that the metrics take the form
  \[
    g_\a=-dt\otimes dt+a_{1}^2\xi_{\a}^1\otimes\xi_\a^1+a_{2}^{2}(\xi_{\a}^2\otimes\xi_\a^2+\xi_{\a}^3\otimes\xi_\a^3).
  \]
  In order for the data not to be isotropic, we have to have either $K_{\a,1}\neq K_{\a,2}$ for $\a\in\{0,1\}$ or $R_{\a,1}\neq R_{\a,2}$ for $\a\in\{0,1\}$.
  This means that $\chi_{*\rod}e_{0,1}=\pm e_{1,1}$, where $(\chi_{*\rod}X)_{x}=\chi_{*}X|_{\chi^{-1}(x)}$. Moreover, (\ref{eq:ezottwoeootwo LRS}) and
  (\ref{eq:abgdecondiso LRS}) have to hold with $\psi$ replaced by $\chi$. Summarising the above observations, it follows that $\chi\times\roId$ is an
  isometry between the developments. 

  Finally, assume that the data are neither isotropic nor LRS. Then, due to Remark~\ref{remark:frame to frame},
  we can choose frames as in the statement of Lemma~\ref{lemma:isometryinducesequalnuK} such that $\chi_{*\rod}e_{0,i}=e_{1,i}$ (just replace
  $\{e_{1,i}\}$ with $\{\e_i e_{1,\sigma(i)}\}$, if necessary). Combining this observation with the fact that the metrics take the form
  (\ref{eq:g alpha format dev}), it is clear that $\chi\times\roId$ is an isometry between the developments.
\end{proof}
Next, we note that there is a result analogous to Lemma~\ref{lemma:isometryinducesequalnuK} concerning data on the singularity. 
\begin{lemma}\label{lemma:isometryinducesequal o P}
  Let $\mfI_\a:=(G_{\a},\msH_{\a},\msK_{\a},\Phi_{\a,0},\Phi_{\a,1})\in\mS$, $\a=0,1$. Let $\{e_{\a,i}\}$ be an orthonormal basis of
  $\mfg_{\a}$ (the Lie algebra of $G_{\a}$) with respect to $\msH_\a$, $\a=0,1$, such that if $o_{\a}$ and $P_{\a}$ are the corresponding commutator
  and Weingarten matrices respectively, then $x_{\a}:=(o_{\a},P_{\a},\Phi_{\a,0},\Phi_{\a,1})\in\sfS$. Assume that there is an isometry $\psi$ between
  $\mfI_0$ and $\mfI_1$. Then, if $x_{0}\in\sfS_{\mfT}^\mfs$, it follows that $x_{1}\in\sfS_{\mfT}^\mfs$. Moreover, there is a $\psi_\sigma^\pm\in\Gamma_\mfT^\mfs$,
  see Definition~\ref{def:BTs GTs}, such that $x_1=\psi_\sigma^\pm(x_0)$. 
\end{lemma}
\begin{remark}\label{remark:frame to frame sing}
  In case $\mfI_0$ is neither isotropic nor LRS, there is a $\sigma\in S_3$ and $\e_i\in\{-1,1\}$, $i=1,2,3$, such that
  $\psi_{*\rod}e_{0,i}=\e_i e_{1,\sigma(i)}$. Here $(\psi_{*\rod}X)_{x}=\psi_{*}X|_{\psi^{-1}(x)}$. 
\end{remark}
\begin{proof}
  The proof is similar to that of Lemma~\ref{lemma:isometryinducesequalnuK}.
\end{proof}

Finally, we are in a position to prove Proposition~\ref{prop:iso idos to iso dev}.

\begin{proof}[Proof of Proposition~\ref{prop:iso idos to iso dev}]
  Define $G_\a$, $\msH_\a$, $\msK_\a$, $\Phi_{\a,0}$ and $\Phi_{\a,1}$ by
  \[
  (G_\a,\msH_\a,\msK_\a,\Phi_{\a,0},\Phi_{\a,1}):=\mfI_{\infty,\a}.
  \]
  Due to the existence of the isometry, $\Phi_{0,\b}=\Phi_{1,\b}$ for $\b\in\{0,1\}$. For that reason, we define $\Phi_{\b}:=\Phi_{0,\b}$, $\b\in\{0,1\}$.
  For each $\a\in\{0,1\}$, let $\{e_{\a,i}\}$ be a frame satisfying the conditions of (\ref{eq:canonicalbasisBianchiA}) with $\msH$ replaced by $\msH_\a$;
  $\msK$ replaced by $\msK_\a$; $e_i$ replaced by $e_{\a,i}$; $p_i$ replaced by $p_{\a,i}$; and $o_i$ replaced by $o_{\a,i}$. Assume, moreover, the frames
  to be as described in one of Lemmas~\ref{lemma:frame isotropic setting}; \ref{lemma:frame LRS case}; and \ref{lemma:frame general case} according
  to whether the data are isotropic; LRS; or neither isotropic nor LRS, respectively. Note, in particular, that the $o_{\a,i}$ have signs as in
  Table~\ref{table:bianchiA}. Due to the existence part of the proof of Theorem~\ref{thm:dataonsingtosolution}, we conclude that there are solutions
  $(M_\a,g_\a,\phi_\a)$ inducing the data $\mfI_{\infty,\a}$ on the singularity. Moreover, $M_\a=G_\a\times (0,t_{\a,+})$ and
  \begin{equation}\label{eq:g alpha format dev sing prel}
    g_{\a}=-dt\otimes dt+\textstyle{\sum}_{i}a_{\a,i}^2(t)\xi_\a^i\otimes\xi_\a^i,
  \end{equation}
  where $\{\xi_\a^i\}$ is the frame dual to $\{e_{\a,i}\}$. 

  Next, combining Lemmas~\ref{lemma:frame isotropic setting}; \ref{lemma:frame LRS case}; and \ref{lemma:frame general case} with the proof of
  Lemma~\ref{lemma:x arises from Iinf}, it is clear that if $o_\a$ and $P_\a$ are the matrices with diagonal components $o_{\a,i}$ and $p_{\a,i}$
  respectively, then $x_\a:=(o_\a,P_\a,\Phi_0,\Phi_1)\in\sfS$. Thus the frames $\{e_{\a,i}\}$ have the properties required in the statement of
  Lemma~\ref{lemma:isometryinducesequal o P}, so that $x_1=\psi_\sigma^\pm(x_0)$, using the terminology of Lemma~\ref{lemma:isometryinducesequal o P}.
  In the isotropic setting, the $o_\a$ and $P_\a$ are multiples of the identity. Moreover, due to the requirements of
  Lemma~\ref{lemma:frame isotropic setting}, the $o_{\a,i}$ have to be non-negative. Combining these observations, it follows that $x_0=x_1$ in the
  isotropic setting. In the LRS setting, Lemma~\ref{lemma:frame LRS case} ensures that $p_{\a,2}=p_{\a,3}$ and that $o_{\a,2}=o_{\a,3}$, but
  that, for each $\a$, either $p_{\a,1}\notin \{p_{\a,2},p_{\a,3}\}$ or $o_{\a,1}\notin \{o_{\a,2},o_{\a,3}\}$. As a consequence, $\sigma(1)=1$, so that
  $\sigma$ is either the identity or a transposition interchanging $2$ and $3$. Since the elements of $o_\a$ have signs as in
  Table~\ref{table:bianchiA}, it follows that $x_0=x_1$. Finally, assume that the data are neither isotropic nor LRS. In the case of Bianchi type I,
  the requirements of Lemma~\ref{lemma:frame general case} ensure that $x_0=x_1$. In the case of Bianchi type II, $o_{\a,1}>0$ and $o_{\a,2}=o_{\a,3}=0$,
  so that $\sigma(1)=1$. Since, in addition, $p_{\a,2}>p_{\a,3}$, we conclude that $\sigma=\roId$. Finally, since $o_{\a,1}>0$, we conclude that
  $x_0=x_1$. In the case of Bianchi type VI${}_0$, the argument is similar, except if $p_{0,2}=p_{0,3}$ and $o_{0,2}=-o_{0,3}$. In that case $\sigma$
  could be the transposition interchanging $2$ and $3$, and we could have $x_1=\psi_\sigma^-(x_0)$. However, we then still have $x_1=x_0$. In the
  remaining cases, arguments similar to the above yield the conclusion that the map $\psi_\sigma^{\pm}$ relating $x_0$ and $x_1$ has to be the identity.
  To conclude, $x_0=x_1$. For this reason, we below use the notation $o_i:=o_{0,i}$ and $p_{i}:=p_{0,i}$. 

  Next, since $x_1=x_0$, the data on the singularity used when appealing to Proposition~\ref{prop:variableexandunique} are the same for the two
  initial data sets $\mfI_{\infty,\a}$; cf. the existence part of the proof of Theorem~\ref{thm:dataonsingtosolution}. More specifically, if $o_i=0$,
  then $\e_i=0$ and $m_i=0$; if $o_i\neq 0$, then $\e_i:=o_i/|o_i|$ and $m_i:=\ln |o_i|$; $\varkappa_\infty=0$; and $\sigma_i:=p_i-1/3$. The corresponding
  solution to (\ref{eq:thetaprime})--(\ref{seq:phizphioev}) such that the statements of Proposition~\ref{prop:variableexandunique} hold is thus the same
  for the two initial data sets. Since the scalar fields $\phi_\a$ and the $a_{\a,i}$ are constructed solely on the basis of this solution, we conclude that
  $\phi_0=\phi_1$ and that $a_{0,i}=a_{1,i}$; cf. (\ref{eq:g alpha format dev sing prel}). Moreover, the relevant existence intervals are the same, so
  that $t_{0,+}=t_{1,+}$. In particular,
  \begin{equation}\label{eq:g alpha format dev sing}
    g_{\a}=-dt\otimes dt+\textstyle{\sum}_{i}a_i^2(t)\xi_\a^i\otimes\xi_\a^i,
  \end{equation}
  where $a_i:=a_{0,i}=a_{1,i}$; cf. (\ref{eq:g alpha format dev sing prel}).

  If both of the data sets are isotropic, the developments take the form
  \[
    g_{\a}=-dt\otimes dt+a_1^2(t)\msH_\a;
  \]
  cf. (\ref{eq:iso form of development}) and (\ref{eq:g alpha format dev sing}). Since $\chi$ is an isometry from $\mfI_{\infty,0}$ to $\mfI_{\infty,1}$, it follows
  that $\chi\times\roId$ is an isometry between the developments. 

  Next, assume that the data are LRS. Then $o_2=o_3$ and $p_2=p_3$, so that $a_{2}=a_3$; cf. (\ref{eq:g form LRS}) and the adjacent
  text. This means that the metrics take the form
  \[
    g_\a=-dt\otimes dt+a_{1}^2\xi_{\a}^1\otimes\xi_\a^1+a_{2}^{2}(\xi_{\a}^2\otimes\xi_\a^2+\xi_{\a}^3\otimes\xi_\a^3).
  \]
  In order for the data not to be isotropic, either $p_1\neq p_2$ or the diagonal Ricci tensor component corresponding to $e_{0,1}$ is
  different from the other two diagonal Ricci tensor components. 
  This means that $\chi_{*\rod}e_{0,1}=\pm e_{1,1}$, where $(\chi_{*\rod}X)_{x}=\chi_{*}X|_{\chi^{-1}(x)}$. Moreover, (\ref{eq:ezottwoeootwo LRS}) and
  (\ref{eq:abgdecondiso LRS}) have to hold with $\psi$ replaced by $\chi$. Summarising the above observations, it follows that $\chi\times\roId$ is an
  isometry between the developments. 

  Finally, assume that the data are neither isotropic nor LRS. Then, due to Remark~\ref{remark:frame to frame sing}, there is a $\sigma\in S_3$
  and $\e_i\in\{-1,1\}$, $i=1,2,3$, such that $\chi_{*\rod}e_{0,i}=\e_i e_{1,\sigma(i)}$. Since $\chi_{*\rod}$ preserves the Lie bracket, this means that
  $x_1=\psi_{\lambda}^{\pm}(x_0)$, where $\lambda=\sigma^{-1}$. By the arguments given above, $\lambda$ has to equal the identity, except if
  the data belong to $\mS_{\mrVIz}^{\roper}$, in which case $\lambda$ is allowed to permute $2$ and $3$ and $x_1=\psi_{\lambda}^{-}(x_0)$. However, in
  this case $a_2=a_3$. Due to these observations, it is clear that $\chi\times\roId$ is an isometry between the developments.
\end{proof}

\printindex

\end{document}